\documentclass{memo-l}

\usepackage[scaled]{helvet} 
\usepackage{courier} 
\usepackage{eulervm} 
\normalfont
\usepackage[T1]{fontenc}

\usepackage{mathrsfs}
\usepackage{hyperref}
\usepackage{amssymb,amsthm}
\usepackage{amsmidx}
\usepackage{imakeidx}
\usepackage{amsmath}
\usepackage{pict2e}
\usepackage{bbold}
\usepackage{color}

\DeclareSymbolFont{greek}{OML}{cmr}{m}{n}  
\DeclareMathSymbol{\varrho}{0}{greek}{"25}

\newtheorem{theorem}{Theorem}[section]
\newtheorem{corollary}[theorem]{Corollary}
\newtheorem{proposition}[theorem]{Proposition}
\newtheorem{lemma}[theorem]{Lemma}

\theoremstyle{definition}
\newtheorem{definition}[theorem]{Definition}

\theoremstyle{remark}
\newtheorem{remark}[theorem]{Remark}
\newtheorem{remarks}[theorem]{Remarks}

\numberwithin{section}{chapter}
\numberwithin{equation}{section}

\makeindex

\frontmatter

\title{The ${\mathscr P}(\varphi)_2$ Model on the de Sitter Space}

\author[J.~Barata]{Jo\~{a}o C.A.~Barata}
\address{
Universidade de S\~ao Paulo (USP), Brasil}
\email{jbarata@if.usp.br}

\author[C.\ J\"akel]{Christian D.\ J\"akel}
\address{
Universidade de S\~ao Paulo (USP), Brasil}
\email{christian.jaekel@icloud.com}

\author[J. Mund]{Jens Mund}
\address{Departamento de Fisica\\
Universidade de Juiz de Fora, Brasil}
\email{mund@fisica.ufjf.br}

\begin{document}



\thanks{The second author is grateful to 
\begin{itemize}
\item[$ i.)$] Walter Wreszinski, for his kind invitation and to the Universidade 
de S\~{a}o Paulo (USP), Brasil, for the hospitality provided during a visiting professorship in 
2004--2005 and during an extended visit in 2009--2010;
\item[$ ii.)$] Ricardo Baeza, for his support and to the Instituto de Matematica y Fisica, Universidad de Talca, 
Chile, for providing perfect working conditions;
\item[$ iii.)$] IHES, Bures-sur-Yvette, France, for generous hospitality offered  during extended visits in 2007 
and 2008. 
\end{itemize}
The third author acknowledges support by CNPq.  The second and the third author were supported by FAPESP}

\dedicatory{Dedication text (use \\[2pt] for line break if necessary)}

\begin{abstract}
In 1975 Figari, H\o egh-Krohn and Nappi~\cite{FHN} constructed the ${\mathscr P}(\varphi)_2$ model on the two-dimensional de Sitter space. 
Here we complement their work with a number of new results. In particular, we show that

\begin{itemize}
\item[$i.)$] the unitary irreducible representations of $SO_0(1,2)$ for both the principal and the complementary series 
can be formulated on the Hilbert space formed by wave functions supported on the Cauchy surface;   
\item[$ii.)$] for $m> - \tfrac{1}{2 r^2}$ physical infrared problems\footnote{As shown in ~\cite{FHN}, the ultraviolet problems
can be resolved in the same manner as on flat Minkowski space. The number $r>0$ is the radius of the time-zero circle 
in de Sitter space.}
 are absent on de Sitter space;
\item[$iii.)$] the interacting quantum fields satisfy the \emph{equations} of motion in their covariant form; 
\item[$iv.)$] the Haag-Kastler and the time-slice axiom hold true. 
In fact, one can choose an arbitrary space-like geodesic and require that 
the local von Neumann algebras for all double cones with bases on this specific geodesic
are the same for the both free and the interacting theory; 
\item[$v.)$] the \emph{generators} of the boosts and the rotation for the interacting quantum field theory 
arise by contracting the \emph{stress-energy tensor} with the relevant \emph{Killing vector fields} and integrating 
over the relevant line segments. They generate 
a reducible, unitary \emph{representation} of the \emph{Lorentz group} on the Fock space for the free field.
\end{itemize}
In addition, we provide a detailed discussion of the causality structure of de Sitter space and a brief review of 
the representation theory of $O(1,2)$. We describe the free classical dynamical system in both its covariant and canonical form, 
and present the associated quantum one-particle KMS structures.  The ${\mathscr P}(\varphi)_2$ interaction is added on the Euclidean 
sphere and the Osterwalder-Schrader reconstruction is carried out in some detail. 
\end{abstract}

\maketitle

\tableofcontents
\addcontentsline{toc}{chapter}{Abstract\hfill}

\chapter*{List of Symbols}

\quad

\smallskip

{\bf Space-Time}

\bigskip

$(\mathbb{R}^{1+2}, \mathbb{g})$ \hfill  Minkowski space-time in 1+2 dimensions \quad \pageref{dSpage}

$(dS, g)$ \hfill  de Sitter space-time $dS \subset \mathbb{R}^{1+2}$ \quad \pageref{dSpage}

$\mathbb{g}$ \hfill  metric  on Minkowski space $\mathbb{R}^{1+2}$   \quad \pageref{metricpage}

$g = \mathbb{g}_{\upharpoonright dS} $ \hfill  metric  restricted to $dS$     \quad \pageref{metricpage}

${\mathcal C}$ \hfill  a Cauchy surface   \quad \pageref{ccCpage}

$S^1$ \hfill  time-zero circle   \quad \pageref{ccCpage}

$I_+$ \hfill  the half-circle $W _1 \cap S^1$      \quad \pageref{halfcirclepage}

$V^+$  \hfill  forward light-cone in $\mathbb{R}^{1+2}$ \quad \pageref{vpluspage}

$\Gamma^\pm (x)$  \hfill  future  and past of a space-time point $x \in dS$
 \quad \pageref{vpluspage}

${\mathcal O}$ \hfill  a open, bounded space-time region 
 \quad \pageref{cOpage}

${\mathcal O}'$ \hfill  space-like complement of  ${\mathcal O} \subset dS$
 \quad \pageref{cOpage}
 
$W$ \hfill   a wedge in $dS$ \quad \pageref{tWpage}

$W _1$ \hfill  the wedge $\bigl\{ x \in dS \mid x_2 > \vert x_0 \vert \bigr\} $  \quad \pageref{tWpage}

${\mathbb{W}}$ \hfill  the double wedge $W \cup W'$ \quad \pageref{dWpage}

${\mathcal O}_I$ \hfill the double-cone $I''$, with basis $I\subset S^1$   
 \quad \pageref{doubleconepage}

$H^\pm_m$ \hfill  mass hyperboloid in $\mathbb{R}^{1+2}$ \quad \pageref{masshyperboloid}

$P_\tau$  \hfill  horosphere \quad \pageref{horosphere}
	
$W^{(\alpha)}$ \hfill  the wedge $R_0(\alpha)W_1$ \quad \pageref{kldypage}

$g_{\upharpoonright S^1} $ \hfill  metric  restricted to $S^1$   \quad \pageref{seinspage}

${\rm d} l (\psi)$  \hfill  induced surface element on $S^1$ \quad \pageref{dlpage}

\bigskip

\goodbreak
{\bf De Sitter Group}

\smallskip

$O(1,2)$  \hfill de Sitter group, \emph{i.e.}, Lorentz group in 1+2 dimensions \quad \pageref{metricpage}

$SO_0(1,2)$  \hfill  proper, orthochronous de Sitter group \quad \pageref{metricpage}

$R_0$  \hfill  a rotation around the $x_0$-axis  \quad \pageref{rropage}

$\Lambda$ \hfill  an arbitrary element in $SO_0(1,2)$  \quad \pageref{lambdapage}

$\Lambda_1 (t)$ \hfill   the boost which leaves $W_1$ invariant  \quad \pageref{Lambdaalpha}

$\Lambda^{(\alpha)} (t)$ \hfill   the boost which leaves $W^{(\alpha)}$ invariant  \quad \pageref{Lambdaalpha}

$K_0$, $L_1$, $L_2$ \hfill generators of Lorentz transformations \quad \pageref{Lambdaalpha}

$L^{(\alpha)}$ \hfill  generator of the boost $t \mapsto \Lambda^{(\alpha)} (t)$  \quad \pageref{Lambdaalpha}

$T \colon  \mathbb{R}^3 \to \mathbb{R}^3$ \hfill time reflection \quad \pageref{timerefl-page}

$P_1 \colon \mathbb{R}^3 \to \mathbb{R}^3$ \hfill parity reflection \quad \pageref{timerefl-page}

$\Theta_{W}$ \hfill reflection at the edge of the wedge $W$ \quad \pageref{PTwedgepage}

\bigskip
\goodbreak
{\bf Test-Functions on $dS$}

\smallskip

${\mathcal D}_{\mathbb{R}} (dS) $ \hfill  real $C^\infty$-functions with compact support on  $ dS $ \quad \pageref{ccCpage}

$f$, $g$ \hfill  elements of ${\mathcal D}_{\mathbb{R}} (dS) $ \quad \pageref{ccCpage}

\bigskip

{\bf Unitary irreducible representations of $SO_0(1,2)$}

\smallskip

$K_0$, $L_1$, $L_2$ \hfill generators of $SO_0(1,2)$
on $L^2 (\partial V^+, \frac{{\rm d} \alpha}{2 \pi} {\rm d} p_0)$
\quad \pageref{bargmangeneratorpage}

$M$ \hfill  the Casimir operator on the light cone
\quad \pageref{casimir}

$p = (p_0, p_1, p_2) \in \partial V^+$ \hfill coordinates on the light-cone
\quad \pageref{bargmangeneratorpage}

$S (S+1) \, \Phi= \mu^2 \, \Phi $ \hfill KG equation on the light-cone $\partial V^+$
\quad \pageref{lightconecoordinateKGpage}

$\widetilde u (\Lambda)$ \hfill  UIR of $SO_0(1,2)$ 
on $\widetilde {{\mathfrak h}}  (\partial V^+)$ \quad \pageref{umLambdapage}

$\Gamma$ \hfill a path on the forward light cone~$\partial V^+$ \quad \pageref{Gammacpluspage}

${\rm d} \mu_\Gamma$  \hfill  restriction of~${\rm d} \mu_{\partial V^+}$  to a path $\Gamma \subset \partial V^+$ 
\quad \pageref{Gammacpluspage}

\bigskip

{\bf (Pseudo-)Differential Operators}

\smallskip

$\square_{dS}+\mu^2$ \hfill Klein--Gordon operator \quad \pageref{squarepage}

$n$ \hfill future
pointing normal vector field ${\tt  n}(t,\psi) = \,  {\rm cos}_\psi ^{-1} \,\partial_t$ \quad \pageref{vnpage}

$\varepsilon$ \hfill generator of the boosts $t \mapsto \Lambda_1(t)$ \quad \pageref{epsilonpage}

$ {\rm cos}_\psi$ \hfill multiplication operator by  $\cos\psi$ \quad \pageref{cosinepsipage}

\bigskip

\goodbreak
{\bf Covariant Dynamical System}

\smallskip

$\sigma$ \hfill  symplectic form associated to  ${\mathcal E}$  \quad \pageref{sigmapage}

${\mathscr E}$ \hfill  the commutator function for the Klein--Gordon equation  \quad \pageref{ccCfpage}

${\mathfrak u} (\Lambda)$ \hfill  representation of $O(1,2)$  on  $({\mathfrak k}(dS), \sigma)$  \quad \pageref{ccTpage}

${\mathfrak k}(dS)$ \hfill  space of solutions of the Klein--Gordon equation   \quad \pageref{kldypage}

$\mathbb{\Phi}$  \hfill a solution of the Klein--Gordon equation (an element  in ${\mathfrak k}(dS)$)  \quad \pageref{kldypage}

${\mathbb P}$ \hfill  projection from ${\mathcal D}_{\mathbb{R}} (dS) $ to $ {\mathfrak k}(dS)$  \quad \pageref{kldypage}

$\mathbb{f}$ \hfill  solution of the KG equation for $f  \in {\mathcal D}_{\mathbb{R}} (dS)$
\quad \pageref{kldypage}

\bigskip

\goodbreak
{\bf Canonical Dynamical System}

\smallskip
 
$\widehat {\mathfrak k} (S^1)$ \hfill the space of Cauchy data for the Klein--Gordon equation
\quad \pageref{hatckpage} 

$(\mathbb{\phi}, \mathbb{\pi})$ \hfill  Cauchy data (an element of $\widehat {\mathfrak k} (S^1)$) \quad \pageref{hatckpage} 

$\widehat\sigma$   \hfill the canonical symplectic form on $\widehat {\mathfrak k} (S^1)$
\quad \pageref{hatckpage} 

$\widehat {\mathfrak u} (\Lambda)$ \hfill  representation of $O(1,2)$  on   
$(\widehat {\mathfrak k} (S^1), \widehat\sigma)$  \quad \pageref{widehatccTpage}

$\widehat {\mathbb P}$  \hfill  a map from ${\mathcal D}_{\mathbb{R}} (dS) $ to $ \widehat {\mathfrak k} (S^1)$  \quad \pageref{widehatcxC}

\bigskip

{\bf Complex Space-Time}

\smallskip

$dS_{\mathbb{C}}$ \hfill complex de Sitter space \quad \pageref{dSccpage}

${\mathcal T}_\pm $ \hfill forward (backward) tuboid \quad \pageref{tuboidspage}

$S^2$ \hfill Euclidean sphere  \quad \pageref{euclidspherepage}

\bigskip

{\bf Fourier Transformation}

\smallskip

$  ( x_\pm \cdot p \, )^s$ \hfill the Harish-Chandra plane-wave  \quad \pageref{planewavepage}

$ \widetilde {f}_\pm ( p, s)  $ \hfill Fourier transform \quad \pageref{Fouriertransformpage}

${\mathcal F}_{+ \upharpoonright \nu} $ \hfill FH-transformation restricted to the mass shell
\quad \pageref{mass-shell-ft-page}

$\widetilde f_\nu (p )$ \hfill  restriction of the Fourier transformation  to the mass shell \quad \pageref{mass-shell-ft-page}

\bigskip

\goodbreak
{\bf Covariant One-Particle Structure}

\smallskip
 
${\mathfrak h}  (dS)$ \hfill completion of ${\mathcal D}_{\mathbb{R}} (dS)/{\ker ( {\mathbb E}_\mu  {\mathcal F}_+ )} $ \quad \pageref{ophs1page}

$\langle  \, . \, , \, . \, \rangle_{{\mathfrak h} (dS)}$ \hfill scalar product on ${\mathfrak h}  (dS)$  \quad \pageref{ophs1page}

$ u (\Lambda)$ \hfill  unitary irreducible representation of $SO_0(1,2)$ on ${\mathfrak h}  (dS)$ \quad \pageref{umLambdapage2}

$K$ \hfill  maps ${\mathfrak k} (S^1) $ into ${\mathfrak h}  (dS)$ \quad \pageref{umLambdapage2}

$\left( K, {\mathfrak h}  (dS), u  \right)$ \hfill   one-particle
structure for $( {\mathfrak k} (dS) ,\sigma, {\mathfrak u} )$ \quad \pageref{covonepartpage}

\bigskip

{\bf Canonical One-Particle Structure}

\smallskip
$\widehat{{\mathfrak h}}  (S^1)$ \hfill   time-zero Hilbert space  \quad \pageref{chfdpage}

$\langle  \, . \, , \, . \, \rangle_{\widehat{{\mathfrak h}}  (S^1)}$ \hfill  scalar product on $\widehat{{\mathfrak h}}  (S^1)$  \quad \pageref{chfdpage}

$\widehat{K}$ \hfill  maps $\widehat {\mathfrak k} (S^1) $ into $\widehat{{\mathfrak h}}  (S^1)$ \quad \pageref{Kmhatpage}

$\bigl(\widehat{K} , \widehat{{\mathfrak h}}  (S^1), \widehat{u} \bigr) $ \hfill   one-particle
structure for $(\widehat {\mathfrak k} (S^1) ,\widehat \sigma, \widehat   {\mathfrak u} )$ \quad \pageref{Kmhatpage}

$\widehat{u}(\Lambda)$ \hfill  unitary irreducible representation of $SO_0(1,2)$ on $\widehat{{\mathfrak h}}  (S^1)$ \quad \pageref{Umhatpage}
 
\bigskip 
\goodbreak
{\bf Operator Algebras and States}

\smallskip

${\mathfrak W} ({\mathfrak k} , \sigma)$ \hfill  Weyl algebra \quad \pageref{weylalgebrapage}

$\alpha_\Lambda $ \hfill automorphic representation of  $SO_0(1,2)$ on  ${\mathfrak W}({\mathfrak k}(dS),\sigma)$ \quad \pageref{alphapage}

$\bigl({\mathfrak W}(dS), \alpha_\Lambda^\circ \bigr)$ \hfill covariant quantum dynamical system  \quad \pageref{cqds-page}
 
$\bigl(\widehat {\mathfrak W}(dS), \widehat \alpha_\Lambda^\circ \bigr)$ \hfill canonical quantum dynamical system  \quad \pageref{cqds-page}

$\alpha_\Lambda $ \hfill automorphic representation of  $SO_0(1,2)$ on  $\widehat {\mathfrak W}({\mathfrak k}(S^1),
\widehat \sigma)$ \quad \pageref{alphapage}

$\omega^\circ$ \hfill free de Sitter vacuum state  \quad  \pageref{freevacuumstatepage}

$\widehat \omega^\circ$ \hfill free de Sitter vacuum state  \quad  \pageref{freevacuumstatepage}

${\mathscr A}  ( {\mathcal O}) $ \hfill v.~N.~algebra for  the free fields in a double cone ${\mathcal O}\subset dS$ \quad \pageref{AO-page}

${\mathcal R} ( I) $ \hfill v.~N.~algebra for  the free fields in the interval $I \subset S^1$ \quad \pageref{RI-page}

\bigskip 
\goodbreak

{\bf Euclidean space-time}

\smallskip

$S^2$  \hfill Euclidean space-time \quad \pageref{eulidspherepage}

$S_\pm$  \hfill upper (resp.~lower) hemisphere \quad \pageref{spherepluspage}

$S^1$ \hfill time-zero circle \quad \pageref{equatoreuclidpage}

$I_\pm$ \hfill half-circle formed by $ W_1 \cap S^1 $ or  $ W_1' \cap S^1 $ 
\quad \pageref{eulidspherepage}

$I_\alpha$ \hfill  the half-circle $I_\alpha = {\tt R}_{0} (\alpha) I_+$ 
\quad \pageref{halfcirclealphapage}

\bigskip

\goodbreak
{\bf Probability space}

\smallskip

${\mathcal Q}={\mathcal D}'_{{\mathbb R}}(S^2)$ \hfill distributions  \quad  \pageref{dualitypage}

$\Sigma$ \hfill  Borel $\sigma$-algebra on ${\mathcal Q}$ \quad \pageref{dualitypage}

$C (f,g)$ \hfill covariance \quad \pageref{dualitypage} 
 
 ${\rm d} \Phi_{C}$ \hfill Gaussian measure \quad \pageref{dualitypage}
 
 $L^{p}({\mathcal Q}, \Sigma, {\rm d}\Phi_{C})$ \hfill $L^p$ spaces \quad \pageref{dualitypage}
 
\bigskip

{\bf Time-zero fields}

\smallskip

$C_{| s |} (h_{1}, h_{2})$ \hfill  time-zero covariance \quad  \pageref{stcpage} 

$\Phi (\theta, h)$ \hfill  sharp-time field \quad  \pageref{stfpage}

\bigskip

\goodbreak
{\bf Sobolev spaces}

\smallskip

$\mathbb{H}^{\pm 1} (S^2)$ \hfill Sobolev spaces   \quad \pageref{sobolevpage}

$\widehat{\mathfrak h} (S^1)$ \hfill a subspace of $\mathbb{H}^{- 1} (S^2)$ \quad \pageref{sobolevpage}

$e ({S^1}), e \left(\overline{S_\pm}\right)$ \hfill orthogonal projections \quad \pageref{sobolevpage}

\bigskip

{\bf Fock space}

\smallskip

$\Gamma( \mathbb{H}_{{\mathbb C}}^{-1}(S^2))$ \hfill Fock space over the Sobolev space $\mathbb{H}_{{\mathbb C}}^{-1}(S^2)$  \quad \pageref{fockpage}

${\mathcal E}_{\Sigma_{S^1}} := \Gamma \left(e({S^1}) \right)$ 
\hfill conditional expectation on $L^{2}({\mathcal Q}, \Sigma, {\rm d}\Phi_{C})$   \quad \pageref{fockpage}

$\Gamma(\widehat{\mathfrak h} (S^1))$ \hfill	the Fock space ${\mathcal H} \cong \Gamma(\widehat{\mathfrak h} (S^1))$ 
\quad \pageref{hospage}

\bigskip

\goodbreak
{\bf Interaction}

\smallskip

${:} \Phi(f)^{n} {:}_c$ \hfill   Wick ordering  \quad \pageref{wickpage}

$ V \left(\overline{S_+} \right)$ \hfill interaction on the upper hemisphere 
\quad \pageref{interactionspherepage}

${\rm d}\mu_{V}$  \hfill perturbed measure  on the sphere \quad \pageref{interactionspherepage}

${V_0} (\cos_\psi) $ \hfill the interaction on the half-circle $I_+$\quad \pageref{vcospage}

\bigskip

{\bf Fock representations}

\smallskip
$\pi_E $ \hfill Fock representation on $\Gamma(\mathbb{H}_{{\mathbb C}}^{-1}(S^2))\cong L^2(Q, \Sigma, {\rm d} \Phi_C)$  \quad \pageref{fockref-2}

\bigskip

{\bf Canonical  von Neumann algebras}

\smallskip

${\mathcal U} (S^1)$ \hfill abelian algebra of  time-zero fields in ${\mathcal H}$ 
\quad \pageref{abelianalgebraospage}

\bigskip

{\bf Osterwalder-Schrader Hilbert spaces}

\bigskip

$ \Theta := \Gamma ( T_*) $ \hfill time reflection on $L^{2}({\mathcal Q}, \Sigma, {\rm d}\Phi_{C})$   \quad 
\pageref{hospage}

$ {\mathcal H}$ \hfill  Osterwalder-Schrader Hilbert space \quad \pageref{hos-page}

${\mathcal V}$ \hfill canonical map from $ L^{2}({\mathcal Q}, \Sigma_{\overline { S_+}}, {\rm d} \Phi_C) $ 
to ${\mathcal H}$ \quad \pageref{hos-page}

$\Omega$  \hfill free vacuum vector in ${\mathcal H}$ \quad \pageref{vacumvecpage}

$A^{\rm os}$ \hfill  multiplication operator on  $L^2 ({\mathcal Q}, \Sigma_{\overline { S_+}}, {\rm d} \Phi_C)$ \quad \pageref{abelianalgebraospage}

$K_0^{\rm os}$ \hfill generator of the rotation on ${\mathcal H}$ \quad \pageref{knullosproppage}	

$\Omega_{\rm int}$  \hfill interacting vacuum vector in ${\mathcal H}$ \quad \pageref{intvacuumpage}

\bigskip

{\bf Generalized path spaces}

\smallskip

$({\mathcal Q}, \Sigma, \Sigma_{0}, {\rm U}(t), \Theta, \mu)$   \hfill generalised path space  \quad  \pageref{gppage}

${\rm U} ^{(\alpha)} (\theta) $ \hfill  measure preserving automorphisms \quad  \pageref{gppageds}

$\Sigma=\bigvee_{\theta \in
S^{1}}{\rm U}^{(\alpha)} (\theta)\Sigma^{(\alpha)}$ \hfill the Borel $\sigma$-algebra on ${\mathcal Q}$
  \quad  \pageref{gppageds}

$\Sigma^{(\alpha)}$ \hfill
smallest $\sigma$-algebra for which 
$\Phi(0,h)$  
is measureable
  \quad  \pageref{sigmaalphapage}

\bigskip

{\bf Local symmetric semigroups}

\smallskip

$\bigl(P^{(\alpha)}(\theta), {\mathscr D}^{(\alpha)}_{\theta} \bigr)$ \hfill local symmetric semigroup for the free dynamics\quad  \pageref{lssgpage}

$L^{(\alpha)}$ \hfill generator of the free boost $\Lambda^{(\alpha)}$ \quad  \pageref{lssgpage}

$ J^{(\alpha)}$ \hfill  modular conjugation associate to 
$({\mathcal R} (I_\alpha), \Omega)$ \quad \pageref{Jalpha}

$( P^{(\alpha)}_V (\gamma), {\mathscr D}_{\theta, V} \bigr)$ \quad  \hfill loc.~sym.~semigroup for interacting dynamics \quad \pageref{intlssgpage}

$L_V^{(\alpha)}$ \hfill generator of the interacting boost $\Lambda^{(\alpha)}$  \quad  \pageref{intlssgpage}

$ J_V^{(\alpha)}$ \hfill  modular conjugation associate to 
$({\mathcal R} (I_\alpha), \Omega_{\rm int})$ \quad \pageref{TTintpage}

\bigskip

\goodbreak
{\bf Virtual representations}

\smallskip

$(G, H , \vartheta)$ \hfill symmetric space  \quad  \pageref{virtreppage}

${\mathfrak g}, {\mathfrak k} , {\mathfrak m}$ \hfill Lie algebras  \quad  \pageref{virtreppage}

${\mathfrak g}^*$ \hfill dual symmetric Lie algebra  \quad  \pageref{virtreppage}

$ \wp$ \hfill virtual representation  \quad  \pageref{virtreppage2}

\bigskip

\goodbreak
{\bf Auxiliary Hilbert spaces for the interacting measure}

\smallskip

$L^2( {\mathcal Q}, \Sigma, {\rm d}\mu_{V} ) $ \hfill $L^p$ spaces for the interacting measure \quad \pageref{interactingsphereLP}

$ {\mathcal H}_V$ \hfill completion of $L^{2}({\mathcal Q}, \Sigma_{\overline { S_+}}, {\rm d} \mu_V) /{\mathcal N}_V$ \quad 
\pageref{inthilbertpage}

${\mathcal V}_V $ \hfill canonical map from $ L^{2}({\mathcal Q}, \Sigma_{\overline { S_+}}, {\rm d} \mu_V)$ to ${\mathcal H}_V$ \quad \pageref{inthilbertpage}

$ \Omega_V$ \hfill interacting vacuum vector in ${\mathcal H}_{V}$ \quad \pageref{inthilbertpage}

$A_V$  \hfill  multiplication operator on  $L^2 ({\mathcal Q}, \Sigma_{\overline { S_+}}, {\rm d} \mu_V)$ 
\quad \pageref{inthilbertpage}

$ {\mathcal U}_V (S^1)$ \hfill abelian algebra of  interacting time-zero fields in ${\mathcal H}_{V}$ 
\quad \pageref{inthilbertpage}

${\mathbb V}$ \hfill  a unitary operator from ${\mathcal H}_V$ to ${\mathcal H}$ \quad \pageref{qtotot-page}

\bigskip

{\bf Unitary groups}

\smallskip

$\widehat{u}$  \hfill unitary irreducible representation of $SO_0(1,2)$ on~$\widehat{\mathfrak h} (S^1)$ \quad \pageref{Umhatpage}	

$\widehat{U}$  \hfill a unitary representation of $SO_0(1,2)$ on~${\mathcal H}$ \quad \pageref{umospage}	

$\widehat {\rm U}_{\rm int}$ \hfill interacting representations of $SO_0(1,2)$ on ${\mathcal H}$ \quad \pageref{unitrospage}

\bigskip
\goodbreak

{\bf Symbols Appendices}

\smallskip

$(K, \sigma, T_t)$ \hfill classical dynamical system \quad  \pageref{page48}

$\bigl({\mathfrak h}_{\scriptscriptstyle \rm AW},K_{\scriptscriptstyle \rm AW},U_{\scriptscriptstyle \rm AW} (t) \bigr)$ \hfill Araki-Woods one-particle structure
\quad  \pageref{page49}

\chapter*{Preface}

There is well-founded trust  in quantum field theory on curved space-times  and its ability to predict 
and explain many of the exciting astrophysical and cosmological phenomena currently discovered  
in one of the most thriving branches of experimental physics. But despite substantial effort, little is 
known about the physics of quantum fields on \emph{general} curved space-times beyond the scope of 
(re\-normalised) perturbation theory\footnote{Even in flat space, the ${\mathscr P}(\varphi)_2$-models, 
do \emph{not} allow {\em Borel summation} of the perturbation series, unless 
the order of the polynomial is less or equal four, as the number of Feynman diagrams grows to rapidly 
for polynomials of higher order. Although in each order of perturbation theory there are no divergences, 
the Green's functions are not analytic in the coupling constant, neither are the proper self energy and 
the two-particle scattering amplitude \cite{AJ1}.
For the $\varphi^4_2$-model on Minkowski space, perturbation theory yields a  Borel summable 
asymptotic series for the Schwinger functions.
}  (see,  \emph{e.g.},~\cite{BFr, BrFr, BFK, H1, H2, HW1, HW2, HW3}  
and references therein). 

However, for  {\em static} space-times\index{static space-time}, \emph{i.e.}, 
solutions to Einstein's equations for which the metric has the form
	\[
		g = \lambda \; {\rm d} t \otimes {\rm d} t 
		- \sum_{i, j=1}^3 \lambda^{-1} h_{ij} \; {\rm d} x_i \otimes {\rm d} x_j   \; , 
	\]
with both $\lambda$ and $h_{ij}$ time-independent\footnote{Strictly speaking, de Sitter space~$dS$ is 
not a static space-time, unless it is restricted to the past $\Gamma^-(W)$ of a wedge $W$ 
(which itself is conformally equivalent to Minkowski space).}, some progress has been made in recent years.
These space-times allow analytic continuations to Riemannian manifolds, and
Ritter and Jaffe \cite{JR1, JR2, JR3} pioneered a non-perturba\-tive, constructive approach to interacting fields
defined on them. 
They have shown that one can reconstruct a unitary representation of the isometry group 
of the static space-time under consideration, starting from the corresponding Euclidean field theory~\cite{JR2}.
Some progress has also been made in case the space-time is
asymptotically flat,  see, \emph{e.g.},~\cite{D84, GeP, GM}.

For \emph{maximally symmetric} space-times, like the (two-dimensional) de Sitter space,
the situation is even more favourable. In fact, in 1975 Figari, H\o egh-Krohn and Nappi~\cite{FHN} 
constructed the first (and up till now the only) 
interacting quantum field theory on a curved space-time, the so-called ${\mathscr P}(\varphi)_2$~model on the de Sitter 
space. In this work we reconsider the contribution of Figari, H\o egh-Krohn and Nappi~\cite{FHN} in the light of more recent work 
by Birke and Fr\"ohlich~\cite{BF}, Dimock~\cite{D} and Fr\"ohlich, Osterwalder and Seiler~\cite{FOS}.
We provide a detailed and very explicit, non-perturbative description of the ${\mathscr P}(\varphi)_2$ 
model on de Sitter space. 
Euclidean methods play an essential role in our approach, despite the fact that they are not available on general 
curved space-times. We, however, would like to mention that in forthcoming complementary work \cite{BMJ}, 
they play a less significant role.

Let us add some comments on the relevance of this particular model. 
In our opinion, the role of ${\mathscr P}(\varphi)_2$ model in quantum field theory 
may well be compared with the role the Ising model plays
in (quantum) statistical mechanics or the role  $SL(2, \mathbb{R})$ plays in harmonic analysis. 
The various ${\mathscr P}(\varphi)_2$ models were the 
first interacting quantum field theories (in Minkowski space), which gained a precise mathematical meaning
and up till now they remain the most thoroughly studied models in the axiomatic 
framework. The original construction of these models (without cutoffs) is due to  
Glimm and Jaffe~\cite{GJ1}\cite{GJ2}\cite{GJ3}\cite{GJ4}\cite{GJ5}\cite{GJ6}. 
An enormous amount of work has been invested to understand 
the scattering theory, the bound states, the low energy particle structure and the properties of the 
correlation functions of these models (see the books by Glimm and Jaffe~\cite{GJ}\cite{GJcp}  
and Simon \cite{S}, and the references therein).  The ${\mathscr P}(\varphi)_2$ models are also 
the only interacting quantum field theories, for which the non-relativistic limit (including bound states) 
has been analysed in detail, demonstrating that the low energy regime of these models can be equally 
well  described by non-relativistic bosons interacting with $\delta$-potentials \cite{D57}\cite{SZ}. 
In addition, Hepp  demonstrated how one can recover the classical field equations for the 
${\mathscr P}(\varphi)_2$ models by taking the classical limit \cite{Hepp}.

\bigskip

Finally, we would like to express our deep gratitude to the mathematics and physics community. 
This work would not be possible without the continuing efforts of colleagues working in differential 
geometry~\cite{Frankel, KN}, harmonic analysis~\cite{DyMcK, Folland, He, Mo1, Vara}, complex analysis in several 
variables~\cite{Fa, Hoer1, Vlad}, operator algebras \cite{BR, KR,Takesaki}, 
the representation theory of semi-simple Lie groups~\cite{Ba, BaF, Knapp, Mack, T, Vil}, 
the theory special functions \cite{Lebedev, snow}, 
axiomatic quantum field theory~\cite{SW, Jost}, local quantum field theory~\cite{A0, H}, constructive quantum 
field theory~\cite{GJ, GJcp, S} and quantum statistical mechanics \cite{BR, Ruelle}. 
Of course, the references given here can only serve as an entry point to the literature
as the scope of the material, on which this work is firmly based on, is exceptionally wide.

\vskip .5cm

\aufm{Jo\~{a}o C.A.~Barata  ,
Christian D.\ J\"akel \&
Jens Mund}

\mainmatter

\part{De Sitter Space}

\chapter{De Sitter Space as a Lorentzian Manifold}
\label{geometry}

We start with a few historical remarks. In 1915 Albert Einstein published his theory of gravitation. The Einstein equations\index{Einstein equations},
	\[
		\underbrace{R_{\mu \nu} - \frac{1}{2} R \, g_{\mu \nu} }_{\doteq G_{\mu \nu}}
		= 8 \pi \, T_{\mu \nu} \; ,  \qquad \mu, \nu = 0, 1, \ldots, 3,
	\]
describe the curvature\index{curvature} of space-time resulting from the distribution of matter fields in space-time.
They relate the metric tensor $g_{\mu \nu}$ and the stress-energy tensor $T_{\mu \nu}$. 
The Ricci tensor\index{Ricci tensor} $R_{\mu \nu}$ and the scalar curvature\index{scalar curvature} 
$R$ both depend only on the metric tensor $g_{\mu \nu}$. 
Given a particular model, one can obtain the stress-energy tensor\index{stress-energy tensor}
	\[
		T_{\mu \nu} = \frac{2}{\sqrt{|g|}} \frac{ \delta S_{m} }{\delta g_{\mu \nu}} 
	\]
from Hilbert's classical prescription of varying the
action $S_{m}$ representing the matter fields with respect to the metric tensor\footnote{The Einstein 
equations themselves may be obtained by demanding that $S_g + \kappa S_m$ is stationary with 
respect to variations of $g_{\mu \nu}$, where $S_g = \frac{1}{2} \int \sqrt {|g|} \, R $.}\index{metric tensor}.   

In 1916 Einstein (re-) introduced a positive (\emph{i.e.}, repulsive) cosmological constant $\Lambda>0$ in 
the Einstein equations, requesting\footnote{In space-time dimension two, the Einstein tensor $G_{\mu \nu}$  is always zero.
Nevertheless, $R$ may be non-zero. 
Note that there is no Bianchi identity in two dimensions.} 
	\[
		R_{\mu \nu} - \frac{1}{2} R \, g_{\mu \nu} + \Lambda \, g_{\mu \nu} 
		= 8 \pi \, T_{\mu \nu} \; , 
	\]
in an attempt to ensure the existence of static solutions. Only a few months later, in 1917, de Sitter showed 
that for $T_{\mu \nu}=0$ (\emph{i.e.}, the empty space), the new constant $\Lambda >0$ leads to a 
universe, which undergoes accelerated expansion \cite{deS1,deS2}. Einstein first discarded de 
Sitter's solution as physically irrelevant because it is not globally static \cite{Schu}. 
However, experimental evidence soon suggested that on a large scale our universe is isotropic, 
homogeneous and indeed undergoing accelerated expansion. The latter may be attributed to the 
existence of a positive cosmological constant \cite{Riess}. 

\section{The metric and the isometry group}

De Sitter space\index{de Sitter space} $( dS , g )$ is the \emph{Lorentzian manifold}\index{Lorentzian manifold} 
analog of the Euclidean sphere. It is maximally symmetric and has constant negative curvature. 
In more than two space-time dimensions, it is simply-con\-nected. 
In two dimensions, $dS$
can be viewed as a one-sheeted \emph{hyperboloid}\index{hyperboloid}, embedded in $1+2$-dimensional Minkowski 
space~$\mathbb{R}^{1+2}$. 

\subsection{Embedding de Sitter space into $\mathbb{R}^{1+2}$}
Following \cite{Schr}, we identify de Sitter space with the submanifold
\color{black}
\label{dSpage}
	\begin{equation}
		\label{eqdSMin}
		dS \doteq \left\{  
		x \equiv (x_0, x_1, x_2) \in \mathbb{R}^{1+2} 
		\mid 
		x_{0}^{2} - x_{1}^{2} - x_{2}^{2} = - r^2 \right\} \; , 
		\quad r >0 \; .
	\end{equation}
Unless the radius $r$ of the time-zero circle plays a significant
role, we will suppress the dependence on $r$. We denote the points of the  $1+2$-dimensional Minkowski
space~$\mathbb{R}^{1+2}$ as either triples $(x_0, x_1, x_2)$ or column vectors
$\left(\begin{smallmatrix} x_0 \\ x_1 \\ x_2\end{smallmatrix}
\right)$, which ever is more convenient. The point $o \equiv (0, 0, r )\in dS$ is called the 
\emph{origin}\index{origin} of~$dS$.  

We denote by $S^1 \subset dS$ the ``time zero'' circle 
\label{ccCpage}
	\begin{equation} 
		\label{SeinsC}
		S^1 \doteq  \bigl\{ ( 0, r \sin \psi, r \cos \psi) \in 
		\mathbb{R}^{1+2} \mid - \tfrac{\pi}{2} \le \psi < \tfrac{3\pi}{2}  \bigr\}   
	\end{equation} 
and by $I_+$ (respectively, by $I_-$) the open subset of $S^1$ with
positive (respectively, negative) $x_1$ coordinate: 
\label{halfcirclepage}
	\begin{equation} 
		\label{halfC}
		I_+\doteq \big\{(0, r \sin \psi, r \cos \psi ) \in 
		\mathbb{R}^{1+2} \mid - \tfrac{\pi}{2} \le \psi < \tfrac{\pi}{2} \big\}
	\end{equation} 
and $I_-\doteq \big\{(0, r \sin \psi, r \cos \psi ) \in 
		\mathbb{R}^{1+2} \mid  \tfrac{\pi}{2} \le \psi < \tfrac{3\pi}{2}\big\}$.

\subsection{The metric}
The \emph{metric}\index{metric} on $dS$ equals the induced metric 
$g= \mathbb{g}_{\upharpoonright dS } $, with   
\label{metricpage}
	\begin{equation}
		\label{metrik}
		\mathbb{g} = {\rm d} x_0 \otimes {\rm d} x_0 
		- {\rm d} x_1 \otimes {\rm d} x_1 - {\rm d} x_2 \otimes {\rm d} x_2   
 	\end{equation}
the metric of the ambient space $(\mathbb{R}^{1+2},  \mathbb{g})$. We denote the Minkowski  
product of two vectors $x, y$ in $\mathbb{R}^{1+2}$ by $x \cdot y \in \mathbb{R}$. 

\subsection{The Isometry Group}

The isometry group of $dS$ is $O(1,2)$. Its linear action on the ambient
space $\mathbb{R}^{1+2}$ is given by $3\times 3$-matrices acting on vectors
$\left(\begin{smallmatrix}x_0\\ x_1 \\ x_2\end{smallmatrix}\right)\in\mathbb{R}^{1+2}$.  The group 
	\[ 
		O(1,2) = O^\uparrow_+ (1,2)  \cup O^\downarrow_+ (1,2)  \cup O^\uparrow_- (1,2) 
				\cup O^\downarrow_- (1,2) 
	\]
has  four connected components~\cite{SW}, namely those (distinguished by $\pm$),
which preserve or change the 
orientation and those (distinguished by $\uparrow \downarrow$), which preserve or change the 
time orientation. Group elements, which preserve the orientation, are called {\em proper}. 
Lorentz transformations, which preserve the time orientation, are called 
{\em orthochronous}. 
The connected component containing the identity is the \emph{proper, orthochronous 
Lorentz group}\index{Lorentz group}, denoted as $SO_0 (1,2) \equiv  O^\uparrow_+ (1,2)$. 
The group $SO_0(1,2)$ acts transitively on the de
Sitter space $dS$. $SO_0(1,2)$ has three uniparametric subgroups 
leaving the coordinate axes in $\mathbb{R}^{1+2}$ invariant: the \emph{rotation} subgroup
$\{R_0(\alpha) \mid \alpha\in [0, 2 \pi)\}$, with
\label{rropage}
	\[
		R_0(\alpha) \; \doteq \; \begin{pmatrix}
					1 &  0 &0 \\
					0 &  \cos \alpha &  - \sin \alpha  \\
					0  & \sin \alpha & \cos \alpha   
				\end{pmatrix} \;,
	\]
and the two subgroups of
\emph{Lorentz boots} $\{\Lambda_1(t) \mid t\in\mathbb{R} \}$ and
$ \{\Lambda_2(s) \mid s\in\mathbb{R} \}$, with
\label{lambdapage}
	\[
		\Lambda_1 (t)   \; \doteq \;   \begin{pmatrix}
				 \cosh t  &  0 &\sinh t \\
     					  0  &  1 & 0  \\
 				  \sinh t &  0 & \cosh t   
				\end{pmatrix} 
	\; \mbox{ and } \;
		\Lambda_2(s) \doteq \begin{pmatrix}
				\cosh s &  \sinh s &0 \\
				\sinh s &  \cosh s &0  \\
					0 & 0 & 1  
				\end{pmatrix} \;.
	\]
According to our convention, the boosts $\Lambda_1 (s)$ (respectively,  $\Lambda_2(t)$) keep 
the $x_1$-axis (respectively,  the $x_2$-axis) invariant, and therefore correspond  to  boosts
in the $x_2$-direction  (respectively, in the $x_1$-direction).

\subsection{Generators}

The \emph{generators}\index{generator} of  
the \emph{boosts}\index{boost} $\mathbb{R}\ni t \mapsto \Lambda_{1}(t)$, $\mathbb{R}\ni s \mapsto \Lambda_{2}(s)$
and the rotations $ [0, 2 \pi )\ni \alpha \mapsto R_{0}(\alpha)$ are 
	\[
		\qquad L_1  =  
			\begin{pmatrix}
							0 &  0 & 1 \\
							0 &  0 & 0  \\
							1 & 0 & 0   
				\end{pmatrix}  , 
		\; L_2  =  
			\begin{pmatrix}
							0 &  1 & 0 \\
							1 &  0 & 0  \\
							0 & 0 & 0   
			\end{pmatrix} 
		\;
				\text{and} \; \; 
			K_0  =  
			\begin{pmatrix}
							0 &  0 & 0 \\
							0 &  0 & -1  \\
							0  & 1 & 0   
			\end{pmatrix} ,  
		\]
respectively. We will also refer to the generators
	\label{Lambdaalphapage}
		\begin{equation} 
		\label{Lalpha}
			L^{(\alpha)} \doteq \cos \alpha \; L_1 + \sin \alpha  \; L_2 \; , 
			\qquad \alpha \in [ 0, 2 \pi ) \; , 
		\end{equation} 
of the boosts  
		\begin{equation*} 
			\label{Lambdaalpha}
				\Lambda^{(\alpha)} (t)= 
				R_{0}(\alpha) \Lambda_{1}(t) R_{0}(-\alpha) \; , 
				\qquad t \in \mathbb{R} \; . 
		\end{equation*} 
Note that $L^{(0)}= L_1$ and $L^{(\pi/2)} = L_2$. 
The \emph{Casimir operator} 
	\[
		C^2 = - K_0^2 + L_1^2 + L_2^2 = 2 \cdot \mathbb{1}_3 \; , 
	\]
with~$\mathbb{1}_3$ the unit $3 \times 3$-matrix, is an element in the centre of the universal enveloping 
algebra\footnote{An analog results holds in $SO(3)$, where the Casimir operator equals 
$2 \cdot \mathbb{1}_3$ as well.} of the Lie algebra $\mathfrak{so}(1,2)$. 

\bigskip
We will continue our discussion of the Lorentz group $SO_0(1,2)$ in Chapter \ref{isometrygroup}.

\section{The causal structure}

\subsection{Time-like and space-like curves}
The intrinsic \emph{causal structure}\index{causal structure} of $dS$ coincides with the 
one inherited from the ambient Minkowski space.  This allows us to freely use the standard terminology.
In particular, we call a smooth curve $t \mapsto \gamma (t)$ on $dS$  (with nowhere 
vanishing tangent vector $\dot \gamma$) {\em causal, time-like, light-like}\index{time-like} 
and {\em space-like}\index{space-like}, according 
to whether the tangent vector  satisfies
	\[ 
		0 \le \dot \gamma \cdot \dot \gamma  \; , 
		\quad 0 < \dot \gamma \cdot \dot \gamma  \;, 
		\quad 0 = \dot \gamma \cdot \dot \gamma  \;, 
		\quad \dot \gamma \cdot \dot \gamma < 0 \;, 
	\]
everywhere along the curve. Similarly, a point $y \in dS$ is called {\em causal, time-like, light-like} 
and {\em space-like separated} to $x\in dS$, if $(y-x) \cdot (y-x) $ is larger or equal than, larger than, 
equal to  or  smaller than zero, respectively. Since $x \cdot x = y \cdot y = -r^2$, these notions are equivalent to
	\[
		x\cdot y \le - r^2, \quad x\cdot y < -r^2, \quad x\cdot y = -r^2, \quad -r^2 < x\cdot y \; , 
	\]
respectively. 

\subsection{The future and the past}
The \emph{future}\index{future} $\Gamma^+ ( x )$and the \emph{past}\index{past} $\Gamma^- ( x )$ of a  point $ x \in dS$ 
are\footnote{In particular, $\Gamma^\pm (0,0,r) = \{ y \in dS \mid \pm y_0 > 0, y_2 \ge r \}$.} 
given by
	\begin{equation} 
		\label{lklkl}
			\Gamma^\pm ( x ) \doteq \bigl\{  y \in dS \mid \pm ( y-  x ) 
			\in \overline{V^+} \bigr\} \; , 
	\end{equation} 
where the bar in \eqref{lklkl} denotes the closure of the \emph{future cone}\index{future cone}\label{vpluspage}
	\[
		V^+ \doteq \bigl\{   y \in \mathbb{R}^{1+2} \mid  y \cdot  y >0,  y_0 > 0 \bigr\}  
	\]
in the ambient space $(\mathbb{R}^{1+2}, g)$.  The boundaries of the future (and the 
past) are given by two \emph{light rays}\index{light ray}, which form the intersection of $ dS $ with a Minkowski  forward (respectively, 
backward) \emph{light cone}\index{light cone}
	\begin{equation*}  
		\label{cone}
				C^\pm ( x ) 
				= \bigl\{  y \in  dS \mid ( y -  x ) \cdot  ( y -  x ) = 0, \;
				\pm ( y_0 - x_0) > 0 \bigr\}  
	\end{equation*}
with apex at $x$. As it turns out, these two light rays  are also given\footnote{For the origin $o$, the light 
rays $\{ o + \lambda (\pm 1, 0 , 1) \mid \lambda \in \mathbb{R} \}$ 
are given by the intersection of $dS$ with the plane $\{ x \in \mathbb{R}^{1+2} \mid x_2= r \}$.}
by the intersection of $dS$ with the tangent plane at $x \in dS$. They  
separate the future, the past and the space-like regions
relative to the point $x$. 
The forward light cone $C^+ ( (0,0,0) ) $ with apex at the origin coincides with the boundary set $\partial V^+$
of the forward cone. 

\goodbreak

\subsection{Cauchy surfaces}
De Sitter space is \emph{globally hyperbolic}\index{globally hyperbolic}, \emph{i.e.}, it has no time-like closed curves 
and for every pair of points $ x,  y \in dS$ the set
	\begin{equation*} 
		\label{globalhyper}
 			\Gamma^- ( x ) \cap  \Gamma^+ ( y)
	\end{equation*} 
is compact (eventually empty). 
These two properties imply 
that $dS$ is diffeomorphic to $\mathcal{C} \times  \mathbb{R} $, 
with $\mathcal{C}$  a \emph{Cauchy surface}\index{Cauchy surface} for $dS$ (see, \emph{e.g.},~\cite{BS}). It is convenient 
to choose $\mathcal{C} = S^1$;  see \eqref{SeinsC}.
One may arrive at this choice by first choosing an arbitrary point $ x \in dS$ and a space-like 
\emph{geodesic}\index{geodesic}\footnote{In the presence of a metric, a geodesic can be defined as the curve joining 
$  x$ and $ y $ with maximum possible length in time --- for a time-like curve --- or the minimum 
possible length in space --- for a space-like curve. 
The null-geodesics on the de Sitter space are 
light rays, \emph{i.e.}, straight lines.}  
$\mathcal{C}$ passing through $ x$, and then introducing coordinates in \eqref{eqdSMin} such 
that~$\mathcal{C}$ equals~\eqref{SeinsC}. 

\subsection{Space-like complements and causal completions}
The complement of the union $\Gamma^+ ( x) \cup \Gamma^- ( x )$  
consists of space-like points.
The \emph{space-like complement}\index{space-like complement} of a simply connected set $\mathcal{O} \subset dS$ is the set 
\label{cOpage}
	\begin{equation*} 
	\label{2.12}
		 \mathcal{O}' \doteq \left\{  y \in dS 
		 \mid  y \notin  \Gamma^+ ( x)  \cup  \Gamma^- ( x)  \;
		\; \forall   x \in \overline{\mathcal{O} }\right\} \; . 
	\end{equation*} 
The \emph{causal completion}\index{causal completion} $\mathcal{O}''$ of $\mathcal{O}$ is defined as the space-like 
complement of~$\mathcal{O}'$. A subset  $\mathcal{O} \subset dS$ is called \emph{causally complete}\index{causally complete}, 
if $\mathcal{O}'' = \mathcal{O} $. (Note that one always has $\mathcal{O} \subset \mathcal{O} ''$.) 

\begin{remark}These notions apply as well to subsets of lower dimension, 
\emph{e.g.},  line-segments  in $dS$. For example, one can easily compute the causal 
completion of an open interval $I\subset S^1$: set
	\[
		x(\psi_\mp) = ( 0, r \sin \psi_\mp, r \cos \psi_\mp) \; , \qquad \text{with} \quad 0\le \psi_- < \psi_+ <\pi \; . 
	\]
The two intersecting (half-) light rays passing through $(0, r \sin \psi_\pm, r \cos \psi_\pm)$ are 
	\begin{equation}
	\label{intersectinglightrays}
				R_0 ( \psi_\mp )
									\left[ 
				\begin{pmatrix}
					0\\
					0\\
					r
				\end{pmatrix}
	+ \lambda 	\begin{pmatrix} 
					 1 \\
					\mp 1\\
					0  
				\end{pmatrix} \right] = 
								\begin{pmatrix}
					\lambda \\
					- r
					\sin\psi_\mp \mp \lambda \cos \psi_\mp \\
					r
					\cos \psi_\mp \mp \lambda \sin \psi_\mp
				\end{pmatrix} \; , 
	\end{equation}
with $\lambda  >0$. They intersect at 
	\begin{equation}
		\label{height}
		\lambda =  	r \tan \left( \tfrac{\psi_+ - \psi_-}{2}\right)  . 
	\end{equation}
Now, any space-like geodesic can be identified with $S^1$ by applying a coordinate transformation.
Therefore the causal completion of an open interval $I$ on an \emph{arbitrary} 
space-like geodesic  is a bounded
space-time region in $dS$, if the length $| I |$ (measured by inserting the endpoints of the interval 
$I$ in~\eqref{dlength}, see below) is less than $\pi \, r $. 
\end{remark}

\section{Geodesics and geodesic distances}
\label{geodesics}

De Sitter space is \emph{geodesically complete}\index{geodesically complete}, \emph{i.e.}, the affine parameter of any 
geodesic passing through an arbitrary point $x \in dS$ 
can be extended to reach arbitrary values. However, given \emph{two} points $  x, y \in dS$, one may ask whether 
there exist geodesics passing through \emph{both} $x$ and $y$: 
\begin{itemize}
\item[$ i.)$] if $y$ is  time- or light-like to the antipode~$-x$ of $x$, then there is \emph{no} geodesic 
passing through both points $x$ and $y$; 
\item[$ ii.)$] the  case $ y = -x$ is degenerated, as \emph{every} space-like geodesics passing 
through $x$ also passes through $- x$;
\item[$ iii.)$] in all the other cases, there exists\footnote{For a proof, we refer the reader to 
\cite[Bemerkung (4.3.14)]{Thirring}.} a \emph{unique} geodesic passing 
through $ x$ and~$ y $. It is a connected component of the intersection 
of~$dS$ with the plane in $\mathbb{R}^{1+2}$ passing through $ x, y $ 
and~$0$~\cite{oneil}. 
\end{itemize}

\begin{remark} If a \emph{time-like} curve is contained in the intersection of $dS$ with a plane 
\emph{not} passing through the origin, then it describes the trajectory of a \emph{uniformly accelerated} observer.
\end{remark}

\subsection{Geodesic distance}
If two points are connected by a geodesic, a \emph{geodesic distance}\index{geodesic distance} can be defined:
\begin{itemize}
\item[$i.)$] if $x$ and $y$ are space-like to each other \emph{and}\footnote{If $y$ is  time- or light-like to 
the \emph{antipode}~$-x$ (\emph{i.e.}, $x \cdot  y >   r^2 $), then $d(x,y)$ is not 
defined.} $ |  x \cdot  y  |  \le  r^2 $,  
a \emph{spatial distance}\index{spatial distance}  
	\begin{equation} 
		\label{dlength}
			d( x ,  y ) \doteq 
			r \arccos \left(-   \tfrac{x \cdot  y}{r^2} \right) 
	\end{equation} 
is defined as the length of the arc on the ellipse 
connecting $x$ and~$y$.
Note\footnote{Recall that the principle values of the function $[-1,1] \ni z \to \arccos (z)$  
are monotonically decreasing between $\arccos (-1) = \pi$ and $\arccos (1) = 0$.} 
 that $d( x ,  x )=0$, iff $x \cdot  x = - r^2 $;

\item[$ii.)$] if $x$ and $y$ are time-like to each other,
a \emph{proper time-difference}\index{proper time-difference} 
	\begin{equation} 
		\label{dlength2}
			d( x ,  y ) \doteq 
			r \; {\rm arcosh} \left(-   \tfrac{x \cdot  y}{r^2} \right) 
			= r \ln \left(-   \tfrac{x \cdot  y}{r^2} + \sqrt{ \tfrac{|x \cdot  y|^2}{r^4} -1} \, \right)  \;  .
	\end{equation} 
is defined as the length of the arc on the hyperbola connecting $x$ and~$y$.
\end{itemize}

\section{Wedges and double cones}

If $x \in dS$, the point $-x$, called the \emph{antipode}\index{antipode}, is in $dS$ as well. 
The light rays going through $x$ and $-x$ lie in the tangent planes at $x$ and $-x$, respectively. 
These tangent planes are parallel to each other. It follows that  a point $x \in dS$ determines 
four closed regions, namely $\Gamma^\pm ( \pm x)$. Since  
	$
	\Gamma^{+}( \pm x) \cap \Gamma^{-} (\pm x) = \{ \pm x \} $,
their union consists of two disjoint, connected components. 
The complement of the union of these two sets consists of two open and disjoint sets, which 
we call \emph{wedges}. 

\subsection{Wedges}
The points $(0, \pm r, 0) \in dS$ are the \emph{edges}\index{edge} of the wedges
	\[ 
		W _1\doteq \bigl\{  x \in dS \mid x_2 > |x_0 | \bigr\} \quad \text{and} \quad
		W _1' \doteq \bigl\{  x \in dS \mid x_2 < |x_0 | \bigr\}  \; .
	\]
Since the proper, orthochronous Lorentz group $SO_0(1,2)$ is transitive\footnote{In fact, 
the orbit $\{g x \mid g \in SO_0(1,2)\}$ of \emph{any} point $x \in dS$ is all of $dS$.} 
on the de Sitter space $dS$, an \emph{arbitrary} wedge~$W$ is of the form 
\label{wedgegeo}
\label{lambdawpage}
\label{tWpage}
	\[  
		W \doteq  \Lambda W _1 \; , \qquad \Lambda \in SO_0(1,2)  \; .
	\]
A one-to-one correspondence \cite[p.~1203]{GL} between points $x \in dS$ and wedges 
is established by requiring that 
\begin{itemize}
\item [---] $x$ is an edge of the wedge $W_x$; 
\item [---] for any point $y$ in the interior of $W_x$ the 
triple $\{ (1,0,0), x, y \}$ has positive orientation. 
\end{itemize}
For example, $y= o$ lies in the wedge
	$ W_1 = W_{(0,  r  , 0)} $. 

\goodbreak
\begin{remarks}
\quad
\begin{itemize}
\item [$i.)$] 
Two wedges $W_x$ and $W_y$  have empty intersection, iff 
$y \in \Gamma^+(-x) \cup \Gamma^- (-x)$ \cite[Lemma 5.1]{GL}. 
\item [$ii.)$] 
Given a wedge $W$, there is exactly one time-like geodesic  ${\mathscr G}$, which lies entirely 
within $W$. Indeed, the wedge $W$ itself 
is the causal completion of~${\mathscr G}$, \emph{i.e.}, ${\mathscr G}'' = W$.  
\item[$iii.)$] 
The union of $ \Gamma^{+}( W )$ with $ \Gamma^{-}( W ')$ 
covers the de Sitter space $dS$; the intersection of $\Gamma^{+}( W )$ and $ \Gamma^{-}( W')$
are two light-rays. 
\item[$iv.)$] The space-like complement $W'$ of a wedge $W$ is a wedge, called
the \emph{opposite wedge}\index{opposite wedge}. The \emph{double wedge}\index{double wedge}
	\label{dWpage}
	\begin{equation} 
			\label{double-wedge}
			\mathbb{W} \doteq W \cup W'  
	\end{equation} 
is uniquely specified by fixing  (one of) its edges (the other one is just the antipode). 
\end{itemize}
\end{remarks}

\subsection{Double cones}
Open, bounded, connected, causally complete space-time regions in $dS$ are 
called \emph{double cones}\index{double cone}. We provide the following characterisation. 

\begin{proposition} Let ${\mathcal O}$ be a double cone. Then there exist
\begin{itemize}
\item [$i.)$] two\footnote{Note that \emph{every} bounded non-empty region ${\mathcal O}$
given as the intersection of wedges, is an intersection of \emph{two} (canonically determined)
wedges \cite[Lemma~5.2]{GL}.} wedges such that ${\mathcal O}$ is equal to their intersection; 
\item [$ii.)$] a time-like geodesic ${\mathscr G}$ and an open bounded 
interval  $J \subset {\mathscr G}$ such that the causal completion $J''$ 
(which lies entirely within the wedge~${\mathscr G}''$) equals ${\mathcal O}$;
\item [$iii.)$] two points $x, y \in dS$ such that 
${\mathcal O}$ is the interior of the intersection of the future of $x$ and the past of~$y \in \Gamma^+ (x)$ 
(both $x$ and $y$ can be identified as boundary points of the segment $J$ appearing in $ii.)$);
\item [$iv.)$] an interval  $I$ of length $| I | < \pi \, r  $ on a space-like geodesic such that  the causal completion $I''$
equals ${\mathcal O}$.
\end{itemize}
\end{proposition} 

For double cones with \emph{base}\index{base} $I$ on 
$S^1$, we introduce the following notation:
\label{doubleconepage}
	\begin{equation*} 
		\mathcal{O}_{ I } \doteq I'' \subset dS \; , \qquad | I | < \pi \, r \; , \quad I \in S^1 \; . 	
	\end{equation*} 
Note that any double cone is of the form 
	\begin{equation*} 
		\Lambda \mathcal{O}_{ I }   \; , \qquad \Lambda \in SO_0(1,2) \; , \quad I \subset S^1 \; , 
		\quad | I | < \pi \, r  \; . 	
	\end{equation*} 
As $| I | \to \pi \,  r $, the light rays in \eqref{intersectinglightrays}
become parallel, and $I''$ itself becomes a wedge~$W$.

\bigskip
Wedges and double cones are  causally
complete. Wedges are also \emph{geodesically closed}\index{geodesically closed}, in the sense 
that if $x, y \in W$, then there is an 
interval $I$ on some geodesic connecting these two points, which lies entirely in~$W$. In fact, 
the causal completion~$I''$ of $I$ automatically lies in $W$ as well. A similar statement 
holds for double cones. 

\section{Finite speed of propagation}
\label{sec:fsop}

The support of Cauchy data that can
influence events at some point $x\equiv (x_0, x_1, x_2)\in dS$ with $x_0>0$ is given by the
intersection $\Gamma^-(x)\cap S^1$ of the past $\Gamma^-(x) $ of $x$ with
the Cauchy surface $S^1$. It will be of particular importance\footnote{We will later on show that the ${\mathscr P}(\varphi)_2$ models
respect both the finite speed of propagation and the particularities of de Sitter space
(\emph{e.g.}, the presence of a cosmological horizon).} to describe the evolution of this
set as the point $x$ is subject to a Lorentz boost.  

\begin{lemma} 
\label{FSoL}
Consider a point $x (r, \psi) =(0, r \sin \psi, r \cos \psi ) \in I_+$. For $\tau>0$ the intersection of the past of  
the point 
	\[
	 	\Lambda_1 \left( \tau \right) x  =   \begin{pmatrix}
										 \cosh \tau &  0 &\sinh \tau \\
										0 &  1&0  \\
										\sinh \tau & 0 & \cosh \tau   
 									\end{pmatrix} 
									\begin{pmatrix}
										0 \\
									 	r \sin \psi \\
									 	r \cos \psi 
									 \end{pmatrix}
	\]
with $S^1$,  
	\[ 
		\Gamma^- (\Lambda_1 \left( \tau \right) x ) \cap S^1 = \{ x (r, \psi) \in S^1 \mid \psi_- \le \psi \le \psi_+ \}  \; , 
	\]
is an interval of length 
	\begin{equation}
	\label{length}
		r (\psi_+ - \psi_-)  = 2 r \arctan  (\sinh  \tau   \cos \psi)  
	\end{equation}
centred at $x \bigl(r, \tfrac{\psi_+  + \psi_- }{2} \bigr)$
with $\frac{\psi_+  + \psi_- }{2} = \arcsin \frac{\sin \psi}{\sqrt{1+ (\sinh   \tau  \cos \psi)^2}}$. 
\end{lemma}

\begin{proof} We compute:
	\[
		\Lambda_1 \left( \tau \right) x = \begin{pmatrix} r \sinh  \tau  \cos \psi \\
									r \sin \psi \\
									r \cosh  \tau \cos \psi
						\end{pmatrix} \; . 
	\]
Eq.~\eqref{length} now follows 
directly from \eqref{height}, and the localisation of the interval follows from 
	\[
		\sin \left( \tfrac{\psi_+  + \psi_- }{2} \right) = \frac{\sin \psi}{\sqrt{\sin^2 \psi 
		+ \cosh^2 \tau \cos^2 \psi}} \; , 
			\qquad \psi_\pm \equiv \psi_\pm (\tau, \psi) \; , 
	\]
using $\sin^2 \psi = 1 - \cos^2 \psi$ and $\cosh^2 \tau= 1+ \sinh^2 (\tau)$.
\end{proof}

The set $I(\alpha, \tau)$ introduced in the following proposition describes the localisation region for the Cauchy data, which 
can influence space-time points in the set $\Lambda^{(\alpha)} (\tau) I$, $\tau > 0$, where $I \subset S^1$ is some open interval. 

\begin{proposition}\label{ialpha}
Let $I $ be an arbitrary interval in $S^1$. 
Consider the boosts $\tau \mapsto \Lambda^{(\alpha)} \left(\tau\right)$. It follows that the set  
	\begin{align}
		\label{Ialphat}
		 I (\alpha , \tau)  \doteq S^1 \cap \Bigl( \bigcup_{y \in \Lambda^{(\alpha)}  \left(\tau\right) I }   
		\Gamma^- (y )\cup  \Gamma^+ (y ) \Bigr) 
				\nonumber
	\end{align}
equals
	\[
		 I (\alpha , \tau) 
		 =  \bigcup_{(0, r \sin \psi, r \cos \psi ) \in I } 
		\bigl\{ x(r, \psi-\alpha) \mid   \psi_- (\pm \tau, \psi+\alpha) \le \psi \le \psi_+ (\pm \tau, \psi+\alpha) \bigr\} .
	\]
As before, $x (r, \psi) =(0, r \sin \psi, r \cos \psi )$. Explicit formulas for the angles $\psi_\pm(\tau, \psi)$ are provided in Lemma \ref{FSoL}. 
\end{proposition}

\goodbreak
\begin{remarks}
\quad
\begin{itemize}
\item [$i.)$] The \emph{speed of propagation}\index{speed of propagation}\footnote{Note that 
speed refers to proper time and spatial geodesics distances as defined in \eqref{dlength}.}
	\[
		{\rm v}_\mp = r \frac{d \psi_\mp \bigl( \tfrac{\tau}{r}, \psi + \alpha \bigr) }{d\tau} 
	\]
(to the left and to the right) along the circle $S^1$ goes to zero as $x$ approaches the \emph{fixed  points}  
$R_0 (\alpha)x$, with $x=(0,\pm r,0)$, for the boost $\tau \mapsto \Lambda^{(\alpha)} \left( \tau \right)$.
\item [$ii.)$] For $\tau$ small the interval $I \left( \alpha , \tfrac{\tau}{r} \right)$ grows at most\footnote{This is the case for small $\tau$, 
if the interval is centred at $R_0 (\alpha)x$, with $x=(0,\pm r,0)$.} with the speed
of light (on both sides), while for~$\tau$ large and increasing the growth rate decreases to zero. In fact, for any interval
$I \subset I_\alpha$ we have
	\[
		I (\alpha , \tau) \subset I_\alpha
		\qquad \forall \tau \ge 0 \; . 
	\]
Recall that $ \bigcup_{t \in \mathbb{R}} \Lambda^{(\alpha)}  (t) I_\alpha = W^{(\alpha)} $ 
and $W^{(\alpha)}  \cap S^1= I_\alpha$.  
\item [$iii.)$] Let $I \subset I_+$ be an open interval. It follows that 
$\lim_{\tau \to \infty} I (\alpha , \tau) = I_+ $. In fact, for every point $x \in W^{(\alpha)} $ one has 
	\[
		\lim_{\tau \to \pm \infty} \Gamma^- \bigr(\Lambda^{(\alpha)} ( \tau ) x\bigr) \cap S^1 = I_\alpha \; .
	\]
\end{itemize}
\end{remarks}

\chapter{Space-time symmetries}
\label{isometrygroup}

Just like Minkowski space, de Sitter space is maximally symmetric. 
Just like the sphere and the plane, it has constant curvature. 

\section{The group $O(1,2)$}

The \emph{isometry group}\index{isometry group} of the ambient space $\mathbb{R}^{1+2}$
is the Poincar\'e group~$E(1,2)$.
The stabiliser of the zero vector~$ \boldsymbol{0} \equiv (0,0,0)  \in \mathbb{R}^{1+2}$ is the
subgroup $O(1,2)$ of $E(1,2)$. It is the group of isometries of  $( dS ,  g )$.

\begin{lemma}
The action of the group $O(1,2)$ splits $\mathbb{R}^{1+2}$ into orbits\footnote{In other words, the sets $\partial V^+ \cup \partial V^-$, 
$dS$ and $H^+_m \cup H^-_m$ are \emph{$G$-sets}\index{$G$-set} for the group $G=O(1,2)$.}:
\begin{itemize}
\item[$i.)$] $\{ g \, \boldsymbol{0} \mid g \in O(1,2) \} = \{ \boldsymbol{0} \} $, \emph{i.e.}, the group $O(1,2)$  leaves the origin $(0, 0, 0)$ invariant;
\item[$ii.)$]  $\bigl\{ g \left( \begin{smallmatrix} m \\ 0 \\ 0 \end{smallmatrix} \right) \mid g \in O(1,2) \bigr\} = H^+_m \cup H^-_m$, where
\label{masshyperboloid}
	\[
		H^\pm_m \doteq  \{ x \in \mathbb{R}^{1+2} \mid x_0^2 - x_1^2 - x_2^2 = m^2 , \pm x_0 >0 \} \; .   
	\] 
More generally, the orbit of any point in the interior of the forward light-cone is a 
two-sheeted \emph{mass hyperboloid}\index{mass hyperboloid} $H^+_m \cup H^-_m$ for some mass $m>0$; 
\item[$iii.)$]  $\{ g o \mid g \in O(1,2) \} = dS$. More generally, the orbit of any point, which is space like to 
the zero vector $\boldsymbol{0}$, 
is a de Sitter space $dS$ of some radius $r>0$; 
\item[$iv.)$] $\bigl\{ g \left( \begin{smallmatrix} 1 \\ 0 \\ -1 \end{smallmatrix} \right) \mid g \in O(1,2) \bigr\} = (\partial V^+ \cup \partial V^-) 
\setminus \{ \boldsymbol{0} \}$. More generally, the orbit of any point, which is light-like to the zero vector  $\boldsymbol{0}$, 
is $(\partial V^+ \cup \partial V^-) 
\setminus \{ \boldsymbol{0} \}$. 
\end{itemize}
The Minkowski space $\mathbb{R}^{1+2}$ is the disjoint union of all of these sets.
\end{lemma}
 
\begin{proof}
If $X$ is $\partial V^+ \cup \partial V^-$, $dS$, or $H^+_m 
\cup H^-_m$,  then
	\[
	\Lambda (\Lambda' x ) = (\Lambda \circ \Lambda') x \in X \; , 
	\quad \Lambda, \Lambda ' \in O(1,2) \; , \quad x \in X \; . 
	\]
In particular, $\Lambda (\Lambda^{-1} x )= (\Lambda \circ \Lambda^{-1} ) x = x$ for all $x \in X$.
Moreover, the group $O(1,2)$ acts transitively
on $\partial V^+ \cup \partial V^-$, $dS$  and $H^+_m 
\cup H^-_m$:
	\begin{align*}
		\partial V^+ \cup \partial V^- & = \{ T^k R_0 (\alpha) \Lambda_1(t) \left( 
		\begin{smallmatrix}
		1\\
		0\\
		-1\\
		\end{smallmatrix} \right)
		\mid k= 0,1, \;  t \in \mathbb{R}, \alpha \in [0, 2 \pi) \} \; , \\
		dS & = \{ R_0 (\alpha) \Lambda_1(t) o \mid t \in \mathbb{R}, \alpha \in [0, 2 \pi) \} \; , \\
		H^+_m \cup H^-_m & = \{ T^k R_0 (\alpha) \Lambda_1(t) \left( 
		\begin{smallmatrix}
		m\\
		0\\
		0\\
		\end{smallmatrix} \right)
		\mid k= 0,1, \;  t \in \mathbb{R}, \alpha \in [0, \pi) \} \; .
	\end{align*}
\end{proof}

\subsection{The action of $SO_0(1,2)$ on the light-cone}
In the sequel, the forward light cone 
	\begin{align*}
		\partial V^+ & = \left\{ R_0 (\alpha) \Lambda_1 (t) \left( \begin{smallmatrix}
		1 \\  0 \\ -1 \end{smallmatrix} \right) \mid t \in \mathbb{R}, \alpha \in [0, 2 \pi) \right\} \\
		& = \left\{ \left( \begin{smallmatrix} p_0	 \\ p_0 \sin \alpha \\ - p_0 \cos \alpha
		\end{smallmatrix} \right) \mid p_0 >0, \alpha \in [0, 2 \pi) \right\} \; ,  
	\end{align*}
will play an important role. In the second line, we have set $p_0 = {\rm e}^{-t}$. 

\goodbreak
We will therefore 
provide explicit formulas for the action of the boosts $\Lambda_{1}(t)$, $\Lambda_{2}(s)$ 
and the rotations $R_0(\beta)$ on $\partial V^+$: 
	\begin{equation} 
	\label{lambda1-s}
				\Lambda_{1}^{-1}(t)  
									\begin{pmatrix}
 								p_0 \\
								p_0 \sin \alpha   \\
								- p_0 \cos \alpha  
									\end{pmatrix}   
									= p_0
								\begin{pmatrix}
 								\cosh t + \sinh t \cos \alpha \\
								\sin \alpha   \\
								- \sinh t - \cosh t \cos \alpha  
									\end{pmatrix}  \; .   
				\end{equation}
and
	\begin{equation} 
	\label{lambda2-s}
					\Lambda_{2}^{-1}(s) 
					\begin{pmatrix}
 								p_0 \\
								p_0 \sin \alpha   \\
								- p_0 \cos \alpha  
									\end{pmatrix}  
									=
									p_0
									 \begin{pmatrix}
 								\cosh s -  \sinh s \sin \alpha \\
							       - \sinh s + \cosh s \sin \alpha  \\
								- \cos \alpha  
					\end{pmatrix} 
													 \;   , 
		\end{equation} 
Finally, 
	\begin{equation} 
	\label{R0-s}
				R_{0}^{-1}(\beta)  
									\begin{pmatrix}
 								p_0 \\
								p_0 \sin \alpha   \\
								- p_0 \cos \alpha  
									\end{pmatrix}   
									= p_0
								\begin{pmatrix}
 								1 \\
								\cos \beta \sin \alpha - \sin \beta  \cos \alpha \\
								- \sin \beta \sin \alpha - \cos \beta  \cos \alpha  
									\end{pmatrix}  \; .   
				\end{equation}

\subsection{Reflections}
The \emph{time reflection}\index{time reflection} and the \emph{parity transformation}\index{parity transformation}
\label{timerefl-page}
		\[ 
		T\doteq  \begin{pmatrix}
					-1 &  0 &0 \\
					0 &  1 & 0  \\
					0  & 0 & 1   
 				\end{pmatrix}  \; , \quad 	P_1\doteq  \begin{pmatrix}
					1 &  0 &0 \\
					0 &  1 & 0  \\
					0  & 0 & -1   
 				\end{pmatrix} \; \in \; O (1,2) \; ,
		\] 
leave the Cauchy surface $S^1$ invariant. 
$P_1 T$ is the  reflection at the edge of the wedge~$W_1$. Together $P_1$ and $T$ generate the Klein four group.  
The reflection 
at the edge of an arbitrary  wedge $W = \Lambda W_1$,  is 
\label{PTwedgepage}
	\begin{equation*} 
	\label{PTwedge}
		\Theta_{ W } \doteq \Lambda\;  P_1 T \; \Lambda^{-1} \;, \qquad \Lambda \in SO_0(1,2) \; .  
	\end{equation*} 
$\Theta_{ W }$ is an isometry of both ${\mathbb W}$ and $dS$. It preserves the orientation but inverts the
time orientation, in other words, $\Theta_{ W }$ is an element of $SO^\downarrow(1,2)$.  

\subsection{Boosts associated to wedges}
The boost $t \mapsto \Lambda_1(t)$ leaves the wedge~$W_1$ invariant. In fact, 
$W_1$ is the causal completion of the worldline $t \mapsto \Lambda_1(t)  o$. 
For an arbitrary  wedge $W = \Lambda W_1$, $\Lambda \in SO_0(1,2)$, 
	\[
	\Lambda_{\scriptscriptstyle W}(t) \doteq \Lambda \Lambda_1(t) \Lambda^{-1} \; ,  \qquad t \in \mathbb{R} \; , 
	\]
defines a boost leaving $W$ invariant, \emph{i.e.},  
	\[ 
		\Lambda_{\scriptscriptstyle W}(t) W =W \; , \qquad t \in \mathbb{R} \; .  
	\] 
In particular,  $\Lambda_1(t) = \Lambda_{W_1} (t)$ for all $t \in \mathbb{R}$. 

The Killing vector field\footnote{Killing fields are the infinitesimal generators of isometries; that is, flows generated 
by Killing fields are continuous isometries of the manifold.} induced by~$\Lambda_{ W } (t)$ leaves the opposite 
wedge~$W'$ invariant  too. It is however past directed in~$W'$. 
One may fix the scaling factor  by normalising the Killing vector field  
on the time-like geodesic~${\mathscr G}$ satisfying~${\mathscr G}'' = W$. 
Uniqueness then implies 
	\[
		\Lambda_{ W } (t) = \Lambda_{ W' } (-t) \; , \qquad t \in \mathbb{R} \; .  
	\]
The double-wedge $\mathbb{W}$ introduced in \eqref{double-wedge} is  invariant under  both $\Lambda_{ W } (t)$ and 
$\Lambda_{ W' } (t)$, $t \in \mathbb{R}$. 

Another interesting property of the boost $t \mapsto \Lambda_1(t)$ is that it leaves the points $(0, \pm r, 0)$ invariant.
In other words, it is the {\em stabilizer} --- within the group $SO_0(1,2)$ --- of the point $(0,r, 0) \in dS$ (and, at the same time, 
the antipode $-(0,  r, 0)$). Similarly, 
the origin $ o \,$ and its antipode $-o\, $ are invariant under the boosts $\Lambda_{2}(s)$, $s \in \mathbb{R}$. More generally, 
the group 
	\[ 
		t \mapsto \Lambda_{ W_x } (t) \; , \qquad t \in \mathbb{R} \; , 
	\]
is the unique --- up to rescaling --- one-parameter subgroup of $SO_0(1,2)$,  which leaves the edges of the wedge $W_x$ 
invariant and induces a future directed Killing vector field in the wedge $W_x$. 
Clearly, it also leaves  the light rays passing through $\pm x$ invariant. 

\begin{remark}
A free falling observer passing through the origin $o$ interprets the  
boost $\upsilon \mapsto \Lambda_{2} ( \upsilon) =\exp  (  \upsilon  L_2 )$  as a Lorentz transformation, 
the boost (re-scaled to proper time) 
	\[ 
		\tau \mapsto \Lambda_{1} \left( \tfrac{\tau}{r} \right) = {\rm e}^{  \frac{\tau}{r} \, L_1 }
	\]
as his geodesic time evolution and the 
rotation~$a \mapsto R_{0} \left( \tfrac{a}{r} \right) = \exp \left( \tfrac{a}{r} \, K_0 \right)$, $a \in [0, 2 \pi \, r)$,  
as a spatial translation\footnote{Alternatively, one may also view a motion along a horosphere, given by the 
map $q \mapsto D(q/r)$ as a translation; see \eqref{1.5.1} below.}. 
Unless $x =o$, the path 
	\[
		\tau \mapsto \Lambda_1 \left( \tfrac{\tau}{r} \right)  x \; , \qquad x \in I_+ \; , 
	\]
describes a uniformly accelerated observer. Note that such a path 
lies on the intersection of $dS$ with a plane parallel to the $(x_2= 0)$-plane, passing through~$x$. 
\end{remark}

\subsection{Coordinates for the wedge $W_1$}
The chart 
	\begin{equation} 
	 	\label{w1psi}
		 x (t, \psi) \doteq \Lambda_{1} \left( t \right) \, R_0 \left( - \psi \right) 
			\begin{pmatrix}
					0 \\
					0  \\
					r  
			\end{pmatrix} \; , \qquad \, t \in \mathbb{R} \; , 
		\quad  {\textstyle  -\frac{\pi}{2}< \psi < \frac{ \pi}{2}}   \; ,
	\end{equation} 
provides coordinates for the wedge $W _1$.  
Allowing $\psi \in \;  [-  \tfrac{\pi}{2} ,  \tfrac{3\pi}{2} )$,  these coordinates extend\footnote{In the sequel, we will 
always take care of the fact that these coordinates are degenerated at the 
fixed-points $(0, r, 0), (0, -r, 0) \in dS$ for the boost $t \mapsto \Lambda_1(t)$.}
to the space-time region   
	\begin{equation} 
		\label{3.32}
		\mathbb{W}_1  \cup \{ (0, r, 0), (0, -r, 0)\} = \bigcup_{ t \in \mathbb{R} } \Lambda_{1}(t) S^1  \; . 
	\end{equation} 
The r.h.s.~is the union of the boosted time-zero circles $\Lambda_{1}(t) S^1$, $t \in \mathbb{R}$.  

\section{Horospheres}

The set of orbits of all maximally \emph{unipotent}\footnote{A unipotent element $a$ of a ring $R$ is one such that $a -1$ is 
a nilpotent element. Any unipotent algebraic group is isomorphic to a closed subgroup of the group of 
upper triangular matrices with diagonal entries $1$, and conversely any such subgroup is unipotent.} subgroups 
--- the so-called \emph{horospheres}\footnote{Horospheres previously appeared in hyperbolic geometry. They are spheres of 
infinite radius with centres at infinity and different from hyperbolic hyperplanes.} ---
play a central role in the construction of induced representations of the Lorentz group.
Gelfand and Gidinkin (see, e.g., \cite{Macfadyen-1, Macfadyen-2, Macfadyen-3})
have shown that the generalized Fourier transform on homogeneous spaces and the horospherical transform 
are connected by the (commutative) \emph{Mellin transform}. 
\begin{lemma}
\label{stabil}
The {\em stabilizer} --- within the group $SO_0(1,2)$ --- of the point  $(1, 0, -1) \in \partial V^+$ 
is the one-parameter group\footnote{One verifies that $D(q) D(q') =D(q+q')$ for all $q, q' \in \mathbb{R}$.}
		\begin{equation} 
			\label{1.5.1}
				D(q)  \doteq  \begin{pmatrix}
										 1+ \frac{q^2}{2} &  q &  \frac{q^2}{2}\\
										q &  1& q  \\
										- \frac{q^2}{2} & - q & 1- \frac{q^2}{2}   
 									\end{pmatrix} , 
							\quad 
				q \in \mathbb{R}\; .   
		\end{equation} 
It has the following properties:
\begin{itemize}
\item [$i.)$] it leaves the light ray $\lambda (1, 0, -1) $, $\lambda \in \mathbb{R}$, pointwise invariant; 
\item [$ii.)$] it is \emph{nilpotent}. In fact, 
	\[
		D(q) = e^{q(L_2 - K_0)} = \mathbb{1} + q (L_2 - K_0) + \tfrac{q^2}{2} (L_2 - K_0)^2  \; , 
		\qquad  q \in \mathbb{R} \; ; 
	\]
\item [$iv.)$] it leaves the half-spaces $\Gamma^+ (W_1)$ and $\Gamma^- (W_1')$ invariant. 
In particular, it leaves the two light rays 
	\[ 
		D(q) \begin{pmatrix}
				0 \\
				 \pm r \\
				  0 \end{pmatrix} 
				  = \begin{pmatrix} \pm r q \\
				   \pm r \\
				    \mp r q \end{pmatrix}
				  \; ,  \qquad  q \in \mathbb{R} \; ,
	\]
which form the intersection of $\Gamma^+ (W_1)$ with $\Gamma^- (W_1')$, invariant; 
\item [$iv.)$] it satisfies\footnote{See, \emph{e.g.}, \cite[Chapter 9.1.1, Equ.~(11)]{Vil}.} 
	\begin{equation}
		\label{scalingH}
		\Lambda_1(t) D(q) \Lambda_1(-t) = D \bigl( {\rm e}^t q \bigr)\; , \qquad t, q \in \mathbb{R} \; . 
	\end{equation}
\end{itemize}
\end{lemma}
\color{black}

\subsection{Coordinates for the half-space $\Gamma^+ (W_1)$}
The boosts $\Lambda_{1}(t)$, $t \in \mathbb{R}$, together with the translations
$D(q)$, $ q \in \mathbb{R}$, give rise to the chart\footnote{These coordinates are frequently 
called \emph{Lema\^itre-Robinson} coordinates in the physics literature, see \cite{Le}. In the mathematics literature
they are called \emph{orispherical} coordinates, see, \emph{e.g.}, \cite[Chapter 9.1.5, Equ.~(16)]{Vil}.}
	\begin{align}
	\label{XR}
	x ( \tau, \xi) & \doteq 	
		 D \left( \tfrac{\xi}{r} \right) \Lambda_{1} \left( \tfrac{\tau}{r} \right)  \begin{pmatrix}
										0 \\
										0  \\
										r   
									\end{pmatrix}
									=
									 \begin{pmatrix}
									  r \sinh \tfrac{\tau}{r} + \tfrac{\xi^2}{2r} {\rm e}^{\frac{\tau}{r}} \\
									   \xi {\rm e}^{\frac{\tau}{r}} \\
									   r \cosh \tfrac{\tau}{r} - \tfrac{\xi^2}{2r} {\rm e}^{\frac{\tau}{r}}
									\end{pmatrix}   
	\end{align}
for the interior of the half-space $\Gamma^+ (W_1)$. In particular, 
	\[ 
		D \left( \tfrac{\xi}{r} \right) \begin{pmatrix}
				0 \\
				 0 \\
				  r \end{pmatrix} 
				  = \begin{pmatrix} \tfrac{\xi^2}{2r} \\
				   \xi \\
				    r - \tfrac{\xi^2}{2r} \end{pmatrix}
				  \qquad \text{for}  \; \; \xi \in \mathbb{R} \; .
	\]
The metric takes the form 
	$ 
		g_{\upharpoonright \Gamma^+ (W_1)} = {\rm d} \tau \otimes {\rm d} \tau -  {\rm e}^{\frac{2\tau}{r}} {\rm d} \xi \otimes {\rm d} \xi $. 

\subsection{Parabolas in $\Gamma^+ (W_1)$}
For $\tau$ fixed, \eqref{XR} parametrizes the \emph{horosphere}\index{horosphere} (which actually is a parabola in $\mathbb{R}^{1+2}$)
	\[ 
	 P_\tau \doteq \left\{ x(\tau, \xi) \mid \xi \in \mathbb{R} \right\} \subset dS \; . 
	\]
General horospheres result from taking the intersection of $dS$ with a plane whose 
normal\footnote{Note 
that the Lorentzian scalar product  $x \cdot p$, $x \in \mathbb{R}^{1+2}$, $p \in \partial V^+$, equals the Euclidean 
scalar product of $x$ with $Pp \in \partial V^+$, with $P$ the space-reflection (see the final paragraph in 
Section~\ref{isometrygroup}). The plane defined by $x \cdot p = 0$, $x \in \mathbb{R}^{1+2}$, $p \in \partial V^+$ fixed, contains
the point $x=p$.} vector~$p$ is light like, \emph{i.e.}, $p  \cdot p = 0$. 
In particular, the horosphere $P_\tau$  
is given by    
	\begin{equation} 
	\label{horosphere}
	P_\tau =  \left\{ x \in dS
	\mid x \cdot \left( \begin{smallmatrix}
					1\\
					0  \\
					- 1   
 			\end{smallmatrix} \right) = r {\rm e}^{\frac{\tau}{r}} \right\} \; . 
	\end{equation}
Thus 
general horospheres are of the form $R_0(\alpha)P_\tau$ with $\alpha \in [0, 2\pi)$ and $ \tau \in \mathbb{R}$.

\subsection{Horospheric distances}
The proper time-difference---given by \eqref{dlength}---of the 
points $\Lambda_1 \bigl( \tfrac{\tau_1}{r} \bigr) o$ and $\Lambda_1 \bigl(\tfrac{\tau_2}{r} \bigr) o$ on the geodesic passing through the 
origin $o = (0, 0, r )$  is\footnote{Simply recall that $\cosh (x - y) = \cosh x \cosh y - \sinh x \sinh y$.} 
	\[
	d \big(\Lambda_1 \bigl( \tfrac{\tau_1}{r} \bigr) o, \Lambda_1 \bigl( \tfrac{\tau_2}{r} \bigr) o \bigr)
	= r \; {\rm arcosh} \, \left( - \begin{pmatrix}
						\sinh \tfrac{\tau_1}{r}  \\
						0  \\
						\cosh \tfrac{\tau_1}{r}   
					\end{pmatrix} \cdot 
					\begin{pmatrix}
						\sinh \tfrac{\tau_2}{r} \\
						0  \\
						\cosh \tfrac{\tau_2}{r}   
					\end{pmatrix} \right) =    | \tau_1 - \tau_2 | \; . 
	\]
In fact, $| \tau_1 - \tau_2 |$ is the minimal distance of \emph{any} two time-like points on the horospheres 
$P_{\tau_1}$ and $P_{\tau_2}$, respectively: if 
	\begin{equation} 
		\label{a}
		x  = \bigl( r \sinh \tfrac{\tau_2}{r} + \tfrac{\xi^2}{2r} {\rm e}^{\tfrac{\tau_2}{r}} , \xi {\rm e}^{\tfrac{t_2}{r}} , 
		r \cosh \tfrac{\tau_2}{r} - \tfrac{\xi^2}{2r} {\rm e}^{\tfrac{\tau_2}{r}} \bigr)
	\end{equation}
is a point in $P_{\tau_2}$, then the minimal distance to time-like points in the horosphere~$P_{\tau_1}$, 
called the \emph{horospheric distance},  is given by\footnote{To show the second identity, one may use the invariance the Minkowski scalar 
product, \emph{i.e.}, $D(q)x \cdot D(q') y= D(q-q')x \cdot y $ for all $q, q' \in \mathbb{R}$ and $x, y \in \mathbb{R}^{1+2}$.} 
	\begin{align}
	\label{distance-to-horosphere}
	d  \bigl(x, P_{\tau_1} \bigr) & \doteq r \; {\rm arcosh} \bigl( \, \min_{y \in P_{\tau_1} } - \tfrac{x \cdot y}{r^2} \bigr)  \nonumber \\
	& 
	= |Ê\tau_1 - \tau_2| 
	= r \ln | \, \tfrac{x}{r} \cdot  p \left( \tfrac{\tau_1}{r} \right) |  
	\;  ,  
	\end{align}
with $			p(t) \doteq 
				 \left( \begin{smallmatrix}
					{\rm e}^{-t}\\
					0  \\
					- {\rm e}^{-t}   
 			\end{smallmatrix} \right) = 
			\Lambda_1  (t) 
			 \left( \begin{smallmatrix}
					1\\
					0  \\
					- 1   
 			\end{smallmatrix} \right)$. 
Note that $P_\tau =  \left\{ x \in dS \mid \tfrac{x}{r} \cdot  p \left( \tfrac{\tau}{r} \right) = 1 \right\} $.

\section{The complex Lorentz group}

The \emph{complex  de Sitter group}\index{complex  de Sitter group} is defined as
the group 
	\[ 
		O_{\mathbb{C}} (1,2)  
		\doteq 	\bigl\{ \Lambda \in M_3 (\mathbb{C}) 
				\mid \Lambda \,  \mathbb{g} \, \Lambda^T =  \mathbb{g} 
			\bigl\}   \; . 
	\]
The elements in $ M_3 (\mathbb{C})$ are $3 \times 3$ matrices with complex  entries and $\mathbb{g}$ is the metric 
on Minkowski space $\mathbb{R}^{1+2}$ given in \eqref{metrik}.
The group $O_{\mathbb{C}} (1,2) $ has two connected components (distinguished by the sign of $\det \Lambda$, which 
takes the values  $\det \Lambda= \pm 1$). Following standard terminology, we set  
	\[
		SO_{\mathbb{C}} (1,2) 
		\doteq	\bigl\{ \Lambda \in M_3 (\mathbb{C}) 
				\mid \Lambda \,  \mathbb{g} \, \Lambda^T =  \mathbb{g} \; , \; \det \Lambda = 1  
			\bigl\} \; .  
	\] 
Note that $SO_{\mathbb{C}} (1,2)$ is isomorphic to $SO_{\mathbb{C}} (3)$; the isomorphism from $SO_{\mathbb{C}} (1,2)$ 
to $SO_{\mathbb{C}} (3)$ is given by the map 
	\[
		\Lambda \mapsto \begin{pmatrix} - i & 0 & 0 \\
		0 & 1 & 0 \\
		0 & 0 & 1
		\end{pmatrix} \Lambda \begin{pmatrix} i & 0 & 0 \\
		0 & 1 & 0 \\
		0 & 0 & 1
		\end{pmatrix} .
	\]

\subsection{Analytic continuation}
In $O_{\mathbb{C}}  (1,2)$ the reflections 
		\[
		P_1 T \doteq \begin{pmatrix}
								-1 &  0 &0 \\
								0 &  1 & 0  \\
								0  & 0 & -1   
 						\end{pmatrix},  \; 
		P_2 T \doteq \begin{pmatrix}
								-1 &  0 &0 \\
								0 & -1 & 0  \\
								0  & 0 & 1   
						\end{pmatrix},  \; 
							P\doteq  \begin{pmatrix}
								1 &  0 &0 \\
								0 &  -1 & 0  \\
								0  & 0 & -1   
						\end{pmatrix}  , 
		\]
are topologically connected to the identity $\mathbb{1}$. In fact, the matrix-valued 
function $t\mapsto \Lambda_{1}(t)$  extends to an entire analytic function 
		\begin{equation} 
		\label{eqBooW} 
			\Lambda_{1}(t+i\theta) 
				= \Lambda_{1}(t) 
					\left[ \begin{pmatrix}
								\cos \theta &  0&0  \\
								0&  1 &0  \\
								0 & 0 & \cos \theta  
 						\end{pmatrix}
					+ i  \begin{pmatrix}
								0 &  0&\sin \theta  \\
								0&  0 &0  \\
								\sin \theta & 0 & 0  
							\end{pmatrix} \right] \, .
		\end{equation} 
The first matrix in the square brackets 
continuously deforms the unit $\mathbb{1}$ to 
$P_1 T$, as~$\theta$ takes values starting at~$\theta=0$ and ending at $\theta =\pm \pi$. The 
second matrix in the square brackets projects the wedge~$W_1$ 
continuously into the $x_1=0$ section of the forward light cone, \emph{cf.}~\cite{H}. 

\section{The Cartan decomposition\index{Cartan decomposition} of $SO_0(1,2)$}

Since every point $\vec x$ of $H_m^+ \cong SO_0(1,2) / SO(2)$ has the form 
	\[
		\vec x = \begin{pmatrix}
		m \cosh \theta  \\
		- m \sinh \theta \sin \alpha   \\
		m \sinh \theta \cos \alpha \end{pmatrix} , 		
	\]
any element $g \in SO_0(1,2)$ can be represented in the form
	\begin{equation}
	\label{cartan}
		g = R_0 (\alpha) \Lambda_1 (t) R_0( \alpha') \; , \qquad \alpha, \alpha' \in [0, 2 \pi) \; , \quad t \in \mathbb{R} \; . 
	\end{equation} 
The corresponding decomposition 
	\[
		SO_0(1,2) = KAK
	\]
with $K= SO(2)$ and $A=SO(1,1)$ is called the \emph{Cartan decomposition}. Note that the decomposition \eqref{cartan} is not unique;
see, \emph{e.g.}, \cite[Chapter 9.1.5]{Vil}.

\section{The Iwasawa decomposition\index{Iwasawa decomposition} of $SO_0(1,2)$}

A brief inspections shows that every point $x \in H_m^+ $ can also be written
in the form
	\begin{align}
	\label{XR2}
	x ( t, q) & = 	
		 D (q) \Lambda_{1} (t)  \begin{pmatrix}
										m \\
										0  \\
										0   
									\end{pmatrix}
									=
									 m \begin{pmatrix}
									  \cosh t + \tfrac{q^2}{2} {\rm e}^{t} \\
									   q {\rm e}^{t} \\
									  \sinh t - \tfrac{q^2}{2} {\rm e}^{t}
									\end{pmatrix}   
	\end{align}	
Now, if $g \in SO_0(1,2)$, then $g \left( \begin{smallmatrix} m \\ 0 \\ 0 \end{smallmatrix} \right)$ is of 
the form \eqref{XR2} for some \emph{unique} $t$ and $q$. It follows that this point can be carried back 
to the point $(m , 0,  0)$ by the action of $\Lambda_1(-t) D(-q)$, \emph{i.e.}, 
	\[
		\Lambda_1(-t) D(-q) \, g \begin{pmatrix} m \\ 0 \\ 0 \end{pmatrix} = \begin{pmatrix} m \\ 0 \\ 0 \end{pmatrix}
	\]
But 
the stabiliser of the point $\left( \begin{smallmatrix} m \\ 0 \\ 0 \end{smallmatrix} \right)$ is the group $K \cong SO(2)$. 
Thus there exists some $\alpha \in [0, 2 \pi)$ such that $R_0 (- \alpha) \Lambda_1(-t) D(-q) \, g = \mathbb{1}$. 
Thus we arrive at the so-called \emph{Iwasawa  decomposition}\footnote{Consider a non-compact semi-simple Lie group $G$.    
If one chooses a maximal compact sub\-group $K$, 
and a suitable abelian subgroup $A$, then there exits a nilpotent sub\-group $N$, 
normalised\footnote{$A$ normalises $N$, if 
$a n a^{-1} \in N$ for all $a \in A$ and all $n\in N$; see also \eqref{scalingH}.} by $A$, 
such that any group element $g \in G$ can be written \emph{uniquely} as $kan$  with $k \in K$, $a \in A$ 
and $n \in N$.} \cite{Iwa}:

\begin{lemma}
\label{iwa}
In case $G$ is the \emph{two-fold covering group}\index{two-fold covering group}\footnote{The two-fold covering group of 
$SO_0(1,2)$ is $SU(1,1)$.} 
of $SO_0(1,2)$, any element $g \in G$ can be written as 
	\begin{align*} 
		\label{Iwaw}
		g 
			& = R_0(2\alpha) P^k\Lambda_{1}(t) D(q)
			\\
			& =
				\begin{pmatrix}
					1 &  0 &0 \\
					0 &  \cos 2 \alpha & - \sin 2 \alpha  \\
					0  & \sin 2 \alpha & \cos 2 \alpha   
				\end{pmatrix}  
				\begin{pmatrix}
					1 &  0 &0 \\
					0 &  (-1)^k & 0  \\
					0  & 0 & (-1)^k   
				\end{pmatrix}
				\begin{pmatrix}
					\cosh t &  0 &\sinh t \\
					0 &  1&0  \\
					\sinh t & 0 & \cosh t   
 					\end{pmatrix} 
					\\
					& \qquad \qquad   \times 
			\begin{pmatrix}
					1+ \frac{q^2}{2} &  q &  \frac{q^2}{2}\\
					q &  1& q  \\
					- \frac{q^2}{2} & -q  & 1- \frac{q^2}{2}   
 			\end{pmatrix}, 	
			\quad \alpha \in [0,   \pi), \; t,q \in \mathbb{R}, \; k\in \{ 0, 1\} \; . 
	\end{align*} 
\end{lemma}

The resulting decomposition, 
$G= KAN $, provides  
\begin{itemize}
\item[$i.)$] a maximal compact subgroup $K$ (the two-fold covering group of $SO(2)$), consisting of the rotations 
$K/M \cong SO(2)$ and the group $M\cong \mathbb{Z}_2$ generated 
by the reflection $P$:
	\[
		M= \left\{
				\begin{pmatrix}
					1 &  0 &0 \\
					0 &  1 & 0  \\
					0  & 0 & 1   
				\end{pmatrix}  ,
				\begin{pmatrix}
					1 &  0 &0 \\
					0 &  -1 & 0  \\
					0  & 0 & -1   
				\end{pmatrix}  	\right\} \cong \mathbb{Z}_2 \; ; 
	\] 
\item[$ii.)$] a Cartan maximal abelian subgroup $A\cong (\mathbb{R}, +)$, namely the boosts 
$\{ \Lambda_1(t) \mid  t \in \mathbb{R} \}$; 
\item[$iii.)$] a nilpotent\footnote{A nilpotent group is one that has a central series of finite length.
Any unipotent group is a nilpotent group, though the diagonal matrices are only nilpotent in $GL (n)$.} 
group $N \cong (\mathbb{R}, +)$, namely the horospheric translations $\{ D(q) \mid  q \in \mathbb{R} \}$.
\end{itemize}

\begin{remark}
The subgroup $M$ is the \emph{centralizer}\index{centralizer} of $A$ in $K$, \emph{i.e.}, 
	\[
		M \cong \{ k \in K \mid k a= a k \; \; \forall a \in A\} \; .
	\]
The group $AN$ is a solvable subgroup and $B= MAN$ 
is the \emph{minimal parabolic} subgroup\footnote{We recall that a {\em normal series} of a group~$G$ 
is a finite sequence $A_1, \ldots, A_N$ of subgroups such that 
$\mathbb{1} \trianglelefteq A_1 \trianglelefteq  \ldots \trianglelefteq A_N \trianglelefteq G$. A {\em normal factor} of $G$ is a quotient 
group $A_{k+1}/ A_k $ for some index $k < N$. The group $G$ is  a {\em solvable group} if all normal 
factors are abelian. Any maximal solvable subgroup $B \subset G$ is called a {\em Borel subgroup}.   
(A maximal subgroup $B$ of a group $G$ is a proper subgroup, such that no proper subgroup $K$ 
contains $B$ strictly.) Given the Iwasawa decomposition, a {\em parabolic subgroup} of $G$ is a closed subgroup containing some conjugate of $MAN$; the conjugates of $MAN $ are called {\em minimal 
parabolic subgroups} (see \cite[Chapter V.5]{Knapp}). 
In the case of the group $SL(2,\mathbb{R})$, there is up to conjugacy only one proper parabolic subgroup, 
the \emph{Borel subgroup} of the upper-triangular matrices of determinant $1$. 
The Iwasawa decomposition shows that it is enough to conjugate to $K$.} of~$G$. 
\end{remark}

\section{The Hannabuss decomposition of $SO_0(1,2)$}

A decomposition, which is closely related to the Iwasawa decomposition, was 
discovered\footnote{Its relevance in the present context was first emphasised by Hannabuss \cite{Hanna}.} by Takahashi \cite{T}: 

\begin{lemma}[Hannabus \cite{Hanna}]
Almost every element $g \in SO_0(1,2)$ can be written uniquely in the form of a product
	\begin{align*} 
		\label{Iwaw}
		g 
			& = \Lambda_{2}(s) P^k \Lambda_{1}(t) D(q)
			\\
			& \quad \\
			& =
				\begin{pmatrix}
 								\cosh s &  \sinh s &0 \\
								\sinh s &   \cosh s &0  \\
								0 & 0 & 1  
				\end{pmatrix} 
				 \begin{pmatrix}
						1 &  0 &0 \\
						0 &  (-1)^k & 0  \\
						0  & 0 & (-1)^k   
					\end{pmatrix} 
					\\
			& \qquad \qquad  \qquad \qquad \qquad  \times  
			\begin{pmatrix}
					\cosh t &  0 &\sinh t \\
					0 &  1&0  \\
					\sinh t & 0 & \cosh t   
 					\end{pmatrix}
			\begin{pmatrix}
					1+ \frac{q^2}{2} &  q & - \frac{q^2}{2}\\
					q &  1& -q  \\
					\frac{q^2}{2} & r & 1- \frac{q^2}{2}   
 			\end{pmatrix}, 	
	\end{align*} 
with  $s, t, q \in \mathbb{R}$ and $k=\{0,1\}$, \emph{i.e.}, almost every element $g \in SO_0(1, 2)$
can be decomposed into a product, which consists of a Lorentz transformation $s \mapsto \Lambda_2(s)$, 
possibly a reflection $P^k$,
a time translation $t \mapsto \Lambda_1(t)$, and a spatial translation $q \mapsto D(q)$. 
\end{lemma}

\begin{proof} Let $g \in G$ be given in its Iwasawa decomposition, \emph{i.e.}, 
	\[
		g=  R_0( \alpha)  \Lambda_1(t'') D(q'') \; , \qquad \alpha \in [0, 2 \pi), \quad t,q \in \mathbb{R}\; . 
	\]
We will show that, unless $\alpha =  \tfrac{\pi}{2}$  or  
$ \tfrac{3\pi}{2}$, $R_0( \alpha)$ can be written in  the form 
	\begin{equation}
	\label{r0H}
		R_0( \alpha) = \Lambda_2(s) P^k D(-q') \Lambda_1(-t')  \; , \qquad s, t', q' \in \mathbb{R}\; , \quad k = 0,1 \; . 
	\end{equation}
Taking \eqref{scalingH} into account,   this will imply that  
	\[
		g = \Lambda_2(s) P^k  \Lambda_1(\underbrace{t''-t'}_{t}) D \bigl( \underbrace{q'' - {\rm e}^{-(t''-t')} q' }_{q} \bigr) \; . 
	\]
Thus it remains to establish \eqref{r0H}. Multiplying 	\eqref{r0H} with $\Lambda_1(t') D(q')$ from the right yields 
	\[
		\Lambda_2(s) P^k = R_0( \alpha)  \Lambda_1(t') D(q') \; , \qquad s \in \mathbb{R}\; , \quad k = 0,1 \; . 
	\]
This is the Iwasawa decomposition of $\Lambda_2(s)P^k$, $k=0, 1$, which  is given by choosing 
	\[
		\begin{matrix}	& \cosh s = (-1)^k \cos^{-1}\alpha \; , \qquad  & {\rm e}^{t'} = \tfrac{1}{| \cos \alpha | }  \; , \\ 
				& \sinh s  = (- 1)^{k+1}\tan \alpha \:,  \qquad &
				q' = (-1)^{k+1} \sin \alpha \; . 
		\end{matrix}
	\]
In fact, 
unless $\cos \alpha = 0$, 
	\begin{align*} 
		&\begin{pmatrix}
 								| \cos \alpha |^{-1}& - \frac{\sin \alpha}{| \cos \alpha|}  &0 \\
								-\frac{\sin \alpha}{| \cos \alpha|}  &  | \cos \alpha |^{-1}  &0  \\
								0 & 0 &   1    
		\end{pmatrix}  \begin{pmatrix}
 								1 &  0  &0 \\
								 0  &  (-1)^k  &0  \\
								0 & 0 &   (-1)^k     
		\end{pmatrix} \\
\\
		& \qquad =						\begin{pmatrix}
					1 &  0 &0 \\
					0 &  \cos \alpha & - \sin \alpha  \\
					0  & \sin \alpha & \cos \alpha   
				\end{pmatrix}  
\begin{pmatrix}
					 \tfrac{ | \cos \alpha |^{-1}   + | \cos \alpha |}{2}  &  0 &  \tfrac{ | \cos \alpha |^{-1}   - | \cos \alpha | }{2} \\
					0 &  1&0  \\
					 \tfrac{| \cos \alpha |^{-1}   - | \cos \alpha |}{2} & 0 &  \tfrac{| \cos \alpha |^{-1}   + | \cos \alpha |}{2}   
 					\end{pmatrix} 
		\\
		&
		\qquad \qquad \times 	\begin{pmatrix}
					1+ \frac{\sin^2 \alpha }{2} &  (-1)^{k+1} \sin \alpha  & - \frac{\sin^2 \alpha }{2}\\
					(-1)^{k+1} \sin \alpha  &  1&  (-1)^{k}\sin \alpha   \\
					\frac{\sin^2 \alpha}{2} & (-1)^{k+1} \sin \alpha & 1- \frac{\sin^2 \alpha }{2}   
 			\end{pmatrix}. 	
	\end{align*} 
with 
	\[
		k = \begin{cases} 0 & \text{if} \quad \cos \alpha >0 \; , \\
		1 & \text{if} \quad \cos \alpha <0 \; .  
		\end{cases} 
	\]
The exceptional group elements, which can not be represented in this form,  
contain the rotations 
$R_0 (\pm \frac{\pi}{2})$ in their Iwasawa decomposition.  
\end{proof}

The resulting 
decomposition of $G$ is of the form
	\[
	 	G= K'AN \; , \qquad K' = 
		\left\{ \Lambda_2 (s) ,  \Lambda_2 (s)P \mid s \in \mathbb{R} \right\} \; . 
	\]
The spatial reflection $P$ is needed to account for elements whose Iwasawa decomposition contains a rotation $R_0(\alpha)$,
with $\frac{\pi}{2} < \alpha< \frac{3 \pi}{2}$. 

\section{Homogeneous spaces, cosets and orbits}
\label{circle-mass-shell}

Consider a closed subgroup $H$ of a topological group $G$. 
Let $\mathbb{\Pi} \colon G \to G/H$ denote the canonical mapping
defined by 
	\[
		\mathbb{\Pi} (g) = gH \; . 
	\]
We equip $G/H$ with the \emph{quotient topology}\index{quotient topology}, \emph{i.e.}, a set $O \subset G/H$ is 
open if $\mathbb{\Pi}^{-1}(O) \subset G$ is open. By construction, $G$ acts \emph{transitively}\index{transitive} on $G/H$: 
	\[
		g \mathbb{\Pi} (g') = \mathbb{\Pi} (g g') \;  . 
	\]
We note that (locally) there is a continuous section $\Xi \colon G/H \to G$, 
which satisfies $\mathbb{\Pi} \circ \Xi = id $. 

\begin{lemma} 
\label{surfaceG/H}
Let $G \cong SU(1,1)$ be the two-fold covering group of $SO_0(1,2)$. 
Furthermore, let $H \subset   G  $ be the stabiliser of an arbitrary point $x \in \mathbb{R}^{1+2}$. 

Then there exists a bijective map $\mathbb{\Gamma} \colon G/H \to X$,
	\[
		g H \mapsto g x \; , \qquad g \in G \; , 
	\]
from the homogeneous space $G / H = \{ gH \mid g \in G \}$ to the orbit $X=\{ gx \in \mathbb{R}^{1+2} \mid g \in G \}$
such that 
	\begin{equation}
	\label{groupaction}
		 \mathbb{\Gamma} (gg'H) = g \mathbb{\Gamma}(g'H) \qquad \forall g, g' \in G \; ; 
	\end{equation}
\emph{i.e.}, $g (g' x) = (g \circ g') x$ for all $g, g' \in G$.
\end{lemma}

\begin{proof}
One easily verifies that  $\mathbb{\Gamma}$ is well-defined: if $g_1 H = g_2 H$,  then $g_1 = g_2 h$ for some $h \in H$, and since 
the $H$ leaves the point $x$ invariant, the map is well-defined. On the other hand, if 
	\[
		g_1	x = g_2 	x \; , 
	\]
then  $g^{-1}_2 g_1 $ fixes $x$ and thus must be in $H$. 
This implies $g_1 H = g_2 H$. Thus the map $\mathbb{\Gamma} \colon G/H \to X$ is bijective.
By construction, it satisfies \eqref{groupaction}. 
\end{proof}

In the following, we concentrate on the cases where $H$ is the stabiliser (within $SO_0(1,2)$ of a point $x \in \mathbb{R}^{1+2}$. 

\subsection{The forward light-cone}

\label{surfaces} Let us first consider the case $x = \left( \begin{smallmatrix}
1 \\ 0 \\ -1 \end{smallmatrix} \right)$. According to Lemma \ref{stabil}, $x= D(q)x $ for all $q \in \mathbb{R}$. 
Since the group $\{ D(q) \mid q \in \mathbb{R} \}$ is nilpotent, it is usually denoted by the letter $N$. Clearly, 
the point $x$ is also invariant under the reflection $P$, which generates the subgroup $M$. 
Thus the stabiliser of $x$ in the two-fold covering group of $SO_0(1,2)$ is $H= MN$. Now
recall that the map 
	\[
		(\alpha, t) \mapsto R_0 (\alpha) \Lambda_{1}(t) \left( \begin{smallmatrix}
						1 \\
						0 \\
						- 1
					\end{smallmatrix} \right) = R_0 (\alpha) \left( \begin{smallmatrix}
						{\rm e}^{-t} \\
						0 \\
						- {\rm e}^{-t}
					\end{smallmatrix} \right) \in  \partial V^+ \; , \quad t \in \mathbb{R} \; , \; \; \alpha \in [0, 2 \pi) \; , 
	\]
provides coordinates for $\partial V^+ \setminus \{ (0,0,0) \}$. 
Note that $\Lambda_{1}(t) $ leaves the light ray connecting the origin $(0,0,0)$ and the 
point $(1, 0,-1)$ invariant.   
Using Lemma \ref{iwa},
the canonical mapping $\mathbb{\Pi} \colon G \to G/MN$ is given by
	\[
		g \mapsto R_0(\alpha) \Lambda_{1}(t) MN \; , \qquad \text{with} \quad g =  R_0(\alpha) P^k \Lambda_{1}(t) D(q) \; .  
	\]

\begin{lemma} 
\label{repV+}
The homogeneous space $G/MN \cong \{ gMN \mid g \in G \}$ can be 
naturally identified with $\partial V^+ \setminus \{(0, 0,0)\}$ by setting
	\[
		\mathbb{\Gamma} (g MN) \doteq    g 	\left( \begin{smallmatrix}
							1 \\
							0 \\
							-1
						\end{smallmatrix} \right) \; , \qquad g \in G \; .  
	\]
Moreover, $g (g' x) = (g \circ g') x$ for all $g, g' \in G$ and $x \in \partial V^+ \setminus \{(0, 0,0)\}$. 
\end{lemma}

\subsection{The mass hyperboloid} 
Next consider the point $x = \left( \begin{smallmatrix}
m \\ 0 \\ 0 \end{smallmatrix} \right)$. 
Clearly, $R_0(\alpha) x = x$ for all $\alpha \in [0, 2 \pi)$. 
In other words, the stabiliser of $x$ is $K$. 

\begin{lemma} The coset space $SO_0(1,2)/K$ can be naturally identified with a 
\emph{two-fold covering} of $H_m^+ $ 
by setting
	\[
		\mathbb{\Xi} (g K) \doteq    g 	\left( \begin{smallmatrix}
							m \\
							0 \\
							0
						\end{smallmatrix} \right) \; , \qquad g \in SO_0(1,2) \; .  
	\]
\end{lemma}

\begin{proof} The rotations $K$ are the stabiliser of the point  $\left( \begin{smallmatrix}
							m \\
							0 \\
							0
						\end{smallmatrix} \right)$
in $SO_0(1,2)$. Note that 
	\[
	\Lambda_{1}(t) x = \left( \begin{smallmatrix}
						 m \cosh t \\
						0 \\
						m \sinh t
					\end{smallmatrix} \right) \in  \partial V^+ \; , \qquad t \in \mathbb{R} \; . 
	\]
Applying the rotations $R_0(\alpha)$, 
$\alpha \in [0, 2 \pi)$, to $\Lambda_{1}(t) x$ results in a two-fold covering 
of the mass hyper\-boloid~$H^+_m$. 
Thus the result follows from Lemma \ref{surfaceG/H}.
\end{proof}

Clearly, this result gives rise to the Cartan decomposition, see \eqref{cartan}.

\subsection{De Sitter space} Finally,  consider the case $x= o$. Clearly,
the boosts $\Lambda_2(t)$, $t \in \mathbb{R}$, form the stabiliser $A'$ of  the origin $o$. 
Moreover, the map
	\begin{equation}
	\label{ds-me}
	x(\alpha, t) \doteq R_0 (\alpha) \Lambda_{1} (t) o = R_0 (\alpha) \begin{pmatrix}
						r \sinh t  \\
						0 \\
						r \cosh t 
					\end{pmatrix}  ,  
	\end{equation}
with $\alpha \in [0, 2 \pi  )$ and $t \in \mathbb{R}$, 
provides coordinates\footnote{This should be compared with the chart introduced in 
\eqref{w1psi}, which only covers $\mathbb{W}_1$.} for $dS$. 

\begin{lemma} The coset space $G/ A' $ can be naturally identified with $dS$, 
by setting
	\[
		\mathbb{I} (g A') \doteq    g 	\left( \begin{smallmatrix}
							0 \\
							0 \\
							1
						\end{smallmatrix} \right) \; , \qquad g \in G \; .  
	\]
Moreover, $g (g' x) = (g \circ g') x$ for all $g, g' \in G$ and $x \in dS$. 
\end{lemma}

\subsection{Circles}
\label{circ-sub} We next consider the choice $H=MAN$.
Clearly, $H$ leaves the light ray 
	\[
	 	\{ \lambda  \left( \begin{smallmatrix}
							1 \\
							0 \\
							-1
						\end{smallmatrix} \right) \mid \lambda >0 \}
	\] 
passing through the origin and the point $\left( \begin{smallmatrix}
							1 \\
							0 \\
							-1
						\end{smallmatrix} \right)$ invariant. 
It is therefore natural to identify the factor group $SO(2) \cong G/MAN $ with the 
projective space 
	\[
		\partial \dot V^+ = \{ \dot y \mid y \in \partial V^+ \} \; ,  
	\]
formed by the light rays  $\dot y =  \{ \lambda y \mid \lambda >0 \}$, $y \in \partial V^+ $.
Each light ray in $\partial \dot V^+$ intersects the circle
	\begin{equation}
		\label{gamma-0}
		\Gamma_0 \doteq \left\{ R_0 (\alpha) \left( \begin{smallmatrix}
							1 \\
							0 \\
							-1
						\end{smallmatrix} \right) \mid \alpha \in [0 , 2 \pi) \right\} 
	\end{equation}
just once. However, it should be emphasised that the boosts in $A= \{ \Lambda_1 (t) \mid t \in \mathbb{R} \}$ do not leave the 
point $\left( \begin{smallmatrix}
							1 \\
							0 \\
							-1
						\end{smallmatrix} \right)$ invariant.   

\subsection{Mass shells} 
\label{massshell-sub}
Using the Hannabuss decomposition of 
$SO_0(1,2)$, almost every element $g$ in $SO_0(1,2)$ can be written in the form 
	\[
		g =  \Lambda_{2}(s) P^j \Lambda_{1}(t) D(q) \; , 
		\; \;   j \in \{0, 1\} \; . 
	\]
Thus almost all of the cosets $gMAN$, $g \in G$, (or light rays in $\partial \dot V^+$) are in one-to-one correspondence 
to points in the two hyperbolas
	\[
		\left\{ \Lambda_2 (s) \left( \begin{smallmatrix}
	m \\
	0 \\
	-m 
	\end{smallmatrix} \right) \mid s \in \mathbb{R} \right\} \cup
	\left\{ \Lambda_2 (s) P \left( \begin{smallmatrix}
	m \\
	0 \\
	-m 
	\end{smallmatrix} \right) \mid s \in \mathbb{R} \right\} \; . 
	\]
Note that $P\Lambda_2 (s)P= \Lambda_2 (-s)$ for all $s \in \mathbb{R}$. 

\begin{remark} It is convenient to choose a parametrization such that
	\[ 
	m   \cosh  s  = \sqrt{p_1^2 + m^2} \; , \qquad m   \sinh s = p_1 \; . 
	\]
Then, using $s = {\rm arcsinh}  \tfrac{p_1}{m}$, the measure is 
	${\rm d} s= \tfrac{ {\rm d} p_1 } { \sqrt{p_1^2 + m^2} }$ 
and $\Lambda_2 (s)$ is of the form
	\[ 
	\Lambda_2 (s) =   \begin{pmatrix}
	\frac{\sqrt{p_1^2 + m^2}}{m} & \frac{p_1}{m} & 0 \\
	\frac{p_1}{m} &  \frac{\sqrt{p_1^2 + m^2}}{m} & 0 \\
	0 & 0 & 1
	\end{pmatrix}  \; .
	\]  
\end{remark}

\section{Invariant measures}

\subsection{Haar measure} Using the Cartan decomposition,
the \emph{Haar measure}\index{Haar measure} can be decomposed as \cite[Chapter 9]{Vil}
	\begin{equation}
	\label{hmeasure}
		{\rm d} g =  {\rm d} \alpha \; \sinh t \, {\rm d} t \; {\rm d}  \alpha'  \; , \qquad  
		g = R (\alpha) \Lambda_1 (t) R (\alpha')\; ,
	\end{equation}
with $\alpha , \alpha' \in [0, 2\pi)$ and $t \in \mathbb{R}$. 

On the other hand, using the Iwasawa decomposition, 
the Haar measure on $SO_0(1,2)$ can be written as
	\[
		{\rm d} g = \frac{{\rm d} \alpha}{2 \pi} \,  {\rm e}^t {\rm d} t \, {\rm d} q \; , 
		\qquad g =  R_0(\alpha) P^k \Lambda_{1}(t) D(q) \; ,  
	\]
with  $\alpha \in [ 0, 2\pi)$, $t, q \in \mathbb{R}$ and $k \in \{ 0, 1\}$. 

\subsection{The invariant measure on the mass hyperboloid}

The restriction of the measure \eqref{hmeasure} to the mass hyperboloid $H_m^+$ equals  
${\rm d} \alpha \,  \sinh t  {\rm d} t$.  The latter 
equals twice the measure  
	\begin{align}
	\label{dSmeasure}
			\int  {\rm d}^3 p \; \theta (p_0) \delta (p_0^2 - p_1^2 - p_2^2  - m^2) 
		& = \int \rho {\rm d} \rho \, {\rm d} \alpha \,  {\rm d} p_0 \; \theta (p_0) \tfrac{ \delta \left(\rho - \sqrt{p_0^2 - m^2} \right)}
		{2 \sqrt{ p_0^2 - m^2} }  
	\end{align}
used by Bros and Moschella in \cite{BM}. This can be seen by setting $p_0= m \cosh t$, which implies ${\rm d} p_0 = m \sinh t {\rm d} t$.
In the last line we have changed coordinates, setting $p_1= \rho \sin \alpha$ and $p_2 = \rho \cos \alpha$. 

\subsection{The invariant measure on the one sheeted hyperboloid}

The measure  used by Bros and Moschella in \cite{BM}, 
	\begin{align}
	\label{hhmeasure}
		\int  {\rm d}^3 x \;   \delta (x_0^2 - x_1^2 - x_2^2 + r^2 ) 
		& = \int \rho {\rm d} \rho \, {\rm d} \psi \,  {\rm d} x_0 \;  \tfrac{ \delta \left(\rho - \sqrt{x_0^2 + r^2} \right)}
		{2 \sqrt{ x_0^2 + r^2} }  
		& = \frac{1}{2} \int_{dS}  r\, {\rm d} \psi \,  {\rm d} x_0 
	\end{align}
differs from the measure we will use, namely 
	\[ 
		{\rm d} \mu_{dS} \doteq {\rm d} x_0 \, r \, {\rm d} \psi 
	\]
by a factor two. Taking \eqref{ds-me} into account we find 
	\[
		{\rm d} \mu_{dS} = r^2 \cosh t \, {\rm d} t  \, {\rm d} \psi \; . 
	\]

\subsection{The invariant measure on the forward light cone}
Setting $p_0 = {\rm e}^{-t }$, we find $ {\rm d} p_0 = {\rm e}^{-t} {\rm d} t $ and, consequently,
the invariant measure on the forward light cone   
	\begin{align*}
		\partial V^+ \setminus \{(0, 0, 0) \} & = \left\{  \left( \begin{smallmatrix} p_0 \\
		p_0 \sin \alpha \\
		- p_0 \cos \alpha
		\end{smallmatrix} \right) \mid p_0 >0 , \alpha \in [0, 2 \pi) \right\} \\
		& \cong \{g MN \mid g \in G \}
	\end{align*}
is given --- up to normalisation --- by the formula
\cite[Chapter 9.1.9, Equ.~(13)]{Vil}
	\[
	 	|p_0|^{-1} {\rm d} p_1 {\rm d} p_2 = {\rm d} p_0 {\rm d} \alpha \; , 
	\]
in agreement with taking the limit $m \to 0$ in \eqref{hhmeasure}. There  one finds
$\tfrac{1}{2} {\rm d} \alpha \, {\rm d} p_0 $. 

\subsection{Measures on contours on the light-cone}

Let $\Gamma$ be a contour on the light cone which intersects every light ray of the light cone at one point. 
Following \cite[Chapter 9.1.9]{Vil}, we denote by ${\rm d} \mu_\Gamma$ a measure on this contour such that 
\label{Gammacpluspage}
	\[
		|p_0|^{-1} {\rm d} p_1 {\rm d} p_2 = {\rm d} \lambda {\rm d} \mu_\Gamma (\eta) \; , \qquad \eta \in \Gamma \; , 
	\]
where $p= \lambda \eta$, $\lambda >0$, $\eta \in \Gamma$. It follows that if a function $f(p)$ on $\partial V^+$ is homogeneous 
of degree $-1$, that is $f(\lambda p) = \lambda^{-1} f(p)$, $\lambda >0$, then the integral 
	\[
		\int_\Gamma {\rm d} \mu_\Gamma (\eta) f (\eta) 
	\]
does not depend on the choice of $\Gamma$; in agreement with \cite[Proposition 10]{BM2}.

\part{Harmonic Analysis}

\chapter{Induced Representations for the Lorentz Group}

Assume we are given a quantum theory, formulated in terms of operators (which 
play the role of the observables) and normalised, positive linear functionals, \emph{i.e.}, states (which play the role
of physical expectation values). Then  
a symmetry operation is realised by a so-called \emph{Wigner automorphism} \cite{Wigner}. 
Given an invariant state $\omega$, any such automorphism can be implemented by 
either a unitary or an anti-unitary operator in the GNS Hilbert space ${\mathcal H}_\omega$. 
These operators are unique up to a factor. Thus, given a group of symmetries, one is led to study
its \emph{projective representations}. The latter are in one-to-one correspondence with the 
unitary representations of the \emph{universal covering group} of the symmetry 
group\footnote{The case $SO_0(1,2)$ is somewhat exceptional, as the universal covering group
does \emph{not} coincide with the two-fold covering group. 
However, for simplicity, here we deal only with the latter. 
These representations describing anyons
will be discussed elsewhere.}. 

Unitary irreducible representations of the Lorentz group $SO_0(1,2)$ (and its two-fold covering group) 
have first\footnote{At the same time, Gelfand and Naimark investigated the group $SL(2, \mathbb{C})$ \cite{GN}.}
been constructed by Bargmann~\cite{Ba}, using multipliers. The latter had appeared in
Schur's theory of projective representations. 
The approach we will follow here was pioneered by Wigner \cite{W} and 
Mackey~\cite{Mack}, based on earlier work by Frobenius (see, \emph{e.g.},~\cite{Vara}). 

\section{The general case\index{induced representation}}

We briefly recall some key elements of 
the general theory of induced representations,  following \cite{Folland} (see also~\cite{Bruhat, GG, 
Knapp, Kostant, Lipsman, Mack, Vara, Vil, VG, VGG}).

\subsection{Modular functions} 

Let $G$ be a locally compact topological group $G$. 
Then $G$ has a left invariant \emph{Haar measure}\index{Haar measure} 
$\mu_G$ on the $\sigma$-algebra of Borel sets, which is unique up to a normalisation 
factor; see, \emph{e.g.}, \cite[Theorems~(2.10) and (2.20)]{Folland}.
In general, the left Haar measure $\mu_G$ on $G$ is not equal to the 
right-invariant Haar measure. However, there always exists a multiplicative $\mathbb{R}^+$-valued function 
$\Delta_G$ on~$G$, called the \emph{modular function}\index{modular function} of $G$, such that
	\[
		\int_G {\rm d} \mu_G (g) \; f(g g') = \frac{1}{\Delta_G (g')} \int  {\rm d} \mu_G (g) \; f(g) 
		\qquad \forall g' \in G
	\]
and for every $\mu_G$-integrable function $f$ on $G$. 
The modular function relates the left and the right invariant Haar measure:
	\[
		\int_G {\rm d} \mu_G (g) \; f(g^{-1}) = \int_G {\rm d} \mu_G (g) \; f(g) \Delta_G (g^{-1})  \;  . 
	\]
In case $ \Delta_G (g)=1$ for all $g \in G$,
the left and the right Haar measure coincide, and~$G$ is called \emph{unimodular}.

\goodbreak
\begin{lemma} The twofold covering group of $SO_0(1,2)$ is unimodular.
\end{lemma}

\begin{proof} The two-fold covering group of $SO_0(1,2)$ is isomorphic to $SL(2, \mathbb{R})$.
For $SL(2, \mathbb{R})$, consider the action $\Delta_* \colon \mathfrak{sl}(2,\mathbb{R}) \to \mathbb{R}$ of the character $\Delta$
(given by the modular function)
on the Lie algebra of $SL(2, \mathbb{R})$. Since $\mathbb{R}$ is
abelian, $\Delta_* ( [ X, Y] ]) = 0$ for all $X, Y \in \mathfrak{sl}(2,\mathbb{R})$. But every element of $\mathfrak{sl}(2,\mathbb{R})$ 
is of this form, hence $\Delta_*=0$. Since $SL(2, \mathbb{R})$ is connected this determines $ \Delta_G (g)=1$ for all $g \in G$.
\end{proof}

\begin{remark}
\label{notunimodular}
Consider the group $AN = \{ \Lambda_1(t) D(q) \mid t, q \in \mathbb{R} \}$. 
We compute (using \eqref{scalingH})
	\begin{align*}
		\int {\rm d}t {\rm d}q \;  f \bigl(  \Lambda_1(t)  D(q)  \Lambda_1(t') D(q') \bigr) 
		& = \int {\rm d}t {\rm d}q \;  f \bigl(  \Lambda_1(t+t') D( {\rm e}^{-t'} q + q') \bigr)  \\
		& = \int {\rm d}t'' {\rm d}q'' {\rm e}^{t'}  \;  f \bigl(  \Lambda_1(t'') D( q'' ) \bigr)  \; ;
	\end{align*}
\emph{i.e.}, $\Delta_{AN} \bigl( \Lambda_1(t') D(q') \bigr) = {\rm e}^{-t'}$. 
Thus $AN$ is \emph{not} unimodular. 
\end{remark}

\subsection{Invariant measures on the quotient space}

Let $H$ be a closed subgroup of a locally compact group $G$. Moreover,
let $\Delta_G$ and $\Delta_H$ denote the modular functions of $G$ and $H$, respectively. 

\begin{theorem}[Theorem 2.49, \cite{Folland}]
The quotient space $G/H$ admits a nonzero positive $G$-invariant measure $\mu_{G / H}$, 
if and only if 
	\begin{equation}
	\label{modular function}
		\Delta_{G} \upharpoonright H = \Delta_H \; . 
	\end{equation}
If \eqref{modular function} holds, then the positive invariant measure $\mu_{G / H}$
is unique (up to multiplication
by a positive constant). Moreover, one can normalize the invariant measure $\mu_{G/H}$ on 
$G/H$ such that for every $f$ in $C_0 (G)$,
	\begin{equation}
		\label{quotient measure}
		\int_{G / H} {\rm d} \mu_{G / H} (gH) \int_H {\rm d} \mu_H (h)  \; f(gh) 
		= \int_G {\rm d} \mu_G (g) \; f(g)  \; , 
	\end{equation}
where $\mu_G$  and $\mu_H$ denote the Haar measures of $G$ and $H$, respectively.
\end{theorem}

\begin{remark}
One can define a linear map $C_0(G) \to C_0(G/H)$ by setting
	\[
		f^H (gH) \doteq \int_H {\rm d} \mu_H (h)  \; f(gh)  \; , \qquad f \in C_0 (G) \; . 
	\]
The identity \eqref{quotient measure} then takes the form 
	\begin{equation}
	\label{GHmeasure}
		\int_{G / H} {\rm d} \mu_{G / H} (gH) \; f^H (gH) 
		= \int_G {\rm d} \mu_G (g) \; f(g)  \; .
	\end{equation}
\end{remark}

\begin{lemma}
Let $H \subseteq G$ be compact, then $\Delta_{G} \upharpoonright H = 1$. In particular, if $G$ is
compact, then $G$ is unimodular.
\end{lemma}

\begin{proof} As $\Delta_G$ 
is continuous, it follows that $\Delta_G (H)$ is a compact subgroup of~$\mathbb{R}^+$
and hence equal to $\{1\}$.
\end{proof}

\subsection{Quasi-invariant measures}

In the case we are interested in, the condition \eqref{modular function} is not satisfied. 
Consequently, there is no $G$-invariant measure on $G/H$. 

\begin{definition}
A regular Borel measure $\mu$ on $G/H$ is called 
\begin{itemize}
\item[$i.)$] \emph{quasi-invariant}\index{quasi-invariant measure},
if the measure  $\mu$ and the measure
	\[
		\mu^g (\, . \, ) \doteq \mu(g \, . \, )
	\]
are mutually absolutely continuous for all $g \in G$;
\item[$ii.)$] \emph{strongly quasi-invariant}, if there exists a continuous $\mathbb{R}^+$-valued 
function $\lambda_g (g'H)$ on $G \times G/H$ such that
	\[
		\qquad \qquad
		{\rm d} \mu^g (g'H) = \lambda_g (g'H) {\rm d} \mu (g'H)
		\qquad \forall g \in G\; ,  \quad \forall g'H \in G/H \; , 
	\]
\emph{i.e.}, $\lambda_g$ is the Radon-Nikodym derivative
	\begin{equation}
	\label{lambda-function}
		\lambda_g (g'H) \doteq\frac{ {\rm d} \mu^g} { {\rm d} \mu}  (g'H)  \; , \qquad g, g' \in G \; . 
	\end{equation}
\end{itemize}
\end{definition}

Quasi-invariant measures send null sets  into null sets under the action 
of $G$. Strongly quasi-invariant measures on $G/H$ are closely related to
rho-functions on~$G$. 

\begin{definition}
A real-valued function $\rho$ on~$G$ is a \emph{rho-function}\index{rho-function} for $(G,H)$, if 
it is positive, continuous, and satisfies
	\[
		\rho(gh) = \frac{\Delta_H(h)}{ \Delta_G (h)} \rho(g)
	\]
for all $g \in G$ and $h \in H$.
\end{definition}
	
\begin{remark} 
\label{rhocomputed}
Let $G= SO_0(1,2)$ and $AN = \{ \Lambda_1(t) D(q) \mid t, q \in \mathbb{R} \}$. It follows that  $\Delta_G(h)=1$ for all $h \in AN$.  
Consequently, 
	\[
		\rho(h) = \Delta_{AN} (h)  \rho(\mathbb{1}) \qquad \forall h \in AN \; . 
	\]
Here $\mathbb{1}$ denotes the identity in $G$. We have already seen (in Remark \ref{notunimodular}) that 
	\[
		\Delta_{AN} \bigl( \Lambda_1(t) D(q) \bigr) = {\rm e}^{-t} \; , \qquad t, q \in \mathbb{R} \; . 
	\]
Thus (possibly up to a constant) $\rho\bigl( \Lambda_1(t) D(q) \bigr) = {\rm e}^{-t} $. 
\end{remark}

\begin{theorem}[see Theorem (2.56), \cite{Folland}]
\label{cocycle-relation}
Let $\rho$ be a rho-function for~$(G,H)$. It follows that
\begin{itemize}
\item[$i.)$] there exists a strongly quasi-invariant measure $\mu$ on $G/H$ such that 
	\begin{equation}
	\label{GHmeasurea}
		\int_{G/H} {\rm d} \mu  (gH) \; f^H(gH)   
		= \int_G {\rm d} \mu_G (g) \;  \rho(g) f(g)  \qquad \forall f \in C_0(G) \; ;
	\end{equation}
\item[$ii.)$] the measures $\mu^g$ and $\mu$  are absolutely continuous to each other;
\item[$iii.)$] the Radon-Nikodym derivative\index{Radon-Nikodym derivative} is given by 
	\begin{equation}
	\label{lambda-function-2}
		\lambda_g (g'H) = \frac{\rho (gg')}{\rho (g') } \; , \qquad g, g' \in G \; ;
	\end{equation}
\item[$iv.)$] the Radon-Nikodym derivative satisfies the \emph{cocycle relation}\index{cocycle}
	\[
		\lambda_{g_1 g_2} (g H ) = \lambda_{g_1} (g_2g H) \lambda_{g_2} (gH)  		
		\qquad 
		\forall g_1, g_2, g  \in G \; .  
	\]	
\end{itemize}
\end{theorem}
	
\goodbreak 
\begin{remarks}
\label{rhot} \quad
\begin{itemize}
\item[$i.)$] Clearly, \eqref{GHmeasurea} should be compared with \eqref{GHmeasure}.
\item[$ii.)$] 
In case $g'H= R_0(\alpha') MAN$, Remark \ref{rhocomputed} implies that $\rho(R_0(\alpha'))= 1$. Thus 
\eqref{lambda-function} implies that
	\[
		\lambda_g (g'H) = \rho (g) \; .
	\]
In other words, 
	\begin{equation}
	\label{rhog}
		\lambda_{R_0(\alpha) \Lambda_1(t) D(q)} (g'H) = {\rm e}^{-t}
	\end{equation}
for all $\alpha \in [0, 2 \pi)$ and $t, q \in \mathbb{R}$.
\end{itemize}
\end{remarks}

\subsection{Induced representations\index{representation}}
Let $G$ be a locally compact group, $H$ a closed subgroup, 
	\[ 
		\mathbb{\Gamma} \colon G \to G / H
	\]
the canonical quotient map, and $\pi \colon H \to {\mathscr B}({\mathfrak H})$ a representation 
of $H$ on some Hilbert space ${\mathfrak H}$. Denote the norm and the inner product on 
${\mathfrak H}$ by $\| u \|_{\mathfrak H}$
and  $ \langle u , v \rangle_{\mathfrak H}$, and denote by 
$C(G,{\mathfrak H})$ the space of continuous functions from $G$ to ${\mathfrak H}$. 
Now consider the following space of vector valued functions \cite{Folland}:
	\[
		\mathcal{F}_0 = \{ f \in C(G,{\mathfrak H}) \mid \mathbb{\Gamma} ({\rm supp} \, f) \;
		\text{is compact}, \; f(gh)= \pi (h^{-1}) f(g) \; \text{for} \; h \in H, g \in G \}.
	\]
Note that if $\pi$ is unitary and $f \in \mathcal{F}_0$, then $\| f \|_{\scriptscriptstyle {\mathfrak H} }$ depends only on 
the equivalence classes $gH$, $g \in G$. 

\begin{definition}
\label{Fmu}
Let $\mu$ be a strongly quasi-invariant measure on $G/H$.
\begin{itemize}
\item[$i.)$] In case $\pi$ is unitary, define an \emph{inner product}\index{inner product} on $\mathcal{F}_0$ by setting
	\begin{equation}
	\label{Spr}
		\langle f, f' \rangle_{\mu} \doteq \int_{G/H} 
		{\rm d} \mu(gH) \; \langle f (g), f' (g) \rangle_{\mathfrak H}  \; , 
		\qquad  f, f' \in \mathcal{F}_0 \; . 
	\end{equation}
\item[$ii.)$]  
In case $\pi$ is a non-unitary character of $H$, the scalar product \eqref{Spr} gives rise to a bounded (non-unitary) representation
of $G$ on the completion of~$\mathcal{F}_0$. However, it may still be 
a possibility to find a new \emph{inner product}\index{inner product} on $\mathcal{F}_0$ 
with respect to which the induced representation becomes unitary. 
If that is possible, we refer to the resulting unitary representation on the completion of $\mathcal{F}_0$
as the complementary series representation (see, \emph{e.g.}, \cite[p.~32]{Lipsman}).
\end{itemize}
In both cases, denote the completing of $\mathcal{F}_0$ w.r.t.~the norm $\| f \|_\mu \doteq \sqrt{ \langle f, f\rangle_{\mu}}$ 
by $\mathcal{F}_\mu$. 
The \emph{induced representation}\index{induced representation} $ \Pi_\mu (g)$ on the Hilbert 
space~$\mathcal{F}_\mu$ is specified by setting
	\begin{align}
		\bigl( \Pi_\mu (g) f \bigr) (g') & \doteq \sqrt{\lambda_g(g'H)} \; f (g^{-1} g') \; , \qquad g, g' \in G \; , 
	\label{indrep}
	\end{align}
with $\lambda_g$ the Radon-Nikodym derivative defined in \eqref{lambda-function}. 
\end{definition}

\begin{remark}
Note that while $\Pi_\mu$ depends on $\mu$, 
its unitary equivalence class depends only on $\pi$.
\end{remark}

The induced representation is equivalent to a representation on
$C_0(G/H,\mathfrak{H})$, as follows. Let $\Xi \colon G/H \to G$ be a smooth global section. Then $\mathcal{F}_0$
is, as a linear space, isomorphic to $C_0(G/H,\mathfrak{H})$ by identifying $f\in\mathcal{F}_0$ with 
$\tilde f\in C_0(G/H,\mathfrak{H})$ defined by\footnote{The inverse map 
$h\in C_0(G/H,\mathfrak{H}) \mapsto \check{h} \in \mathcal{F}_0$
is given by 
	\[
		\check{h}(g) \doteq \pi\big(g^{-1}\Xi(p)\big)\, h(p) \; , \quad p\doteq \Pi(g) \; .
	\]
}
	\begin{equation} 
		\label{eqFeqFGH}
			\tilde{f}(p) \doteq f\big(\Xi(p)\big) \; ,  \quad p\in G/H \; .
	\end{equation} 
The scalar product~\eqref{Spr} in $\mathcal{F}_0$ goes over, under
this equivalence, into the $L^2$-product in $C_0(G/H;\mathfrak{H})$: 
	\begin{equation} 
		\label{eqScalProdGH} 
		\| f\|^2_\mu = \|\tilde f\|^2_{L^2(G/H;\mathfrak{H})} = \int_{G/H} {\rm d} \mu(gH)\, \|f(gH)\|^2_{\mathfrak{H}}. 
	\end{equation}
Further, for $g\in G$ and $p\in G/H$, the group elements $g^{-1} \Xi(p)$
and $\Xi(g^{-1}\cdot p)$\footnote{We denote the action of $G$ in $G/H$
by a dot, $g\cdot (gH)\equiv (gg')H$.} differ by an element in $H$, the so-called Wigner
rotation $\Omega(g,p)$: 
	\begin{equation} 
		\label{eqWignerRot}
		g^{-1} \Xi(p) = \Xi(g^{-1}\cdot p) \, \Omega(g,p)^{-1} \; ,
		\quad 
 		\Omega(g,p) \doteq \Xi(p)\,g \, \Xi(g^{-1}\cdot p) \, \in H \; . 
	\end{equation}
Using this fact, the induced representation~\eqref{indrep} on
$\mathcal{F}_0$ is equivalent, via the isomorphism~\eqref{eqFeqFGH},
to the representation $ \widetilde{\Pi}_\mu$ defined on $C_0(G/H,\mathfrak{H})$ by  
	\begin{equation} 
		\label{eqRepGH}
		\bigl( \widetilde{\Pi}_\mu (g) h \bigr) (p) \doteq 
		\lambda_g\big(\Xi(p)H\big)^{\frac{1}{2}} \; \pi\big(\Omega(g,p)\big)\; h(g^{-1}\cdot p) \; .
	\end{equation}
As the isomorphism~\eqref{eqFeqFGH} intertwines the respective
scalar products, an \hbox{(anti-)} unitary operator in $\mathcal{F}_0$ goes
over into an (anti-) unitary operator in $L^2(G/H; \mathfrak{H})$. Thus, 
if the representation $\pi$ of $H$ in $\mathfrak{H}$ is unitary, then 
the representation $\widetilde{\Pi}_\mu$ is unitary in
$L^2(G/H; \mathfrak{H})$.

\section{Induced representations for $SO_0(1,2)$}

A representation $\pi_{\nu}^{\pm} \colon MAN \to \mathbb{C}$
of the closed solvable subgroup $MAN$ of the  two-fold covering group of $SO_0(1,2)$ 
on $\mathbb{C}$ is defined by lifting a character $\chi_\nu$
of $A$ to $AN$, and taking its product with a representation of $M$: set 
$\pi_{\nu}^{\pm} =  (\sigma_\pm \otimes \chi_\nu \otimes \mathbb{1})$, where
	\begin{equation}
	\label{indrepc} 
		\sigma_\pm \otimes \chi_\nu \otimes \mathbb{1} \colon \;   man 
		\mapsto \chi_\nu (a  )\sigma_\pm (m) \; , 
	\end{equation} 
with $\sigma_\pm (\mathbb{1}_3)= 1$, 
	\[ 
		\sigma_\pm    \begin{pmatrix}
					1 &  0 &0 \\
					0 &  - 1 & 0  \\
					0  & 0 & - 1   
				\end{pmatrix}  = \pm 1 \;  ,  \quad \text{and} \quad	\chi_\nu
								\begin{pmatrix}
					\cosh t &  0 &\sinh t \\
					0 &  1&0  \\
					\sinh t & 0 & \cosh t   
 					\end{pmatrix} 
 			= {\rm e}^{i \nu t} \; . 
		\]
Thus
	\[
		\pi_\nu^\pm\big(P^k\Lambda_1(t)D(q)\big)\doteq(\pm 1)^ke^{i\nu t} \; . 
	\]
The induced representation is initially defined on the space\footnote{As we have seen in 
Section \ref{circle-mass-shell}, the cosets $\{gMAN \in g \in G \}$ can be identified either with a 
circle on the forward light cone (using the Iwasawa decomposition of the group) or with a pair of mass-hyperbolas 
on the forward light cone (using the Hannabus decomposition of $SO_0(1,2)$).} 
(see De\-fini\-tion~\ref{Fmu})
	\[ 
		{\mathfrak h}_{\nu,0}^\pm \doteq \left\{ f \in C (G, \mathbb{C}) \mid \mathbb{\Gamma} ({\rm supp} \, f) \;
		\text{is compact,} \; f(gh)= \pi_{\nu}^{\pm} (h^{-1}) f(g) \; \text{for} \; h \in MAN \right\}   .  
	\] 
This definition implies that 
\begin{itemize}
\item [$i.)$] a function $f \in {\mathfrak h}_{\nu,0}^+$ depends only on
  the cosets $gMN$, $g \in G$, as	\begin{equation}
	\label{fcoset}
		f (gmn) = f(g) \qquad \forall g \in G \; , \; \forall m \in M \; , \; \forall n \in N \; ;
	\end{equation}
\item [$ii.)$] if $f$ is a function in ${\mathfrak h}_{\nu,0}^+$ or ${\mathfrak h}_{\nu,0}^-$, then
	\[
		f \bigl( g \Lambda_1(t)D(q) \bigr)  = p_0^{i \nu} f (g) \; , \qquad \text{ with } p_0 = {\rm e}^{-t} >0 \; ; 
	\]
\item [$iii.)$] in case $\nu \in \mathbb{R}$, the representation $\pi_{\nu}^{\pm} $ is unitary; 
\item [$iv.)$] in case $\nu$ is purely imaginary, 
the representation \eqref{indrepc} is {\em no longer} a unitary representation 
of $MAN$ in~$\mathbb{C}$. However, \eqref{indrepc} implies that 
	\begin{equation}
	\label{Compseries}
	\int_G {\rm d} g \, | f(gh) |^2 = \int_G {\rm d} g \,  \bigl(\pi_{\nu}^{\pm} (h) f\bigr)(g) 
	\bigl(\pi_{-\nu}^{\pm} (h) f\bigr)(g) \; \qquad \forall h \in MAN \; .  
	\end{equation}
\end{itemize}
We will explore these facts further in the next subsection. 

Let us denote by $\Pi_\nu^{\pm}$ the representation of $G$ induced from the  
representation~$\pi_{\nu}^{\pm}$ of the closed subgroup $MAN$, 
	\begin{equation} 
		\label{eqIndRepSO12}
		\bigl( \Pi_\nu^{\pm}(g)f \bigr)(g')=\sqrt{\lambda_g(g'MAN)}\,f(g^{-1}g') \; , \quad f\in 	{\mathfrak h}_{\nu,0}^\pm \; .
	\end{equation}

\begin{remark} To compute explicit expressions for 
the representation~\eqref{eqIndRepSO12}
(for specific choices of $g \in G$)
one can take advantage of \eqref{fcoset}. According to 
Lemma \ref{repV+} the map 
	\[
		\mathbb{\Gamma} (g MN) \doteq    g 	\left( \begin{smallmatrix}
							1 \\
							0 \\
							-1
						\end{smallmatrix} \right) \; , \qquad g \in G \; ,  
	\]
defines a bijection, which identifies the homogeneous space 
	\[
		\{gMN \mid g \in SO_0(1,2) \}= \{ R_0(\alpha)  \Lambda_{1}(t)MN \mid \alpha \in [0, 2\pi), t \in \mathbb{R}  \}
	\]
with the forward light cone
	\begin{equation}
	\label{X-lightcone}
		\partial V^+  \cong \left\{ (\alpha, {\rm e}^{-t}) \in S^1 \times \mathbb{R}^+  \mid \alpha \in [0, 2\pi), t \in \mathbb{R} \right\} \; . 
	\end{equation}
Setting $p_0 = {\rm e}^{-t}$, 
the action of $SO_0(1,2)$ on the forward light cone
$\partial V^+$ is given by~\eqref{lambda2-s}, \emph{i.e.},
\label{umLambdapage}
	\begin{align}
		\label{udrei} 
 			\Lambda_2 (s)^{-1}  (\alpha', p_0') 
			& =  \bigl( \alpha_2 \,  , \,  p_0' (\cosh s - \sinh s \sin \alpha') \bigr)
	\nonumber \\  
			\Lambda_1 (t)^{-1} (\alpha', p_0')  
		&= \bigl( \alpha_1  \,  , \,  p_0' (\cosh t - \sinh t \cos \alpha') \bigr) 
	 \nonumber \\
			R_{0}^{-1} (\alpha)   (\alpha', p_0')  
		&=  (\alpha' + \alpha \,  , \, p_0')     \; , 
	 \nonumber \\
			P  (\alpha', p_0')  
		&
		= (\alpha'  + \pi \,  , \, p_0')     \; , 
	\end{align}
with
	\begin{align*}
	( \sin \alpha_2 , \cos \alpha_2)  & = \left( \tfrac{-\sinh s + \cosh s \sin \alpha'}
	{\cosh s -  \sinh s \sin \alpha'} \; , 
	  \tfrac{\cos \alpha'}{\cosh s -  \sinh s \sin \alpha'} \right) \; , \\
	(\sin \alpha_1 , \cos \alpha_1)  &= 
	\left( \tfrac{\sin \alpha'}{\cosh t - \sinh t \cos \alpha'} ,  \tfrac{-\sinh t 
	+ \cosh t \sin \alpha'}{\cosh t - \sinh t \cos \alpha'} \right) \; .
	\end{align*}
\end{remark}

\begin{lemma} 
The restriction of the Lorentz invariant measure on $\mathbb{R}^{1+2}$ to the forward 
light-cone~\eqref{X-lightcone}, given by
	\begin{equation}
	\label{measure-lightcone}
		{\rm d} \mu (\alpha, p_0) \doteq  \tfrac{{\rm d} \alpha}{2 \pi}  {\rm d} p_0 \; , 
		\qquad p_0 = {\rm e}^{-t} \; , 
	\end{equation}
defines a strongly quasi-invariant measure on the homogeneous space $G/MN \cong
 \{ gMN \mid g \in SO_0(1,2) \}$. Its Radon--Nikodym derivative is
	\begin{align}
	\label{RN}
	 	 \frac{{\rm d} \mu^g}  
		{ {\rm d} \mu } (g' MN) =  \chi_{i} \left( \Lambda_1(t )\right) 		\; , 
		\qquad  \text{with} \quad g^{-1}g' =  R_0( 2 \alpha) P^k \Lambda_{1}(t) D(q) \; .     
	\end{align}
\end{lemma}

\begin{proof} 
Set $ g^{-1}g' =  R_0(2 \alpha) P^k \Lambda_{1}(t) D(q)$ and let $f$ be a measurable function 
on~\eqref{X-lightcone}. Then 
	\begin{align*}
		\int {\rm d} \mu^g (\alpha', p_0') \;  f(\alpha', p_0') & = \int {\rm d} \mu (\alpha', p_0') \; f(\alpha, p_0) \\
		& =  \int {\rm d} \mu (\alpha, p_0) \;   p_0^{-1}   f(\alpha, p_0) \; .     
	\end{align*}
The second identity follows from Remark \ref{rhot} $ii.)$. Note that $\chi_{i} \left( \Lambda_1(t )\right) = {\rm e}^{-t}$.
Thus \eqref{RN} follows.
As expected, $\lambda_g (g'H)$ satisfies the cocycle relation stated in Proposition \ref{cocycle-relation}.
This result is in agreement with \cite[p.169, 170] {Knapp}.
\end{proof}

\section{Unitary representations on a circle on the light cone}
\label{UIRc}

The Iwasawa decomposition together with the definition of ${\mathfrak h}_{\nu,0}^\pm$ 
imply that a function $f \in {\mathfrak h}_{\nu,0}^\pm$ is determined 
by the restriction $f_{\upharpoonrightÊK}\, $ of $f$ to $K$. We have seen that $\{ gMN \mid g \in G \}$ can be identified 
with $\partial V^+$, while $\{ gMAN \mid g \in G \}$ can be identified with the projective space formed by the light rays on the 
forward light cone, see Subsection \ref{circ-sub}.

The latter can be identified with the subgroup $SO(2)$ of $G$,
considered as a topological space, and we have 
	\[
		G/MAN \cong SO(2) \; .
	\]      
This can be also directly seen by considering the unique Iwasawa decomposition
$G=KAN = SO(2)\, MAN$. The projection $G\to G/MAN$ is then given by 
	\[
		R_0(2 \alpha)P^k \Lambda_1(t)D(q)\mapsto R_0(2 \alpha) \; ,  \qquad \alpha \in [0, \pi) \; , 
	\]
and the embedding of $SO(2)$ into $G$ can be considered as a global
smooth section  
	\[ 
		\Xi \colon G/MAN \to G,\quad R_0(\alpha) \mapsto R_0(\alpha) \; .
	\]
We wish to translate the induced represention~\eqref{eqIndRepSO12} as
in Eq.~\eqref{eqRepGH} to a representation acting  on $C_0(SO(2))$. 
For given $g\in G$ and $R_0(\alpha')\in G/MAN\cong SO(2)$ there are unique
$\alpha,k,t$ and  $q$ such that
	\begin{equation}
		\label{aktq}
		g^{-1} R_0(\alpha')= R_0 (\alpha) \; P^k \Lambda_1(t) D(q)\; .  
	\end{equation}
Taking the class w.r.t.\ $MAN$, this implies that  $g^{-1} R_0(\alpha')=R_0(\alpha)$  in the
sense of the action of $G$ on $G/MAN\cong SO(2)$
and that $\Xi \bigl( g^{-1} R_0(\alpha')\bigr) = R_0(\alpha)\in G$.  Eq.~\eqref{aktq} then implies that 
$P^k \Lambda_1(t) D(q)=\Omega\big(g,R_0(\alpha')\big)^{-1}$, see \eqref{eqWignerRot}.

Let us denote by $\widetilde{\Pi}_\nu^{\pm}$ the representation 
living on $C(SO(2))$ equivalent to the induced
representation ${\Pi}_\nu^{\pm}$~\eqref{eqIndRepSO12}.
According to \eqref{eqRepGH}, it acts as   
	\begin{align*} 
		\big(\widetilde{\Pi}_\nu^{\pm}(g) f \big)_{\upharpoonrightÊK} (R_0(\alpha')) & = 
			\sqrt{\lambda_g \bigl( R_0 (\alpha') MAN \bigr)} \;
			\pi_\nu^{\pm}\big(\Omega(g,R_0(\alpha')\big) \; f_{\upharpoonrightÊK} (g^{-1} \cdot R_0(\alpha')) \\
		&= {\rm e}^{-\frac{1}{2}t} \;\pi_\nu^{\pm}\big(P^k \Lambda_1(t)
			D(q)\big)^{-1} \; f_{\upharpoonrightÊK} (R_0(\alpha)) \\
		&= {\rm e}^{(-\frac{1}{2}-i\nu)  t}	(\pm 1)^{k } f_{\upharpoonrightÊK} \left( R_0 ( \alpha )\right)  \; . 
\end{align*}
We have used $ \lambda_g ( R_0 (\alpha') MAN ) = \rho (g) = {\rm e}^{- t}$ and 
	\[
		\pi_\nu^{\pm}\big(P^k \Lambda_1(t) D(q)\big)^{-1}=e^{-i\nu  t}(\pm 1)^k \; , 
	\]
as well as $g^{-1}\cdot R_0(\alpha')=R_0(\alpha)$. 
Note that if $\nu \in \mathbb{R}$, then $\pi_\nu^\pm$ \eqref{indrepc} is a {\em unitary} 
representation of $MAN$ in~$\mathbb{C}$. 

Identifying $SO(2)$ with the circle $\Gamma_0$ 
introduced in \eqref{gamma-0} by setting 
	\[
	h (\alpha) \doteq f_{\upharpoonrightÊK} (R_0 (\alpha)) \; , \qquad \alpha \in [0, 2 \pi) \; , 
	\]
the representation $\widetilde{\Pi}_\nu^{\pm}$ extends to a unitary representation on $L^2(\Gamma_0,{\rm d}\mu_{\Gamma_0})$, with
${\rm d}\mu_{\Gamma_0}= \tfrac{{\rm d} \alpha}{2 \pi}$ the strongly quasi-invariant 
measure on $G/MAN \cong \Gamma_0$;  see the remark after Eq.~\eqref{eqRepGH}. 

\begin{proposition}
\label{Prop:2.1}
Let $\widetilde {\mathfrak h}_\nu^\pm$ denote the completion of $C_0(\Gamma_0)$     
with respect to one of the following norms:
\begin{itemize}
\item[$i.)$]
in case $0 < \zeta < \frac{1}{2}$, define for $\nu =  \pm i \sqrt{\frac{1}{4} - \zeta^2}$ a norm on ${\mathfrak h}_{\nu,0}^\pm$ by setting
	\[ 
		\| h \|_\nu^2 \doteq 
			\int_{\Gamma_0} \frac{{\rm d} \alpha}
			{2 \pi}  \; 
			\overline{  h (\alpha) }  \int_{\Gamma_0} \frac{{\rm d} \alpha'}
			{2 \pi} \; \varrho_\nu (\alpha - \alpha') \,   h ( \alpha') \; ,
	\]
with $\varrho_\nu (\alpha) 
		\doteq \frac{\Gamma (\frac{1}{2} - i \nu )}{\Gamma ( \frac{1}{2})
		\Gamma ( - i \nu )} \; \left(\sin \tfrac{\alpha}{2} \right)^{-\frac{1}{2}  - i \nu } 	\pi \; $;
\item[$ii.)$]
in case $\frac{1}{2} \le \zeta$, define for $\nu = \pm \sqrt{\zeta^2 -
  \frac{1}{4}}$ a norm on ${\mathfrak h}_{\nu,0}^\pm$ by setting
	\[ 
		\| h \|^2 \doteq 
			\frac{1}{2 \pi} \int_{\Gamma_0} {\rm d}   \alpha \;  | h ( \alpha)|^2 \; . 
	\]
\end{itemize}
It follows that for all $\zeta>0$ the operators $\Pi_{\nu}^{\pm} (g)
$, $g \in G$, extend from $C_0(\Gamma_0)$  
to a unitary representation
	\begin{align}
		\left( \widetilde u_{\nu}^{\pm} (g) h \right) (\alpha')  
		& = ( \pm 1 )^{k } {\rm e}^{( - \frac{1}{2} - i\nu) t }
		h \left( \alpha \right)  \; 
	\label{FG1}
	\end{align}
of the two fold covering group of the Lorentz group $SO_0(1,2)$. 
The parameters $\alpha, k, t, q$ on the r.h.s.~ 
are given by \eqref{aktq}.
\end{proposition}

\begin{proof}
The case $\nu \in \mathbb{R}$ follows from the discussion preceding the proposition. 
In the case $-\tfrac{1}{2} < i \nu < \tfrac{1}{2}$, note that the 
norm reads 
	\[
		\| h\|_\nu^2 = \langle h,A_\nu h\rangle_{L^2(\Gamma_0)} \; , 
	\]
where $A_\nu$ is the operator acting on $C_0(\Gamma_0)$ 
as  
	\[ 
		(A_\nu h) (\alpha) \doteq  \int_{\Gamma_0} \frac{ {\rm d} \alpha' }{2 \pi}  \; 
        			\varrho_\nu (\alpha-\alpha') h(\alpha') \; , 
			\qquad \alpha \in [0, 2 \pi) \; . 
	\]
We show below that this map intertwines $\widetilde{\Pi}_{-\nu}^{\pm}$
and $\widetilde{\Pi}_{\nu}^{\pm}$; see \eqref{intertwiner}. 
Using this fact and the fact that $\overline{\pi_\nu^{\pm}(man)}\,
\pi_{-\nu}^{\pm}(man) =1$ for all $man \in MAN$, one verifies that 
$\widetilde\Pi_{\nu}^{\pm} (g) $, $g \in G$, is a unitary operator in
${\mathfrak h}_{\nu,0}^\pm$.  
\end{proof}

\goodbreak
\begin{remarks} 
\quad
\begin{itemize}
\item[$i.)$]
Note that in case $\frac{1}{2} \le \zeta$, the norm does not depend on $\nu$. 
\item[$ii.)$]
In Bargmann's classification \cite{Ba} of the unitary irreducible representations of $SO_0(1, 2)$, the 
\emph{principle series}\index{principle series} and the \emph{complementary series}\index{complementary series} 
are both denoted by $C_{\zeta^2}^0$. They are distinguished by the eigenvalue of $\zeta^2$ of the Casimir 
operator $C^2$, with $\zeta^2$ being larger or equal {\em or} smaller than $1/4$.  
\end{itemize}
\end{remarks}

Choosing $p_0 = 1$ in \eqref{udrei} and using the notation introduced in \eqref{aktq}, one finds 
(see Equ.~(4.41) and Equ.~(4.42) in \cite{BM})
	\begin{align}
		\label{udrei23} 
 			\bigl( \widetilde  u_{\nu}^{\pm} (\Lambda_2 (s))h \bigr)  (\alpha') 
		&= {\rm e}^{( - \frac{1}{2} - i\nu)  t_2 } h (  \alpha_2 ) 
	\nonumber \\  
			\bigl( \widetilde u_{\nu}^{\pm} (\Lambda_1 (t))h \bigr)   (\alpha')  
		&= {\rm e}^{( - \frac{1}{2} - i\nu)  t_1 }  
		h ( \alpha_1 ) 
	 \nonumber \\
			\bigl( \widetilde u_{\nu}^{\pm} (R_{0}(\alpha))h \bigr)   (\alpha')  
		&= h (\alpha + \alpha' )  \; , 
	\end{align}
with
	\begin{align*}
		t_2   & = \ln (\cosh s - \sinh s \sin \alpha')  \; , 
		\\
		 t_1   & = \ln (\cosh t - \sinh t \cos \alpha') \; , 		  
	\end{align*}
and
	\begin{align*}
		{\rm e}^{i \alpha_1} & 
		=  \tfrac{\cos \alpha' -i \sinh s + i \cosh s \sin \alpha'}{\cosh s -  \sinh s \sin \alpha'}  \; , 
		\\
		{\rm e}^{i \alpha_2} & 
		=  \tfrac{-\sinh t + \cosh t \sin \alpha' +i \sin \alpha' }{\cosh t  +  \sinh t \cos \alpha'}  \; . 
	\end{align*}

\begin{theorem}[Bargmann, \cite{Ba}]
\label{TH-irr}
The representations $\widetilde u_{\nu}^{\pm}$ given by \eqref{FG1} are irreducible. 
The representations for $\nu$ and $- \nu$, $\nu \in \mathbb{R}$, are unitarily equivalent
both for the principal and the complementary series\footnote{See, \emph{e.g.}, 
\cite[p.~104]{Pukanszky}.}. 
\end{theorem}

\begin{proof} 
Let us consider $C^\infty_0$ functions on the forward light cone. It follows that the generators $L_2$, $L_1$ and $K_0$ take 
the form (see~\cite[\S 6a]{Ba})
\label{bargmangeneratorpage}
	\begin{align}
		i L_2  &= \cos \alpha  \frac{\partial}{ \partial \alpha}
			+ \sin \alpha \, p_0  \frac{\partial}{ \partial p_0} \; , \nonumber \\  
		i L_1  &= \sin \alpha  \frac{\partial}{ \partial \alpha} 
			- \cos \alpha \,    p_0  \frac{\partial}{ \partial p_0} \; , \nonumber \\ 
		i K_0 &= -  \frac{\partial}{ \partial \alpha} \; .
	\label{qqdq}
	\end{align}
Note that $K_0^2=-  \frac{\partial^2}{ \partial \alpha^2}$ is a positive operator. 
The eigenfunctions of $K_0^2$ on the light cone for the eigen\-value~$k^2$ are of the form $h (p_0) e_k$ with 
	\[
		e_k= \frac{{\rm e}^{ik\alpha}}{\sqrt{2\pi}} \; ,  \qquad k \in \mathbb{Z}  \; .
	\]
The generator of the horospheric translations is $i (L_2 - K_0)$ and 
the Casimir operator is
	\begin{equation}
		\label{casimir}
		C^2 = - K_0^2 + L_1^2 + L_2^2   \; . 
	\end{equation}
The latter equals \cite[Eq.~(6.5)]{Ba}
\label{lightconecoordinateKGpage}
	\begin{equation} 
	\label{casimir2}
	C^2 = - S (S+1) = - \partial_{p_0} p_0^2  \partial_{p_0}
        \; , \qquad \text{with} \quad S= p_0  \partial_{p_0} \; . 
	\end{equation} 
It is positive, since 
	\begin{align}
	 	\langle g, C^2 g \rangle 
		&=-\int_0^\infty {\rm d} p_0 \int_0^{2 \pi} \frac{{\rm d} \alpha}{2 \pi} \; \overline{g (p_0, \alpha)}\, 
	 	\partial_{p_0} p_0^2 \partial_{p_0} g (p_0, \alpha)   \\
	 	&= \int_0^\infty {\rm d} p_0 \int_0^{2 \pi} \frac{{\rm d} \alpha}{2 \pi} \; 
	 	p_0^2 \, | \partial_{p_0} g (p_0, \alpha)|^2 \ge 0 \;  .  \quad \nonumber
	\end{align}
The eigenvalue equation $\zeta^2= -s(s+1)$ has the solutions
	\begin{equation} 
		\label{dd1} 
			s^\pm= -\frac{1}{2}  \mp i \nu \; , \quad \text{with} \quad \nu =  
			\begin{cases}
				i \sqrt{\frac{1}{4} -\zeta^2} & \text{if $ 0< \zeta < 1 /2$}  \, ,\\
				 \sqrt{\zeta^2 - \frac{1}{4} } & \text{if $ \zeta \ge 1 /2$} \, .
			\end{cases} 
	\end{equation} 
Eq.~\eqref{casimir2} implies \cite[Eq.~(6.6b)]{Ba} that
the generalised eigenfunctions for the eigenvalue $\zeta^2$ of~$C^2$
are homogenous functions of the form 
	\[
		(\alpha, p_0) \mapsto {p_0}^{-\frac{1}{2} - i \nu} f ( \alpha, 1) \; , 
	\]
in agreement with \eqref{FG1}. 
Thus, in the representation $\widetilde u_{\nu}^{\pm}$ the Casimir operator is a 
multiple of the identity with eigenvalue $s^+ = - \tfrac{1}{2} - i \nu$. 

Now let $A$ be  a bounded linear operator $A$ on $\widetilde {\mathfrak h}_\nu^\pm$, which 
commutes with all $\widetilde u_{\nu}^{\pm} (g)$, $g \in SO_0(1,2)$. It follows \cite[p.~608]{Ba} that  
	\begin{align}
		K_0  A \,  f_k & = A K_0 \,  f_k  \; , \qquad f_k \doteq {p_0}^{-\frac{1}{2} - i \nu} e_k  \; , 
		\nonumber
		 \\
		L_i  A \, f_k & = A L_i  \, f_k  \; , 
		\qquad \; \; i = 1, 2 \; , \quad  k \in \mathbb{Z}  \; . \label{key-irred}
	\end{align}
The first equation implies that $A \, f_k  = \alpha_{\nu,k} \cdot  f_k   $ for some $\alpha_{\nu,k} \in \mathbb{C}$. 
To explore the content of the second and third equation in \eqref{key-irred}, we introduce
the ladder operators $L_\pm = L_1 \pm L_2$. They satisfy
	\begin{align}
		L_+ \, f_k & = c_{k+1} \sqrt{ \zeta^2+ k (k+1)} \, f_{k+1} \; , \nonumber \\
		L_- \, f_k & = c_{k}^{-1} \sqrt{ \zeta^2+ k (k-1)} \, f_{k-1} \; ,    
	\label{ladder-op}
	\end{align}
with $| c_k | = 1$ some constants of absolute value $1$. Since $L_1 = \tfrac{1}{2} (L_+ + L_-)$ 
and $L_2 = \tfrac{i}{2} (L_- - L_+)$, we obtain from \eqref{ladder-op} a set of equations, which may be written in the form 
\cite[Equ.~(5.34)]{Ba}
	\[
		L_i f_k = \sum_{k'} h_{k,k'} {f_{k'}} \; , 
	\]
where $h_{k,k'} = \overline{h_{k',k}} \, $, and where $h_{k,k'}= 0$ if $| k - k' | >1 $. 
We therefore obtain from the second and third equation in \eqref{key-irred}  equations of the form
	\[
		( \alpha_{\nu,k} - \alpha_{\nu,k'} ) h_{k, k'} = 0 \qquad \forall k, k'  \in \mathbb{Z} \; . 
	\]
A brief inspection shows 
that all $\alpha_{\nu,k}$ have to be equal to each other (for $\nu$ fixed), \emph{i.e.}, 
that $A = \alpha_{\nu, 0} \cdot \mathbb{1}$. 
\end{proof}

\subsection{Representations of $SO_0(1,2)$ on the forward light cone}
For $0< \zeta <1/2$, the eigenfunctions 
of the Casimir operator \eqref{casimir} are \emph{not} in $L^2 (\partial V^+, 
\frac{{\rm d} \alpha}{2 \pi} {\rm d} p_0 )$, 
as their decay in the variable $p_0$ is not fast enough to ensure the existence of the integral. 
Thus the unitary irreducible representations in the \emph{complementary series} (corresponding to 
$0< \zeta < 1/2$) do not appear, if one decomposes the reducible representation 
on $L^2(\partial V_+, \frac{{\rm d} \alpha}{2 \pi} {\rm d} p_0 )$ given by the pull-back:
 
\begin{theorem}[Spectral theorem] 
\label{spectheo}
As an operator on $L^2(\partial V_+, \frac{{\rm d} \alpha}{2 \pi} {\rm d} p_0 )$ with
domain~${\mathcal D}_\mathbb{R} (\partial V^+)$, the Casimir operator $C^2$ given in \eqref{casimir2} is essentially 
self-adjoint and positive. The positive square root of its self-adjoint extension, denoted by~$C$, has 
spectrum ${\rm Sp} (C) = [1/2,\infty)$. The corresponding spectral 
decomposition is
	\[
L^2 \bigl(\partial V_+, \tfrac{{\rm d} \alpha}{2 \pi} {\rm d} p_0 \bigr) 
=  \int_{\frac{1}{2}}^\infty {\rm d} \zeta \;  {\mathfrak H}_\zeta \; , \qquad 
         {\mathfrak H}_\zeta \cong L^2 \bigl(S^1, \tfrac{{\rm d} \alpha}{2 \pi} \bigr)\otimes
  \mathbb{C}^2 \; . 
	\]
\end{theorem}

\begin{proof} We exploit the properties of the Mellin transform \cite{Oberh2, Oberh}:  for $ g \in L^2(0,\infty)$ 
one finds
	\begin{equation}
		\label{mellin}
		g(p_0) = \frac{1}{2\pi} \int_\mathbb{R} {\rm d}\nu \;
		 p_0^{-\frac{1}{2}-i\nu} \int_0^\infty dp_0' \; {p_0'}^{-\frac{1}{2}+i\nu} g(p_0')  \; .		
	\end{equation}
The integral w.r.t.~${\rm d} \nu$ is over the whole real axis.   This implies:
\begin{itemize}
\item[$ i.)$] The identity operator on $L^2(\partial V_+, \frac{{\rm d} \alpha}{2 \pi} {\rm d} p_0 )$ is 
	\[
		\mathbb{1} =  \int_\mathbb{R} {\rm d} \nu  \biggl( \sum_j \;  | p_0^{- \frac{1}{2}-i\nu} h_j \rangle 
		\langle  {p_0'}^{-\frac{1}{2}-i\nu} h_j | 
		\biggr) \; , 
	\]
with $\{ h_j \in L^2(S^1, \frac{{\rm d} \alpha}{2 \pi})\mid  j \in \mathbb{N} \}$ an orthonormal 
basis in $L^2(S^1, \frac{{\rm d} \alpha}{2 \pi})$. Thus
$L^2(\partial V_+, \frac{{\rm d} \alpha}{2 \pi} {\rm d} p_0 )$ is the direct integral over $\nu \in\mathbb{R}$ of
the Hilbert spaces~$\widetilde {\mathfrak h}_\nu  $ consisting of functions
	\[
	  (p_0,\alpha)  \mapsto     p_0^{-\frac{1}{2} - i\nu}  h(\alpha) \; .
	 \]
The scalar product in $\widetilde {\mathfrak h}_{\nu}$ is just the scalar product in $L^2(S^1, \frac{{\rm d} \alpha}{2 \pi})$;
\item[$ ii.)$] the spectrum of $C$ in $L^2(\partial V_+ , \frac{{\rm d} \alpha}{2 \pi} {\rm d} p_0 )$ is $[\frac{1}{2},\infty)$; 
\item[$ iii.)$] for $\zeta^2= \frac{1}{4} + \nu^2$,
	\[
		{\mathfrak H}_\zeta = 
		\widetilde {\mathfrak h}_\nu   \oplus \widetilde {\mathfrak h}_{-\nu}   \; ; 
	\]
\emph{i.e.},  homogeneous functions of degree $s^+ $ and $s^- $ (see \eqref{dd1}) both  appear.
\end{itemize}
\end{proof}

\subsection{Intertwiners}

\begin{proposition}
\label{Prop:2.1.0}
Consider the representations described in  \eqref{FG1}.
It follows that the map 	
	\begin{equation}
	\label{intertwiner-1}
		(A_\nu h)(\alpha) \doteq \int_{\Gamma_0}   
		\frac{{\rm d} \alpha' }{2 \pi} \; \varrho_\nu (\alpha-\alpha') h(\alpha') \; , 
		\qquad \alpha \in [0, 2 \pi) \; , 
	\end{equation}
with\footnote{The Euler function\index{Euler function} $B(x,y)=\Gamma (x)\Gamma (y) / \Gamma (x+y)$ appears here.} 
	\begin{equation} 
		\label{dd2} 
		\varrho_\nu (\alpha) 
		\doteq \frac{\Gamma (\frac{1}{2}  - i \nu  )}{\Gamma ( \frac{1}{2})
		\Gamma ( - i \nu )} \;  \left(\sin^2 \tfrac{\alpha}{2} \right)^{-\frac{1}{2} - i \nu} 	\pi  \; , 
	\end{equation} 
defines an operator~ $A_\nu$, which {\em intertwines}\index{intertwiner} $\widetilde u_{-\nu}^{\pm}$ and $\widetilde u_{\nu}^{\pm}$, \emph{i.e.},  
	\begin{equation}
	\label{intertwiner}
 		A_\nu \widetilde u_{-\nu}^{\pm} (g)  = \widetilde u_{\nu}^{\pm} (g)  A_\nu   \qquad \forall g \in G \; . 
	\end{equation}
\end{proposition}
\color{black}

\goodbreak
\begin{remarks}\quad
\begin{itemize}
\item [$i.)$] 
The integral kernels appearing in \eqref{intertwiner-1} were first derived by Barg\-mann~\cite{Ba}. In the literature they are frequently 
written in the following alternative form:
	\begin{equation} 
	\label{Barg-factor}
		\varrho_\nu (\alpha) 
		=	\frac{\Gamma (\frac{1}{2}  - i \nu  )}{\Gamma ( \frac{1}{2})
		\Gamma (  - i \nu  )} \;  \left(\tfrac{1- \cos \alpha}{2} \right)^{-\frac{1}{2} - i \nu} \pi \; . 
	\end{equation} 
\item[$ii.)$] In case $\nu =  \pm i \sqrt{\frac{1}{4} - \zeta^2}$ with $0 < \zeta < \frac{1}{2}$, 
the sesquilinear form  
	\[
		h, h' \mapsto \int_{\Gamma_0} {\rm d} \alpha \; \overline{h(\alpha)} (A_\nu h')(\alpha)
	\] 
is positively definite~\cite{Ner}. 
This implies 
	\[
		\int_{\Gamma_0} {\rm d} \alpha \; \overline{ h( \alpha ) }  (A_\nu h')(\alpha)   
		 = \int_{\Gamma_0} {\rm d} \alpha \;  \overline{ A_\nu h (  \alpha ) }  h' ( \alpha  ) \; ,  
	\]
and, consequently, \eqref{intertwiner}
defines a positive operator\footnote{The operator $A_\nu$  {\em intertwines} the pullback action of 
$SO_0(1,2)$ on homogeneous functions of degree 
$-\frac{1}{2} + \sqrt{\frac{1}{4}-\mu^2}$ and $-\frac{1}{2} - \sqrt{\frac{1}{4}-\mu^2}$, respectively.}.
\item [$iii.)$] In case $\nu $ is real, 
we have $A_\nu^*  = A_{-\nu}$ \cite[Lemma 2.1]{Sally}. In fact, $A_\nu$ is unitary as  $A_\nu^* A_\nu = \mathbb{1}$; 
see Remark \ref{A-unitary} below. 
\item [$iv.)$] 
The bilinear form-valued function $\nu \to ( \, . \,  , A_\nu \, . \,  )_{L^2({\tt S}^1, {\rm d} \alpha)}$ is 
meromorphic in $\mathbb{C}$. The poles of this function are the points 
$i \nu = 0, \frac{1}{2}, 1, \frac{3}{2}, \ldots$. 
\item [$v.)$] 
The integral \eqref{intertwiner-1} is convergent if $i \nu < 0$ \cite[p.~605]{Ner}. In this case
	\begin{align*}		
		\frac{  \Gamma (\frac{1}{2})
		\Gamma (  - i \nu  )}{ \Gamma (\frac{1}{2}   - i \nu )} & 
		( {\rm e}^{i n \psi}  , A_\nu {\rm e}^{i m \psi} )_{L^2({\tt S}^1, {\rm d} \psi)} = 
		\nonumber \\
		& = 
		\frac{2^{ i \nu} }{ 4 \, B( 1/2 - i \nu - n , 1/2 - i \nu + n ) } \; \delta_{n,m} \;  . 
	\end{align*}
The beta function is 
	\[
		\qquad B( 1/2 - i \nu - n , 1/2 - i \nu + n ) 
		= \frac{\Gamma(1/2 - i \nu - n)\Gamma(1/2 - i \nu + n)}{\Gamma(1 - 2i \nu)} \; . 
	\]
Using $\Gamma(z+1) = z \Gamma (z)$ and
the Stirling formula for the $\Gamma$-functions implies that the pre-factor  on the r.h.s.~diverges as  
	\[
		|n|^{2 i \nu} \left( 1 + O \left( \frac{1}{n}\right)\right) \; , \qquad n \to \infty \; . 
	\]
Hence the form $( \, . \,  , A_\nu \, . \,  )_{L^2({\tt S}^1, {\rm d} \psi)}$ is well-defined on the 
Sobolev space  $\mathbb{H}^{i \nu} (S^1)$; see Definition \ref{sob-circ}.
\end{itemize}
\end{remarks}

\begin{proof}
\color{black}
Inspecting \eqref{udrei23} we find that 
	\begin{equation} 
		\widetilde u_{-\nu}^{\pm} (R_{0}(\alpha)) = \widetilde u_{\nu}^{\pm} (R_{0}(\alpha)) \; , \qquad \alpha \in [0, 2 \pi) \; . 
	\label{remark-identity}
	\end{equation}
Moreover, 
	\begin{align*}
	& \int_{\Gamma_0} \frac{ {\rm d} \beta' }{2 \pi} \; \varrho_\nu (\beta-\beta') \bigl( \widetilde u_{-\nu}^{\pm} (R_{0}(\alpha)) h\bigr)(\beta') 
	\nonumber \\
	& \qquad = \int_{\Gamma_0} \frac{ {\rm d} \beta' }{2 \pi} \;  \varrho_\nu (\beta-\beta') h (\beta' + \alpha) \nonumber \\
	& \qquad = \widetilde u_{\nu}^{\pm} (R_{0}(\alpha))  
	\Bigl( \int_{\Gamma_0} \frac{ {\rm d} \beta' }{2 \pi} \;  \varrho_\nu (\beta - \beta') h (\beta') \Bigr) \; . 
	\end{align*}
It remains to be shown that
	\[
 		A_\nu \widetilde u_{-\nu}^{\pm} (\Lambda_2 (s))  = \widetilde u_{\nu}^{\pm} (\Lambda_2 (s))  A_\nu   \qquad \forall s \in \mathbb{R} \; . 
	\]
Compute
	\begin{align}
	\label{I-1}
	& \int_{\Gamma_0} \frac{ {\rm d} \beta' }{2 \pi} \;  \varrho_\nu (\beta-\beta') 
	\bigl( \widetilde u_{-\nu}^{\pm} (\Lambda_{2}(s)) h\bigr)(\beta') \nonumber	\\
	& \quad = \int_{\Gamma_0} \frac{ {\rm d} \beta' }{2 \pi} \;  \varrho_\nu (\beta-\beta') \, 
	{\rm e}^{s^- t } h \left(\alpha \right)  \; ,  
	\end{align}
with $\alpha= \alpha(s, \beta')$ given by 
	\[
	R_0 (\alpha) P^k \Lambda_1(t) D(q) \doteq \Lambda_{2}(s)^{-1} R_0(\beta') \; . 
	\]
Note that $1+s^- = -s^+$. Thus
	$\frac{ {\rm d}  \alpha(s, \beta') }{ {\rm d} \beta'} = (\cosh s - \sinh s \sin \beta')^{- 1 } $, 
	\[
	\sin \beta'  = \tfrac{\sinh s + \cosh s \sin \alpha}{\cosh s + \sinh s \sin \alpha} \; , 
	\qquad 
	\cos \beta'  = \tfrac{\cos \alpha}{\cosh s + \sinh s \sin \alpha} \; , 
	\]
and
	$
	\cosh s - \sinh s \sin \beta' = \bigl(\cosh s  + \sinh s \sin \alpha \bigr)^{-1} $.
This allows us to rewrite \eqref{I-1} as
	\begin{align}
	\label{I-1-1}
	& \int_{\Gamma_0} \frac{ {\rm d} \beta' }{2 \pi} \;  \varrho_\nu (\beta-\beta') \bigl( \widetilde u_{-\nu}^{\pm} (\Lambda_{2}(s)) h\bigr)(\beta')  \\
	& \qquad = \int_{\Gamma_0} \frac{ {\rm d} \alpha }{2 \pi} \;  \varrho_\nu \bigl(\beta-\beta' (s, \alpha) \bigr) 
	\bigl(\cosh s  + \sinh s \sin \alpha \bigr)^{s^+}  h ( \alpha ) \; . \nonumber
	\end{align}
The kernel $\varrho_\nu  \bigl(\beta-\beta' (s, \alpha) \bigr) $ can be rewritten using the formula
	\begin{align}
	\label{I-7}
	& \bigl( 1 - \cos \bigl( \beta - \beta' (s, \alpha) \bigr)\bigr)^{s^+}
	 \\
	& \qquad \qquad = \Bigl( 1 - \cos  \beta  \cos \bigl(\beta' (s, \alpha)\bigr) 
	- \sin  \beta  \sin \bigl(\beta' (s, \alpha)\bigr) \Bigr)^{s^+} \; . 	
	 \nonumber \\
	& \qquad \qquad = \Bigl(
	\tfrac{ \cosh s  + \sinh s \sin \alpha - \cos \beta \cos \alpha - \sin \beta \sinh s - \sin \beta \cosh s \sin \alpha}
	{\cosh s + \sinh s \sin \alpha}  \Bigr)^{s^+}. 	
	\nonumber
	\end{align}
For \eqref{I-1} this yields
	\begin{align}
	\label{I-9}
	& \int_{\Gamma_0} \frac{ {\rm d} \beta' }{2 \pi} \;   \varrho_\nu (\beta-\beta') \bigl( \widetilde u_{-\nu}^{\pm} (\Lambda_{2}(s)) h\bigr)(\beta')   \\
	& \qquad = \frac{\Gamma (\frac{1}{2}+ \nu) \; \pi }{\Gamma ( \frac{1}{2})
		\Gamma (\nu)}   \int_{\Gamma_0} \frac{ {\rm d} \alpha }{2 \pi} \;  
	\Bigl( \cosh s  + \sinh s \sin \alpha - \cos \beta \cos \alpha 
	\nonumber \\
	& \qquad \qquad \qquad \qquad  \qquad \qquad 
	- \sin \beta \sinh s 
	- \sin \beta \cosh s \sin \alpha \Bigr)^{s^+}  h^-  \bigl( \alpha \bigr) \; . \nonumber
	\end{align}
On the other hand, using \eqref{udrei}
	\begin{align}
	\label{I-9a}
	& \widetilde u_{\nu}^{\pm} (\Lambda_{2}(s)) \Bigl( \int_{\Gamma_0} \frac{ {\rm d} \beta' }{2 \pi} \; 
	 \varrho_\nu (\beta -\beta')  h (\beta') \Bigr) \nonumber \\
	& \qquad = \bigl(\cosh s - \sinh s \sin \beta \bigr)^{s^+}  
		\nonumber \\
		& \qquad \qquad \qquad \qquad  \times	 \int_{\Gamma_0} \frac{ {\rm d} \alpha }{2 \pi} \;  \varrho_\nu 
	\bigl( \arccos \bigl( \tfrac{\cos \beta}{\cosh s -  \sinh s \sin \beta} \bigr) -\alpha \bigr)  
		h  \bigl( \alpha \bigr)
 		\nonumber \\
	& \qquad = \bigl(\cosh s - \sinh s \sin \beta \bigr)^{s^+} 
	 \int_{\Gamma_0} \frac{ {\rm d} \alpha }{2 \pi} \;  \varrho_\nu 
	\bigl( \beta' (s, \alpha) -\alpha \bigr)  
		h  \bigl( \alpha \bigr) \; . 
	\end{align}
Next we compute $\varrho_\nu \bigl(\beta' (s, \alpha)-\alpha) \bigr)$:  
	\begin{align*}
	\bigl( 1 - \cos \bigl( \beta' (s, \alpha)  - \alpha \bigr) \bigr)^{s^+}
	& = \Bigl( 1 - \cos  \bigl(\beta' (s, \alpha)\bigr)   \cos \alpha 
					- \sin  \bigl(\beta' (s, \alpha)\bigr)  \sin \alpha \Bigr)^{s^+}   	
	 \nonumber \\
	& = \Bigl(
	\tfrac{ \cosh s  - \sinh s \sin \beta - \cos \beta \cos \alpha + \sin \alpha \sinh s - \sin \beta \cosh s \sin \alpha}
	{\cosh s - \sinh s \sin \beta}  \Bigr)^{s^+}. 	
	\nonumber
	\end{align*}
Inserting this result into \eqref{I-9a} shows that 
\begin{align}
	& \int_{\Gamma_0} \frac{ {\rm d} \beta' }{2 \pi} \;  \varrho_\nu (\beta-\beta') \bigl( \widetilde u_{-\nu}^{\pm} (\Lambda_{2}(s)) h\bigr)(\beta') 
	\nonumber \\
	& \qquad = \widetilde u_{\nu}^{\pm} (\Lambda_{2}(s)) 
	\Bigl( \int_{\Gamma_0} \frac{ {\rm d} \beta' }{2 \pi} \;  \varrho_\nu (\beta - \beta') h (\beta') \Bigr)  \; . 
\end{align} 
Since $ R_0(\alpha)$ and $\Lambda_2(s)$ generate $SO_0(1,2)$, this verifies \eqref{intertwiner}. 
\end{proof}

\begin{remark}
\label{A-unitary}
Clearly  \eqref{remark-identity} implies
	\[ 
		[ A_\nu ,  \widetilde u_{\nu}^{\pm} (R_{0}(\alpha))] = 0 \; , \qquad \alpha \in [0, 2 \pi) \; .  
	\]
Therefore $A_\nu$ has diagonal form in the spectral representation of the generator 
of the rotations $\alpha \mapsto R_{0}(\alpha)$. In fact \cite[p.~619]{Ba},
	\[ 
	 	\frac{\Gamma (\frac{1}{2} - i \nu)} { \Gamma(\frac{1}{2}) \Gamma(- i \nu)}
		\left( \sin^2 \frac{ \alpha }{2} \right)^{- \frac{1}{2} - i \nu} \pi  = 1+ 
		\sum_{k\in \mathbb{Z} \setminus \{0\} }  \prod_{j=1}^{| k |}
		\frac{  ( j  + \frac{1}{2} + i \nu)}{  (  j + \frac{1}{2} - i \nu)} \; {\rm e}^{i k \alpha  } \; , 
	\]
and, consequently, the Fourier coefficients\footnote{These coefficients refer to the 
eigenfunctions $ \frac{{\rm e}^{i k \alpha  }}{\sqrt{2 \pi} } $.}
of $\varrho_\nu$ are 
	\[ 
		\widetilde \varrho_\nu (k) 
		= \sqrt{2 \pi} \; \frac{\Gamma ( | k | + \frac{1}{2} + i \nu )}{\Gamma ( | k | + \frac{1}{2} - i \nu)} 
		\frac{\Gamma (\frac{1}{2} - i \nu)} { \Gamma(\frac{1}{2} + i \nu)} \; .
	\]
Hence, for $\nu \in \mathbb{R}$, 
	\begin{align*}
		\int_{\Gamma_0} \frac{ {\rm d} \alpha' }{2 \pi} \, \overline{ \rho_\nu(\alpha-\alpha') } \, \rho_\nu(\alpha'-\alpha'') 
		&= \int_{\Gamma_0} \frac{ {\rm d} \alpha' }{2 \pi} \, \overline{
		\sum_k \widetilde \varrho_\nu (k)  \frac{ {\rm e}^{- i k (\alpha -\alpha')}}{\sqrt{2 \pi}}
		} \, \sum_j \widetilde \varrho_\nu (j)  \frac{ {\rm e}^{- i j (\alpha' -\alpha'')}}{\sqrt{2 \pi}} 
		\\
		& = \sum_k \overline{
		\frac{ \widetilde \varrho_\nu (k)  }{\sqrt{2 \pi} }}
		\sum_j \frac{ \widetilde \varrho_\nu (j)  }{\sqrt{2 \pi}}
		\\
		& 
		\qquad \qquad \times 
		\int_{\Gamma_0} \frac{ {\rm d} \alpha' }{2 \pi} {\rm e}^{- i k (\alpha -\alpha')} {\rm e}^{ i j (\alpha' -\alpha'')} 
		\\
		& =  	\sum_k {\rm e}^{i k (\alpha - \alpha'')} 
		\\
		& =  	2\pi \;  \delta(\alpha-\alpha'')
		\; .   
	\end{align*}
We have used that $ |\widetilde \varrho_\nu (k)|^2 = 2\pi $ for all $k \in \mathbb{Z}$. 
Note that $2\pi \;  \delta(\alpha-\alpha'')$ is the kernel of the identity operator with respect to the 
measure $\tfrac{{\rm d} \alpha''}{2\pi}$. 
\end{remark}

\subsection{The time reflection\index{time reflection}}
\label{sec:time reflection} 

Our next aim is to extend the unitary irreducible representations of
the two fold covering group of $SO_0(1,2)$ discussed so 
far to (anti-)unitary representations of $O(1,2)$. 

We start with the induced representation $ \Pi_\nu^{\pm}$
\eqref{eqIndRepSO12}, and consider first the case  $\nu\in
\mathbb{R}$, \emph{i.e.}, $\zeta\geq 1/2$. Let $\Pi_{\nu,0}^\pm(T)$ be the anti-linear map from
$\mathfrak{h}_{-\nu,0}^\pm$ to $\mathfrak{h}_{\nu,0}^\pm$ 
defined by 
	\begin{equation} 
		\label{eqUT0}
		(\Pi_{\nu,0}^\pm(T) f)(g) \doteq \overline{f(Pg)} \; ,  \quad
 		 f\in  \mathfrak{h}_{\nu,0}^\pm   \; ,  
	\end{equation}
where $P$ is the space-reflection, $P=R_0(\pi)\in SO_0(1,2)$. 
Since $\lambda_P(g MAN)=1$, this is an isometric
map. Now pick am intertwiner 
	\[
		A_\nu \colon \mathfrak{h}_{\nu,0}^\pm  \to \mathfrak{h}_{-\nu,0}^\pm  
	\]
between the representations $\Pi_\nu^\pm$ and $\Pi_{-\nu}^\pm$ and
define $\Pi_{\nu}^{\pm}(T) \doteq \Pi_{\nu,0}^{\pm}(T)\circ A_\nu$: 
	\begin{equation} 
		\label{eqInRepT}
		\big(\Pi_{\nu}^\pm(T) f\big)(g) \doteq \overline{(A_\nu f)(Pg)} \; . 
	\end{equation}
Then one has 
	\[
		\Pi_{\nu}^\pm(T)\, \Pi_{\nu}^\pm(g)\, \Pi_{\nu}^\pm(T)^{-1}=
		\Pi_{\nu,0}^{\pm}(T) \, \Pi_{-\nu}^\pm(g)\, \Pi_{\nu,0}^{\pm}(T)^{-1}  
	\]
and 
	\[
		\big(\Pi_{\nu}^\pm(T)\, \Pi_{\nu}^\pm(g)\, \Pi_{\nu}^\pm(T)^{-1} \; f \big)(g')
		= \lambda_g(Pg' \,H)^{\frac{1}{2}}\, f(Pg^{-1}P\,g') \; ,
	\]
while  on the other hand
	\[
		\big(\Pi_{\nu}^\pm(TgT^{-1}) \; f \big)(g')= \lambda_{PgP}(g' \,H)^{\frac{1}{2}}\, f((PgP)^{-1}\,g') \; , 
	\]
where it has been used that the adjoint action of $T$ on $SO(_0(1,2)$ 
coincides with that of the space-reflection $P$, $TgT=PgP$.  Since
$\lambda_{PgP}(g' \,H)= \lambda_g(Pg' \,H)$, this proves that 
	\[
		\Pi_{\nu}^\pm(T)\, \Pi_{\nu}^\pm(g)\, \Pi_{\nu}^\pm(T)^{-1}=\Pi_{\nu}^\pm(TgT^{-1}) \;  .
	\]
Thus, $\Pi_{\nu}^\pm(T)$ is a representer of $T$ which, in addition, 
is easily seen to be anti-unitary. 
 
We wish to find the equivalent representer in the representation space 
$C_0(G/H)$. The intertwiner $A_\nu$ corresponds uniquely to an
operator $\widetilde A_\nu$ in $C_0(G/H)$ by the equivalence~\eqref{eqFeqFGH}, 
	\[
		\widetilde{A}_\nu \tilde{f} \doteq \widetilde{A_\nu f} \; , 
	\]
which intertwines the representations $\widetilde u_\nu^\pm$ and  $\widetilde u_{-\nu}^\pm$. 
Now this equivalence translates $\Pi_{\nu}^\pm(T)$ into the
anti-unitary operator $\widetilde u_{\nu}^\pm(T)$ in $C_0(SO(2))$ given by  
	\begin{align*} 
\big(\widetilde u_{\nu}^\pm(T)\,\tilde f\big)(R_0(\alpha)) 
& \doteq   \big(\widetilde{\Pi_{\nu}^\pm(T)\,f}\big)(R_0(\alpha))
=  \big({\Pi_{\nu}^\pm(T)\,f}\big)(R_0(\alpha)) = 
\overline{\big(A_\nu\,f\big)(P R_0(\alpha))}\\
&=\overline{\big(\widetilde{A_\nu\,f}\big)(P R_0(\alpha))}
\; = \; \overline{\big(\widetilde A_\nu\,\tilde f\big)(P R_0(\alpha))}. 
\end{align*}
In the second and fourth equation we have used the fact that
$\Xi(R)=R\in G$ for any rotation $R\in SO(2)\cong G/MAN$ and that
$PR_0(\alpha)\equiv R_0(\alpha+\pi)$ is a rotation.  In short, $\widetilde u_{\nu}^\pm(T)$ acts on $C_0(SO(2))$ as 
	\begin{equation} 
		\big(\widetilde u_{\nu}^\pm(T)\, h \big)(R_0(\alpha)) = 
			\overline{\big(\widetilde A_\nu\,h \big)(P R_0(\alpha))} \; .
	\end{equation}
Note that $\widetilde u_{\nu}^\pm(T)^2=A_\nu^* A_\nu = \mathbb{1}$.

In the case  $\nu\in i\mathbb{R}$,  \emph{i.e.}, $0<\zeta< 1/2$, the
anti-linear map $\Pi_{\nu,0}^\pm(T)$ defined above leaves
$\mathfrak{h}_{\nu,0}^\pm$ invariant, and we take this operator to be
the representer of $T$ in $\mathfrak{h}_{\nu,0}^\pm$. The proof of the 
representation property goes as above.  
We then define $\widetilde u_{\nu}^\pm(T)$ as the equivalent representer
in the  representation space $C_0( \Gamma_0 )$, namely, 
	\[
		\big(\widetilde u_{\nu}^\pm(T)\,  h  \big)(R_0(\alpha)) \doteq 
		 \overline{h(P R_0(\alpha))} \; . 
	\]
Anti-unitarity can be seen as follows: 
	\begin{align*}
		\|\widetilde u_{\nu}^\pm(T)h\|_\nu &= \big\langle \widetilde u_{\nu}^\pm(T)h \, , \, A_\nu\,
		\widetilde u_{\nu}^\pm(T)h \big\rangle_{ L^2(\Gamma_0, {\rm d} \mu_{\Gamma_0}) } \\
		&=  \int_{SO(2)} \frac{d\alpha}{2\pi} \; h(PR_0(\alpha)) \, 
			\int_{\Gamma_0} \frac{d\alpha'}{2\pi} \; 
			\rho_\nu(\alpha-\alpha')\overline{h(PR_0(\alpha'))} \\
		&=  \int_{ \Gamma_0 } \frac{d\alpha'}{2\pi} \; \overline{h(R_0(\alpha'))}) \, 
			\int_{\Gamma_0} \frac{d\alpha}{2\pi} \; 
			\rho_\nu(\alpha-\alpha') h(R_0(\alpha)) \\
		&=  \big\langle h \, , \, A_\nu\, h \big\rangle_{ L^2(\Gamma_0, {\rm d} \mu_{\Gamma_0}) } = \|h\|_\nu. 
\end{align*}
In the fourth equation we have used the symmetry
$\rho_\nu(\alpha)=\rho_\nu(-\alpha)$. 

Note that the preceding discussion also shows that in both cases,
$\zeta<1/2$ and $\zeta \geq1/2$, the representer of the
space-reflection $P$ is given by
	\[ 
		\big(\widetilde u_{\nu}^{\pm}(P)h\big)(R_0(\alpha)) =
		h(PR_0(\alpha))\equiv h(R_0(\alpha-\pi)).
	\] 
Summarising, our discussion leads to the following definition.

\begin{definition} For $h(\alpha) \in \widetilde {\mathfrak h}_\nu$
define an antilinear operator by setting 
	\begin{equation}
		\label{tilde-time-reflection}
			\big(\widetilde u_{\nu}^{\pm} (T)h\big)(\alpha) \doteq
				\begin{cases}  
                       			\overline{(A_\nu h) (\alpha-\pi) }
					& \text{if \ $1/2 \le \zeta $\, , }\\
					\overline{h (\alpha-\pi) }
					& \text{if \ $0<  \zeta < 1/2$\, . }
				\end{cases}
	\end{equation}
\end{definition} 

We have shown: 

\begin{proposition}
The anti-unitary operator $\widetilde  u_{\nu}^{\pm} (T)$ is an anti-unitary representer of the time-reflection $T$
on $\widetilde{\mathfrak h}_\nu^{\pm}$. Together with
  $\widetilde u_{\nu}^{\pm} (P_2)$ it extends the representation $\widetilde u_{\nu}^{\pm}$ from $ SO_0(1,2)$ to $O(1,2)$. 
\end{proposition}

\section{Unitary representations on the mass shell}
\label{UIRm}

The Hannabus decomposition implies that a function  $f \in {\mathfrak h}_0^{\pm}$ is determined by 
$f (k' MAN) $ with $k' \in K'= \left\{ \Lambda_2 (s) ,  \Lambda_2 (s)P \mid s \in \mathbb{R} \right\}$. 
We have seen in Subsection \ref{massshell-sub} that 
we can identify the cosets\footnote{The cosets $ k' MN$, $ k' \in K' $, can be identified with 
points in $\partial V^+$.}
of the form
	\[
		k' MAN\; , \qquad k' \in K' \; , 
	\]
with points in the two hyperbolas
	\begin{align} 
		\label{eqGamma1}
		\Gamma_1 & \doteq  \Gamma_+ \; \dot \cup \; \Gamma_-  \nonumber \\
		& = \{ p_+ (s)  \in \partial V^+ \mid s \in \mathbb{R} \} \; \dot \cup \;  \{ p_- (s) \in \partial V^+ \mid s \in \mathbb{R} \} \; , 
	\end{align} 
where 
	\begin{equation}
		\label{mass-hyperboloids}
		p_\pm(s) \doteq \Lambda_{2} (s)  p_\pm(0)  
		= 
		\begin{pmatrix} m  \cosh s \\
		\pm m  \sinh s \\
		\pm m
		\end{pmatrix}   \; , \qquad s \in \mathbb{R} \; .  
	\end{equation}
Note that $p_\pm(0) = \left( \begin{smallmatrix} m \\ 0 \\ \pm m \end{smallmatrix}\right)$.
By construction, the boosts $s \mapsto \Lambda_{2}(s)$, $s \in \mathbb{R}$, leave the curves  $\Gamma_\pm$ invariant.
The reflections $P_1, P_2 $ and $P$ act on $\Gamma_1$:
	\[
	P_1 \Lambda_{2} (s)  p_\pm(0) 
	= \begin{pmatrix} m \cosh s \\
		m \sinh s\\
		\mp m
		\end{pmatrix} \; , \qquad 
	P_2 \Lambda_{2} (s)  p_\pm(0) 
		= \begin{pmatrix} m \cosh (-s) \\
		m \sinh (-s)\\
		\pm m
		\end{pmatrix} \; , 
	\]
and, consequently, 
	$
	P \Lambda_{2} (s)  p_\pm(0) 
		= \left( \begin{smallmatrix} m \cosh (-s) \\
		\pm m \sinh (-s)\\
		\mp m
		\end{smallmatrix} \right)$. 
		
\begin{remark} Given a homogeneous function on the forward light cone, one can define 
a function on $\Gamma_1$ by restriction. On the contrary (see \cite[Equ.~4.44]{BM}), given a pair of functions 
$p_1 \mapsto (h_+ (p_1), h_-(p_1))$ on $\Gamma_+ \; \dot \cup \; \Gamma_-$, 
	\[
	 	f(p)= \begin{cases}
				\left(\tfrac{p_2}{m}\right)^s h_+ \left(\tfrac{m p_1}{p_2}\right)  & \text{if $p_2 >0$, }\\ 
				\left(\tfrac{p_2}{m}\right)^s h_- \left(\tfrac{m p_1}{p_2}\right)   & \text{if $p_2 <0$, }
			\end{cases}
	\]
defines a homogeneous function of degree $s$ on the light cone $\partial V^+$.
\end{remark}

The restriction of the quasi-invariant 
measure  $\tfrac{{\rm d} \alpha}{2 \pi} {\rm d} p_0$ to
$\Gamma_\pm$ is~$\tfrac{{\rm d} s}{2}$. It follows that $L^2(\Gamma_1,{\rm d} \mu_{\Gamma_1})$ consists of 
two copies of $L^2(\mathbb{R},\tfrac{{\rm d} s}{2})$:
	\begin{equation} 
		L^2(\Gamma_1,{\rm d} \mu_{\Gamma_1})
		\cong L^2(\mathbb{R},\tfrac{{\rm d} s}{2})\oplus L^2(\mathbb{R},  - \tfrac{{\rm d} s}{2} ) \; . \label{eqUBM}
	\end{equation} 
Note that the two disjoint parts of $\Gamma_1$ form a closed curve enclosing the origin; thus the minus sign in the second component.

\begin{lemma}
The generator of the boost $s \mapsto \widetilde u_{\nu}^{\pm} ( \Lambda_2(s)) $ on $L^2(\Gamma_1,{\rm d} \mu_{\Gamma_1})$ is 
\begin{equation}
\label{boostgenerator}
 - i 
\bigl( \tfrac{\partial}{\partial s} \oplus \tfrac{\partial}{\partial s} \bigr) \; .
\end{equation}
Its spectrum is absolutely continuous on the whole real line. 
\end{lemma}

\begin{theorem} The induced representation \eqref{indrep} is given by
	\begin{align}
		\left( \widetilde u_{\nu}^{\pm} (g) h_{+} \right) (s')  
		& = 
		\chi_{\frac{i}{2}-\nu} \bigl( \Lambda_1(t)\bigr) h_{(-)^j} (s) \nonumber \\ 
		& = {\rm e}^{(-\frac{1}{2} - i \nu) t} h_{(-)^j} (s) \; ,  
	\label{GHK}
	\end{align}
where\footnote{In particular,  $\Lambda_2(s)^{-1}  \Lambda_2(s') = \Lambda_2(s'-s)$.}
	\begin{equation}
	\label{sktq}
		\Lambda_2(s)P^j \Lambda_1(t) D(q) \doteq g^{-1} \Lambda_2(s')\; , 
	\end{equation}
with $s, j, t$ and $q \in \mathbb{R}$. 
\end{theorem}

\begin{remark}
Recalling \eqref{distance-to-horosphere}, we find
	\begin{equation} 
	\label{h-limit}
		t = d ( x,  P_{  t } ) \;  \qquad \forall \, x = 
		\left( \begin{smallmatrix}
		\tfrac{q^2}{2r}  \\
		q \\ 
		r  - \tfrac{q^2}{2r} 
		\end{smallmatrix}
		\right) \in P_{\tau=0} \; . 
	\end{equation}
Thus \eqref{GHK}  yields 
	\begin{align}
		\left( u_{\nu}^{\pm} (g) h_+  \right) (s')  
		& = \bigl| x \cdot p_\pm  (t) \bigr|^{- \frac{1}{2} -i \nu} 
		h_{(-)^j} (s)  \; , 
	\label{H}
	\end{align}
with $s$, $j$ and $t$ defined by \eqref{sktq} and $x \in P_{\tau=0}$. 
\end{remark}

\begin{proof} 
If  $p \in \Gamma_1$ and $p_\pm = \Lambda_2 (s)  p_\pm (0)$, then the cosets $g MAN$ can be identified with $\Gamma_1$, 
and we may thus consider 
	\[
	\bigl(  \Pi_{\nu}^{\pm}  (g) f \bigr)  (\Lambda_2 (s'))
	= \sqrt{\lambda_g \bigl( \Lambda_2 (s') MN \bigr)} \; f \bigl(  g^{-1}  \Lambda_2 (s') \bigr) 
	\]
Thus the induced representation takes the form \eqref{GHK}. 
\end{proof}

\begin{remark} In particular \cite[Equ.~(4.45--4.47)]{BM}, recalling \eqref{R0-s},
we find
	\begin{align*}
		\left( \widetilde u_{\nu} (R_0(\alpha)) f \right) (p)  
		& =  
		\begin{cases} 
			\left(\tfrac{  | p_2 |}{m} \right)^{-\frac{1}{2} - i \nu} 
			\begin{pmatrix} 
					h_+  ( (p_1 \cos \alpha - m \sin \alpha) / p_2) \\
					h_- (  (p_1 \cos \alpha - m \sin \alpha) / p_2)
			\end{pmatrix}  & \text{if $p_2>0$ \; ,}
		\\
		\\
			\left(\tfrac{  | p_2 |}{m} \right)^{-\frac{1}{2} - i \nu}
			\begin{pmatrix} 
		 		h_- (  (p_1 \cos \alpha - m \sin \alpha) / p_2) \\
				h_+ ( (p_1 \cos \alpha - m \sin \alpha) / p_2) 
			\end{pmatrix}  & \text{if $p_2<0$\; ,}
		\end{cases}
	\end{align*}
with 
	\[	p_2 = p_1 \sin \alpha + m \cos \alpha \; . \]
Recalling \eqref{lambda1-s}, we find
	\begin{align*}
		\left( \widetilde u_{\nu} (\Lambda_1(t)) f \right) (p)  
		& =  
				\begin{cases} 
			| p'_2 |^{-\frac{1}{2} - i \nu} 
			\begin{pmatrix} 
					h_+  (  p_1 / p'_2 ) \\
					h_-  (  p_1 / p'_2 )
			\end{pmatrix}  & \text{if $p'_2 > 0$ \; ,}
		\\
		\\
			| p'_2 |^{-\frac{1}{2} - i \nu} 
			\begin{pmatrix}
		 		h_-  (  p_1 / p'_2 )  \\
				h_+  (  p_1 / p'_2 ) 
			\end{pmatrix}  & \text{if $p'_2 < 0$\; ,}
		\end{cases}
	\end{align*}
with 
	\[
		p'_2 = \tfrac{m \cosh t - \sqrt{p_1^2+ m^2} \sinh t }{m} \;. 
	\]
Finally, 
	\begin{align*}
		\left( \widetilde u_{\nu} (\Lambda_2(s)) f \right) (p)  
		& =  \begin{cases}
		h_+ ( p_1 \cosh s - \sqrt{p_1^2+ m^2} \sinh t) \\
		h_- ( p_1 \cosh s - \sqrt{p_1^2+ m^2} \sinh t)		
		\end{cases} \; . 
	\end{align*}
\end{remark}

\subsection{The group contraction $SO_0(1,2)$ to $E_0(1,1)$}

On the two-dimensional Minkowski space, the plane waves 
	\[
		(t, q) \mapsto {\rm e}^{i (t, q) \cdot (\sqrt{p_1^2+ m^2}, p_1) } 
	\] 
can be interpreted as 
improper common eigenvectors of  the space-time translations $T(t', q')$,  $(t', q') \in \mathbb{R}^{1+1}$.
They form an improper basis in the eigenspace of the Casimir operator 
	\[
		{\it M}^2 =\sqrt{ {\it P}_0^2- {\it P}_1^2}  
	\] 
for the eigenvalue $m^2>0$.

The energy operator ${\it P}_0$ and the momentum operator~${\it P}_1$ act like multiplication operators in Fourier space:
	\begin{align*}
		{\it P}_0 \; {\rm e}^{i \,   (t, q) \cdot (\sqrt{p_1^2+ m^2}, p_1) } & = \sqrt{p_1^2+ m^2} \;  {\rm e}^{i \,  (t, q) \cdot (\sqrt{p_1^2+ m^2}, p_1)} \; , \\
		\qquad {\it P}_1 \;  {\rm e}^{i \,   (t, q) \cdot (\sqrt{p_1^2+ m^2}, p_1) } & = p_1 \; {\rm e}^{i \,  (t, q) \cdot (\sqrt{p_1^2+ m^2}, p_1)} \; . 
	\end{align*}
We note that the inner product 
$	 (t, q) \cdot (\sqrt{p_1^2+ m^2}, p_1) $ 
equals $m$ times the Euclidean distance of the point $(t, q)$ 
from the line passing through the origin whose normal vector is $(\sqrt{p_1^2+ m^2}, - p_1) $.

Now let us compare this with the situation on de Sitter space. Consider a point $x \in  \Gamma^+ (W_1) \subset dS$, 
parametrized\footnote{Note that in \cite[p.~358]{BM} the order of the group elements in the parametrisation is reversed.} 
 by $t$ and $q$, \emph{i.e.}, 
	\begin{align*}
		x (t, q) & = \Lambda_1\left( \tfrac{t}{r} \right) D \left( \tfrac{q}{r} \right) 
		\left( \begin{smallmatrix} 0 \\ 0 \\ r \end{smallmatrix} \right) \\
		& = 	\begin{pmatrix}
										 \cosh \frac{t}{r} &  0 &\sinh \frac{t}{r} \\
										0 &  1&0  \\
										\sinh \frac{t}{r} & 0 & \cosh \frac{t}{r}   
 									\end{pmatrix} \begin{pmatrix} \frac{q^2}{2r} \\ q \\ r- \frac{q^2}{2r} \end{pmatrix} 
				 = 	\begin{pmatrix}
										 \frac{q^2}{2r} {\rm e}^{- \frac{t}{r}}  +  r \sinh \frac{t}{r} \\
										q  \\
										\frac{q^2}{2r} {\rm e}^{- \frac{t}{r}} +   r \cosh \frac{t}{r}   
 									\end{pmatrix} 
							\; .  
	\end{align*}
We have $\lim_{r \to \infty} \left[ x  (t, q) - \left( \begin{smallmatrix}
									  0  \\
									  0 \\
									  r
									\end{smallmatrix} \right) \right] = \left( \begin{smallmatrix}
									   t  \\
									  q \\
									   0
									\end{smallmatrix} \right) $. 
In particular, for $r$ very large the $x_2$-component of $x  (\tau, \xi)$ approaches $r$.

\begin{lemma} [Bros \& Moschella \cite{BM}] As $r \to \infty$,
the generalised plane wave\footnote{Note that  in a neighbourhood of the origin 
the curves in $dS$, for which $x \cdot p $ is constant, become straight lines in $dS$
perpendicular to $p$ as the radius $r \to \infty$ in \eqref{eqdSMin}.} (see Section \ref{planwaves})
	\begin{equation}
	\label{pw-23}
		\mathbb{R}^{1+1} \ni (t, q)
		\mapsto  \left(  \tfrac{x (t, q)}{r} \cdot  \tfrac{p}{ m }
					\right)^{-\frac{1}{2} \mp i m r } 
	\end{equation}
approaches---see \eqref{distance-to-horosphere} and \eqref{H}---the plane wave 
	\begin{equation}
	\label{mpw}
		(t, q)  \mapsto {\rm e}^{\pm i t \sqrt{p_1^2 + m^2}  \mp i q p_1}  \; ,
	\end{equation}
of a Minkowski space particle with mass $m$. 
\end{lemma}

\begin{proof} 
We set $ p  = \left( \begin{smallmatrix}
		\sqrt{p_1^2+ m^2} \\
		p_1 \\
		\pm m
		\end{smallmatrix} \right)$, $ m >0 $. Now consider \eqref{H}. Note that 
$ x \cdot p_\pm  (t) =  \Lambda_1 (-t) x \cdot p_\pm$.
It follows that 
	\begin{align*}
	\lim_{r \to \infty} & \frac{1}{mr}
	\left[ \begin{pmatrix}
										 \frac{q^2}{2r} {\rm e}^{\frac{t}{r}}  - r \sinh \frac{t}{r} \\
										q  \\
										\frac{q^2}{2r} {\rm e}^{\frac{t}{r}} +   r \cosh \frac{t}{r}   
 									\end{pmatrix} \cdot \begin{pmatrix}
		\sqrt{p_1^2+ m^2} \\
		p_1 \\
		\pm m
		\end{pmatrix} \right]^{-\frac{1}{2} \mp i m r }  
	\\
	& = 
	\lim_{r \to \infty} {\rm e}^{ \ln ( 1- \frac{t \sqrt{p_1^2 + m^2}  -  q p_1}{mr}) (-\frac{1}{2} \mp i m r )} \; . 
	\end{align*}
The result now follows by expanding the logarithm \cite[Equ.~(4.5)]{BM}.
\end{proof}

Let $ {\mathscr D}_{\pm m}$ be the unitary irreducible representation of the Poincar\'e group $E_0(1,1)$
for mass $m$ and spin zero induced by the representation 
	\[
		(t,q) \mapsto {\rm e}^{\pm i t \sqrt{p_1^2+ m^2} \mp i q p_1}
	\]
of the translation subgroup and by the representation $\sigma_{+}$ of the little group $\{ \mathbb{1}, P_1 \}$, 
\emph{i.e.}, 
	\[
		\Bigl( {\mathscr D}_{m} \bigl(\Lambda_2(s) T(t, q) \bigr) g_+ \Bigr) (p_1)= {\rm e}^{i (t, q) \cdot (\sqrt{p_1^2+ m^2}, p_1) } 
		g_+ (p_1+ p_1' ) \; , 
		\quad (t, q) \in \mathbb{R}^{1+1} \; ,  
	\]
and $m \sinh s=  p_1' $. Here $T(t, q)$ denotes the space-time translations in $\mathbb{R}^{1+1}$.

The following result is closely related to a theorem on group contractions by Mickelsson and Niederle \cite[Theorem 2]{MN}. 

\begin{theorem}
\label{r-unendlich}
Consider the representation $\widetilde u_{mr}^{+}  \equiv \widetilde u_{mr} $ of $SO_0(1,2)$ acting on the de Sitter space $dS$. 
Let $g \in L^2(\Gamma_1,{\rm d} \mu_{\Gamma_1})$. Then 
	\[
	\lim_{r \to \infty}
	\left\| \widetilde u_{mr} \left( \Lambda_2 \left( s \right) \Lambda_1 \left( \tfrac{t}{r} \right) D \left( \tfrac{q}{r} \right) \right) g
	- {\mathcal D}_{m} \bigl( \Lambda_2 (s) T( t,q) \bigr) g \right\|_{L^2(\Gamma_1,{\rm d} \mu_{\Gamma_1})} \; . 
	\]
Here ${\mathcal D}_{m}$ is the reducible representation of the Poincar\'e group $E_0(1,1)$ given by
	\[
		{\mathcal D}_{m} = {\mathscr D}_{m} \oplus {\mathscr D}_{-m} \; . 
	\]
\end{theorem}

\begin{proof} We have to show that 
	\begin{equation}
	\label{contraction}
	\left(\widetilde u_{m r} \left( \Lambda_2 (s ) 
		\Lambda_1 ( \tfrac{t}{r} ) D ( \tfrac{q}{r} ) \right) g_\pm \right) (p_1)
		\to  {\rm e}^{\pm i ( t \sqrt{p_1^2 +m^2} -  q p_1 ) } g_\pm ( p_1 + p_1' ) 
	\end{equation}
in $L^2 \bigl(\mathbb{R}, \tfrac{{\rm d} p_1}{2 \sqrt{p_1^2 +m^2}} \bigr)$ as $r \to \infty$, with $p_1'= m \sinh s$.
Note that, by definition, 
	\[
		\widetilde u_{m r} \left( \Lambda_2 ( s ) g_\pm \right) (p_1) = g_\pm  (p_1+p_1') \; .
	\]
Thus it is sufficient to show that 
	\[
	\widetilde u_{m r} \left( 
		\Lambda_1 ( \tfrac{t}{r} ) D ( \tfrac{q}{r} ) \right) g_\pm  
		\mapsto  {\rm e}^{\pm i ( t \sqrt{\, .\,^2 +m^2} -  q \, .\, ) } g_\pm  
		\qquad \text{in $L^2(\Gamma_1,{\rm d} \mu_{\Gamma_1})$ as $r \to \infty$}.
	\]
In order to be able to interchange the limit with the integration, we 
approximate $g_\pm$ with continuous functions with compact support.
Set
	\[
		F^\pm_{r} (p_1) \doteq \left| \left(\widetilde u_{m r} \left( 
		\Lambda_1 ( \tfrac{t}{r} ) D ( \tfrac{q}{r} ) \right) g_\pm  \right) (p_1) - {\rm e}^{\pm i ( t \sqrt{p_1^2 +m^2} -  q p_1 ) } g_\pm (p_1)  \right|^2 \; . 
	\]
It follows that  
	\begin{equation}
	\label{contraction-2}
	\lim_{r \to \infty} \left( \int_{\Gamma_1} {\rm d} p_1 \; F^\pm_{r} (p_1) \right)^{1/2} 
	= \left( \int_C {\rm d} p_1  \lim_{r \to \infty} F^\pm_{r} (p_1) \right)^{1/2} = 0 \; . 
		\end{equation}
We have used \eqref{H} and the fact 
that $F^\pm_{r} (p_1)$ is zero outside of some compact region $C \subset \Gamma_1$ for $t, q$ fixed and $r$ 
sufficiently large. (This follows form  $g_\pm \in C_0(\mathbb{R})$.)
Note that the set  $C_0(\mathbb{R}) \subset L^2 \bigl(\mathbb{R}, \tfrac{{\rm d} p_1}{2 \sqrt{p_1^2 +m^2}} \bigr)$ 
of continuous functions with compact support is dense in $L^2$.  (In Minkowski space, these 
functions would be called finite energy wave-functions.)  Thus \eqref{contraction} follows for  
$g_\pm \in L^2(\Gamma_1,{\rm d} \mu_{\Gamma_1})$.
\end{proof}

\chapter{Harmonic Analysis on the Hyperboloid}
\setcounter{equation}{0}
\label{Harmanaly}

Harmonic analysis on semi-simple Lie groups originated with the monumental work of 
Harish-Chandra \cite{HC1} -- \cite{HC9}. The subject  has been developed further in 
particular by Helgason. But strictly speaking, the framework considered by Helgason~\cite{He}, 
namely  symmetric spaces~$G / H$, does not cover the case of the one-sheeted hyperboloid, 
as the stabilizer $H \cong (\mathbb{R}, +)$ of~the Lie group $G\cong SO_0(1,2)$ fails to be a compact subgroup. 
The  necessary alterations can be found in the work of Molchanov~\cite{Mo1, Mo2} and Faraut~\cite{Fa}.  

In this work we follow a more recent approach to Fourier(-Helgason) transformation on de Sitter space, 
due to Bros and Moschella~\cite{BM2},  which emphasises the role of  {\em tuboids}. The advantage of this 
approach is that the analyticity properties of the Fourier transform are evident from the definition.

\section{Tuboids}
\label{tuboidsds}

Consider  the complex de Sitter space
\label{dSccpage}
	\[  
		dS_{\mathbb{C}} \doteq \{  z \in \mathbb{C}^{1+2} 
		\mid z_0^2 - z_1^2 - z_{2}^2 = - r^2 \} \;  .
	\] 
A {\it tuboid} is a subset of $dS_{\mathbb{C}}$, which is bordered 
by  real de Sitter space $dS$ and whose {\em shape} (called its {\em profile})
near each point $x$ of $dS$ can be mapped to
a  cone ${\mathcal P}_{ x}$ in the tangent space $T_{ x }dS$. The precise definition
of a profile  can be found in \cite{BM}; for the benefit of the reader we recall it 
in Appendix \ref{deftub}. 

\goodbreak
Complex de Sitter space $dS_{\mathbb{C}}$ is equipped \cite{BM2} with four distinguished tuboids, 
which are invariant under the action of the proper orthochronous Lorentz group $SO_0(1,2)$:
	\begin{align*} 
		\label{t+1}  
		{\mathcal T}_\pm &\doteq  \{  \Lambda \; ( i r \sin \theta , 0, r \cos \theta) 
		\in dS_{\mathbb{C}} \mid  0 <  \mp \theta < \pi ,  \;  \Lambda \in SO_0 (1,2)\} \; , 
		 \nonumber \\  
				{\mathcal T}_{\hskip -.1cm {\quad \atop \rightarrow} \atop  \hskip -.1cm \leftarrow} \hskip -.2cm
		&\doteq \{ \Lambda  ( 0, i r \sinh t , r \cosh t) \in dS_{\mathbb{C}} \mid  \mp t >0 ,  \;  \Lambda \in SO_0 (1,2) \} \; . 
	\end{align*}
The {\em chiral} tuboids ${\mathcal T}_{\leftarrow}$ and~${\mathcal T}_{\rightarrow}$ are not simply-connected. Their profiles 
at the origin~$ o = (0,0, r)$ of $dS$ are  the cones
	\[
		   T_{ o} dS \cap \{  y \in \mathbb{R}^{1+2} \mid  \mp y_1 > |y_0|   \} \; . 
	\]
The chiral tuboids play a key role for quantum fields on anti-de Sitter space \cite{BuSu, AdS}, but are of no relevance for this work. 
\label{tuboidspage}

\goodbreak
The {\em Lorentzian  tuboids}  ${\mathcal T}_\pm$ are  similar in many respects to the 
tubes\footnote{Of course, $V_+$ here 
denotes the future light cone in $\mathbb{R}^{1+2}$. Note that our sign convention follows \cite{SW}, 
in contrast to the less common sign convention chosen in~\cite{BM2}.} 
	\begin{equation*}
		{\mathfrak T}_\pm = \mathbb{R}^{1+2} \mp i V^+ 
	\end{equation*}
defined in complex Minkowski space. In fact~\cite[Proposition 2]{BM2}, 
	\begin{equation*} 
 		{\mathcal T}_\pm = {\mathfrak T}_\pm \cap dS_{\mathbb{C}} \; .  
	\end{equation*} 

\begin{proposition}[Proposition 1, \cite{BM2}]
The tuboids ${\mathcal T}_\pm$ consists of all points $z \in dS_{\mathbb{C}}$ for which 
the inequality $\mp \Im z \cdot p >0$ holds for all $p \in \partial V^+$, \emph{i.e.}, 
	\begin{equation} 
		\label{tube-0}
 		{\mathcal T}_\pm = \bigl\{ z \in dS_{\mathbb{C}} \mid  \mp \Im z \cdot p >0 \; \forall p \in \partial V^+ 
		\setminus \{(0,0,0)\}
		 \bigr\} \; .  
	\end{equation} 
\end{proposition}

\begin{proof}
Consider the vectors $p =(1,0,-1)$ and $q= (0,r,0)$. Since $p \cdot p=0$ and $p \cdot q = 0$, we find 
	\begin{equation}
	\label{zero-plane}
	( \lambda p + \mu q) \cdot p = 0 \; , \qquad \lambda , \mu \in \mathbb{R} \; . 
	\end{equation}
The plane\footnote{We note that this plane contains the light rays
$(0, \pm r,0) + \lambda (1, 0, -1)$, $\lambda \in \mathbb{R}$, which form the horosphere $P_{- \infty }$, see \eqref{horosphere}.}
spanned by $p$ and $q$ separates the regions in $\mathbb{R}^3 \ni x$ for 
which $x \cdot p <0$ and $x \cdot p >0$, respectively. The latter half-space includes the positive $x_0$-axis. 
Rotating the vector $p$ and taking the intersection of the 
resulting regions  for which the scalar product $x \cdot p$
is positive, yields the forward light cone $\partial V_+ \setminus \{(0,0,0)\}$.   
\end{proof}

The  profile of the {\em forward tuboid}  ${\mathcal T}_+$ near each point $ x$ of $dS$ (in the space 
of $\Im  z$ and for $\Im  z \searrow 0$) is the cone
	\begin{equation}  
		\label{px}
		{\mathcal P}^+_{ x}=  T_{ x } dS \cap (-V_+)  
	\end{equation} 
in the tangent space $T_{ x } dS$  at the point $ x \in dS$. 
(Note that in \eqref{px} the tangent space $T_{ x } dS \cong \mathbb{R}^2$ at $ x \in dS$ 
is viewed as a subspace of $T_{ x } \mathbb{R}^3 \cong \mathbb{R}^3$.) For the origin $o \in dS$
this yields ${\mathcal P}^+_{o}=  \{ y \in \mathbb{R}^3 \mid - y_0 > |y_1|, y_2 = 0 \}$.

\subsection{The Euclidean sphere}
Applying\footnote{Applying the boosts $\Lambda_1 (t)$, $t \in \mathbb{R}$, to the half-circles \eqref{half-circle-S1}, followed by  
the rotations $R_{0} (\alpha)$, $\alpha \in [0, 2 \pi)$, yields the interior of the one-sheeted hyperboloid in $( i \mathbb{R}) \times 
\mathbb{R}^2$ with the (closed) past light cone and the interior of the
future mass shell for the value $m=r$ removed.}
the rotations $R_{0} (\alpha)$, $\alpha \in [0, 2 \pi)$, to the half-circles 
	\begin{equation}
	\label{half-circle-S1}
	 \{  ( i r \sin \theta , 0, r \cos \theta) \in dS_{\mathbb{C}} \mid  0 <  \mp \theta < \pi \} \; 
	\end{equation}
we find  the (open) \emph{lower}\index{lower hemisphere} and  
\emph{upper hemispheres}\index{upper hemisphere} 
	\begin{equation} 	
		\label{spl}
		S_\mp =\{( i\lambda_0,x_1, x_2) \in  (i\mathbb{R}) 
		\times \mathbb{R}^2 \mid \lambda_0^2 +  x_1^2 +  x_2^2 = r^2, \mp \lambda_0 > 0\}
	\end{equation} 
of the \emph{Euclidean sphere}\index{Euclidean sphere}\footnote{Note that the definition of 
$S^2$ in \eqref{euclidsphere} refers to the Lorentz 
metric~\eqref{metrik}. }
\label{euclidspherepage}
	\begin{equation} 
		\label{euclidsphere}
		S^2 = \Bigl\{ (i r \sin \theta \cos \psi,  r \sin \psi, r \cos \theta \cos \psi)  \in \mathbb{C}^3 
		\mid \theta \in (- {\textstyle  \frac{\pi}{2}, \frac{\pi}{2}}  ], 
		\psi \in ( - {\textstyle  \frac{\pi}{2}, \frac{3\pi}{2}} ]   \Bigr\} \;  .
	\end{equation} 
Thus ${\mathcal T}_\pm =  \{  \Lambda \; S_\mp \mid    \Lambda \in SO_0 (1,2)\} $  and, consequently,
	\[
		S^2 \subset \overline{{\mathcal T}_+ \cup {\mathcal T}_-} \subset dS_{\mathbb{C}} \; . 
	\]

\subsection{Rotations}
The rotations, which leave the 
Euclidean sphere \eqref{euclidsphere} invariant,
form the subgroup $SO(3)$ of $SO_\mathbb{C} (1,2)$; the imaginary part in the square bracket on 
the right hand side of~\eqref{eqBooW} is in agreement with \eqref{euclidsphere}. 
We denote the generators of the rotations around the three coordinate axis by $K_0$, $K_1$ 
and $K_2$, and  set
	\[
		K^{(\alpha)} \doteq \cos \alpha \; K_1 + \sin \alpha  \; K_2\;  , 
		\qquad \alpha \in [ 0, 2 \pi) \; ,
	\]
in agreement with the definition of $L^{(\alpha)}$ in \eqref{Lalpha}.

\goodbreak
\begin{remark}
Clearly, the decomposition of the Euclidean sphere into a lower and an upper hemisphere in
\eqref{spl} distinguishes a Cauchy surface $S^1= \partial S_\mp$. However, 
as ${\mathcal T}_+$ is invariant under the action of $SO_0 (1,2)$, one might as well 
consider the Lorentz transformed Cauchy surface $\Lambda S^1\subset dS$ together with a 
Lorentz transformed sphere 
	\[
		\Lambda S^2 \subset \overline{{\mathcal T}_+ \cup {\mathcal T}_-} \subset dS_{\mathbb{C}} \; , 
		\qquad \Lambda \in SO_0(1,2) \; . 
	\]
\end{remark}

\begin{lemma} 
\label{lemma1.2}
Let $x = R_{0} (\psi) o \in S^1$, $\psi \in [0, 2 \pi)$. Then 
	\begin{equation} 
		\label{waw}
		\Gamma^\pm ( x  )
		= 	\left\{ 
				{\Lambda^{(\alpha)}}(t)  x  \in dS  
				\mid t  \in \mathbb{R}^\pm \; , \; 
				\alpha \in  ( {\textstyle  \psi - \frac{\pi}{2}, \psi + \frac{\pi}{2}  })  \right\} \; . 
	\end{equation} 
Moreover,  the map 
	\begin{equation} 
		\label{waw1}
		\tau \mapsto{\Lambda^{(\alpha)}} (\tau)   x \; , 
		\qquad \alpha \in (  {\textstyle  \psi-\frac{\pi}{2}, \psi+\frac{\pi}{2} } )  \; , 
	\end{equation} 
is entire,  and for $\tau \in {\mathbb S} = \mathbb{R} \mp i ( 0,\pi ) $   
the map \eqref{waw1} takes values in $ {\mathcal T}_\pm$. 
\end{lemma}

\begin{remark} 
Given an arbitrary point $ x \in dS$, formulas analogous to \eqref{waw} 
and~\eqref{waw1} hold true for all possible choices of  space-like geodesics passing through
the point~$ x$. 
Note that a space-like geodesic is used to define ${\Lambda^{(\alpha)}}$. 
\end{remark}

For $\alpha \ne 0$ the map $\mathbb{R} \ni t \mapsto {\Lambda^{(\alpha)}}(t) o$ no longer describes the geodesic 
motion of a free falling observer. As $\alpha \to \pm \pi /2$, the observer following the path 
$\{ {\Lambda^{(\alpha)}}(t) o \mid t \in \mathbb{R} \}$ is exposed to a {\em uniformly accelerated motion}, namely 
a boost, and will  observe a temperature~$\bigl( (2 \pi) \cos \alpha \bigr)^{-1}$. This result follows by parameterising 
the path~\eqref{waw} in the proper time, see~\eqref{Lalpha} and also~\cite{NPT}.
In other words,  the result of {\em Bisognano-Wichmann}~\cite{BiWia, BiWib} and Unruh~\cite{Unruh} 
remains valid on $dS$ (see also~\cite{DeBeM}). 

\begin{lemma} 
\label{t+2}
Let 
	$
	M \doteq \{ (\psi, \alpha) \in S^1 \times S^1\mid   | \alpha -\psi  | < \pi /2 \}
	$
be the double twisted M\"obius strip. Here $| \alpha -\psi  |$ denotes the minimal distance on~$S^1$. 
The map
	\begin{align} 
		\label{ganzetube}
				 {\mathbb S}  \times  M 
				& \to  {\mathcal T}_+ \nonumber \\ 
				 (\tau, \psi, \alpha) & \mapsto  {\Lambda^{(\alpha)}} (\tau) R_{0} (  \psi)  o   
	\end{align} 
is surjective and, if restricted to   $- \pi/2 < \Im \tau <0$,  it is one-to-one onto 
$ {\mathcal T}_+ \setminus \{  z \in dS_\mathbb{C} \mid \Re  z = 0  \} $. 
\end{lemma}

\goodbreak
\begin{proof} 
Let $\tau=t +i\theta$, with $- \pi/2 < \theta <0$. Then
	\begin{equation} 
		\label{3.16}
		{\Lambda^{(\alpha)}} (\tau) R_{0} ( \psi)  o =  u +i y
	\end{equation} 
 with
	\begin{equation} 
		\label{3.17}
 u = \kern - .1cm 
		\begin{pmatrix}
			\cos (\psi- \alpha) \cos \theta \sinh t \\
			- \cos \alpha \sin (\psi- \alpha)- \sin \alpha \cos (\psi- \alpha) \cos \theta \cosh t \\
			- \sin \alpha \sin (\psi- \alpha)- \cos \alpha \cos (\psi- \alpha) \cos \theta \cosh t  \\
		\end{pmatrix}
	\end{equation} 
and
	\begin{equation} 
		\label{3.18}
		    y = \sin \theta \, \cos (\psi - \alpha) 
		   	\begin{pmatrix}
					\cosh t \\
					- \sin \alpha \sinh t \\
					\cos \alpha \sinh t \\
			\end{pmatrix} \; . 
	\end{equation} 
The vector $y$ is time-like, i.e.,  $0 \le  y \cdot y \le 1$, and 
	\[
		 x =\frac{1}{\sqrt{1-  y \cdot y }} \;  u \in dS \; . 
	\]
Moreover, $  u  \cdot  y = 0$. The equality $  u \cdot  y = 0$ implies that $ u + i  y \in dS_{\mathbb{C}}$, as 
	\begin{equation} 
		\label{rmp3.20}
		dS_{\mathbb{C}} = \left\{ ( u,  y) \in \mathbb{R}^6 \mid  
		u \cdot u -  y \cdot y = -1,   u \cdot  y =  0 \right\} \; .
	\end{equation} 
Now assume that  $u+i y$ can be written (see \eqref{3.16}) as $\Lambda^{(\alpha')} (\tau') R_{0} ( \psi') o$. 
A short calculation, using \eqref{3.17} and~\eqref{3.18}, shows that $\psi'=  \psi$, $\alpha' = \alpha$ and 
$ \tau' = \tau$, using  the restriction $-\frac{\pi}{2}< \Im \tau, \Im \tau' < 0$ to ensure the latter equality. Thus 
there are no further ambiguities, and uniqueness of the restriction follows.
\end{proof}

The coordinates provided by the map \eqref{ganzetube} can not be extended  to the whole  
tuboid~${\mathcal T}_+$: the south pole of the Euclidean sphere $(i, 0, 0) \in S^2$  would 
correspond to  $\theta= - \frac{\pi}{2}$ and $\psi \in S^1$. Similarly the coordinate system would  
be degenerated at every single point in the purely imaginary negative unit  mass-shell (compare to 
Eq.~\eqref{rmp3.20})
	\[
		\{  z \in {\mathcal T}_+ \mid \Re  z = 0  \} 
		= \{ \Lambda (-i,0,0) \mid \Lambda \in SO_0(1,2) \}\; .
	\]
For $ - \pi < \theta < 0$,  $\theta \ne \frac{\pi}{2}$, the identity 
	\[ 
		\Lambda^{(\psi)}(i\theta) R_{0} ( \psi)  o 
			= \Lambda^{(\psi  + \pi )}(i (-\pi- \theta)) R_{0} ( \psi + \pi )  o 
	\]
exemplifies the two possibilities to reach a single point on the Euclidean sphere 
(within the tuboid ${\mathcal T}_+$) from the circle~$S^1$; the two points $R_{0} ( \psi) o$ 
and $R_{0} ( \psi + \pi ) o$ are opposite to each other on the circle, and 
	\[
		\Lambda^{(\alpha + \pi) }(\tau) =  {\Lambda^{(\alpha)}} (-\tau)
	\]
for all $\tau \in \mathbb{C}$ with $\Re \tau = 0$.  

\begin{lemma}
\label{eqzWsx} 
For every point $z\in{\mathcal T}_+$ one can find two wedges $ W_1$, $W_2$, close to each 
other\footnote{Two wedges are {\em close} to each other if their edges are close to each other 
in $dS$ w.r.t.~the Euclidean~metric.}, and two angles $\theta_1,  \theta_2 \in(0,\pi)$
as well as two points $x_1 \in W_1$ and $x_2 \in W_2$ such that   
\begin{itemize} 
\item[$i.)$] $\Lambda_{W_1}(i\theta_0)x_1 = z = \Lambda_{W_2}(i \theta_2)  x_2$; 
\item[$ii.)$] the map   
	\begin{equation} 
		\label{eqHoloChart}
		(\tau_1, \tau_2) \mapsto \Lambda_{W_2}(\tau_2) \Lambda_{W_1}(\tau_1) z 
	\end{equation} 
gives rise to a holomorphic chart in a neighbourhood of $z$. 
\end{itemize}
\end{lemma}
 
\begin{proof}
It is known that for every $z\in{\mathcal T}_+$ there exists~\cite[Lemma
A.2]{MSY} an angle $\theta_1 \in(0,\pi)$ and some  $\Lambda\in SO_0(1,2)$ such that 
	\[
	\Lambda^{-1}\,	z =  z_0 =	   	\begin{pmatrix}
					0 \\
					0 \\
					r \\
			\end{pmatrix}  
\cos \theta_1 +i 		   	\begin{pmatrix}
					r \\
					0 \\
					0 \\
			\end{pmatrix}  
\sin \theta_1 \in S_+ \; , 
	\]
i.e.,  $z_0=\Lambda_1(i\theta_1)o$, with
$o=(0,0,r)$ the origin; see also \eqref{eqBooW}. Hence $z$ is of the form as claimed in~(i), namely 
$z=\Lambda_{W_1}(i\theta_1)x_1$, with $W=\Lambda W_1$ and
$x_1=\Lambda x_0$. 

It is noteworthy that $\theta_1$, $x_1$ and $W_1$ can be directly characterized
in a coordinate-independent manner: 
let $z=u+iy$. The real part $u$ satisfies 
$\tfrac{u \cdot u}{r^2} \in (-1,0]$ and is orthogonal to $y$.  Then 
	\[
	\theta_1 =  \arccos  \sqrt{-\frac{ u\cdot u}{r^2} } \; , \qquad   x_1 = 
	\frac{ u }{ r \cos \theta_1 }  \; ,
	\]
and $W_1$ is the causal completion of the unique time-like geodesic in
$dS$ starting at $x_1$ with initial velocity
$y$. (Note that $y$ is orthogonal to $x_1$ and can therefore be
identified with a tangential vector at $x_1$.) 

By construction, the boosts $\Lambda_{W_1} (t)$ leave the $u$-$y$-plane in the ambient space
$\mathbb{R}^{1+2}$ invariant. Hence the generator $L_{W}$
leaves the complex $u$-$y$-plane in ambient $\mathbb{C}^{3}$ invariant\footnote{In fact,  
$L_{W} u \parallel y$ and  $L_W y \parallel u$.}. 
Now pick a different wedge $W_2$ sufficiently close to $W_1$ and such that the
vector $L_{W_2} z$ is not in the $u$-$y$-plane. (This implies that $W_2 \ne R_{0} (\alpha) W_1$ for all 
$\alpha \in [0, 2 \pi)$.) 
It follows that the vectors $L_{W_1}z$ and $L_{W_2}z$ are 
linearly independent and \eqref{eqHoloChart} is a holomorphic chart in
a neighbourhood of~$z$.
Furthermore, if $W_2$ is close enough to $W_1$, then 
the line segment 
	\[	
		\{\Lambda_{W_2}(-i\theta_2)z \mid 0< \theta_2  <\pi \}
	\]
intersects $dS$ in some point, say $x_2 \in dS$ (just like the line
segment $\{\Lambda_{W_1}(-i\theta_1)z \mid 0< \theta_1  <\pi \}$, which
intersects $dS$ in $x_1 \in dS$).  Thus there is some $
\theta_2$ such that $z= \Lambda_{W_2}(i \theta_2) x_2$, as claimed in~i.). 
\end{proof}

Next we provide a  {\em flat tube theorem}  (see, e.g., \cite{BB, BEGS}; an early result of this type 
is due to Malgrange and Zerner) for the de Sitter space. 
   
\begin{theorem}
\label{flattube}
Let $f$ be a tempered distribution on $dS$ with the
following property: for any wedge $W \subset dS$ and any $x\in W $, the map 
	\begin{equation}
  		\label{eqStripHolo}
		t\mapsto f(\Lambda_{W}(t)x)
	\end{equation}
can be uniquely extended to a function defined  and  analytic in the strip ${\mathbb S}  
= \mathbb{R}+i ( 0,\pi ) $, whose boundary values are described by \eqref{eqStripHolo}.

Then $f$ is the boundary value, in the sense of distributions, of a unique function~$F$, 
which is analytic in the tuboid ${\mathcal T}_+$. 
\end{theorem}

\begin{proof}
In a first step, assume that $f$ is a continuous function. Let 
\[
	z \in {\mathcal T}_+,  \quad x_1 \in W_1, 
	\quad  x_2 \in  W_2 \quad 
	\text{and} \quad - \pi < \theta_1, \theta_2 < 0
\] 
as in Lemma  \ref{eqzWsx}. By hypotheses, 
$f$ can be analytically continued into the point $z$  in the variable
$t+i\theta$ (defined in Eq.~\eqref{eqHoloChart})  to $z$ along the path 
	\[
		\theta \mapsto \Lambda_{W_1}(i\theta)x_1, \qquad \theta_1 \leq \theta\leq 0 \; . 
	\]
$f$ can as well be analytically continued to $z$ 
in the variables $t_2 +i  \theta_2$, namely along the path 
$\theta' \mapsto \Lambda_{W_2}(i \theta') x_2$.
Both continuations coincide at~$z$, 
yielding a function~$F$, which is holomorphic separately in the variables $t+i\theta$ and 
$t'+i \theta'$ in a neighbourhood~$U$ of~$z$. By the flat tube 
theorem~\cite[Vol.~I, Theorem~IX.14.2]{RS}, $F$ is 
holomorphic on an open convex subset of $U$. Since $z$ was an arbitrary point 
in~${\mathcal T}_+$, it follows that $F$ is holomorphic in ${\mathcal T}_+$. This proves 
the statement in case $f$ is a continuous function. 

The necessary generalization to tempered distributions, together with an appropriate generalization 
of the Bros-Epstein-Glaser Lemma \cite[Vol.~II, Theorem~IX.15]{RS},
will be given elsewhere. 
\end{proof}

The following result clarifies the relation between the  tuboid 
${\mathcal T}_+$, as described in Lemma~\ref{t+2}, and the tangent bundle $TdS$.  
 
\begin{lemma}[Bros and Moschella \cite{BM}, p.~339]
\label{tds}
The map
	\begin{align} 
		\label{tm}
				\bigcup_{ x \in dS} 
				\left(  x, T_{ x } dS \right) & \to  dS_\mathbb{C}   \nonumber \\
				 ( x,  y ) & \mapsto  \sqrt{1- y \cdot y} \; \;  x + i  y    
	\end{align}
is a one-to-one $C^\infty$-mapping from $\cup_{ x \in dS} ( x, T_{ x} dS)$ onto 
$dS_\mathbb{C} \setminus  \{  z \in dS_\mathbb{C} \mid \Re  z = 0  \}$, and
if $ y \in V^+$, the map \eqref{tm} defines a diffeomorphism from 
	\[
		\bigcup_{ x \in dS} \left(  x,  {\mathcal P}^+_{ x} 
			\cap \{  y \in V^+\mid  y \cdot y < 1 \}  \right) 
	\]
onto  $ {\mathcal T}_+ \setminus \{  z \in dS_\mathbb{C} \mid \Re  z = 0  \} $.  
$ {\mathcal P}^+_{ x}$ is defined in \eqref{px}. 
\end{lemma}

\section{Plane waves}
\label{planwaves}

The eigenfunctions of the Casimir operator on the light-cone $\partial V^+$
are homogeneous functions of degree $s=s^\pm$; see \eqref{dd1}.
Thus, in order to construct a {\em plane wave} on $dS \ni x$, 
one considers homogeneous functions of the scalar product 
	\begin{equation}
	\label{arbitrary-plane}
	x \cdot p = ( x + \lambda p + \mu q) \cdot p   \; , \qquad \lambda , \mu \in \mathbb{R} \; . 
	\end{equation}
In \eqref{arbitrary-plane} we have used \eqref{zero-plane}, with $q \in S^1$ such that $q \cdot p= 0$. 
Given a point $x \in \Gamma (W_1)$, the intersection of the plane\footnote{Of course, 
the planes \eqref{x-plane} for different $x \in \Gamma (W_1)$ are all parallel to each other.}
	\begin{equation}
	\label{x-plane}
	\left\{ x  +  \left[ \lambda 
	 \left( \begin{smallmatrix} 1 \\ 0 \\ -1 
	 \end{smallmatrix} \right)
	  + \mu q \right]  \mid  \lambda , \mu \in \mathbb{R}  \right\} \; ,  
	\end{equation}
with the de Sitter space $dS$ is just the horosphere
	\[
	 	P_{ \tau } = \left\{ y  \in dS  \mid  r {\rm e}^{\frac{\tau}{r}}  = y \cdot \left( \begin{smallmatrix} 1 \\ 0 \\ -1 
	 \end{smallmatrix} \right) \right\} \; ,  
	\]
where $\tau \in \mathbb{R}$ is fixed by requesting  
$r {\rm e}^{\frac{\tau}{r}}  = x \cdot \left( \begin{smallmatrix} 1 \\ 0 \\ -1  \end{smallmatrix} \right)$. 
The angle between the plane \eqref{x-plane}
in $\mathbb{R}^{1+2}$ containing the horosphere $P_{ \tau } $
and the $x_0$-axis is always $\pi/4$. 

For the two\footnote{We will soon integrate over $p\in \Gamma_0$. As $p$ rotates on the light cone $\partial V^+$, 
{\em all} light rays in $dS$ are affected.}
light rays forming the 
horosphere $P_{- \infty}$, \emph{i.e.}, the 
intersection of~$dS$ with the plane \eqref{zero-plane},
the scalar $x \cdot p$ vanishes and powers with negative real part have to be defined in distributional sense.
One possibility, which we will pursue, 
is to define them as the boundary values of analytic functions, using the \emph{principal value}\index{principle values} 
of the complex powers. The characterisation of the tuboid given in \eqref{tube-0} guarantees that the 
functions~
	\[   z \mapsto  ( z \cdot  p)^{s} \]
are holomorphic both in ${\mathcal T}_+$ and ${\mathcal T}_-$. Their boundary values  
as~$ z \in dS_\mathbb{C}$ tend to $x \in dS$ 
from within the respective tuboids ${\mathcal T}_+$ and ${\mathcal T}_-$ of $dS$
are denoted as 
\label{planewavepage}
	\begin{equation}
		\label{eqPW}
		 x \mapsto   (  x_\pm \cdot  p )^{s} \; , \qquad x \in dS \; .
	\end{equation}
As expected, we encounter a discontinuity as $  \Im 
x_+ \cdot p \nearrow  0$ or $  \Im   x_- \cdot p  \searrow 0$, respectively.
Another way of denoting the function \eqref{eqPW} is \cite[Eq.~(45)]{BM2}
	\begin{equation}
		\label{eqPW-b}
		(x_\pm \cdot p)^{s^+} = \mathbb{1}_{[0, \infty)} ( x \cdot p ) \;  | x \cdot p |^{s^+} 
			+ {\rm e}^{\mp i \pi s^+} \mathbb{1}_{(0, \infty)} ( - x \cdot p )\;  | x \cdot p |^{s^+} \; , 
	\end{equation}
where $\mathbb{1}_{[0, \infty)}$ is the \emph{Heaviside step function}\index{Heaviside step function}, \emph{i.e.}, 
$\mathbb{1}_{[0, \infty)} (t) = 0$ if $t<0$ and $\mathbb{1}_{[0, \infty)} (t) = 1$ if $t \ge 0$. In case $\Re s^+> -1$, the singularity is 
integrable and the equality \eqref{eqPW-b}
holds in the sense of $L^1$-functions.  

An explicit computation\footnote{The Lapalce-Beltrami operator $\square_{dS} =|g|^{-1/2}\partial_\mu g^{\mu\nu}|g|^{1/2}\partial_\nu$ on $dS$,
can be expressed as a trace of $\square_{\mathbb{R}^3}$.} \cite[Eq.~(4.3)]{BM} shows that the
plane waves given in \eqref{eqPW} satisfy the Klein--Gordon equation  
	\begin{align}
	\label{eqPW-new}
		(\square_{dS} +\mu^2)\left(  x_\pm  \cdot   p \right)^s  & =
		r^{-2} \left(  K_0^2 - L_1^2 - L_2^2 + \tfrac{1}{4} + m^2 r^2\right) \left(  x_\pm  \cdot   p \right)^{-\frac{1}{2} \mp  i mr}  \nonumber \\
		& = 0 \; , \qquad -s(s+1) = \mu^2 r^2 = \frac{1}{4} + m^2 r^2\; , 
	\end{align}
on the de Sitter space $dS$. As we have seen in Equ.~\eqref{qqdq}--\eqref{dd1}
they also satisfy the Klein--Gordon equation on the forward light cone~$\partial V^+$ (see also \eqref{casimir2}):
	\[
		(\square_{\partial V^+} +\zeta^2) \left(  x_\pm  \cdot  p \right)^s = 0 \; , \qquad -s(s+1) = \zeta^2 \; . 
	\]
Note that in contrast to the Minkowsi space case, the operators $K_0, L_1$ and $L_2$ do not commute, so they can not be represented 
as commuting multiplication operators in Fourier space.

\section{The Fourier-Helgason transformation}
\label{Fourier-HelgasonTransformation}

\begin{definition}
Let $p \in \partial V^+ $ and  $s \in \{ z \in \mathbb{C}  \mid -z (z+1)> 0 \}$.
The {\em Fourier-Helgason transforms}~$ {\mathcal F}_\pm $  are 
defined~\cite[Eq.~(44), see also Definition 2]{BM2} by
\label{Fouriertransformpage}
	\begin{equation}
	\label{ftps}
 	 {\mathcal D} (dS)  \ni f \mapsto	\widetilde {f}_\pm (  p, s) 
		= \int_{dS} {\rm d} \mu_{dS} ( x ) \; f( x ) \; (  x_\pm \cdot  p )^s \; . 
	\end{equation}
\end{definition}

For $p$ fixed, the functions  $\widetilde {f}_\pm ( p , \, . \,)$
are holomorphic with respect\footnote{Note that a function  analytic in the  strip $- 1 < \Re s < 0$ 
is uniquely determined by its values  
on one of the two symmetry axis given in \eqref{s1} and \eqref{s2}.} to $s$ in the strip $- 1 < \Re s < 0$ 
[Bros und Moschella~\cite{BM2}, Prop.~8.a]. 

\begin{lemma} 
\label{nu-analyicity}
The function
\[
\nu \mapsto \widetilde {f}_\pm (  p, -\tfrac{1}{2} - i \nu ) 
\]
is analytic in the open strip $\{ \nu \in \mathbb{C} \mid | \Im \nu | < \tfrac{1}{2} \}$.
\end{lemma}

For $s$ fixed, the two functions $\widetilde {f}_\pm ( \, . \, , s)$ are
continuous, homogeneous functions of 
degree $s$ on~$\partial V^+$. 
Together with \eqref{casimir2} this implies that 
$\widetilde {f}_\pm (\, . \, , s)$ is an eigenfunction 
of the Casimir operator $M^2$ on $\partial V^+$ iff $s$ lies on    
\begin{itemize}
\item[(i)]  the symmetry axis 
	\begin{equation}
	\label{s1} 
		s = -1/2 \mp i \nu \; , \qquad \nu = \sqrt{\mu^2 r^2 - \tfrac{1}{4}} = m r\in \mathbb{R}^+_0 \; , 
	\end{equation}
of the strip $- 1 < \Re s < 0$. This choice corresponds to  $\mu^2  = \tfrac{1}{4r^2}+ m^2 \ge \tfrac{1}{4r^2}$, 
\emph{i.e.}, to a \emph{bare mass} $m  \ge  0$.;
\item[(ii)] the symmetry axis 
	\begin{equation}
		\label{s2} 
		s = -1/2 \mp i  \nu \; , \qquad  \nu = i \sqrt{\tfrac{1}{4} - \mu^2 r^2} = i mr  \; , 
	\end{equation}
of the strip $- 1 < \Re s < 0$. This choice corresponds to  $0 < \mu^2  \le \tfrac{1}{4r^2} $, \emph{i.e.}, to 
a \emph{negative bare mass}\index{bare mass} $-\tfrac{1}{2r} < m \le 0$.
\end{itemize}
Thus  the critical mass~$\mu_c$, 
which separates the two cases, is $\mu_c=\sqrt{-s (s+1)}= (2r)^{-1}$. 
Note that the factor $(2r)^{-1}$ may be interpreted as a contribution to the mass coming from the 
curvature of space-time (see, \emph{e.g.}, \cite{GaT}).

\begin{remark}
\label{splitting}
Taking advantage of \eqref{eqPW-b},  the Fourier-Helgason transforms~$ {\mathcal F}_\pm$ can  be written in the following form
(see \cite[Eq.~(50)]{BM2})
	\begin{align*}
 	\widetilde {f}_\pm (  p, -1/2 - i \nu )
		& = \int_{ \{ x \in dS \mid x \cdot p >0 \}} {\rm d} \mu_{dS} ( x ) \; f( x ) \; |  x \cdot  p |^{-1/2 - i \nu } \\
		& \qquad    + {\rm e}^{\mp i \pi s^+}  
		\int_{ \{ x \in dS \mid x \cdot (-p) >0 \}} {\rm d} \mu_{dS} ( x ) \; f( x ) 
		\; |  x \cdot  p |^{-1/2 - i \nu } \; . 
	\end{align*}
This identity is valid in the open strip $\{ \nu \in \mathbb{C} \mid | \Im \nu | < 1/2 \}$. The second term can be viewed as 
a continuous, homogeneous function of degree $s^+$ on~$\partial V^-$. 
\end{remark}

\section{The Plancherel theorem on the hyperboloid}

\label{hardyspacepage}
Denote by $H^2 ({\mathcal T}_+)$, 
$H^2 ({\mathcal T}_-)$, $H^2 ({\mathcal T}_\leftarrow)$ and 
$H^2 ({\mathcal T}_\rightarrow)$ the Hardy spaces of functions~$F(z)$ 
characterised by the following properties~\cite[Sect.~3.2]{BM2}\cite[Sect.~3.3]{Ne}: 
\begin{itemize}
\item [$i.)$] $F$ is holomorphic in the tuboid considered; 
\item[$ii.)$] the function $F(z)$ admits  boundary values $f(x)$ on $dS$ from this tuboid, 
which belong to $L^2 (dS, {\rm d} \mu_{dS})$;  
\item[$iii.)$] $F$ is `sufficiently regular at infinity in its domain' (in the sense made precise 
in \cite[p.~10]{BM2}).
\end{itemize}

\begin{theorem}[Bros \& Moschella~\cite{BM2}, Theorem 1]
\label{hardy}
Any given function $f \in L^2 (dS, {\rm d} \mu_{dS})$ admits a decomposition of the form
\label{hardyspacedecomposition}
	\begin{equation}
		\label{decompose}
 		f = f_+ + f_- + f_\leftarrow + f_\rightarrow \equiv \sum_{tub} f_{(tub)} \; , 
		\qquad (tub) = +, -, \leftarrow, \rightarrow  \; ,  
	\end{equation}
where $f_{(tub)} ( x ) \in L^2 (dS, {\rm d} \mu_{dS})$ is the boundary value of the function 
	\begin{equation}
		\label{ck}
 		F_{(tub)} ( z) 
		= \epsilon_{(tub)} \frac{1}{\pi^2} \int_{dS} {\rm d} 
		\mu_{dS} ( x) \; \frac{f( x)}{{( x -  z ) \cdot ( x -  z )} } \in H^2 ({\mathcal T}_{(tub)}) \; .
	\end{equation}
The sign function $\epsilon_{(tub)}$ takes the value $-1$ for  ${\mathcal T}_\pm$,  and
$+1$ for ${\mathcal T}_\leftarrow$ and~${\mathcal T}_\rightarrow$.
\end{theorem}

\begin{remark}In Minkowski space $\mathbb{R}^{1+1}$, a similar decomposition can 
be gained by simply dividing the support of the Fourier transform $\widetilde {f}$ 
into the four cones $\{ (E, p ) \in \mathbb{R}^{1+1} \mid \pm E > |p|\}$ and 
$\{ (E, p ) \in \mathbb{R}^{1+1} \mid \pm p > |E|\}$. Note that for $m>0$ the boundary sets 
$\{ (E, p ) \in \mathbb{R}^{1+1} \mid \pm p = E \}$ are of measure zero. The inverse Fourier 
transform of each of these functions is then the boundary  of a function analytic 
in a tube. For the first two, the tube is $T^\pm = \mathbb{R}^2 \pm i \{ (x_0, x_1 ) 
\in \mathbb{R}^{1+1} \mid \pm x_0 > |x_1|\}$. 
The situation is similar for the two other cases. 
\end{remark}

The Cauchy kernel on $dS_\mathbb{C}$ introduced in  \eqref{ck} is~\cite[Proposition~11]{BM2}
	\begin{equation}
		\label{ac}
		\frac{1}{( z' -  z ) \cdot ( z' -  z )} 
		= - \frac{\pi^2}{2} \int_0^\infty {\rm d} \mu_\pm (\nu)
			\int_\Gamma {\rm d} \mu_\Gamma (p )\;  
			( z \cdot p)^{- \frac{1}{2} + i \nu}   ( p \cdot  z')^{- \frac{1}{2} - i \nu} \; .  
	\end{equation}
The integral in \eqref{ac}  is absolutely convergent for 
$( z,  z') \in {\mathcal T}_+ \times {\mathcal T}_-$ for~${\rm d} \mu_+$ and for 
$( z,  z') \in {\mathcal T}_- \times {\mathcal T}_+ $ for~${\rm d} \mu_-$, respectively. 
The measure ${\rm d} \mu_\pm (\nu)$ on $\mathbb{R}^+$ is 
 {\rm (see~\cite[Sect.~4.1]{BM2})}
	\[  
		\label{dmu}
		{\rm d} \mu_\pm (\nu) 
		= \frac{1}{2\pi^2} \; \frac{\nu \tanh \pi \nu}{ {\rm e}^{\pm \pi \nu} \cosh \pi \nu }  {\rm d} \nu  \; . 
	\]
Combine \eqref{ck}, \eqref{ac} and \eqref{ftps} to find the {\em inversion formula} \cite[Eq.~(80)]{BM2}
	\begin{equation}
	\label{inverse}
		F_\pm (  z) = - \int_0^\infty {\rm d} \mu_\pm (\nu) 
		\int_\Gamma {\rm d} \mu_\Gamma ( p )\;  
			( z \cdot  p)^{- \frac{1}{2} + i \nu} \;  
			\widetilde {f}_\pm \left(  p, - {\textstyle \frac{1}{2} } - i \nu \right) \; .
	\end{equation}
The functions $f_\pm (  x) $ introduced in \eqref{decompose} now appear as  boundary values of the 
holomorphic functions $F_\pm (  z)$, $z \in {\mathcal T}_\pm$. 

\begin{remark}
For every function $F_\pm$ in the Hardy space~$H^2 ({\mathcal T}_\pm)$ the  transform 
$\widetilde {f}_\pm \left(  p, - {\textstyle \frac{1}{2} } - i \nu \right)$ 
vanishes \cite[Proposition~8]{BM2}. This follows from the analyticity properties stated in 
Proposition~\ref{prop:4.1}. A similar result holds true in the Minkowski space-time:
The functions $f$ on $\mathbb{R}^{1+d}$, which 
are boundary values of holomorphic functions in either tube ${\mathfrak T}_\pm = \mathbb{R}^{1+d} \mp i V^+ $
are exactly the functions characterised by the fact that their Fourier transforms 
	\[
	\widetilde f (k) = \frac{1}{(2\pi)^{\frac{1+d}{2}}}\int_{\mathbb{R}^{1+d}} {\rm d} y \; f(y) \; {\rm e}^{i   k \cdot y} \; , 
	\qquad f \in {\mathcal D} (\mathbb{R}^{1+d}) \; , 
	\]
have their support contained in the closure of either one of the 
cones~$ V^\pm$; see, \emph{e.g.},~\cite[Ch.~8]{Schwartz}. 
\end{remark}

\begin{theorem}[Molchanov \cite{Mo1}]
\label{plancherel}
For any pair of functions $f, g$ in $L^2(dS, {\rm d} \mu_{dS})$ and their corresponding 
decomposition given in \eqref{decompose}, one has the {\em Plancherel theorem}\footnote{These 
are the Eq.~(118) and Eq.~(119) in \cite{BM2}.}:
	\begin{align}
		\label{planch}
 		& \int_{dS} {\rm d} \mu_{dS} ( x)\;  \overline{f_\pm( x )} g_\pm (  x )  \\
		& \qquad  = 
		\int_0^\infty {\rm d} \mu_\pm (\nu)
		\int_\Gamma {\rm d} \mu_\Gamma ( p ) \; \overline{\widetilde {f}_\pm 
		\left(  p, - {\textstyle \frac{1}{2} } - i \nu \right) } 
		\; \widetilde {g}_\pm \left(  p, - {\textstyle \frac{1}{2} } - i \nu \right) \; . \nonumber
	\end{align}
The measures ${\rm d} \mu_{dS}$ and ${\rm d} \mu_\Gamma$ denote the Lorentz 
invariant measures on $dS$ and the restriction of the Lorentz invariant measures 
on $\partial V^+$ to $\Gamma$.  
\end{theorem}

\part{Classical Fields}

\chapter{Classical Field Theory}
\setcounter{equation}{0}

In this chapter we consider the classical {\em Lagrangian density}
	\[
		{\mathcal L}(\mathbb{\Phi}) = \frac{1}{2}  \bigl( \nabla \mathbb{\Phi} \cdot \nabla \mathbb{\Phi} 
		- \mu^2 \mathbb{\Phi}^2  - P( \mathbb{\Phi}) \bigr)\; .  
	\]
Here $\mathbb{\Phi}$ is a real valued scalar field, $\nabla$ denotes the Levi-Civita connection on $dS$.
The  polynomial $P$ is bounded from below and $\mu >0$. 

\section{The classical equations of motion}

Let ${\rm K}$ be a compact submanifold of $dS$. The action associated to ${\mathcal L}$ and $K$ is 
	\begin{align*}
	S({\rm K}, \mathbb{\Phi}) & = \int_{\rm K} {\rm d} \mu_{dS} (x)\;  {\mathcal L}(\mathbb{\Phi}(x)) \\
	& = \frac{1}{2}
	\int_{\rm K}  {\rm d}^2 x  \; \sqrt{|g|}  \; \Bigl(  g^{\mu \nu}\partial_ \mu \mathbb{\Phi}(x) \partial_ \nu \mathbb{\Phi}(x)
	- \mu^2 \mathbb{\Phi}^2 - P( \mathbb{\Phi}(x))\Bigr)
	\; . 
	\end{align*}
The (inhomogeneous) Klein--Gordon equation is obtained by demanding that for every such ${\rm K}$,
the action $S({\rm K}, \mathbb{\Phi}) $ is stationary with respect to smooth variations $\mathbb{\Phi} \mapsto \mathbb{\Phi} + \delta 
\mathbb{\Phi} $ of $\mathbb{\Phi} $ that vanish on the boundary $\partial K$ of $K$.  In other words, we require that
\[
0 = \frac{\delta S({\rm K}, \mathbb{\Phi})}{\delta \mathbb{\Phi}(y) } 
=  \int_{\rm K} {\rm d} \mu_{dS} (x) \Bigl( \frac{\partial {\mathcal L}(\mathbb{\Phi}(x))}{\partial \mathbb{\Phi}(y)} 
+ \frac{\partial {\mathcal L}(\mathbb{\Phi}(x))}
{\partial (\partial_\mu \mathbb{\Phi}(x))} \frac{\delta(\partial_\mu \mathbb{\Phi}(x))}{\delta\mathbb{\Phi}(y) } \Bigr)
\]
for every such ${\rm K}$. 
The resulting Euler-Lagrange equation
\[
\partial_\mu \frac{\partial {\mathcal L}}{\partial (\partial_\mu \mathbb{\Phi})} -  \frac{\partial {\mathcal L}}{\partial \mathbb{\Phi}}  = 0 
\]
is 
the \emph{equation of motion}\index{equation of motion}\footnote{In local coordinates, the Laplace-Beltrami operator 
$\square_{dS}$ equals $ |g|^{-1/2}\partial_\mu g^{\mu\nu}|g|^{1/2}\partial_\nu $,  with $ |g|\equiv | {\rm det} g| $. } 
\[
\partial_\mu \Bigl( \sqrt{|g|} \, g^{\mu \nu} \partial_\nu  \mathbb{\Phi}  \Bigr) +  \sqrt{|g|} \Bigl( \mu^2 \mathbb{\Phi} + P'  (\mathbb{\Phi} ) \Bigr) = 0 
\]
on $dS$. In a more compact notation, this equation is rewritten as 
\label{squarepage}
	\begin{equation}
		\label{3.25}
		(\square_{dS}+\mu^2) \mathbb{\Phi}  = - P'  (\mathbb{\Phi} ) \; , 
		\qquad \mathbb{\Phi} \in C^\infty (dS) \;  ,   \quad \mu>0 \; ,
	\end{equation}
where $\mu$ is a mass\footnote{How $\mu$ is related (or identical) to 
the physically observable mass of a particle on de Sitter space  will be 
discussed elsewhere. See \cite{GaT} 
for a discussion of several interpretations of $\mu$ found in the literature.} parameter. In the sequel we keep $\mu>0$ fixed, and 
although almost all quantities we encounter depend on $\mu$, we will suppress this dependence on $\mu$ in the notation. 

\section{Conservation Laws}
\label{SET}

\color{black}
The advantage of the Lagrangian formulation is that any one-parameter subgroup, which leaves the Lagrangian density invariant, 
gives rise to a conservation law\footnote{In \cite[p.~269]{FHN}, 
the authors have chosen $K= W_1$, and thus the action $S({\rm K}, \mathbb{\Phi}) $ yields
	\[
		S(W_1, {\mathbb \Phi}) = \frac{1}{2} 
		\int_{W_1} r^2 \cos \psi \, {\rm d} \psi  \, {\rm d} t  \; 
		\Bigl(  r^{-2} \cos^{-2} \psi \bigl( \tfrac{ \partial {\mathbb \Phi} }{\partial t}\bigr)^2 - 
		r^{-2} \bigl( \tfrac{ \partial {\mathbb \Phi} }{\partial \psi}\bigr)^2 - \mu^2 {\mathbb \Phi}^2 
		- P (\mathbb{\Phi}) \Bigr) \; .
	\]
The invariance with respect to translations of $t$ yields the conserved quantity 
\begin{align*}
L_{1 \upharpoonright I_+ } 
& =  
\int_{I_+}  r \cos \psi \, {\rm d} \psi \;   T_{00} \; . 
\end{align*}
In the last equation we have used  $n= r^{-1} \cos^{-1}\psi \partial_t$ and $n \mathbb{\phi} = \mathbb{\pi}$. 
Here $n$ denotes unit normal, future pointing vector field, restricted to the Cauchy surface~${\mathcal C}$.}. 
If the 2-form ${\mathcal L}$ depends on the scalar $\mathbb{\Phi}$, then its variation 
	\[
		\delta {\mathcal L} =  \delta \mathbb{\Phi} \wedge 
		\left[ \tfrac{\partial {\mathcal L} }{ \partial \mathbb{\Phi}}  -  d \tfrac{\partial {\mathcal L} }{ \partial (d \mathbb{\Phi}) } \right] 
		+ d \Bigl(  \delta\mathbb{\Phi} \wedge  \tfrac{\partial {\mathcal L} }{ \partial (d \mathbb{\Phi}) } \Bigr) \; , 
	\]
and the equations of motion  
	\[
		\tfrac{\partial {\mathcal L} }{ \partial \mathbb{\Phi}}  -  d \tfrac{\partial {\mathcal L} }{ \partial (d \mathbb{\Phi}) } = 0 \; . 
	\]
imply that 
	\[
		\delta {\mathcal L} = 0 \quad \Rightarrow \quad
		d \Bigl(  \delta\mathbb{\Phi} \wedge  \tfrac{\partial {\mathcal L} }{ \partial (d \mathbb{\Phi}) } \Bigr) = 0 \; . 
	\]
If the variation results from a Lie derivative $L_v$, with $v$ some vector field, then 
	\[
		\delta {\mathcal L} = L_v {\mathcal L} = d ( i_v {\mathcal L}) 
	\]
as the exterior derivative of the 2-form vanishes in two-dimensional space. It follows that 
	\[
		\sum_j \delta\mathbb{\Phi} \wedge  \frac{\partial {\mathcal L} }{ \partial (d \mathbb{\Phi}) }
	\]
is a closed 1-form, or, using the Hodge $*$,  a conserved current. 

\begin{theorem}[Noether]
\label{noether} 
Let $v \in E_1$ with $\delta \Phi = L_v \Phi$ and $\delta {\mathcal L} = L_v {\mathcal L}$. It follows that
	\[
		d \left[  L_v \mathbb{\Phi} \wedge \tfrac{\partial {\mathcal L}}{\partial (d \mathbb{\Phi})} 
		\right] =: - d * T_v = 0 \; . 
	\]
\end{theorem} 

\begin{remark}
If $L_v$ is a translation (\emph{i.e.}, $v = dx_\mu$, $\mu= 0, 1, \ldots , d $), 
on Min\-kowski space $\mathbb{R}^{1+d}$, then $T^\mu$ (as defined in Theorem \ref{noether}) is called the energy-momen\-tum current.
Its components with respect to a basis define the \emph{energy-momentum tensor}~${T^\mu}_\nu$:
	\[
		{T^\mu} = {T^\mu}_\nu dx^\nu \; . 
	\]
If one integrates over a space-like surface ${N_d}$, than one finds the energy $P_0$ and the momentum $P^j$:
	\[
		P^ \mu = \int_{N_d} * {T^\mu} \; . 
	\]
$T_{\mu \nu}$ describes the flux of the $\mu$-th component of the conserved energy-momentum vector across a surface 
with constant $x_\nu$ coordinate (see, \emph{e.g.}, \cite[p.~35]{Thirring}). 
\end{remark}

\bigskip
A convenient basis to derive explicit expressions for the stress-energy tensor on de Sitter space is the following:
	\[
		\left(\begin{matrix} 
			x_0 	\\
			x_1	\\
			x_2
		\end{matrix}\right) 
		= 
		\left(\begin{matrix} 
			x_0 \\ 
			\sqrt{ r^2 + x_0^2 } \sin \psi \\ 
			\sqrt{ r^2 + x_0^2 } \cos \psi
		\end{matrix}\right)\; , \qquad x_0 \in \mathbb{R} \; , \quad \psi \in [0, 2 \pi) \; . 
	\] 
The metric takes the form $g= {\rm d} x_0 \otimes {\rm d}  x_0 -  \sqrt{r^2 + x_0^2}  {\rm d} \psi \otimes {\rm d} \psi $ and
the stress-energy tensor is given by
\begin{align*}
	{T^{\mu}}_{\nu} & = \partial^\mu \mathbb{\Phi} \partial_\nu \mathbb{\Phi}  -  g^{\mu\kappa} 
	g_{\kappa \nu}{\mathcal L}(\mathbb{\Phi}) \\
        &= \partial^\mu \mathbb{\Phi} \partial_\nu \mathbb{\Phi} - \tfrac{1}{2} {\delta^{\mu}}_{\nu} 
	\bigl(  \partial^\kappa \mathbb{\Phi}\partial_\kappa \mathbb{\Phi} \bigr)
	+ \tfrac{1}{2} {\delta^{\mu}}_{\nu}  \bigl(\mu^ 2 \mathbb{\Phi}^2  + P (\mathbb{\Phi}) \bigr) 	\; , \qquad \mu, \nu = x_0, \psi \; . 
\end{align*}
In particular, 
	\begin{align*}
		T_{00} &= \frac{1}{2} \left( \mathbb{\pi}^2 +   r^{-2} \bigl( \partial_\psi  \mathbb{\Phi} \bigr)^2 + 
					\mu^2 \mathbb{\Phi}^2 + P  (\mathbb{\Phi} ) \right) \; ,  
	\end{align*}
with $\mathbb{\pi} = \tfrac{\partial}{\partial x_0} \mathbb{\Phi}$, and 
	\begin{equation}
	\label{T01}
		T_{0\psi} =  \frac{1}{r} \partial_{x_0} \mathbb{\Phi}\, \partial_\psi \mathbb{\Phi} \; .  
	\end{equation}

The Killing vector fields on $dS$ are given by $\partial_t $ (within the double cone $\mathbb{W}_1$, using the coordinates 
introduced in \eqref{w1psi}) and
$\partial_\psi $. The corresponding conserved quantities 
	\[
		L^1 = \int_{S^1} *  T^t \quad \text{and} \quad
		K_0 = \int_{S^1} * T^\psi  		 
	\]
generate the $\Lambda_1$-boosts 
and the rotations around the $x_0$-axis, respectively.  Rewriting $\partial_t\mathbb{\Phi}$ as 
	\[
		\partial_t\mathbb{\Phi}= r\cos\psi \, n \mathbb{\Phi} \equiv r\cos\psi \, \mathbb{\pi} \; , 
	\]
where $n$ is the future directed normal vector field to the time-circle $x_0 =0$, we have
\footnote{Using the coordinates introduced in \eqref{w1psi} we have 
$g=r^2 ( \cos^2 \psi \, {\rm d} t\otimes {\rm d} t  - {\rm d}\psi \otimes {\rm d}\psi )$ and $ |g|^{1/2}= r^2 |\cos\psi| $.}
	\begin{align}
		L^1 = \tfrac{1}{2}\,\int_{S^1} r \, \cos\psi \; {\rm d} \psi \; \big( 
							\mathbb{\pi}^2 + r^{-2} (\partial_\psi\mathbb{\Phi})^2 +
							\mu^2\mathbb{\Phi}^2 + P(\mathbb{\Phi})\big) \; . 
	\end{align}
Using \eqref{T01}
yields the formula for the angular momentum   
	\[ 
		K_0 = \int_{S^1} \,  r^2 \,  |\cos\psi| \; {\rm d} \psi \;  T_{0\psi}
		=  \int_{S^1} \,  r \, {\rm d} \psi \;   \mathbb{\pi} \, (\partial_\psi \mathbb{\Phi})  \; .
	\]

\begin{remark} 
\label{classical-energy}
Integrating $T_{00}$ over the time-zero circle $S^1$ yields a positive quantity, 
	\[
		\int_{S^1} r\, {\rm d} \psi \; 	T_{00} (\psi) > 0 \; ,  
	\]
which may be interpreted as the energy density for 
the classical $P(\varphi)_2$ model on the \emph{Einstein universe} (see, \emph{e.g.}, \cite{Fe1}\cite{Fe2}), \emph{i.e.}, 
the space-time of the form $S^1 \times \mathbb{R}$. 
\end{remark}

Although there are interesting results concerning the  non-linear Klein-Gordon equation in two space-time dimensions
(see, \emph{e.g.}, \cite{Delort-1, Delort-2, Delort-3, Delort-4, Delort-5, GV, Lebeau}), we will concentrate on free fields for the rest of this chapter. 

\section{The covariant classical dynamical system}
\label{CCDS}

As mentioned in Section \ref{geodesics}, the de Sitter space-time $dS$ is  globally hyperbolic. Thus the
inhomogeneous Klein--Gordon equation
	\begin{equation}
		 \label{4.3n}
			 (\square_{dS}+\mu^2)  \mathbb{\Phi}   
		 =  f \; , \qquad f \in{\mathcal D}_\mathbb{R} (dS) \; , 
	\end{equation}
has smooth solutions, which are uniquely specified by fixing their support properties  
(see \cite{C-B,D77, L, Lich}):

\begin{theorem}
\label{fundamental}
There exist unique operators 
	\[
		\mathbb{E}^{\pm} \colon{\mathcal D}_{\mathbb{R}} (dS) \to C^\infty (dS) 
	\]
such that  $\mathbb{E}^\pm f $ is a solution of the inhomogenous equation \eqref{4.3n} with
	\[
		{\rm supp\,}  (\mathbb{E}^\pm f)   \subset \Gamma^\pm ( {\rm supp\,} f) \quad  \text{and} \quad 
		{\rm supp\,} ( \mathbb{E}^\pm f ) \cap \Gamma^\mp ( {\rm supp\,} f) \;  \text{compact}. 
	\] 
\end{theorem}

$\mathbb{E}^{\pm}f$ are called the \emph{retarded}\index{retarded solution} and the 
\emph{advanced solution}\index{advanced solution}, respectively. 
The difference between the retarded and the advanced solution of the inhomogeneous equation \eqref{4.3n}, 
namely 
	\begin{equation}
	\label{Ef}
	\mathbb{\Phi}=\mathbb{E} f \; , \qquad \text{with} \quad  \mathbb{E} = \mathbb{E}^{+} - \mathbb{E}^{-} \; , 
	\end{equation}
is a solution of the homogenous Klein--Gordon equation \eqref{3.25}.

\begin{remark}
For comparison, we briefly recall the situation on Minkowski space $\mathbb{R}^{1+1}$.
After Fourier transformation, the inhomogeneous equation
	\[
		(\square_{\mathbb{R}^{1+1}} + m^2)G = - \delta \; , 
	\]
takes the simple form $(- {\tt P}_0^2 + {\tt P}_1^2 + m^2) \widetilde {G} = -1$; the latter has the retarded and advanced 
propagators as its solution: 
	\[
		{\mathscr E}_{adv}( {\tt x}, {\tt y} ) = \lim_{\epsilon \downarrow 0} \frac{1}{(2 \pi)^2} 
			\int {\rm d}^2 p \frac{ {\rm e}^{- i {\tt p} \cdot ({\tt x}- {\tt y})}}
			{({\tt p}_0 - i \epsilon)^2 - {\tt p}_1^2 - m^2} \; , 
	\]
and
	\[
		{\mathscr E}_{ret}( {\tt x}, {\tt y} ) = \lim_{\epsilon \downarrow 0} \frac{1}{(2 \pi)^2} 
			\int {\rm d}^2 p \frac{ {\rm e}^{- i {\tt p} \cdot ({\tt x}- {\tt y})}}
			{({\tt p}_0 + i \epsilon)^2 - {\tt p}_1^2 - m^2} \; , 
	\]
The difference ${\mathscr E}( {\tt x}, {\tt y} ) \doteq {\mathscr E}_{ret}( {\tt x}, {\tt y} ) - {\mathscr E}_{adv}( {\tt x}, {\tt y} )$ 
is a bi-solution of the Klein-Gordon equation. 	Note that if 
	\[
	{\tt p} \cdot ({\tt x}- {\tt y})=0
	\]
the integral is defined in distributional sense only. We will soon encounter a similar problem on de Sitter space. 
\end{remark}

\bigskip
As we will see next, {\em any} smooth solution of \eqref{3.25} is of the type \eqref{Ef}.

\goodbreak
\begin{theorem}[B\"ar, Ginoux and Pf\"affle \cite {BGP},Theorem 3.4.7]
\label{solutions}
\quad
\begin{itemize}
\item[$ i.)$]Any smooth solution $\mathbb{\Phi}$ of the Klein--Gordon equation \eqref{3.25} 
may be written in the form 
	\[
	\mathbb{\Phi} = \mathbb{E} f \; , \qquad \text {for some} \; f \in {\mathcal D}_{\mathbb{R}}(dS) \; ; 
	\]
and, given any neighbourhood \footnote{In \cite{BMJ}
we will demonstrate that for the $P(\varphi)_2$ 
model on the de Sitter space, the expectation values of all observables can 
be predicted from the expectation values of observables, 
which can be measured within an {\em arbitrarily small} time interval. Thus 
the $P(\varphi)_2$ model on the de Sitter space satisfies the Time-Slice
Axiom~\cite{CF}.}  ${\mathscr N}$  of a Cauchy surface ${\mathcal C}$, one may choose such an $f \in {\mathcal D}({\mathscr N})$. 
\item[$ ii.)$] We have 
	\[
		\ker \mathbb{E} = (\square_{dS}+\mu^2) {\mathcal D}_{\mathbb{R}} (dS) \; . 
	\]
In consequence, the {\em space of smooth real-valued solutions} equals 
	\begin{equation*}
	\mathbb{E}  {\mathcal D}_{\mathbb{R}} (dS) \cong {\mathcal D}_{\mathbb{R}} (dS) / 
		(\square_{dS}+\mu^2) {\mathcal D}_{\mathbb{R}} (dS) \doteq {\mathfrak k}(dS) \; .  
	\end{equation*}
\end{itemize}
\end{theorem}

Taking advantage of the properties $ i.)$ and $ ii.)$, we 
can define a projection
	\begin{align*}
			\mathbb{P} \colon {\mathcal D}_\mathbb{R}(dS) 
					& \to   {\mathfrak k} (dS) \nonumber \\
			 f & \mapsto   [f] \; . 
	\end{align*}
This yields a one-to-one correspondence between 
	\[ 
	{\mathfrak k} (dS) \ni [f] \longleftrightarrow \mathbb{f} \in \mathbb{E}  {\mathcal D}_{\mathbb{R}} (dS)  \; , 
	\]
with $[f]  \doteq \{ f + (\square_{dS}+\mu^2)h \mid h \in {\mathcal D}_{\mathbb{R}}(dS) \}$
and $\mathbb{f} \doteq \mathbb{E} f$. 

\begin{definition}
Subspaces of $ {\mathfrak k} (dS) $ associated to open space-time regions ${\mathcal O} \subset dS$ are defined by restricting $\mathbb{P}$
to ${\mathcal D}_\mathbb{R}({\mathcal O})$, \emph{i.e.}, 
	\[
			{\mathfrak k} ({\mathcal O}) 
			\doteq{\mathcal D}_\mathbb{R}({\mathcal O})/ \ker \mathbb{E}  \;. 
	\]
\end{definition}

${\mathfrak k} ({\mathcal O})$ will be used in Chapter \ref{2Q} to define 
{\em local} von Neumann algebras.

\bigskip
Embedding $ C^\infty (dS)$ into ${\mathcal D}'(dS)$ 
(see \cite[Sect.~2]{Fe1}\cite{Fe2}), the map~$\mathbb{E}$ 
gives rise to a bi\-distribution ${\mathcal E}$ on $dS\times dS$
	\begin{align} 
		\label{integralkernel}
			{\mathcal E}  ( f,g) \; &
			\doteq \int {\rm d} \mu_{dS} ( x) \;  f (x)( \mathbb{E} g)(x) 
			\nonumber \\
				&\doteq 
			\int {\rm d} \mu_{dS} ( x) {\rm d} \mu_{dS} ( y) \; 
			f ( x ) {\mathscr E} ( x ,  y )  
			g ( y  )  \; ,  
	\end{align}
antisymmetric in $f, g \in{\mathcal D}_\mathbb{R} 
(dS)$, whose kernel ${\mathscr E} ( x ,  y ) $, called the \emph{fundamental solution}\index{fundamental solution}, 
is a weak bisolution for the Klein--Gordon equation, 
\label{ccCfpage}
\label{vnpage}
	\begin{equation}
		 \label{cmzero} 
		  {\mathcal E}  \left( (\square_{dS}+\mu^2)
		  f,g \right)  
		  = {\mathcal E}  \left( f, 
		  (\square_{dS}+\mu^2)g \right) = 0 \; , 
	\end{equation}
with initial data\footnote{Micro-local analysis shows  that ${\mathscr E}$ and
its normal derivatives can be restricted to 
${{\mathcal C}} \times {{\mathcal C}}$, see~\cite{Hoerm}.}
	\begin{align}
		{{\mathscr E}}_{ \upharpoonright {\mathcal C} \times {\mathcal C} } 
				& =  0 \; , \label{eqESigma} \\  
				(n_\ell {\mathscr E})_{ \upharpoonright {{\mathcal C}} 
				\times {{\mathcal C}}}
				&=  -\delta_{{\mathcal C}} \; .  
				\label{eqESigma2}				
	\end{align}
Here $n_\ell $ denotes the vector field $n$ acting on  
the left variable $ x$ in ${\mathscr E} ( x,  y)$
and $\delta_{{\mathcal C}}$ is the integral kernel of the unit operator with 
respect to the 
induced measure on ${{\mathcal C}}$. 
The map 
\[
{\mathcal D}(dS) \ni f \mapsto \mathbb{f} = \mathbb{E} f 
\]
can now be viewed as  a convolution\footnote{On Minkowski space,  Fourier transformation
converts a convolution 
in position space to a multiplication in momentum space.
For the situation on $dS$, see Section \ref{Harmanaly}.} of a test function $f$ 
with the kernel ${\mathscr E}$, \emph{i.e.}, 
	\begin{equation}
		\label{feb-1}
			\mathbb{f} ( x ) 
				\doteq \int {\rm d} \mu_{dS} ( y ) \; {\mathscr E}( x ,  y ) f ( y ) \; , 
			\qquad  f \in{\mathcal D}_\mathbb{R} (dS) \; .  
	\end{equation}
Eq.~(\ref{feb-1}) implies  that $\mathbb{f}  ( x) =0$ for all $ x \in dS$,  iff  
	\begin{equation}
		 \label{eqEf}
		f \in 
		\ker {\mathcal E}  
		\doteq \left\{ f \in{\mathcal D}_\mathbb{R} (dS) 
		\mid  {\mathcal E}  (g, f)= 0 \; \; 
		\forall g \in{\mathcal D}_\mathbb{R} (dS) \right\}\;  .
	\end{equation}
In other words, $\ker \mathbb{E} = \ker {\mathcal E} $.  
Consequently, the 
bidistribution ${\mathcal E}$ provides a non-degenerated symplectic form  
$\sigma$ on the space of solutions ${\mathfrak k} (dS)$:
\label{sigmapage}
	\begin{equation}
	 	\label{eqSplForm}
		\qquad \qquad  \qquad \qquad 
		\sigma \bigl( [f] ,  [g] \bigr)
		\doteq {\mathcal E} (f,g) \; , 
		\qquad  f, g 
		\in{\mathcal D}_\mathbb{R}(dS)\; . 
	\end{equation}
As a consequence of \eqref{cmzero}, 
the right hand side does not dependent on the choice of the representatives
in the equivalence classes $[f]$ and $[g]$.  
Thus $({\mathfrak k} (dS), \sigma)$ is a the symplectic vector space.

\goodbreak
\begin{lemma} 
\label{Lm3.8}
Let $f \in{\mathcal D}_\mathbb{R}({\mathcal O})$, 
${\mathcal O}\subset dS$ an open region. Then 
$\mathbb{f} = \mathbb{E} f$ is a solution of the Klein--Gordon equation with
	\[
		{\rm supp\,} (\mathbb{f} ) 
		\subset\Gamma^+ ({\mathcal O}) 
		\cup \Gamma^- ({\mathcal O}) \; .  
	\]
In particular, if ${\mathcal O} \subset W$, then 
${\rm supp\,} (\mathbb{f} ) 
\subset dS \setminus \overline{\, W'}$.
\end{lemma}

\begin{proof}
The support properties of~$\mathbb{E}^{\pm}$ force 
${\mathcal E} ( f,g)$  to vanish, 
whenever the support of $f$ is space-like separated from that 
of $g$. Thus, for $ y \in dS$ fixed,  the distribution 
$ x \mapsto {\mathscr E} ( x ,  y )$
has support in  $\Gamma^+ ( y) \cup \Gamma^- ( y ) $. 
The final statement follows from  this fact as well.
\end{proof}

Exploring the one-to-one correspondence between 
${\mathfrak k} (dS) \ni [f]$ and $ \mathbb{f} \in \mathbb{E}  {\mathcal D}_{\mathbb{R}} (dS) $, 
this result can be rephrased in the following way.  

	\begin{lemma} 
\label{Lm4-9}
Let  $f \in{\mathcal D}_\mathbb{R}({\mathcal O})$, ${\mathcal O}\subset dS$ a bounded open region, and 
$g \in{\mathcal D}_\mathbb{R}({\mathcal O}')$, where ${\mathcal O}'$ denotes the space-like complement
of ${\mathcal O}$. Then 
	\begin{equation}
	\label{def:sigma}
	 		\sigma \bigl(   [f]  , [g] \bigr)= 0 \; .  
	\end{equation}
\end{lemma}
 
\begin{proposition}
\label{Prop4-10}
The symplectic space $({\mathfrak k}(dS), \sigma)$ carries a representation 
	\[
		\Lambda \mapsto {\mathfrak u} (\Lambda)\; ,  \qquad \Lambda \in O(1,2) \; , 
	\]
of the Lorentz group. 
\label{ccTpage}
\end{proposition}

\begin{proof}
The group of isometries $\Lambda \in O(1,2)$ of $dS$ gives rise to a
group of symplectic transformations $\Lambda \mapsto T_\Lambda$ on $({\mathfrak k} (dS), \sigma)$  
induced by the pull-back~$\Lambda_*$, which maps
\label{lambdasternpage}
	\begin{equation}
		\label{eqTt}
		f+\ker \mathbb{E} \mapsto \Lambda_* f+\ker \mathbb{E}\, \;  . 
	\end{equation}
The map \eqref{eqTt} is well-defined, because $g \in \ker \mathbb{E}$ implies $\Lambda_* g \in \ker \mathbb{E}$.  
\end{proof}

\begin{definition}
\label{Def4-11} The triple  $({\mathfrak k}(dS) , \sigma, {\mathfrak u} )$  is  the 
\emph{covariant classical dynamical system}\index{covariant classical dynamical system}
associated to the Klein--Gordon equation \eqref{3.25}. 
\label{kldypage}
\end{definition}

\section{The restriction of the KG equation to a (double) wedge}

Our next objective is to provide an explicit formula for ${\mathscr E} ( x ,  y )$. In some sense, it is sufficient to solve this problem in the 
causal dependence region of a half-circle: given an arbitrary point $ x \in dS$~and the Cauchy surface $ S^1$, there exists  
a wedge $W^{(\alpha)}= R_0(\alpha)W_1$, which contains both $x$ and $\Gamma^- ( x) \cap S^1$ (or, if this intersection is empty, 
$\Gamma^+ ( x) \cap S^1$). On the other hand,  all the formulas we will derive in this section naturally extend to the double-wedge
$\mathbb{W}^{(\alpha)}=W^{(\alpha)} \cup W^{(\alpha + \pi)}$, so it is natural to state them in their extended form.  

In order to keep the notation simple, we work out explicit expressions for the double wedge $\mathbb{W}_1$ in the chart 
\eqref{w1psi} for $ x   \equiv  x (t,\psi)$ and $ y  \equiv  y (t',\psi')$.  (The points $ \psi = \pm \frac{\pi}{2}$ in this chart 
correspond  to the points  $ (0, \pm r, 0) \in dS$.)
However, we would like to emphasize that all computations in this subsection can be carried out for arbitrary 
double wedges $\Lambda \mathbb{W}_1$, $\Lambda \in SO_0(1,2)$. 

The restriction of the metric~$g$ to $\mathbb{W}_1$ is 
	\[ 
		g_{\upharpoonright {\mathbb{W}}_1} =    r^2 \cos^2 \psi {\rm d} t \otimes {\rm d} t -  r^2 {\rm d} \psi \otimes {\rm d} \psi \;   . 
	\]
The restriction of the Lorentz invariant  measure ${\rm d} \mu_{dS} $ to $\mathbb{W}_1$ is
\label{dlpage}
	\begin{equation}
		\label{new-surface} 
		{\rm d} \mu_{\mathbb{W}_1 } (t,\psi) = r {\rm d} t\, {\rm d} l (\psi) \;  , 
		\qquad \text{with} \quad 
		{\rm d} l (\psi) =  |\cos \psi | \; r {\rm d} \psi  \; . 
	\end{equation}
The line element on the circle $S^1$ is
\label{seinspage}
\begin{equation}\label{new-eqVolInd}
| g_{\upharpoonright S^1} |^{1/2} \, r {\rm d}\psi=r {\rm d}\psi \; .  
\end{equation}
Restricted to the double wedge $\mathbb{W}_1$ the Klein--Gordon operator is
\label{epsilonpage}
	\begin{equation}
		\label{varepsilon}
			\square_{\mathbb{W}_1}+\mu^2
				=  \frac{1}{r^2 \cos^2 \psi}\,(\partial_t^2+  \varepsilon^2) \; , 
	\end{equation}
with 
	\[
		\varepsilon^2  \doteq  - (\cos \psi  \, \partial_\psi)^2 + (\cos \psi )^2 \, \mu^2 r^2 \; . 
	\]

\begin{remark}
For $\psi \in ( -\pi/2,\pi/2 )$ define a new spatial coordinate $\chi=\chi(\psi)$  by
	\begin{equation*}
		\label{eqphichi}
		\frac{{\rm d}\chi}{{\rm d}\psi} = \frac{1}{\cos \psi} \; ,
		\qquad \chi(0)=0 \; . 
\end{equation*}
Find
	$
		\chi(\psi)= \ln \tan (\psi/2+\pi/4) 
	$
and
	\[
			\cos\psi=(\cosh\chi)^{-1},\quad \sin\psi=\tanh\chi \; .
	\]
$\chi$ is a diffeomorphism from $(-\pi/2,\pi/2) $ onto $\mathbb{R}$. 
	\[
		g_{\upharpoonright W_1}
		=   \tfrac{r^2}{\cosh^2 \chi} (  {\rm d} t \otimes {\rm d} t - 
		{\rm d} \chi \otimes {\rm d} \chi) \; .
	\]
Thus $W_1$ is conformally equivalent to Minkowski space $\mathbb{R}^{1+1}$ {\rm \cite{FHN}}. In these coordinates 
	\[
			\square_{\mathbb{W}_1}+\mu^2
				=  \frac{\cosh^2 \chi}{r^2} \,(\partial_t^2+  \epsilon^2) \; , 
	\]
with $\epsilon^2  \doteq  - \partial_\chi^2 + (\cosh \chi )^{-2} \, \mu^2 r^2 $. 
\end{remark}

\begin{lemma}
\label{Asadj}
Identify $S^1 \cong [-\frac{\pi}{2} ,\frac{3\pi}{2}) $,  $I_+  \cong (- \frac{\pi}{2}, \frac{\pi}{2} )$ and 
$I_- \cong (\frac{\pi}{2} ,\frac{3\pi}{2})$.
It follows that $\varepsilon^2$ is positive and symmetric on  
	\[
		{\textstyle{\mathcal D} \left(S^1 \setminus \{ - \frac{\pi}{2}, \frac{\pi}{2}\} \right) }
		\subset L^2(S^1, | \cos \psi |^{-1}  \, r {\rm d} \psi) \; . 
	\]
Denote its Friedrich extension by the same symbol.  Then  ${\rm Sp}
(\varepsilon^2) = [ \, 0, \infty ) $. 
\end{lemma}

\begin{proof} 
Clearly, $\varepsilon^2$ is positive and symmetric on 
	\[
		{\mathcal D} \left(S^1 
		\setminus \{ - \tfrac{\pi}{2}, \tfrac{\pi}{2}\} \right) 
		={\mathcal D} (I_+) \oplus{\mathcal D}(I_-) \; . 
	\]
We next show that ${\mathcal D} \left(S^1 
\setminus \{ - \frac{\pi}{2}, \frac{\pi}{2}\} \right)$ is dense in 
$L^2(S^1, | \cos \psi |^{-1}  \, r {\rm d}\psi)$.
First note that ${\mathcal D} \left(S^1 \setminus \{ - \frac{\pi}{2}, \frac{\pi}{2}\} \right)$ 
is dense in~$L^2(S^1,  \, r {\rm d}\psi)$. Moreover, a function 
	\begin{equation}
		\label{4.9}
		 \,  {\rm cos}_\psi^{1/2}h \in L^2(S^1, | \cos \psi |^{-1}  \, 
		 r {\rm d}\psi) \qquad \text{iff}  \quad 
		 h \in L^2(S^1, \,  r {\rm d}\psi) \; . 
	\end{equation}
It follows that
	\[
	 	\textstyle { \,  {\rm cos}_\psi^{1/2}{\mathcal D} 
	 	\left(S^1 \setminus \{ - \frac{\pi}{2}, \frac{\pi}{2}\} \right) =
		{\mathcal D} \left(S^1 \setminus \{ - \frac{\pi}{2}, \frac{\pi}{2}\} \right) }
	\]
is dense in $L^2(S^1, | \cos \psi |^{-1}  \, r {\rm d}\psi)$. 
Thus \cite[Theorem X.23, p.177]{RS} applies and defines the 
Friedrich extension. 
\end{proof}

\begin{remark}
Since the spectrum of $\varepsilon^2$ has no gap around the discrete eigenvalue zero, 
the choice of coordinates \eqref{w1psi} may  lead to 
artificial infrared problems if one adds an interaction, similar to the ones encountered in~\cite{FHN}. 
We will avoid this problem later on by working with functions in the Hilbert space $\widehat{{\mathfrak h}}  (S^1)$, whose scalar 
product is rotation-invariant;
see Section~\ref{Sect: canon-HS}.
\end{remark}

$\varepsilon^2 $ is a differential operator, thus $\varepsilon^2$ acts locally and maps the subspaces 
	\[
 		{\mathscr D}^\pm \doteq{\mathscr D} (\varepsilon^2) 
		\cap L^2 \bigl(I_\pm,|\cos\psi|^{-1} r {\rm d} \psi \bigr)
	\] 
into $L^2 \left(I_\pm,|\cos\psi|^{-1} r {\rm d} \psi \right)$, respectively. It therefore is consistent to define 
	\begin{equation}
		\label{vaepsdef}
		\varepsilon (h_+ + h_-) 
			\doteq \sqrt{{\varepsilon^2}_{\upharpoonright 	 {\mathscr D}^+}} \; h_+   
			- \sqrt{{\varepsilon^2}_{\upharpoonright 	{\mathscr D}^-}} \; h_-\; , 
				\qquad h_\pm\in {\mathscr D}^\pm\;  . 
	\end{equation}
$\varepsilon$ is densely defined by \eqref{vaepsdef}, as ${\mathscr D}^+ \oplus {\mathscr D}^- 
={\mathscr D}(\varepsilon^2)$. The pseudo-differential operator~$\varepsilon$ is non-local, but  does not mix 
functions supported on the half-circles $I_+$ and~$I_-$. Denote the restrictions by 
$\varepsilon_{\upharpoonright I_+}$ and $\varepsilon_{\upharpoonright I_-}$.

\begin{lemma}
\label{Lm3.5}
There exits a self-adjoint  operator $\varepsilon$ on $L^2(S^1, | \cos \psi |^{-1}  \, r {\rm d}\psi)$  
such that~\eqref{vaepsdef} holds.  ${\rm Sp} (\varepsilon) = \mathbb{R} $, 
${\rm Sp} (\varepsilon_{\upharpoonright I_+}) =  [0, \infty) $ and 
${\rm Sp} (\varepsilon_{\upharpoonright I_-}) =  (-\infty, 0] $. 
Moreover, zero is not an eigenvalue of $\varepsilon$.
\end{lemma}

\begin{proof} Use the spectral theorem to define the square roots in 
\eqref{vaepsdef} as self-adjoint operators. 
${\mathscr D}^+ \cap {\mathscr D}^- =\{0\}$, in fact ${\mathscr D}^+
$ and ${\mathscr D}^-$ are orthogonal to each other in  $L^2(S^1, | \cos \psi |^{-1}  \, r {\rm d}\psi)$.  
Thus  the sum of the square roots is self-adjoint on the direct sum of their domains (see \cite[Theorem VIII.6]{RS})
and ${\rm Sp} (\varepsilon) = \mathbb{R} $.
\end{proof}

\paragraph{\it Notation.} If $P$ is a pseudo-differential  operator on $L^2(S^1, | \cos \psi |^{-1}  \, r {\rm d}\psi)$,    
define its kernel~$P (\psi, \psi') $ for all $h \in L^2(S^1, | \cos \psi |^{-1}  \, r {\rm d} \psi) \cap{\mathscr D} (P)$, 
for which the following expressions exist, by
	\[ 
	(P h) (\psi)  = \int_{S^1} \frac{ r {\rm d}\psi' }{ | \cos \psi' |} \; P (\psi, \psi') h (\psi')   
		 		 = \int_{S^1} \frac{ {\rm d} l (\psi') }{| \cos\psi' |^2}\;  P (\psi, \psi') h (\psi')  \; . 
	\] 
${\rm d} l (\psi')$ was defined in \eqref{new-surface}.

If $P$ is hermitian  with domain  
${\mathscr D} \subset L^2(S^1, | \cos \psi |^{-1}  \, r {\rm d}\psi)$, then 
	\[
		P(\psi, \psi') = \overline{P(\psi', \psi)}\; , \qquad \psi, \psi' \in S^1 \; . 
	\]
In the next lemma, $\left( \frac{\sin(\varepsilon(t-t')) } { | \varepsilon | } \right)$ is considered as such a
pseudo-differential  operator on $L^2(S^1, | \cos \psi |^{-1}  \,  r {\rm d}\psi)$. 

\bigskip

\begin{lemma} \label{PropWedge} Use the coordinates \eqref{w1psi}. Then 
	\begin{equation}
		 \label{eqPropWedge'}
		 	{\mathscr E} (   x,  y )  
		 		= - \left( \frac{\sin(\varepsilon(t-t')) } { | \varepsilon | } \right)  (\psi,\psi')
		\;   . 
	\end{equation}
\end{lemma}

\begin{proof}
For $f, g \in{\mathcal D}_\mathbb{R}(\mathbb{W}_1)$, set $f_t(\psi) 
\doteq f(t,\psi)$ and $g_{t'}(\psi') \doteq g(t',\psi')$. 
Clearly, $f_t, g_t \in L^2(S^1, | \cos \psi |^{-1}  \, r {\rm d}\psi)$.
Consider 
\label{cosinepsipage}
	\begin{equation}
		\label{3.41}
		{\mathcal E}_{{\mathbb W}_1} ( f,g)  
   			 \doteq   - \int  r^3  {\rm d} t \, {\rm d} t' \left\langle \cos^2_\psi {f}_t \,, \,
				\tfrac{\sin(\varepsilon(t-t')) } { | \varepsilon | }\, 
				\cos^2_\psi g_{t'}\right\rangle_{L^2(S^1, | \cos \psi |^{-1}  \, r {\rm d}\psi)} \; ,
	\end{equation}
with  $\,  {\rm cos}_\psi$  the multiplication operator by  $\cos\psi$. 
Clearly,  ${\mathcal E}_{{\mathbb W}_1}$  is anti-symmetric with respect to permutation of $f$ and~$g$. 
Moreover, according to \eqref{varepsilon}  
	\begin{align*} 
	 & {\mathcal E}_{{\mathbb W}_1}\big(f,(\square_{\mathbb{W}_1}+\mu^2) h\big)   \nonumber \\
	 	 &=  {\mathcal E}_{{\mathbb W}_1} \left( f, r^{-2} \cos^{-2}_\psi
		(\partial_t^2+\varepsilon^2) h\right) 
		\label{eqEKGE}\\ 
	&= - \int r^3  {\rm d} t {\rm d} t'\left\langle  \cos^2_\psi {f}_t \,,\,
	\tfrac{\sin(\varepsilon(t-t'))}{ | \varepsilon |}
		r^{-2} (\partial_{t'}^2+\varepsilon^2)h_{t'}\right\rangle_{L^2(S^1, 
		| \cos \psi |^{-1}  \, r {\rm d}\psi)}  \; , 	\nonumber	
	\end{align*}
where $h_{t'} (\psi) \doteq h(t', \psi)\in{\mathcal D}_{\mathbb{R}} \bigl(S^1 \setminus \{ -\frac{\pi}{2}, \frac{\pi}{2} \} \bigr)$. Now  
	\begin{equation*}
		\label{seceq}
		\int {\rm d} t' \sin(\varepsilon t') \partial_{t'}^2h_{t'} 
		= \int {\rm d} t'  \big(\partial_{t'}^2 \sin(\varepsilon t')\big) h_{t'}
		= \int {\rm d} t'  \sin(\varepsilon t')(-\varepsilon^2)  h_{t'} 
	\end{equation*}
by partial integration and using $(\partial_{t'} h)_{t'}= \partial_{t'}(h_{t'})$.  
Thus  
	\[ 
		{\mathcal E}_{{\mathbb W}_1} \big(f,(\square_{\mathbb{W}_1}+\mu^2) 
		h\big) = 0 \; . 
	\]
A similar argument can be used to show 
${\mathcal E}_{{\mathbb W}_1} \big((\square_{\mathbb{W}_1}+\mu^2) h, f\big) = 0 $. It follows that
the kernel ${\mathscr E}_{{\mathbb W}_1} ( x ,  y )$, defined by    
	\begin{equation*}
		\label{integralkernel2}
				  \int {\rm d} \mu_{dS} ( x) 
				  {\rm d} \mu_{dS} ( y) \; 
			f ( x ) {\mathscr E}_{{\mathbb W}_1} ( x ,  y )  
			g ( y  ) 
			\doteq {\mathcal E}_{{\mathbb W}_1}( f,g)  \; ,  
	\end{equation*}
is anti-symmetric and satisfies the Klein--Gordon equation in both entries. 
Furthermore,
	\begin{align*} 
		\label{integralkernel3}
			{\mathcal E}_{{\mathbb W}_1}( f,g) 
   			 & =    - \int  r^3   {\rm d} t \, {\rm d} t' \left\langle  \cos^2_\psi {f}_t \,, \,
				\tfrac{\sin(\varepsilon(t-t')) } { | \varepsilon | }\, 
				\cos^2_\psi g_{t'}
				\right\rangle_{L^2(S^1, | \cos \psi |^{-1}  \, r {\rm d}\psi)} 
					\nonumber 			\\ 			
				& =    - \int r^2 {\rm d} t \, {\rm d} t' \int  
				\tfrac{r\, {\rm d} \psi }{| \cos \psi |} \, 
				\cos^{2} (\psi) {f}_t (\psi) 
				\\
			& \qquad \times
			\int \frac{ r \, {\rm d} \psi' }{|\cos \psi'|}\left( 
			\tfrac{\sin(\varepsilon(t-t')) } { | \varepsilon | }\right) (\psi, \psi')  
				\cos^{2} (\psi') \, g_{t'}(\psi')   \nonumber \\ 
			&=  -  \int r {\rm d} t \, {\rm d} l (\psi)   \int  r {\rm d} t' \, {\rm d} l (\psi') \; 
			{f}_t (\psi) \; \left(\tfrac{\sin(\varepsilon(t-t'))}
			{| \varepsilon |} \right) (\psi,\psi') \;  g_{t'}(\psi') \; . 
			\nonumber
	\end{align*}
In the last equation we used \eqref{new-surface}, \emph{i.e.}, 
${\rm d} l (\psi) = r \, | \cos \psi | \, {\rm d} \psi $. Thus 
	\begin{equation*}
		\label{3.46}
		{\mathscr E}_{{\mathbb W}_1} \left( x ,  y \right) 
			=-  \left( \tfrac{\sin(\varepsilon(t-t'))}{ | \varepsilon | }\right) 
			(\psi,\psi') \; ,  
	\end{equation*}
using $ x   \equiv  x (t,\psi)$ and $ y  \equiv  y (t',\psi')$. 
Clearly, ${\mathscr E}_{{\mathbb W}_1}$ satisfies \eqref{eqESigma} with ${\mathcal C}= S^1$. 

The unit normal future pointing vector field on 
$S^1 \setminus \{ - \frac{\pi}{2}, \frac{\pi}{2}\}$ is  
	\begin{equation}
		\label{eqnSigma}
		n(t,\psi) = \, r^{-1} {\rm cos}_\psi^{-1} \; \partial_t \; . 
	\end{equation}
In $I_-$ the vector field $\partial_t$ is past 
directed and $\cos\psi<0$, thus equation (\ref{eqnSigma}) holds for 
both half-circles $I_+$ and $I_-$. 
From \eqref{3.41} read off 
	\begin{equation}
		\label{3.49}
		r^{-1} \partial_t {\mathscr E}_{{\mathbb W}_1} \bigl(  x(t,\psi);  y (0,\psi')
		\bigr)_{\upharpoonright t=0}= -  
		\left( \frac{\varepsilon}{|\varepsilon|}  \,  \mathbb{1} \right) (\psi,\psi') \; , 
	\end{equation}
where 
	\[
		\mathbb{1}(\psi,\psi') = \tfrac{1}{r} |\cos\psi|\, \delta(\psi-\psi')
	\] 
is the kernel of the unit 
in $L^2 \left(S_1 \setminus \{ - \frac{\pi}{2}, \frac{\pi}{2}\} ,
|\cos\psi|^{-1} r {\rm d} \psi \right)$. 
Now $\,  {\rm cos}_\psi^{-1}\varepsilon|\varepsilon|^{-1}=|\,  {\rm cos}_\psi|^{-1}$.
Hence the r.h.s.~in \eqref{3.49} is $-\cos\psi \, \delta(\psi-\psi')$ 
and  (\ref{eqnSigma}) implies
	\begin{equation*}
		\label{eqdtESigma}
		n_{\rm \ell}\, {\mathcal E}  \bigl( x (t,\psi);  y (0,\psi')
		 \bigr) = -  \delta(\psi-\psi') \; .   
	\end{equation*}
$\delta(\psi-\psi')$ is the  kernel of the unit with respect to the induced line element 
$r {\rm d}\psi$ on~$S^1$, 
see Equation~(\ref{new-eqVolInd}). 
Thus \eqref{eqESigma2} holds, 
and hence, by the uniqueness result mentioned, ${\mathscr E}_{{\mathbb W}_1} 
={\mathscr E}_{\upharpoonright \mathbb{W}_1}$  
and ${\mathcal E}_{{\mathbb W}_1} = {\mathcal E}_{\upharpoonright \mathbb{W}_1}$ 
within the double wedge~$\mathbb{W}_1$.
\end{proof}

Thus, for $f\in{\mathcal D}_\mathbb{R}(\mathbb{W}_1 )$,   
$x \equiv x (t,\psi) \in \mathbb{W}_1 $ and $f_{t'}(\psi) \doteq f (x(t', \psi))$,
	\begin{equation} 
		\label{pmf-2}
				\mathbb{f} ( x  )   
					= -  \int r \, {\rm d} t' 
					\left( \frac{ \sin(\varepsilon(\, t-t'))}{|\varepsilon|} \, 
					\cos^2_\psi f_{t'}\right)  (\psi) \;  . 
	\end{equation} 
Note that \eqref	{pmf-2} describes $\mathbb{f} $ only on a proper 
subset of its support, namely the intersection of its support with $\mathbb{W}_1$. 

\goodbreak

\begin{remark}
For $ h\in{\mathcal D}_{\mathbb{R}} \left(I_+\right)$
one can extend the domain of $\mathbb{E}$ to distributions of the form 
	\begin{align}
	 \label{sharp-timetestfunction}
	 f (x) \equiv (\delta \otimes h) (x) &=	\delta (t)  \;  \frac{ h  (0,\psi )}{ r \cos \psi }\;  , 
	 \nonumber \\
	g (x) \equiv (\delta' \otimes h) (x) &= \left( \frac{\partial_t}{ r \cos \psi } \delta \right) (t)  \;  \frac{ h  (0,\psi )}{ r \cos \psi }\;  , 
	\end{align} 
with $x \equiv x (t, \psi)$, using the coordinates introduced in \eqref{w1psi}, and 
	\[ 
		{\rm d} \mu_{\mathbb{W}_1 } (t,\psi) = r^2  {\rm d} t\,  {\rm d} \psi   \cos \psi  \; . 
	\]  
The properties of the convolution ensure that $\mathbb{f} ,  \mathbb{g}$ are 
$C^\infty$-solutions of the Klein--Gordon equation~\eqref{3.25}, whose support is contained 
in~$dS \setminus \overline{\, W'}$. Within the region~$\mathbb{W}_1$ these solutions are given by
	\begin{align*}
				\mathbb{f} ( x  )   
					&=  -   \frac{ \sin(\varepsilon t)}{|\varepsilon|} \, 
					\cos \psi  \cdot h  (0,\psi )  \; , 
		\\
				\mathbb{g} ( x  )   
					&=   \frac{ \cos  (\varepsilon  t)}{  r } \  h  (0,\psi ) \; . 
	\end{align*} 
\end{remark}

\section{The canonical classical dynamical system}
\label{CaCDS}

Let $(n \, \mathbb{\Phi}) _{\upharpoonright {\mathcal C}}$ denote the Lie 
derivative of $\mathbb{\Phi}$ along the unit normal, future pointing vector field~$n$, restricted to 
the Cauchy surface~${\mathcal C}$.

\begin{theorem}[Dimock \cite{D77}, Theorem 1]
\label{cauchyproblem}
 Let ${\mathcal C} \subset dS$ be a Cauchy surface and let~$(\mathbb{\phi}, \mathbb{\pi}) \in  
 C^\infty ({\mathcal C}) \times C^\infty ({\mathcal C}) $. 
Then there exists a unique $\mathbb{\Phi} \in C^\infty (dS)$ satisfying the Klein--Gordon 
equation~\eqref{3.25} with Cauchy data
	\begin{equation}
		\label{3.26} 
		\mathbb{\Phi}_{\upharpoonright {\mathcal C}} = \mathbb{\phi} \; , 
		\quad (n \mathbb{\Phi})_{\upharpoonright {\mathcal C}} = \mathbb{\pi} \; . 
	\end{equation}
Furthermore, ${\rm supp\,} \mathbb{\Phi} \subset \bigcup_\pm \Gamma^\pm ({\rm supp\,} \mathbb{\phi} \cup {\rm supp\,} \mathbb{\pi} ) $. 
\end{theorem}

\begin{remark}
For functions in the \emph{Sobolev space}\index{Sobolev space} $\mathbb{H}^2_{\rm \, loc}(dS)$, 
this is the classical existence and uniqueness theorem of Leray \cite{L}. 
\end{remark}

Now consider the space
\label{hatckpage}
	\begin{equation*}
		\label{calSHat}
		\widehat {\mathfrak k} (S^1) 
		\doteq C^\infty_\mathbb{R}(S^1)\times C^\infty_\mathbb{R}(S^1)
	\end{equation*}
together with the {\em canonical symplectic form} 
	\begin{equation}
		\label{eqSplFormHat}
		\widehat \sigma\big((\mathbb{\phi}_1,\mathbb{\pi}_1),(\mathbb{\phi}_2,\mathbb{\pi}_2)\big)
		\doteq \langle \mathbb{\phi}_1,\mathbb{\pi}_2 \rangle_{L^2(S^1, r\,  {\rm d} \psi ) }-  
		\langle\mathbb{\pi}_1,\mathbb{\phi}_2 \rangle_{L^2(S^1, r \, {\rm d} \psi)} \; , 
	\end{equation}
where $r \, {\rm d} \psi $ is the line
element on $S^1$, see (\ref{new-eqVolInd}). The right hand side in \eqref{eqSplFormHat}
is zero, if $(\mathbb{\phi}_1,\mathbb{\pi}_1)$ and $(\mathbb{\phi}_2,\mathbb{\pi}_2)$ have disjoint support. 
For the open half\-circles~$I_\pm$ define 
	\begin{equation*}
		\label{calSHat2}
		\widehat {\mathfrak k} (I_\pm ) 
		\doteq{\mathcal D}_\mathbb{R}(I_\pm)\times{\mathcal D}_\mathbb{R}(I_\pm) \; . 
	\end{equation*}
 
\begin{proposition}
\label{nocheinlabel} 
The symplectic space $(\widehat {\mathfrak k} (S^1),  \widehat \sigma)$ carries a 
representation $ \Lambda  \mapsto \widehat  {\mathfrak u} ({\Lambda}) $,
$\Lambda \in O(1,2)$, defined  by  
\label{widehatccTpage}
	\begin{equation}
		\label{eqTtHat}
		\widehat {\mathfrak u} (\Lambda) (\mathbb{\phi},\mathbb{\pi})  
		\doteq\big(  [ \Lambda_* f]_{ \upharpoonright S^1} \; ,
		[ n  \Lambda_* f]_{\upharpoonright S^1}\big) \; . 
	\end{equation}
The triple  $\bigl(\, \widehat {\mathfrak k}(S^1), \widehat \sigma, \widehat  {\mathfrak u}  \bigr)$  is the 
{\em canonical classical dynamical system} associated to 
the covariant classical dynamical system specified in Definition~\ref{Def4-11}.
\end{proposition}

\begin{proof}This result follows directly from Theorem
  \ref{cauchyproblem} and  the invariance of the Klein--Gordon operator
under the adjoint pull-back action of $O(1,2)$. 
\end{proof}

\begin{proposition}
The map
	\begin{align*}
	 	{\mathbb  T} \colon ( {\mathfrak k}(dS),  \sigma,  {\mathfrak u} (\Lambda) )
		&\to (\widehat {\mathfrak k}(S^1), \widehat \sigma, \widehat  {\mathfrak u} (\Lambda) ) \\
	 	 [f]  & \mapsto   ( \mathbb{f}_{\upharpoonright S^1} \; ,  
		 (n \mathbb{f})_{\upharpoonright S^1} )
		\equiv (\mathbb{\phi} ,  \mathbb{\pi}) 
	\end{align*}
is symplectic. 
\end{proposition}

\begin{proof} Let $f, g \in {\mathcal D}_{\mathbb{R}}(dS)$. Then Stokes' theorem 
implies (see\footnote{Note that Dimock's operator $E$ differs from our conventions by a sign, as can be seen 
by comparing Corollary 1.2 in \cite{D77} with \eqref{eqESigma2}.} \cite[Lemma~A.1]{D77}) that 
	\begin{align*}
		\sigma ([f], [g]) & = {\mathcal E} (f,g) 
		\\
		&= \int_{dS} {\rm d} \mu_{dS}(x) \; f (x) (\mathbb{E} g) (x)   \\
		& = \int_{S^1}  r \, {\rm d} \psi \; \Bigl(  
		(\mathbb{E} f)_{\upharpoonright S^1} (\psi) (n \mathbb{E}  g)_{\upharpoonright S^1} (\psi) 
		-
		( n  \mathbb{E} f)_{\upharpoonright S^1} (\psi) (\mathbb{E} g)_{\upharpoonright S^1} (\psi)  \Bigr) \; 
		\\
		& =   
		\langle \mathbb{f}_{\upharpoonright S^1}, 
		(n \mathbb{g})_{\upharpoonright S^1} \rangle_{L^2(S^1, r \, {\rm d} \psi ) }
		-
		\langle (n \mathbb{f})_{\upharpoonright S^1}, 
		\mathbb{g}_{\upharpoonright S^1} \rangle_{L^2(S^1,  r\, {\rm d} \psi)} 
		\\
		& =  \widehat \sigma\Bigl( \bigl( \mathbb{f}_{\upharpoonright S^1},
		(n \mathbb{f})_{\upharpoonright S^1} \bigr), \bigl( \mathbb{g}_{\upharpoonright S^1},
		(n \mathbb{g})_{\upharpoonright S^1} \bigr)\Bigr) \; .
	\end{align*}
Thus $\mathbb{T}$ is symplectic. 
\end{proof}

The canonical projection  
\label{widehatcxC}
	\begin{align}
		\label{eqCovCan}
		\widehat {\mathbb{P}} \colon  \; \;{\mathcal D}_{\mathbb{R}} (dS) 
		& \to  C^\infty_\mathbb{R} (S^1) 
		\times C^\infty_\mathbb{R} (S^1)  \nonumber  \\ 
		f  & \mapsto   
		\bigl( \mathbb{f}_{ \upharpoonright S^1}, 
		\left( n \, \mathbb{f} \right)_{\upharpoonright S^1} \bigr) \equiv \widehat f
	\end{align}
maps a smooth, real valued function $f \in{\mathcal D}_\mathbb{R} 
(dS) $ with compact support
to the Cauchy data of a $C^\infty$-solution $\mathbb{f}$ of the Klein--Gordon equation~\eqref{3.25}. 

\begin{remark}
For the special case $f\in{\mathcal D}_\mathbb{R}(\mathbb{W}_1 )$, 
Eq.~\eqref{eqCovCan} yields
	\begin{align}  
		\label{pmfhut}
			\mathbb{f}_{ \upharpoonright S^1}(\psi)  	
 &=    \int r \, {\rm d} t' \Bigl( \tfrac{ \sin(t' \varepsilon )}{|\varepsilon|} \, 
			 	\cos^2_\psi f_{t'}\Bigr)  ( \psi) \; ,  \\ 
		\label{pmfhutableitung} \quad
			(n \, \mathbb{f} )_{\upharpoonright S^1}
			(\psi) 
			&= - \tfrac{1}{ r |\cos  \psi|} \int r \, {\rm d} t' \,  \left( \cos(t' \varepsilon) \, 
			\cos^2_\psi f_{t'}\right)  ( \psi) 
			\; ,  
	\end{align}
where $f_{t}(\psi):=f(  x (t, \psi)) $,
using again $\,  {\rm cos}_\psi^{-1}\varepsilon|\varepsilon|^{-1}=|\,  {\rm cos}_\psi|^{-1}$.
An explicit formula, which generalizes both \eqref{pmfhut} and \eqref{pmfhutableitung}
to $ f\in{\mathcal D}_\mathbb{R}(dS)$ will follow from Eq.~\eqref{comfkt} in Section \ref{new-3.4}.
\end{remark}

\begin{proposition}
\label{fsol}
Let $\widehat f \in \widehat {\mathfrak k} (I)$, $ I \subset S^1$. Then 
	\[
		\widehat{\mathfrak u} (\Lambda) \, \widehat f \in 
		\widehat {\mathfrak k} \; \Bigl( \bigl( \Gamma^{+}(\Lambda I ) 
		\cup \Gamma^{-}(\Lambda I) \bigr) \cap S^1 \Bigr) \; .
	\]
\end{proposition}

\begin{proof} Let $[f] = {\mathbb  T}^{-1} \, \widehat f \; $ be the element in ${\mathfrak k}(dS)$ associated 
to the smooth solution~$f$ of the Klein--Gordon equation with Cauchy data given by $\widehat f$. It follows that 
\[
{\mathbb  T}^{-1} \bigl( \widehat{\mathfrak u} (\Lambda) \widehat f  \, \bigr) = {\mathfrak u} (\Lambda) [f] \; . 
\]
The smooth solution of the Klein--Gordon equation associated to $ {\mathfrak u} (\Lambda) [f]$ has support in 
$\Gamma^{+}(\Lambda I ) \cup \Gamma^{-}(\Lambda I) $; thus
the Cauchy data of the solution associated to $ {\mathfrak u} (\Lambda) [f]$ have support in 
$\bigl( \Gamma^{+}(\Lambda I ) \cup \Gamma^{-}(\Lambda I) \bigr) \cap S^1$.
\end{proof}

\goodbreak

\begin{proposition} 
\label{porp4.13}
The  rotations $\widehat {\mathfrak u} (R_0(\alpha))$, $\alpha \in  [0, 2\pi )$, which map
	\[
		\bigl( \mathbb{\phi} (\psi), \mathbb{\pi}(\psi) \bigr) 
		\mapsto  \bigl(\mathbb{\phi} (\psi - \alpha) , \mathbb{\pi} (\psi - \alpha) \bigr) , \qquad \alpha \in  [0, 2\pi ) \; , 
	\]
and the boosts $\widehat  {\mathfrak u} (\Lambda_1(t))$, $t \in\mathbb{R}$, which map
	\begin{equation}
		\label{3.77}
		(\mathbb{\phi}, \mathbb{\pi}) 
		\mapsto  
		  (\mathbb{\phi}_t , \mathbb{\pi}_t )  
		  \; ,
	\end{equation}
with
	\begin{align}
		\label{varphi-pi-t} 
			\mathbb{\phi}_t  
			&=  \cos(\varepsilon t)\mathbb{\phi}
					- \sin(\varepsilon t) \,
                                        \varepsilon^{-1} \,  {\rm cos}_\psi \,
                                        \mathbb{\pi}   \nonumber	 \\ 
			\mathbb{\pi}_t &=   
		( r \,  {\rm cos}_\psi)^{-1}\big(\varepsilon\sin(\varepsilon t)\mathbb{\phi} 
              + \cos(\varepsilon t) \,  {\rm cos}_\psi \; \mathbb{\pi} \big) \; ,  
	\end{align}
generate the representation $ \Lambda  \mapsto \widehat  {\mathfrak u} (\Lambda) $  of $SO_0(1,2)$ introduced in (\ref{eqTtHat}).
The 	points $(\mathbb{\phi} (\pm\frac{\pi}{2}) , \mathbb{\pi} (\pm\frac{\pi}{2}) )$ are fixed
points of the map  $t\mapsto (\mathbb{\phi}_t (\psi) , \mathbb{\pi}_t (\psi))  $. 
The representers of the reflections $P_1$ and $ T$ are 
	\begin{align}  
		\label{eqJCan}
		\widehat  {\mathfrak u} (P_1 ) \colon  (\mathbb{\phi},\mathbb{\pi})& \mapsto  
		\bigl( (P_1)_*\mathbb{\phi}, (P_1)_*\mathbb{\pi} \bigr) \; ,  \\ 
		\label{eqJCan2}
		\widehat  {\mathfrak u} (  T) \colon  (\mathbb{\phi},\mathbb{\pi})& \mapsto 
		 ( \mathbb{\phi}, - \mathbb{\pi} ) \; .  
	\end{align}
\end{proposition}

\begin{proof} Recall \eqref{3.32} and 
consider the boosts
$t \mapsto \Lambda_1 (t)$, acting on the Cauchy data on~$S^1$.  
Now combine  (\ref{eqTtHat}) and  the definition of $\Lambda_*$ to conclude that the boosts
$ \widehat {\mathfrak u} (\Lambda_1 (t)) (\mathbb{\phi},\mathbb{\pi}) $, $ t \in\mathbb{R}, $ 
are determined by  
$  \mathbb{\Phi}_{ \upharpoonright S^1 \cup \, \mathbb{W}_1} $, 
where  $\mathbb{\Phi}$ is the solution of the Klein--Gordon equation \eqref{3.25} specified in Theorem~\ref{cauchyproblem}. 
Evaluate~\eqref{3.77}
with care---write it out explicitly and take advantage of the fact that $\varepsilon^2$ is a differential operator which satisfies 
$(\varepsilon^2 h)  (\psi\pm\frac{\pi}{2}) = O(\psi)$ for $h \in C^\infty (S^1)$, just as $\cos (\psi\pm\frac{\pi}{2})$---to show that the 
map is well-defined  for $\psi  \to \pm \frac{\pi}{2}$ and 
	\[
		\textstyle{
		\bigl(\mathbb{\phi}_t (\pm \frac{\pi}{2}) , \mathbb{\pi}_t (\pm \frac{\pi}{2}) \bigr) = 
		\bigl(\mathbb{\phi} (\pm \frac{\pi}{2}), \mathbb{\pi}(\pm \frac{\pi}{2}) \bigr) }\qquad \forall t \in\mathbb{R} \; .
	\]
(This ensures that $\mathbb{\phi}_t$ and $\mathbb{\pi}_t$ are both well-defined despite the fact that
the coordinate system is degenerated 
at $\psi = \pm \frac{\pi}{2}$.) It remains to construct $\widehat {\mathfrak u} (\Lambda_1 (t))$ in the 
space-time region~$\mathbb{W}_1$. On~$\mathbb{W}_1$, the Klein--Gordon equation \eqref{3.25} reads 
	\begin{equation}
		\label{3.79}
		\frac{1}{r^2 \cos^2 \psi}\,(\partial_t^2+\varepsilon^2) \mathbb{\Phi}=0 \; , 
	\end{equation}
using \eqref{varepsilon}. 
The real valued solution  of  \eqref{3.79} in the region 
$\mathbb{W}_1$ with  Cauchy data   (see  \eqref{eqnSigma})
	\begin{equation}
		\label{3.80}
			\mathbb{\phi}  = \mathbb{\Phi}_{ \upharpoonright S^1} \; , 
			\qquad 
			\mathbb{\pi} = \frac{1}{r \cos \psi} \, (\partial_t \mathbb{\Phi})_{ \upharpoonright S^1} \; ,
	\end{equation}
 is   
	$
		\mathbb{\Phi}  (t,\cdot )= \cos(\varepsilon t)\mathbb{\phi} + 
 		 \sin(\varepsilon t) \, \varepsilon^{-1} \,  {\rm cos}_\psi\, \mathbb{\pi} $. 
Hence 
	\[
		\mathbb{\phi}_t  \equiv (\Lambda_1(t)_*\mathbb{\Phi}) _{ \upharpoonright S^1} 
		\qquad \text{and} \qquad
 		\mathbb{\pi}_t \equiv n (\Lambda_1(t)_*  \mathbb{\Phi}) _{ \upharpoonright S^1} \;   
	\]
are determined by 
	\begin{align*}
		(\Lambda_1(t)_*\mathbb{\Phi})(t',\cdot ) 
		&= \cos(\varepsilon (t'-t))\mathbb{\phi} + 
 				 \sin(\varepsilon (t'-t)) \, \varepsilon^{-1} \,  {\rm cos}_\psi \mathbb{\pi} \; ,\\ 
		(\partial_{t'} \Lambda_1(t)_*\mathbb{\Phi})(t',\cdot ) 
		&= -\varepsilon\sin(\varepsilon (t'-t))\mathbb{\phi} + 
 			 \cos(\varepsilon (t'-t)) \,  {\rm cos}_\psi \mathbb{\pi} \; . 
	\end{align*}
$\varepsilon^2 $ maps $ C_\mathbb{R}^\infty (S^1)$ to the functions in $C_\mathbb{R}^\infty (S^1 )$, which vanish at 
$\psi =\pm \frac{\pi}{2}$. Thus $\widehat  {\mathfrak u} (\Lambda_1(t))$ maps
	\[	
		C^\infty_\mathbb{R}(S^1)\times C^\infty_\mathbb{R}(S^1) 
		\mapsto C^\infty_\mathbb{R}(S^1)\times C^\infty_\mathbb{R}(S^1) \; . 
	\] 
The boosts $\Lambda_1(t)$, $t \in\mathbb{R}$, together with the rotations $R_0(\alpha)$, 
$\alpha \in [0, 2 \pi)$, generate $SO_0(1, 2)$. 

$(P_1)_*$ and $ T_*$  commute with the restriction to $S^1$, and $(P_1)_*$ 
commutes with $n$ while $ T_*$  anti-commutes with $n$. On the time-zero 
circle $S^1$, the spatial reflection $P_1$ acts as 
	\[
		P_1 \colon \psi \mapsto 
		\pi - \psi \;, \qquad \psi \in [ -\tfrac{\pi}{2}, \tfrac{3 \pi}{2}) \; . 
	\]
Thus \eqref{eqJCan} and \eqref{eqJCan2} follow. 
\end{proof}
\goodbreak

\goodbreak
The reflection at the edge of the wedge $W_1$
	\[
		(P_1T)_* \colon \; g (x_0, x_1,  x_2) 
		\mapsto g (-x_0, x_1,  -x_2) \; , \qquad g \in{\mathcal D} (dS) \; , 
	\]
gives rise to a double classical system in the sense of Kay~\cite{Kay1}: 

\goodbreak 
\begin{proposition}
\label{ssdd} The symplectic space $\widehat {\mathfrak k} 
\bigl(S^1 \setminus \{ \pm \frac{\pi}{2} \}\bigr) $ is the direct sum of  $\widehat {\mathfrak k} (I_+ )$ 
and~$ \widehat {\mathfrak k} (I_- )$. Moreover, 
\begin{itemize}
\item [$ i.)$] 	$
		\widehat \sigma (  \widehat f ,  \widehat g ) = 0  
	$
 for all $\widehat{f} \in \widehat {\mathfrak k} (I_+ )$ and $\widehat{g} \in \widehat {\mathfrak k} (I_- )$;  
\item [$ ii.)$] the maps $\{ \, \widehat  {\mathfrak u} (\Lambda_1 (t)) \}_{t \in\mathbb{R}}$ leave the 
subspaces $\widehat {\mathfrak k} (I_+ )$ and $\widehat {\mathfrak k} (I_- )$ invariant;
\item [$ iii.)$] $\widehat  {\mathfrak u} (P_1T)$ is an anti-symplectic involution,   which satisfies
	\[
		\widehat  {\mathfrak u} (P_1T)  \, \widehat {\mathfrak k} (I_+ ) 
		= \widehat {\mathfrak k} (I_- ) 
		\quad
		\text{and}
		\quad
		\Bigl[ \, \widehat  {\mathfrak u} (\Lambda_1 (t))  \, , 
		\, \widehat  {\mathfrak u} (P_1T) \Bigr]=0 \quad \forall t \in\mathbb{R} \; . 
	\]
\end{itemize} 
Thus $\bigl( \, \widehat {\mathfrak k} \bigl(S^1 
\setminus \{ \pm \frac{\pi}{2} \}\bigr) ,  \widehat  {\mathfrak u} (\Lambda_1 ) , 
\widehat  {\mathfrak u} (P_1T) \bigr)$ is a double classical linear dynamical system 
in the sense of~\ref{dcldsy}.
\end{proposition}

\bigskip
\goodbreak 
In other words, the following diagram commutes:
\vskip -.5cm

\begin{picture}(200,140)

\put(130,100){$\longrightarrow$}
\put(60,100){$\bigl(\, \widehat {\mathfrak k} (I_+ ), \widehat \sigma \bigr)$}
\put(125,110){$\widehat  {\mathfrak u} (P_1T)$}
\put(35,70){$ {\widehat  {\mathfrak u} (\Lambda_1 (t))}$}
\put(210,70){$\widehat {\mathfrak u} (\Lambda_1 (t))$}
\put(125,50){$\widehat  {\mathfrak u} (P_1T)$}

\put(170,100){$ \bigl(\, \widehat {\mathfrak k} (I_- ), \widehat \sigma\bigr)$}

\put(60,40){$\bigl(\, \widehat {\mathfrak k} (I_+ ), \widehat \sigma \bigr)$}

\put(170,40){$\bigl(\, \widehat {\mathfrak k} (I_- ), \widehat \sigma \bigr)\, .$}

\put(130,40){$\longrightarrow$}

\put(85,85){\vector(0,-3){20}}
\put(200,85){\vector(0,-3){20}}

\end{picture}
\vskip - 1.8cm
\quad

\bigskip
\begin{remark}
\label{re-3-20}
The domain of the map $\widehat {\mathbb{P}} $ extends 
to testfunctions $f, g$ of the form given in \eqref{sharp-timetestfunction}. 
Use \eqref{pmfhut} and \eqref{pmfhutableitung} to compute the corresponding Cauchy data:
	\begin{equation}
		\label{sharp-timetestfunction-phipi}
		\big( \mathbb{f}_{ \upharpoonright S^1}, 
							 (n \, \mathbb{f} )_{\upharpoonright S^1})
  		 = \left( 0, - \tfrac{h}{  r } \right) \equiv  (\mathbb{\phi}, \mathbb{\pi}) \; , 
	\end{equation}
and, by partial integration, 
	\begin{equation}
	\label{sharp-timetestfunction-phipi2}
	\big(  \mathbb{g}_{ \upharpoonright S^1}, 
							 (n \, \mathbb{g})_{\upharpoonright S^1})
 		  = \left(  \tfrac{h}{ r } , 0 \right) \equiv  (\mathbb{\phi}, \mathbb{\pi}) \; . 
	\end{equation}
All elements in~$\widehat {\mathfrak k}(I_+)$ are linear combinations of the Cauchy data 
		  arising from  sharp-time testfunctions $f, g$ 
		  of the form described in \eqref{sharp-timetestfunction}.
\end{remark}

\part{Free Quantum Fields}

\chapter{One-Particle Hilbert Spaces}

\section{The covariant one-particle Hilbert space} 
\label{new-3.4}

Our basic strategy is to use, just as in Minkowski space (see, \emph{e.g.},~\cite{RS}), 
the restriction of the Fourier-Helgason transform to the upper mass shell,
\label{mass-shell-ft-page}
	\begin{align}
	\label{mass-shell-ft}
				{\mathcal F}_{+ \upharpoonright \nu} \colon  \; \; {\mathcal D}_{\mathbb{R}} (dS) 
				& \to  \widetilde {\mathfrak h}_\nu (\partial V^+)  \nonumber  \\ 
				 f  & \mapsto  \sqrt{ \frac{c_\nu {\rm e}^{-  \pi \nu}  r   }{\pi}} 	\; 		
				 	 \widetilde {f}_+  (  \, . \, , s^+) \doteq  {\widetilde f}_\nu \; , 
	\end{align}
with $s^+$ given by \eqref{dd1},
to define a (complex valued) semi-definite quadratic form 
	\begin{equation}
	\label{nuclear}
	{\mathcal D}_{\mathbb{R}} (dS) \ni f, g \mapsto 
	\langle \widetilde f_\nu , \widetilde g_\nu \rangle_{ \widetilde {{\mathfrak h}} (\partial V^+)  }
	\end{equation}
on the test-functions.  (We will suppress the index $\nu$ when possible, for example, we 
will frequently write $\widetilde {{\mathfrak h}} (\partial V^+) $
instead of $\widetilde {{\mathfrak h}}_\nu (\partial V^+)$.) 
 The value of the positive normalisation constant (see Harish-Chandra~\cite{HC1, HC2}) 
	\[
		c_\nu = -\frac{1}{2\sin  (\pi s^+ )} = 
		\frac{1}{2 \cos (i \nu \pi ) }  
	\]  
is chosen such that twice the imaginary part of the scalar product \eqref{nuclear} equals the 
value of the symplectic form $\sigma$ of the classical dynamical system given in \eqref{eqSplForm}; 
for further details, see the discussion preceding \eqref{comfkt} below.
Using \eqref{eq:gamma-2}, one can show that
	\[
		c_\nu = \frac{\Gamma\left(1 + s^+ \right)\Gamma\left(1 +s^- \right)}{2\pi} = c_{-\nu} \; . 
	\]

The following result of Faraut was pointed out to us by J.~Bros.

\begin{proposition}[Faraut \cite{Fa}, Prop.~II.4]
\label{propo2.5}
Let $f \in {\mathcal D}_{\mathbb{R}} (dS)$ and $0 <  -i \nu < \frac{1}{2}$. Then
	\begin{align*} 
		\int_{dS} {\rm d} \mu_{dS} (x) \; f (x) \; (  x_\pm \cdot  p )^{- { \frac{1}{2} } - i \nu}
		&= \frac{\Gamma(\frac{1- i \nu}{2})}{\Gamma (\frac{3}{4}
			+\frac{i \nu}{2})}  \; \frac{\sqrt{\pi} \; \Gamma(\frac{1}{2}+ i \nu)}{ 2^{-\frac{1}{2}+ i \nu} \, \Gamma( i \nu)}
			 \int_{\Gamma} {\rm d} \mu_{\Gamma} ( p') \; (  p \cdot p')^{-\frac{1}{2}+ i \nu} 
			\nonumber \\
		& \qquad \qquad \quad \times
			\int_{dS} {\rm d} \mu_{dS} (x) \; f (x) \; 
			(  x_\pm \cdot  p' )^{- { \frac{1}{2} } + i \nu} \;  ,  
	\end{align*}
where $\Gamma$ is a closed curve on the forward light cone~$\partial V^+$,  which encloses the origin.
\end{proposition}

\begin{remark} Choosing 
$p  = (1, \cos \alpha,  \sin \alpha)$  and $ p'  = (1, \cos \alpha', \sin \alpha')$ 
we find 
	\[
		p  \cdot p' = 1- \cos (\upsilon-\upsilon') \; .   
	\]
Thus we have recovered the factor \eqref{Barg-factor} first introduced by Bargmann; see the definition of the intertwined $A_\nu$.  
\end{remark}

\begin{lemma}
Let $\mu^2  = \frac{1}{4r^2} + m^2$, \emph{i.e.}, $\nu^2 = m^2 r^2$. It follows that 
	\[
		\ker {\mathcal F}_{+ \upharpoonright \nu} = {\ker {\mathbb P}} = (\square_{dS}+\mu^2) {\mathcal D}_{\mathbb{R}} (dS) \; . 
	\]
\end{lemma}

\begin{proof}
If $f \in \ker {\mathbb P}$, 
then \eqref{cmzero} implies that
there exists $g \in {\mathcal D}_{\mathbb{R}} (dS)$ 
such that $f = (\square_{dS}+\mu^2) g$. 
Evaluate $ {\mathcal F}_{+ \upharpoonright \nu} \bigl( (\square_{dS}+\mu^2) g \bigr)$ 
using the definition of the Fourier-Helgason transform 
(see  \eqref{eqPW-new}) and  
	\[ 
		(\square_{dS}+\mu^2) (  x_+  \cdot  p )^{s^\pm} = 0
	\]
for $s^\pm$ given by \eqref{dd1} with $\zeta^2= \mu^2 r^2$. This shows that  
$\ker {\mathcal F}_{+ \upharpoonright \nu} \supset {\ker {\mathbb P}}$. The inclusion 
$\ker {\mathcal F}_{+ \upharpoonright \nu} \subset {\ker {\mathbb P}}$ will follow from the fact that 
	\[
		{\mathcal E} (f, g) = 2 \Im \langle \widetilde f , \widetilde g \rangle_{ \widetilde{\mathfrak h} (\partial V^+)} \; .
	\]
This will be verified in \eqref{comfkt} below.
\end{proof}

\subsection{Real Hilbert Spaces}
The kernel of the quadratic form \eqref{nuclear} equals $\ker {\mathcal F}_{+ \upharpoonright \nu}$.  This allows us to 
turn the real symplectic spaces ${\mathfrak k} (X)$ into {\em real} pre-Hilbert spaces
\label{FHtransformationmasspage}
\label{ophs1page}
	\begin{equation}
		\label{ophs2}
		{\mathfrak h}^\circ (X) \doteq \bigl( {\mathfrak k} (X), 
		\Re \langle \, \widetilde \cdot , \widetilde \cdot \,  \rangle_{ \widetilde {{\mathfrak h}} (\partial V^+)  } \bigr) \; , 
		\qquad X = dS,  {\mathcal O }, W \; .  
	\end{equation}
The completion of ${\mathfrak h}^\circ (X)$ 
defines the real Hilbert spaces 
${\mathfrak h} (X)$, $X = dS,  {\mathcal O }, W$. 
Their real valued scalar product is given by the real part 
	\[
	\Re \langle f , g \rangle_{{\mathfrak h} (dS)}
	 \doteq \tfrac{1}{4} \Bigl( \| f + g \|_{{\mathfrak h} (dS)}
	- \| f - g \|_{{\mathfrak h} (dS)} \Bigr) \; 
	\]
of the {\em complex valued} scalar product 
	\begin{equation}
		\label{ophs2a}
		\langle [f] , [g] \rangle_{{\mathfrak h} (dS)} \doteq 
		\langle \, \widetilde f_\nu , \widetilde g_\nu  \,  \rangle_{ \widetilde {{\mathfrak h}} (\partial V^+)  }\;   , 
		\qquad [f], [g] \in {\mathfrak k} (X)\; ,   
	\end{equation}
with $\nu = \nu(\mu)$ given by \eqref{dd1} with $\zeta^2 = \mu^2 r^2$. 

\bigskip
\subsection{Complex Hilbert Spaces}
The question now arises, whether these real Hilbert spaces can be interpreted as complex Hilbert spaces, \emph{i.e.}, 
whether or not they carry an intrinsic complex structure. The answer to this question depends on the choice of~$X \subset dS$. 
In case $X = dS$,  the real-valued  scalar product $  f , g \mapsto  
\Re \langle  f , g \rangle_{{\mathfrak h} (dS)}  $ can be used to define an 
operator~${\mathscr I}$, 
	\[
		\Re \langle {\mathscr I} f , g \rangle_{{\mathfrak h} (dS)} 
		\doteq\Im \langle  f , g \rangle_{{\mathfrak h} (dS)}    
		\qquad \forall   g \in {\mathfrak h} (dS) \; . 
	\]
The Riesz lemma (applied on the real Hilbert space ${\mathfrak h} (dS)$) fixes the vector 
	\[
		{\mathscr I}  f  \in {\mathfrak h} (dS)
	\]
uniquely, since the 
symplectic form $\Im \langle \, . \, ,  \, . \, \rangle_{{\mathfrak h} (dS)}  $ is non-degener\-ated on 
${\mathfrak h}^\circ (dS)$. The operator ${{\mathscr I}}$ satisfies 
	\[
	 	\Im \langle {\mathscr I}  f, g \rangle_{{\mathfrak h} (dS)}  
		= - \Im \langle f , {\mathscr I}  g \rangle_{{\mathfrak h} (dS)}   
	\] 
and $ {\mathscr I}^2 = -1 $, and therefore defines a complex structure: for $f \in {\mathfrak h} (dS)$ we have
	\begin{equation}
		\label{CS}
		(\lambda_1 + i \lambda_2) f  = \lambda_1 f  + \lambda_2 ({\mathscr I} f)  \; , 
		\qquad \lambda_1, \lambda_2 \in \mathbb{R} \; . 
	\end{equation}
This turns the real Hilbert space $\bigl( {\mathfrak h} (dS), \Re \langle \, . \, ,  \, . \, \rangle_{{\mathfrak h} (dS)}  \bigr)$  into a 
complex Hilbert space $\bigl({\mathfrak h} (dS), \langle \, . \, ,  \, . \, \rangle_{{\mathfrak h} (dS)} \bigr)$. 
The scalar product $f, g \mapsto \langle f ,  g \rangle_{{\mathfrak h} (dS)}$  
is anti-linear in $f$ and linear in $g$ with respect to the complex structure~defined in \eqref{CS}. 

\begin{remark}
In case $X= {\mathcal O }$ (with ${\mathcal O }$ bounded) or $X=W$,  the spaces ${\mathfrak h} ({\mathcal O })$ 
and ${\mathfrak h} (W)$ are only {\em real} subspaces of~${\mathfrak h} (dS)$. We will later show that their 
\emph{complex linear span} 
is dense in ${\mathfrak h} (dS)$. 
\end{remark}


\subsection{A representation of $O(1,2)$}
In \cite{BMJ} we will identify the quantum one-particle space with some abstract Hilbert space~${\mathfrak h}$ 
carrying a unitary irreducible representation of the Lorentz group $SO_0(1,2)$. The {\em real} subspaces ${\mathfrak h} ({\mathcal O })$ 
associated to open bounded subsets of ${\mathcal O }$ will be identified using the concept of 
{\em modular localization}~\cite{BGL}. 
Here we show that ${\mathfrak h}  (dS)$ carries a representation of $O(1,2)$. 

\begin{proposition}
\label{UIR-FH}
There is a unitary representation $ u$ of $SO_0(1,2)$ on~${\mathfrak h}  (dS)$ such that for $f \in {\mathcal D}_{\mathbb{R}} (dS)$
\[ u (\Lambda) [f] =  [ \Lambda_*f ]  \; , \qquad \Lambda \in SO_0(1,2) \; ; \]
and consequently, 
$u(\Lambda) {\mathfrak h}  ({\mathcal O}) =   {\mathfrak h}  (\Lambda {\mathcal O}) $, $\Lambda \in SO_0(1,2)$.  
In other words, $u$ acts geometrically on ${\mathfrak h}  (dS)$. 
\end{proposition}

\begin{proof} In order to extend the pull-back from ${\mathfrak h}^ \circ (dS)$ to a unitary representation  of 
$SO_0(1,2)$ on ${\mathfrak h} (dS)$, we have to show that
$ \| [ \Lambda_* f ] \|_{{\mathfrak h} (dS)} = \| [f] \|_{{\mathfrak h} (dS)}  $. 
By construction, 
	\begin{align*}
 		\| [ \Lambda_* f ] \|_{{\mathfrak h} (dS)} & =
		\Bigl\| \int_{dS} {\rm d} \mu_{dS} ( x ) \; f( \Lambda^{-1} x ) \; (  x_+ \cdot  p )^{s^+} \Bigr\|_{\widetilde {\mathfrak h} (\partial V^+)} \\
		&= \Bigl\| \int_{dS} {\rm d} \mu_{dS} ( x ) \; f( x ) \; (  \Lambda x_+ \cdot  p )^{s^+} \Bigr\|_{\widetilde {\mathfrak h} (\partial V^+)}  \\
		&= \Bigl\| \int_{dS} {\rm d} \mu_{dS} ( x ) \; f( x ) \; (  x_+ \cdot \Lambda^{-1} p )^{s^+}  \Bigr\|_{\widetilde {\mathfrak h} (\partial V^+)} \\
		&= \|  \widetilde u_\nu^+   (\Lambda)   \widetilde f_\nu  \|_{\widetilde {\mathfrak h} (\partial V^+)}
		=  \| \widetilde f_\nu  \|_{\widetilde {\mathfrak h} (\partial V^+)} = \| [f] \|_{{\mathfrak h} (dS)}\;   , 
		\end{align*}
where ${s^+}$ is given by \eqref{dd1} with $\zeta^2 = \mu^2 r^2$. 
\end{proof}

\begin{proposition} 
Let $u(T)$ and $u(P)$ be defined by 
	\[ 
		\widetilde{u}_\nu^+ (T) {\mathcal F}_{+ \upharpoonright \nu} f \doteq {\mathcal F}_{+ \upharpoonright \nu} T_* \overline{f} 
		\; , 
		\qquad u(P_2)[f] \doteq  [ P_{2*}f ] \; , \qquad f \in {\mathcal D}_{\mathbb{R}} (dS) \; .
	\]
The operators ${u}(T)$  and $u(P_2)$ extend to well-defined (anti-)unitary operators on 
${\mathfrak h}(dS)$.  They extend the representation $u$ form $SO_0(1,2)$ to $O(1,2)$.
\end{proposition}

\begin{proof}
Let us calculate the action of $\widetilde{u}(T)$ in $\widetilde{\mathfrak{h}}_\nu (\partial V^+)$: 
	\[
		\big({\mathcal F}_{+ \upharpoonright \nu} T_*\bar{f}\big)(p) = 
			\int_{dS} {\rm d} \mu_{dS} ( x ) \; \overline{f( x )} \; (  (Tx)_+ \cdot  p )^{s^+} \; .
	\]
Using the fact that for $t\in\mathbb{R}$ 
	\begin{equation} 
		\big(-(t\pm i\epsilon)\big)^{s}=e^{\mp i\pi s}(t\pm i\epsilon)^{s},
	\end{equation}
we write 
	\begin{align*}
		\big(Tx\cdot p + i\epsilon\big)^{s^+}  \equiv \big(-x\cdot (-Tp)+ i\epsilon\big)^{s^+}  & = 
		\Big(-\big(x\cdot (-Tp)- i\epsilon\big)\Big)^{s^+}  \\
		& = e^{i\pi s^+}\, \big(x\cdot (-Tp)- i\epsilon\big)^{s^+} \\
		& \equiv
		e^{i\pi s^+}\, \overline{\big(x\cdot (-Tp)+ i\epsilon\big)^{\overline{s^+}}} \; .
	\end{align*}
Now for $0<m< 1/2$ the number $s^+$ is real, hence 
(note that $P= -T$ leaves the light cone invariant) 
	\[
		\big({\mathcal F}_{+ \upharpoonright \nu} T_*\bar{f}\big)(p) =
		e^{i\pi s^+}\, \overline{\big({\mathcal F}_{+ \upharpoonright \nu} {f}\big)(-Tp) } \; . 
	\]
For $\mu \geq 1/2r$, the complex conjugate of $s^+$ is $s^-$, hence
	\[
		\big({\mathcal F}_{+ \upharpoonright \nu} T_*\bar{f}\big)(p) = 
		e^{i\pi s^+}\, \overline{\big({\mathcal F}_{- \upharpoonright \nu}{f}\big)(-Tp) } 
		\equiv 
		e^{i\pi s^+}\, \overline{\big(A_\nu\,{\mathcal F}_{+ \upharpoonright \nu} {f}\big)(-Tp) } \; , 
	\]
where we have used that (see Proposition \ref{propo2.5})
	\begin{align*}
		A_\nu \colon \widetilde{\mathfrak{h}}_\nu (\partial V^+)  & \to 
		\widetilde{\mathfrak{h}}_{- \nu}  (\partial V^+)
		\nonumber \\ 
			 {\mathcal F}_{+ \upharpoonright \nu}  f & \mapsto {\mathcal F}_{- \upharpoonright \nu}    f 
			 \; . 
	\end{align*} 
Comparing this result with the corresponding result for~$\widetilde u_\nu^+ (P)$, and 
inspecting \eqref{tilde-time-reflection}, proves the claim. 
\end{proof}

\subsection{The Wightman two-point function}
Finally, we apply the nuclear theorem to the quadratic form \eqref{nuclear}. It follows that 
there exist tempered distribution  ${\mathcal W}^{(2)} (  x_1 ,  x_2 )$ on $dS\times dS$ such that
\label{twopointpage}
	\[
	\int_{dS \times dS} {\rm d} \mu_{dS} 
	( x_1) {\rm d} \mu_{dS} (  x_2)   { f ( x_1 ) } 
		{\mathcal W}^{(2)} (  x_1 ,  x_2 ) g ( x_2) 
		\doteq \langle [f]  ,   [g] \, \rangle_{{\mathfrak h} (dS)}  \; .
	 \]
The distributions ${\mathcal W}^{(2)} (  x_1 ,  x_2 )$ are called the {\em two-point functions}.  

\goodbreak
\begin{theorem}[Bros and Moschella \cite{BM}, Theorem 4.1 \& 4.2] 
\label{prop:4.1}
The Wightman two-point function ${\mathcal W}^{(2)} ( x_1,  x_2)$ is
a tempered distribution, which is the boundary value 
of the function
	\begin{equation}
		\label{tpf-1}
	\quad {\mathcal W}^{(2)}  ( z_1 ,  z_2) =  c_\nu \frac{{\rm e}^{-   \pi \nu} r } { \pi} 
	\int_\Gamma {\rm d} \mu_\Gamma ( p )
		\;  (  z_1 \cdot  p )^{s^- } (  p \cdot  z_2)^{ s^+} 		 
	\end{equation}
defined and holomorphic for $( z_1,   z_2) \in {\mathcal T}_ + \times {\mathcal T}_-$. 
The boundary values of \eqref{tpf-1} are taken as 
$\Im z_1 \nearrow 0$ and $\Im z_2 \searrow 0$, 
$( z_1,   z_2) \in {\mathcal T}_ + \times {\mathcal T}_-$.
As before, the exponents $s^\pm $ are given by \eqref{dd1} 
and for the measure ${\rm d} \mu_\Gamma ( p )$ one has
	\begin{equation}
	\label{measure-on-circle} 
		{\rm d} \mu_\Gamma ( p ) = \frac{{\rm d} \alpha} {2} ; 
	\end{equation}
in agreement with the normalisation used in \cite[Section 4.2]{BM2}.
\end{theorem}

\begin{remark}
In Minkowski space, after Fourier transformation, the two-point function 
	\[ 
	{\mathcal W}^{(2)}_m (x, y) = \int_{\mathbb{R}^{1+d}} {\rm d} k \; \theta (k^0) \delta (k \cdot k-m^2) \; 
	{\rm e}^{-i    k \cdot x} {\rm e}^{i   k \cdot y}
	\]
is the boundary value of a holomorphic function as $x \in \mathfrak{T}^+$ and 
$y \in \mathfrak{T}^-$ approach the reals. 
For $x= (x_0, \vec x\, )$, $y = (y_0, \vec y\, )$ and $p = \bigl( \sqrt{{\vec k \, }^2+m^2}, \vec k\, \bigr)$ this yields  
	\begin{equation}
	\label{flat2point}
	{\mathcal W}^{(2)}_m (x_0, \vec x, y_0, \vec y \, )
	= \frac{1}{2\pi}\int_{\mathbb{R}^d}  \tfrac{ {\rm d} \vec k}
	{2 \sqrt{ {\vec k \, }^2+m^2} } \; {\rm e}^{i  \vec k ( \vec x- \vec y \, ) - i  (x_0- y_0) \sqrt{{\vec k \, }^2+m^2}}  \; . 
	\end{equation}
\end{remark}

\bigskip
A direct consequence of this result is the one-particle Reeh-Schlieder theorem:

\begin{theorem}[Bros and Moschella \cite{BM},  Proposition 5.4]
\label{oprs}
Let ${\mathcal O}$ be an open region in $dS$. It follows that 
${\mathfrak h}({\mathcal O}) + i {\mathfrak h}({\mathcal O})$
is dense in ${\mathfrak h}(dS)$. 
\end{theorem}

\begin{proof} It is sufficient to show that if $[f] \in {\mathfrak h}^\circ (dS)$ is orthogonal 
to ${\mathfrak h}({\mathcal O}) + i {\mathfrak h}({\mathcal O})$, then~$[f] $ 
is the zero-vector. Consider the complex valued function  
\[
z \mapsto F(z) = c_\nu \frac{{\rm e}^{-   \pi \nu}  r } {\pi} \int_\Gamma {\rm d} \mu_\Gamma ( p )
		\;  (  z \cdot  p )^{s^- } \; {\widetilde f}_\nu (p)\; , \]
which is holomorphic within ${\mathcal T}_ +$. Assume that\footnote{Note that $[g] \in 
{\mathfrak h}^\circ({\mathcal O}) + i {\mathfrak h}^\circ({\mathcal O})$ for  $g \in {\mathcal D}_{\mathbb{C}} ({\mathcal O}) $, 
by linearity of \eqref{ftps}.}
\[
\langle [g], [f] \rangle_{{\mathfrak h}(dS)}
=\int  {\rm d} \mu_{dS} \; g (x) F(x_+) = 0 \qquad \forall g \in {\mathcal D}_{\mathbb{C}} ({\mathcal O}) \; . 
\]
This implies that $F(z)$ vanishes on its boundary (as $\Im z \nearrow 0$) in 
the open region ${\mathcal O}$. It follows that its boundary values vanish on $dS$. This means that 
$[ f ] $ is orthogonal to any vector in ${\mathfrak h}(dS)$; thus it is the zero-vector. 
\end{proof}

\begin{proposition}[Bros and Moschella \cite{BM}, Proposition 2.2] 
\label{prop:4.1b}
The Wightman two-point function ${\mathcal W}^{(2)} ( z_1,  z_2)$ can be analytically continued into the cut-domain 
	\[
		\Delta = dS_{\mathbb{C}} \times dS_{\mathbb{C}} \setminus \Sigma
	\]
where the cut $\Sigma$ is the set 
	\[
		\Sigma = \bigl\{ (z_1, z_2) \in dS_{\mathbb{C}} \times dS_{\mathbb{C}} \mid (z_1 -z_2) \cdot (z_1 -z_2)   \ge 0 \bigr\} \; . 
	\]
Within $\Delta$ the two point function is invariant under the transformation 
	\[
		{\mathcal W}^{(2)} ( z_1 ,  z_2) = {\mathcal W}^{(2)} ( \Lambda_*z_1 ,  \Lambda_* z_2)  	\; , 
		\qquad \Lambda \in SO_{\mathbb{C}} (1,2) \; . 	 
	\]
Moreover, the permuted Wightman function ${\mathcal W}^{(2)} ( x_2,  x_1)$ is the boundary value of 
the analytic function ${\mathcal W}^{(2)} ( z_2 ,  z_1)$ from its domain $\{ (z_1, z_2) \mid z_1 \in {\mathcal T}^+ , 
z_2 \in {\mathcal T}^- \}$. 
\end{proposition}

\begin{proof} Proposition \ref{UIR-FH} guarantees that the distribution ${\mathcal W}^{(2)} ( z_1,  z_2)$ 
is invariant under Lorentz transformations, \emph{i.e.}, if 
$\Lambda \in SO_0 (1,2) $. 
Invariance under  the complexified group then follows by analytic continuation in the group parameter. 
For further details see [Bros and Moschella \cite{BM}, Proposition 2.2]. 
\end{proof}

\begin{remark} The cut $\Sigma$ contains all pairs of points $(x, y) \in dS \times dS$, which are causal to each other. In other words, 
\[ \Sigma \cap (dS \times dS) = \{(x, y) \in dS \times dS \mid y \in \Gamma^+ (x) \cup \Gamma^- (x) \} \; . \]
\end{remark}

The two-point function \eqref{tpf-1} can be expressed in terms of Legendre 
functions. 
	
\begin{proposition}
[Proposition 12 \cite{BM2}]
\label{legendre}
The following integral representation holds for $( z_1,   z_2) \in {\mathcal T}_ + \times {\mathcal T}_-$:
	\begin{equation}
		\label{tpf-2}
	{\mathcal W}^{(2)} ( z_1 ,  z_2) =  c_\nu \, P_{s^+} \left(  \tfrac{z_1 \cdot  z_2}{r^2} \right) 	\; , \qquad m> 0 \; . 	 
	\end{equation}
The boundary values of \eqref{tpf-2} can be taken as 
$\Im z_1 \nearrow 0$ and $\Im z_2 \searrow 0$. 
\end{proposition}

\begin{remark}
\label{Wdomain}
The image of the domain ${\mathcal T}_ + \times {\mathcal T}_-$ by the mapping 
\[ ( z_1,   z_2) \mapsto \tfrac{z_1 \cdot  z_2}{r^2} \]
is\footnote{See \cite[Proposition 3]{BM2}.} the cut-plane $\mathbb{C} \setminus (- \infty, -1]$: consider the following points
	\begin{align*}
		z_1 & = (i r \sin (u_1+i {\rm v}_1), 0, r \cos (u_1 +i{\rm v}_1))\; ,  \\
		z_2 & = (-i r \sin u_2, 0, r \cos u_2) \; ,  \qquad \quad \qquad \qquad  0< u_1, u_2 < \pi \; , \; \; {\rm v}_1\in \mathbb{R} \; . 
	\end{align*}
It follows that 
	\[
		\tfrac{z_1 \cdot  z_2}{r^2} = - \cos (u_1 + u_2 + i {\rm v}_1) \;, \qquad 0<u+u_2 < 2\pi \; , \quad {\rm v}_1 \in \mathbb{R} \; . 
	\]
Thus $\mathbb{C} \setminus (- \infty, -1]$ is contained in the image. 
The fact that  $\mathbb{C} \setminus (- \infty, -1]$ equals the image follows from an argument in the ambient space, see \cite[Proposition 3]{BM2}.
The region $\mathbb{C} \setminus (- \infty, -1]$
is exactly the domain of analyticity of the \emph{Legendre function}\index{Legendre function}~$P_{s^+} $. 
\end{remark} 

\begin{proof} Let $p \in \Gamma= \{ (1, r \cos \alpha, r \sin \alpha) \in \partial V^+ \mid  - \pi \le  \alpha \le \pi \}$. Because of the invariance 
properties of ${\mathcal W}^{(2)} (  z_1,  z_2)$ it is sufficient to consider the choice
$z_1 = (-i r \cosh \beta , 0 , i r \sinh \beta)$, $z_2 = (ir,0,0)$ such that $\tfrac{z_1 \cdot  z_2}{r^2} = \cosh \beta \in \mathbb{R}^+$. It follows 
that\footnote{The second line in the following formula is exactly the one given in \cite[Eq.~(4.18)]{BM}.}
	\begin{align*}
		{\mathcal W}^{(2)} ( z_1 ,  z_2) & = c_\nu  \frac{{\rm e}^{-   \pi \nu} r}{\pi} 
		\int_\Gamma {\rm d} \mu_\Gamma ( p )
		\;  (  z_1 \cdot  p  )^{ s^-} (  p \cdot  z_2)^{ s^+} \\
		&=  c_\nu \int_{-\pi}^\pi \frac{{\rm d} \alpha}{2 \pi}
		\;   (  \cosh \beta + \sinh \beta \sin \alpha  )^{ s^-} \; . 	 
	\end{align*}
In the second equality we have used \eqref{measure-on-circle}. Finally, recall that according to \cite[Eq.~7.4.2]{Lebedev}  
	\[
	P_{s^+} (\cosh \beta ) = \frac{1}{\pi} \int_0^\pi \frac{{\rm d} \alpha }
	{(\cosh \beta - \sinh \beta \cos \alpha)^{s^+ +1} }\; , 
	\]
and $-s^+ -1 = s^-$. 
\end{proof}



\bigskip
Note that
	$
	\overline { {\mathcal W}^{(2)} (  x_1,  x_2) } = {\mathcal W}^{(2)} (  x_2,  x_1)
	$,
and
	\begin{equation}
	\label{comfun}
		{\mathcal W}^{(2)} (  x_1,  x_2) - {\mathcal W}^{(2)} (  x_2,  x_1) 
		=  2 i \Im {\mathcal W}^{(2)} (  x_1,  x_2) \; .
	\end{equation}
The commutator function $2  \Im {\mathcal W}^{(2)} (  x_1,  x_2)$ 
is an anti-symmetric distribution on $dS \times dS$, which satisfies the Klein--Gordon equation in both entries, 
with initial conditions described in (\ref{eqESigma}) and (\ref{eqESigma2}). In fact, for $x_1,x_2$ space-like, this 
is obvious and for $x_1=x_2$ this follows from \cite[page 199]{Lebedev}:
	\[
	P_{s^+} (-1 +i0) - P_{s^+} (-1 -i0) = 2 i \sin s^+ \pi \; . 
	\]
In other words, 
	\[
	{\mathcal W}^{(2)} (  x_+,  x_-) - {\mathcal W}^{(2)} (  x_-,  x_+) =  c_\nu \cdot 2 i \sin s^+ \pi = - i \; . 
	\]
The constants introduced in
\eqref{mass-shell-ft} were chosen to ensure that 
	\[
		\frac{\partial}{\partial x_0} 2\Im {\mathcal W}^{(2)} (  x,  y) = -  \delta_{S^1} \; .  
	\]
As before, $\delta_{S^1}$ is the integral kernel of the unit operator with 
respect to the induced measure on $S^1$.  It follows that 
	\begin{equation}
		\label{comfkt}
			{\mathscr E} (  x_1,  x_2) = 2 \Im {\mathcal W}(  x_1,  x_2) \;  
	\end{equation}
is the kernel of the {\em commutator function} defined in \eqref{integralkernel}.
Equation~(\ref{comfkt}) extends the formula for the propagator given in \eqref{eqPropWedge'} from 
$\mathbb{W}_1$  to  $dS$. To show that $ {\mathscr E} (  x_1,  x_2) $ as given in (\ref{comfkt})
is invariant under the rotations~$R_{0}(\alpha)$, $\alpha \in [0, 2 \pi)$, choose a circle $\Gamma_0$  
on $\partial V^+$ with $p_0 = 1$  in~\eqref{tpf-1}. Rotation invariance of the propagator  
now follows from $ z_1 \cdot  R_0 p = R_0^{-1}z_1 \cdot   p$ and rotation 
invariance of the measure ${\rm d} \mu_{\Gamma_0} = \frac{{\rm d} \alpha}{2}$; see \eqref{lambda2-s}. 

\section{The canonical one-particle Hilbert space}
\label{Sect: canon-HS}

We now define a Hilbert space for functions supported on the time-zero circle~$S^1$.
As we have seen, the two-point function ${\mathcal W}^{(2)} (  x,  y)$ is analytic 
for~$x$ space-like to $y$. For $x= (0, r \sin \psi, r \cos \psi)$ and $y= (0, r \sin \psi, r \cos \psi)$,  
the (minimal) spatial distance $d$ define in \eqref{dlength} is given by
	\[
		d( x ,  x' ) = 
		r \arccos \bigl( - \tfrac{x \cdot x'}{r^2} )= | \psi'-\psi | \,  r \; . 
	\]
Thus ${\mathcal W}^{(2)} (  x,  y) =  c_\nu \, P_{s^+}  ( - \cos ( \psi'-\psi)) $. 		 
Note that the singularity at $\psi = \psi' $ is integrable. 
This suggest the following definition. 

\begin{definition}
\label{zeit-null-Hilbert}
The completion of $C^\infty (S^1)$ with respect to the scalar product
	\begin{equation}
		\langle h , h' \rangle_{\widehat{{\mathfrak h}} (S^1)}  = c_\nu
				\int_{S^1} r \, {\rm d}\psi \int_{S^1} r \, {\rm d}\psi' \; \overline{h(\psi)} \, 
				P_{s^+}\big( - \cos(\psi' - \psi) \big) \, h'(\psi') \;.
				\label{eq:def-scalar-product-2}
	\end{equation}
is a Hilbert space, which we denote by $\widehat {{\mathfrak h}} (S^1)$.
As before, $s^+$ is given by \eqref{dd1}.
\end{definition}

\goodbreak 

\begin{proposition} 
\label{thhm}
The scalar product \eqref{eq:def-scalar-product-2} can be 
\begin{itemize}
\item[$ i.)$] expressed as
	\[
		\langle h , h' \rangle_{\widehat{{\mathfrak h}} (S^1)}  = 
		\bigl\langle h  , \tfrac{1}{2\omega} 
		h' \bigr\rangle_{L^2( S^1, r {\rm d} \psi)} \; ,  
	\]
with $\omega $ a strictly positive self-adjoint operator on $L^2( S^1, r {\rm d} \psi)$ with Fourier coefficients 
	\begin{equation}
	\label{eq:omega-1}
		\widetilde {\omega}(k) = {r}^{-1} \, (k+s^+)
					 \frac{\Gamma \left( \frac{k+s^+}{2} \right)}{ \Gamma \left( \frac{k-s^+}{2} \right)}
			\frac{ \Gamma \left( \frac{k+1-s^+}{2} \right) }{ \Gamma \left( \frac{k+1+s^+}{2} \right)} \; ,
			\quad k\in \mathbb{Z} \; ;
	\end{equation}
\item[$ ii.)$] 	written as
	\begin{equation}
	\langle h , h' \rangle_{\widehat{{\mathfrak h}} (S^1)}  =  r \, \bigl\langle  h,\tfrac{1}{2  |\varepsilon| } 
					\bigl(\coth   \pi |\varepsilon| +  
					\tfrac{(P_1)_*}{\sinh \pi |\varepsilon | }\bigr) |\,  {\rm cos}_\psi| h'  
					\bigr\rangle_{L^{2}(S^1, r {\rm d} \psi)} \; ,  
					\label{eqMagicFormB} 
	\end{equation}
with $\varepsilon^2  =-  (\cos \psi \, \partial_\psi)^2 +(\cos \psi)^{2}\mu^2r^2$.  

\end{itemize}
\end{proposition}

\begin{proof} 
$ i.)$ Set
	\begin{equation}
		P_{s^+}\big( - \cos(\psi' - \psi) \big) =  \sum_{k\in\mathbb{Z}} p_k \frac{{\rm e}^{ik (\psi' - \psi)}}{ \sqrt{2\pi r} } \; . 
			\label{eq:Fourrier-Pmu}
	\end{equation}
This yields
	\begin{align}
	 \langle h,  h'   \rangle_{\widehat{{\mathfrak h}} (S^1)} 
		& =  \sqrt{2\pi r} \, c_\nu  \sum_{k\in\mathbb{Z}} p_k 
				\overline{ \left(   \int_{S^1} r\,  {\rm d}\psi \, h(\psi)  \frac{{\rm e}^{-ik \psi}}{\sqrt{2\pi r} } \right) }
						\left( \int_{S^1}  r\, {\rm d}\psi' \,  h'(\psi')  \frac{{\rm e}^{-ik \psi'}}{\sqrt{2\pi r} } \right)
			\nonumber \\
		& =  \sqrt{2\pi r} \, c_\nu  \sum_{k\in\mathbb{Z}} p_k \, \overline{h_k} \, h'_k \;,
		\label{eq:Jhoitrytfsd}
	\end{align}
where $h_k$ and $h'_k$ are the Fourier coefficients of $h$ and $h'$, respectively. 
Comparing \eqref{eq:def-scalar-product-2}
with \eqref{eq:Jhoitrytfsd} we see that 
\begin{itemize}
\item [$a.)$]
$\omega$ is a diagonal operator w.r.t.~the orthonormal basis 
$\bigl\{e_k \in L^{2}(S^1, \, r {\rm d} \psi) 
	\mid e_k (\psi)=\frac{{\rm e}^{-ik \psi}}{\sqrt{2\pi r}} \; , k\in \mathbb{Z} \bigr\} $.
\item [$b.)$]
the Fourier coefficients 
$\widetilde \omega(k) \doteq \left\langle e_{k}, \; 
\omega e_k \right\rangle_{L^{2}(S^1, \, r {\rm d}\psi)} $ of $\omega$ are given by
	\begin{equation}
		\widetilde \omega(k) =  -\frac{2\sin(\pi s^+)}{ \sqrt{2\pi r} } \; \frac{1 }{r \, p_k} \;,
		\qquad k\in\mathbb{Z} \; . 
		\label{eq:definition-omegak-0}
	\end{equation}
\end{itemize}
Using Proposition \ref{w-coefficients}, we arrive at  \eqref{eq:omega-1}.

\bigskip
\noindent
$ ii.)$ The proof of this statement exploits 
the fact that on the Euclidean sphere $S^2 \subset \mathbb{R}^3$ one has
	\[
		 (- \Delta_{S^2}
		 	+\mu^2)^{-1}
		 	=   \left( - \partial^2_\theta + \varepsilon^2 \right)^{-1} r^2 \cos^{2} \psi  \; ,  
	\] 
using the coordinates
	\begin{equation}
		\label{ps-cord}
		\begin{pmatrix}
			{\tt x}_0 \\
			{\tt x}_1 \\
			{\tt x}_2 
		\end{pmatrix} 
			= 	\begin{pmatrix}
 						r \sin \theta \cos \psi  \\
						r \sin  \psi  \\
						r \cos \theta \cos \psi  \\
				\end{pmatrix} \; , \qquad 0 \le \theta < 2 \pi  \; , \quad - \frac{ \pi}{2} < \psi <  \frac{ \pi}{2} \; ,  
	\end{equation}
while in cartesian coordinates 
	\begin{align*}
		 (- \Delta_{S^2} +\mu^2)^{-1}  ({\tt x}, {\tt y}) & 
		 =  -  \frac{P_{s^+} 
		  \left( - \tfrac{{\tt x}_0 {\tt y}_0 + {\tt x}_1 {\tt y}_1+ {\tt x}_2 {\tt y}_2}{r^2}  \right) } {4 \sin  (\pi s^+ )}  \; .   
	\end{align*}
The details are given in the proof of Lemma \ref{3.9}. 
\end{proof}

\begin{remark}
  In (\ref{eq:symmetry-of-tildeomega}) we will establish that
  $\widetilde \omega(k)=\widetilde \omega(-k)$ for all
  $k\in\mathbb{Z}$.  For the case of the principal series, one has
  $s^\pm=-\frac{1}{2}\mp i\nu$, with~$\nu \in\mathbb{R}^+_0$, and
\[ \Gamma\left(\tfrac{k+\frac{3}{2}-i\nu }{2}\right)
\stackrel{(\ref{eq:gamma-1})}{=} \tfrac{k-\frac{1}{2}-i\nu }{2}\Gamma\left(\tfrac{k-\frac{1}{2}-i\nu }{2}\right) 
\]
implies, from (\ref{eq:omega-1}),
	\begin{align}
          \label{eq:def-tilde-omega}
		\widetilde \omega(k) & 
				 \stackrel{(\ref{eq:gamma-5})}{=} \; {r}^{-1}
						\left( \tfrac{ (k-1/2)^2 + \nu^2 }{4}\right)
							\frac{\bigl|\Gamma \bigl(\frac{k-\frac{1}{2}+i\nu}{2}\bigr) \bigr|^2}{
								\bigl|\Gamma \bigl( \frac{k+\frac{1}{2}+i\nu}{2} \bigr)\bigr|^2}
				\;,
	\end{align}
showing that $\widetilde \omega(k)>0$ for all $k\in
\mathbb{Z}$.  
This positivity property also holds in the case of the complementary
series. In that case, one has $\nu=i\sqrt{\frac{1}{4}-\zeta^2}$, with
$0<\zeta\leq 1/2$.   Hence, $-1/2<s^+\leq 0$.
Since $\widetilde \omega(k)=\widetilde \omega(-k)$ for all
$k\in\mathbb{Z}$, it is enough to consider $k\geq 0$.  We know from (\ref{eq:useful-1}) that
$\widetilde\omega(k) \widetilde\omega(k+1) = r^{-2} k (k+1) + \mu^2 >0$. Hence $\widetilde\omega(k+1)$ and
$\widetilde\omega(k)$ have the same sign and, therefore, in order to prove that
$\widetilde\omega(k)>0$ for all $k\geq 0$, it is enough to establish that
$\omega(0)>0$. But, from (\ref{eq:omega-1}), one has
$$
\widetilde {\omega}(0) = {r}^{-1} \, s^+
					 \frac{\Gamma \left( \frac{s^+}{2} \right)}{ \Gamma \left( \frac{-s^+}{2} \right)}
			\frac{ \Gamma \left( \frac{1-s^+}{2} \right) }{ \Gamma \left( \frac{1+s^+}{2} \right)} 
\; > \; 0
\;,
$$
since $s^+<0$, and since $\Gamma(x)>0$ for all $x>0$ and $\Gamma(x)<0$
for all $x\in (-1, \; 0)$ (one has $\frac{s^+}{2} \in (-1, \; 0)$, but
$\frac{1-s^+}{2}$, $\frac{-s^+}{2} $ and $ \frac{1+s^+}{2}$ are all positive).

\end{remark} 

\begin{remark}
As we shall see in
(\ref{eq:tilde-omega-and-the-traditional-dispersion-relation}), 
for both the principal and complementary series one has
$\frac{1}{2} \big( \widetilde\omega(k) \widetilde\omega(k+1) +\widetilde\omega(k) \widetilde\omega(k-1)  \big)
= \frac{k^2}{{r}^2}  + \mu^2$
and, hence, we conclude that $\widetilde\omega(k)$ behaves 
for large $|k|$ as $\sqrt{\frac{k^2}{{r}^2}  + \mu^2}$, approaching
thus the dispersion relation of the Minkowski space-time.  

\end{remark}

\begin{corollary} 
\label{MagicForm}
The operator $\omega $ on $L^2(S^1, r{\rm d}\psi)$ satisfies the operator identity 
	\begin{equation}
		\label{eqMagicForm}
 		\omega =   | r \,  {\rm cos}_\psi|^{-1}\,|\varepsilon|\,  \bigl(\coth\pi |\varepsilon|
		- \tfrac{(P_1)_*}{\sinh\pi|\varepsilon|} \bigr)  \; .  
	\end{equation}
\end{corollary}

\begin{proof} Comparing $i.)$ and $ii.)$ in Proposition \ref{thhm} yields
	\[
		\omega^{-1} = \tfrac{1}{ |\varepsilon| } \, 
		\bigl(\coth\pi |\varepsilon|+\tfrac{(P_1)_*}{\sinh\pi|\varepsilon|} \bigr) \,
				| r \,  {\rm cos}_\psi| \, ,   
	\]
which is equivalent to~\eqref{eqMagicForm}. 
\end{proof}

\begin{proposition} 
\label{th4.20}
For $h, h' \in {\mathscr D}(\omega)$ have
	\[
	 \bigl\langle \omega r \, h  , \omega  r \, 
	 h' \bigr\rangle_{\widehat{{\mathfrak h}} (S^1)} =   c_\nu 
				\int_{S^1} r\, {\rm d}\psi \int_{S^1} r\, {\rm d}\psi' \; \overline{h(\psi')} \, 
				P'_{s}\big( - \cos(\psi' - \psi) \big) \, h'(\psi) \;.
	\]
Note that $C^\infty(S^1) \subset {\mathscr D}(\omega)$. 
\end{proposition}

\begin{proof} See Proposition \ref{prop:E5}.
\end{proof}

The following definition respects the causal structure of the globally hyperbolic manifold $dS \supset S^1$. 
This will become evident in the sequel. 

\begin{definition}
\label{local-h-hat}
For $I$ a bounded open interval in $S^1$, we define a  {\em real}
subspace of $\widehat{{\mathfrak h}} (S^1)$ by  
	\begin{equation*}
		\widehat{{\mathfrak h}}  ( I) 
		\doteq
		\left\{ h \in \widehat{{\mathfrak h}}  (S^1)   \mid  
		{\rm supp\,} \left( \Re h,    \,  \omega^{-1}\Im h \right) \subset I \times I \right\} \; .
	\end{equation*}
\end{definition}

Clearly, $\widehat{{\mathfrak h}}  (J)$ is in the symplectic complement of $\widehat{{\mathfrak h}}  ( I) $ if $J \subset S^1 \setminus I$. This follows
directly from the definition: 
\[
\Im \langle h, g \rangle_{\widehat{{\mathfrak h}}  (S^1)} =  \langle  \Re h, \omega^{-1} \Im g \rangle_{L^2(S^1, r {\rm d} \psi)} -
 \langle  \omega^{-1}\Im f,  \Re g \rangle_{L^2(S^1, r {\rm d} \psi)} = 0  
\]
for $h\in \widehat{{\mathfrak h}}  ( I) $ and $g \in \widehat{{\mathfrak h}}  ( J)$.  

\section{Time-symmetric and time-antisymmetric test-functions}

The restriction of the Fourier transform to the mass shell 
allows an extension from~${\mathcal D}_\mathbb{R}(dS)$ to distributions  supported 
on the time-zero circle.   
We shall identify~$dS$ with $\mathbb{R}\times S^1$ via the coordinate system 
	\begin{equation}
	\label{w1psitau}
	x (x_0,\psi) = \begin{pmatrix}
				x_0 \\
				\sqrt{r^2 +x_0^2} \; \sin  \psi  \\
				\sqrt{r^2 + x_0^2} \; \cos  \psi 
				\end{pmatrix} \in dS
	\end{equation}
and write $(f\otimes h)(x):=f(x_0)h(\psi)$ for $f\in{\mathcal D}(\mathbb{R})$ 
and $h\in{\mathcal D}(S^1)$ if $x=x(x_0,\psi)$. The metric on $dS$ is
	\[
	g = \frac{1}{r^2+x_0^2} \; {\rm d} x_0 \otimes {\rm d} x_0 
	- ( r^2+ x_0^2) \; {\rm d}  \psi \otimes {\rm d} \psi  
	\]
and $|g|=1$. Thus ${\rm d} \mu_{dS} ( x)= {\rm d} x_0 r {\rm d}  \psi $. 

\begin{theorem} 
\label{st-kappa}
Let $h \in C^\infty_\mathbb{R} (S^1)$ and let $\delta_k$ be a sequence of absolutely 
integrable smooth functions, supported in a neighbourhood of the origin in $\mathbb{R}$, 
approximating the Dirac $\delta$-function. It follows that for all $\mu>0$
the limits 
	\begin{equation}
		\label{tz}
			\lim_{k\to \infty} \| [ \delta_k \otimes h ] \|_{{\mathfrak h} (dS)}
\quad \text{and} \quad 
			\lim_{k \to \infty}  \| [ n \, (\delta_k \otimes g ) ] \|_{{\mathfrak h} (dS)}
	\end{equation}
exist and equal $\| h \|_{ \widehat{{\mathfrak h}}  (S^1)} $ 
and $\|  \omega  r \, g \|_{ \widehat{{\mathfrak h}}  (S^1)}$, respectively. Here $n \, ( \delta_k \otimes h)$ denotes the Lie derivative\footnote{Recall that 
 $\int_{\mathbb{R}^{1+d}} {\rm d} t {\rm d} \vec x  \, \delta' (t) h(\vec x\, )  
{\rm e}^{i  (\omega t -  \vec p \cdot \vec x)} 
	= 	- i \omega
	\int_{\mathbb{R}^d} {\rm d} \vec x \,  h(\vec x\, ) {\rm e}^{-i  \vec p \cdot \vec x}$. }
of $( \delta_k \otimes h)$ along the 
unit normal future pointing vector field~$n$. 
\end{theorem}

\begin{proof} 
According to Proposition \ref{legendre}
	\begin{align*}
	& \lim_{k, k' \to \infty}  \langle [ \delta_k \otimes h ] ,  [ \delta_{k'} \otimes h' ] \rangle_{{\mathfrak h}(dS)} 
	 \\ 	
	 &  \quad = \lim_{k, k' \to \infty}\int_{dS \times dS} {\rm d} \mu_{dS} ( x ) \; {\rm d} \mu_{dS} (  x')  \; 
	 \overline { (\delta_k  \otimes h) ( x) }  \,    
	{\mathcal W}^{(2)} (  x  ,  x' )    (\delta_{k'}   \otimes h') ( x' )   \\
					& \quad = c_\nu 
				\int_{S^1} r\, {\rm d}\psi \int_{S^1} r\, {\rm d}\psi' \; \overline{h(\psi)} \, 
				P_{s^+}\big( - \cos(\psi' - \psi) \big) \, h'(\psi') \\
				&  \quad = \langle h , h'  \rangle_{ \widehat{{\mathfrak h}}  (S^1)}  \;.
	\end{align*}
For the derivative of the delta-function, partial integration yields
	\begin{align}
	\label{deriv}
	& \lim_{k, k' \to \infty}  \langle [ \delta'_k \otimes h ] ,  [\delta'_{k'} \otimes h' ] \rangle_{{\mathfrak h}(dS)} 
	 \nonumber \\ 	
	 &  \quad = \lim_{k, k' \to \infty}\int_{dS \times dS} \kern -.3cm {\rm d} \mu_{dS} ( x ) \; {\rm d} \mu_{dS} (  x')  \; 
	 \overline { (\delta'_k  \otimes h) ( x) }  \,    
	{\mathcal W}^{(2)} (  x  ,  x' )    (\delta'_{k'}   \otimes h') ( x' )   \nonumber \\
					& \quad = c_\nu 
				\int_{S^1 \times S^1}  r^2 {\rm d}\psi {\rm d}\psi' 
				\int_{{\mathbb R} \times {\mathbb R}} {\rm d}x_0 {\rm d}x'_0 \; \delta(x_0)\delta(x'_0)
				\nonumber \\
				& 
				\qquad \qquad \qquad \qquad \qquad \qquad \qquad \times 
				\; \overline{h(\psi)} \, h'(\psi')
				\left( \tfrac{\partial}{\partial x_0 }\tfrac{\partial}{\partial x'_0 } P_{s}\left( \tfrac{x_+ \cdot x'_- }{r^2} \right)  
				\right)
				 \nonumber \\
					& \quad = c_\nu 
				\int_{S^1 \times S^1} r^2 {\rm d}\psi {\rm d}\psi' 
				\int_{{\mathbb R} \times {\mathbb R}} {\rm d}x_0 {\rm d}x'_0 \; \delta(x_0)\delta(x'_0)
				\nonumber \\
				& 
				\qquad \qquad \qquad \qquad  \qquad \times 
				\; \overline{h(\psi)} \, h'(\psi')
				 \tfrac{\partial}{\partial x_0 } \left( P'_{s}\left( \tfrac{x_+ \cdot x'_- }{r^2} \right) \tfrac{\partial}{\partial x'_0 } 
				(x_+ \cdot x'_- )\right)
				 \nonumber \\
					& \quad = c_\nu 
				\int_{S^1} r\, {\rm d}\psi \int_{S^1} r\,  {\rm d}\psi' \; \overline{h(\psi)} \, 
								P'_{s}\big( - \cos(\psi' - \psi) \big) \, h'(\psi') 
				\nonumber  \\
				&  \quad = \langle \omega r \, h , \omega r \, 
				h'  \rangle_{ \widehat{{\mathfrak h}}  (S^1)}  \;.
	\end{align}
The second but last equality follows from 
	\begin{align*}
	 \tfrac{\partial}{\partial x'_0 }  (x \cdot x') 
	 & =  \tfrac{\partial}{\partial x'_0 } \left(x_0 x'_0 - \sqrt{1 + {\tt x_0}} \sqrt{1 + {\tt x'_0}} \cos (\psi- \psi') \right)
	 \\ 	
	 &  = x_0 - \tfrac{{\tt x_0} \sqrt{1 + {\tt x'_0}} }{ \sqrt{1 + {\tt x_0}} }  \cos (\psi- \psi')   
	\end{align*}
and $\tfrac{\partial}{\partial x_0 }\tfrac{\partial}{\partial x'_0 } (x \cdot x') 
	 = 1 - \tfrac{ {\tt x_0}}{ \sqrt{1 + {\tt x_0}} } \tfrac{ {\tt x'_0}}{ \sqrt{1 + {\tt x'_0}} } \cos (\psi- \psi')   $. Thus 
	\[
	 \tfrac{\partial}{\partial x_0 } (x \cdot x')_{ \upharpoonright x_0=x'_0 = 0}
	 =  \tfrac{\partial}{\partial x'_0 } (x \cdot x')_{\upharpoonright x_0=x'_0 = 0}
	 = 0 \; , \quad 
	 \tfrac{\partial}{\partial x_0 }\tfrac{\partial}{\partial x'_0 } (x \cdot x')_{ \upharpoonright x_0=x'_0 = 0}
	= 1 \; . 
	\]
The last equality in \eqref{deriv} follows from Proposition \ref{th4.20}.
\end{proof}	 

\begin{remark} We can now extend the class of distributions considered in Remark \ref{re-3-20}.
The domain of $\mathbb{E}$ contains distributions of the form 
	\begin{align*}
	 f (x) \equiv (\delta \otimes h) (x) &=	\delta (x_0)   h  (\psi ) \;  , 
	 \nonumber \\
	g (x) \equiv (\delta' \otimes h) (x) &= \delta'  (x_0)  h  ( \psi )\;  , 
	\end{align*} 
with $ h\in{\mathcal D}_{\mathbb{R}}(S^1)$ and $x \equiv x ( x_0, \psi)$, using the coordinates introduced in \eqref{w1psitau}. 
The Lorentz invariant measure is   
${\rm d} \mu_{dS} (x_0,\psi) = {\rm d} x_0\,   r   {\rm d} \psi $. Using \eqref{comfkt},
the properties of the convolution \eqref{feb-1} ensure that there exist  
$C^\infty$-solutions $\mathbb{f} ,  \mathbb{g}$ of the Klein--Gordon equation~\eqref{3.25} with Cauchy data:
	\begin{equation}
		\label{sharp-timetestfunction-phipi3}
		\big( \mathbb{f}_{ \upharpoonright S^1}, 
							 (n \, \mathbb{f} )_{\upharpoonright S^1})
  		 = (0, - h) \equiv  (\mathbb{\phi}, \mathbb{\pi}) \; , 
	\end{equation}
and, by partial integration, 
	\begin{equation}
	\label{sharp-timetestfunction-phipi4}
	\big(  \mathbb{g}_{ \upharpoonright S^1}, 
							 (n \, \mathbb{g})_{\upharpoonright S^1})
 		  = (  h, 0) \equiv  (\mathbb{\phi}, \mathbb{\pi}) \; . 
	\end{equation}
All elements in~$\widehat {\mathfrak k}(S^1)$ are linear combinations of the Cauchy data 
arising from  sharp-time testfunctions $f, g$ of the form described above.
\end{remark}

The functions $  \delta \otimes h $ and $\delta' \otimes h$ provide examples of test-functions, which are symmetric  
and anti-symmetric, respectively, under time-reflection. 
In fact, the time-reflection $T$  induces a conjugation\footnote{An anti-linear   
isometry~$C$ satisfying  $C^2= 1$ is called a {\em conjugation}. } $\kappa$ on~${\mathfrak h} (dS)$, as the map 
$f \mapsto T_* f $ leaves the kernel of ${\mathcal F}_{+ \upharpoonright \nu}$  invariant.
The subspace consisting of functions  invariant under time-reflection is 
	\begin{equation}
	\label{kappat}
 	{\mathfrak h}^\kappa (dS) = \{ f \in {\mathfrak h} (dS) \mid \kappa f = f \}  \; . 
	\end{equation}
One can decompose
any testfunction into a symmetric and an anti-symmetric part with respect to time-reflections: 
	\[
	f= \frac{1}{2} ( f + \kappa f) + \frac{1}{2} ( f - \kappa f)  \; , \qquad f \in {\mathfrak h} (dS) \;  .
	\]
For $[f] , [g] \in {{\mathfrak h}^\circ} (dS) \cap {{\mathfrak h}^\kappa} (dS)$, 
polarisation  yields
	\[
		\langle  [f] , [g] \rangle_{{\mathfrak h} (dS)}  
			= \langle [ T_* f ] , [ T_* g ] \rangle_{{\mathfrak h} (dS)} 
			= \langle [g] , [f] \rangle_{{\mathfrak h} (dS)} \; . 
	\]
Since ${{\mathfrak h}^\circ} (dS)$ is dense in ${\mathfrak h} (dS)$,  this implies 
\[ 
\Im \langle f, g \rangle_{{\mathfrak h} (dS)}  = 0 \quad \text{for all $f, g \in {\mathfrak h}^\kappa (dS)$}.
\]
Thus ${\mathfrak h}^\kappa (dS)$ is a real subspace of ${\mathfrak h} (dS)$.

\bigskip
The following result shows that the functions introduced above are already the most general
elements in ${\mathfrak h}^\kappa(dS)$ and its symplectic complement ${\mathfrak h}^\kappa(dS)^{\perp}$, respectively. 

\goodbreak
\begin{corollary}
\label{tztf}
Let $I\subset S^1$ be an open interval (or $I= S^1$) and let $h, g \in {\mathcal D}_\mathbb{R} (I)$. It follows that 
\begin{itemize}
\item[$ i.)$] $\delta \otimes h \in {\mathfrak h}^\kappa(dS) \cap {\mathfrak h}(I)$ and $h \in \widehat{{\mathfrak h}}  ( I)$ is real valued; 
\item[$ ii.)$] $\delta' \otimes g \in {\mathfrak h}^\kappa(dS)^{\perp}\cap {\mathfrak h}(I)$ and $i \omega g \in \widehat{{\mathfrak h}}  ( I)$
has purely imaginary values; 
\item[$ iii.)$] for every time-symmetric function $f \in {\mathcal D}_\mathbb{R} ({\mathcal O}_I)$ there exists a function 
$h \in {\mathcal D}_\mathbb{R} (I)$ such that 
\[ 
[f] = [\delta \otimes h] \; ; 
\] 
\item[$ iv.)$] for every anti-time-symmetric function $e \in {\mathcal D}_\mathbb{R} ({\mathcal O}_I)$ there exists a function 
$g \in {\mathcal D}_\mathbb{R} (I)$ such that 
\[ [e] = [n(\delta \otimes g)] \; . \] 
\end{itemize}
\end{corollary}

\begin{remark}
The statements $iii.)$ and $iv.)$ imply that there is a one-to-one relation between the image of
time-symmetric (time-antisymmetric) testfunctions in ${\mathfrak h}(dS)$ and real (purely imaginary) valued functions 
in~$\widehat{{\mathfrak h}} (S^1)$. The Minkowski space case of this result
is proven in \cite[Vol.~II p.~217]{RS}. It also follows directly by differentiation from Eq.~\eqref{flat2point}.
\end{remark}

\begin{proof} 
$i.)$ By assumption, $h$ is real valued, and we have already seen that $\delta \otimes h \in {\mathfrak h}(dS)$
is equivalent to $h \in \widehat{{\mathfrak h}}  ( I)$; thus we have only to show that $\delta \otimes h \in {\mathfrak h}^\kappa(dS)$. 
This can be achieved by approximation the delta function with a sequence of functions which are all symmetric around the origin. 

$ii.)$ By assumption, $g$ is real valued, and we have already seen that $\delta' \otimes g \in {\mathfrak h}(dS)$
is equivalent to $\omega g \in \widehat{{\mathfrak h}}  (S^1)$. Clearly, the definition of  $\widehat{{\mathfrak h}} (I)$ 
together with $g \in {\mathcal D}_\mathbb{R} (I)$ implies that the function $i \omega g$ takes purely imaginary values
and lies in~$\widehat{{\mathfrak h}} (I)$. Thus it only remains to show that $\delta' \otimes g \in {\mathfrak h}^\kappa(dS)^{\perp}$. 
This can be achieved by approximation the derivative of the delta function with a sequence of functions which are all anti-symmetric 
around the origin. 

$iii.)$  For every time-symmetric function $f \in {\mathcal D}_\mathbb{R} ({\mathcal O}_I)$, the $C^\infty$-solution $\mathbb{f}$ of 
the Klein--Gordon equation is time-symmetric. This implies that  $(n \mathbb{f})_{\upharpoonright S^1} $ vanishes. On the other hand, 
we can define $h \doteq \mathbb{f}_{\upharpoonright S^1} $. It then follows from Theorem \ref{cauchyproblem} that $[f] = [\delta \otimes h]$.

$ iv.)$ For every anti-time-symmetric function $e \in {\mathcal D}_\mathbb{R} ({\mathcal O}_I)$,  the corresponding
$C^\infty$-solution $\mathbb{e}$ of 
the KG equation is anti-time-sym\-metric. This implies that $ \mathbb{e}_{\upharpoonright S^1} $ vanishes. On the other hand, 
we can define $g \doteq - (n \mathbb{e})_{\upharpoonright S^1} $. It then follows from Theorem \ref{cauchyproblem} that 
$[ e]  = [ \delta \otimes h]$. 
\end{proof}

\begin{proposition}
The linear extension of the map
	\begin{equation}
	\label{s1-to-dS}
		h_1 + i \omega r h_2 \mapsto [\delta \otimes h_1] + 
		[\delta' \otimes h_2] 
	\end{equation}
defines a unitary map ${\mathbb U}
\colon \widehat{{\mathfrak h}}  (S^1) \to {\mathfrak h}(dS)$,  which respects the local structure, \emph{i.e.},
	\[ 
		{\mathbb U} \colon \widehat{{\mathfrak h}}  (I) \to {\mathfrak h}({\mathcal O}_I) \; , 
	\]
with ${\mathcal O}_I= I''$ the causal completion of $I \subset S^1$. 
\end{proposition}

\begin{proof} We have already seen  that the image of $[\delta \otimes h_1] + 
		[\delta' \otimes h_2]$ is dense in~${\mathfrak h}(dS)$. Moreover,
	\[
		\| h_1 + i \omega r h_2 \|_{\widehat{{\mathfrak h}}  (S^1)} 
		= \|  [\delta \otimes h_1] + 
		[ \delta' \otimes h_2 ] \|_{{\mathfrak h}(dS)} \; . 
	\]
The result now follows by linear extension. The local part follows from Corollary~\ref{tztf}. 
\end{proof}

\begin{corollary} 
Let $I$ be an open interval in $S^1$. Then
$\widehat{{\mathfrak h}} (I) +i \widehat{{\mathfrak h}} (I)$
is dense in the Hilbert space $\widehat{{\mathfrak h}} (S^1)$.
\end{corollary}

\begin{proof} This result follows directly form Proposition \ref{oprs}:
\[
\overline{\widehat{{\mathfrak h}} (I) +i \widehat{{\mathfrak h}} (I)} 
= {\mathbb U}^{-1} \overline{ {\mathfrak h} ({\mathcal O}_I) +i {\mathfrak h} ({\mathcal O}_I)}
= {\mathbb U}^{-1} {\mathfrak h} (dS) = \widehat{\mathfrak h} (S^1)\; . 
\]
(A direct proof might be based on arguments similar to those given in \cite{V}. However, we have not fully investigated this question.)
\end{proof}

\begin{corollary}
\label{WdSprop}
For any double wedge ${\mathbb W}$, we have ${\mathfrak h} ({\mathbb W})={\mathfrak h} (dS)$.
\end{corollary}

\begin{proof} 
The completion of 
${\mathcal D}_{\mathbb{C}} \left(S^1 \setminus \{ - \frac{\pi}{2} , 
\frac{\pi}{2}\} \right)$ with respect to the scalar product \eqref{eq:def-scalar-product-2}
is $\widehat {{\mathfrak h}} (S^1)$. Thus, by Corollary \ref{tztf}, one has  ${\mathfrak h} ({\mathbb W}_1)={\mathfrak h} (dS)$. 
The general result follows from ${\mathbb W}= \Lambda {\mathbb W}_1$ for some $ \Lambda \in 
SO_0(1,2)$.
\end{proof}

\chapter{Quantum One-Particle Structures}
\setcounter{equation}{0}

Given a classical dynamical system for the Klein--Gordon equation on the de Sitter space (in either the covariant or the canonical formulation)
there is a unique one-particle quantum system associated to it, characterised by 
the \emph{geodesic KMS condition}\index{geodesic KMS condition}.

\section{The covariant one-particle structure} 
\label{new-3.4b}

As we have seen, the Hilbert space ${\mathfrak h} (dS)$ carries an  (anti-)unitary irreducible representation $u$ of $O(1,2)$.

\label{umLambdapage2}

\begin{theorem} 
\label{1PStrucHe} 
Consider the identity map 
\label{covonepartpage}
	\begin{align*}
	\label{K1PStrucHe}
		K  \colon  \qquad {\mathfrak k}(dS) 
		& \rightarrow  {\mathfrak h} (dS) \nonumber \\
		  [f] 
		  & \mapsto   [f]\; , \qquad f  \in {\mathcal D}_{\mathbb{R}} (dS) \; . 
	\end{align*}
It follows that the triple $ \bigl(K , {\mathfrak h} (dS), u \bigr) $
is a {\em de Sitter one-particle structure} for the classical dynamical system 
$\bigl( {\mathfrak k}  (dS) , \sigma, {\mathfrak u}   \bigr)$. In other words, 
\begin{itemize}
\item [$ i.)$] $K$ defines a symplectic map from $({\mathfrak k}  (dS ) , \sigma)$ 
to $\left({\mathfrak h} (dS), \Im \langle \, . \, , 
\, . \, \rangle_{{\mathfrak h} (dS)} \right)$ and the image  
of ${\mathfrak k}  (dS )$ is dense in ${\mathfrak h} (dS)$;
\item [$ ii.)$] 
there exists a (anti-) unitary representation $u$ of $O(1,2)$ on ${\mathfrak h} (dS)$ satisfying 
	\begin{equation*} 
		\label{eqUHe}
		 u (\Lambda)\circ K = K \circ  {\mathfrak u} (\Lambda)  \;, 
		 \qquad \Lambda \in O(1,2)\; ;  
	\end{equation*}
\item [$ iii.)$]  for any wedge $W$, the \emph{geodesic KMS condition} holds: 
	\begin{equation} 
		\label{5.1}
		K {\mathfrak k}(W)  \subset {\mathscr D} \bigl( u (\Lambda_W ( i \pi)) \bigr) \; ,
	\end{equation}
and
	\begin{equation}
		\label{5.4}
		u (\Lambda_W ( i \pi)) [f] 
		= u (\Theta_{W}) [f] \; , 
		\qquad [f] \in K {\mathfrak k}(W) \; , 
	\end{equation}
$\Theta_{W}$ is the reflection on the edge of the 
wedge $W$.
\end{itemize}
\end{theorem}

\begin{proof} $K$ is well-defined, as $\ker {\mathcal F}_{+ \upharpoonright \nu} = {\ker {\mathbb P}}$.

\begin{itemize}
\item[$ i.)$] It follows from 	
	\begin{equation*}
	\label{comfkt-3}
	{\mathscr E} (  x_1, x_2) = 2 \Im {\mathcal W}^{(2)}_m(  x_1,  x_2) \;  
	\end{equation*} 
 that $K$ is symplectic; see \eqref{comfkt}. Recall that ${\mathfrak h}^\circ (dS) 
$ is dense in ${\mathfrak h} (dS)$;

\item[$ ii.)$]  Both ${\mathfrak u} (\Lambda)$ and $u (\Lambda)$ are induced by the pullback action 
on the test functions: for $f \in {\mathcal D}_{\mathbb{R}} (dS)$
	\begin{align*}  
		\qquad K \circ  {\mathfrak u} (\Lambda)\; [f] &
		= K \circ {\mathbb P} (\Lambda_* f)
		=  [\Lambda_* f] 
		\nonumber \\
		&= u (\Lambda) [f]  
		= u (\Lambda)\circ K \, [f]  \;, 
		 \quad \Lambda \in O(1,2)\; .  \quad
	\end{align*} 
The second but last identity follows from the definition of the Fourier-Helgason transform 
(see \eqref{mass-shell-ft}), and Proposition \ref{UIR-FH}.

\item[$ iii.)$]  One can read of from \eqref{eqBooW} that
	\[
		\Lambda_1( t+i \pi ) = \Lambda_1( t) TP_1 \; . 
	\]
Now, let $f, g \in  {\mathcal D}_{\mathbb{R}} (W_1)$. It 
follows that the map 
	\begin{align*}
	 	t \mapsto & \langle [f] ,
				u (\Lambda_W ( t)) [g] \rangle_{\mathfrak{h}(dS)} \\
				& = \langle [f] , [\Lambda_1 (t)_* g] \rangle_{\mathfrak{h}(dS)} 
		\\
		& =
			\int_{dS \times dS} {\rm d} \mu_{dS} ( x_1) {\rm d} \mu_{dS} (  x_2)   { f ( x_1 ) } 
				{\mathcal W}^{(2)} (  x_1 ,  x_2 ) g ( \Lambda_1^{-1}(t) x_2) \\
				& =
			\int_{dS \times dS} {\rm d} \mu_{dS} ( x_1) {\rm d} \mu_{dS} (  x_2)   { f ( x_1 ) } 
				{\mathcal W}^{(2)} (  x_1 ,  \Lambda_1 (t) x_2 ) g (  x_2) 
	\end{align*}
allows an analytic continuation into the strip $\{ t \in \mathbb{C} \mid 0< \Im t <  \pi \, r \}$ with continuous boundary values. 
The boundary values are 
	\begin{align*}
	 	 \langle [f] , & \; [\Lambda_1 (i \pi r)_* g] \rangle_{\mathfrak{h}(dS)} \\
		& =
			\int_{dS \times dS} {\rm d} \mu_{dS} ( x_1) {\rm d} \mu_{dS} (  x_2)   { f ( x_1 ) } 
				{\mathcal W}^{(2)} (  x_1 ,  TP_1 x_2 ) g ( x_2) \\
				& =
			\int_{dS \times dS} {\rm d} \mu_{dS} ( x_1) {\rm d} \mu_{dS} (  x_2)   { f ( x_1 ) } 
				{\mathcal W}^{(2)} (  x_1 ,  x_2 ) g ( TP_1 x_2) \\
			& = \langle [f] ,u (TP_1) [g] \rangle_{\mathfrak{h}(dS)} \; . 
	\end{align*}
This identity holds for the total set of vectors $\{ [f] \in \mathfrak{h}(dS) \mid f \in  {\mathcal D}_{\mathbb{R}} (W_1)$. It follows 
that the identity \eqref{5.4} holds. 
\end{itemize}
\end{proof}

The space ${\mathfrak h}(W)$ is a real subspace in ${\mathfrak h}(dS)$. Moreover, ${\mathfrak h}(W)
+ i {\mathfrak h}(W)$ is dense in ${\mathfrak h}(dS)$ and ${\mathfrak h}(W)  \cap i {\mathfrak h}(W)  = \{0\}$. 
Thus one can define, following Eckmann and Osterwalder \cite{EO} (see also \cite{LRT}), a closeable operator 
	\begin{equation} 
		\begin{matrix}
			s_{\scriptscriptstyle W} \colon & {\mathfrak h}(W) & 
				+ & i {\mathfrak h}(W)  
				& \to&  {\mathfrak h}(W) & 
				+ & i {\mathfrak h}(W)  \\
				& f & + & i g & \mapsto & f & - & i g
						\end{matrix} \; \; . 
	\end{equation}
The polar decomposition of its closure $\overline {s}_{\scriptscriptstyle W} = j_{\scriptscriptstyle W} \delta_{\scriptscriptstyle W}^{1/2}$
provides 
\begin{itemize}
\item[---] an anti-unitary involution (\emph{i.e.}, a conjugation)
	\begin{equation} 
		\begin{matrix}
		j_{\scriptscriptstyle W} \colon & {\mathfrak h}\oplus \overline{{\mathfrak h}}  & \to&  {\mathfrak h}\oplus \overline{{\mathfrak h}}  \\
			& f  \oplus   g & \mapsto &    \overline{g}  \oplus   \overline{ f }
		\end{matrix}  \; \;  ; 
	\end{equation}
\item[---] a complex linear, positive operator $\delta^{1/2}_{\scriptscriptstyle W} $.  
\end{itemize}

\begin{theorem}[One-particle Bisognano-Wichmann theorem]
\label{one-p-BW}
The one-particle Tomita operator $s_{\scriptscriptstyle W_1}$  has the polar decomposition 
	\begin{equation} 
		\label{eqSPolarDecomp}
			s_{\scriptscriptstyle W_1} = u (TP_1)\, u\bigl(\Lambda_1( i\pi) \bigr) \; .
	\end{equation}
\end{theorem}

\begin{proof} According to Theorem \ref{1PStrucHe} $iii.)$ we have 
	\[
		u (TP_1) ( [f] + i [g]) = u (\Lambda_1 ( i \pi)) ([f] +i [g])\; , 
		\qquad [f] , [g] \in {\mathfrak h}^\circ (W_1) \; , 
	\]
Since $u (TP_1)$ is idempotent and anti-linear, this implies   
	\[ 
		( [f] -  i [g]) = u (TP_1) u (\Lambda_1( i\pi)) ( [f] + i [g]) \; , 
		\qquad [f] , [g] \in {\mathfrak h}^\circ (W) \; . 
	\]
The left hand side coincides with $s_{\scriptscriptstyle W_1} ( [f] + i [g])$. 
The space ${\mathfrak h}^\circ (W)$ is invariant under $u (\Lambda_1(t))$ and therefore is a core  for
$u(\Lambda_1(i\pi))$. Therefore the above equation implies that
$s_{\scriptscriptstyle W_1}$ has the polar decomposition~\eqref{eqSPolarDecomp}. 
\end{proof}

\color{black}

\begin{corollary} The
the quadrupel $ \bigl(K , {\mathfrak h} (\mathbb{W}), u(\Lambda_W (\tfrac{.}{r})), 
u (\Theta_{W})\bigr) $, with $W$ an arbitrary wedge, forms a double $2\pi r$-KMS one-particle 
structure for the classical double dynamical system 
$\bigl( {\mathfrak k}  (\mathbb{W} ) , \sigma, {\mathfrak u} (\Lambda_W (\tfrac{.}{r})),  {\mathfrak u} (\Theta_{W})  \bigr)$ in the 
sense of~\ref{dbops}.
\end{corollary}

\begin{proof}
We verify the properties listed in~\ref{dbops}: 
\begin{itemize}
\item [$ a.)$] ${\mathfrak h} (\mathbb{W})$ is a complex Hilbert space;  in fact, it equals ${\mathfrak h} (dS)$, 
see Proposition~\ref{WdSprop}.
\item [$ b.)$] The map $K \colon  {\mathfrak k} (\mathbb{W}) \to {\mathfrak h} (\mathbb{W})$ 
is  real linear and symplectic (Theorem~\ref{1PStrucHe}~(i)). Moreover,   
	\[
		K {\mathfrak k} (W) 
		+ i \, K  {\mathfrak k} (W) 
		= {\mathfrak h}^\circ (W) +  i \, {\mathfrak h}^\circ (W)  
	\]
is dense in ${\mathfrak h} (dS) = {\mathfrak h} (\mathbb{W})$. This follows from Theorem \ref{oprs}. 

\item [$ c.)$] 
$t \mapsto u (\Lambda_W(t)) $ is a strongly continuous  one-parameter group of unitaries, and 
	\begin{equation}
		u (\Lambda_W(t)) \circ K  
		= K  \circ   {\mathfrak u} (\Lambda_W(t) )   \; . 
		\label{eqUeps2}
	\end{equation}
This is a special case of (ii). By construction, the generator of the boost $ t \mapsto \Lambda_W(t)$ is unitarily equivalent to 
\eqref{boostgenerator}. It has no zero eigenvalue and according to \eqref{5.1}
	\[
	 \bigl( K {\mathfrak k} (W) + i \,   K {\mathfrak k} (W) \bigr) 
		\subset {\mathscr D} \bigl( u (\Lambda_W(i \pi)) \bigr) \; ; 
	\]
\item [$d.)$] $u (\Theta_{W})$ is a conjugation, and 
	\[
		u (\Theta_{W}) \circ K  =
                K  \circ  {\mathfrak u} (\Theta_{W}) \; .  
	\] 
The pre-Bisognano-Wichmann condition (see \cite[p. 75]{Kay3}) holds:  
	\[
		u (\Lambda_W(i \pi)) \, K [f] =  u (\Theta_{W}) \, 
		 K [f] \; , \qquad [f] \in {\mathfrak k}(W_1) .
	\]
Both properties follow from Theorem \ref{1PStrucHe} $iii.)$ and the fact that $\Theta_{W} \in O(1,2)$. 
\end{itemize}
\end{proof}

\goodbreak
\section{One-particle structures with positive and negative energy}
\label{posneg}

Let $\widehat {\mathfrak  d} (S^1)$ be the completion of ${\mathcal D}_{\mathbb{C}} \left(S^1 \setminus \{ - \frac{\pi}{2}, 
\frac{\pi}{2}\} \right)$ with respect to the
scalar product 
\label{hfdpage}
	\begin{equation}
		\label{eqSkalProd}
		\langle h_1,h_2 {\rangle}_{\widehat {\mathfrak  d} (S^1)}
		 \doteq \langle h_1,(2 | \varepsilon| )^{-1}h_2 
		\rangle_{L^2(S^1, |\cos \psi |^{-1} r {\rm d} \psi)} \; ,   
	\end{equation}
for $ h_1, h_2 \in {\mathcal D}_{\mathbb{C}} \left(S^1 \setminus \{ - \frac{\pi}{2}, \frac{\pi}{2}\} \right)$.  Let $\widehat {\mathfrak  d} (I_\pm)$ be the completion of 
${\mathcal D}_{\mathbb{C}} (I_\pm)$ with respect to the scalar product~\eqref{eqSkalProd}. Then
	\[
		\widehat {\mathfrak  d} (S^1) = \widehat {\mathfrak  d} (I_+) \oplus \widehat {\mathfrak  d} (I_-) \; . 
	\]
This follows from Eq.~\eqref{4.9} and Lemma \ref{Lm3.5}. 

\begin{proposition} 
\label{OnePartdS} 
Let $\widehat{K}_\infty \colon \widehat {\mathfrak k}  (S^1) \to \widehat {\mathfrak  d} (S^1) $ be the map given by 
	\begin{equation}
		\label{eqKdotf} 
			\bigl( \widehat{K}_\infty (\mathbb{\phi}, \mathbb{\pi}) \bigr) (\psi)  
				\doteq   \cos \psi \; \mathbb{\pi} (\psi)-i
                                \,(\varepsilon \mathbb{\phi}) (\psi)  \; . 
	\end{equation}
Then $\bigl( \widehat{K}_\infty,\widehat {\mathfrak  d} (S^1) , {\rm e}^{it\varepsilon } \bigr)$
forms a one-particle structure for the classical 
dynamical system $\bigl(\, \widehat{\mathfrak k} (S^1) ,\widehat \sigma, \widehat   {\mathfrak u}(\Lambda_1(t))\bigr)$.
\end{proposition}

\begin{proof} The map \eqref{eqKdotf} is well-defined  for 
$(\mathbb{\phi}, {\mathbb \pi}) \in C^\infty (S^1) \times C^\infty (S^1)$. 
This follows from the fact that $\varepsilon^2$ is a differential operator, which satisfies 
	\[ 
		(\varepsilon^2 h)  (\psi \pm \tfrac{\pi}{2}) = O(\psi)\quad \text{for $h\in C^\infty (S^1)$} \; , 
	\]
just as $\cos (\psi\pm\frac{\pi}{2})$.

Use  that  $\varepsilon|\varepsilon|^{-1}=\,  {\rm cos}_\psi|\,  {\rm cos}_\psi|^{-1}$ equals  $1$ 
on $L^2( I_+)$ and $-1$ on $L^2( I_-)$ to show
	\begin{align*} 
		& \qquad 2 \Im \langle   \, \widehat{K}_\infty (\mathbb{\phi}_1, {\mathbb \pi}_1),
		\widehat{K}_\infty (\mathbb{\phi}_2, {\mathbb \pi}_2)  {\rangle}_{\widehat {\mathfrak  d} (S^1)}
		\nonumber \\ 
			&\qquad \qquad =  2\Im  \bigl\langle   \,  {\rm cos}_\psi \,
                         {\mathbb \pi}_1-i\varepsilon \, \mathbb{\phi}_1    \, , \,  \tfrac{1}{2 |\varepsilon| } \;
			( \,  {\rm cos}_\psi \, {\mathbb \pi}_2-i \varepsilon  \,
                         \mathbb{\phi}_2  ) \bigr\rangle_{L^2(S^1,   |\cos \psi|^{-1}  r {\rm d} \psi)}  \qquad \qquad \nonumber \\  
  			& \qquad \qquad =   \langle \mathbb{\phi}_1   , {\mathbb \pi}_2  \rangle_{L^2(S^1,    r {\rm d} \psi)} 
				 -    \langle    {\mathbb \pi}_1   ,   \mathbb{\phi}_2    \rangle_{L^2(S^1, r {\rm d} \psi)}     \; . \qquad \qquad 
	\end{align*} 
Thus $\widehat{K}_\infty$ is symplectic.

Moreover, $\widehat{K}_\infty$ intertwines $\widehat {\mathfrak u} (\Lambda_1(t))$ and
${\rm e}^{it\varepsilon}$:  according to  (\ref{varphi-pi-t})  
	\[
		\widehat  {\mathfrak u} (\Lambda_1(t)) (\mathbb{\phi},{\mathbb \pi})=(\mathbb{\phi}_t,{\mathbb \pi}_t)
	\]
with
	\begin{align*} 
		\mathbb{\phi}_t (\psi) 
			&= \big(\cos(\varepsilon t)\mathbb{\phi} 
				- \sin( \varepsilon t) \varepsilon^{-1} \,  {\rm cos}_\psi \;  {\mathbb \pi}\big) (\psi)  \nonumber \\  
		{\mathbb \pi}_t (\psi) 
			&=  \cos^{-1} (\psi)   \big(\varepsilon\sin(\varepsilon t)\mathbb{\phi}  
			+ \cos(\varepsilon t) \,  {\rm cos}_\psi \;  {\mathbb \pi} \big)(\psi)  \; . 
	\end{align*} 
Consequently, 
	\begin{align*} 
		\widehat{K}_\infty \circ \widehat  {\mathfrak u} (\Lambda_1(t)) (\mathbb{\phi},{\mathbb \pi})
			& =  \,  {\rm cos}_\psi \, {\mathbb \pi}_t - i\varepsilon  \,
                \mathbb{\phi}_t   \nonumber \\  
			& =  \Big(\varepsilon\sin(\varepsilon t)\; \mathbb{\phi}   + \cos(\varepsilon t) \;
				\,  {\rm cos}_\psi \, {\mathbb \pi}  \Big)  \nonumber  \\ 
			&   \qquad \qquad -i \varepsilon  \, \Bigl(  \cos(\varepsilon t) \; \mathbb{\phi}   
					- \varepsilon^{-1} \sin( \varepsilon t) \; \,  {\rm cos}_\psi \,  {\mathbb \pi}    \Bigr)   \nonumber \\  
			& =   \cos (\varepsilon t)  \bigl( \,  {\rm cos}_\psi \, {\mathbb \pi} -i\varepsilon \, \mathbb{\phi}  \bigr)
			+ i  \sin (\varepsilon t)  \bigl(\,  {\rm cos}_\psi \, {\mathbb \pi} -i\varepsilon \, \mathbb{\phi}  \bigr)  \nonumber \\
			&= {\rm e}^{it\varepsilon} (\, \widehat{K}_\infty (\mathbb{\phi}, {\mathbb \pi}))   \; . 
	\end{align*}
Since $|\cos \psi|^{-1}$ is finite  away from the boundary of~$I_+$,  the 
set~$\widehat{K}_\infty \bigl( \, \widehat {\mathfrak k}  (S^1 )  \bigr) $ is dense 
in $\widehat {\mathfrak  d} (S^1)$. 
\end{proof}

\begin{proposition} \label{restrictedOnePartdS} 
Consider the one-particle structure $\bigl(\widehat{K}_\infty,\widehat {\mathfrak  d} (S^1) , 
{\rm e}^{it\varepsilon} \bigr)$.
It follows that 
\begin{itemize} 
\item [$ i.)$] 
the restricted structure $ \bigl(\widehat{K}_\infty, \widehat {\mathfrak  d} (I_+) ,
{\rm e}^{it\varepsilon_{\upharpoonright I_+ }}\bigr)$
is a positive energy one-particle  structure for $\bigl(\, \widehat {\mathfrak k} (I_+) ,
\widehat \sigma, \widehat  {\mathfrak u} (\Lambda_1(t)) \bigr)$,
 \emph{i.e.}, 
\begin{itemize}
\item[---] the group 
$t \mapsto {\rm e}^{it\varepsilon_{\upharpoonright I_+ }}$ has a  
positive  generator   $\varepsilon_{\upharpoonright I_+ } \ge 0$;
\item[---] 
$  \widehat{K}_\infty\, \widehat {\mathfrak k}   (I_+)$ is dense 
in $\widehat {\mathfrak  d} (I_+)$. 
\end{itemize}
\item [$ ii.)$] 
the restricted structure $\bigl(\widehat{K}_\infty, \widehat {\mathfrak  d} (I_-) , 
{\rm e}^{it\varepsilon_{\upharpoonright I_- }}\bigr)$
is a negative energy one-particle  structure for $\bigl(\,\widehat {\mathfrak k} (I_-) ,
\widehat \sigma, \widehat  {\mathfrak u} (\Lambda_1(t))\bigr)$, 
\emph{i.e.}, 
\begin{itemize}
\item[---] the group 
$t \mapsto {\rm e}^{it\varepsilon_{\upharpoonright I_- }}$ has a 
 negative generator $\varepsilon_{\upharpoonright I_- } \le 0$;
\item[---] 
$  \widehat{K}_\infty\, \widehat {\mathfrak k}   (I_-)$ is dense in 
$\widehat {\mathfrak  d} (I_-)$. 
\end{itemize}
\item [$ iii.)$] the parity and time-reflections are represented (anti-) unitarily, namely
	\begin{align}
		\label{5.46}
		\widehat{K}_\infty \circ\widehat {\mathfrak u} (P_1) &= - (P_1)_*  \circ \widehat{K}_\infty  \; ;  
		\\
		\label{5.47}
		\widehat{K}_\infty \circ\widehat {\mathfrak u} (T) &= - C  \circ \widehat{K}_\infty  \; , 
	\end{align}
where 
	\[
		(Ch)(\psi)\doteq \overline{h(\psi)} \; , \qquad h \in C^\infty (S^1) \; ,
	\]
extends to $ \widehat {\mathfrak  d} (S^1)$;
\item [$ iv.)$] zero is not an eigenvalue of $\varepsilon$; thus the one-particle structures given in $i.)$ and~$ii.)$ are 
unique, up to unitary equivalence. 
\end{itemize}
\end{proposition}


\begin{proof} $i.)$ and $ ii.)$ follow from \eqref{vaepsdef} as well as the final statement in the proof of 
Proposition \ref{OnePartdS}. For $ iii.)$ use that $({P_1})_*$ anti-commutes with~$\varepsilon$ and
with the multiplication operator $\,  {\rm cos}_\psi$. Eq.~\eqref{5.46} follows from
	\[
		\widehat{K}_\infty \bigl( (P_1)_*\mathbb{\phi}, (P_1)_*{\mathbb \pi} \bigr) 
		= -(P_1)_*\,  {\rm cos}_\psi {\mathbb \pi} + i (P_1)_* \varepsilon \mathbb{\phi}
	\] 
and Eq.~\eqref{5.47} follows from 
$\widehat{K}_\infty (\mathbb{\phi}, -{\mathbb \pi}) = - \overline{(\,  {\rm cos}_\psi {\mathbb \pi} - i \varepsilon \mathbb{\phi})} $. 
Finally, $iv.)$ follows from Lemma \ref{Lm3.5} and Proposition \ref{Kay Th}.
\end{proof}

\begin{proposition}   \label{Jeins}
The operator 
	$
		j \doteq C (P_1)_*
	$
acting on $\widehat {\mathfrak  d} (S^1)$ is an anti-unitary involution (\emph{i.e.}, a conjugation), which implements
the $P_1 T$ transformation and anti-commutes with the generator
$\varepsilon$ of the boosts $t \mapsto \Lambda_1 (t)$: 
	\begin{align} 
		j \circ \widehat{K}_\infty &= \widehat{K}_\infty \circ \widehat {\mathfrak u} (P_1 T) \; ,\label{eqJKKP1T}  \\
		j\, \varepsilon &=-\varepsilon \, j \; .  \label{eqJeps} 
	\end{align}
\end{proposition}

Note that Eq.~\eqref{eqJeps} and anti-linearity imply 
	\begin{equation}
	\label{eqJUt} 
		{\rm e}^{it\varepsilon} \; j=j  \; {\rm e}^{it\varepsilon}\qquad\text{and}
	\qquad j \, |\varepsilon| =|\varepsilon| \, j \; . 
	\end{equation}
	
\begin{proof} 
Clearly $({P_1})_*$ commutes with $\varepsilon^2$ and hence with
its positive square root $|\varepsilon|$. Now~$\varepsilon$ may be
written 
	\[
		\varepsilon = |\varepsilon|(\chi_{I_+}- \chi_{I_-}) \; ,
	\]
where $\chi_{I_\pm}$ denotes multiplication by the characteristic
function of $I_\pm$. Since 
	\[
		({P_1})_* \circ \chi_{I_\pm}
		= \chi_{I_\mp}\circ ({P_1})_*
	\] 
and pointwise complex conjugation commutes with $\varepsilon$, this proves~\eqref{eqJeps}. 
Equation~\eqref{eqJKKP1T}
follows from Proposition \ref{restrictedOnePartdS} $iii.)$. 
\end{proof}

\section{One-particle KMS structures}
\label{KMSops}

Define the {\em real} linear map $\widehat{K}_\beta \colon\widehat {\mathfrak k}(S^1)\to\widehat {\mathfrak  d} (S^1) $, $\beta >0$,  by 
	\begin{equation*}
		\widehat{K}_\beta (\mathbb{\phi},{\mathbb \pi}) \doteq
			\big((1+\rho_\beta )^{{\frac{1}{2}}}+\rho_\beta^{\frac{1}{2}}\, j \big) \;
 			\widehat{K}_\infty (\mathbb{\phi},{\mathbb \pi}) 
	\end{equation*}
with 
	\[ 
		\rho_\beta \doteq \frac{{\rm e}^{-\beta|\varepsilon|}}{ 1+ {\rm e}^{-\beta|\varepsilon|} } 
		\quad \text{and} \quad 
		(1+\rho_\beta) = \frac{1}{ 1+ {\rm e}^{-\beta|\varepsilon|} } \; .
	\]
The domain of $\rho_\beta$ and $(1+\rho_\beta)$ contains 
${\mathscr D} (|\varepsilon|^{1/2})$, 
as can be seen from the elementary bound \cite[\S A2]{Kay1}
	\[
		 0 < 
		\frac{ {\rm e}^{-\lambda} }{ 1+ {\rm e}^{-\lambda} } \; , \; 
		\frac{ 1 }{ 1+ {\rm e}^{-\lambda} } \le max (1, \lambda^{1/2}) \; , 
		\qquad \lambda \in \mathbb{R}^+ \; . 
	\]
Note that $\widehat{K}_\beta$ is {\em not} the Araki-Woods map $K_{\scriptscriptstyle \rm AW}$
discussed in~\ref{Kbeta}, as $K_{\scriptscriptstyle \rm AW}$ would map 
$\widehat {\mathfrak k}(I_+)$ to $\widehat {\mathfrak  d} (I_+) \oplus \overline{\widehat {\mathfrak  d} (I_+)}$.

\begin{proposition} 
\label{OnePart}  
The quadruple $\bigl( \widehat{K}_\beta , \widehat {\mathfrak  d} (S^1), {\rm e}^{i \frac{.}{r} \varepsilon},  j \bigr) $
is a double $\beta r $-KMS one-particle structure for the classical 
double dynamical system 
\[ \bigl(\, \widehat {\mathfrak k}  (S^1 ) , \widehat \sigma, 
\widehat  {\mathfrak u} (\Lambda_1 (\tfrac{.}{r})), \widehat {\mathfrak u} (P_1 T)  \bigr)
\]
in the sense of Kay, see~\ref{dbops}.  
\end{proposition}

\begin{proof}
Let $\widehat{\mathbb{\Phi}}_i=(\mathbb{\phi}_i,{\mathbb \pi}_i) \in \widehat {\mathfrak k}(S^1)$, $i=1,2$ and denote the scalar 
product in $\widehat {\mathfrak  d} (S^1)$ just by~$\langle \, . \, ,\, . \,  \rangle$.  Then 
	\begin{align} 
	\Im	\langle \, \widehat{K}_\beta \widehat{\mathbb{\Phi}}_1,
	\widehat{K}_\beta \widehat{\mathbb{\Phi}}_2\rangle &= 
	\Im \big\{	\langle \, \widehat{K}_\infty \widehat{\mathbb{\Phi}}_1,
	(1+\rho_\beta) \widehat{K}_\infty \widehat{\mathbb{\Phi}}_2\rangle 
		+ \langle \, j  \widehat{K}_\infty \widehat{\mathbb{\Phi}}_1, 
		\rho_\beta j \widehat{K}_\infty \widehat{\mathbb{\Phi}}_2\rangle \big\} 
		\nonumber  \\
	&= \Im \big\{ \langle \, \widehat{K}_\infty \widehat{\mathbb{\Phi}}_1,
	(1+\rho_\beta) \widehat{K}_\infty \widehat{\mathbb{\Phi}}_2\rangle 
		+ \overline{ \langle \, \widehat{K}_\infty \widehat{\mathbb{\Phi}}_1,
		 \rho_\beta \widehat{K}_\infty \widehat{\mathbb{\Phi}}_2 \rangle}\big\} \nonumber  \\
		&= 
		\Im\big\{ \langle \, \widehat{K}_\infty \widehat{\mathbb{\Phi}}_1,(1+\rho_\beta) 
		\widehat{K}_\infty \widehat{\mathbb{\Phi}}_2\rangle 
		- \langle \, \widehat{K}_\infty \widehat{\mathbb{\Phi}}_1, \rho_\beta \widehat{K}_\infty 
		\widehat{\mathbb{\Phi}}_2\rangle\big\} \nonumber \\
		&= \Im\langle \,\widehat{K}_\infty \widehat{\mathbb{\Phi}}_1, \widehat{K}_\infty \widehat{\mathbb{\Phi}}_2\rangle
			=\frac{1}{2}\,\widehat\sigma(\widehat{\mathbb{\Phi}}_1,\widehat{\mathbb{\Phi}}_2) \; . 
		\label{eqImbeta} 
	\end{align}
Now verify the properties listed in Definition~\ref{dbops}: 
\begin{itemize}
\item [$i.)$] $\widehat {\mathfrak  d} (S^1)$ is a complex Hilbert space;  
\item [$ii.)$] the map $\widehat{K}_\beta \colon \widehat {\mathfrak k}
 \left(S^1\right) \to \widehat {\mathfrak  d} (S^1)$ is  real linear and symplectic, as can be 
seen from Eq.~\eqref{eqImbeta}.  
Moreover,   
	\begin{align*}
		\qquad & 
		\widehat{K}_\beta \;  \widehat {{\mathfrak k} } (I_+)   
		+ i \; \widehat{K}_\beta \; \widehat {{\mathfrak k} } (I_+) 
		\nonumber \\ 
		&\quad =    \big((1+\rho_\beta)^{{\frac{1}{2}}}+\rho_\beta^{\frac{1}{2}}\,j \big)  
		\widehat{K}_\infty\, \widehat {{\mathfrak k} } (I_+) 
			+ i \big((1+\rho_\beta)^{{\frac{1}{2}}}+\rho_\beta^{\frac{1}{2}}\,j \big) 
			\widehat{K}_\infty\, \widehat {{\mathfrak k} } (I_+)
		\nonumber \\ 
		&\quad =    (1+\rho_\beta)^{{\frac{1}{2}}}  \bigl( \widehat{K}_\infty \, \widehat {{\mathfrak k} } (I_+) 
		+i \widehat{K}_\infty \, \widehat {{\mathfrak k} } (I_+)  \bigr)
		\nonumber \\ 
		&	\qquad \qquad	+ \rho_\beta^{\frac{1}{2}}  \bigl( 
		\widehat{K}_\infty \circ \widehat {\mathfrak u} (P_1 T) \widehat {{\mathfrak k} } (I_+) 
		+ i \widehat{K}_\infty \circ \widehat {\mathfrak u} (P_1 T) \widehat {{\mathfrak k} } (I_+)  \bigr) 
				\nonumber \\ 
		&\quad =    (1+\rho_\beta)^{{\frac{1}{2}}}   \bigl( \widehat{K}_\infty \, \widehat {{\mathfrak k} } (I_+) 
		+i \widehat{K}_\infty \, \widehat {{\mathfrak k} } (I_+)  \bigr)
 			+ \rho_\beta^{\frac{1}{2}}  \bigl( \widehat{K}_\infty  \,  \widehat {{\mathfrak k} } (I_-) 
			+ i \widehat{K}_\infty \, \widehat {{\mathfrak k} } (I_-) \bigr)\; .  
	\end{align*}
It follows from Proposition \ref{restrictedOnePartdS} $i.)$ and $ii.)$ that this set is dense in $\widehat {\mathfrak  d} (S^1)$. 
We also used \eqref{eqJKKP1T} and the fact that $(1+\rho_\beta)$ and $\rho_\beta$ are 
strictly positive, and therefore invertible on 
$\widehat{K}_\infty \, {\mathcal D}_{\mathbb{C}} \left(S^1 \setminus \{ - \frac{\pi}{2}, 
\frac{\pi}{2}\} \right)$. 

\item [$iii.)$] 
$t \mapsto {\rm e}^{it \varepsilon } $ is a strongly continuous  group of unitaries, and \eqref{eqJUt} implies
	\begin{align}
		{\rm e}^{it \varepsilon } \circ \widehat{K}_\beta  
		&= {\rm e}^{it  \varepsilon}\circ 
		\big((1+\rho_\beta)^{{\frac{1}{2}}}+\rho_\beta^{\frac{1}{2}}\,j \big) \circ  
		\widehat{K}_\infty \nonumber \\
		&=
			\big((1+\rho_\beta)^{{\frac{1}{2}}}+\rho_\beta^{\frac{1}{2}}\,j \big) \circ 
			{\rm e}^{it \varepsilon}\circ \widehat{K}_\infty  \nonumber \\
		&= \widehat{K}_\beta  \circ \widehat  {\mathfrak u} (\Lambda_1(t) ) \; .  
		\label{eqUeps}
	\end{align}
Let $\widehat{\mathbb{\Phi}}=(\mathbb{\phi},{\mathbb \pi})$, where $\mathbb{\phi}$ and
${\mathbb \pi}$ 
have compact supports in the open half-circle $I_+$. 
Then 
	\[
		\varepsilon \, \widehat{K}_\infty \,  \widehat{\mathbb{\Phi}} 
		=|\varepsilon|\, 
		\widehat{K}_\infty \, \widehat{\mathbb{\Phi}} 
		\qquad \text{and} \qquad \varepsilon \,  j \, \widehat{K}_\infty \, \widehat{\mathbb{\Phi}} 
		= -|\varepsilon| \, j \, \widehat{K}_\infty
		\, \widehat{\mathbb{\Phi}} \; . 
	\]
This implies   
	\begin{align*} 
		{\rm e}^{-\beta\varepsilon/2}\, \widehat{K}_\beta \widehat{\mathbb{\Phi}} 
		&= 
			{\rm e}^{-\beta\varepsilon/2}
			\big((1+\rho_\beta)^{{\frac{1}{2}}} 
			+\rho_\beta^{\frac{1}{2}}\,j\big)\,  \widehat{K}_\infty \, \widehat{\mathbb{\Phi}} \nonumber \\ 
		& =  {\rm e}^{-\beta|\varepsilon|/2}(1+\rho_\beta)^{{\frac{1}{2}}} \,\widehat{K}_\infty\, \widehat{\mathbb{\Phi}}+ 
				{\rm e}^{\beta|\varepsilon|/2}\, \rho_\beta^{\frac{1}{2}}\,j \, \widehat{K}_\infty \, \widehat{\mathbb{\Phi}}
		\\
		&=
			\left(\tfrac{{\rm e}^{-\beta|\varepsilon|}}
			{1-{\rm e}^{-\beta|\varepsilon|}}\right)^{\frac{1}{2}} \,\widehat{K}_\infty \, 
			\widehat{\mathbb{\Phi}} 
			+ \left(
			\tfrac{1}
			{1-{\rm e}^{-\beta|\varepsilon|}} \right)^{\frac{1}{2}}\,
			j  \, \widehat{K}_\infty \, \widehat{\mathbb{\Phi}}  \nonumber \\
		&= \rho_\beta^{\frac{1}{2}} \, \widehat{K}_\infty \, \widehat{\mathbb{\Phi}} + 
		(1+\rho_\beta)^{\frac{1}{2}} j \, \widehat{K}_\infty \, \widehat{\mathbb{\Phi}} \; . 
	\end{align*}
Thus (by linearity) 
	\begin{equation}
	 \label{can-geo}
	 \widehat{K}_\beta \, \widehat {\mathfrak k}   (I_+ ) 
	 + i \,\widehat{K}_\beta \, \widehat {\mathfrak k}   (I_+ )
		\subset {\mathscr D} \bigl( {\rm e}^{-\beta\varepsilon / 2 } \bigr) \; ; 
	\end{equation}
Moreover, according to Lemma \ref{Lm3.5},
zero is not an eigenvalue of the generator $ \varepsilon$. 
\item [$iv.)$] $j $ is a conjugation, and  
	\[
		j \circ \widehat{K}_\beta  =
                \widehat{K}_\beta  \circ \widehat {\mathfrak u} (P_1T) \;  
	\] 
by Lemma~\ref{Jeins} and the fact that $j$ commutes with $\rho_\beta$. 
The KMS condition holds: we have already seen that 
	\begin{align*} 
		{\rm e}^{-\beta\varepsilon/2}\, \widehat{K}_\beta \widehat{\mathbb{\Phi}} 
		&= \rho_\beta^{\frac{1}{2}} \, \widehat{K}_\infty \, \widehat{\mathbb{\Phi}} + 
		(1+\rho_\beta)^{\frac{1}{2}} j \, \widehat{K}_\infty \, \widehat{\mathbb{\Phi}} =  j \, 
			\widehat{K}_\beta \,\widehat{\mathbb{\Phi}} \;  .
	\end{align*}
\end{itemize}
\end{proof}

\begin{lemma} Let $h\in {\mathcal D}_{\mathbb{C}} \left(S^1 \setminus \{ - \frac{\pi}{2}, 
\frac{\pi}{2}\} \right)$. Then 
\begin{align}
\bigl\| \bigl((1+& \rho_\beta)^{{\frac{1}{2}}}+\rho_\beta^{\frac{1}{2}}\, (P_1)_* \bigr) h \bigr\|_{\widehat {\mathfrak  d} (S^1)}^2 
 \\
& =   \bigl\langle \, | \,  {\rm cos}_\psi | \, h \, , 
		\tfrac{1}{2 |\varepsilon| }\, \bigl(\coth \tfrac{\beta\varepsilon}{2}+\tfrac{(P_1)_*}
					{\sinh \frac{ \beta |\varepsilon|}{2}} \bigr) \, | \,  {\rm cos}_\psi | \,  h  
					\bigr\rangle_{L^{2}(S^1, \frac{r {\rm d} \psi}{|\cos \psi|})} \; . 
\nonumber
\end{align}
\end{lemma}

\begin{proof} Write $h=h_+ + h_-$, where the support of $h_\pm$ is contained in
$I_\pm$, respectively. Note that $(P_1)_*$ commutes with $|\varepsilon|$ and with the
multiplication operator~$|\,  {\rm cos}_\psi|$. It follows that  
	\begin{align*} 
		 \bigl\| \bigl((1+ \rho_\beta)^{{\frac{1}{2}}}+\rho_\beta^{\frac{1}{2}} (P_1)_* \bigr) h_\pm \bigr\|_{\widehat {\mathfrak  d} (S^1)}^2 
			 &= 
					\bigl\langle  \, |\,  {\rm cos}_\psi|  \, h_\pm \,  , (1+2\rho_\beta)  
					|\,  {\rm cos}_\psi | \, h_\pm \bigr\rangle_{\widehat {\mathfrak  d} (S^1)} 
					\\
					&= \bigl\langle \, \,  {\rm cos}_\psi h_\pm  \,  ,  \frac{\coth   \beta |\varepsilon| }
			{2  |\varepsilon|   }  
			\,  {\rm cos}_\psi  h_\pm   \bigr\rangle_{L^{2}(S^1 ,
                          \frac{r {\rm d} \psi}{|\cos \psi|}) } \; . 
	\end{align*}
For the mixed terms, we find
	\begin{align}
		\bigl\langle 
		 \bigl( (1+ &\rho_\beta)^{ \frac{1}{2} }+\rho_\beta^{ \frac{1}{2} }\, (P_1)_* \bigr)   
		h_+ \, , 
		\bigl((1+ \rho_\beta)^{ \frac{1}{2} }+\rho_\beta^{ \frac{1}{2} }\, (P_1)_* \bigr)
		h_-   \bigr\rangle_{ \widehat {\mathfrak  d} (S^1) }
 		\nonumber \\
		&= 
				\bigl\langle \,  {\rm cos}_\psi  h_+ \,  ,  \tfrac{{\rm e}^{- \frac{\beta}{2} |\varepsilon|} }
					{ |\varepsilon| (1- {\rm e}^{- \beta |\varepsilon|  })}  
					\,  {\rm cos}_\psi (P_1)_*h_-   \bigr\rangle_{L^{2} (I_+ ,
                        			 \frac{ r{\rm d} \psi}{|\cos \psi|}) }  \nonumber \\
		&= 
				\bigl\langle \,  {\rm cos}_\psi h_+  \, ,  \tfrac{1}{2  |\varepsilon| \sinh \frac{\beta}{2} 
 					|\varepsilon|}\,  {\rm cos}_\psi (P_1)_*h_-   \bigr\rangle_{L^{2}(I_+ ,
                          			\frac{r {\rm d} \psi}{|\cos \psi|}) } \; .  
	\end{align}
We have used the identities $1+2\rho_\beta=\coth \frac{\beta}{2} |\varepsilon|$ and 
	\[
		2(\rho_\beta (1+\rho_\beta ))^{\frac{1}{2}}= \bigl(\sinh \tfrac{\beta}{2} |\varepsilon| \bigr)^{-1}  \; . 
	\]
The term with $h_+$ and $h_-$ interchanged yields a similar
expression. Putting together the four terms, and noting that $\varepsilon$ leaves
the subspaces $L^{2}(I_\pm,\frac{r{\rm d} \psi}{|\cos \psi|})$ invariant,  
completes the proof.
\end{proof}

\section{The canonical one-particle structure}
\label{co-ps}

It was recognised by Borchers and Buchholz~\cite{BoB} that the 
proper, ortho\-chron\-ous Lorentz group $SO_0(1,2)$ 
group can be unitarily implemented iff~$\beta$ is equal to the Hawking\footnote{In the present context, the temperature $T= 2 \pi r$
was first derived by Figari, H\o egh-Krohn  and Nappi~\cite{FHN}. The article by Hawking was submitted soon afterwards.}
temperature $2 \pi r$ \cite{Ha, Se1, Se2}. In fact, we will now show that if~$\beta = 2 \pi r$, then the unitary map 
	\begin{align*} 
		\mathbb{u} \colon \widehat {\mathfrak  d} (S^1) & \to \widehat{{\mathfrak h}}  (S^1) \\
		h & \mapsto 
		\tfrac{1}{\sqrt{r}}  \; 
		|\,  {\rm cos}_\psi|^{-1}\, \big( \rho_{2 \pi}^{\frac{1}{2}} (P_1)_* - (1+\rho_{2 \pi})^{\frac{1}{2}}  \big) h \; , 
	\end{align*}
allows us to implement the rotations $R_0 (\alpha)$, $\alpha \in [0, 2 \pi)$, in the double $(2 \pi r)$-KMS 
one-particle structure introduced in Proposition \ref{OnePart}.

\begin{proposition}
\label{Prop5.6}
The operator $\mathbb{u}$ is unitary, \emph{i.e.}, 
\[
	\| \mathbb{u} h \|_{\widehat{{\mathfrak h}}  (S^1)} = \| h \|_{\widehat {\mathfrak  d} (S^1)} \; . 
\]
Its inverse $\mathbb{u}^{-1} \colon \widehat{{\mathfrak h}}  (S^1)\to  \widehat {\mathfrak  d} (S^1)$ is given by 
\label{chfdpage}
	\begin{equation}
  		\label{eq2piH1}
			\mathbb{u}^{-1} = - \sqrt{r}  \big( (1+\rho_{2 \pi})^{\frac{1}{2}}+\rho_{2 \pi}^{\frac{1}{2}}
			(P_1)_*\big) \,|\,  {\rm cos}_\psi|   \, . 
	\end{equation}
\end{proposition}

\begin{proof} 
Let $h\in \widehat{{\mathfrak h}}  (S^1)$. Using again  that $(P_1)_*$ commutes with $|\varepsilon|$ and with the
multiplication operator~$|\,  {\rm cos}_\psi|$, we find 
	\begin{align} 
		\|\mathbb{u}^{-1}h\|_{\widehat {\mathfrak  d} (S^1)}^2 
			&=   r \,  
			\bigl\| \big((1+ \rho_{2 \pi})^{{\frac{1}{2}}}
			+\rho_{2 \pi}^{\frac{1}{2}}\,(P_1)_* \big) \,|\,  {\rm cos}_\psi| h \bigr\|_{\widehat {\mathfrak  d} (S^1)}^2    
					\nonumber
					\\
				& =   r \, 
					\bigl\langle |\,  {\rm cos}_\psi| h \,,\, (1+2\rho_{2 \pi})  |\,  {\rm cos}_\psi| h \bigr\rangle_{\widehat {\mathfrak  d} (S^1)} 
					\nonumber \\
					&
					\qquad  + 2  r \, 
					  \bigl\langle |\,  {\rm cos}_\psi| h \,,\, (\rho_{2 \pi}(1+\rho_{2 \pi}))^{\frac{1}{2}} (P_1)_* |\,  {\rm cos}_\psi|  h 
					\bigr\rangle_{\widehat {\mathfrak  d} (S^1)} 
										\nonumber \\
			&=  r \,   \bigl\langle  h,\tfrac{1}{2  |\varepsilon| } 
					\bigl(\coth   \pi |\varepsilon| +  
					\tfrac{(P_1)_*}{\sinh \pi |\varepsilon | }\bigr) |\,  {\rm cos}_\psi| h  
					\bigr\rangle_{L^{2}(S^1, r {\rm d} \psi)} 
					\nonumber
					\\
			&=  \|  h \|_{\widehat{{\mathfrak h}}  (S^1)} 	\; . \label{eqMagicFormC} 
	\end{align}
The last equality follows from Proposition \ref{thhm}.
We have again used the identities $1+2 \rho_{2 \pi}=\coth \pi |\varepsilon|$ and 
	\[
		2( \rho_{2 \pi} (1+\rho_{2 \pi}))^{\frac{1}{2}}= (\sinh\pi |\varepsilon|)^{-1}  \; . 
	\]
\end{proof}

\begin{proposition}
Consider the map 
\label{Kmhatpage}
	\begin{align*}
	\label{K1PStrucHe}
		 \widehat{K} \colon  \qquad \widehat{\mathfrak k}(S^1) 
		& \rightarrow  \widehat{{\mathfrak h}}  (S^1) \nonumber \\
		 (\mathbb{\phi},\mathbb{\pi})  		  & \mapsto   \tfrac{1}{\sqrt{r}}  ( - \mathbb{\pi} +  i \,\omega
		 r \, \mathbb{\phi}) \; .
	\end{align*}
It follows that the quadruple 
	\[
		\bigl(\widehat{K} , \widehat{{\mathfrak h}}  (S^1), {\rm e}^{i t \, \omega \,  {\rm cos}_\psi} , 
		C (P_1)_* \bigr) 
	\]
forms a double $2\pi r$-KMS one-particle structure for the classical 
double dynamical system 
$\bigl( \widehat {\mathfrak k}  \left(S^1 \setminus \{ - \frac{\pi}{2} , 
\frac{\pi}{2} \} \right) , \widehat \sigma, \widehat {\mathfrak u} (\Lambda_1( \frac{t}{r}) ),  
\widehat {\mathfrak u} (P_1T )  \bigr)$ in the sense
of~\ref{dbops}, unitarily equivalent to 
$\bigl(\, \widehat{K}_{2 \pi} , \widehat {\mathfrak  d} (S^1), {\rm e}^{i \frac{t}{r} \varepsilon},  j \bigr) $, 
in agreement with Theorem~\ref{ThB1}.
\end{proposition}

\begin{proof} We first show that $\widehat{K} =  \mathbb{u}  \circ \widehat{K}_{2 \pi} $.   
Using $j = C (P_1)_*$, where $(C h)(\psi)\doteq \overline{h(\psi)}$, one gets  
	\begin{align*}
		& \quad \; \; \mathbb{u}  \circ 
		\widehat{K}_{2 \pi} (\mathbb{\phi},\pi) 
		\\
				& \qquad =  \mathbb{u}  \circ 
			\big((1+\rho_{2 \pi})^{{\frac{1}{2}}}+\rho_{2 \pi}^{\frac{1}{2}}\,j \big) \;
 			\widehat{K}_\infty (\mathbb{\phi},\pi) 
			\nonumber \\
		&\qquad = - \tfrac{1}{\sqrt{r}}  \, |\,  {\rm cos}_\psi|^{-1}\,
 							\big( (1+\rho_{2 \pi})^{\frac{1}{2}} - \rho_{2 \pi}^{\frac{1}{2}} (P_1)_*\big)
			\big((1+\rho_{2 \pi})^{{\frac{1}{2}}}+\rho_{2 \pi}^{\frac{1}{2}}\,j \big) \;
 			\widehat{K}_\infty (\mathbb{\phi},\pi) 
			\nonumber \\
		&\qquad =	 - \tfrac{1}{\sqrt{r}}  \, |\,  {\rm cos}_\psi|^{-1}\,\Big( 1+\big(\rho_{2 \pi}-(\rho_{2 \pi }(1+\rho_{2 \pi}))^{\frac{1}{2}}
				(P_1)_*\big)(1-C ) \Big) \,\widehat{K}_\infty (\mathbb{\phi},\pi) \; .
			\nonumber  				
	\end{align*}
Taking $1+2\rho_{2 \pi} =\coth \pi  |\varepsilon|$ and 
$2(\rho_{2 \pi} (1+\rho_{2 \pi} ))^{\frac{1}{2}}= (\sinh\pi  |\varepsilon|)^{-1}$ into account, we find 
	\[
		\mathbb{u}  \circ \widehat{K}_{2 \pi} (\mathbb{\phi},\pi) 
		=  \begin{cases}  - \tfrac{1}{\sqrt{r}}  \, |\,  {\rm cos}_\psi|^{-1}\, \widehat{K}_\infty (\mathbb{\phi},\mathbb{\pi})&
		\text { if } \widehat{K}_\infty
						(\mathbb{\phi},\mathbb{\pi})  \in \widehat {\mathfrak  d} (S^1, {\mathbb{R}})\; ,\\
				-   \sqrt{r}   \, 
				\omega \, \varepsilon^{-1}\, \widehat{K}_\infty (\mathbb{\phi},\mathbb{\pi}) & \text{ if }  \widehat{K}_\infty
				(\mathbb{\phi},\mathbb{\pi}) \in i \widehat {\mathfrak  d} (S^1, {\mathbb{R}})\, \nonumber .
				\end{cases}  
	\]
In the last equation we have used~\eqref{eqMagicForm} and $(P_1)_* \varepsilon = - \varepsilon (P_1)_*$. 
By $\widehat {\mathfrak  d} (S^1, {\mathbb{R}})$ 
we have denoted the real subspace of real valued functions 
in $\widehat {\mathfrak  d} (S^1)$.
Use $\widehat{K}_\infty (\mathbb{\phi}, \mathbb{\pi}) =  \,  {\rm cos}_\psi \,\mathbb{\pi}  -i  \varepsilon \mathbb{\phi} $ to prove that
	\begin{equation}
	\label{can-geo2}
		\widehat{K} =  \mathbb{u}  \circ \widehat{K}_{2 \pi} \; . 
	\end{equation}  
It remains to show that the unitary map $\mathbb{u}$ satisfies 
	\[ 
	  \mathbb{u} \circ\varepsilon\circ \mathbb{u}^{-1} = \omega  \,  {\rm cos}_\psi \quad 
	 	\text{and} \quad
	  \mathbb{u} \circ j\circ \mathbb{u}^{-1} =  C (P_1)_*  \quad \text{on $\widehat{{\mathfrak h}}  (S^1)$} \, . 
	 \]
Using again $(P_1)_* \varepsilon = - \varepsilon (P_1)_*$,  we can verify the first of these two 
identities:
	\begin{align*}
			\mathbb{u} \circ\varepsilon\circ \mathbb{u}^{-1} &=    |\,  {\rm cos}_\psi|^{-1}\,
 							\big( (1+\rho_{2 \pi})^{\frac{1}{2}} - \rho_{2 \pi}^{\frac{1}{2}} (P_1)_*\big) 
						\varepsilon 
							\big( (1+\rho_{2 \pi})^{\frac{1}{2}}+\rho_{2 \pi}^{\frac{1}{2}}  (P_1)_*\big)\,
						|\,  {\rm cos}_\psi| 
			\nonumber \\
			& =  |\,  {\rm cos}_\psi|^{-1} \varepsilon 
			\big( (1+2\rho_{2 \pi})  
			+ 2\rho_{2 \pi}^{\frac{1}{2}}(1+\rho_{2 \pi})^{\frac{1}{2}} (P_1)_*\big) \,|\,  {\rm cos}_\psi| 
			\nonumber \\
			&= |\,  {\rm cos}_\psi|^{-1} \varepsilon  \,\bigl(\coth \pi  |\varepsilon|
		+ \tfrac{(P_1)_*}{\sinh\pi |\varepsilon|} \bigr)^{-1} |\,  {\rm cos}_\psi|
			\nonumber \\
			&= \omega r  \,  {\rm cos}_\psi \; . 
	\end{align*}
In the second but last equality we have used the identity \eqref{eqMagicForm}.

The second identity follows from the fact that $j$ commutes with the multiplication operator
$|\,  {\rm cos}_\psi|$: 
	\[ 
	  \mathbb{u} \circ j\circ \mathbb{u}^{-1} =  C (P_1)_*  \quad \text{on $\widehat{{\mathfrak h}}  (S^1)$} \, . 
	 \]
We have thus established unitarily equivalence of the two double $2\pi r$-KMS 
one-particle structure under consideration, in agreement with Theorem~\ref{ThB1}.

It is now straight forward to verify that $\bigl(\widehat{K} , \widehat{{\mathfrak h}}  (S^1), 
{\rm e}^{i t \omega \,  {\rm cos}_\psi} , 
C (P_1)_* \bigr)$ forms a double $2\pi r$-KMS one-particle structure for the classical double dynamical system 
$\bigl( \widehat {\mathfrak k}  \left(S^1 \setminus \{ - \frac{\pi}{2} , 
\frac{\pi}{2} \} \right) , \widehat \sigma, \widehat {\mathfrak u} (\Lambda_1(\frac{.}{r}) ),  
\widehat {\mathfrak u} (P_1T )  \bigr)$ in the sense
of~\ref{dbops}:   
	\begin{align}
		 \widehat{K}\circ \widehat {\mathfrak u} (\Lambda_1 (t))  
		& =  \mathbb{u} \circ \widehat{K}_{2 \pi } 
		\circ \widehat {\mathfrak u} (\Lambda_1 (t)) \nonumber \\
		& = \mathbb{u} \circ {\rm e}^{it \varepsilon} \;  
		\widehat{K}_{2 \pi } \nonumber \\
		& =  {\rm e}^{it\, \omega r \,  {\rm cos}_\psi} \circ  \mathbb{u} \circ  \widehat{\mathbb k}_{2 \pi r} \nonumber \\
		& =  \widehat{u}(\Lambda_1(t)) \circ \widehat{K} \; , 
		\label{canonical-boosts}
	\end{align}
see Eq.~\eqref{eqUeps}; and 
	\begin{equation}
		\label{reflections}
		\widehat{K} \circ \widehat {\mathfrak u} (P_1T) 
		= \mathbb{u} \circ \widehat{K}_{2 \pi r} \circ 
		\widehat {\mathfrak u} (P_1T)
		= \mathbb{u} \circ j \circ \widehat{K}_{2 \pi r}
		=C (P_1)_* \circ \widehat{K}  \; . 
	 \end{equation}
This also shows that $\widehat{u} (P_1T) = C (P_1)_* $ is anti-unitary. 
\end{proof}

\begin{theorem} 
\label{UIR-S1}
The rotations 
	\[
	\bigr(\widehat{u} (R_0(\alpha)) h \bigl) (\psi) = h (\psi - \alpha) \; , \qquad \alpha \in [0, 2\pi) \; , \quad h \in \widehat{{\mathfrak h}}  (S^1) \; , 
	\]
and the boosts 
\label{Umhatpage}
	\[
		\widehat{u} (\Lambda_1(t)) = {\rm e}^{i t \omega  r \, {\rm cos}_\psi } \; , \qquad t \in \mathbb{R} \; , 
	\]
generate a representation of $SO_0(1,2)$ on $\widehat{{\mathfrak h}}  (S^1)$. 
\end{theorem}

\begin{proof}
The generator of the rotations  $K_0 =- i \partial_\psi$ has purely discrete spectrum.  
Its eigenfunctions are 
	\[
		e_k= \frac{{\rm e}^{ik\psi}}{\sqrt{2\pi  r } } \; ,  \qquad k \in \mathbb{Z}  \; .
	\]
The generators of the boosts,
	\[
		L_1 = \omega r \,  {\rm cos}_\psi \; , \qquad L_2 = \omega r \, \sin_\psi \; , 
	\]
satisfy the commutator relations 
	\[
		[K_0 \, , \, L_1] =   i L_2 \; , 
		\qquad
		[L_2\, , \, K_0] =  i L_1 \;  . 
	\]
The latter follow from
	\[
		[ - i \partial_\psi \, , \,  \omega r \,  {\rm cos}_\psi ] = i \omega r \,\sin_\psi \; , 
		\qquad 
		[ \omega r \, \sin_\psi \, , \, - i \partial_\psi ] = i \omega r \,  {\rm cos}_\psi \; . 
	\]
To  verify the commutation relation 
$[L_1 \, , \, L_2] = - i K_0 $, we consider the  ladder operators
	\[
		L_\pm = L_1  \pm i L_2  = \omega r \, {\rm e}^{ \pm i \psi }\; . 
	\]
We will show that 
	\[
		\langle e_{k'},  \big[L_+ ,  L_- \big] e_k \rangle_{\widehat{{\mathfrak h}}  (S^1)}
		 = \ -2 \langle e_{k'},  K_0 e_k \rangle_{\widehat{{\mathfrak h}}  (S^1)} \; .
	\]
The latter is equivalent to  
	\begin{equation}
		\label{eq:venoirgtuuwp}
		\widetilde \omega(k)\big(\widetilde \omega(k-1) - \widetilde \omega(k+1) \big) =  - \frac{2k} { {r}^2} \;, \qquad \forall k \in \mathbb{Z} \;.
	\end{equation}
In order to verify this identity, let us first consider only the $\Gamma$-factors occurring in \eqref{eq:omega-1}. 
Define, for $k\in\mathbb{Z}$,
	\[ 
		w (k) \doteq  \frac{ \Gamma\left(\frac{k+s^+}{2}\right) }{
      							\Gamma\left( \frac{k-s^+}{2} \right)}
					\frac{ \Gamma\left(\frac{k-s^+ +1}{2}\right)
						}{\Gamma\left( \frac{k+s^+ +1}{2} \right) }\;.
	\]
One has
	\begin{align*}
		w (k+1) & =  
					\frac{ \Gamma\left(\frac{k+s^+ +1}{2}\right) }{
      							\Gamma\left( \frac{k-s^+ +1}{2} \right) }
					\frac{ \Gamma\left(\frac{k-s^+ +2}{2}\right)
						}{\Gamma\left( \frac{k+s^+ +2}{2} \right)}  
				 =  
				 	\frac{ \Gamma\left(\frac{k+s^+ +1}{2}\right) }{
     						 \Gamma\left( \frac{k-s^+ +1}{2} \right) }
					\frac{ \Gamma\left(\frac{k-s^+}{2}+1\right)
						}{\Gamma\left( \frac{k+s^+}{2} +1\right)}   \\
				& 
				=   
					\frac{k-s^+}{k+s^+} \;  \frac{ \Gamma\left( \frac{k+s^+ +1}{2} \right) }{
      										\Gamma\left( \frac{k-s^+ +1}{2} \right) }
									\frac{ \Gamma\left(\frac{k-s^+}{2}\right)
										}{\Gamma\left( \frac{k+s^+}{2} \right) }  
				= \; \frac{k-s^+}{k+s^+} \; \frac{1}{w (k)} \:,
	\end{align*}
as one easily verifies. Hence, 
	\begin{equation}
		w(k)w(k+1)=\frac{k-s^+}{k+s^+} \;  .
		\label{eq:JUytiuytuiu}
	\end{equation}
Since $\widetilde \omega(k)  =  \frac{{ (k +s^+) } }{r} \, w (k)$, we have
	\[ \widetilde \omega(k) \widetilde \omega(k+1) 
	= {r}^{-2} (k +s^+) (k +s^+ +1)  \frac{k-s^+}{k+s^+}
	= {r}^{-2} (k-s^+)(k +s^+ +1)  \;.
	\]
Hence, we get the two following useful relations:
	\begin{align}
		\label{eq:useful-1}
		\widetilde \omega(k) \widetilde \omega(k+1) & = {r}^{-2} (k-s^+)(k +s^+ +1) = r^{-2} k (k+1) + \mu^2 \\
		\label{eq:useful-2}
		\widetilde \omega(k) \widetilde \omega(k-1) & = {r}^{-2} (k +s^+)(k-s^+-1) 
		= r^{-2} k (k-1) + \mu^2\;.
	\end{align}
The last one is obtained from the previous one by taking $k\to k-1$.
We note that 
	\begin{equation}
		           \label{eq:tilde-omega-and-the-traditional-dispersion-relation}
	\frac{1}{2} \Bigl( 
				\widetilde \omega(k) \widetilde \omega(k+1) +	\widetilde \omega(k) \widetilde \omega(k-1)  \Bigr) = \frac{k^2}{{r}^2}  + \mu^2 \; , 
	\end{equation}
which allows us to establish the usual dispersion relation in the limit $r \to \infty$. 

Relation~(\ref{eq:venoirgtuuwp}) can now be verified using \eqref{eq:useful-1}--\eqref{eq:useful-2}:
	\begin{align*}
		\widetilde \omega(k)\widetilde \omega(k-1) -\widetilde \omega(k)\widetilde \omega(k+1) & = {r}^{-2}
		(k +s^+)(k-s^+-1) \\
		& \qquad - {r}^{-2} (k-s^+)(k +s^+ +1) \\
		& =  - \frac{2k}{ {r}^{2}} \;,
	\end{align*}
as desired. We conclude that 
	\[
		[L_+ , L_-] = - 2 K_0 \; , 
		\qquad [ {\mathfrak k}_0 , L_\pm ] = \pm L_\pm \; ,
	\]
in agreement with $[L_1 \, , \, L_2] = - i K_0 $.
\end{proof}
 
\goodbreak
The operator 
	\[
		\omega r \,  {\rm cos}_\psi 
		= (\omega r \,  {\rm cos}_\psi)_{\upharpoonright I_+} 
		+ (\omega r \,  {\rm cos}_\psi)_{\upharpoonright I_-}
	\]
is the sum of a positive operator $(\omega r \,  {\rm cos}_\psi)_{\upharpoonright I_+}$ 
acting on $\widehat{{\mathfrak h}}  (I_+)$, 
and a negative operator $ (\omega r\,  {\rm cos}_\psi)_{\upharpoonright I_-}$ 
acting on $\widehat{{\mathfrak h}}  (I_-)$. 
Both operators have absolutely continuous spectrum. 
Similar results hold for $I_\alpha$, $\alpha \in [0, 2 \pi)$. 


\begin{theorem} 
\label{1PStrucHe2} 
The triple $  \bigl(\widehat{K} , \widehat{{\mathfrak h}}  (S^1), \widehat{u} \bigr) $
is a {\em one-particle de Sitter structure} for the canonical classical dynamical system 
$\bigl( \, \widehat{\mathfrak k}  (S^1 ) , \widehat\sigma, \widehat{\mathfrak u}   \bigr)$ introduced 
in Proposition~\ref{nocheinlabel}.
In other words, 
\begin{itemize}
\item [$ i.)$] $\widehat{K}$ defines a symplectic map from $(\widehat{\mathfrak k}  (S^1 ) , \widehat\sigma)$ 
to $\bigl( \, \widehat{{\mathfrak h}}  (S^1), \Im \langle \, . \, , 
\, . \, \rangle_{\widehat{{\mathfrak h}}  (S^1)} \bigr)$ and $\widehat{K} \, \widehat{\mathfrak k}  (S^1 )$ 
is dense in $\widehat{{\mathfrak h}}  (S^1)$;
\item [$ ii.)$] 
there exists a unique (anti-) unitary representation of $O(1,2)$ satisfying 
	\begin{equation} 
		\label{eqUHe2}
		 \widehat{u} (\Lambda)\circ \widehat{K} = \widehat{K} \circ \widehat {\mathfrak u} (\Lambda) \; . 
	\end{equation}
Moreover, $\widehat{u} (R_{0}(\alpha)) = R_{0}(\alpha)_*$ for $\alpha \in [0, 2 \pi)$;
\item [$ iii.)$]  for any half-circle\footnote{Given the fact that we consider $\widehat{{\mathfrak h}}  (S^1)$,  it is
more natural to specify a half-circle 
$I_\alpha= R_{0}(\alpha) I_+$. Recall that 
$W^{(\alpha)} = I_\alpha''$.}
 $I_\alpha$, the pre-Bisognano-Wichmann property \cite[p.~75]{Kay3} holds: 
	\begin{equation} 
		\label{5.3b}
		 \widehat{K}  \, \widehat{{\mathfrak k}} (I_\alpha)  \subset {\mathscr D} \bigl( \widehat{u} 
		 (\Lambda_{W^{(\alpha)}} ( i \pi)) \bigr) \; ,
	\end{equation}
and
	\begin{equation}
		\label{5.4b}
		\widehat{u} (\Lambda_{W^{(\alpha)}} ( i \pi)) h	= \widehat{u} (\Theta_{W^{(\alpha)}}) h \; , 
		\qquad h \in \widehat{K} \, \widehat{\mathfrak k}(I_\alpha) \; . 
	\end{equation}
\end{itemize}
\end{theorem}

\goodbreak
\begin{proof} \quad 
\begin{itemize}
\item [$ i.)$] 	
Clearly, $C^\infty (S^1) + i \omega  C^\infty (S^1)$ is dense in $\widehat{{\mathfrak h}}  (S^1)$. 
To verify that $\widehat{K}$ is a symplectic map, compute
\begin{align*}
2 \Im \langle \widehat{K}(\mathbb{\phi}_1,\mathbb{\pi}_1) , \widehat{K}(\mathbb{\phi}_2,\mathbb{\pi}_2) \rangle_{\widehat{{\mathfrak h}}  (S^1)}
& = 2 \Im \langle -\mathbb{\pi}_1 +  i \,\omega r \, \mathbb{\phi}_1  , -\mathbb{\pi}_2 
+  i \,\omega r \, \mathbb{\phi}_2  \rangle_{\widehat{{\mathfrak h}}  (S^1)}
\\
& =   \langle \mathbb{\phi}_1,\mathbb{\pi}_2 \rangle_{L^2(S^1, r {\rm d} \psi ) } - \langle \mathbb{\pi}_1   , 
\mathbb{\phi}_2  \rangle_{L^2(S^1, r {\rm d} \psi)} 
\\
&=		\widehat \sigma\big((\mathbb{\phi}_1,\mathbb{\pi}_1),(\mathbb{\phi}_2,\mathbb{\pi}_2)\big) \; . 
\end{align*}
\item[$ ii.)$]
For $\Lambda=
R_{0}$ a rotation, we have 
	\begin{align*}
		(\widehat{u}(R_{0}) \circ \widehat{K})(\mathbb{\phi},\mathbb{\pi})&
		= (R_{0}) _* \left(-\mathbb{\pi}+ i\omega r \, \mathbb{\phi} \right) \nonumber \\
		&=- (R_{0})_*\mathbb{\pi} +  i\omega r \,  (R_{0})_*\mathbb{\phi}
		\nonumber \\
		&= \widehat{K} \bigl( (R_{0})_*\mathbb{\phi}, (R_{0})_*\mathbb{\pi} \bigr) 
		= \bigl( \widehat{K} \circ \widehat{{\mathfrak u}} (R_{0})\bigr) (\mathbb{\phi},\mathbb{\pi})  \; , 
	\end{align*}
since $\omega$ commutes with the pullback $(R_{0})_*$ of a rotation. 
For the boosts, see \eqref{canonical-boosts}; and for the reflections, 
see \eqref{reflections}.

\item [$ iii.)$] for $(\mathbb{\phi},\mathbb{\pi}) \in  \widehat{\mathfrak k}(I_\alpha)$, 
	\begin{align*}
		\widehat u (\Lambda_W ( i \pi)) \widehat{K} (\mathbb{\phi},\mathbb{\pi}) &=
		\widehat K \circ \widehat {\mathfrak u} (\Lambda_W ( i \pi)) \,  (\mathbb{\phi},\mathbb{\pi})
		\nonumber \\
		&= \widehat K \circ  \widehat {\mathfrak u} (\Theta_{W}) \, (\mathbb{\phi},\mathbb{\pi})
		\nonumber \\
		&= \widehat u (\Theta_{W}) \, \widehat{K} (\mathbb{\phi},\mathbb{\pi})  \; , 
	\end{align*}
which demonstrates both \eqref{5.3b} and \eqref{5.4b}. 
The first equality follows from combining 
\eqref{can-geo} and \eqref{can-geo2}.
\end{itemize}
\end{proof}

\goodbreak

\begin{proposition} 
\label{Prop5.7}
There exists a unitary map ${\mathbb U}$ from $\widehat{{\mathfrak h}}  (S^1)$ to ${\mathfrak h} (dS)$, 
which  intertwines  the representations
$\widehat{u} (\Lambda)$ and 
$ u (\Lambda)$, $ \Lambda \in O(1,2)$, and the one-particle structures.
In other words, the following diagram commutes:

\begin{picture}(250,140)


\put(60,80){$\nearrow$}
\put(0,70){${\mathcal D}_{\mathbb{R}} (dS)$}
\put(220,70){${\mathbb U} $}
\put(130,70){${\mathbb T} $}
\put(50,95){$\widehat  {\mathbb P}$}
\put(50,45){$ {{\mathbb P}}$}
\put(60,55){$\searrow$}

\put(90,110){$ (\widehat {\mathfrak k} (S^1 ), \widehat {\mathfrak u})$}
\put(160,110){$\longrightarrow$}
\put(190,110){$(\widehat{{\mathfrak h}}  (S^1), \widehat{u})$}

\put(90,30){$({\mathfrak k} (dS), {\mathfrak u})$}
\put(160,30){$\longrightarrow$}
\put(190,30){$({\mathfrak h} (dS), u)\; .$}

\put(165,40){$K$}
\put(165,120){$ \widehat{K}$}

\put(120,85){\vector(0,-3){20}}
\put(210,85){\vector(0,-3){20}}

\end{picture}

\vskip -.8cm
\noindent
Moreover, the restricted map 
	\begin{equation}
		\label{localpart}
		{\mathbb U}  \colon
		 \widehat{{\mathfrak h}}  (I)
		  \mapsto 
		  {\mathfrak h} ({\mathcal O}_I ) \;,  \qquad I \subset S^1 \; , 
	\end{equation}
is unitary too. 
\end{proposition}

\goodbreak
\begin{proof} The existence of ${\mathbb U}$ 
follows from the uniqueness of the de Sitter one-particle structure.
The latter is a direct consequence of the uniqueness of the $(2 \pi r)$-KMS structure for the double wedge, 
see~\ref{ThB1}. The local part, Eq.~\eqref{localpart}, follows from Lemma \ref{tztf}:
for $ h\in {\mathcal D} (S^1)$ and
	\begin{align*} 
		 f (x ) \equiv (\delta \otimes h) (x ) &=	
		 \delta (x_0)  \;  h  ( \psi ) \;  , 
	 \nonumber \\
	g (x ) \equiv (\delta' \otimes h) (x ) &= 
	\left( \frac{\partial}{\partial x_0 } \delta \right) (x_0)  \;   h  (\psi ) \;  , 
	\end{align*} 
with $x  \equiv x  (t, \psi)$ the coordinates introduced in \eqref{w1psitau}, 
the  Cauchy data for the corresponding 
solutions $\mathbb{f}, \mathbb{g}$ of the 
Klein--Gordon equation are:
	\begin{align*}
		\big( \mathbb{f}_{ \upharpoonright S^1}, 
							 (n \, \mathbb{f} )_{\upharpoonright S^1})
   &= (0, - h) \equiv  (\mathbb{\phi}, \mathbb{\pi}) \; , 
\nonumber \\
	\big( \mathbb{g}_{ \upharpoonright S^1}, 
							 (n \, \mathbb{g} )_{\upharpoonright S^1})
 		  &= (  h, 0) \equiv  (\mathbb{\phi}, \mathbb{\pi}) \; . 
	\end{align*}
Together with $\widehat{K}\,(\mathbb{\phi},\mathbb{\pi}) = - \mathbb{\pi} +  i \,\omega r \, \mathbb{\phi}$ this gives 
	\begin{align*}
		\widehat{K}\, \big( \mathbb{f}_{ \upharpoonright S^1}, 
							 (n \, \mathbb{f} )_{\upharpoonright S^1})
   		&= h  \; , 
\nonumber \\
	\widehat{K}\, \big( \mathbb{g}_{\upharpoonright S^1}, 
							 (n \, \mathbb{g} )_{\upharpoonright S^1})
 		  &= i \omega r \, h  \; , 
	\end{align*}
both elements\footnote{As mentioned before,  $C^\infty (S^1) \subset {\mathscr D}(\omega)$.} 
of $\widehat{{\mathfrak h}}  (S^1)$. Finally, the unitary map ${\mathbb U}
\colon \widehat{{\mathfrak h}}  (S^1) \to {\mathfrak h}(dS)$ is the linear extension of the map
	\[
		h_1 + i \omega r \, h_2 \mapsto [ \delta \otimes h_1 ] + 
		[ \delta' \otimes h_2 ] \; . 
	\]
The latter shows that ${\mathbb U}
\colon \widehat{{\mathfrak h}}  (I) \to {\mathfrak h}({\mathcal O}_I)$, with ${\mathcal O}_I= I''$ 
the causal completion of $I \subset S^1$. 
\end{proof}

\goodbreak
\begin{corollary} 
\label{halpha}
Let $I$ be an open subset in $S^1$. 
The unitary group $t \mapsto {\rm e}^{i t \omega r \,  {\rm cos}_{\psi}}$ 
maps $\widehat{{\mathfrak h}}  (I)$ to
\[
\widehat{{\mathfrak h}}  \Bigl( \bigl( \Gamma^+(\Lambda_1(t)I) \cup \Gamma^-(\Lambda_1(t)I)  \bigr) \cap S^1 \Bigr) \; . 
\]
In particular, 
the unitary group $t \mapsto {\rm e}^{it (\omega r \,  {\rm cos}_\psi)_{\upharpoonright I_\pm}}$ 
leaves $\widehat{{\mathfrak h}}  (I_\pm)$ invariant.
\end{corollary}

\begin{proof} This is a direct consequence of the fact that ${\rm e}^{it \omega \,  {\rm cos}_{\psi}}$ implements the 
Lorentz boost $\Lambda_1(t)$. The latter act geometrically on ${\mathfrak k} (dS)$, \emph{i.e.}, 
a distribution supported at $I \subset S^1 \subset dS$ is mapped to a distribution supported at $\Lambda_1(t) I \subset dS$. 
This result extends by continuity to $\widehat{{\mathfrak h}}  (I)$. 
\end{proof}

\begin{corollary}
The unitary representation $\widehat{u} (\Lambda)$, $\Lambda \in SO_0(1,2)$, defined by 
Eq.~\eqref{eqUHe2},
is irreducible.
\end{corollary}

\begin{proof} By construction, 
$\widehat{u}$ is unitarily equivalent to the unitary irreducible representation 
$\widetilde{u}$ on~$\widetilde{{\mathfrak h}}(\partial V^+)$. In fact, the Casimir operator takes the form 
	\[		C^2  = - K_0^2 + \frac{1}{2} (L_+ L_- + L_- L_+) \; . 
	\]
Its off-diagonal matrix elements vanish and the diagonal matrix elements equal~$\zeta^2$: 
	\begin{align*}
		\frac{ \langle e_{k}, \; C^2  e_k \rangle_{\widehat{{\mathfrak h}}  (S^1)}}{
					\big\| e_{k}\big\|^2_{\widehat{{\mathfrak h}}  (S^1)} } 
		& =  -k^2 +\frac{{r}^2}{2}\big( \widetilde \omega(k)\widetilde \omega(k-1)+\widetilde \omega(k)\widetilde \omega(k+1) \big) \\ \; 
		& =
			-k^2 +\frac{1}{2} \big( (k +s)(k-s-1) + (k-s)(k +s+1) \big) \\
		&=  -s(s+1) = \tfrac{1}{4} + \nu^2 = \zeta^2 \; , 
	\end{align*}
as expected. In the second but last equality we have used \eqref{eq:useful-1}--\eqref{eq:useful-2}.
\end{proof}

\chapter{Second Quantization}
\label{2Q}
\setcounter{equation}{0}


Let $({\mathfrak k} , \sigma)$ be a symplectic space. The unique $C^*$-algebra ${\mathfrak W} ({\mathfrak k}, \sigma)$ generated
by nonzero elements ${\rm W}({\mathfrak f})$, ${\mathfrak f} \in {\mathfrak k}$, satisfying 
\label{weylalgebrapage}
	\begin{align}
		{\rm W} ({\mathfrak f}_1){\rm W} ({\mathfrak f}_2) 
		&= {\rm e}^{- i \sigma ( {\mathfrak f}_1,{\mathfrak f}_2 ) /2} {\rm W} ({\mathfrak f}_1+{\mathfrak f}_2) \; ,
		\nonumber \\ 
		\label{weylalgebra} 
		{\rm W} ^*({\mathfrak f}) &= {\rm W} (-{\mathfrak f}) \; ,  \qquad {\rm W} (0)=\mathbb{1} \; ,  
	\end{align}
is called the {\em Weyl algebra} associated to $({\mathfrak k} , \sigma)$;
see, \emph{e.g.}, \cite{BR}. 
In case ${\mathfrak k}$ is a Hilbert space, 
we suppress the dependence on the symplectic form given by twice the imaginary part of the scalar product.

\goodbreak
\begin{definition} 
Set 
\[
{\mathfrak W}(X) \equiv {\mathfrak W}({\mathfrak k}(X), \sigma) \; , \qquad X= {\mathcal O}, W, dS \; . 
\]
Let $\Lambda \mapsto {\mathfrak u}  (\Lambda)$ be the  representation of $SO_0(1,2)$ on~${\mathfrak k} (dS)$; 
see Proposition~\ref{Prop4-10}. Define a group of automorphisms 
$\alpha^\circ \colon \Lambda \mapsto \alpha_\Lambda^\circ$ acting on ${\mathfrak W}(dS)$ by
\label{alphapage}
	\[
		\alpha_\Lambda^\circ (W ( [f] ))=W ( {\mathfrak u}  (\Lambda) [f] ) \; , 
		\qquad  [f] \in {\mathfrak k} (dS) \; . 
	\]
The pair $\bigl({\mathfrak W}(dS), \alpha^\circ \bigr)$ is called the {\em covariant quantum 
dynamical system} associated to the Klein--Gordon equation on the de Sitter space.
\label{cqds-page}
\end{definition}

There is no global time evolution (in terms of a 1-parameter group of automorphisms) 
on ${\mathfrak W}(dS)$. In fact, there is not even a globally time-like Killing vector field on the 
de Sitter space. Nevertheless, the automorphisms~$\alpha^\circ $ respect the local structure:
	\[
		\alpha_\Lambda^\circ \bigl({\mathfrak W}({\mathcal O})\bigr)
		={\mathfrak W}(\Lambda {\mathcal O}) \;  , \qquad {\mathcal O} \subset dS \; . 
	\]
The map $ \alpha^\circ \colon \Lambda \mapsto \alpha_\Lambda^\circ$ fails to be strongly continuous in the $C^*$-norm; 
thus strictly speaking $\bigl({\mathfrak W}(dS), \alpha^\circ \bigr)$ is not a $C^*$-dynamical system. 

\begin{definition} Set 
	\[
		\widehat {\mathfrak W} (I) \doteq {\mathfrak W} \bigl( \widehat {\mathfrak k} (I), \widehat {\sigma} \bigr) \; , 
		\qquad I \subseteq S^1 \; . 
	\]
Let $\Lambda \mapsto \widehat {\mathfrak u}  (\Lambda)$ be the  representation of $O(1,2)$ 
on~$\widehat {\mathfrak k} (S^1 )$; see Proposition~\ref{nocheinlabel}. Define a group of automorphisms 
$\widehat \alpha^\circ \colon \Lambda \mapsto \widehat \alpha_\Lambda^\circ$ acting on 
$\widehat {\mathfrak W} (S^1 )$
by
	\[
		\widehat\alpha_\Lambda^\circ (\widehat {\rm W} (\widehat f ))
		=\widehat {\rm W} \big(\widehat {\mathfrak u}(\Lambda) \widehat f \, \big) \; , 
		\qquad \widehat f \in \widehat {\mathfrak k} (S^1 ) \; , \qquad \Lambda \in O(1,2) \; .  
	\]
The pair $\bigl(\widehat {\mathfrak W} (S^1 ), \widehat \alpha^\circ \bigr)$ is the 
{\em canonical quantum dynamical system} associated to the Klein--Gordon equation on the de Sitter 
space.
\label{Wcqds-page}
\end{definition}

\begin{proposition} Let $I \subseteq S^1$. Then 
	\[
		\widehat\alpha_\Lambda^\circ \bigl( \, \widehat {\mathfrak W} (I) \bigr) \subset 
		\widehat {\mathfrak W} \bigl( \bigl( \Gamma^{+}(\Lambda I ) \cup \Gamma^{-}(\Lambda I) \bigr) \cap S^1 \bigr) \; . 
	\]
\end{proposition}

\begin{proof} This is a direct consequence of Proposition \ref{fsol}.
\end{proof}

\goodbreak
\section{De Sitter vacuum states}
\label{Sec-6-1}

Let  $ \alpha \colon \Lambda \mapsto \alpha_\Lambda$ be a representation (in terms of automorphisms) 
of $SO_0(1,2)$ on the $C^*$-algebra ${\mathfrak W}(dS)$. 

\begin{definition}
\label{vs}
A normalised positive linear functional $\omega$  is called a {\em de Sitter 
vacuum state} for the quantum dynamical system 
$\left({\mathfrak W} (dS) , \alpha  \right) $, if  
\begin{itemize}
\item[$ i.)$] $\omega$ is invariant under the action of the proper, orthochronous 
Lorentz group $SO_0(1,2)$, \emph{i.e.}, 
	\[
		\omega = \omega \circ \alpha_\Lambda  \qquad \forall \Lambda \in  SO_0(1,2)\; ;
	\]
\item [$ ii.)$] $\omega$ satisfies the {\em geodesic KMS condition\/}: for every wedge $W = \Lambda W_1$,  
$\Lambda \in  SO_0(1,2)$, the restricted (partial) state $\omega_{\upharpoonright {\mathfrak W}(W)}$ satisfies 
the KMS-condition at inverse temperature $2 \pi r$ with respect to the one-parameter  group 
	\[
		t \mapsto \Lambda_{W}\bigl( \tfrac{t}{r} \bigr)
	\]
of boosts, which leaves the wedge $W$ invariant. In other words:  
for each pair  
$[f],  [g] \in {\mathfrak k}(W)$ 
there exists a  
function~${\mathfrak F}_{f,g}(\tau)$ holomorphic in the strip 
	\[
	{\mathbb S}_{2 \pi r} =\{ \tau\in {\mathbb C} \mid 0  < \Im \tau <  2\pi r \}
	\] 
and continuous on~$\overline{{\mathbb S}_{2 \pi r}  } $ such that
	\begin{align*}
		{\mathfrak F}_{f,g}(t)&= \omega_{\upharpoonright {\mathfrak W}( W )} 
		\bigl({\rm W}([f])\alpha_{\Lambda_W (\frac{t}{r})} ({\rm W}([g])) \bigr) 
			\nonumber \\
		{\mathfrak F}_{f, g}(t + i 2\pi r)&=
		\omega_{\upharpoonright {\mathfrak W}( W)}
		\bigl(\alpha_{\Lambda_W (\frac{t}{r})} ({\rm W}([g])){\rm W}([f]) \bigr)   \qquad \forall  t\in {\mathbb R} \; .
	\end{align*}
\end{itemize}
\end{definition}

\begin{remark} It is sufficient to verify the geodesic KMS condition for 
{\em one} wedge, as the invariance property $i.)$ then implies that it 
holds for any wedge. 
\end{remark}

\bigskip
The {\em de Sitter vacuum state} for the free field is presented next.

\begin{theorem}
\label{th2.5} The state $\omega^\circ$ on ${\mathfrak W}( {dS})$ given by 
\label{freevacuumstatepage}
	\[
		\label{statedef}
		\omega^\circ({\rm W}([f]))
		= {\rm e}^{-\frac{1}{2} \| [f] \|_{ {\mathfrak h} (dS) }}, \qquad  f \in {\mathcal D}_{\mathbb{R}} ({dS}) \; ,
	\]
is the {\em unique} de Sitter vacuum state for the pair~$\bigl({\mathfrak W}( dS ), \alpha^\circ \bigr)$. 
Moreover, the GNS representation $\pi^\circ$ associated to the pair 
$\bigl({\mathfrak W}( dS) ,\omega^\circ \bigr) $
is (unitarily equivalent to) the  Fock representation (see Appendix \ref{Fockspace}) over the one-particle space ${\mathfrak h} (dS)$, \emph{i.e.}, 
	\[
 		\pi^\circ ({\rm W}( [f] )) 
		= {\rm W}_F ( [f]) \; , \qquad {\mathcal H} (dS) 
		= \Gamma ({\mathfrak h}(dS)) \; .
 	\]
Note that $\ker \mathbb{P} = \ker {\mathcal F}_{+ \upharpoonright \nu}$, thus $\omega^\circ$ is well-defined. 
\end{theorem}

\begin{proof} Recall the commutative diagram in Proposition \ref{Prop5.7}.
By construction,  $\omega^\circ$ is invariant under Lorentz transformations: for $f \in {\mathcal D}_{\mathbb{R}} (dS)$, we have
	\begin{align*}
		\label{invariantstate}
		\omega^\circ \bigl(\alpha_\Lambda ({\rm W}( \mathbb{P} f))\bigr) 
			&= {\rm e}^{-\frac{1}{2}\| [ \Lambda_*  f ] \|_{{\mathfrak h} (dS)}^2}
			= {\rm e}^{-\frac{1}{2}\| u (\Lambda)  [f] \|_{{\mathfrak h} (dS)}^2}\\
			&= {\rm e}^{-\frac{1}{2}\|  [f] \|_{{\mathfrak h} (dS)}^2}  \quad \qquad \qquad
			\qquad \qquad   \forall \Lambda \in  SO_0 (1, 2) \; .
	\end{align*}
The geodesic KMS condition follows from Proposition \ref{1PStrucHe} $iii.)$: for $f, g \in {\mathcal D}_{\mathbb{R}}(W)$ we find
\begin{align*}
& \omega^\circ \bigl({\rm W}([f])\alpha_{\Lambda_W (t)} ({\rm W}([g])) \bigr)
\\
& \qquad = {\rm e}^{- i \sigma ( [f], \Lambda_W (\frac{t}{r})_* [g] ) /2} \omega 
		\bigl({\rm W}([f]+ \Lambda_W (\tfrac{t}{r})_* [g]) \bigr) \\
		& \qquad =  {\rm e}^{- i \Im \langle [f], u (\Lambda_W (\frac{t}{r})) [g] \rangle_{{\mathfrak h}(dS)} } 
		{\rm e}^{-\frac{1}{2}\| [f] + u(\Lambda_W (\frac{t}{r}))[g] \|_{{\mathfrak h} (dS)}^2} \\
		& \qquad =  {\rm e}^{- i \langle [f], u (\Lambda_W (\frac{t}{r})) [g] \rangle_{{\mathfrak h}(dS)} } 
		{\rm e}^{-\frac{1}{2}\| [f]  \|_{{\mathfrak h} (dS)}^2-\frac{1}{2}\| [g] \|_{{\mathfrak h} (dS)}^2} \; . 
\end{align*}  
Moreover, Proposition \ref{1PStrucHe} $iii.)$ implies that for $f, g \in {\mathcal D}_{\mathbb{R}}(W)$ 
\begin{align*}
\langle u (\Lambda_W ( i \pi )) [f], u (\Lambda_W ( i \pi )) [g] \rangle_{{\mathfrak h}(dS)} 
& = \langle u (\Theta_W ) [f], u (\Theta_W ) [g] \rangle_{{\mathfrak h}(dS)}
\\
& = \langle [g], [f] \rangle_{{\mathfrak h}(dS)} \; . 
\end{align*}  
Together these two identities imply 
\[
\omega^\circ \bigl({\rm W}([f])\alpha_{\Lambda_W (\frac{t}{r} + i 2 \pi)} ({\rm W}([g])) \bigr)
= \omega^\circ \bigl(\alpha_{\Lambda_W (\frac{t}{r})} ({\rm W}([g])) {\rm W}([f]) \bigr)\; . 
\]  
Uniqueness of the geodesic KMS state follows by expressing the $n$-point functions in terms of two-point functions, 
using the KMS condition for the wedge (see \cite[Vol.~II, p.~49] {BR}).  Uniqueness of the two-point function
in the wedge follows from Proposition \eqref{restrictedOnePartdS} $iv.)$, \emph{i.e.}, the fact that zero is not an eigenvalue
of $\varepsilon_{\upharpoonright I_+} \ge 0$.

\goodbreak
The GNS representation $({\mathcal H} (dS), \pi^\circ, \Omega)$ is characterised, up to unitary equivalence, 
by the following two properties:
\begin{itemize}
\item [$ i.)$] 
$(\Omega, \pi^\circ ({\rm W}([f])) \Omega) 
= \omega^\circ \bigl( {\rm W}([f])\bigr)$ 
		for all ${\rm W}([f]) \in {\mathfrak W}( dS)$; 
\item [$ ii.)$] the vector $\Omega \in {\mathcal H} (dS)$ is cyclic for 
$\pi^\circ \bigl( {\mathfrak W}( dS )\bigr)''$. 
\end{itemize}
A short inspection of the Fock representation (see Appendix \ref{Fockspace}) 
verifies that the properties $i.)$ and $ii.)$ hold. 
\end{proof}

\emph{Notation}.
Denote the generators implementing the automorphisms corresponding to the subgroups 
$t \mapsto \Lambda_1(t)$, $s \mapsto\Lambda_2(s)$ and $ \alpha \mapsto R_0(\alpha)$ in the GNS 
representation~$\pi^\circ$ associated to a de Sitter vacuum state by $L_1$, $L_2$ and~$K_0$.

\label{AO-page}
\begin{definition}
The local von Neumann algebras for the free covariant field are defined by setting
	\[
		{\mathscr A}_\circ ({\mathcal O}) \doteq  \pi^\circ \bigl({\mathfrak W}
		({\mathcal O}) \bigr)'' \; ,  \qquad {\mathcal O} \subset dS \; . 
	\] 
It follows from Theorem \ref{th2.5} that the algebra ${\mathscr A} ({\mathcal O}) $ is equal to the von Neumann algebra generated 
by ${\rm W}_{F} (f )$, $f \in {\mathfrak h}({\mathcal O})$.
\end{definition}

\begin{definition}[Extension to the weak closure]
\label{weak-Ext}
The (anti-) unitary operators 
\[ U  (\Lambda) \doteq \Gamma (u(\Lambda))\; , \qquad \Lambda \in O(1,2) \; , \]
implement the free dynamics in the Fock space ${\mathcal H} (dS)$: for $f \in {\mathcal D}_{\mathbb{R}}(dS)$ we have
	\begin{align*}
		\pi^\circ \bigl( \alpha_\Lambda^\circ ( {\rm W}( [f] ) ) \bigr) 
		&= U  (\Lambda) {\rm W}_{F} \bigl( [f]) \bigr) U  (\Lambda)^{-1} \; , 
		\qquad \Lambda \in O(1,2) \; . 
	\end{align*}	
The right hand side extends to arbitrary elements $h \in {\mathfrak h}(dS)$, 
and, in the sequel, to arbitrary 
elements in the weak closure $\pi^\circ \bigl({\mathfrak W} (dS) \bigr)''$ of ${\mathfrak W} (dS)$. 
We denote this extension of the automorphism 
$\alpha_\Lambda^\circ$ by the same letter, \emph{i.e.}, for $h \in {\mathfrak h}(dS)$,
 	\begin{equation}
	\label{aut-ext}
		\alpha_\Lambda^\circ ( {\rm W}_F( h) ) \bigr) 
		\doteq U  (\Lambda) {\rm W}_{F} ( h)  U  (\Lambda)^{-1} \; , 
		\qquad \Lambda \in O(1,2) \; . 
	\end{equation}
Similarly, the GNS vector can be used to extend the free de Sitter vacuum state $\omega^\circ$ 
to the weak closure $\pi^\circ \bigl({\mathfrak W} (dS) \bigr)''$:
	\begin{equation}
	\label{state-ext}
	\omega^\circ (A) = ( \Omega , A \Omega ) \; , 
	\qquad A \in \pi^\circ \bigl({\mathfrak W} (dS) \bigr)'' \; . 
	\end{equation}
Here $\Omega$ denotes the vacuum vector in the Fock space ${\mathcal H}(dS)$. 
\end{definition}

\begin{proposition} The state \eqref{state-ext} satisfies the geodesic KMS condition 
with respect to the automorphisms \eqref{aut-ext}. 
\end{proposition}

\begin{proof} 
The geodesic KMS property for $\left({\mathfrak W} (W) , \alpha_{\Lambda_W}  \right) $
is part of Theorem \ref{th2.5}.
The fact that the KMS property 
extends to the weak closure is a standard result, see, \emph{e.g.}, \cite[Corollary 5.3.4]{BR}. 
\end{proof}

\section{The canonical free field}
\label{sec:freedesittercircle}

As $C^*$-algebras, the Weyl algebras $\widehat {\mathfrak W} (S^1)$ and  ${\mathfrak W} (dS)$ are isomorphic, and 
can be identified using the map (see Proposition \ref{nocheinlabel})
	\[	
		\widehat W ( \widehat f ) \mapsto  W([f]) \; , \qquad f \in {\mathcal D}_{\mathbb{R}}(dS) \;  . 
	\]
Moreover, for $f \in {\mathcal D}_{\mathbb{R}}(dS)$ we have  (see Proposition \ref{Prop5.7})
	\[
		{\rm e}^{-\frac{1}{2} \| \widehat K \widehat f \|_{\widehat {\mathfrak h} (S^1) }}
		= {\rm e}^{-\frac{1}{2} \| [f] \|_{ {\mathfrak h} (dS) } } .
	\]
Consequently, the state
	\begin{equation}
	\label{129a}
	\widehat \omega^\circ  \bigl( \widehat {{\rm W}} (\widehat f ) \bigr) 
	\doteq {\rm e}^{ - \frac{1}{2} \| \widehat K \widehat f \|_{\widehat{\mathfrak h} (S^1) } }\; , \qquad f \in {\mathcal D}_{\mathbb{R}}(dS) \; , 
	\end{equation}
describes the {\em same} (we will clarify exactly in which sense) state as the one given in Theorem \ref{th2.5}.

\begin{theorem}
\label{canonicalfreevacuumstatepage}
The state \eqref{129a} is the unique  normalised positive linear functional on 
$\widehat{\mathfrak W} \bigl( S^1 )$, which satisfies the 
following properties:
\begin{itemize}
\item[$ i.)$] $\widehat \omega^\circ$ is invariant under the action of $SO_0(1,2)$, \emph{i.e.}, 
	\[
		\widehat \omega^\circ = \widehat \omega^\circ \circ
		\widehat{\alpha}^\circ_\Lambda  \qquad \forall \Lambda \in  SO_0(1,2)\; ;
	\]
\item [$ ii.)$] $\widehat \omega^\circ$ satisfies the {\em geodesic KMS condition\/}: for every half-circle 
$I_\alpha$  the restricted (partial) state 	
	\[ 
	\widehat \omega_{\upharpoonright \widehat{\mathfrak W}( I_\alpha)}^\circ
	\]
satisfies the KMS-condition at inverse temperature $2 \pi r$ with respect to the one-parameter  group 
$t \mapsto \Lambda^{(\alpha)}(\frac{t}{r})$ of boosts. 
\end{itemize}
\end{theorem}

\begin{proof}
Property $i.)$ follows from the definition; property $ii.)$ will follow from Proposition \ref{identification-dS-EdS} and the properties of the 
time-zero covariance.  
\end{proof}

It is convenient to take the weak 
closure in the GNS representation $(\widehat{\pi}^\circ, \widehat{\mathcal H} (S^1), \widehat \Omega)$ for the 
pair $\bigl( \, \widehat{\mathfrak W} \bigl( S^1 ), \widehat \omega^\circ \bigr)$. The latter
is (unitarily equivalent to) the  Fock representation (see Appendix \ref{Fockspace}) over the one-particle 
space~$\widehat {\mathfrak h} (S^1)$, \emph{i.e.}, 
	\[
 		\widehat{\pi}^\circ \bigl( \, \widehat {\rm W}( \widehat f ) \bigr) 
		= {\rm W}_F ( \widehat {K}  \widehat f ) \; , \qquad \widehat{\mathcal H} (S^1) 
		= \Gamma \bigl( \, \widehat{\mathfrak h}(S^1) \bigr) \; .
 	\]
The GNS vacuum vector $\widehat \Omega$ can be used to extend $\widehat \omega$ to the weak closure: 
\[
	\widehat \omega (W_F (h)) \doteq ( \widehat \Omega, W_F (h) \widehat \Omega) 
	= {\rm e}^{-\frac{1}{2} \| h \|_{ \widehat{\mathfrak h} (S^1) } } \; , \qquad h  \in \widehat{\mathfrak h} (S^1) \; .
\]

\label{RI-page}
\begin{definition}
The local von Neumann algebras for the free canonical field are defined by 
	\[
		{\mathcal R} (I) \doteq  \pi^\circ \bigl( \,
		\widehat{\mathfrak W} ({\mathfrak k}(I) )\bigr)'' \; ,  \qquad I \subset S^1 \; . 
	\] 
A similar argument to the one given in the proof of Theorem \ref{th2.5} shows that 
the algebra ${\mathcal R} (I) $ is equal to the von Neumann algebra generated 
by $\widehat {\rm W}_{F} (h)$, $h \in \widehat{\mathfrak h}(I)$.
\end{definition}

\begin{proposition} \quad
\label{l6.1}
\begin{itemize}
\item[$i.)$] The local von Neumann algebras for the canonical free field are regular from the inside and regular from the outside:
	\[
		\bigcap_{ J \supset \overline {I} } 
		{\mathcal R} (J)={\mathcal R}  (I)
		= \bigvee_{ \overline{J}\subset I }   {\mathcal R}(J)\; ; 
	\]
\item[$ii.)$]  The net $I \mapsto {\mathcal R}(I)$ of local von Neumann 
algebras for the canonical free field is {\rm additive}:
	\[
		{\mathcal R}(I)= \bigvee_{J_i} {\mathcal R}(J_i) \qquad 
		\hbox{ if }  I = \cup_i J_i \;  .  
	\]
Moreover, 
	\[
	{\mathcal R}(S^1) = {\mathcal B} \bigl(\Gamma (\widehat{\mathfrak h} (S^1))\bigr)\; , 
	\qquad {\mathcal R}(S^1)' = {\mathbb C} \cdot \mathbb{1} \; ; 
	\]
\item[$iii.)$] For each open  interval $I \subset S^1$, 
the local observable algebra ${\mathcal R}(I)$ is $*$-isomorphic to the
unique hyper-finite factor of type~{\rm III}$_1$. 
\end{itemize}
\end{proposition}

\begin{proof}
Clearly, 
	\begin{equation}
	\bigcap_{J \supset \overline {I} }  \widehat{\mathfrak h} ( J)
	=\bigvee_{\overline {J} \subset I } \widehat{\mathfrak h}( J)
	= \widehat{\mathfrak h} ( I) \; ,
	\label{ekg5}
	\end{equation}
which together with Proposition \ref{araki} (see  Eq.~\eqref{et.2n}, Appendix E) implies $i.)$ and~$ii.)$. 
The proof of~$iii.)$ will be given in \cite{BMJ}. 
\end{proof}

\begin{remark}
\label{cor2}
A special case of $i.)$ is the following: let $I$ be an open interval contained in a half-circle. 
Then 
	\[
	{\mathcal R} (I) = \bigcap_{I \subset I_\alpha} 
	{\mathcal R} ( I_\alpha ) \; , 
	\]
where the $I_\alpha$'s are the half-circles containing $I$.
\end{remark}

\begin{theorem}
[Finite speed of propagation]
\label{fst-theorem}
Let $I \subset S^1$ 
be an open interval. Then 
	\begin{equation}
		\widehat 
		\alpha^{\circ}_{\Lambda^{(\alpha)} (t)}
 		\colon  
		{\mathcal R} (I)\hookrightarrow 
		{\mathcal R} \bigl( I (\alpha , t)  \bigr) \; .
		\label{e6.1f}
	\end{equation}
\end{theorem}

\begin{proof} 
This result follows directly from Proposition \ref{ialpha} and Corollary \ref{halpha}.
\end{proof}

\section{Analyticity properties of the correlation functions}

Consider the Weyl operators ${\rm W}( f_{j})$, $ f_j \in {\mathfrak h}^\circ (W_1)$, $1\leq j\leq n$, and the {\em correlation functions}
	\[
		G \bigl(t_{1}, \dots, t_{n}; {\rm W}( f_{1}), \dots, {\rm W}( f_{n}) \bigr)
		\doteq  \omega^\circ \Bigl(\prod_{j=1}^{n} \alpha_{\Lambda_1(t_j)} ({\rm W}( f_{j})) \Bigr) \; . 
	\]
More specifically, let us first consider an element $[f] \in {\mathfrak k} ({\rm W}_1)$, together with its Cauchy data 
$\widehat f \in {\mathfrak k} (I_+)$. It follows that 
	\[
	\| [f] \|_{{\mathfrak h} (dS)} = \| \widehat K \widehat f \|_{\widehat {\mathfrak h} (S^1)}
	=  \| \mathbb{u} \circ \widehat K_{2 \pi } \widehat f \|_{\widehat {\mathfrak h} (S^1)}
	=  \sqrt{r} \, \|  \widehat K_{2 \pi } \widehat f \|_{\widehat {\mathfrak d} (S^1)} \; . 
	\]
We have used Eq.~\eqref{can-geo2} and Proposition \ref{Prop5.6}. Recall that 
	\begin{equation*}
		\widehat{K}_{2 \pi } \widehat f \doteq
			\big((1+\rho_{2 \pi} )^{{\frac{1}{2}}}+\rho_{2 \pi}^{\frac{1}{2}}\, j \big) \;
 			\widehat{K}_\infty \widehat f \; ,  
	\end{equation*}
with $\rho_{2 \pi} \doteq \frac{{\rm e}^{-2 \pi |\varepsilon|}}{ 1+ {\rm e}^{-{2 \pi}|\varepsilon|} }$ and
$(1+\rho_{2 \pi}) = \frac{1}{ 1+ {\rm e}^{-{2 \pi} |\varepsilon|} }$. Since $\widehat{K}_\infty \widehat f 
\subset \widehat {\mathfrak d} (I_+)$, no cross terms arise, and consequently
	\[
		\|  \widehat K_{2 \pi } \widehat f \|_{\widehat {\mathfrak d} (S^1)}
		=  \| (1+ 2 \rho_{2 \pi })^{\frac{1}{2}} \widehat K_{\infty} \widehat f \|_{\widehat {\mathfrak d} (S^1)} \; . 
	\]
Now compute  
	\begin{align*}
		 &G\bigl(t_{1}, \dots, t_{n}; {\rm W}( f_{1}), \dots, {\rm W}(f_{n})\bigr)
			= \omega^\circ \Bigl(\prod_{j=1}^{n} {\rm W} \bigl( u (\Lambda_1( \tfrac{t_j}{r}) f_{j} ) \bigr)  \Bigr)
			\nonumber \\
			&\qquad =  \Bigl( \prod_{1\leq i< j\leq n}{\rm e}^{- i 2 \Im \langle u (\Lambda_1 (\frac{t_i}{r}) ) f_{i} \, , \,  
			u (\Lambda_1(\frac{t_j}{r})) f_{j}  \rangle_{\mathfrak{h}(dS)} } \Bigr)
			\omega^\circ \Bigl({\rm W} \bigl(\sum_{j=1}^{n} u \bigl(\Lambda_1\bigl(\tfrac{t_j}{r} \bigr)\bigr) f_{j} \bigr) \Bigr) 
			\nonumber \\
			&\qquad = \Bigl(\prod_{1\leq i< j\leq n}{\rm e}^{- i 2  r 
			\Im  \langle \, \widehat K_{\infty} \widehat f_{i} \, , \,  
					{\rm e}^{i \frac{t_{j}-t_{i}}{r} \varepsilon } 
					\widehat K_{\infty} \widehat f_{j}  \rangle_{ \widehat  {\mathfrak  d}  (I_+ )} } \Bigr)
					\nonumber \\
			&\qquad 			 \qquad	\qquad \qquad\times 
					{\rm e}^{-\frac{r }{4}\langle  \sum_{1}^{n}
					{\rm e}^{i \frac{t_j}{r} \varepsilon } \widehat K_{\infty} 
					\widehat f_{j} \; , \; (1+ 2 \rho_{2 \pi })\sum_{1}^{n}
					{\rm e}^{i \frac{t_j}{r} \varepsilon } \widehat K_{\infty} 
					\widehat f_{j} \rangle_{ \widehat  {\mathfrak  d}  (I_+ )} } 
					\nonumber \\
			&\qquad = \Bigl( \prod_{1\leq i< j\leq n}
					{\rm e}^{-\frac{1}{2}R_{\frac{t_{j}- t_{i}}{r}}
					( \,  \widehat K_{\infty} \widehat f_{i},\,  \widehat K_{\infty} \widehat f_{j})}\Bigr)
					\Bigl(
					\prod_{i=1}^{n}{\rm e}^{-\frac{1}{4} \left\langle \, \widehat K_{\infty} \widehat f_{i}, (1+ 2 \rho_{2 \pi }) 
					\widehat K_{\infty} \widehat f_{i}\right\rangle_{ \widehat  {\mathfrak  d}  (I_+ )} }\Bigr)  \; ,
	\end{align*}
where
	\begin{align*}
		R_{\frac{t}{r}} (  \, \widehat K_{\infty} \widehat f_1,  \widehat K_{\infty} \widehat f_2)
			&=  r \, \left\langle \, \widehat K_{\infty} \widehat f_1\; , \; \tfrac{{\rm e}^{i t \varepsilon }} 
				{\varepsilon(1-{\rm e}^{- 2 \pi \varepsilon} ) } \, \widehat K_{\infty} 
				\widehat f_2 \right\rangle_{L^2( I_+ , | \cos \psi|^{-1} r {\rm d} \psi)}
				\nonumber \\  
			& \qquad +  r \, \left\langle \, \widehat K_{\infty} \widehat f_2 \; , \; \tfrac{{\rm e}^{- 2\pi  \varepsilon  }}
			{\varepsilon(1-{\rm e}^{- 2 \pi  \varepsilon}) } \, 
			{\rm e}^{i t\varepsilon } \widehat K_{\infty} \widehat f_1 \right\rangle_{L^2( I_+ , | \cos \psi|^{-1} r {\rm d} \psi)}  \; .
	\end{align*}
For $ f_1,  f_2\in {\mathfrak h} (W_1)$ the function $t\mapsto R_{t/r} (  \, \widehat K_{\infty} \widehat f_1,  \widehat K_{\infty} \widehat f_2)$ allows a holomorphic
extension to the strip $\{ \tau \in {\mathbb C} \mid 0 < \Im \tau < 2\pi r\}$.  Consequently,  the function 
	\[ 
		(t_{1}, \dots, t_{n})\mapsto G\bigl(t_{1}, \dots, t_{n}; {\rm W}( f_{1}), \dots, {\rm W}( f_{n})\bigr) 
	\] 
is holomorphic in the set
	\[
		I_{2 \pi r}^{n+}= \bigl\{(\tau_{1}, \dots, \tau_{n})\in {\mathbb C}^{n} 
						\mid \Im\tau_{i}< \Im\tau_{i+1},  \; \; \Im\tau_{n}- \Im\tau_{1}<2\pi r \bigr\} \; ,
	\]
and continuous on $\overline{I_{2 \pi r}^{n+}}$. The  holomorphic extension is
	\[
		(\tau_{1}, \dots, \tau_{n})\mapsto \prod_{i=1}^{n}{\rm e}^{-\frac{ r }{4}\left\langle \, \widehat K_{\infty} \widehat f_{i}, 
		(1+2 \rho_{2 \pi }) \widehat K_{\infty} \widehat f_{i}\right\rangle_{ \widehat  {\mathfrak  d}  (I_+ )} }
		\prod_{1\leq i< j\leq n}{\rm e}^{-\frac{1}{2}R_{\frac{\tau_{j}- \tau_{i}}{r}} 
		( \,\widehat K_{\infty} \widehat f_{i}, \, \widehat K_{\infty} \widehat f_{j})} \; .
	\]
Hence the {\em Euclidean Green's functions}  
	\begin{align*}
			G^{\scriptscriptstyle E} & \bigl( s_{1}, \dots, s_{n}; {\rm W}( f_{1}),  \dots, {\rm W}( f_{n})\bigr) 
				\doteq  G \bigl(  i s_{1}, \dots,  is_{n}; {\rm W}(  f_{1}), \dots, {\rm W}(  f_{n})\bigr) \\ 
				&=  \prod_{i=1}^{n}{\rm e}^{- \frac{1}{2}  C_0 ( \cos_\psi^{-1} \widehat K_{\infty} \widehat f_{i} ,  \, 
				\cos_\psi^{-1} \widehat K_{\infty} \widehat f_{i})}
					\prod_{1\leq i<j\leq n}{\rm e}^{- C_{ \frac{| s_{j}-s_{i}|}{r}} 
					(  \cos_\psi^{-1} \widehat K_{\infty} \widehat f_{i}, \cos_\psi^{-1}  \widehat K_{\infty} \widehat f_{j})} \; , \nonumber  
	\end{align*}
where 
	\begin{align}
			\label{coideq}
			& C_{\frac{|s|}{r}} \bigl( \, \cos_\psi^{-1} \widehat K_{\infty} \widehat f_1, \cos_\psi^{-1} \widehat K_{\infty} \widehat f_2 \bigr) 
			\nonumber \\
			& \qquad \qquad = r \, 
			\Bigl\langle  \, \overline{   \widehat K_{\infty} \widehat f_1 },  \tfrac{{\rm e}^{-   |s|\varepsilon   }
					+ {\rm e}^{- (2 \pi  -|s|) \varepsilon    }}{
						2 \varepsilon  (1-{\rm e}^{- 2 \pi  \varepsilon  })} 
					\,  \widehat K_{\infty} \widehat f_2 \Bigr\rangle_{L^{2}(I_+  , \frac{r {\rm d} \psi}{\cos \psi})} \; ,  
	\end{align}
with $\varepsilon^2  =-  (\cos \psi \, \partial_\psi)^2 +(\cos \psi)^{2}\mu^2 r^2$.  
For $\widehat K_{\infty} \widehat f_{j} $ real valued\footnote{As a consequence of Lemma \ref{tztf},
$\widehat {f}_j $  (and in the sequel $\widehat K_{\infty} \widehat f_j$) is real valued if $[f] \in {\mathfrak h}^\kappa$.}
and $1\leq j\leq n$,
	\begin{align}
		\label{mf1}
			&G^{\scriptscriptstyle E} \bigl( s_{1}, \dots, s_{n}; {\rm W}( f_{1}), \dots, {\rm W}( f_{n})\bigr) 
			=\prod_{1\leq i, j\leq n}{\rm e}^{-\frac{1}{2} C_{\frac{ |s_{i}-s_{j}|}{r} } (\, \cos_\psi^{-1}
			\widehat K_{\infty} \widehat f_{i}, \, \cos_\psi^{-1} \widehat K_{\infty} \widehat f_{j})} \; . 
	\end{align}

We  now specialise this result to time-zero test-functions.   

\begin{proposition} 
\label{identification-dS-EdS}
Let $f_i = \delta \otimes h_i$ with $h_i \in {\mathcal D}_{{\mathbb R}}(I_+ )$, $ i= 1, \ldots, n$. 
It follows that the map  
	\begin{align*}
		& 
		 (t_{1}, \dots, t_{n}) \mapsto 
		  G\bigl(t_{1}, \dots, t_{n}; {\rm W} \bigl( [f_1] \bigr), \dots, {\rm W} \bigl( [ f_n ] \bigr)\bigr)    
		\\  
		& \qquad   \doteq  \Bigl(\Omega ,  {\rm e}^{ it_1 L_{1}} 
		{\rm W}_{F} \bigl( [f_1]  \bigr) 
		{\rm e}^{i (t_2 - t_{1}) L_{1}} 
		{\rm W}_{F} \bigl( [f_2] \bigr)   
		\cdots 
		{\rm e}^{i(t_{n} -t_{n-1}) L_{1}} 
		{\rm W}_{F} \bigl( [f_n] \bigr) \Omega \Bigr)  \nonumber 
	\end{align*}
is holomorphic in the set
	\[
		{\mathcal I}_{2 \pi r}^{n+}= \bigl\{(z_{1}, \dots, z_{n})\in {\mathbb C}^{n} \mid \Im z_{i}<
		\Im z_{i+1}, \; \Im z_{n}- \Im z_{1}<2\pi r \bigr\},
	\]
and continuous on $\overline{{\mathcal I}_{2 \pi r}^{n+}}$.  Moreover, 
	\[
		G\bigl( i \theta_{1}, \dots,  i\theta_{n}; {\rm W} \bigl( [ f_1 ] \bigr), \dots, {\rm W} \bigl([ f_n ] \bigr)\bigr)
		=\prod_{1\leq i, j\leq n}{\rm e}^{-\frac{1}{2} C_{|\theta_{i}-\theta_{j}|} (   h_{i},  h_{j}) } \; , 
	\]  
where $C_{|\theta_{i}-\theta_{j}|} (   h_{i},  h_{j})$ is defined in Eq.~(\ref{coideq}). 
\end{proposition}

\begin{proof}
Recall that for $ h\in {\mathcal D} \left(I_+ \right)$
one can extend the domain of ${\mathcal F}_{+ \upharpoonright \nu}$ to distributions of the form 
	 \[
		f (x) \equiv (\delta \otimes h) (x) =	\delta (t)  \;  \frac{ h  (0,\psi )}{ r \cos \psi }\;  , 
	\] 
with $x \equiv x (t, \psi)$ the coordinates introduced in \eqref{w1psi} and 
	$
		{\rm d}\mu_{{\mathbb W}_1 } (t,\psi) = r^2 {\rm d}t\,  {\rm d}\psi   \cos \psi  $. 
As we have seen, this leads to
	\[
		\big( \mathbb{f}_{\upharpoonright S^1}, 
							 (n \, \mathbb{f})_{\upharpoonright S^1})
   = (0, - h) \equiv  (\mathbb{\phi}, \mathbb{\pi}) \; , 
	\]
and finally, recalling	\eqref{eqKdotf}, \emph{i.e.},  $\bigl( \widehat{K}_\infty (\mathbb{\phi}, \mathbb{\pi}) \bigr) (\psi)  
\doteq   \cos \psi \; \mathbb{\pi} (\psi) - i \,(\varepsilon \mathbb{\Phi}) (\psi)$, we find
	\[
		\widehat{K}_\infty \big( \mathbb{f}_{\upharpoonright S^1}, 
							 (n \, \mathbb{f})_{\upharpoonright S^1})
   = - \cos_ \psi h \; . 
	\]
Together with \eqref{mf1} this proves the result.
\end{proof}

\part{Euclidean Quantum Fields}

\chapter{The Euclidean Sphere}

We will now introduce Markov fields\index{Markov field} on the sphere, which 
later on will allow us to reconstruct 
free quantum field on the de Sitter space. In order to define a probability measure on the sphere, we
briefly review the geometrical setting. We consider the \emph{Euclidean sphere}\index{Euclidean sphere}\footnote{We have
changed the notation; see  \eqref{euclidsphere} for comparison.}
	\[ 
		S^2 \doteq \{ \vec x \in \mathbb{R}^{3} \mid  x_0^2 +  x_1^2 + x_2^2 = r^2  \} \; , 
	\]
embedded in ${\mathbb R}^3$.  Let  $\vec 0 =(0, 0, 0)$ denote the origin in ${\mathbb R}^3$. 
The upper (resp.~lower) hemisphere\index{hemisphere} is 
\label{spherepluspage}
	\[
		S_\pm \doteq \{ \vec x \in S^2 \mid \pm x_0 >  0 \}  \; . 
	\]
The equator\index{equator} is $S^1 \doteq \{  \vec x \in S^2 \mid x_0=0 \}$.  
The hemispheres $S_\pm$ are open, their boundaries $\partial S_{\pm} = S^1 $ coincide with 
the equator, and $S^2$ is the disjoint union 
\label{equatoreuclidpage}
\label{eulidspherepage}
	\[ 
		S^2 = S_-  \mathbin{\dot{\cup}} S^1 \mathbin{\dot{\cup}} S_+ \; . 
	\]
The \emph{Euclidean time reflection}\index{time reflection (Euclidean)}  
	\begin{equation}
		\label{deftimerefl}
		T \colon ( x_0, x_1, x_2)\mapsto (- x_0, x_1, x_2)
	\end{equation}
maps~$S_\pm$ onto $ S_\mp$ and leaves $S^1$
invariant. $S^1$ itself is the disjoint union 
	\[  
		I_+  \mathbin{\dot{\cup}} \{ ( 0,- r, 0) , ( 0,r, 0)  \} \mathbin{\dot{\cup}} I_- \; ,  
	\]
with $I_\pm \doteq \{  \vec x \in S^1 \mid   \pm x_2> 0 \}$ open half-circles.

\section{The symmetry group of the sphere}
\label{SO3}

The group of rotations\index{group of rotations} $SO(3)$  leaves the 
sphere $S^2$ invariant. We denote the generators of the rotations  
		\[ 
				R_{0}(\alpha) \doteq \begin{pmatrix}
										1 &  0 &0 \\
										0 &  \cos \alpha & - \sin \alpha  \\
										0  & \sin \alpha & \cos \alpha   
								\end{pmatrix} \; , \qquad  
				R_{1}(\beta)  \doteq  \begin{pmatrix}
										 \cos \beta &  0 & - \sin \beta \\
										0 &  1&0  \\
										 \sin \beta & 0 & \cos \beta   
 								\end{pmatrix} \; , 
		\] 
and
		\[ 
				R_{2}(\gamma) \doteq \begin{pmatrix}
 								\cos \gamma &  - \sin \gamma &0 \\
								 \sin \gamma &  \cos \gamma &0  \\
								0 & 0 & 1  
									\end{pmatrix} \; , 
				\qquad \alpha, \beta,\gamma\in [0, 2 \pi ) \; , 
		\] 
around the three coordinate axis by $K_0$, $K_1$ and $K_2$, and  set
	\[
		K^{(\alpha)} \doteq \cos \alpha \; K_1 + \sin \alpha  \; K_2\;  , \qquad \alpha \in [ 0, 2 \pi) \; .
	\]
Denote by $R^{(\alpha)}$ the rotations generated by $K^{(\alpha)}$, namely the rotations 
	\[
		R^{(\alpha)} (\theta) 
		= R_{0} (\alpha) R_{1} (\theta) R_{0} (-\alpha) \; , \qquad \alpha, \theta \in [0, 2\pi) \; , 
	\]
which leave the boundary 
points $x_\alpha = (0, r \sin \alpha, r \cos \alpha)$ 
and $- x_\alpha$ of the time-zero half-circles 
	\begin{equation}
		\label{halfcirclealphapage}
		I_\alpha = R_{0} (\alpha) I_+
	\end{equation}
invariant. 

\section{Geographical and path-space coordinates}

We will now now define two charts, which will be convenient in the sequel. Together they provide an atlas for the 
sphere.
\subsection{Geographical coordinates}
\label{geo-chart}
The chart  
	\[
			\begin{pmatrix}
					        x_0 \\
						x_1 \\
						x_2 
			\end{pmatrix} 
				= \begin{pmatrix}
						r \sin \vartheta  \\
 						r \cos \vartheta \sin \varrho  \\
 						r \cos \vartheta \cos \varrho  \\
				\end{pmatrix} \; , \qquad - \frac{\pi}{2} <\vartheta <  \frac{\pi}{2} \;, \quad  - \pi \le \varrho <  \pi \; , 
	\]
covers the sphere, except for the geographical poles\index{geographical pole}
$(\pm r, 0, 0) \in {\mathbb R}^3$. Refer to $(\vartheta, \varrho )$ as {\em geographical coordinates}\index{geographical coordinates}. 
The equator
$S^1 \cong \{ (\vartheta, \varrho ) \mid \vartheta =0 \}$ and the point $ (\vartheta, \varrho) \equiv (0, 0)$ 
is mapped to the origin $ \vec o =(0, 0,r)$. The restriction of the Euclidean metric to this chart  is
	\[
	g = r^2 {\rm d} \vartheta \otimes {\rm d} \vartheta + r^2 \cos^2 \vartheta \; ({\rm d} \varrho \otimes {\rm d} \varrho )
	\]
and
	\begin{align}
		\label{L1}
		\Delta
					&= |g|^{-1/2} \partial_\mu \Bigl( | g|^{1/2} g^{\mu \nu} \partial_\nu \Bigr)
						\nonumber \\
					&=  \frac{1}{r^2 \cos_\vartheta^2 } 
						\left( \Bigl( \cos_\vartheta \frac{\partial}{\partial \vartheta} \Bigr)^2
						+ \frac{\partial^2}{\partial \varrho^2}  \right) \; .
	\end{align}
Here $\cos_\vartheta$ denotes the multiplication operator
$(\cos_\vartheta f)(\vartheta) \doteq \cos \vartheta  \, f(\vartheta)$ acting on functions of $\vartheta$.
The surface element on $S^2$ is ${\rm d}\Omega(\vartheta, \varrho ) = r^2 \cos \vartheta \,{\rm d} \vartheta {\rm d} \varrho$. 

\subsection{Path-space coordinates}
The chart 	\eqref{ps-cord},
	\begin{equation}
		\begin{pmatrix}
			x_0 \\
			x_1 \\
			x_2 
		\end{pmatrix} 
			= 	\begin{pmatrix}
 						r \sin \theta \cos \psi  \\
						r \sin  \psi  \\
						r \cos \theta \cos \psi  \\
				\end{pmatrix} \; , \qquad 0 \le \theta < 2 \pi  \; , \quad - \frac{ \pi}{2} < \psi <  \frac{ \pi}{2} \; , 
	\end{equation}
covers the sphere with exception of 
the two points $(0, \pm r, 0) \in {\mathbb R}^3$. Refer to this chart as {\em path-space coordinates}\index{path-space coordinates}.
The point $(\theta, \psi)\equiv (0, 0)$ is mapped to the origin $\vec o =(0, 0,r)$.
The restriction of the Euclidean metric to this chart  is
\label{Laplacepage}
	\[
	g =  \cos^2 \psi  \, ({\rm d} \theta \otimes {\rm d} \theta) + {\rm d} \psi \otimes {\rm d} \psi  
	\]
and
	\begin{equation}
	\label{L2}	
	\Delta
		= \frac{1}{r^2 \cos_\psi^2} \left( \frac{\partial^2}{\partial \theta^2}  
		+ \Bigl( \cos_\psi \frac{\partial}{\partial \psi} \Bigr)^2\right) .
	\end{equation}
The surface element on $S^2$ is ${\rm d}\Omega(\theta, \psi ) = r^2 \cos \psi \,{\rm d} \theta {\rm d} \psi$.

\subsection{The Laplace operator}
The expressions in \eqref{L1} and \eqref{L2} both extend to the self-adjoint \emph{Laplace operator}\index{Laplace operator}  
$\Delta_{S^2}$ on $L^2(S^2, {\rm d} \Omega)$.  $- \Delta_{S^2}$~has non-negative
discrete spectrum and an isolated simple eigenvalue at zero with eigenspace the constants. 

\chapter{Gaussian Measures}

\section{The definition of the measure}
\label{GMPATH}

Let $\Sigma$ be the $\sigma$-algebra\index{$\sigma$-algebra} generated by the \emph{Borel cylinder sets}\index{Borel cylinder set} 
of ${\mathcal Q}:={\mathcal D}'_{{\mathbb R}}(S^2)$, the dual space\index{dual space} of $C^\infty_{{\mathbb R}}(S^2)$. 
For $f\in C^\infty_{{\mathbb R}}(S^2)$, let $\Phi(f) \colon  {\mathcal Q}  \to  {\mathbb C}$ be the \emph{evaluation map}\index{evaluation map}  
\label{dualitypage}
	\begin{equation}
		\label{cmap}
			q  \mapsto   \langle q, f \rangle  \; .  
	\end{equation}
$\langle \, . \, , \, . \, \rangle$ is the duality bracket\index{duality bracket}.
For $F$ a Borel function\index{Borel function} on~${\mathbb R}$, define  
	\[
		\begin{matrix}
		 F (\Phi(f)) \colon & {\mathcal Q} & \to & {\mathbb C}  \\
		& q & \mapsto & F \bigl( \langle q, f\rangle \bigr) \; .\\
		\end{matrix} 
	\]
The Fourier transform\index{Fourier transform of measure} of  
a \emph{Gaussian measure}\index{Gaussian measure}  ${\rm d} \Phi_{C}$  on~$({\mathcal Q}, \Sigma)$ with 
\emph{covariance}\index{covariance}	
	\begin{equation}
		\label{c}
		C (f,g)=\bigl\langle \overline{f},(- \Delta_{S^2} +  \mu^{2})^{-1}g \bigr\rangle_{L^{2}( S^2, {\rm d} \Omega)} \;
		,\qquad f,g\in C^\infty(S^2) \; ,
	\end{equation}
is
	\begin{equation}
		\label{e1.00}
			\int_{{\mathcal Q}} {\rm d}\Phi_{C} \; {\rm e}^{i \Phi(f)} 
			=  {\rm e}^{- C (f,f)/2}, \quad f\in C^\infty_{{\mathbb R}}(S^2) \; .
	\end{equation}
According to \emph{Minlos' theorem}\index{Minlos theorem} \cite{GJ}\cite{S}\cite{Bour}, Equ.~\eqref{e1.00}~defines 
a unique probability measure\index{probability measure}, namely ${\rm d}\Phi_{C}$, on~$Q$.
 
The mean of the Gaussian measure~${\rm d}\Phi_{C}$ is zero, and the covariance $C$ coincides with  the second 
momentum.
More generally,  
	\begin{equation}
		\int_{{\mathcal Q}}{\rm d}\Phi_{C} \; \Phi(f)^{p}=
			\left\{
				\begin{array}{l}
					0 \, , \qquad \qquad \qquad \qquad \; \; \,  p\hbox{ odd}\\
					(p-1)!! \, C(f,f)^{p/2},\quad p\hbox{ even}
				\end{array}
			\right. 
			\; .
		\label{e1.0}
	\end{equation}
It follows from (\ref{e1.0}) that 
	\[
		{\rm e}^{\Phi(f)}\in L^{1}({\mathcal Q}, \Sigma, {\rm d}\Phi_{C}) 
		\quad \text{if} \quad f\in C^\infty_{{\mathbb R}}(S^2) \; .
	\]
For $\mu>0$ the covariance (\ref{c})  defines a norm
$\| \, . \, \|_{-1} = C( \, . \, , \, .\, )^{1/2}$ on $C_{{\mathbb R}}^\infty(S^2)$. The completion of 
$C^\infty_{{\mathbb R}} (S^2)$ with respect to $\| \, . \, \|_{-1}$ is the \emph{Sobolev space}\index{Sobolev space} $\mathbb{H}^{-1} (S^2) $. 

\newpage

\begin{proposition}
\label{3.9a}
Let $f, g \in \mathbb{H}^{-1} (S^2)$. It follows that 
	\begin{align*}
		 \langle f, g \rangle_{ \mathbb{H}^{-1} (S^2) } 
		 & =  \left\langle f , (- \Delta_{S^2} +\mu^2)^{-1} g \right\rangle_{L^2( S^2, {\rm d} \Omega)}  
		 \\
		& = \frac{1}{ 2}  \int_{S^2}  {\rm d} \Omega (x)  \int_{S^2}  {\rm d} \Omega (y) \; 
		f(x)  \, c_\nu \, P_{s^+}  
		 \bigl( - \tfrac{  \vec x \cdot \vec y }{r^2} \bigr) \,   g(y)  \; .  
	\end{align*}
\end{proposition}

\begin{remark}
This result should be compared with Proposition \ref{legendre}, which says that for 
two functions $f, g \in C^\infty_0 (dS)$ one has 
	\begin{align*}
	 \langle [f],   [g] \, \rangle_{{\mathfrak h} (dS)} 		 
	 & =   \int_{dS} {\rm d} \mu_{dS} 
	( x)  \int_{dS} {\rm d} \mu_{dS} (  y)  \;  { f ( x ) } \, \underbrace{ c_\nu \, 
	P_{s^+} \left(  \tfrac{x_+ \cdot  y_-}{r^2} \right)}_{\mathcal W}(  x_+,  y_-) \, g ( y )  \; .
	\end{align*}
\end{remark}

\begin{proof} We recall that the spherical harmonics\index{spherical harmonics} 
	\begin{equation}
	\label{spherical harmonics}	
		Y_{l,k} (\theta, \psi)= \sqrt{ \frac{2l +1}{4 \pi} \frac{(l-k)!}{(l+k)!} } \; 
		P_l^k (\cos \theta) \,  {\rm e}^{i k \psi} \; , \qquad \theta \in [0, \pi] \; , 
	\end{equation}
with $\theta = 0$ at the north pole, are orthonormal 
	\[
	\int_{S^2} {\rm d} \Omega \, \overline{Y_{l',k'} (\theta, \psi)} Y_{l,k} (\theta, \psi) = r^2 \delta_{l, l'} \delta_{k, k'} \; , 
	\]
and satisfy 
	\[
	 \Delta_{S^2} Y_{l,k} (\theta, \psi)= - \frac{l (l+1)}{r^2} Y_{l,k} (\theta, \psi) \;  .  
	\]
Now consider two vectors $\vec x \equiv \vec x (\theta,\psi)$ and $\vec y \equiv \vec y (\theta',\psi')$ 
of length $ | \vec x | = | \vec y | = r$. It follows that the kernel of the operator
	\begin{align*}
		 (- \Delta_{S^2} + \mu^2)^{-1}  (\vec x, \vec y) & =   \sum_{l=0}^\infty \sum_{k=-l}^l 
		\frac{\overline{Y_{l,k} (\theta', \psi')} Y_{l,k} (\theta, \psi)}{l(l+1) + \mu^2r^2} 
		 \\
		& =   \frac{1} {4 \pi}\sum_{l=0}^\infty 
				\frac{2l+1}{l(l+1) + \mu^2r^2} P_l \bigl( \tfrac{\vec x \cdot \vec y}{r^2} \bigr) 
		 \\
		& =   \frac{1} {4 \pi}\sum_{l=0}^\infty  
				\frac{2l+1}{l(l+1) + \mu^2r^2} 
				(-1)^l   P_l \bigl( - \tfrac{  \vec x \cdot \vec y}{r^2} \bigr) \; . 
	\end{align*}
In the second equality we have used the summation formula \cite[p. 395]{WaWi}: 
	\[
	P_l \bigl( \tfrac{\vec x \cdot \vec y}{r^2} \bigr) = \frac{4 \pi}{2l+1}	\sum_{k=-l}^l 
		\overline{Y_{l,k} (\theta', \psi')} Y_{l,k} (\theta, \psi) \; . 
	\]
In the third equality we have used $	(-1)^l  P_l (z)  =  P_l (-z)$. 

The identity\footnote{This identity extends by analyticity from $\nu \in \mathbb{R}$ to the case
$\nu =i \sqrt{\frac{1}{4} -\mu^2r^2}$,  $ 0< \mu < 1 /2r$.} 
 \cite[Eq.~(23), page 205]{Neumann}
	\begin{align*}
	\int_{-1}^1 {\rm  d} z \; P_l (z) P_{-\frac{1}{2} - i \nu} (z)  & = \frac{2 \cos (i \nu  \pi ) }{\pi}  \frac{(-1)^l}{(l+\frac{1}{2})^2 +  \nu^2} 	\\
	& = \frac{2 \cos (i \nu  \pi ) }{\pi}  \frac{(-1)^l}{l(l + l) +\mu^2r^2}
	\end{align*}
shows that 
	\begin{align*}
	P_{-\frac{1}{2} - i \nu} (z)  
	& = \sum_{l=0}^\infty \left( \frac{2l +1}{2}\int_{-1}^1 {\rm  d} z' \; P_l (z') P_{-\frac{1}{2} - i \nu} (z') ) \right) P_l (z) \\
	& = \frac{ \cos (i \nu  \pi ) }{\pi} \sum_{l=0}^\infty (-1)^l \frac{ 2l +1 }{l(l + l) + \mu^2r^2}  P_l (z) \; .  
	\end{align*}
Thus 
	\begin{align*}
		 (- \Delta_{S^2} + \mu^2)^{-1}  ( \vec x, \vec y) & 
		 =     \frac{P_{-\frac{1}{2} - i \nu}  
		 \bigl( - \frac{  \vec x \cdot \vec y }{r^2} \bigr)   } {4 \cos (i \nu \pi )}  \; .  
	\end{align*}
\end{proof}

Next use  (\ref{c}) and (\ref{e1.0}) to show  the following result. 

\begin{lemma}
The map
	\begin{equation}
		\label{e1.6bb}
			\begin{matrix}
				& C^\infty_{{\mathbb R}}  (S^2) & \to  & \bigcap_{1\leq
				p<\infty}L^{p}({\mathcal Q}, \Sigma, {\rm d}\Phi_{C}) \\
				& f & \mapsto & \Phi(f) \\
			\end{matrix}  \,  
	\end{equation}
extends to a  continuous map from $\mathbb{H}^{-1} (S^2)$ to 
$\bigcap_{1\leq p<\infty}L^{p}({\mathcal Q}, \Sigma, {\rm d}\Phi_{C})$. 
\end{lemma}

\begin{lemma}
The cylindrical functions $F \bigl(\Phi(f_{1}), \dots, \Phi(f_{n}) \bigr)$, $f_{i}\in
C^\infty_{{\mathbb R}}(S^2)$, $F$ a Borel function on~${\mathbb R}^{n}$ and $n\in
{\mathbb N}$, are dense
in $L^{p}({\mathcal Q}, \Sigma, {\rm d}\Phi_{C})$ for  $1\leq p<\infty$. 
\end{lemma}

\section{Properties of the covariance}

The short distance behaviour of the covariance $C$ has been studied in \cite{AFG}. 
We briefly present their findings. The covariance $C$ can be expressed \cite{wald} 
in terms of the heat kernel $p(t, x, y)$ on the sphere~$S^2$ : 
	\[
		C (x, y) = \int_0^\infty {\rm d}t \; {\rm e}^{-t \mu^2} p(t, x, y)
	\]
This allows us to introduce the usual multi-scale decomposition (see, \emph{e.g.}, \cite{Benfatto}). The covariance C(x,y) is given by
	\[
		C(x,y)= \sum_{l=0}^\infty C_l (x,y) \; , 
	\]
where, for some fixed constant $\gamma$ larger than $1$,
	\[
		C_l (x,y) = \int_0^\infty {\rm d}t \; \left( {\rm e}^{-t \gamma^{2l} \mu^2} - {\rm e}^{-t \gamma^{2l+2} \mu^2}  \right) p(t, x, y)
	\]
is the kernel of the operator
	\[
		C_l (x,y) = \frac{1}{- \Delta_{S^2} +  \mu^{2}\gamma^{2l}  } - \frac{1}{- \Delta_{S^2} +  \mu^{2}\gamma^{2l+2} }
	\]
Following \cite{AFG}, we introduce the regularized covariance  
	\[
		C^{(k)} (x,y)= \sum_{l=0}^{\log_\gamma k - 1} C_l (x,y) \; , 
	\]
where, of course, $k$ is such that $\log_\gamma k $ ranges on the positive integers. 
$C^{(k)}$ represents the covariance of the field with length cutoff $\gamma (\mu k)^{-1}$, the analog, in the flat case, 
of a momentum cutoff of order $ \mu k$. In this sense, $C^{(k)}$  
compares with the~$\delta_k C \delta_k$ in the equations 
LR1, LR2, and LR3 of \cite[p.~160--161]{GJ}.

\begin{theorem}[De Angelis, de Falco and Di Genova \cite{AFG}]
\label{c-log}
Let $1 \le q < \infty$. With the notation introduced above,  we have   
	\begin{align*}
		\sup_{x \in S^2} \| C( x, \, . \, ) \|_{L^q(S^2, {\rm d} \Omega)}  & < + \infty \; , 
		\\
		\| C^{(k)} ( \, . \, , \, . \, ) - C ( \, . \, , \, . \, ) \|_{L^q(S^2 \times S^2, {\rm d} \Omega \otimes {\rm d} \Omega )}& \le O(k^{-2/q} ) \; , 
		\\
		\sup_{x \in S^2}  C^{(k)} ( x, x )  & \le O( \log_\gamma k) \; .  
	\end{align*}
\end{theorem}

\begin{remark}
The logarithmic nature of the singularity of the covariance
$C(x, y)$ at coinciding points 
	\[
		C(x,y) \sim \frac{1}{2\pi} \log  \mu d(x,y) 
	\]
follows from the asymptotic behaviour of the heat kernel as $t \downarrow 0$ \cite{Bellaiche}:
	\[
		p(t, x,y) \sim \frac{1} {4 \pi t} H(x, y) {\rm e}^{- \frac{d^2(x,y)}{4t}}
	\]
uniformly on all compact sets in $S^2 \times S^2$,  which do not intersect the  cut locus of~$S^2$. 
Here $d(x,y) $ is the geodesic distance and $H(x,y)$ the Ruse invariant (see, \emph{e.g.}, \cite{Walker}). 
The latter is equal to one in the case of constant curvature. \end{remark}

\section{Sobolev spaces}
\label{sec:3.1}
In \cite{D} Dimock defines real
Sobolev spaces over a Riemannian manifold $({\mathcal M},  g)$, with ${\mathcal M}$ an oriented compact connected finite-dimensional manifold 
and $g$ a positive definite metric. The Euclidean sphere  $S^2$ is a special case:
define for  $m >0$ fixed, the spaces $\mathbb{H}^{1} (S^2)$ as the completion of $C_{\mathbb R}^\infty (S^2)$ in the norm
\label{sobolevpage}
	\[
		\| h \|^2_{1} 
			= \left\langle h, (- \Delta_{S^2} +\mu^2)h \right\rangle_{L^2( S^2, {\rm d} \Omega)} \;  . 
	\]
The function spaces $\mathbb{H}^{\pm 1} (S^2)$ are real Hilbert spaces, $ | (f,g) | \le \| f \|_1 \| g \|_{-1} \; $, and
	\[
		C_{\mathbb R}^\infty (S^2) 
			\subset \mathbb{H}^1(S^2) \subset L^2_{\mathbb R} (S^2, {\rm d} \Omega) 
				\subset \mathbb{H}^{-1} (S^2)  \; . 
	\]
The inner product extends to a bilinear pairing of $\mathbb{H}^1(S^2)$ and $\mathbb{H}^{-1}(S^2)$. 
In fact, $\mathbb{H}^1(S^2)$ and $\mathbb{H}^{-1}(S^2)$ 
are dual to each other with respect to this pairing,  and the map
$f \mapsto (- \Delta_{S^2} + \mu^2)f$ is unitary from $\mathbb{H}^1 (S^2)$ to $\mathbb{H}^{-1}(S^2)$.

For a closed subset $C \subset S^2$, which contains an open subset of $S^2$, 
we define a closed subspace $\mathbb{H}^{-1}_{\upharpoonright C} (S^2)$ of $\mathbb{H}^{-1}(S^2)$:
	\begin{equation}
		\label{square}
		\mathbb{H}^{-1}_{\upharpoonright C} (S^2) = \{ f \in \mathbb{H}^{-1}(S^2) \mid {\rm supp\,} f \subset C \}   \; . 
	\end{equation}
For the open sets $S_\pm \subset S^2 $, let $\mathbb{H}^1_0(S_\pm)$ be the closure of 
$C^\infty_{0 {\mathbb R}} (S_\pm)$ in~$\mathbb{H}^1 (S^2)$. 

\begin {lemma}[Dimock \cite{D}, Lemma 1, p.~245] 
\label{detheo}  
	\begin{align*}
		\mathbb{H}^{-1} (S^2)  &
		= \mathbb{H}^{-1}_{\upharpoonright \overline{S_\mp}} (S^2) \oplus (- \Delta_{S^2} + \mu^2) \mathbb{H}^1_0 (S_\pm) \; , 
		\nonumber \\ 
		\mathbb{H}^{-1} (S^2) &
		= (- \Delta_{S^2} + \mu^2) \mathbb{H}_0^{1}(S_-) \oplus \mathbb{H}^{-1}_{\upharpoonright S^1} (S^2) \oplus 
			(- \Delta_{S^2} + \mu^2) \mathbb{H}^1_0 (S_+) \; . 
	\end{align*} 
\end{lemma} 

\smallskip
Let $e ({S^1})$ and $e \left(\overline{S_\pm}\right)$ denote the orthogonal projections from 
$\mathbb{H}^{-1}(S^2)$ onto $\mathbb{H}^{-1}_{\upharpoonright S^1} (S^2)$
and $\mathbb{H}^{-1}_{\upharpoonright \overline{S_\mp}} (S^2)$, respectively. 

\begin{lemma}[Dimock's pre-Markov property\index{pre-Markov property} \cite{D}, Lemma 2, p.~246] 
\label{dlemma} 
	\[
			e \left({\overline{S_\mp}}\right)  e \left({\overline{S_\pm}}\right) = e( S^1 )
			\qquad \text{on $\mathbb{H}^{-1} \left( S^2 \right)$.}
	\]
Thus $\mathbb{H}^{-1}_{\upharpoonright S^1} (S^2) =   \mathbb{H}^{-1}_{\upharpoonright \overline{S_+}} (S^2) 
\cap   \mathbb{H}^{-1}_{\upharpoonright \overline{S_-}} (S^2)$.
\end{lemma}

We note that the origins of Dimock's work can be traced back to \cite{GRS} and even further to \cite{N1, N2, N3, N4}. 

\section{Conditional expectations}
\label{sec:3.1x}

The Markov property for the sphere, presented in Theorem~\ref{martheo}~$ vi.)$ below, 
 is satisfied, iff  for any function of the Euclidean field  in 
$\overline{S_\pm}$, conditioning to the fields in $\overline{S_\mp}$ 
is the same as conditioning to the fields in $\partial S_\pm=S^1$. 

\begin{theorem}[Simon \cite{S}, Theorem III.7, p.~91]
\label{bsimon}
Let $({\mathcal Q}, \Sigma, \mu)$ be a probability space  and let $\Sigma'$ be a sub-$\sigma$-algebra of $\Sigma$.
Let~$F$ be an element of $L^1_{\mathbb R} ({\mathcal Q}, \Sigma, \mu)$. Then  there exists a unique function 
$E_{\Sigma'} F$ such that
\begin{itemize}
\item[$ i.)$] $ E_{\Sigma'} F $ is $\Sigma'$-measurable;
\item[$ ii.)$] for all $G \in L^\infty_{{\mathbb R}} ({\mathcal Q}, \Sigma, {\rm d} \mu) $ which are $\Sigma'$-measurable
	\[
		\int_{\mathcal Q} {\rm d} \mu \; G E_{\Sigma'} F   = \int_{\mathcal Q} {\rm d} \mu \;  G F \;  .
	\]
\end{itemize}
$ E_{\Sigma'} F$ is called the {\em conditional expectation} of $F$ given $\Sigma'$.
\end{theorem}

Next consider (see \cite[Vol.~II]{BR}) the Fock space $\Gamma( H_{{\mathbb C}}^{-1} (S^2))$  over the 
complexification $H_{{\mathbb C}}^{-1} (S^2)$ of $H^{-1}(S^2)$. Let $\Phi_E$ be the Fock space field operator and 
$\Omega_E $  the  vacuum vector in~$\Gamma( H_{{\mathbb C}}^{-1} (S^2))$. The map
\label{fockpage}
\label{fockref-2}
	\begin{equation} 
		\label{dfock}
		{\rm e}^{i \Phi (f) } \mapsto {\rm e}^{i \Phi_E (f)} \Omega_E , \quad f \in H_{{\mathbb C}}^{-1} (S^2) \; ,
	\end{equation}
extends to a unitary operator from 
$L^2( {\mathcal Q}, \Sigma, {\rm d} \Phi_C ) $ to the Euclidean Fock space $\Gamma( H_{{\mathbb C}}^{-1}(S^2))$.
Set
	\begin{equation}
		\label{pfock}
		E_{\Sigma_{S^1}} \doteq \Gamma \left(e({S^1}) \right), 
		\qquad 
		E_{\Sigma_{\overline{S_\pm}}} \doteq \Gamma \left(e({\overline{S_\pm}}) \right) \; .
	\end{equation}
Use (\ref{dfock}) and denote the 
linear operators on $L^2( {\mathcal Q}, \Sigma, {\rm d} \Phi_C ) $ corresponding to (\ref{pfock})  by ${\mathcal E}_{\Sigma_{S^1}} $ and ${\mathcal E}_{\Sigma_{\overline{S_\pm}}}$, respectively. 

\paragraph{\it Notation}
For any closed and simply connected region $X \subset S^2$  denote by $\Sigma_{X}$ the smallest sub-$\sigma$-algebra of
$\Sigma$ for which the functions 
	$
		\bigl\{ \Phi(f) \mid f \in H^{-1}_{\upharpoonright X} (S^2) \bigr\}  
	$
are measurable. 

\goodbreak
\begin{theorem}[Dimock \cite{D}, Theorem 1, p.~247] 
\label{martheo}   
\quad 
\begin{itemize} 
\item[$ i.)$] 
	${\mathcal E}_{\Sigma_{S^1}}$ is the conditional expectation for $\Sigma' = \Sigma_{S^1}$. 
\item[$ ii.)$] 
	${\mathcal E}_{\Sigma_{\overline{S_\pm}}} $ is the conditional expectation for $\Sigma' = \Sigma_{\overline{S_\mp}}$; 
\item[$ iii.)$] 
	If $F \in L^2 ( {\mathcal Q}, \Sigma_{\overline { S_\pm}},  {\rm d} \Phi_C)$, then ${\mathcal E}_{\Sigma_{\overline{S_\mp}} }
	F = {\mathcal E}_{\Sigma_{S^1} } F $. 
\item[$ iv.)$] 
	${\mathcal E}_{\Sigma_{\overline{S_\mp}}}{\mathcal E}_{\Sigma_{\overline { S_\pm}} }
	= {\mathcal E}_{\Sigma_{S^1}}$, acting  on $L^2 ( {\mathcal Q}, \Sigma,  {\rm d} \Phi_C)$  {\em (Markov property).}
\end{itemize}
\end{theorem}
\begin{proof} Parts $ i.)$ and $ ii.)$ follow directly from the definitions.
Parts $ iii.)$ and $ iv.)$ follow from Lemma~\ref{dlemma}.
\end{proof}

\section{Sharp-time fields}
\label{sec:3.2n}

Let  $\mathbb{H}^{-1}_{\mathbb C} (S^2)$ be the complexification of $\mathbb{H}^{-1} (S^2)$.
The Sobolev space~$\mathbb{H}^{-1}_{\mathbb C} (S^2)$ contains the distribution 
	\begin{equation} 
	\label{eqDeltaTensorh}
	 (\delta \otimes h) ( \vec x) \doteq   r^{-1} \delta (\vartheta)  h (0, \varrho) \; , 
	 \qquad h (0, \, . \,)\in C^\infty(S^1)\; , \qquad
	  \vec x \equiv  \vec x (\vartheta, \varrho) \; , 
	\end{equation}
using geographical coordinates\footnote{The tensor product notation will be used for functions as well.} (see Section \ref{geo-chart}), 
which is supported on $S^1$. If ${\rm supp\,} h$ does not contain $(0, \pm r, 0)$,
then \eqref{eqDeltaTensorh} equals, as an element  in $\mathbb{H}^{-1}_{\mathbb C} (S^2)$,
	\begin{equation} 
	\label{eqDeltaTensorh'}
	 (\delta \otimes h) ( \vec x)=	\delta (\theta)  \;  \frac{ h  (0,\psi )}{r \cos \psi }\;  -
	\delta (\theta - \pi)  \;  \frac{ h(\pi,\psi)}{r \cos \psi }\; , \qquad 
			  \vec x \equiv  \vec x (\theta, \psi) \; , 
	\end{equation} 
in path-space coordinates. 
The sign on the right hand side is due to the different orientations of the one-forms 
${\rm d} \varrho$ and  ${\rm d} \psi$ on $I_-$.  

\begin{lemma}
\label{3.9}
Consider distributions of the form \eqref{eqDeltaTensorh}.  It follows that the 
time-zero covariance\index{time-zero covariance}
	\begin{align}
		\label{tratra} 
			C_0 (\overline{h_{1}}, h_{2}) 
				&\doteq C \bigl(\delta \otimes \overline{h_{1}} , \delta 
					\otimes h_{2} \bigr)  \nonumber \\  
				&=  \tfrac{1}{2}
				 \int_{S^1} r \, {\rm d} \varrho  \int_{S^1} r \, {\rm d} \varrho' \; \overline{h_{1}(\varrho) }  \; c_{\nu} \, 
				P_{-\frac{1}{2} + i \nu} (- \cos (\varrho - \varrho')) \, h_{2} (\varrho')  
	\end{align}
 exists as a positive quadratic form on $C^\infty (S^1)$. 
 Moreover, $C_0$ is invariant under  rotations around the  axis connecting the geographical poles. 
The constant $c_{\nu}$  appearing in \eqref{tratra} is given by
	\[
		c_{\nu} = -\frac{1}{2\sin  ( \pi (- \tfrac{1}{2} + i \nu))} = \frac{1} {2 \cos (i \nu  \pi )} \; . 
	\]   
and---as in \eqref{s1} and \eqref{s2}---
	\[ 		
		\nu =  
			\begin{cases}
				i \sqrt{\frac{1}{4} -\mu^2 r^2} & \text{if \ $ 0< \mu < \tfrac{1}{2r} $}   \; , \\
				 \sqrt{\mu^2 r^2 - \frac{1}{4} } & \text{if \ $ \mu \ge \tfrac{1}{2r} $}  \; .
			\end{cases}  
	\] 
\end{lemma}

\begin{proof}
Recall Proposition \ref{3.9a}.
This allows us to compute
	\begin{align*}
		 & \left\langle \delta \otimes \overline{h_{1}}, (- \Delta_{S^2} +\mu^2)^{-1}  \delta \otimes h_{2} 
		\right\rangle_{L^2( S^2, {\rm d} \Omega)}  
		 \\
		& \qquad = \int_{S^2} r {\rm d} \psi' {\rm d} \theta'  	\int_{S^2} r {\rm d} \psi {\rm d} \theta \; 
		\sin \theta'  
		\; \delta (\theta'- \tfrac{\pi}{2}) h_{1}(\psi')
		\\
		& \qquad \quad \qquad \quad \times  \sum_{l=0}^\infty \sum_{k=-l}^l 
		\frac{\overline{Y_{l,k} (\theta', \psi')} Y_{l,k} (\theta, \psi)}{l(l+1) + \mu^2r^2} 
		\sin \theta \; \delta (\theta - \tfrac{\pi}{2})  h_{2} (\psi)    
		 \\
		& \qquad = \frac{1} {4 \cos (i \nu  \pi )} \int_{S^1} r {\rm d} \psi'  \int_{S^1}r  {\rm d} \psi'  \; h_{1}(\psi') 
		P_{-\frac{1}{2} + i \nu}  ( - \cos (\psi-\psi') )  		h_{2} (\psi)  \; .  
	\end{align*}
%
%
We have used that for $\vec x= (0, \sin \psi, \cos \psi)$ and $\vec y= (0, \sin \psi', \cos \psi')$  
	\[
		\tfrac{\vec x \cdot \vec y}{r^2} = \cos ( \psi'-\psi) \; , 
	\]
as $\cos  \psi  \cos \psi' + \sin  \psi \sin \psi' = \cos ( \psi'-\psi)$. 
\end{proof}

An explicit formula for $C_0 (h_{1}, h_{2})$ (restricted to a half-circle) in {\em path-space coordinates} 
is provided next. Let, for $ h\in {\mathcal D} (I_+) $ and $0\leq\theta_0\leq2\pi$, 
	\[
		(\delta_{\theta_0}\otimes h)(\vec{x})\doteq \delta(\theta-\theta_0)
		\frac{h(\psi)}{r\cos\psi}\;, \qquad  \vec x \equiv \vec x (\theta, \psi) \; .
	\]
For $\theta=0$, $\delta_{\theta=0}\otimes h$ coincides with
$\delta\otimes h$ defined in Eq.~\eqref{eqDeltaTensorh},
cf.~Eq~\eqref{eqDeltaTensorh'}. 
\begin{lemma}
\label{coid}
For $h_{1}, h_{2}\in {\mathcal D} (I_+) $, $0\leq \theta_{1},
\theta_{2}< 2 \pi $, 
\label{stcpage}
	\begin{align}
		\label{coideq2}
			&C_{|\theta_1 -\theta_2|} (h_1, h_2) \doteq 
 C \bigl(\delta_{\theta_{1}}\otimes 
			h_{1} , \delta_{\theta_{2}}\otimes h_{2} \bigr)\\ 
			&\qquad =  r \, \Bigl\langle  \overline{ \cos_\psi  h_{1} },  \frac{{\rm e}^{-   |\theta_{2}
						-\theta_{1}|  \varepsilon   }
					+ {\rm e}^{- (2 \pi-|\theta_{2}-\theta_{1}|)  \varepsilon    }}{
						2 \varepsilon  (1-{\rm e}^{- 2 \pi   \varepsilon  })} 
					\cos_\psi  h_{2} \Bigr\rangle_{L^{2}(I_+ , \frac{r {\rm d} \psi} {\cos \psi})} \, , \nonumber 
	\end{align}
with $\varepsilon^2  =-  (\cos \psi \, \partial_\psi)^2 + \mu^2r^2 \cos^2 \psi $.  
\end{lemma}

\begin{proof} 
An approximation of the Dirac $\delta$-function is given by $\delta_{k}$, $k\in {\mathbb N}$, with
	\begin{equation}
		\label{deltaka}
		\delta_{k}(\theta) =(2 \pi )^{-1}\sum_{|\ell|\leq
                  k}{\rm e}^{i  \ell \theta } 
   \; \chi_{[0,2\pi)}(\theta) 
\; , 
		\quad \theta \in [0, 2 \pi )  \; ,
	\end{equation}
and $\chi_{[0,2\pi)}$  the characteristic function of the interval $ [0, 2 \pi ) \subset {\mathbb R}$. 
Use 
	\[
		 (- \Delta
		 	+\mu^2)^{-1}
		 	=   \left( - \partial^2_\theta + \varepsilon^2 \right)^{-1} r^2 \cos^{2} \psi   
	\] 
and
	\[
		 \sum_{\ell' \in {\mathbb Z}} \int_0^{2 \pi} \frac{{\rm d} \theta} {2 \pi }  \; {\rm e}^{i  \ell (\theta - \theta_1) } 
		 \left( - \partial^2_\theta + \varepsilon^2 \right)^{-1} {\rm e}^{i  \ell' (\theta - \theta_2)} 
		 =  \frac{{\rm e}^{i  \ell (\theta_{2}- \theta_{1})}}{
                   {\ell^2} + \varepsilon^2}  
	\]
to show that for $0\leq \theta_{1} , \theta_{2}< 2 \pi $ 
	\begin{align*}
			&   \lim_{k \to \infty} \lim_{k' \to \infty} C \bigl(\delta_{k'}(.-\theta_{1})
			 	\otimes h_{1}, \delta_{k}(.-\theta_{2})\otimes h_{2} \bigr) \nonumber \\
			&\quad  
			= \lim_{k \to \infty} \lim_{k' \to \infty} \Bigl\langle \overline{\delta_{k'}(.-\theta_{1})\otimes  h_{1}}, 
				\bigl(-\Delta_{S^2}+m^2\bigr)^{-1} \delta_{k}(.-\theta_{2})\otimes
				 h_{2} \Bigr\rangle_{L^2 (S^2,  {\rm d} \Omega) } \nonumber \\ 
			& \quad  = \lim_{k \to \infty} (2 \pi)^{-1}\sum_{|\ell|\leq k} 
				\Bigl\langle \overline{ \cos_\psi^{-1}\, h_{1}} , \frac{r {\rm e}^{i  \ell (\theta_{2}- \theta_{1})}}{  {\ell^2} 
			+ \varepsilon^2} \, 
					\cos_\psi  \, h_{2}   \Bigr\rangle_{L^{2}(I_+,  r \cos \psi \, {\rm d} \psi)}
				\nonumber \\ 
			& \quad  = \lim_{k \to \infty} (2 \pi)^{-1}\sum_{|\ell|\leq k} 
					\Bigl\langle  \overline{\cos_\psi  h_{1} } ,  
					\frac{ r {\rm e}^{i  \ell (\theta_{2}- \theta_{1})}}{  {\ell^2} 
						+ \varepsilon^2} \,  \cos_\psi h_{2} 
						\Bigr\rangle_{L^{2}(I_+,  \cos \psi^{-1} \,r  {\rm d} \psi)} \; .
		\end{align*}
$\varepsilon^2 $ is a differential operator, thus $\varepsilon^2$ acts locally and maps the subspaces 
\[
 	{\mathscr D}_\pm \doteq {\mathscr D} (\varepsilon^2) \cap L^2 \bigl(I_\pm,|\cos\psi|^{-1} r {\rm d} \psi \bigr)
\] 
into $L^2 \left(I_\pm,|\cos\psi|^{-1} r {\rm d} \psi \right)$, respectively. It therefore is consistent (see \eqref{vaepsdef}) to define 
	\begin{equation}
		\varepsilon (h_+ + h_-) 
			\doteq \sqrt{{\varepsilon^2}_{\upharpoonright 	 {\mathscr D}_+}} \; h_+   
			- \sqrt{{\varepsilon^2}_{\upharpoonright 	{\mathscr D}_-}} \; h_-\; , 
				\qquad h_\pm\in {\mathscr D}_\pm\;  . 
	\end{equation}
$\varepsilon = \varepsilon_{\upharpoonright I_+}+\varepsilon_{\upharpoonright I_-}$ 
is densely defined by \eqref{vaepsdef}, as ${\mathscr D}_+ \oplus {\mathscr D}_- 
= {\mathscr D}(\varepsilon^2)$. On the half\-circle~$I_+ $ the operator $\varepsilon$ has,  
just like $\varepsilon^2$, purely a.c.~spectrum on all of~${\mathbb R}^+$. 
Especially, $\varepsilon$ does not have a discrete eigenvalue at zero.
Thus, despite the fact that $\varepsilon$ has no mass gap,  one can apply the Poisson sum formula (see, \emph{e.g.},~\cite{KL2}) 
	\[
		\frac{1}{2 \pi  } \sum_{\ell\in {\mathbb Z}}\frac{{\rm e}^{i \ell  \theta}}{ \ell^{2}+
			\varepsilon^{2}}= \frac{{\rm e}^{-| \theta | |\varepsilon |}+
			{\rm e}^{-(2 \pi -| \theta |)|\varepsilon|}}{ 2  |\varepsilon|  (
			1-{\rm e}^{- 2 \pi  |\varepsilon|  })}  \: \: \hbox{ for } \: \: 0\leq |  \theta |< 2 \pi \;  
	\]
to conclude that \eqref{coideq2} holds.
\end{proof}

\goodbreak

Note that for $\theta=\pi$, we have 
	\begin{equation} 
		\label{eqDeltaPi}
		\delta_{\pi}\otimes  h = 	\delta \otimes {P_1}_* h \; , 
	\end{equation} 
where ${P_1}_*$ is the pull-back of the reflection at the
$x_0$-$x_1$ plane, \emph{i.e.},
	\[ 
		({P_1}_* h)(\psi) \doteq h(\pi-\psi) \; .
	\] 

\begin{remark} Consider distributions of the form \eqref{eqDeltaTensorh}. 
The closure of this set of distributions within the Sobolev space $\mathbb{H}^{-1}(S^2)$ can be 
naturally identified with the Hilbert space $H^{-1}_{S^1} (S^2, \mathbb{R} ) \cong \widehat{\mathfrak h} (S^1, \mathbb{R} )$
introduced in Definition \ref{zeit-null-Hilbert}, 
which itself is the completion of $C_{\mathbb R}^\infty (S^1)$ 
in the norm
	\[
		\| h \|^2_{\widehat{\mathfrak h}(S^1)} = 
		2 
		C \bigl(\delta \otimes \overline{h } , \delta 
					\otimes h  \bigr)  \;  . 
	\]
We denote the scalar product in $\widehat{\mathfrak h}(S^1)$ by
	\[
	\langle h_1, h_2 \rangle_{\widehat{\mathfrak h}(S^1)} \doteq  
	  \bigl\langle h_1, \tfrac{1}{2 \omega}  h_2 \bigr\rangle_{L^2( S^1, r {\rm d} \psi)} \; . 
	\]
\end{remark}

\goodbreak
\begin{corollary}[Time-zero Covariances] 
\label{wichtig} \quad 
\begin{itemize}
\item [$ i.)$] Let $h \in \widehat{\mathfrak h} (S^1)$. Then 
	\begin{align}
		\label{corwichtig}
		&
		\| h \|_{\widehat{\mathfrak h}(S^1)}  
		= r  \bigl\langle |\cos_\psi| \,  h , \big( \tfrac{\coth  \pi |\varepsilon|}
			{2  |\varepsilon|   }
   		+
		\tfrac{	{P_1}_*}{2 \varepsilon \sinh \pi \varepsilon  }\big)\,  
		|\cos_\psi|\,  h   \bigr\rangle_{L^{2}(S^1 , 
                          \tfrac{r {\rm d} \psi}{|\cos \psi|}) }.  
	\end{align}
\item [$ ii.)$] The pseudo-differential operator 
$\omega $ on $L^2(S^1, r {\rm d}\psi)$ 
satisfies the operator identity 
	\begin{equation}   
		 \label{eqMagicForm-a}
 		\omega  =  | \color{red}r\color{black}  \cos_\psi|^{-1}\,|\varepsilon|\,\bigl(\coth\pi\varepsilon
			-\tfrac{{P_1}_*}{\sinh\pi|\varepsilon|} \bigr)^{-1}   \; .  
  	\end{equation}
\end{itemize}
\end{corollary}

\begin{proof}
Write $h=h_++h_-$, where the support of $h_\pm$ is contained in
$I_\pm$, respectively. By Lemma~\ref{3.9} and Lemma~\ref{coid}, 
	\begin{align*}
		 \bigl\langle  h_+, \tfrac{1}{2\omega} h_+ \bigr\rangle_{L^2(S^1, r {\rm d} \psi)} 
		&= r \bigl\langle \cos_\psi  h_+   ,  \tfrac{\coth   \pi  |\varepsilon| }
			{2  |\varepsilon|   }  
			\cos_\psi h_+   \bigr\rangle_{L^{2}(S^1 ,
                          \frac{r {\rm d} \psi}{|\cos \psi|}) }.  
	\end{align*}
Now note that 
$\langle  h_-,\frac{1}{2\omega}h_-\rangle_{L^2(I_-, r {\rm d} \psi)} =  
\langle {P_1}_*  h_- ,\frac{1}{2\omega}{P_1}_*h_-\rangle_{L^2(I_+, r {\rm d}
\psi)} $ and again apply Lemma~\ref{3.9} and Lemma~\ref{coid} to  find a
similar expression. 
For the mixed term, use Eq.~\eqref{eqDeltaPi} 
to write 
	\begin{align*}
		 \bigl\langle h_+, \tfrac{1}{2\omega} h_- \bigr\rangle_{L^{2}(S^1 ,
                          \frac{r {\rm d} \psi}{|\cos \psi|}) } 
 		& =  C(\delta\otimes  \overline{h_+} ,\delta  \otimes h_-) 
		\nonumber \\
		&= 
				C \bigl( \delta_0\otimes  \overline{h_+ },\delta_\pi  \otimes {P_1}_*h_- \bigr) 
		= C_\pi( \overline{h_+},{P_1}_* h_-)  
		\nonumber \\
		&= 
				r \bigl\langle \cos_\psi h_+ ,  \tfrac{ {\rm e}^{- \pi  |\varepsilon|}}
					{ |\varepsilon| (1- {\rm e}^{-2\pi  |\varepsilon|  })}  
					\cos_\psi {P_1}_*h_-   \bigr\rangle_{L^{2}(I_+ ,
                         			 \frac{r {\rm d} \psi}{|\cos \psi|}) }  \nonumber \\
		&= 
				r \bigl\langle \cos_\psi  h_+ ,  \tfrac{1}{2  |\varepsilon| \sinh \pi 
 					|\varepsilon|}\cos_\psi {P_1}_*h_-   \bigr\rangle_{L^{2}(I_+ ,
                          			\frac{r {\rm d} \psi}{|\cos \psi|}) } \; .  
	\end{align*}
The term $\langle h_- , \frac{1}{2\omega} h_+\rangle$ yields a similar
expression, with $h_+$ and $h_-$ interchanged. 
We can now put the four terms together. The proof of Eq.~\eqref{corwichtig}
is completed by  noting that $\varepsilon$ leaves
the two subspaces $L^{2}(I_\pm,\frac{r{\rm d} \psi}{|\cos \psi|})$ invariant. The latter implies, by 
polarization, the  operator identity  
	\[
		\omega^{-1} =  |\varepsilon|^{-1}\, 
					\bigl( \coth\pi\varepsilon+\tfrac{{P_1}_*}
					{\sinh\pi|\varepsilon|} \bigr) \, |  r \, \cos_\psi |   \; ,   
	\]
which is equivalent to~\eqref{eqMagicForm}. 
\end{proof}

\goodbreak
The existence of Euclidean sharp-time fields now follows from~\eqref{e1.0} and 
Lemma~\ref{3.9}:

\begin{proposition}
\label{3.10}
The  {\em time-zero fields} 
	\begin{equation}
 		\Phi (0, h) =\lim_{n\to \infty}\Phi(\delta_{n}(\, .\,)\otimes h) \; , 
		\qquad h \in C_{{\mathbb R}}^\infty (S^1) \; ,
		\label{e1.4}
	\end{equation}
exist as elements of $L^{p}({\mathcal Q}, \Sigma, {\rm d} \Phi_{C})$, $1 \le p < \infty$.
\end{proposition}

\section{Foliation of the Euclidean field}
\label{sec:foliations}

Consider a foliation of the upper hemisphere in terms of half-circles
	\[
		R_1 (\theta) I_+ \; , \qquad 0< \theta < \pi \; .
	\]
The \emph{Euclidean sharp-time fields}\index{Euclidean sharp-time field}
\label{stfpage}
	\begin{equation}
		\label{iwasnix}
		\Phi (\theta, h) =\lim_{k\to \infty}\Phi( \delta_{k}(\, .\,- \theta) \otimes  h ) \; , 
		\qquad h \in 
		{ \textstyle {\mathcal D}_{\mathbb R} \left( I_+\right) } \; ,
	\end{equation}
belong to $\bigcap_{1\leq p<\infty} L^{p}({\mathcal Q}, \Sigma, {\rm d} \Phi_{C})$. 
Use (\ref{coideq})   to show that the map
	\begin{equation}
		\label{e1.6bbb}
		\begin{matrix} 
			& S^{1}\times C^\infty_{{\mathbb R}} \left(I_+ \right) &
					\to & \bigcap_{1\leq p<\infty}L^{p}({\mathcal Q}, \Sigma, {\rm d}\Phi_{C})  \\
			& (\theta,  h) & \mapsto & \Phi ( \theta, h) \\ 
		\end{matrix}
	\end{equation}
is continuous. 

\begin{lemma}
\label{1.0b}
The following identity holds on $\bigcap_{1\leq p<\infty}L^{p}({\mathcal Q}, \Sigma, {\rm d}\Phi_{C})$:
	\begin{equation}
	\label{fieldfoliation}
		\int_{S^{1}}  r \cos_\psi \, {\rm d}  \theta  \,   \Phi  (\theta, f_{\theta})=\Phi(f)\; , 
		\quad f_{\theta} \equiv f (\theta, \, . \, ) \in C^\infty_{{\mathbb R}} \left( I_+\right)  \; , \; \; f\in C^\infty_{{\mathbb R}} (S^2) \; . 
	\end{equation}
\end{lemma}

\begin{proof}
Let $f\in C^\infty_{{\mathbb R}} (S^2)$ and consider the approximation of the 
Dirac $\delta$-function given in~\eqref{deltaka}.  For $k\in {\mathbb N}$ fixed, the map
	\[
		\begin{matrix} 
			&  S^1 & \to &  \mathbb{H}^{-1}(S^2) \\ 
			& \theta & \mapsto &  \delta_{k}(\, .-\theta ) \otimes    f_{\theta} 
		\end{matrix}
	\]
is continuous. Since $f\in C^\infty_{{\mathbb R}} (S^2)$, the expression 
$\|  \delta_{k}(.-\theta ) \otimes   f_{\theta} \|_{\mathbb{H}^{-1}(S^2)}$ is bounded.
Hence by (\ref{e1.6bbb}) the map
	\begin{align*}
		  S^1 & \to  \bigcap_{1\leq p<\infty}L^{p}({\mathcal Q}, \Sigma, {\rm d}\Phi_{C})  \nonumber \\  
		 \theta & \mapsto  \Phi \bigl( \delta_{k}(\, .-\theta ) \otimes   f_{\theta} \bigr) 
	\end{align*}
is continuous and $\bigl\| \Phi \bigl( \delta_{k}(\, .-\theta ) \otimes   f_{\theta}\bigr) \bigr\|_{L^{p}({\mathcal Q}, 
\Sigma, {\rm d}\Phi_{C})}$ is bounded.  Therefore 
	\[
		 \int_{0}^{2\pi} r \cos \psi  \, {\rm d} \theta \; \Phi \bigl( \delta_{k}(\, .-\theta ) \otimes   f_{\theta} \bigr)  
	\]
is well defined as an element of $\bigcap_{1\leq p<\infty}L^{p}({\mathcal Q}, \Sigma, {\rm d}\Phi_{C})$. Moreover,  
	\begin{align*}
 		\int_{0}^{2\pi} r \cos \psi  \,  {\rm d} \theta  \;  \Phi \bigl( \delta_{k}(\, .-\theta ) \otimes   f_{\theta}\bigr)
		&= \Phi \Bigl(\int_{0}^{2\pi} r \cos \psi  \,   {\rm d}  \theta \;  (\delta_{k}(\, .-\theta ) \otimes    f_{\theta} ) \Bigr)
		= \Phi( \delta_{k}* f) \; ,
	\end{align*}
where the convolution product $*$ acts only in the variable $\theta$. Use  \eqref{eqDeltaTensorh'}  and 
	\[
		 \lim_{k\to \infty} \delta_{k} *   f 
		 = \lim_{k\to \infty} \int_{0}^{2\pi}  r {\rm d} \theta \; r^{-1} \delta_{k} (\, .\, - \theta)    f(\theta , \, .\, ) 
		 =  f 
	\]
in~$\mathbb{H}^{-1}(S^2)$ to obtain from (\ref{e1.6bb}) 
	\[
		\lim_{k\to \infty} \int_{0}^{2\pi} r \cos \psi \,  {\rm d}  \theta  \; \Phi \bigl( \delta_{k}(\, .-\theta ) \otimes    
		 f_{\theta} \bigr)\, 
		=
		\Phi(f) \quad \hbox{in} \:  \bigcap_{1\leq p<\infty}L^{p}({\mathcal Q}, \Sigma, {\rm d}\Phi_{C}) \; .
	\] 
Furthermore, it follows from (\ref{e1.6bb}) that 
	\[
		 \lim_{ k\to \infty}\sup_{\theta \in S^1} 
		 		\bigl\|\Phi \bigl( \delta_{k}(\, .-\theta ) \otimes    f_{\theta} \bigr)
					- \Phi (\theta, f_{\theta}) \bigr\|_{L^{p}({\mathcal Q}, \Sigma, {\rm d}\Phi_{C})} 
		= 0 
	\]
for $f\in C^\infty(S^2)$. Hence \eqref{fieldfoliation} follows. 
\end{proof}
\goodbreak

\chapter{Non-Gaussian Measures}

\section{Wick ordering of random variables}
\label{sec2.2}
Recall the normal ordering of
Gaussian random variables: Let~$({\mathcal Q}, \mu)$ be a probability space and  $X$ 
a real vector space equipped with a positive quadratic form~$f\mapsto
c(f,f)$, \emph{i.e.},  a {\em covariance}. Let $f\mapsto \Phi(f)$ be a ${\mathbb R}$-linear map from~$X$
into the space of real measurable functions on ${\mathcal Q}$.
{\em Normal ordering} ${:} \Phi(f)^{n}{:}_c$ with respect to a
covariance $c$ is defined by 
\label{wickpage}
	\begin{equation}
		\label{wick}
			{:} \Phi(f)^{n}{:}_{c} 
				= \sum_{m=0}^{[n/2]}\frac{n!}{m!(n-2m)!}
					\Phi(f)^{n-2m} \Bigl(-\tfrac{1}{2}  c(f,f) \Bigr)^{m} \; ,
	\end{equation}
where $[.]$ denotes the integer part. 

Let, for~$m \in {\mathbb N}$,
	\[  
	 \delta^{(2)}_{m}(\, .-\theta, \, .- \psi)
	 =	\frac{1}{r^2}
	 \sum_{\ell = 0}^m \sum_{k= - \ell}^{\ell} \overline{Y_{\ell, k} (\theta, \psi)} 
	 Y_{\ell, k} (\, . \, , \, . \, ) \;   , 
	 \qquad \vec x \equiv \vec x (\theta, \psi) \in S^2 \; , 
	\] 
if $\vec x \ne (0, \pm r , 0)$. It follows from  \eqref{F-S2} that 
$ \bigl\{ \delta^{(2)}_m \bigr\}_{m \in {\mathbb N}}$ approximates
the  two-dimensional  
Dirac $\delta$-function $\delta_m^{(2)}$ for $m \to \infty$ unless  $\vec x = (0, \pm r, 0)$.
Note that $\delta^{(2)}$ is  supported at the point 
$(0,0,r) \in S^2$  and $\int_{S^2} {\rm d} \Omega \, \delta^{(2)}  = 1$.

\begin{theorem}[Ultraviolet renormalization]
\label{uvtheo}
For $n \in {\mathbb N}$ and $f\in L^{2}(S^2, {\rm d} \Omega)$
the following  limit exists  in $\kern -.2cm \bigcap \limits_{1\leq p<\infty} \kern -.2cm L^{p}({\mathcal Q}, \Sigma, {\rm d}\Phi_{C})$:
	\[
		\lim_{m \to \infty}\int_0^{2 \pi} r \, {\rm d} \theta \int_{-\pi/2}^{\pi/2} r \cos \psi \, {\rm d} \psi \;  f(\theta,\psi) \; 
		{:}\Phi \bigl(  \delta^{(2)}_{m}(\, .-\theta, \, .- \psi)  \bigr)^n {:}_{C} \; .
	\]
It is denoted by $\int_{S^2} {\rm d} \Omega \; f(\theta,\psi)  \; {:} \Phi( \theta,\psi)^{n}{:}_{C}$. 
\end{theorem}

\begin{proof} 
Use the identification of $L^{2}(Q, \Sigma, \mu)$ with
$\Gamma(\mathbb{H}^{-1} (S^2))$ given in \eqref{dfock}. Then Wick ordering
with respect to $C$ coincides with Wick ordering with
respect to the Fock vacuum on $\Gamma(\mathbb{H}^{-1} (S^2))$ and 
	\[
		{:} \Phi(g)^{n} {:}_{C} = \frac{1}{2^{n/2}} 
							\sum_{j=0}^n \begin{pmatrix} 
								n \\ 
								j
							\end{pmatrix} 
						a^*(g)^j a(g)^{n-j} \; , \qquad g \in  \mathbb{H}^{-1} (S^2) \; . 
	\]
In particular,
	\begin{align}
		{:} \Phi(\delta^{(2)}_{m})^{n} {:}_{C} & = \frac{1}{ 2^{\frac{n}{2}} r^{2n} } 
							\sum_{j=0}^n \begin{pmatrix} 
								n \\ 
								j
							\end{pmatrix} 
	 \sum_{\ell_1 = 0}^m \sum_{k_1= - \ell_1}^{\ell_1} 
	 \cdots
	 \sum_{\ell_n = 0}^m \sum_{k_n= - \ell_n}^{\ell_n} \nonumber \\
	 & \qquad  \overline{Y_{\ell_1, k_1} (\theta, \psi) } \cdots \overline{Y_{\ell_j, k_j} (\theta, \psi)} 
	Y_{\ell_{j+1}, k_{j+1}} (\theta, \psi) \cdots Y_{\ell_{n}, k_{n}} (\theta, \psi) \nonumber \\
	& \qquad 	\quad \times	a^*(Y_{\ell_1, k_1}) \ldots a^*(Y_{\ell_j, k_j}) 
						a(Y_{\ell_{j+1}, k_{j+1}} ) \ldots a(Y_{\ell_{n}, k_{n}}) \; . 
	\label{sum-m-1}
	\end{align}
Using
	$
		\overline{ Y_{\ell, k} (\theta, \psi)} = (-1)^k Y_{\ell, - k} (\theta, \psi) 	$,
we find that (see, \emph{e.g.}, \cite{SHK} or \cite[Sect.~6]{DG} for a recent survey) 
	\[
		P^{(n)}_m (f) \doteq \int_0^{2 \pi} r \, {\rm d} \theta \int_{-\pi/2}^{\pi/2} r \cos \psi \, {\rm d} \psi \;  f(\theta,\psi) \; 
		{:}\Phi \bigl(  \delta^{(2)}_{m}(\, .-\theta, \, .- \psi)   \bigr)^n {:}_{C}  
	\]
is a linear combination of Wick monomials of the form 
	\begin{align}
		& 		\sum_{j=0}^{n}  \begin{pmatrix} 
								n \\ 
								j
							\end{pmatrix}  \sum_{\ell_1 = 0}^m \sum_{k_1= - \ell_1}^{\ell_1} 
	 \cdots
	 \sum_{\ell_n = 0}^m \sum_{k_n= - \ell_n}^{\ell_n}\;  (-1)^{\sum_{i=1}^j k_i} 
		\nonumber \\
		& \qquad  \qquad    \qquad  \qquad
		\qquad  \times		 w^{(n)} (\ell_{1}, k_1, \ldots, \ell_{j}, k_{j}, 
		\ell_{j+1}, k_{j+1}, \ldots, \ell_{n}, k_{n}) 
		\nonumber \\
		& \qquad  \qquad  \qquad \qquad  \qquad  \qquad
		\qquad   \times		a^{*}_{\ell_1, k_1} \cdots a^{*}_{\ell_j, k_j}
				a_{\ell_{j+1}, - k_{j+1}} \cdots a_{\ell_{n}, - k_{n}} \; , 
	\label{sum-m-2a}
	\end{align}
where $a^{(*)}_{\ell_i, k_i} \equiv a^{(*)} (Y_{\ell_i, k_i})$ and 
	\begin{align*}
		w^{(n)} (\ell_{1}, & k_1, \ldots,  \ell_{j}, k_{j}, 
		\ell_{j+1}, k_{j+1}, \ldots, \ell_{n}, k_{n}) \\
		& = \frac{1}{ 2^{\frac{n}{2}} r^{2n} } 
		\int_0^{2 \pi} r \, {\rm d} \theta \int_{-\pi/2}^{\pi/2} r \cos \psi \, {\rm d} \psi \;  f(\theta,\psi) \; \prod_{i=1}^n
			 \overline{Y_{\ell_i, k_i} (\theta, \psi)}   \; .
	\end{align*}
Next we apply the Wick monomials \eqref{sum-m-2a} 
to the Fock vacuum $\Omega_E$ in the Fock space over the one-particle space ${\mathbb H}^{-1} (S^2)$; 
see Definition \ref{sobolev-S2}. Only the term with $j=n$ contributes in the sum over $j$. 
However, when removing the cut-off (\emph{i.e.}, for $m \to \infty$) there are still $n$ infinite double sums. 
The ultraviolet divergencies, stemming from high (angular) momenta, are 
now visible from the sum  
	\[ 
		\sum_{k_i= - \ell_i}^{ \ell_i } \left( \ell_i + \tfrac{1}{2} \right)^{-2}
	\]
over the Fourier coefficients, which goes like $l_i^{-1}$ for large $l_i \in \mathbb{N}$.
However, we can use the bound \cite[Lemma 6.1]{DG}
	\begin{equation}
	\label{sum-n}
		\prod_{i=1}^{n } l_i^{-1} \le \sum_{p=1}^n \left( \prod_{i\ne p} l_i^{-\frac{n}{n-1}} \right)
	\end{equation} 
which follows from $\left( \prod_{p=1}^{n } \lambda_p \right)^{1/n}  \le \sum_{p=1}^n \lambda_p $, 
applied to $\lambda_p = \prod_{i\ne p} l_i^{-\frac{n}{n-1}}$. 
The remaining sum over the $l_p$'s can be estimated using the
Plancherel theorem for the expansion into spherical harmonics which ensure that 
	\begin{equation}
	\label{first-sum}
		\sqrt{ \sum_{l_p=0}^{\infty} \sum_{m_p=- l_p}^{l_p}  | \widetilde f_{l_p, m_p}|^2 } = \| f \|_{L^2(S^2 , {\rm d} \Omega)} \; . 
	\end{equation}
We conclude that $P^{(n)}_m (f) \Omega$ 
converges to a vector $P^{(n)}_\infty (f) \Omega$ in $\Gamma^{(n)} (\mathbb{H}^{-1} (S^2))$,
or equivalently that $P^{(n)}_m (f)$ converges to $P^{(n)}_\infty (f)$ in $L^{2}( Q, \Sigma, \mu)$. 
Since $P^{(n)}_m (f)\Omega$ is a finite particle vector, it follows from a standard argument (see, \emph{e.g.},
\cite[Theorem~1.22]{S} or \cite[Lemma 5.12]{DG}) that 
	\[
		P^{(n)}_m (f) \to P^{(n)}_\infty (f)\in L^{p}(Q, \Sigma_0, \mu)
	\] 
for all $1\leq p<\infty$.
\end{proof}

If ${\mathscr P}={\mathscr P}(\lambda)$ is a real valued polynomial, then 
\label{interactionspherepage}
	\[
		\int_{S^2}   {\rm d} \Omega \; f(\theta,\psi) \; {:} {\mathscr P}(\Phi( \vec x )) {:}_{C}   \; , 
		\qquad \vec x  \equiv  \vec x (\theta,\psi) \in S^2 \; , 
	\]
is well defined, by linearity, for $f \in L^2 (S^2, {\rm d} \Omega)$.
On subsets $K \subset S^2$ with non-empty interiors the interaction is defined by
	\[
		V  (K)=  \int_{K} {\rm d} \Omega \;  {:} {\mathscr P} (\Phi( \vec x )) { :}_{C} \; ,
		\qquad K \subset S^2 \; .
	\]
$V  (K)$ is a densely defined operator, but it is unbounded from below. 

If, in addition, ${\mathscr P}$ is bounded from below, then the function ${\rm e}^{- tV (K)}$ is in $L^1$ for any 
$t>0$ \cite[Lemma 3.15]{SHK}; see also
\cite{G}\cite{Se}. In particular, one has 
	\begin{equation}
	 	\label{intL1}
		{\rm e}^{- V (S^2)}  \in L^{1}({\mathcal Q}, \Sigma, {\rm d}\Phi_{C}) \; . 
	 \end{equation}
Hence one can  define the {\em perturbed
measure ${\rm d}\mu_{V}$ on the sphere}:
	\begin{equation}	 
		\label{pmeasure}
		{\rm d}\mu_{V}= \frac{{\rm e}^{- V (S^2) }{\rm d}\Phi_{C}}{\int_{{\mathcal Q}}{\rm d}\Phi_{C} \; {\rm e}^{- V (S^2) }} \; . 
	\end{equation} 

\section{Sharp-time interactions}
\label{sec:3.5}

The results of this section will be used  to define Feynman-Kac-Nelson kernels  in Section \ref{FKN}. 

\begin{lemma}
\label{wickooo} 
The following limit exists  
in $\bigcap_{1\leq p<\infty}L^{p}({\mathcal Q}, \Sigma, {\rm d}\Phi_{C})$:
\[
 \lim_{k\to \infty}\int_{S^1} r {\rm d} \psi \;  h (\psi)\,  {:}\Phi(0, \underline{\delta}_{k}(.-\psi))^{n} {:}_{C_{0}} \; , 
 \qquad h \in L^2 ( S^1, {\rm d} \psi) \; .  
 \]
It is denoted by $\int_{S^1}  r {\rm d} \psi \;  h (\psi)\,  {:}\Phi(0, \psi)^{n} {:}_{C_{0}} $.
\end{lemma}

\begin{proof} 
See the proof of Theorem \ref{uvtheo}. 
\end{proof}

Thus for ${\mathscr P}$ a real valued polynomial,  the expression
	\label{vcospage}
	\begin{equation}
		\label{bbv-interaction}
		{V_0} (h) = \int_{S^1}  r {\rm d} \psi \;  h(\psi) \, {:} {\mathscr P}(\Phi(0,\psi)) {:}_{C_{0}}\; ,
		 \qquad h \in L^2 (S^1, r {\rm d} \psi)\; ,
	\end{equation}
is well defined, by linearity.  
The  interaction\footnote{The interaction $V^{(\alpha)} $ 
is defined by rotating $V^{(0)}$ around the $x_0$-axis by an angle $\alpha$, 
see Equ.~\eqref{fkn17}, Section \ref{FKN}.}
	\begin{equation}
		\label{v0newdefinition}
		V^{(0)}   \doteq V_0 ( \cos_\psi \chi_{ I_+  }) \; , 
	\end{equation}
with $\chi_{ I_+ }$  the characteristic function of the interval
$I_+ \subset S^1$,  
has two  equivalent meanings: first of all (see Lemma~\ref{wickooo}), 
it can be viewed as a $\Sigma^{(0)}$-measurable function 
	\[
		V^{(0)}  \in \bigcap_{1\leq p<\infty}L^{p}({\mathcal Q}, \Sigma^{(0)}, {\rm d}\Phi_{C}) \; . 
	\]
Here $\Sigma^{(0)}$ is the smallest sub $\sigma$-algebra of $\Sigma$ for which the functions 
$\{ \Phi(0,h) \mid h\in {\mathcal D}_{{\mathbb R}} ( I_+ ) \} $ are measureable. 

Secondly,~$V^{(0)} $~can be considered, applying \eqref{hos}, 
as a self-adjoint operator affiliated to the abelian von
Neumann algebra ${\mathcal U} (S^1) $ acting on the Hilbert space~${\mathcal H}$.  
This second possibility will be discussed in Section \ref{interactingdesitter}.

\section{Foliations of the interaction}
\label{sec1.3}
The $(2 \pi)$-periodic one-parameter group 
$ [0, 2 \pi)   \ni \theta  \mapsto  R^{(\alpha)} (\theta)_* $ induces a representation 
	\begin{equation}
		\label{wawa}
		\theta \mapsto {\rm U} ^{(\alpha)} (\theta) 
	\end{equation}
of $U(1)$ in terms of  automorphisms of $L^{\infty}({\mathcal Q}, \Sigma, {\rm d} \Phi_C)$, 
which extends to a strongly
continuous representation in terms of isometries of $L^{p}({\mathcal Q}, \Sigma, {\rm d} \Phi_C)$. 

\begin{lemma} 
\label{2.2}
For $h\in L^{2} \bigl( I_+ \, , r {\rm d} \psi \bigr)$ and $g \in C^\infty_{{\mathbb R}} (S^{1})$
	\begin{equation}  
		\int_{S^{1}} r \, {\rm d} \theta \; g(\theta) ({\rm U}^{(0)} (\theta) {V_0} )(\cos_\psi h)  
		= \int_{S^2} {\rm d} \Omega \;   g(\theta)  h(\psi) \; {:}{\mathscr P}(\Phi(\theta,\psi)) {:}_{C}
		\label{e1.11}
	\end{equation}
as functions on $ {\mathcal Q}$. 
\end{lemma}

\begin{proof} Consider path-space coordinates and follow an argument given in~\cite{GeJII}: 
let $F \in L^{p}({\mathcal Q}, \Sigma, {\rm d} \Phi_{C})$ for some
$1\leq p<\infty$ and $g \in C^\infty_{{\mathbb R}} (S^{1})$. Then
	\[ 
		\int_{S^{1}} r \, {\rm d} \theta \; g(\theta){\rm U}^{(0)} (\theta)F 
	\]
belongs to $L^{p}({\mathcal Q},\Sigma, {\rm d}\Phi_{C})$.
Together with Lemma~\ref{wickooo} this implies that the functions given in~(\ref{e1.11}) 
are in $L^{p}({\mathcal Q}, \Sigma, {\rm d}\Phi_{C})$. Next prove that they are identical: by linearity, 
one may assume that ${\mathscr P}(\lambda)= \lambda^{n}$.
Use Lemma \ref{wickooo} and the identity~\eqref{wick} to derive
	\[
		\int_{S^2} {\rm d} \Omega \; g(\theta)\,h(\psi)  \, {:} {\mathscr P}(\Phi(\theta,\psi)) {:}_{C} \; 
		=\lim_{(k, k')\to \infty} F(k, k')\hbox{ in }L^{p}({\mathcal Q}, \Sigma, {\rm d}\Phi_{C}) \; ,
	\]
where
	\begin{align*}
		F(k, k') &= 
		\sum_{m=0}^{[n/2]} \frac{\bigl( -\frac{1}{2} C  (\delta^{(2)}_{k, k'}, \delta^{(2)}_{k, k'}) 	
		\bigr)^{m} n!  }{m!(n-2m)!} 
		\\
		& \qquad	\qquad \times	
		\int_{S^2} 	{\rm d} \Omega \,    g(\theta)  h(\psi) \,
		\Phi \bigl(\delta_{k}( \, .\, -\theta)\otimes  \delta_{k'}( \, .\, -\psi) \bigr)^{n-2m} 
	\end{align*}
and $\delta^{(2)}_{k, k'} = \delta_{k} \otimes \delta_{k'} $ provides an
approximation of the Dirac $\delta$-function $\delta^{(2)}$ on~$S^2$. 
According to Proposition \ref{3.9}
	\[
		\lim_{k\to \infty}C \bigl(\delta^{(2)}_{k,k'}, \delta^{(2)}_{k, k'} \bigr)
		= C_{0}(\delta_{k'}, \delta_{k'}) \; . 
	\]
The definition of sharp-time fields in (\ref{e1.4})  implies that 
	\[
 		\lim_{k\to \infty}F(k, k')=  \int_{S^{1}} r \, {\rm d}
		\theta \; g(\theta)V_{k'}(\theta, \cos_\psi h) \: \hbox{ in } \: L^{p}({\mathcal Q}, \Sigma, {\rm d}\Phi_{C}) \; , 
	\]
where
	\begin{align*}
		V_{k'} (\theta, \cos_\psi h) &= \sum_{m=0}^{[n/2]} \frac{n!}{m!(n-2m)!} 
		\Bigl(-\frac{1}{2} C_{0}(\delta_{k'}, \delta_{k'}) \Bigr)^{m} \nonumber \\
		& \qquad \qquad \times 
		\int_{- \pi/2}^{\pi /2} r \cos \psi \,  {\rm d} \psi \;  h(\psi)  \Phi \bigl(\theta, \delta_{k'}( \, .\, -\psi) \bigr)^{n - 2m}\; .    
	\end{align*}
Note that  $V_{k'}(\theta, \cos_\psi h)={\rm U}^{(0)}  (\theta)V_{k'}(0, \cos_\psi h)$.
By Lemma \ref{wickooo}    
	\[
		\lim_{k'\to \infty}V_{k'}(0, \cos_\psi h)
		= \int_{- \pi/2}^{\pi /2}  r \cos \psi \, {\rm d} \psi  \;  h(\psi)\, {:}{\mathscr P}(\Phi(0, \psi)){:}_{C_{0}}    
	\]
in $L^{p}({\mathcal Q}, \Sigma, {\rm d} \Phi_{C})$ and hence 
	\[
	\lim_{k' \to \infty}\int_{S^{1}} r {\rm d} \theta \, g(\theta)V_{k'}(\theta, \cos_\psi h)
	= \int_{S^{1}} r {\rm d} \theta \; g(\theta)({\rm U}^{(0)} (\theta) {V_0}) ( \cos _\psi  h)  
	\]
in $L^{p}({\mathcal Q}, \Sigma,  {\rm d} \Phi_{C}) $. Apply  Lemma \ref{a1} below 
with $E= L^{p}({\mathcal Q}, \Sigma, {\rm d} \Phi_{C})$ to obtain the identity in (\ref{e1.11}).  
\end{proof}

For completeness we recall a simple technical lemma from \cite{GeJII}.
\begin{lemma}
\label{a1}
Let $F \colon {\mathbb R}^{2} \to E$ be a map with values in
a metric space $E$. 
\begin{itemize}
\item[$ i.)$] Assume that
\begin{align*}
	\lim_{k,k'\to \infty}& F(k, k')= F_{\infty} \text{ exists} \, ,\\
	\lim_{k'\to \infty}\; & F(k, k')= G(k) \text{ exists }\forall k\in {\mathbb N} \, ,\\
	\lim_{k\to \infty}\; & G(k)=G_{\infty} \text{ exists}\, .
\end{align*}
Then $F_{\infty}= G_{\infty}$.
\item[$ ii.)$] Assume that
\begin{align*}
\lim_{k'\to \infty} & F(k, k')= G(k)\hbox{ exists} \, ,\\
\lim_{k\to \infty} & G(k)= G_{\infty}\hbox{ exists} \, , \\
\lim_{k\to \infty} & F(k, k')= H(k')\hbox{ exists and the convergence} \\
& \qquad \qquad \qquad \quad \; \text{is uniform w.r.t. }k' \, .
\end{align*}
Then $\lim_{k'\to \infty}H(k')=G_{\infty}$.
\end{itemize}
\end{lemma}

\part{The Osterwalder-Schrader Reconstruction}

\chapter{The Reconstruction of Free Quantum Fields}

The  {\em Markov property} implies that the time-zero quantum fields acting on the vacuum vector 
generate the physical Hilbert space. In case it holds, Nelson's reconstruction theorem 
(see \cite{N1}\cite{N2}\cite{N3}\cite{N4}) applies, and the more sophisticated 
reconstruction theorem of Osterwalder and Schrader \cite{OS1}\cite{OS2} is not necessary. 

\section{Reflection positivity}
\label{sec:4.3n}

The time reflection diffeomorphism $T$ (see \eqref{deftimerefl})
induces a map $T_*$ on $C^\infty (S^2)$:
\label{hospage}
	\[
 		T_* h = h \circ T^{-1} , \qquad h \in C^\infty (S^2) \; . 
	\]
$T_*$ extends to a unitary operator on $H_{{\mathbb C}}^{-1} (S^2)$.
Use (\ref{dfock}) to define a {\em unitary} operator~$\Theta$ on $L^2({\mathcal Q}, \Sigma, {\rm d} \Phi_C)$
corresponding to $\Gamma (T_*)$.  For $f_1, \ldots ,  f_n \in H^{-1}(S^2)$,
	\[
		\Theta \left( \Phi (f_1) \cdots \Phi (f_n) \right) 
		= \Phi (T_* f_1) \cdots \Phi (T_* f_n) \; . 
	\]
$\Theta$ induces a measure preserving automorphism\footnote{Consider the 
Hilbert space $L^2({\mathcal Q}, \Sigma, {\rm d} \Phi_C)$ and let ${\mathcal M}$ be the von Neumann algebra 
of all multiplications by bounded measurable functions, acting on $L^2({\mathcal Q}, \Sigma, {\rm d} \Phi_C)$. 
The Lebesgue measure lifts to a countably additive probability measure $\mu$ on the projections of ${\mathcal M}$. 
A $*$-automorphism $\alpha$ of a von Neumann algebra ${\mathcal M}$ is said to preserve the measure $\mu$  
if $\mu \circ \alpha = \mu$	on the projections of ${\mathcal M}$. See, \emph{e.g.},~\cite{Arv}.} 
of $L^{\infty}({\mathcal Q}, \Sigma, {\rm d} \Phi_C)$, which extends to an isometry of 
$L^{p}({\mathcal Q}, \Sigma, {\rm d} \Phi_C)$ for $1\leq p<\infty$.  

\begin{theorem}[Dimock \cite{D}, Theorem 2, p.~248]
\label{reftheo}
Let $F \in L^2 ({\mathcal Q}, \Sigma_{\overline { S_+}}, {\rm d} \Phi_C)$. 
Then  
\begin{itemize}
\item[$ i.)$] $\Theta (F) \in L^2 ({\mathcal Q}, \Sigma_{\overline { S_-}}, {\rm d} \Phi_C)\, $; 
\item[$ ii.)$] $\int  {\rm d} \Phi_C \; \overline{\Theta (F)} F  \ge 0 $ {\rm (reflection positivity)}; 
\item[$ iii.)$] $\Theta {\mathcal E}_{\Sigma_{S^1}}= {\mathcal E}_{\Sigma_{S^1}}\Theta$.  
\end{itemize}
\end{theorem}

\begin{proof}
Properties $ i.)$ and $ iii.)$ follow directly from the definitions;  $ ii.)$ is a direct consequence (see \cite{D}\cite{KL1}) 
of the Markov property  (Theorem~\ref{martheo} $ iv.)$).
\end{proof}

We note that an alternative proof of reflection positivity for Riemannian manifolds with a suitable symmetry 
(the sphere being in the class considered) was given in \cite{GRS}.

\section{The reconstruction of the Hilbert space}
\label{sec:4.4}
Define ${\mathcal N}\subset L^2 ({\mathcal Q}, \Sigma_{\overline { S_+}},  {\rm d} \Phi_C)$ as the kernel of
the positive quadratic form
	\begin{equation}
		\label{ossp}
		\langle F , G \rangle_{\rm os} =\int_{{\mathcal Q}} {\rm d} \Phi_C \;  \overline{ \Theta (F) }G , \qquad F,G \in L^{2}({\mathcal Q}, \Sigma_{\overline{ S^+}}, {\rm d} \Phi_C) \; .
	\end{equation} 
In other words, 
	\[
	 {\mathcal N}= \left\{ F \in L^2 ({\mathcal Q}, \Sigma_{\overline { S_+}},  {\rm d} \Phi_C) \mid \langle G, F \rangle_{\rm os} = 0 
	 \quad \forall G \in L^2 ({\mathcal Q}, \Sigma_{\overline { S_+}},  {\rm d} \Phi_C)
	\right\} \; . 
	\]
Complete the quotient space with respect to the positive definite scalar
product~\eqref{ossp}: 
		\label{hos-page}
	\begin{equation}
					\label{hos}
		{\mathcal H}=\hbox{ completion of } L^{2}({\mathcal Q}, \Sigma_{\overline { S_+}}, 
		{\rm d} \Phi_C) /{\mathcal N} \; .
	\end{equation}
Let ${\mathcal V}$ be the canonical projection 
	\[
		{\mathcal V}\colon L^{2}({\mathcal Q}, \Sigma_{\overline { S_+}}, {\rm d} \Phi_C) 
		\to L^{2}({\mathcal Q}, \Sigma_{\overline { S_+}}, {\rm d} \Phi_C)/{\mathcal N} \; .  
	\]
There is a {\em distinguished unit vector}
\label{vacumvecpage}
	\begin{equation}
		\label{dv} 
		\Omega ={\mathcal V}(1) \in {\mathcal H} \; . 
	\end{equation}
$1\in L^2 ({\mathcal Q}, \Sigma_{\overline { S_+}}, {\rm d} \Phi_C)$ is the constant 
function equal to $1$ on ${\mathcal Q}$.

If $A\in L^{\infty}({\mathcal Q}, \Sigma_{S^1}, {\rm d} \Phi_C)$, 
multiplication by $A$ preserves ${\mathcal N}$, since 
$A$ is by assumption $\Sigma_{S^1}$ measurable. Define a  bounded operator
$A^{\rm os}\in {\mathcal B}({\mathcal H})$ by
	\begin{equation}
		A^{\rm os}{\mathcal V} (F) = {\mathcal V}(AF ) \; , 
		\qquad F \in L^2 ({\mathcal Q}, \Sigma_{\overline { S_+}}, {\rm d} \Phi_C) \; , 
		\label{identi}	
	\end{equation}
and denote by ${\mathcal U} (S^1) \subset {\mathcal B}({\mathcal H})$ the 
abelian von Neumann algebra 
\label{abelianalgebraospage}
	\[
		{\mathcal U} (S^1) = \left\{A^{\rm os} \in {\mathcal B} ({\mathcal H} ) 
		\mid A\in L^{\infty}({\mathcal Q}, \Sigma_{S^1} ,
		{\rm d} \Phi_C) \right\} \; . 
	\]

\begin{lemma}[Klein \& Landau \cite{KL1}, Lemma 8.1]
\label{uiso}
The map $A\mapsto A^{\rm os}$ induces a 
weakly continuous \hbox{$^{*}$-isomorphism} between $L^{\infty}({\mathcal Q}, \Sigma_{S^1},
{\rm d} \Phi_C)$ and ${\mathcal U} (S^1)$. 
\end{lemma}

\begin{lemma}[Klein \& Landau \cite{KL1}, Theorem 11.2; see also \cite{GeJ}]
\label{ostheo} 
	\[
		\overline{{\mathcal U} (S^1) \Omega} = {\mathcal H} \; .
	\] 
In other words, application of bounded functions of time-zero fields to the vector~$\Omega$ 
generates a  dense set in ${\mathcal H}$.
\end{lemma}

\begin{proof}
Let $\Phi_F$ be the  Fock space field operator  and $\Omega_F $   the   vacuum 
vector in $\Gamma( \widehat{\mathfrak h} (S^1))$. 
The decomposition of $H_{{\mathbb C}}^{-1} (S^2)$ provided by Lemma \ref{detheo} allows us 
to restrict the map \eqref{dfock} to the Euclidean time-zero fields, 
	\[
			{\rm e}^{i \Phi (0, h ) } \mapsto {\rm e}^{i \Phi_F (h)} \Omega_F , 
			\quad h \in \widehat{\mathfrak h}(S^1) \; ,
	\]
thereby demonstrating that ${\mathcal H} \cong \Gamma(\widehat{\mathfrak h} (S^1))$
and that $ {\mathcal U} (S^1)$ can be identified with the von Neumann algebra 
generated by the Fock space field operators~$\Phi_F (h)$, with 
$h \in \widehat{\mathfrak h} (S^1, \mathbb{R})$, defined in terms of creation and annihilation 
operators (see, \emph{e.g.},~\cite{BR}). This establishes the result. 
\end{proof}

The conditional projection~\eqref{identi} identifies the 
time-zero fields $\Phi (0, h)\in L^{2}({\mathcal Q}, \Sigma_{S^1}, {\rm d} \Phi_{C})$ 
with unbounded operators $\Phi^{\rm os} (0,h)$ affiliated to the abelian 
algebra ${\mathcal U} (S^1)$: recall
	\[
 	\Phi (0, h) \in \bigcap_{1 \le p < \infty}  L^{p}({\mathcal Q}, \Sigma_{S^1} , {\rm d} \Phi_{C}) \; , 
		\qquad h \in C_{{\mathbb R}}^\infty (S^1)  \; , 
	\]
and approximate $ \Phi (0, h)$ by a sequence of 
$L^{\infty }({\mathcal Q}, \Sigma_{S^1}, {\rm d} \Phi_{C})$-functions 
	\[
		 \Phi_{n} (0, h) = \mathbb{1}_{[-n, n]} (h) \Phi (0, h) \; , \qquad  n \in {\mathbb N} \; . 
	\]
$\mathbb{1}_{[-n, n]} (h)$ is the characteristic function of the set $ \{ q \in Q \mid | \Phi (0, h) (q) | \le n \}$. 
According to~\eqref{identi}
	\[
		\Phi_{n}^{\rm os} (0,h) \in {\mathcal U} (S^1) \; . 
	\]
It follows that $\Phi^{\rm os} (0,h)$ is affiliated to ${\mathcal U} (S^1)$.

\begin{proposition} 
\label{knullosprop}
Consider the map (\ref{identi}). It follows that
\label{knullosproppage}	
	\begin{equation}
		\label{knullos-2}
		{\rm e}^{i \alpha K_0^{\rm os}}  A^{\rm os} \Omega 
		\doteq {\mathcal V}\bigl( {\rm U} (R_{0} (\alpha)) A \bigr) \; ,
		\qquad A \in L^{\infty}({\mathcal Q}, \Sigma_{S^1}, {\rm d} \Phi_C) \;  , 
	\end{equation}
extends to a strongly continuous unitary representation of the rotation group $U(1)$ 
on the Hilbert space ${\mathcal H}$. 
\end{proposition}

\begin{proof} 
The abelian algebra $ {\mathcal U} (S^1)$ and the vector $\Omega$ are  invariant w.r.t.~rotations around the 
$x_0$-axis. Use Lemma \ref{ostheo} to extend (\ref{knullos-2})  to a unitary representation 
of~$U(1)$ on~${\mathcal H}$.
\end{proof}

\section{Generalised path spaces}
\label{sec5.1}

Recall the following notions from \cite[Definition~1.3, p.~47]{KL4}:
 
\begin{definition}
\label{gps}
A {\em generalised pathspace} $({\mathcal Q}, \Sigma, \Sigma_{0}, {\rm U}(t), \Theta, \mu)$  consists of 
\label{gppage}
\begin{itemize}
\item [$ i.)$]  a probability space $({\mathcal Q}, \Sigma, \mu)$;
\item [$ ii.)$]  a distinguished sub $\sigma$-algebra $\Sigma_{0}\subset \Sigma$;
\item [$ iii.)$]  
a one-parameter group $ t\mapsto {\rm U}(t)$, $t \in {\mathbb R}$, of measure preserving automorphisms
of $L^{\infty}({\mathcal Q}, \Sigma, \mu)$, strongly continuous in
measure, such that 
	\[
	\Sigma=\bigvee_{t\in {\mathbb R}}{\rm U}(t)\Sigma_{0} \; ;
	\]
\item [$ iv.)$]   
	a measure preserving automorphism $\Theta$ of $L^{\infty}({\mathcal Q}, \Sigma,
\mu)$ such that $\Theta^{2}=\mathbb{1}$, 
	\[
		\Theta{\rm U}(t) = {\rm U}(-t)\Theta \; , \qquad t \in {\mathbb R} \; , 
	\] 
and $\Theta E_{0}= E_{0}\Theta$, where
$E_{0}$ is the conditional expectation with respect to~$\Sigma_{0}$.  
\end{itemize}
The generalised path space $({\mathcal Q}, \Sigma, \Sigma_{0}, {\rm U}(t), \Theta, \mu)$ is said to be
\emph {supported} by the probability space~$({\mathcal Q}, \Sigma, \mu)$.
\end{definition}

The properties $ iii.)$ and $ iv.)$ imply that ${\rm U}(t)$ extends to a strongly
continuous group of isometries of $L^{p}({\mathcal Q}, \Sigma, \mu)$  
and $\Theta$ extends to an isometry of $L^{p}({\mathcal Q}, \Sigma, \mu)$ for $1\leq
p<\infty$ \cite{KL1}. 
Two generalised path
spaces $({\mathcal Q}, \Sigma, \Sigma_{0}, {\rm U}_i (t), \Theta, \mu_i)$, $i=1,2$, are {\em equivalent}, if 
	\begin{equation}
		\label{equipathspace}
		\int_{\mathcal Q} \, {\rm d} \mu_1  \; {\rm U}_1 (t_1) F_1 \cdots {\rm U}_1 (t_n) F_n    
		= 
		\int_{\mathcal Q} \, {\rm d} \mu_2  \; {\rm U}_2 (t_1) F_1  \cdots {\rm U}_2 (t_n) F_n  
	\end{equation}
for all $ t_1, \ldots, t_n $ and $ F_1, \ldots, F_n \in C_{\mathbb R} ( {\mathcal Q})$, 
where $C_{\mathbb R} ( {\mathcal Q})$ is the space of real valued continuous 
functions on~${\mathcal Q}$. By convention, 
	\[
	{\rm U}(t_1 ) F_1 {\rm U}(t_2) F_2 \doteq {\rm U}(t_1 ) \Bigl( F_1 \bigl({\rm U}(t_2) F_2 \bigr) \Bigr) \; , 
	\]
which needs to be  distinguished from ${\rm U}(t_1 ) (F_1) {\rm U}(t_2) (F_2)$.

\begin{definition}
\label{gps2}
For $I\subset {\mathbb R}$, denote by $E^{I}$ the conditional expectation
with respect to the $\sigma$-algebra $\bigvee_{t\in I}\Sigma_{t}$, where
$\Sigma_{t}= {\rm U}(t)\Sigma_{0}$. 
The generalised path space $({\mathcal Q}, \Sigma, \Sigma_{0}, {\rm U}(t), \Theta, \mu)$ 
\begin{itemize}
\item[$ i.)$]
	is {\em periodic},  if ${\rm U}(2 \pi )=\mathbb{1}$;
\item[$ ii.)$] 
	is {\em OS-positive}, if $E^{[0,\pi ]}\Theta E^{[0, \pi ]}\geq 0$ as an operator on $L^{2}({\mathcal Q},
	\Sigma, \mu)$;
\item[$ iii.)$] satisfies the {\em two-sided Markov property} 
for semi-circles, if 	
\[
	E^{[0,\pi ]}\Theta E^{[0, \pi ]} =  E^{\{ 0, \pi  \} } \; .
\]
\end{itemize}
\end{definition}

\section{Path-spaces on the sphere}
\label{sec:3.6}

Recall the definition of the unitary group $\theta \mapsto {\rm U} ^{(\alpha)} (\theta) $ from 
Equ.~\eqref{wawa}, Section~\ref{sec1.3}. 

\begin{lemma} Let $\Sigma$ be the Borel $\sigma$-algebra on 
${\mathcal Q}\doteq {\mathcal D}'_{{\mathbb R}}(S^2)$.
Then 
	\[
	 {\rm U}^{(\alpha)} (\theta) \Theta = \Theta {\rm U}^{(\alpha)}(-\theta)  
	 \]
and, for $h\in {\mathcal D}_{{\mathbb R}} \left(I_+\right)$,  
	\begin{equation}
		\label{zeke}
		{\rm U}^{(0)}(\theta) \Phi(0, h) = \Phi(\theta, h) \; .
	\end{equation}
Note that the right hand side of \eqref{zeke} was defined in Equ.~\eqref{e1.4}. 
\end{lemma}

\begin{remark}
Results similar to \eqref{zeke} hold for time-zero fields supported on arbitrary 
half-circles $I_\alpha$, $\alpha \in
[0, 2\pi) \,$, with respect to the appropriate rotations $\theta \mapsto R^{(\alpha)} (\theta)$.
\end{remark}

Let $\Sigma^{(\alpha)}$ be the smallest sub $\sigma$-algebra of $\Sigma$ for which the functions 
$\{ \Phi(0,h) \mid h\in {\mathcal D}_{{\mathbb R}} \bigl( I_\alpha \bigr) \} $ are measureable. 
\label{sigmaalphapage}

\label{gppageds}
\begin{proposition}
\label{gspfree} For each $\alpha \in [0, 2 \pi)$, 
$({\mathcal Q}, \Sigma, \Sigma^{(\alpha)},{\rm U}^{(\alpha)} ( \,.\, ), \Theta, {\rm d} \Phi_{C})$ 
is a $(2 \pi)$-periodic, OS-positive  generalised path space (in the sense of Definition~\ref{gps}), 
which satisfies the two-sided Markov property for semi-circles.
\end{proposition}

\begin{proof} Use  Lemma \ref{1.0b}  to deduce that for $\alpha$ fixed 
$\Sigma=\bigvee_{\theta \in S^{1}}{\rm U}^{(\alpha)} (\theta)\Sigma^{(\alpha)} $. 		
The two-sided Markov property for semi-circles follows from the Markov property, Theorem
\ref{martheo} $ iv.)$.
\end{proof}

\section{Local symmetric semi\-groups}
\label{sec:3.7}

The rotations $\theta \mapsto R^{(\alpha)} (\theta)$ do not preserve the (closed) upper 
hemisphere~$\overline { S_+}$. In fact, the map
	\begin{equation}
		\label{motivate}
		 {\mathcal V}( F )\mapsto  {\mathcal V} \bigl( {\rm U}^{(\alpha)}(\theta)F) \; , 
	\end{equation}
is only defined  if both  $ F $ and 
${\rm U}^{(\alpha)}(\theta) F $ are in $L^{2}({\mathcal Q}, \Sigma_{\overline { S_+}}, {\rm d} \Phi_C)$.
The domain problems which arise, if one tries to  associate self-adjoint operators 
	\[
		P^{(\alpha)}(\theta) \colon {\mathscr D} 
		\to {\mathcal H} \; , \qquad {\mathscr D} \subset {\mathcal H} \; ,
	\]
to \eqref{motivate},  are addressed by  the theory of local symmetric semi\-groups developed by Fr\"ohlich~\cite{F80} and, 
independently, by Klein \& Landau \cite{KL1}\cite{KL2}:

\goodbreak
\begin{definition} 
\label{lsssg}
A pair  $\bigl(P(\theta), {\mathscr D}_{\theta} \bigr)$ forms a {\em local symmetric
semigroup} on a Hilbert space~${\mathcal H}$, if
\begin{itemize}
\item[$ i.)$] 
for each $\theta$, $0 \le \theta \le \pi$, fixed, the set ${\mathscr D}_\theta$ is a linear subset of ${\mathcal H}$. 
The union   
	\[
		{\mathscr D} =\bigcup_{0 < \theta \le \pi} {\mathscr D}_\theta 
	\] 
is dense in ${\mathcal H}$ and ${\mathscr D}_\theta \supset {\mathscr D}_{\theta'} $ if $\theta \le \theta' $;
\item[$ ii.)$] 
for each $\theta$, $0 \le \theta \le \pi$, $P(\theta)$ is a linear operator on ${\mathcal H}$ with 
domain ${\mathscr D}_\theta$  and
	\[
		P(\theta') {\mathscr D}_\theta \subset {\mathscr D}_{\theta-\theta'} 
		\quad  \hbox{for} \quad 0 \le \theta' \le \theta \le \pi \;  ;
	\]
\item[$ iii.)$] 
$ P(0) = \mathbb{1}$, and the {\em semi-group property}  
	\[
	 	P(\theta) P(\theta') = P(\theta + \theta')
	 \]
holds on ${\mathscr D}_{\theta+ \theta'} $ for $\theta, \theta', \theta + \theta' \in [0, \pi] $; 
\item[$ iv.)$] $P(\theta)$ is {\em symmetric}, i.e., 
	\[
	 \langle u, P(\theta) v \rangle_{\mathcal H} 
	 = \langle P(\theta) u, v \rangle_{\mathcal H} \quad  \forall u, v \in {\mathscr D}_\theta  \; , 
	\quad 0 \le \theta \le \pi  \; ;
	\] 
\item[$ v.)$] 
the map $\theta \mapsto P(\theta)$ is {\em weakly continuous}, \emph{i.e.}, if $ u \in {\mathscr D}_{\theta'}$, $0 \le \theta' \le \pi$, then 
	\[
	 	\theta \mapsto \langle u, P(\theta) u \rangle_{\mathcal H} 
	 \]
is a continuous function  for $ 0 < \theta < \theta'$.
\end{itemize}
\end{definition}

It is remarkable that  a local symmetric semi-group has a unique self-adjoint generator:

\begin{theorem}[Fr\"ohlich \cite{F80}; Klein \& Landau \cite{KL2}]  
\label{klein-l-f}
\label{K-L--F}
Let $\bigl(P (\theta), {\mathscr D}_{\theta}\bigr)$ be 
a local symmetric semigroup, acting on a Hilbert space ${\mathcal H}$. Then
there exists a unique self-adjoint operator $L$, the {\em generator} of the local
symmetric semigroup $\bigl(P (\theta), {\mathscr D}_{\theta}\bigr)$
on~${\mathcal H}$, such that 
	\[
	 P(\theta') \Psi = {\rm e}^{-\theta' L} \, \Psi  \; ,   
	\qquad \Psi \in {\mathscr D}_{\theta} \; , \quad 0\leq \theta'\leq \theta \; . 
	\]
\end{theorem}

\bigskip

Return to \eqref{motivate}. For $0\leq \theta\leq \pi$ fixed, set
	\begin{equation} 
		\label{malphatheta}
		{\mathcal M}^{(\alpha)}_{\theta} =L^{2} \Bigl({\mathcal Q}, \bigvee_{\theta' \in [0, \pi -\theta]}
		{\rm U}^{(\alpha)} (\theta' )\Sigma^{(\alpha)}, {\rm d} \Phi_C \Bigr), \quad 0\leq \theta\leq \pi \; .
	\end{equation}
${\mathcal M}^{(\alpha)}_{\theta}$ is the set of all $ F \in L^{2}({\mathcal Q}, 
\Sigma_{\overline { S_+}}, {\rm d} \Phi_C) $ for 
which ${\rm U}^{(\alpha)}(\theta) F \in L^{2}({\mathcal Q}, \Sigma_{\overline { S_+}}, {\rm d} \Phi_C) $.

\begin{definition}
Set  ${\mathscr D}_{\theta}^{(\alpha)}={\mathcal V} \bigl({\mathcal M}^{(\alpha)}_{\theta} \bigr)$. Define, for $0\leq \theta'\leq \theta$,
\label{palpha-page}
	\begin{align*}
		P^{(\alpha)}(\theta') \colon {\mathscr D}^{(\alpha)}_{\theta}& \to 
		 {\mathcal H} \nonumber \\
		{\mathcal V} (F)
		& \mapsto 
		 {\mathcal V} \bigl( {\rm U}^{(\alpha)}(\theta')F \bigr) \; , \quad F \in {\mathcal M}^{(\alpha)}_{\theta} \; .
	\end{align*}
\end{definition}

\begin{proposition} 
\label{3.19} 
$\bigl(P^{(\alpha)}(\theta), {\mathscr D}^{(\alpha)}_{\theta}\bigr)$ is a local symmetric semigroup. Its  generator $L^{(\alpha)}$ satisfies
\label{lssgpage}
	\[
		P^{(\alpha)}(\theta')\Psi
		= {\rm e}^{-\theta' L^{(\alpha)}}\Psi \; ,   
		\qquad \Psi \in {\mathscr D}^{(\alpha)}_{\theta}\; , \quad 0\leq \theta'\leq \theta \; . 
	\]
\end{proposition}

\begin{proof} 
Verify the conditions $ i.)$--$ v.)$ of Definition \ref{lsssg}.
\end{proof}

For the free dynamics, we can provide an explicit formula:

\begin{proposition}
\label{propfreeboost} 
Identifying ${\mathcal H}$ with $\Gamma (\widehat{\mathfrak h}(S^1))$
the generator of the boost can be identified with 
	\[
	L^{(\alpha)} = {\rm d} \Gamma ( \omega r \, \cos_{\psi + \alpha}) \; . 
	\]
Moreover, the spectrum ${\rm Sp} \bigl({L^{(\alpha)}}_{\upharpoonright I_\alpha} \bigr) \ge 0$ and 
${\rm Sp}\bigl( {L^{(\alpha)}}_{\upharpoonright I_{\alpha+\pi}} \bigr) \le 0$.
\end{proposition}

\begin{remark} Note that for $f, g \in {L^2 (S^1)}$
	\begin{align*}
	\langle f,  \omega \cos_{\psi + \alpha} g\rangle_{\widehat{\mathfrak h}(S^1)} 
	&= \tfrac{1}{2} \langle f,    \cos_{\psi + \alpha} g\rangle_{L^2 (S^1, r {\rm d} \psi)}
	\nonumber \\
	&= \langle  \omega  \cos_{\psi + \alpha}  f, g\rangle_{\widehat{\mathfrak h}(S^1)} \; . 
	 \end{align*}
This shows that $ \omega r \,  \cos_{\psi + \alpha}$ is symmetric. In fact, it is self-adjoint by construction.
\end{remark}

\begin{proof} We proceed in several steps. 

\bigskip
\noindent
$i.)$ First, we show that for $f \in {\mathcal D}_{\mathbb R} (  S^+ )$
	\[
		\int_{{\mathcal Q}} {\rm d} \Phi_C \;  \overline{ \Theta ( \Phi (f)) } \Phi(f) =  C(T_* f,  f) 
	\]
is equal to
	\[
		r \, \Bigl\| \int_0^\pi  r {\rm d}  \theta  \,  {\rm e}^{ -\theta  \varepsilon }
		 \frac{ (\mathbb{1}+\rho_{2 \pi})^\frac{1}{2}+\rho_{2 \pi}^\frac{1}{2}
			({\tt P}_1)_*} {\sqrt{2 | \varepsilon|}}\,|\cos_\psi|  f_{\theta}  
					\Bigr\|_{L^{2}(I_+ , \cos \psi^{-1} r {\rm d} \psi)}^2 \, . 
	\]
with
	\begin{equation}
	\label{rho}
	\rho_{2 \pi} =\frac {{\rm e}^{-2 \pi  |\varepsilon|}} {\mathbb{1}- {\rm e}^{-2 \pi   |\varepsilon|}}  \; , \qquad  
	\mathbb{1}+\rho_{2 \pi} =\frac {1} {\mathbb{1}- {\rm e}^{-2 \pi |\varepsilon|}}  \; ,
	\end{equation}
To show this, recall that according to Lemma \ref{1.0b} 
	\[
		\int_0^\pi  \, {\rm d}  \theta  \, r \cos \psi \, \Phi  (\theta, f_{\theta})=\Phi(f)\; , 
		\qquad f_{\theta} \equiv f (\theta, \, . \, ) \in {\mathcal D}_{\mathbb R} \left( I_+\right)  \; .  
	\]
Lemma \ref{coid} implies that $C(T_* f,  f)$ equals  
	\begin{align*}
&\int_0^\pi  r  \,  {\rm d}  \theta_1  \, \int_0^\pi  r  \,  {\rm d}  \theta_2  \, \times
		 \\
& \qquad \times  r\, \Bigl\langle  \cos_\psi  f_{\theta_1} ,  
					\tfrac{{\rm e}^{-   (\theta_{2}+\theta_{1} ) \varepsilon_{\upharpoonright I_+}   }
					+ {\rm e}^{- (2 \pi - (\theta_{2}+\theta_{1}) ) \varepsilon_{\upharpoonright I_+}    }}{
						2 \varepsilon_{\upharpoonright I_+}  
						(\mathbb{1}-{\rm e}^{- 2 \pi   \varepsilon_{\upharpoonright I_+}  })} 
					\cos_\psi  f_{\theta_2} \Bigr\rangle_{L^{2}(I_+ , \frac{r {\rm d} \psi} {\cos \psi} )}  \, . 
	\end{align*}
Using \eqref{rho}
we find that $C(T_* f,  f)$ equals
	\begin{align*}
		& \int_0^\pi  r  \,  {\rm d}  \theta_1   \int_0^\pi  r \, {\rm d}  \theta_2    \; 
		r\, \Bigl\langle    {\rm e}^{ -  \theta_1  \varepsilon_{\upharpoonright I_+}   } \cos_\psi f_{\theta_1}  ,  
			\frac{\mathbb{1}+\rho_{2 \pi}}{2 | \varepsilon|}
			{\rm e}^{ -  \theta_2  \varepsilon_{\upharpoonright I_+}   } \cos_\psi f_{\theta_2} 
			\Bigr\rangle_{L^{2}(I_+ , \frac{r {\rm d} \psi} {\cos \psi} )}
	\\		 		
	& \qquad  + \; \int_0^\pi r   \,  {\rm d}  \theta_1   \int_0^\pi r \, {\rm d}  \theta_2  \;  
				r\, \Bigl\langle  {\rm e}^{  \theta_1  \varepsilon_{\upharpoonright I_+}    } 
				\cos_\psi f_{\theta_1}  ,  \frac{\rho_{2 \pi}}{2 | \varepsilon|} \, 
					{\rm e}^{ \theta_2 \varepsilon_{\upharpoonright I_+}    } \cos_\psi f_{\theta_2}  
					\Bigr\rangle_{L^{2}(I_+ , \frac{r {\rm d} \psi} {\cos \psi} )} \, . 
						\nonumber
	\end{align*}
This formula is symmetric  up to a reflection. 
Note that $\pi- \theta$ is the angle measured starting from $I_-$. 
The second term in this sum equals
	\begin{align*}
		&\int_0^\pi  r \, {\rm d}  \theta_1   \int_0^\pi  r  \,  {\rm d}  \theta_2  \,   
				r\, \Bigl\langle  {\rm e}^{  -\theta_1  \varepsilon_{\upharpoonright I_-}    } (P_1)_*
				|\cos_\psi| f_{\theta_1} \;  ,  
				\\
			& \qquad \qquad \qquad \qquad \qquad  \frac{\rho_{2 \pi} }{2 | \varepsilon|} \, 
					{\rm e}^{ -\theta_2  \varepsilon_{\upharpoonright I_-}    } (P_1)_*
					|\cos_\psi| f_{\theta_2}  
					\Bigr\rangle_{L^{2}(I_- , |\cos \psi|^{-1} r {\rm d} \psi)} \, . 
						\nonumber
	\end{align*}
\color{black}
Consequently, 
	\[
		C(T_* f,  f)=  r \, \Bigl\| \int_0^\pi  r \, {\rm d}  \theta  \;   {\rm e}^{ -\theta  \varepsilon }
		 \frac{(\mathbb{1}+\rho_{2 \pi})^\frac{1}{2}+\rho_{2 \pi}^\frac{1}{2}
			(P_1)_*}{\sqrt{2 | \varepsilon|}} \,|\cos_\psi|  f_{\theta}  
					\Bigr\|^2_{L^{2}(S^1 , \frac{r {\rm d} \psi} {| \cos \psi |} )} \, . 
	\]

\smallskip
\noindent $ii.)$
Next, we show that
	\[
		C(T_* f,  f) = r \, 
		\Bigl\| \int_0^\pi  r\,  {\rm d}  \theta  \,  {\rm e}^{ -\theta \omega \cos_\psi  } f_{\theta}  
					\Bigr\|_{\widehat{\mathfrak h}(S^1)}^2 \, . 
	\]
We first note that Corollary \ref{wichtig} $ i.)$ implies that 
	\[ 
 		\| f_{\theta} \|^2_{\widehat{\mathfrak h} (S^1)} =  r \,  \Bigl\| 
		 \frac{(\mathbb{1}+\rho_{2 \pi})^\frac{1}{2}+\rho_{2 \pi}^\frac{1}{2}
			(P_1)_*}{\sqrt{2 | \varepsilon|}} \,|\cos_\psi|  f_{\theta}  
					\Bigr\|^2_{L^{2}(S^1 , \frac{r {\rm d} \psi} {| \cos \psi | } )} \, . 
	\]
The map\footnote{Note that the map $\mathbb{u}^{-1}$ was defined in Setion \ref{co-ps} as a map from $\widehat{\mathfrak h} (S^1)$
to ${\mathfrak d} (S^1)$, instead of to $L^{2}(S^1 , \frac{r{\rm d} \psi} { | \cos \psi | } )$.}
 ${\mathbb u}^{-1} \colon \widehat{\mathfrak h} (S^1) \to L^{2}(S^1 , \frac{r {\rm d} \psi} { | \cos \psi | } )$ given by
	\begin{equation*}
			{\mathbb u}^{-1} \doteq - 
			 \sqrt{r} \, 
			\frac{(\mathbb{1}+\rho_{2 \pi})^\frac{1}{2}+\rho_{2 \pi}^\frac{1}{2}
			(P_1)_*}{\sqrt{2 | \varepsilon|}} \,|\cos_\psi|   \, . 
	\end{equation*}
is unitary with inverse 
	\[
		 {\mathbb u} 
		= \frac{1}{
		 \sqrt{r} \,   |\cos_\psi|}\, \sqrt{2 | \varepsilon|} \big(  \rho_{2 \pi}^{\frac{1}{2}} (P_1)_* - (1+\rho_{2 \pi})^{\frac{1}{2}}  \big) \; .  
	\]
The generator of the boost is $ {\mathbb u} \circ\varepsilon\circ {\mathbb u}^{-1}$. 
Using $(P_1)_* \varepsilon = - \varepsilon (P_1)_*$,  we can now compute 
	\begin{align*}
			& {\mathbb u} \circ\varepsilon\circ {\mathbb u}^{-1} \\
			&\qquad =    |\cos_\psi|^{-1}\,
 							\big( (1+\rho_{2 \pi})^{\frac{1}{2}} - \rho_{2 \pi}^{\frac{1}{2}} (P_1)_*\big) 
						\varepsilon 
							\big( (1+\rho_{2 \pi})^{\frac{1}{2}}+\rho_{2 \pi}^{\frac{1}{2}}  (P_1)_*\big)\,
						|\cos_\psi| 
			\nonumber \\
			& \qquad =  |\cos_\psi|^{-1} \varepsilon 
			\big( (1+2\rho_{2 \pi})  
			+ 2\rho_{2 \pi}^{\frac{1}{2}}(1+\rho_{2 \pi})^{\frac{1}{2}} (P_1)_*\big) \,|\cos_\psi| 
			\nonumber \\
			& \qquad = |\cos_\psi|^{-1} \varepsilon  \Bigl(\coth\pi |\varepsilon|
		+ \tfrac{(P_1)_*}{\sinh\pi|\varepsilon|} \Bigr)^{-1} |\cos_\psi|
			\nonumber \\
			& \qquad = \omega r \,  \cos_\psi \; . 
	\end{align*}
Thus
	\[
	\int_{{\mathcal Q}} {\rm d} \Phi_C \;  \overline{ \Theta ( \Phi (f)) } \Phi(f) = r \, 
		\Bigl\| \int_0^\pi  r \,  {\rm d}  \theta  \,  {\rm e}^{ -\theta  \omega  \cos_\psi  } f_{\theta}  
					\Bigr\|_{\widehat{\mathfrak h}(S^1)}^2 \, . 
	\]
This identity verifies reflection positivity, and it also verifies that the generator of the boost on 
$\widehat{\mathfrak h} (S^1)$ is $\omega r \,  \cos_\psi$.

\goodbreak
\bigskip
\noindent
$ iii.)$ Finally, second quantization yields
	\[
	{L^{(\alpha)} }_{\upharpoonright I_\alpha}  = 
	{\rm d} \Gamma ( \chi_{I_\alpha} \omega r \, \cos_{\psi + \alpha} \chi_{I_\alpha}) \; , 
	\]
where $\chi_{I_\alpha}$ is the characteristic function of the half-circle $I_\alpha$. 
Since $\omega >0$, the spectral properties now
follow from the sign the cosine function takes on the circle:
	\[
	\langle h, \chi_{I_\alpha} \omega r\,  \cos_{\psi + \alpha} \chi_{I_\alpha} 
	h \rangle_{\widehat{\mathfrak h} (S^1)} 
	= 
	\langle h , r \, {\cos_{\psi + \alpha} }_{\upharpoonright I_\alpha} 
	h \rangle_{L^2(I_\alpha, r {\rm d} \psi)} \ge 0 \; . 
	\]
A similar result holds for $I_{\alpha+\pi}$.
\end{proof}

We complement this result with an explicit formula for the generator of rotations:

\begin{proposition}
\label{propfreeboost2} 
Identifying ${\mathcal H}$ with $\Gamma (\widehat{\mathfrak h}(S^1))$
the generator of the rotations $R_{0}$ can be identified with 
	\[
	K_0 = {\rm d} \Gamma \left( -i   \partial_\psi   \right) \; . 
	\]
\end{proposition}

\begin{proof} By definition (see \eqref{knullos}) we have 
	\[	{\rm e}^{i \alpha K_0}  {\rm e}^{i \Phi^{\rm os} (0,h)} \Omega 
		= {\rm e}^{i \Phi^{\rm os} (0,R_0 (\alpha)_* h)} \Omega  \;  , \qquad h \in \widehat{\mathfrak h}(S^1) \; . 
	\]
Thus $K_0 = {\rm d} \Gamma \left( -i  \partial_\psi   \right) $.
\end{proof}

Finally, recall that according to Proposition \ref{UIR-S1} the rotations 
	\[
	\bigr(\widehat{u} (R_0(\alpha)) h \bigl) (\psi) = h (\psi - \alpha) \; , 
	\qquad \alpha \in [0, 2\pi) \; , \quad h \in \widehat{\mathfrak h} (S^1) \; , 
	\]
and the boosts 
	\[
		\widehat{u} (\Lambda_1(t)) = {\rm e}^{i t \omega r\, \cos_\psi } \; , \qquad t \in \mathbb{R} \; , 
	\]
generate a unitary representation of $SO_0(1,2)$ on $\widehat{\mathfrak h} (S^1)$. 

\section{Tomita-Takesaki modular theory}
\label{sec:5.4}

The abelian  algebra 
	\[
	{\mathcal U} (I_\alpha)
	=\{ A^{\rm os} \in {\mathcal B}({\mathcal H}) \mid A\in 
	L^{\infty}({\mathcal Q}, \Sigma_{I_\alpha},{\rm d} \Phi_C)\} 
	\] 
together with the group of unitary operators $\{{\rm e}^{-it  L^{(\alpha)}} \mid t \in {\mathbb R} \}$ 
generate the non-abelian von Neumann algebra
	\begin{equation}
	\label{RvN}
		{\mathcal R}  \bigl(I_\alpha\bigr) \doteq 
		\bigvee_{t \in {\mathbb R}} 
		\left( {\rm e}^{-it  L^{(\alpha)}} {\mathcal U} (I_\alpha)  
		{\rm e}^{it  L^{(\alpha)}}\right) \; . 
	\end{equation}
The vector $\Omega$ is cyclic and separating for ${\mathcal R}(I_\alpha) $: 
assume that for $A, B \in {\mathcal U} (I_\alpha)$ the maps
	\begin{equation}
	\label{cycsep}
	t \mapsto \langle \Psi , A {\rm e}^{it  L^{(\alpha)}} B \Omega \rangle_{\rm os} = 0  
	\end{equation}
vanish identically. Then analyticity of the maps \eqref{cycsep}  implies that 
	\[
	\langle \Psi , A {\rm e}^{-\pi L^{(\alpha)}} B \Omega \rangle_{\rm os} = 0  \;, 
	\]
showing that $\Psi$ vanishes, as its scalar product with a dense set of vectors is zero. 
(As before, we use that ${\rm e}^{-\pi L^{(\alpha)}} $
maps ${\mathcal U} (I_\alpha)$ 
to the opposite algebra ${\mathcal U} (I_{\alpha + \pi})$, and the fact that 
${\mathcal U} (I_\alpha) \vee {\mathcal U} (I_{\alpha + \pi})
= {\mathcal U} (S^1)$.) 

Since $L^{(\alpha)}$ does not mix $\widehat {{\mathfrak h}}(I_{\alpha})$ and 
$\widehat {{\mathfrak h}}(I_{\alpha + \pi})$, the same argument can be applied to the 
commutant ${\mathcal R}  \bigl(I_\alpha\bigr)' = {\mathcal R}  \bigl(I_{\alpha + \pi}\bigr)$. 
This implies that $\Omega$ is separating for ${\mathcal R}  \bigl(I_\alpha\bigr)$.

Thus the map
	\[
	A \Omega \mapsto A^* \Omega \; ,
	\qquad A \in {\mathcal R} (I_\alpha )\;  ,
	\]
is well-defined and closeable, and one can study its polar decomposition directly
(see, \emph{e.g.},~\cite{BR}); but alternatively, we can reconstruct it from the Euclidean theory:

\begin{definition}
\label{def4}
Let $J^{(\alpha)}$ denote the unique extension of
\label{Jalpha}
	\begin{equation}
	  J^{(\alpha)} {\mathcal V} (F) \doteq {\mathcal V}\bigl(\Theta^{(\alpha)}_{\pi/2} \overline{F} \bigr)\; , 
	  \qquad F \in L^2( {\mathcal Q}, \Sigma_{\overline { S_+}}, {\rm d} \Phi_C) \; ,
	\label{est1.2}
	\end{equation}
with
	\begin{equation}
	\label{thetapi}
	\Theta^{(\alpha)}_{\pi/2}\doteq{\rm U}^{(\alpha)} (\pi/2) \Theta {\rm U}^{(\alpha)} (-\pi/2) \;  
	\end{equation}
on $L^2( {\mathcal Q}, \Sigma , {\rm d} \Phi_C)$.
Thus $J_1 \equiv J^{(0)}$ stems 
from the reflection of $\overline { S_+} $ at the ($x_0$-$x_1$)-plane 
(which clearly preserves $\overline { S_+}$).
\end{definition}

\begin{theorem}[Klein and Landau, Theorem 12.1 \cite{KL1}] 
\label{th3.10}
The operator ${\rm e}^{- \pi  L^{(\alpha)}} $ is the Tomita-Takesaki 
\emph{modular operator}\index{modular operator} for the pair
$\bigl( {\mathcal R} (I_\alpha), \Omega \bigr)$,  and 
$J^{(\alpha)}$ is the corresponding \emph{modular conjugation}\index{modular conjugation}, \emph{i.e.}, 
	\[
		J^{(\alpha)} {\rm e}^{- \pi  L^{(\alpha)}} A \Omega = A^* \Omega 
		\qquad \forall A \in {\mathcal R} (I_\alpha )
	\]
and  $J^{(\alpha)} {\mathcal R} (I_\alpha )  J^{(\alpha)} 
= {\mathcal R} (I_\alpha ) '  $.
\end{theorem}

\section{Multi-analyticity of the correlation functions} 
\label{sec:multan}
The following result is a direct consequence of the reconstruction theorem. 

\begin{proposition}[Klein \& Landau, Lemma 8.4 \cite{KL1}]
\label{boostanaliticity}
Let $\theta_1, \ldots, \theta_n \ge 0$ and $\sum_{j=1}^n \theta_j \le \pi$.  Then, for all 
$A^{\rm os}_1, \ldots, A^{\rm os}_n \in {\mathcal U} \left(I_\alpha \right)$,  
$\alpha \in [0, 2 \pi )$,  the vector
	\[
		A^{\rm os}_n {\rm e}^{-\theta_{n-1} L^{(\alpha)}} A^{\rm os}_{n-1} 
			\cdots {\rm e}^{-\theta_{1} L^{(\alpha)}} A^{\rm os}_1 \Omega 
			\in {\mathscr D} \bigl({\rm e}^{-\theta_{n} L^{(\alpha)}}\bigr) \; . 
	\]
The linear span of such vectors is dense in ${\mathcal H}$ and 
	\begin{align*}		 
				& {\rm e}^{-\theta_{n} L^{(\alpha)}} A^{\rm os}_n 
		{\rm e}^{-\theta_{n-1} L^{(\alpha)}} A^{\rm os}_{n-1} 
		\cdots 
		{\rm e}^{-\theta_{1} L^{(\alpha)}} A^{\rm os}_1 \Omega \\  
			& \qquad \qquad = {\mathcal V} \bigl( {\rm U}^{(\alpha)}(\theta_n) A^{\rm os}_n 
			{\rm U}^{(\alpha)}(\theta_{n-1}) A_{n-1}\cdots {\rm U}^{(\alpha)}(\theta_1) A_1 \bigr) \; .  
	\end{align*}
\end{proposition}

\begin{remark} It is easy to extend 
this result to Euclidean time-zero fields:
for $ h_1,  \ldots, h_n \in {\mathcal D}_{\mathbb R} (I_+)$ and 
$\theta_1, \ldots, \theta_n \ge 0$, $\sum_{j=1}^n \theta_j \le \pi$,  
	\begin{align}
		\label{26nna}
		&{\rm e}^{-\theta_{n} L^{(0)}} \Phi^{\rm os} (0, h_n) 
		\cdots 
		{\rm e}^{-\theta_{1} L^{(0)}} \Phi^{\rm os} (0, h_1) \Omega 
			\nonumber \\ 
		&\qquad \qquad 
		= {\mathcal V} \bigl( {\rm U}^{(0)}(\theta_n) \Phi (0, h_n) 
			\cdots {\rm U}^{(0)}(\theta_1)  \Phi (0, h_1) \bigr) \; . 
	\end{align}
Formula \eqref{26nna} can be justified by verifying that products of Euclidean sharp-time fields are in $L^2 ( Q, \Sigma, {\rm d}\Phi_C)$. 
Similar results hold for arbitrary half\-circles $I_\alpha$, $\alpha \in [0, 2 \pi)$. 
\end{remark}

\begin{corollary} 
Let ${\rm W} (h_1) = {\rm e}^{i \Phi^{\rm os} (0,h_n)}$, with $h_i \in {\mathcal D}_{{\mathbb R}}(I_+)$, $ i= 1, \ldots, n$. 
It follows that the map  
	\begin{align*}
		&  G_{\rm os} \bigl(t_{1}, \dots, t_{n}; {\rm W} (h_1), \dots, {\rm W}  (h_n) \bigr) \doteq
		 \\
		& \qquad \doteq \bigl(\Omega , {\rm e}^{-i \frac{t_{n}}{r} L^{(0)}} 
		{\rm W}  (h_n) {\rm e}^{i \frac{t_n - t_{n-1}}{r} L^{(0)}} 
		{\rm W}  (h_{n-1}) \cdots {\rm e}^{i \frac{t_2 -t_{1}}{r} L^{(0)}} 
		{\rm W}  (h_1) \Omega \bigr)  
	\end{align*}
is holomorphic in the set
	\begin{equation}
		{\mathcal I}_{2\pi r}^{n+}= \bigl\{(z_{1}, \dots, z_{n})\in {\mathbb C}^{n} \mid \Im z_{i}<
		\Im z_{i+1}, \; \Im z_{n}- \Im z_{1}<2\pi r \bigr\},
	\end{equation}
and continuous on $\overline{{\mathcal I}_{2\pi r}^{n+}}$.  Moreover, 
	\[
		G_{\rm os} \bigl(  i r \theta_{1}, \dots,  i r \theta_{n}; {\rm W}  (h_1), \dots, {\rm W}  (h_n)\bigr)
		=\prod_{1\leq i, j\leq n}{\rm e}^{-\frac{1}{2} C_{|\theta_{i}-\theta_{j}|} (   h_{i},  h_{j}) } \; , 
	\] 
where $C_{|\theta_{i}-\theta_{j}|} (   h_{i},  h_{j})$ is defined in Equ.~(\ref{coideq}). 
\end{corollary}

\begin{proof} Multi-analyticity  has been shown by Araki \cite{A2}. To derive the final formula
compute
	\begin{align*}
		G_{\rm os} \bigl(  i r \theta_{1}, \dots,  i r \theta_{n};{\rm W}  (h_1), \dots, {\rm W}  (h_n)\bigr)
		&= \left(\Omega_{E}  ,  {\rm e}^{i \Phi_{E} (\theta _n, h_n)}     \cdots  
		{\rm e}^{i \Phi_{E} (\theta _1, h_1)} \Omega_{E} \right) \nonumber \\
		&=  \int_{{\mathcal Q}} {\rm d}\Phi_{C} \; {\rm e}^{i \Phi\bigl( \sum_{i=1}^n 
		\delta (\, . \, - \theta_i )\otimes h_i \bigr)} \nonumber \\
		&= 
		\prod_{1\leq i, j\leq n}{\rm e}^{-\frac{1}{2} C_{|\theta_{i}-\theta_{j}|} (   h_{i},  h_{j}) } \; . 
	\end{align*}  
We note that we have assumed that the $h_i$, $i=1, \ldots, n$, are real valued.
\end{proof}

\goodbreak

\section{The interacting vacuum vector}

Equ.~\eqref{intL1} implies $ {\rm e}^{- V \left(\overline{S_+} \right)} 
\in L^{2}({\mathcal Q}, \Sigma, {\rm d}\Phi_{C})$.  Thus the vector 	
\label{intvacuumpage}
	\begin{equation}
	\label{intvacuum} 
		\frac{ {\mathcal V}  \bigl({\rm e}^{- V \left(\overline{S_+}\right)} \bigr) }
		{ \left\|{\mathcal V}  \bigl({\rm e}^{- V \left(\overline{S_+}\right)} \bigr) \right\| }
		\doteq \Omega_{\rm int} \in {\mathcal H}
	\end{equation}
is well-defined. The corresponding vector state will be identified as the {\em interacting de Sitter
vacuum} in the sequel. 
We emphasise that there is no need to distinguish a wedge to define 
the vector (\ref{intvacuum}).

\begin{lemma} Let $h \in C^\infty (S^1)$. Then
\label{kappa}
	\begin{align*}
		& \int_{\overline{S^2_+}} {\rm d} \Omega(x) \int_{\overline{S^2_-}} {\rm d} \Omega(y) \, (\delta \otimes h) (x) \, 
		P_{- \frac{1}{2} - i \nu} \bigl( - \tfrac{x \cdot y}{r^2} \bigr) \\
				& \qquad 
		=  \underbrace{ \frac{ \sqrt{\pi} \, r}{ | \Gamma ( \frac{3}{4} + i \frac{\nu}{2} ) |^4 } 
		\frac{ (2\pi)^2}{  | \Gamma ( \frac{1}{4} + i \frac{\nu}{2} ) |^2 } }_{:= \kappa_0} \; 
		c_\nu
			\int_{S^1} r {\rm d} \psi \int_{S^1} r {\rm d} \psi'  \; h (\psi) \,  P_{-\frac{1}{2} - i \nu} (- \cos (\psi - \psi') )  \; . 
	\end{align*}
\end{lemma}

\begin{proof} We use geographical coordinates. As
	\[
		\frac{x \cdot y}{r^2} = \sin \theta \sin \theta' + \cos \theta \cos' \theta \cos (\psi - \psi')  \; , 
	\]
we find 
	\begin{align*}
		& \int_{S^2_+} {\rm d} \Omega(x) \int_{S^2_-} {\rm d} \Omega(y) \,  (\delta \otimes h) (x) \, 
		P_{- \frac{1}{2} - i \nu} ( - \tfrac{x \cdot y}{r^2} ) \\
		& \qquad = r^3
			\int^{\frac{\pi}{2}}_0  \cos \theta' {\rm d} \theta'
			\int_{S^1}  {\rm d} \psi \int_{S^1}  {\rm d} \psi' \; 
		h (\psi) \, P_{- \frac{1}{2} - i \nu} \bigl( -  \cos \theta' \cos (\psi - \psi') \bigr) \; . 
\end{align*}
A special case of \eqref{eq:addition-formula} is the following formula.
	\begin{align}
		P_s  \bigl( - & \cos(\psi-\psi') \cos \theta' \bigr)
		\label{crucial}
		 \\
		&	= \frac{ \sqrt{\pi} } { \Gamma {(\frac{1}{2} - \frac{s}{2})}  \Gamma {(\frac{s}{2} +1)} } 
		P_s(\sin(\psi-\psi')) \nonumber
		\\
		& \qquad + 2 \sum_{k=1}^\infty (-1)^k \frac{\Gamma(s -k + 1)}{\Gamma(s+k+1)} 
 				\cos(k \theta') \; P_s^k (0) P_s^k (\sin(\psi-\psi')) \;.
		\nonumber
	\end{align}
We have used 
$P_{s}(0) = \frac{ \sqrt{\pi} } { \Gamma {(\frac{1}{2} - \frac{s}{2})}  \Gamma {(\frac{s}{2} +1)} } 
= \frac{ \sqrt{\pi} } { \Gamma {(\frac{3}{4} + i \frac{\nu}{2})}  \Gamma {(\frac{3}{4} - i \frac{\nu}{2})} }$.
Next recall that 
	\[
		\Gamma(2z)= \frac{2^{2z-1}}{\sqrt{\pi}}\; \Gamma(z)\Gamma\left(z+\tfrac{1}{2}\right) \; . 
	\]
Thus
	\begin{align*}
		c_\nu  & = \frac{\Gamma\left(\frac{1}{2} - i \nu  \right) \Gamma\left(\frac{1}{2} + i \nu \right)}{2\pi}  \\
		& = \frac{\ \Gamma(\frac{1}{4} - i \frac{\nu}{2})
		\Gamma\left(\frac{3}{4} - i \frac{\nu}{2}\right)
		 \Gamma(\frac{1}{4} + i \frac{\nu}{2})
		\Gamma\left(\frac{3}{4} + i \frac{\nu}{2}\right) }{(2\pi)^2} 
		\; . 
	\end{align*}
When integrating out the $\theta'$ variable, only the first term on the r.h.s.~contributes. 
Thus
\begin{align*}
		& \int_{S^2_+} {\rm d} \Omega(x) \int_{S^2_-} {\rm d} \Omega(y) \,  (\delta \otimes h) (x) \, 
		P_{- \frac{1}{2} - i \nu} ( - \tfrac{x \cdot y}{r^2} ) \\
		& \qquad 
		=  \frac{ \sqrt{\pi} \, r}{ | \Gamma ( \frac{3}{4} + i \frac{\nu}{2} ) |^4 } 
		\frac{ (2\pi)^2}{  | \Gamma ( \frac{1}{4} + i \frac{\nu}{2} ) |^2 } \; 
		c_\nu
			\int_{S^1} r \, {\rm d} \psi \int_{S^1}  r \, {\rm d} \psi'  \; h (\psi) \,  P_{-\frac{1}{2} - i \nu} (- \cos (\psi - \psi') )  \; . 
\end{align*}
The last equality follows from shifting the integration in the $\psi'$ variable. 
\end{proof}

\begin{remark}
The integral over $\psi'$  can be computed using the 
formula 
	\[
		\lambda P_\lambda (x) =   x P_\lambda' (x) - P_{\lambda -1}' (x) \; , 
	\] 
which implies that 
\begin{align*}
	\lambda \int_0^1 {\rm d} x \; P_\lambda (x) & = \int_0^1 {\rm d} x \; x P_\lambda' (x) - P_{\lambda -1} (1) + P_{\lambda -1} (0)
	\\
	& = P_\lambda (1) - \int_0^1 {\rm d} x \; P_\lambda (x) - P_{\lambda -1} (1) + P_{\lambda -1} (0) \; . 
\end{align*}
Note that $P_\lambda (1) = P_{\lambda -1}(1) = 1$ and 
$
P_{\lambda}(0) = \frac{ \sqrt{\pi} } { \Gamma {(\frac{1}{2} - \frac{\lambda}{2})}  \Gamma {(\frac{\lambda}{2} +1)} } 
$. Thus 
	\[
		\int_{0}^1 {\rm d} u  \;  P_s(u) = \frac{ \sqrt{\pi} }{ (1+s) \Gamma ( 1- \frac{s}{2} ) \Gamma ( \frac{s}{2} + \frac{1}{2} ) } \; . 
	\]
\end{remark}

The function $ {\rm e}^{- V \left(\overline{S_+} \right)} $ is $\Sigma_{S_+}$-measurebale. We now 
construct the conditional expectation value.

\begin{theorem}[Conditional Expectation]
Let\footnote{Note that $\lim_{|y| \to \infty} | \Gamma(x+iy) | = \sqrt{2 \pi} \, | y |^{x- \frac{1}{2}} \, {\rm e}^{- \pi |y|  /2}$.}
	\[
		V (S^1) 
		\doteq \int_{S^1}  {\rm d} \psi \;   {:} {\mathscr P}( \kappa_0 \Phi(0,\psi)) {:}_{C_{0}}\; ,
		\qquad \kappa_0 = 
		\frac{ \sqrt{\pi} \, r}{ | \Gamma ( \frac{3}{4} + i \frac{\nu}{2} ) |^4 } 
		\frac{ (2\pi)^2}{  | \Gamma ( \frac{1}{4} + i \frac{\nu}{2} ) |^2 } 
		 \; . 
	\]
Here $1$ is the function, which is constant on $S^1$ and equal to $1$ at every point. It follows that
	\[ 
		{\mathcal V}  \bigl({\rm e}^{- V \left(\overline{S_+}\right)} \bigr) = 
		 {\rm e}^{-  V (S^1) }  
		\qquad \text{on $L^2 (Q, \Sigma_0 , {\rm d} \Phi_C)$}\; . 
	\]
\end{theorem}

\begin{remark}
We note that $V(S^1)$ has to be distinguished from $V_0 (\chi_{S^1})$; 
the two expressions refer to different polynomials in $\Phi( 0, \psi)$ integrated over the time-zero circle $S^1$. 
\end{remark}

\begin{proof} By linearity, one may assume that ${\mathscr P}(\lambda)= \lambda^{n}$.
We can than expand the exponential function, and identifying $L^2( {\mathcal Q}, \Sigma, {\rm d} \Phi_C ) $ 
with $\Gamma( H_{{\mathbb C}}^{-1}(S^2))$
we may apply the projection  (see \eqref{pfock})
	\[
		E_{\Sigma_{S^1}} \doteq \Gamma \left(e({S^1}) \right) 
	\]
to the  vacuum vector $\Omega_E $ in~$\Gamma( H_{{\mathbb C}}^{-1} (S^2))$; \emph{i.e.}, 
	\begin{align*}
		{\mathcal V}  \bigl({\rm e}^{- V \left(\overline{S_+}\right)} \bigr) 
		& = E_{\Sigma_{S^1}} {\rm e}^{- V \left(\overline{S_+}\right)} \Omega_E \\
		& = \Omega_E  - \Gamma \left(e({S^1}) \right)  V \left(\overline{S_+}\right) \Omega_E 
		+ \Gamma \left(e({S^1}) \right) \bigl( V \left(\overline{S_+}\right) \bigr)^2 \Omega_E - \ldots 		
	\end{align*}
Let us first consider the term $\Gamma \left(e({S^1}) \right)  V \left(\overline{S_+}\right) \Omega_E $ in some more
detail. Let $\chi_{\overline{S_+}}$ denote the characteristic function of the closed upper hemisphere. 
We have seen in the proof of Theorem \ref{uvtheo} that 
	\[
		\lim_{m \to \infty} \int_0^{\pi} r \, {\rm d} \theta \int_{-\pi/2}^{\pi/2} r \cos \psi \, {\rm d} \psi \; 
				{:}\Phi \bigl(  \delta^{(2)}_{m}(\, .-\theta, \, .- \psi)   \bigr)^n {:}_{C}  
	\]
is a linear combination of Wick monomials of the form 
	\begin{align}
		& 		\sum_{j=0}^{n}  \begin{pmatrix} 
								n \\ 
								j
							\end{pmatrix}  \sum_{\ell_1 = 0}^\infty \sum_{k_1= - \ell_1}^{\ell_1} 
	 \cdots
	 \sum_{\ell_n = 0}^\infty \sum_{k_n= - \ell_n}^{\ell_n}\;  (-1)^{\sum_{i=1}^j k_i} 
		\nonumber \\
		& \qquad  \qquad    \qquad  \qquad
		\qquad  \times		 w^{(n)} (\ell_{1}, k_1, \ldots, \ell_{j}, k_{j}, 
		\ell_{j+1}, k_{j+1}, \ldots, \ell_{n}, k_{n}) 
		\nonumber \\
		& \qquad  \qquad  \qquad \qquad  \qquad  \qquad
		\qquad   \times		a^{*}_{\ell_1, k_1} \cdots a^{*}_{\ell_j, k_j}
				a_{\ell_{j+1}, - k_{j+1}} \cdots a_{\ell_{n}, - k_{n}} \; , 
	\label{sum-m-2}
	\end{align}
where $a^{(*)}_{\ell_i, k_i} \equiv a^{(*)} (Y_{\ell_i, k_i})$ and 
	\begin{align*}
		w^{(n)} (\ell_{1}, & k_1, \ldots,  \ell_{j}, k_{j}, 
		\ell_{j+1}, k_{j+1}, \ldots, \ell_{n}, k_{n}) \\
		& = \frac{1}{ 2^{\frac{n}{2}} r^{2n} } 
		\int_0^{2 \pi} r \, {\rm d} \theta \int_{-\pi/2}^{\pi/2} r \cos \psi \, {\rm d} \psi \;  \chi_{\overline{S_+}}(\theta,\psi) \; \prod_{i=1}^n
			 \overline{Y_{\ell_i, k_i} (\theta, \psi)}   \; .
	\end{align*}
When \eqref{sum-m-2} is applied to the free vacuum vector $\Omega_E$, only the term $j=n$ is non-zero. Moreover, 
we can undo the expansion into spherical harmonics: 
	\begin{align*}
		\lim_{m \to \infty} & \int_0^{\pi} r \, {\rm d} \theta \int_{-\pi/2}^{\pi/2} r \cos \psi \, {\rm d} \psi \; 
				{:}\Phi \bigl(  \delta^{(2)}_{m}(\, .-\theta, \, .- \psi)   \bigr)^n {:}_{C}  \Omega_E  \\
		& = \int_{S^2} {\rm d} \Omega (\vec x) \; \chi_{\overline{S_+}}(\vec x) \; 
		\underbrace{a^* ( \delta^2_{\vec x} ) \cdots a^* ( \delta^2_{\vec x} )}_{n-times}  \Omega_E \; . 
	\end{align*}
The resulting $n$-particle wave-function   
	\[
		f (\theta_1, \psi_1, \ldots, \theta_n, \psi_n) =  \chi_{\overline{S_+}}(\theta_1, \psi_1) \delta (\theta_1- \theta_2) \delta (\psi_1- \psi_2) 
		\cdots \delta (\theta_1- \theta_n) \delta (\psi_1- \psi_n)   
	\]
is in $\mathbb{H}^{-1} (S^2) \otimes_s \cdots \otimes_s \mathbb{H}^{-1} (S^2) $. 

In second order,  a term of the form \eqref{sum-m-2} is applied to the $n$-particle wave-function 
$f (\theta_1, \psi_1, \ldots, \theta_n, \psi_n)$. Now all terms in the sum over $j$ 
will contribute. A generic term is of the form
	\[
		\int_{S^2} {\rm d} \Omega (\vec y) \; \chi_{\overline{S_+}}(\vec y) \; 
		\underbrace{a^* ( \delta^2_{\vec y} ) \cdots a^* ( \delta^2_{\vec y} )}_{j-times} 
		\underbrace{a ( \delta^2_{\vec y} ) \cdots a ( \delta^2_{\vec y} )}_{(n-j)-times} 
		f (\theta_1, \psi_1, \ldots, \theta_n, \psi_n) \; . 
	\] 
The resulting wave-function is in $\Gamma^{(2j)}(\mathbb{H}^{-1}(S^2))$. 

Higher orders clearly have a very involved structure.
Nevertheless, we are able to apply the projection $e(S^1) \otimes_s  \cdots \otimes_s e(S^1)$. In first order,  
we can compute the conditional expectation by integrating against wave-functions of the form 
	\[
		\left( \delta ( \theta_1) \otimes h_1 (\psi_1) \right) \otimes_s \cdots \otimes_s
		\left( \delta ( \theta_n) \otimes h_n (\psi_n) \right) \; .  
	\]
This yields 
	\begin{align*}
		& \int_{S^2} {\rm d} \Omega(x_1) \int_{S^2} {\rm d} \Omega(y_1) \cdots 
		\int_{S^2} {\rm d} \Omega(x_n) \int_{S^2} {\rm d} \Omega(y_n) \; 
		 (\delta \otimes h_1) (x_1) \cdots (\delta \otimes h_n) (x_n) \, 
		\\
		& \qquad \quad \times 
		 P_{- \frac{1}{2} + i \nu} \bigl( - \tfrac{x_1 \cdot y_1}{r^2} \bigr) 
		\cdots P_{- \frac{1}{2} + i \nu} \bigl( - \tfrac{x_n \cdot y_1}{r^2} \bigr) \;  
		\chi_{\overline{S_+}}(y_1) \delta_{y_1} (y_2) \cdots \delta_{y_1} (y_n) \\
		& \int_{S^2} {\rm d} \Omega(y_1) \int_{S^2} {\rm d} \Omega(x_1) \cdots \int_{S^2} {\rm d} \Omega(x_n) \; 
		 (\delta \otimes h_1) (x_1) \cdots (\delta \otimes h_n) (x_n) \, 
		\\
		&  
		\qquad \quad \times P_{- \frac{1}{2} + i \nu} \bigl( - \tfrac{x_1 \cdot y_1}{r^2} \bigr) 
		\cdots P_{- \frac{1}{2} + i \nu} \bigl( - \tfrac{x_n \cdot y_1}{r^2} \bigr) \;  \chi_{\overline{S_+}}(y_1) \\
		& = \int_{S^1} r\, {\rm d} \psi_1 \int_{S^1} r \, {\rm d} \psi_1' \cdots \int_{S^1} r\, {\rm d} \psi_n \int_{S^1} r\, {\rm d} \psi_n' \; 
		h_1 (\psi_1) \cdots h_n (\psi_n) 
		\\
		& \qquad \quad \times  \kappa_0 c_\nu P_{- \frac{1}{2} + i \nu} \bigl( - \cos (\psi_1 - \psi_1') \bigr)
		\cdots  \kappa_0 c_\nu  \color{black} P_{- \frac{1}{2} + i \nu} \bigl( - \cos (\psi_n - \psi_n') \bigr) 
		\\
		& \qquad \quad \qquad
		\times \chi_{S^1}(\psi'_1)  \delta (\psi'_1- \psi'_2)  \cdots \delta (\psi'_1- \psi'_n) 
		 \; .  
	\end{align*}
Note that the characteristic function $ \chi_{S^1}(\psi'_1)  $ can be removed without alliterating the result. 
Moreover, note that symmetrisation is not required, as the expression is already symmetric. 
We have again used \eqref{crucial} and the fact that when integrating out the $\theta'$ variable, 
all the terms containing $ \cos k \theta'$ will not contribute. 

Thus
	\begin{align}
	\label{ce}
		& \bigl( \;  \underbrace{ e(S^1) \otimes_s  \cdots \otimes_s e(S^1) }_{n-times} 
		f \bigr) (\psi_1, \ldots, \psi_n  )   \nonumber \\
	& \qquad =  \kappa_0^n \chi_{S^1}(\psi_1)  \underbrace{ \delta (\psi_1- \psi_2)  \cdots \delta (\psi_1- \psi_n) }_{(n-1)-times}  \in 
		\underbrace{\mathbb{H}^{-1}_{\upharpoonright S^1} (S^2) \otimes_s  
		\cdots \otimes_s \mathbb{H}^{-1}_{\upharpoonright S^1} (S^2)}_{n-times} \; . 
	\end{align}
This shows that \eqref{ce} is equal to 
	\begin{align*}
		\int_{S^1} r\,  {\rm d} \psi \; 
				{:}\bigl( \kappa_0 \Phi \bigl(  0, \psi)   \bigr)^n {:}_{C_0}  \Omega_E  \; , 
	\end{align*}
where $\kappa_0$ is the constant first appearing in Lemma \ref{kappa}. Note that normal ordering is with respect to the 
time-zero covariance $C_0$. Similar arguments now show that 
	\begin{align*}
		\Gamma \left(e({S^1}) \right) \bigl( V \left(\overline{S_+}\right) \bigr)^n \Omega_E = 
		V (S^1)^n \Omega_E \;,	
	\end{align*}
proofing the statement. 
\color{black}
\end{proof}

\chapter{The Reconstruction of Interacting Quantum Fields}
\label{interactingdesitter}

\setcounter{equation}{0}

We start by introducing an auxiliary Hilbert space ${\mathcal H}_V$, associated to the interacting measure 
${\rm d} \mu_V$ on the sphere $S^2$. 

\section{The Hilbert space for the interacting measure}
\label{sec:6.1}

Replace the Gaussian measure ${\rm d} \Phi_C$ by the interacting measure ${\rm d} \mu_V$ 
first introduced in \eqref{pmeasure},
and repeat the Osterwalder-Schrader reconstruction: 

\begin{itemize}
\item[$i.)$] Since ${\rm d} \mu_V$ is absolutely continuous with respect to the Gaussian measure~${\rm d} \Phi_C$, reflection 
positivity extends to $L^2 ({\mathcal Q}, \Sigma_{\overline { S_+}}, {\rm d} \mu_V)$:
\label{interactingsphereLP}
	\[
		\int_{{\mathcal Q}}  {\rm d} \mu_V \; \overline{\Theta (F)} F  \ge 0 \; , \qquad 
		F \in L^2 ({\mathcal Q}, \Sigma_{\overline 		{ S_+}}, {\rm d} \mu_V) \; . 
	\]
Define
\label{inthilbertpage}
	\begin{equation}
		\label{phs}
 		{\mathcal H}_V 
		=\hbox{ completion of } L^{2}({\mathcal Q}, \Sigma_{\overline { S_+}}, {\rm d} \mu_V) /{\mathcal N}_V \; ,
	\end{equation}
where ${\mathcal N}_V \subset L^2 ({\mathcal Q}, \Sigma_{\overline { S_+}},  {\rm d}\mu_V)$ 
is the kernel of the positive quadratic form
	\begin{equation}
		\label{osspv}
		\langle F,G \rangle_V =\int_{{\mathcal Q}} {\rm d} \mu_V \;  \overline{ \Theta (F) }G \; , 
		\qquad F, G \in L^{2}({\mathcal Q}, \Sigma_{\overline{ S^+}}, {\rm d}\mu_V) \; ,
	\end{equation} 
and the completion in (\ref{phs}) is with respect to the scalar product induced by~(\ref{osspv}) 
on the quotient space
$L^{2}({\mathcal Q}, \Sigma_{\overline { S_+}}, {\rm d} \mu_V) /{\mathcal N}_V$. Let ${\mathcal V}_V$ denote 
the canonical map 
	\[ 
		{\mathcal V}_V \colon  L^{2}({\mathcal Q}, \Sigma_{\overline { S_+}}, {\rm d} \mu_V) 
		\to L^{2}({\mathcal Q}, \Sigma_{\overline { S_+}}, {\rm d} \mu_V)/{\mathcal N}_V
	\]
and let
	\begin{equation}
		\label{dpv}
		\Omega_V \doteq{\mathcal V}_V (1)
	\end{equation}
denote the distinguished unit vector in ${\mathcal H}_V$. 
\item[$ii.)$]
For $A\in L^{\infty}({\mathcal Q}, \Sigma_{S^1}, {\rm d} \mu_V)$ define
$A^{V} \in {\mathcal B}({\mathcal H}_V)$ by
	\begin{equation}
	\label{qqqq}
	A^{V}{\mathcal V}_V (F) \doteq {\mathcal V}_V (AF) \; , 
		\qquad F \in L^2 ({\mathcal Q}, \Sigma_{\overline { S_+}}, {\rm d}\mu_V) \; .
	\end{equation}
Denote by ${\mathcal U}_V (S^1) $ the abelian von Neumann algebra 
	\[
		{\mathcal U}_V (S^1) \doteq \left\{A^{V}  
		\in {\mathcal B}({\mathcal H}_V) \mid A\in L^{\infty}({\mathcal Q}, 
		\Sigma_{S^1}, {\rm d} \mu_V) \right\} \; . 
	 \]
The vector $ \Omega_V$ is cyclic for ${\mathcal U}_V (S^1)$ (see again Theorem 11.2 of \cite{KL1}).
\end{itemize}

\goodbreak

\begin{proposition}
\label{gspint}
$({\mathcal Q}, \Sigma, \Sigma^{(\alpha)},{\rm U}^{(\alpha)} ( \,.\, ), \Theta, {\rm d}\mu_{V})$ 
is a $(2 \pi)$-periodic, OS-positive  generalised path space in the sense of Definition \ref{gps}.
\end{proposition}

For $ 0\leq \theta\leq \pi$ fixed, set  
	\[
	 {\mathcal M}^{(\alpha)}_{\theta, V}\doteq L^{2} \Bigl({\mathcal Q}, \bigvee_{\theta'' \in [0, \pi -\theta]}
	 {\rm U}^{(\alpha)} (\theta'' )\Sigma^{(\alpha)} , {\rm d} \mu_V \Bigr)\; , \quad 0\leq \theta\leq \pi \; .
	 \]

\begin{definition}
Let ${\mathscr D}_{\theta, V}^{(\alpha)}\doteq{\mathcal V}_V \bigl({\mathcal M}^{(\alpha)}_{\theta, V}\bigr)$ and define, 
for $0\leq \theta'\leq \theta$,
	\begin{align*}
		P^{(\alpha)}_V(\theta') \colon {\mathscr D}_{\theta, V}^{(\alpha)} &\to  {\mathcal H}_V
		\nonumber \\
		{\mathcal V}_V (F)
			&\mapsto  {\mathcal V}_V \bigl( {\rm U}^{(\alpha)}(\theta')F) \; , 
			\quad F \in {\mathcal M}^{(\alpha)}_{\theta, V} \; .
	\end{align*}
\end{definition}

\begin{proposition} 
\label{3.20}
$\bigl(P^{(\alpha)}_V (\theta) , {\mathscr D}^{(\alpha)}_{\theta, V} \bigr)$ is a 
local symmetric semigroup on $  {\mathcal H}_V$. 
Its generator  $L_V^{(\alpha)}$  satisfies
\label{intlssgpage}
	\[
	P^{(\alpha)}_V(\theta')\Psi= {\rm e}^{-\theta' L^{(\alpha)}_V}\Psi \; ,   
	\qquad \Psi \in {\mathscr D}^{(\alpha)}_{\theta, V} \; , \quad 0\leq \theta'\leq \theta \; . 
	\]
\end{proposition}

Define $J^{(\alpha)}_{V}$ as the unique extension of the anti-linear operator  
\label{TTintpage}
	\begin{equation}  
	\label{wichtig3}
	J^{(\alpha)}_{V} {\mathcal V}_V (F) \doteq
	{\mathcal V}_V \bigl(\Theta^{(\alpha)}_{\pi/2} \overline{F} \bigr) \; , 
	\qquad F \in L^2( {\mathcal Q}, \Sigma_{\overline{S_+}}, {\rm d} \mu_V) \; ,
	\end{equation}
where $\Theta^{(\alpha)}_{\pi/2}$ was introduced in (\ref{thetapi}). 
Furthermore, define the von Neumann algebra
	\begin{equation}
		{\mathcal R}_V   (I_\alpha ) \doteq \bigvee_{t \in {\mathbb R}} 
		\left( {\rm e}^{-it  L_V^{(\alpha)} } {\mathcal U}_{V} (I_\alpha)  {\rm e}^{it  L_V^{(\alpha)} }\right) \; ,
	\end{equation}
where $ {\mathcal U}_V (I_\alpha ) \doteq \left\{A^{V} \in {\mathcal B}({\mathcal H}_V) \mid 
A\in L^{\infty}({\mathcal Q}, \Sigma_{I_\alpha},
{\rm d} \mu_V) \right\} $.  

\begin{theorem} 
\label{ToTa11}
The positive operator 
${\rm e}^{- \pi  L^{(\alpha)}_V} $ is the Tomita-Takesaki modular operator for the pair
$\bigl({\mathcal R}_V  (I_\alpha ), \Omega_V \bigr)$,  and
$J^{(\alpha)}_{V}$ is the modular conjugation for the pair 
$\bigl({\mathcal R}_V  (I_\alpha ), \Omega_V \bigr)$, \emph{i.e.}, 
	\[
	 J^{(\alpha)}_V {\rm e}^{- \pi  L^{(\alpha)}_V} A \Omega_V = A^* \Omega_V \; ,
	\qquad  A \in {\mathcal R}_V (I_\alpha ) \; .
	\]
\end{theorem}

The rotations can be implemented easily; however, we will have to show later on that the boosts and the rotations 
\emph{together} generate a representation of $SO_0 (1,2)$. 

\begin{proposition} 
\label{knull-V-prop}
Consider the map (\ref{qqqq}). It follows that
\label{knull-V-proppage}	
	\begin{equation}
		\label{knullos}
		{\rm e}^{i \alpha K_0^{V}} A^{V} \Omega_{V} 
		\doteq {\mathcal V}_V \bigl( {\rm U} (R_{0} (\alpha)) A \bigr) \; ,
		\qquad A \in L^{\infty}({\mathcal Q}, \Sigma_{S^1}, {\rm d} \mu_V) \;  , 
	\end{equation}
extends to a strongly continuous unitary representation of the rotation group $SO(2)$ 
on the Hilbert space ${\mathcal H}_{V}$. 
\end{proposition}

\section{Virtual representations}
\label{virtual}
The basic object in the theory of \emph{virtual representations}\index{virtual representation}\footnote{See \cite{JO} for
recent work and an extensive list of references on this topic.} \cite{FOS} is a symmetric space\index{symmetric space}:  
let~$G$ be a Lie group\index{Lie group}, 
$K$ a closed subgroup  of~$G$, with Lie algebras\index{Lie algebra} ${\mathfrak g}$ 
and ${\mathfrak k}$, respectively.
A Lie algebra~${\mathfrak g}$ is  {\em symmetric} (giving rise to a symmetric space 
$(G, K , {\mathfrak T})$), if it  allows a decomposition
\label{virtreppage}
	\begin{equation}
		\label{VR1}
		{\mathfrak g} = {\mathfrak k} \oplus {\mathfrak m}
	\end{equation}
(where $\oplus$ indicates a direct sum of vector spaces) such that
	\begin{equation}
		\label{VR2}
		[{\mathfrak k}, {\mathfrak k}] \subset {\mathfrak k} \; , \qquad [{\mathfrak k}, {\mathfrak m}] 
			\subset {\mathfrak m}\; , \qquad [{\mathfrak m}, {\mathfrak m}] 
			= {\mathfrak k} \; .
	\end{equation}
On a symmetric Lie algebra there exists (use  \eqref{VR1}--\eqref{VR2} to derive this fact)  a natural 
involutive automorphism ${\mathfrak T}$ of ${\mathfrak g}$, such that
	\[
		{\mathfrak T}_{| {\mathfrak k} }= \mathbb{1} \; , 
		\qquad {\mathfrak T}_{|{\mathfrak m}} = - \mathbb{1} \; . 
	\]
The {\em dual symmetric Lie algebra} ${\mathfrak g}^*$ is (see, \emph{e.g.}, \cite{KN})
	\begin{equation}
		\label{VR3}
		{\mathfrak g}^* = {\mathfrak k} \oplus i{\mathfrak m}.
	\end{equation}
\eqref{VR2} implies that  ${\mathfrak g}^*$ is  the {\em real} Lie algebra of a  
simple connected Lie group~$G^*$. 

\begin{lemma} $so(3)$ is a symmetric Lie algebra with sub-Lie algebra $SO(2)$; 
the dual symmetric Lie algebra is $so(1,2)$. 
\end{lemma}

\begin{proof}
Decompose $so(3)$ according to \eqref{VR1} and verify \eqref{VR2}.
\end{proof}

\begin{definition}
\label{defvirtrep}  
A {\em virtual representation} $( \wp, {\mathfrak H})$ of a symmetric space $(G, K , {\mathfrak T})$
consists of a separable Hilbert space ${\mathfrak H} $ together with a local group homomorphism $\wp$ 
from $G$ into linear operators densely defined on ${\mathfrak H}$, with the following properties: 
\begin{itemize}
\item [$ i.)$] 
$\wp_{|   H}$ is a continuous unitary representation of $K$ on ${\mathfrak H}$; 
\item [$ ii.)$]
there exists a neighbourhood $N$ of $\mathbb{1} \in G$, invariant under right translation by $K$, and
a linear subspace ${\mathfrak D}$, dense in ${\mathfrak H}$, 
such that 
\begin{itemize}
\item[---] ${\mathfrak D} \subset {\mathscr D}(\wp(g))$ for all $g \in N$;  and 
\item[---] if $g_1, g_2$ and $g_1 \circ g_2$ are all in $N$, then
\label{virtreppage2}
	\begin{equation} 
		\label{49th}
		\wp (g_2) \Psi \in {\mathscr D}(\wp(g_1))\; , \qquad \Psi \in {\mathfrak D}\; ,  
	\end{equation}
	and
	\[ 
		\wp(g_1) \wp(g_2) \Psi = \wp(g_1 \circ g_2) \Psi \; , \qquad \Psi \in {\mathfrak D}\; ; 
	\]
\end{itemize}
\item[$ iii.)$] if $\ell \in {\mathfrak m}$, $0 \le t \le 1$, and
	\[
		{\rm e}^{- t \ell} \in N  \; ,  \qquad 0 \le t \le 1 \; , 
	\]
then $ \wp ({\rm e}^{-t \ell })$ is a hermitian operator defined on ${\mathfrak D}$ and 
	\begin{equation}
	\label{52th}
		s-\lim_{t \to 0} \wp ({\rm e}^{-t \ell}) \Psi = \Psi \; , \qquad \Psi \in {\mathfrak D} \; . 
	\end{equation}
\end{itemize}
\end{definition}

\noindent 
The main result in the theory of virtual representations is the following:

\begin{theorem}[Fr\"ohlich, Osterwalder, and Seiler \cite{FOS}] 
\label{FOS}
Let $\wp$ be a virtual representation of 
a symmetric space $(G, K, {\mathfrak T})$, with $K$ compact. 
Then $\wp$ can be analytically continued to a unitary representation $\wp^*$ of $G^*$. 
\end{theorem} 

\section{A unitary representation of the Lorentz group}
\label{sec5.8}

We now apply the main result of  the previous subsection to the special case relevant in the present context. 

\label{umospage}
\begin{theorem} 
\label{uo}
The self-adjoint operators $K^V_0 $, $L^V_1 \doteq L^{(0)}_V$ and 
$L^V_2\doteq L^{(\pi/2)}_V$ defined in Proposition~\ref{knull-V-prop} and
Proposition  \ref{3.20}, respectively,  generate a unitary representation 
$\Lambda \mapsto U_V(\Lambda)$ of $SO_0(1,2)$ on~${\mathcal H}_V$.
\end{theorem}

\begin{proof} It is sufficient to construct  a virtual representation $\wp$ of  $SO(3)$ on~${\mathcal H}_V$ and 
then apply Theorem \ref{FOS}.

\begin{itemize}
\item[$ i.)$]
The self-adjoint operator $K_0^V$ generates a unitary representation of  the subgroup $K = SO(2)$ 
of  $G = SO(3)$; 
\item[$ ii.)$] For~$0 <  \theta < \pi/2$ let $N_{\theta}$ 
be the subset of elements in $SO(3)$ consisting of the rotations which (when acting on it) do not move the north pole $(r, 0, 0)$ outside of 
the polar cap 
	\[
		\left\{ 
		\left( \begin{smallmatrix} r \cos \theta' \\  r \cos \alpha \sin \theta' \\ 
		r \cos \alpha \sin \theta' \end{smallmatrix} \right) 
		\mid 0 \le \theta' < \theta, \alpha \in [0, 2 \pi) \right\}  \; . 
	\] 
Now recall \eqref{malphatheta} and set, for~$0 <  \theta < \pi/2$,
	\[
		{\mathscr D}_{\theta, V}={\mathcal V}_V ({\mathfrak M}_{\theta} ) \quad
			\hbox{with} \quad
		{\mathfrak M}_{\theta}= \bigcap_{\alpha \in [0, 2 \pi)} {\mathcal M}^{(\alpha)}_{\theta} \; . 
	\]
In other words, given the polar cap $\frown_\theta$ with distance $\theta$ to the equator, 
the set~${\mathfrak M}_{\theta}$ is the set of all $ F \in L^{2}({\mathcal Q}, \Sigma_{\frown_\theta}, {\rm d} \mu_V ) $, 
with $ \Sigma_{\frown_\theta}$ the smallest $\sigma$-algebra for which the functions 
$\{ \Phi(f) \mid f \in {\mathscr D}_{{\mathbb R}} ( \frown_\theta) \} $ are measureable.
In Lemma \ref{qd}
we will show that $ {\mathscr D}_{\theta, V} $ is dense in ${\mathcal H}_{V}$. 

\smallskip
It follows from the definitions of ${\mathfrak T}$ and $\Theta$ that for $R \in SO(3)$
	\[
		{\rm U}(R) (\Theta F)  = \Theta {\rm U}({\mathfrak T}(R)) F  \; , 
		\qquad F \in C^\infty (S^2) \; .
	\]
Clearly, for $R \in N_\theta$ and $F, G \in {\mathfrak M}_{\theta}$
	\begin{align*}
		\langle G, {\rm U}(R) F \rangle_{V} 
		&=  \int_{{\mathcal Q}} {\rm d} \mu_V \;  \overline{ \Theta (G) } \; {\rm U}(R) F  \nonumber \\ 
		&=  \int_{{\mathcal Q}} {\rm d} \mu_V \;  \overline{ {\rm U}(R^{-1})\Theta (G) } \,  F  \nonumber \\ 
		&=  \int_{{\mathcal Q}} {\rm d} \mu_V \;  \overline{ \Theta ( {\rm U}( {\mathfrak T}(R^{-1}) G) } \,  F \; .		\end{align*} 
Now use the Schwarz inequality to show that the intersection of the kernel~${\mathcal N}_V$ of~${\mathcal V}_V$
with~${\mathscr D}_{\theta, V}$ is invariant under the map $F \mapsto {\rm U}(R) F$ for $R \in N_\theta$, $\theta > 0$. 
Thus, for each $R \in N_\theta  \subset SO(3)$, the map 
	\begin{align*}
	\wp (R) \colon \qquad {\mathscr D}_{\theta, V} \quad & \to  \quad {\mathcal H}_V
		\nonumber \\
	 {\mathcal V}_V (F) & \mapsto   {\mathcal V}_V ( {\rm U}(R) F) \; , 
	 \qquad   F \in {\mathfrak M}_{\theta}  \; , 
	\end{align*}
is well-defined, verifying \eqref{49th}. 

\smallskip
Now, if $F \in {\mathfrak M}_{\theta}$ and $R_1 , R_2 \in N_{\theta}$ 
as well as $R_1 R_2 \in N_{\theta} $, then 
	\[
	{\mathcal V}_V \bigl( {\rm U}(R_1){\rm U}(R_2) F \bigr) = 
	{\mathcal V}_V \bigl( {\rm U}(R_1R_2) F \bigr)\; . 
	\]
For example, for $\gamma$ and $\theta$ sufficiently small and $F \in {\mathfrak M}_{\theta}$
	\[
		{\rm e}^{i \gamma K_0^{V}}   {\rm e}^{-\theta L^{(\alpha)}_{V}} {\mathcal V}_V F  
		=  
			{\mathcal V}_V \bigl( {\rm U} \bigl( R_{0} (\gamma)  R^{(\alpha)} (\theta) \bigr) F \bigr)\;   . 
	\]
\item[$ iii.)$]
The group   $R \mapsto {\rm U}(R)$, $R \in SO(3)$, acts continuously on  $L^2 ( {\mathcal Q}, \Sigma, {\rm d} \Phi_C)$ and ${\mathcal V}_V$ 
is continuous on~$ {\mathcal M}_\theta$. Thus the vector valued function
$N_\theta \ni R \mapsto \wp (R) \Psi$ is continuous for each $\Psi \in {\mathscr D}_{\theta, V}$.
Thus $\wp (R)  \to \mathbb{1}$ as $R  \to \mathbb{1}$, verifying \eqref{52th}.
\end{itemize}
It follows that $\wp$ is a virtual representation of $SO(3)$ on~${\mathcal H}_{V}$.
\end{proof}

\begin{remark} The theory of virtual representations provides a general argument, which also applies in the free case. 
However, for the free field, we have already verified---by directly computing the Lie brackets on the eigenfunctions of 
$K_0$ in the proof of Proposition \ref{UIR-S1}---that 
the generators  $L^{(\alpha)} 
= {\rm d} \Gamma ( \omega \cos_{\psi + \alpha}) $, $\alpha \in [0, 2 \pi)$, of the boosts together with the generator $K_0$ of the rotations
generate a representation of $SO_0(1,2)$ on the Fock space $\Gamma \bigl(\widehat{\mathfrak h}(S^1) \bigr)$. 
In the general case, such a direct proof is 
not available.  
\end{remark}

\begin{lemma} 
\label{qd}
Let $0< \theta < \pi/2$. 
Then the set ${\mathfrak M}_\theta$ is a quantization 
domain\index{quantization domain}~\cite{HJJ}; \emph{i.e.}, the set $ {\mathscr D}_{\theta, V} $  is dense in ${\mathcal H}_{V}$. 
\end{lemma}

\begin{remark} 
In fact, if $O$ is any open set in $S_+$, then the image of 
	\[
	{\mathfrak M} (O) \doteq L^{2}({\mathcal Q}, \Sigma_{O}, {\rm d} \mu_V) 
	\]
under ${\mathcal V}_V$ is dense in ${\mathcal H}_{V}$. Here $ \Sigma_{O}$ is the smallest 
$\sigma$-algebra for which the functions 
$\{ \Phi(f) \mid f \in {\mathcal D}_{{\mathbb R}} ( O) \} $ are measureable.
\end{remark}

\begin{proof} In order to show that ${\mathscr D}_{\theta, V}$ is dense in  ${\mathcal H}_{V}$, 
it is sufficient to show that 
if  $\Psi  \perp {\mathscr D}_{\theta, V}$ is a vector in the orthogonal complement of ${\mathscr D}_{\theta, V}
\subset {\mathcal H}_{V}$, then it equals the zero-vector.
We have already seen that ${\mathcal U}_{V} (S^1) \Omega_{V}$ is 
dense in ${\mathcal H}_{V}$.
Thus it is sufficient to show that	
	\[
		\langle \Psi , {\rm e}^{i\Phi_{V}(h)} \Omega_{V} \rangle_{V} = 0 
		\qquad \forall \, {\rm e}^{i\Phi_{V}(h)}  \in {\mathcal U}_{V} (S^1) \;  ,  
	\]
as this would imply that $\Psi $ is the zero-vector. Moreover,  
	\[
	{\mathcal U}_{V} (S^1) = \bigvee_{ I } \,  {\mathcal U}_{V} ( I)
	\]
with $\cup I$ a covering of $S^1$ in terms of open intervals. Thus it is sufficient to show that
	\[
		\langle \Phi , {\rm e}^{i\Phi_{V}(h)} \Omega_{V} \rangle_{V}  = 0 
		\qquad \forall \, {\rm e}^{i\Phi_{V}(h)}  \in {\mathcal U}_{V} ( I) \;  , 
	\]
with  $I$ an arbitrary fixed interval in the covering of $S^1$. 
For the covering we choose sufficiently many circle segments 
	\[
	I_{\alpha, \theta + \epsilon} 
	= \{ \vec x \in I_\alpha \mid {\rm dist} (\vec x , \partial I_\alpha ) > \theta +\epsilon \} \; ,
	\qquad 0 < \epsilon \ll \theta \; , 
	\]
of equal size, consisting of points in the interior of the half-circle $I_\alpha$, 
which are more than $\theta + \epsilon$ away from the end points of the half-circle. 	

Now consider, for $h \in {\mathcal D}_{\mathbb R} (I_{\alpha, \theta + \epsilon} )$ fixed,  the analytic function
	\begin{equation}
		\label{f24} 
		z \mapsto 
		\langle \Psi , {\rm e}^{-z L^{(\alpha)}_{V}} {\rm e}^{i\Phi_{V}(h)} \Omega_{V} \rangle_{V}  \;  , 
		\qquad \{ z \in {\mathbb C} \mid 0 < \Re z  < \pi\} \;  . 	
		\end{equation} 
By construction there exists an open interval $J$ (whose size depends on $\epsilon$) such that 
	\[
	R^{(\alpha)} (\theta') I_{\alpha, \theta + \epsilon} \subset \; \frown_\theta \; ,
	\qquad 
	\theta' \in J \; , 
	\]
and consequently
	\[
	 {\rm e}^{-\Re z L^{(\alpha)}_{V}} {\rm e}^{i\Phi_{V}(h)} \Omega_{V} \in {\mathscr D}_{\theta, V} \; ,
	 \qquad 
	\Re z  \in J \; , 
	\] 
and, since $\Psi  \perp {\mathscr D}_{\theta, V}$,
the analytic function \eqref{f24} vanishes on an open line segment in 
the interior of its domain, and is therefore identical zero. 
\end{proof}

\section{Perturbation theory of generalised path spaces}
\label{FKN}

In the previous section we have constructed a representation of $SO_0(1,2)$ on a new Hilbert space~${\mathcal H}_{V}$.
In this section, we discuss the action of the interacting field on the Fock space of the free field.  
To do so, construct (see \cite{KL1}) a generalised path space $({\mathcal Q}, \Sigma, \Sigma^{(\alpha)}, 
{\rm U}^{(\alpha)}_{\rm int} (\theta), 
\Theta, {\rm d} \Phi_C)$, which is equivalent (see~\eqref{equipathspace}) 
to the one given in Proposition \ref{gspint}. 

Set
	\begin{equation}
	\label{fkn17}
	V^{(\alpha)} \doteq V_0 (\cos_{\psi + \alpha} \chi_{I_\alpha} ) \; , 
	\end{equation}
where $V_0$ was defined in~\eqref{bbv-interaction} and $\chi_{I_\alpha}$ denotes the characteristic function of the half-circle $I_\alpha \subset S^1$. It follows from 
Lemma \ref{wickooo} and \cite[Lemma~3.15]{SHK}, respectively, that
	\[
	V^{(\alpha)} \in L^{1}({\mathcal Q},\Sigma^{(\alpha)}, {\rm d} \Phi_C) \quad \text{and} \quad 
	{\rm e}^{- 2 \pi V^{(\alpha)}}\in L^{1}({\mathcal Q}, \Sigma^{(\alpha)}, {\rm d} \Phi_C) \; . 
	\]
Hence  the Feyman-Kac-Nelson kernels $\bigl\{F^{(\alpha)}_{[0,\theta']}\bigr\}_{0\leq \theta'\leq \pi}$,   
	\[
		F^{(\alpha)}_{[0,\theta']}
		={\rm e}^{-\int_{0}^{\theta'}{\rm d} \theta \; {\rm U}^{(\alpha)}(\theta)V^{(\alpha)}},
		\quad 0\leq \theta'\leq \pi \;  ,
	\]
belong to $L^{2}({\mathcal Q}, \Sigma_{\overline{S_+}}, {\rm d} \Phi_C)$. 

\bigskip
Next consider the sets  
	\[
		{\mathcal M}^{(\alpha)}_{\theta, {\rm int}} \doteq \hbox{ linear span of }  
		\bigcup_{0\leq \theta''\leq \pi-\theta}F^{(\alpha)}_{[0,\theta'']}
		\, L^{\infty} \left({\mathcal Q}, \Sigma^{(\alpha)}_{[0, \pi -\theta]}, {\rm d} \Phi_C \right) \; . 
	\]
Here $\Sigma^{(\alpha)}_{[0, \pi -\theta]}\doteq \bigvee_{\theta' \in [0, \pi -\theta]} {\rm U}^{(\alpha)} (\theta' )\Sigma^{(\alpha)}$. 

\begin{definition}
Set, for $0\leq \theta\leq \pi$,  
	\begin{align*}
				 {\rm U}^{(\alpha)}_{\rm int}(\theta')\colon   
				 {\mathfrak M}^{(\alpha)}_{\theta,  {\rm int}} & \to  L^{2} 
				 \Bigl({\mathcal Q},\Sigma^{(\alpha)}_{[0, \pi -\theta- \theta']}, 
			 	 {\rm d} \Phi_C \Bigr) \nonumber \\ 
			  G & \mapsto  F^{(\alpha)}_{[0,\theta']}{\rm U}^{(\alpha)} (\theta')G \; . 
	\end{align*}
The map $\theta' \mapsto {\rm U}^{(\alpha)}_{\rm int}(\theta')$ defines the interacting rotations on the sphere.
\end{definition}

The one-parameter group $\theta' \mapsto {\rm U}^{(\alpha)}_{\rm int}(\theta')$ induces 
a local symmetric semigroup on~${\mathcal H}$: 

\begin{proposition}
\label{hos-section}
\label{spaces}
Let ${\mathscr D}_{\theta, {\rm int}}^{(\alpha)}\doteq{\mathcal V} \bigl({\mathcal M}^{(\alpha)}_{\theta, {\rm int}} \bigr)$ 
and set, for  $0\leq \theta'\leq \theta $, 
	\begin{align*}
				P^{(\alpha)}_{\rm int}(\theta')\colon  {\mathscr D}_{\theta, {\rm int}}^{(\alpha)} & \to 
				 {\mathcal H} 
			 \\
				{\mathcal V} (G) & \mapsto  {\mathcal V}\bigl({\rm U}_{\rm int}^{(\alpha)}(\theta')G \bigr)   \; . 
	\end{align*}
$\bigl(P_{\rm int}^{(\alpha)} (\theta), {\mathscr D}_{\theta, {\rm int}}^{(\alpha)} \bigr)$ is
a local symmetric semigroup on~${\mathcal H}$ with generator  $H^{(\alpha)}_{\rm os} $.
\end{proposition}


\part{Interacting Quantum Fields}

\chapter{The ${\mathscr P}( \varphi)_2$ Model on the de Sitter Space}

\section{Identification of Hilbert spaces}
\label{sec6.2}
The Radon-Nikodym theorem\index{Radon-Nikodym theorem} implies that the interacting measure
$ {\rm d} \mu_V$ is absolutely continuous with respect to~the Gaussian measure 
${\rm d} \Phi_C$. Consequently, 
	\begin{equation}
		\label{jss} L^{\infty}({\mathcal Q}, \Sigma_{S^1} , {\rm d} \mu_V) \cong  
		L^{\infty}({\mathcal Q}, \Sigma_{S^1}, {\rm d} \Phi_C) \; . 
	\end{equation}
Moreover, 
${\mathcal U}_V (S^1) \Omega_V$ is dense in ${\mathcal H}_V$ 
and ${\mathcal U} (S^1) \Omega$
is dense in ${\mathcal H}$. As we will see next, ${\mathcal U} (S^1) \Omega_{\rm int}$
is dense in ${\mathcal H}$ as well.
The vectors $\Omega_V$ and
$ \Omega_{\rm int} $ were defined in~(\ref{dpv}) and (\ref{intvacuum}), respectively. 
Note that 
\begin{equation}
	\label{Vhalbkugel} 
	{\rm e}^{- \pi H^{(\alpha)}} \Omega 
	= {\mathcal V} \bigl( F^{(\alpha)}_{[0, \pi]} {\rm U}^{(\alpha)} (\pi) 1 \bigr) 
	= {\mathcal V} \bigl( F^{(\alpha)}_{[0, \pi]}  \bigr) 
	= {\mathcal V}  \bigl ({\rm e}^{- V  (\overline{S_+} )} \bigr) \; .
\end{equation}
The first equality follows from the reconstruction theorem, but it can also be verified directly using the Trotter product\index{Trotter
product formula} formula:
\begin{align*}
& {\mathcal V}  \Bigl( {\rm e}^{- \int_0^\pi U (\theta) V_0 (\cos_\psi) {\rm d} \theta }  \Bigr) \\
& \qquad = \lim_{n \to \infty}
{\mathcal V} \Bigl( {\rm e}^{- \sum_{k=1}^n U (k \pi / n ) V_0 (\cos_\psi)  }  \Bigr) \\
& \qquad =  \lim_{n \to \infty}
{\mathcal V} \Bigl( 
\underbrace{ 
{\rm e}^{- U (\pi/n ) V_0 (\cos_\psi)  } \cdots {\rm e}^{- U (\pi  ) V_0 (\cos_\psi)  } }_{n \  terms}  \Bigr) \\
& \qquad = s\mbox{-}\lim_{n \to \infty}
\Bigl(  
\underbrace{ {\rm e}^{- \frac{\pi}{n} {\rm d} \Gamma (\omega r) \cos_\psi    } {\rm e}^{- V_0 (\cos_\psi)  } 
\cdots {\rm e}^{- \frac{\pi}{n} {\rm d} \Gamma (\omega r) \cos_\psi    }
{\rm e}^{- V_0 (\cos_\psi)  }  }_{n \  terms} \Bigr) \Omega \\
& \qquad = 
{\rm e}^{- \pi H^{(0)}} \Omega \; . 
\end{align*}
In fact, the  identity \eqref{Vhalbkugel} holds for all $\alpha \in [0, 2\pi)$, \emph{i.e.}, 
	\[
		\frac{{\rm e}^{-\pi H^{(\alpha)}}\Omega }{ 
		\|{\rm e}^{-\pi H^{(\alpha)}}\Omega \| } = \Omega_{\rm int} \; , 
		\qquad   \alpha \in [ \, 0, 2\pi ) \; .  
	\]
The vector $\Omega_{\rm int}$ was defined in Equ.~\eqref{intvacuum}.
\color{black}

\begin{theorem} The vector $\Omega_{\rm int}$ is cyclic and separating for ${\mathcal R}(I_\alpha)$, $\alpha \in [0, 2 \pi)$.
\end{theorem}

\begin{proof}
Since $F^{(\alpha)}_{[0,\pi]}$ is in $L^2({\mathcal Q}, \Sigma_{S^1}, {\rm d} \Phi_C)$, it follows that for
$0\le \theta_1 \le \ldots \le  \theta_n \le  \pi$ the vectors 
	\[
		{\rm U}^{(\alpha)}_{\rm int}(\theta_n) A_n 
		{\rm U}^{(\alpha)}_{\rm int}(\theta_{n-1}- \theta_n) A_{n-1}
			\cdots  {\rm U}^{(\alpha)}_{\rm int} (\theta_1-\theta_2)  
			A_1 F^{(\alpha)}_{[0, \pi -\theta_1]}   \; , 
	\]
with $A_1, \ldots, A_n \in L^{\infty}({\mathcal Q}, \Sigma_{I_\alpha}, {\rm d} \Phi_C)$, are in 
$L^2({\mathcal Q}, \Sigma_{[0, \pi]}, {\rm d} \Phi_C)$. 
Since 
	\[
	F^{(\alpha)}_{[0,\pi]} > 0 \qquad \mu-{\rm a.e.} \; , 
	\]
they form a total set in $L^2 (Q, \Sigma_{[0, \pi]}, {\rm d}\Phi_C)$. Therefore the vectors 
	\begin{align}
		\label{26nn}
		&{\rm e}^{-\theta_{n} H^{(\alpha)} } A_n^{\rm os} 
		{\rm e}^{- (\theta_{n-1}- \theta_n) H^{(\alpha)} } A_{n-1}^{\rm os} 
		\cdots 
		{\rm e}^{-(\theta_{1} - \theta_2) H^{(\alpha)} } A_1^{\rm os} 
		{\rm e}^{-(\pi - \theta_{1}) H^{(\alpha)} }\Omega 
					\nonumber \\ 
		&\qquad \qquad 
		= {\mathcal V} \bigl( {\rm U}^{(\alpha)}_{\rm int}(\theta_n) A_n 
		{\rm U}^{(\alpha)}_{\rm int}(\theta_{n-1}- \theta_n) A_{n-1}
			\cdots 
			\nonumber \\
			& \qquad \qquad \qquad \qquad \qquad \qquad  
			\qquad \cdots {\rm U}^{(\alpha)}_{\rm int} (\theta_1-\theta_2)  
			A_1 F^{(\alpha)}_{[0, \pi -\theta_1]}  \bigr)  
	\end{align}
form a total set in ${\mathcal H}$. By analyticity it follows that the vectors 
	\[
	A_n^{\rm os} (t_n)  
		\cdots 
		A_1^{\rm os} (t_1)
		{\rm e}^{- \pi H^{(\alpha)} }\Omega 
	\]
with $A_i^{\rm os} (t) = {\rm e}^{ i t H^{(\alpha)} } 
A_i^{\rm os} {\rm e}^{- i t H^{(\alpha)} } $, $t \in {\mathbb R}$, 
form a total set in ${\mathcal H}$. 
\end{proof}

We next show that ${\mathcal U} (S^1) \Omega_{\rm int}$
is dense in ${\mathcal H}$ \cite[15.4 Remark]{KL1}.

\begin{theorem} The vector $\Omega_{\rm int}$ is cyclic and separating for ${\mathcal U}(S^1)$.
\end{theorem}

\begin{proof} Recall from \eqref{intvacuum}  that 
	\[
	\Omega_{\rm int} = 	\frac{ {\mathcal V}  \bigl({\rm e}^{- V \left(\overline{S_+}\right)} \bigr) }
		{ \left\|{\mathcal V}  \bigl({\rm e}^{- V \left(\overline{S_+}\right)} \bigr) \right\| }  \; . 
	\]
Since $\Omega$ is cyclic for ${\mathcal U}(S^1)$, and  ${\mathcal U}(S^1)$ is abelian, the result follows
from the fact that ${\mathcal V}  \bigl({\rm e}^{- V \left(\overline{S_+}\right)} \bigr)$ 
is affiliated to  ${\mathcal U}(S^1)$ and strictly positive.
\end{proof}

\begin{proposition}
\label{prop:12} 
The map 
\label{qtotot-page}
	\begin{equation}
		\label{qtotot}
		A^{V}  \Omega_V 
			\mapsto A^{\rm os} \Omega_{\rm int} \; , 
			\qquad A \in L^{\infty}({\mathcal Q}, \Sigma_{S^1} , {\rm d} \mu_V) \; ,
	\end{equation}
extends to a unitary operator ${\mathbb V} \colon {\mathcal H}_V \to {\mathcal H}$. 
\end{proposition}

\begin{proof} We have already seen that ${\mathcal U}_V (S^1) \Omega_V$ is dense in ${\mathcal H}_V$ 
and  ${\mathcal U} (S^1) \Omega_{\rm int}$ is dense in ${\mathcal H}$.
It remains to show that the map \eqref{qtotot} is norm preserving. This follows from 
	\begin{align*}
		\| A^{V}  \Omega_V \|^2 
		&= \int {\rm d} \mu_V |A|^2 =
		\frac{\int_{{\mathcal Q}}{\rm d}\Phi_{C} \; 
		{\rm e}^{- V (S^2) } \; |A|^2}{\int_{{\mathcal Q}}{\rm d}\Phi_{C} \; {\rm e}^{- V (S^2) }} 
		 \\
		&= 
		\| A^{\rm os} \Omega_{\rm int}\|^2 \; , \qquad A \in L^{\infty}({\mathcal Q}, 
		\Sigma_{S^1} , {\rm d} \mu_V) \; .
	\end{align*}
The second indenity follows from the definition of the interacting measure ${\rm d} \mu_V$, see~\eqref{pmeasure}.
\end{proof}

\begin{proposition}
\label{lm12}
The map  ${\mathbb V} \colon {\mathcal H}_V \to {\mathcal H}$ introduced in Proposition \ref{prop:12}
\begin{itemize} 
\item[$ i.)$] intertwines ${\mathcal U}_V (S^1)$ and ${\mathcal U} (S^1)$;
\item[$ ii.)$] respects the local structure, \emph{i.e.},  
	\[
	{\mathcal U} (I) = {\mathbb V} \, {\mathcal U}_V (I ) \, {\mathbb V}^{-1} \; , 
	\qquad I \subset S^1 \; ; 
	\]
\item[$ iii.)$] and, intertwines $K_0^V$ and $K_0$, \emph{i.e.},  ${\mathbb V}  K_0^V {\mathbb V} ^{-1}= K_0$.
\end{itemize}
\end{proposition}

\begin{theorem} 
\label{uovos}
Set 
\label{unitrospage}
	\begin{equation}
		\label{auchwichtig}
 		L_{\rm int}^{(\alpha)} \doteq {\mathbb V}  \, L_V^{(\alpha)} \, {\mathbb V} ^{-1} \; . 
	\end{equation}
It follows that the generators $K_0 $, $L^{\rm int}_1 \doteq   L_{\rm int}^{(0)}$ and $L^{\rm int}_2 
\doteq   L_{\rm int}^{(\pi/2)}$ 
generate a unitary representation $\widehat {U}_{\rm \, int}$ of $SO_0(1,2)$ 
on the Hilbert space ${\mathcal H}$.
\end{theorem}

We end this section by defining the interacting automorphisms. 

\begin{definition}
\label{def:6}
Given the unitary representation $\widehat {U}_{\rm \, int}$ of $SO_0(1,2)$ on the 
Hilbert space ${\mathcal H}$, we define the corresponding automorphisms: set 
\label{autintseinspage}
	\[
	\widehat \alpha_{\Lambda}^{\rm \, int} ( A) =
	\widehat {U}_{\rm \, int} (\Lambda) \, A \, \widehat {U}_{\rm \, int} (\Lambda)^{-1} \; , 
	\qquad
	A \in {\mathcal B} ({\mathcal H}) \; , \quad \Lambda \in SO_0(1,2) \; .  	
	\]
We call this group of automorphisms the {\em interacting dynamics} on the Cauchy surface $S^1$.
\end{definition}
 

\section{Perturbation theory for modular automorphisms}
\label{arpermodaut}
Next recall Araki's perturbation theory for modular automorphisms~\cite{A2}\cite{A3}, 
which has been generalised to unbounded perturbations by Derezinski, Jaksic and Pillet~\cite{DJP}.

\begin{theorem} 
\label{th5.2} Set
	\begin{equation}
		\label{halberkreis}
			V^{(\alpha)} 
			=  \int_{I_\alpha}  r {\rm d} \psi \; \cos (\psi + \alpha) \, 
			{:}{\mathscr P}(\Phi^{\rm os} (0, \psi)) {:}_{C_{0}}   \: .
	\end{equation}
It follows that
\begin{itemize}
\item[$ i.)$]
the operator sum $L^{(\alpha)} +V^{(\alpha)}$ is essentially self-adjoint on 
${\mathscr D} \bigl(L^{(\alpha)} \bigr) \cap {\mathscr D} \bigl(V^{(\alpha)} \bigr)$ and 
	\begin{equation}
	\label{Araki1}
			\overline{L^{(\alpha)}+V^{(\alpha)}} 
		= H^{(\alpha)} \; , \qquad \alpha \in [ \, 0, 2 \pi  )\; ;
	\end{equation}
\item[$ ii.)$]
the vector $\Omega $ defined in Equ.~(\ref{dv}) belongs to
${\mathscr D}\bigl({\rm e}^{-\pi H^{(\alpha)}}\bigr)$, $\alpha \in [ \, 0, 2 \pi )$, and 
	\begin{equation}
		\label{Araki}
		\frac{{\rm e}^{-\pi H^{(\alpha)}}\Omega }{ 
		\|{\rm e}^{-\pi H^{(\alpha)}}\Omega \| } = \Omega_{\rm int} \; , 
		\qquad   \alpha \in [ \, 0, 2\pi ) \; .  
	\end{equation}
The vector $\Omega_{\rm int}$ was defined in Equ.~\eqref{intvacuum};
\item[$ iii.)$]
the vector $\Omega_{\rm int}$ satisfies the Peierls-Bogoliubov  and the Golden-Thompson inequalities:
	\[
		{\rm e}^{- {\pi} (\Omega , V^{(\alpha)} \Omega)   } 
		\le  \| {\rm e}^{-\pi H^{(\alpha)}}\Omega \|  
		\le  \|{\rm e}^{- {\pi} V^{(\alpha)}   } \Omega \| \; ;
	\]
\item[$ iv.)$]
the operator $H^{(\alpha)}- J^{(\alpha)} V^{(\alpha)} J^{(\alpha)}$ is essentially self-adjoint 
on the domain 
	\[
		{\mathscr D} \bigl(H^{(\alpha)} \bigr) \cap {\mathscr D} \bigl(J^{(\alpha)} V^{(\alpha)} J^{(\alpha)} \bigr)
	\]
and the closure equals $L^{(\alpha)}_{\rm int} $, 
	\[
		\overline{H^{(\alpha)}- J^{(\alpha)} V^{(\alpha)} 
		J^{(\alpha)}} = L^{(\alpha)}_{\rm int} \; ,  \qquad \alpha \in  [ \, 0, 2\pi ) \; , 
	\]
where $L^{(\alpha)}_{\rm int}$ is defined in Equ.~\eqref{auchwichtig}.
Moreover, $L^{(\alpha)}_{\rm int} \Omega_{\rm int}= 0$; 
\item[$ v.)$]
the conjugation $J^{(\alpha)}$ defined in (\ref{est1.2})  satisfies
	\begin{equation}
		\label{squaresquare}
		 J^{(\alpha)}  {\rm e}^{-\pi L_{\rm int}^{(\alpha)}} A \Omega_{\rm int}
		 = A^*  \Omega_{\rm int}  \qquad 
 		\forall A \in {\mathcal R} \left(I_\alpha\right)\; .
 	\end{equation}
Thus $J^{(\alpha)}$ is the modular conjugation for the pair 
$\bigl({\mathcal R}\left(I_\alpha\right), \Omega_{\rm int} \bigr)$. 
Note that Equ.~\eqref{squaresquare} implies 
${\mathbb V}  J^{(\alpha)}_{V}  {\mathbb V} ^{-1}= J^{(\alpha)}$,  
with $J^{(\alpha)}_{V}$ defined in (\ref{wichtig3}).
\end{itemize}
\end{theorem}

\begin{proof}
Most of the arguments rely on results from the literature:
\begin{itemize}
\item [$ i.)$]  Essential selfadjointness follows from the results on local symmetric semigroups by Fr\"ohlich \cite{F80} and 
Klein and Landau \cite{KL1}\cite{KL2}. The key step in the proof is to show that 
$L^{(\alpha)}+V^{(\alpha)} = H^{(\alpha)}$ on 
		\[
		\widehat{\mathscr D}= \bigcup_{0< \theta \le \pi} \left( \bigcup_{0< \theta' <\theta}
		P^{(\alpha)}_{\rm int}(\theta') {\mathscr D}_{\theta, {\rm int}}^{(\alpha)} \right) \; ,
		\]
which itself  is a core for $H^{(\alpha)}$.
\item [$ ii.)$] The expression on the l.h.s.~of (\ref{Araki}) is a formula, which is well know from the perturbation theory of KMS states 
(see \cite{DJP}\cite{KL1}).
The identification (\ref{Araki}) follows from \eqref{Vhalbkugel}.
Note that 
the final expression on the r.h.s.~is independent of $\alpha$.
\item [$ iii.)$] The  Peierls-Bogoliubov  and 
the Golden-Thompson inequalities (see \cite{A4}) 
were generalised to the present case in~\cite[Theorem 5.5]{DJP}.
\item [$ iv.)$] Since ${\rm e}^{- 2 \pi V^{(\alpha)}}\in L^{1}({\mathcal Q}, \Sigma^{(\alpha)}, {\rm d} \Phi_C)$ and
	\[
		V\in L^{p}({\mathcal Q}, \Sigma^{(\alpha)}, {\rm d} \Phi_C) \; ,\quad
		{\rm e}^{-\pi V^{(\alpha)}}\in L^{q}({\mathcal Q}, \Sigma^{(\alpha)}   , {\rm d} \Phi_C) \; , 		
	\]
with $p^{-1}+ q^{-1}=\frac{1}{2} $ and $ 2\leq p, q\leq \infty$, property (iv) follows from \cite[Theorem~6.12]{GeJ}.
$L^{(\alpha)}_{\rm int} \Omega_{\rm int}= 0$ follows from \eqref{qtotot} and \eqref{auchwichtig}.
\item [$ v.)$] This result is due to Klein and Landau  \cite{KL1}; see also \cite[Lemma 7.13]{GeJ}.
\end{itemize}
\end{proof}

Note that $H^{(\alpha)}$ does not implement the Lorentz boosts on ${\mathcal H}$. However, 
	\begin{equation}
		\label{hosaut}
		\widehat \alpha_{\Lambda^{(\alpha)} (t)}^{\rm \, int} (A)
		= {\rm e}^{ i t H^{(\alpha)} } A {\rm e}^{- i t H^{(\alpha)} } \qquad \forall
		A \in {\mathcal R}  (I_\alpha ) \; .
	\end{equation}
The difference between $H^{(\alpha)}$ and $L_{\rm int}^{(\alpha)}$
is an unbounded operator affiliated to the commutant ${\mathcal R}  (I_\alpha )'$ 
of ${\mathcal R}  (I_\alpha )$. 
This is in agreement
with the perturbation theory of modular automorphisms  \cite{BR}\cite{DJP}.

\begin{theorem}[Uniqueness of the interacting de Sitter vacuum state]
\label{keyresult3}
For each $\alpha \in [ \, 0, 2 \pi ) $, the restricted state 
	\[
		\omega^{\rm \, int}_{\upharpoonright {\mathcal R}  (I_\alpha)} (A) =  
		\langle \Omega_{\rm int} , A \Omega_{\rm int} \rangle \; , 
		\qquad A \in {\mathcal R}  (I_\alpha) \; , 
	\]
is the unique $\widehat {\alpha}^{\rm \, int}_{\Lambda^{(\alpha)}}$-KMS state 
on~${\mathcal R}  (I_\alpha)$ and (therefore)
$\omega^{\rm \, int}$ is the unique de Sitter vacuum state 
for the $W^*$-dynamical system $\bigl({\mathcal R}  (S^1), 
\widehat {\alpha}_{\Lambda^{(\alpha)}}^{\rm \, int} \bigr)$.
\end{theorem}

\begin{proof}
Note that ${\mathcal R} (I_\alpha)$ and 
	\[
		{\mathcal R}_{\rm int} (I_\alpha)
		 = \bigvee_{t \in {\mathbb R}} 
		\left( {\rm e}^{-it  H^{(\alpha)} }  {\mathcal U} (I_\alpha)  
		{\rm e}^{it  H^{(\alpha)} }\right) \; ,
	\]
coincide: as consequence of \eqref{Araki1}, \eqref{hosaut} and the Trotter product formula
	\[
	\widehat {\alpha}_{\Lambda^{(\alpha)}}^{\rm \, int} ( A) = \lim_{n \to \infty} 
	\left( {\rm e}^{i  \frac{t}{n} V^{(\alpha)} } 
	\widehat {\alpha}_{\Lambda^{(\alpha)} (t/n)}^\circ ( A)	
	{\rm e}^{i \frac{t}{n} V^{(\alpha)} }\right)^n  \; , 
	\quad A \in {\mathcal U} (I_\alpha)  \; ,
	\]
we find ${\mathcal R} (I_\alpha)={\mathcal R}_{\rm int} (I_\alpha)$, 
as ${\rm e}^{i t V^{(\alpha)} } \in {\mathcal R}  (I_\alpha)$ for $t \in {\mathbb R}$. 
For the free field the $\widehat {\alpha}^{\circ}_{\Lambda^{(\alpha)}}$-KMS state 
on~${\mathcal R} (I_\alpha)$ is unique, thus 
${\mathcal R} (I_\alpha)$ is a factor, and uniqueness of the interacting state now is 
a direct consequence of \cite[Proposition~5.3.29]{BR}, as was kindly pointed out to us by Jan Derezinski. 
\end{proof}

\begin{theorem}
\label{keyresult1}
The operator sum $L^{(\alpha)}+V_0 (\cos_{\psi + \alpha})$ 
is essentially self-adjoint and the closure
	\[
	 \overline{L^{(\alpha)}+V_0 (\cos_{\psi + \alpha})} = L^{(\alpha)}_{\rm int} \; , 
	\] 
where ${V_0} (h) $ was defined in (\ref{bbv-interaction}). 
\end{theorem}

Note that the integration in (\ref{bbv-interaction})  is over the 
whole circle~$S^1$, while the integration in~(\ref{halberkreis}) is restricted 
to the halfcircle $I_\alpha$. 

\begin{proof} This result 
follows from the fact that $J^{(\alpha)}$ implements the space-reflec\-tion 
$P^{(\alpha)} = R_{0} (\alpha) P_1 R_{0} (\alpha)^{-1}$
on ${\mathcal U} (S^1)$ and $\cos (\frac{\pi}{2} + \psi)  = - \cos (\frac{\pi}{2} - \psi)$. Thus
	\[
	 V_0 (\cos_{\psi + \alpha}) 
	 = \overline{ V^{(\alpha)} - J^{(\alpha)} V^{(\alpha)} J^{(\alpha)} }  \; , 
	\] 
The statement now follows from Theorem \ref{th5.2} $i.)$ and $ iv.)$. 
\end{proof}
 
\section{Local von Neumann algebras on the circle $S^1$}
\label{sec:loccirc}

Let ${\mathcal A}^{(\alpha)}_{r}(I_\alpha)$  denote the
von Neumann algebra generated by 
	\[
		\left\{ \widehat {\boldsymbol \alpha}^{\circ}_{\Lambda^{(\alpha)} (t )} (A) 
		\mid A\in {\mathcal U} (I_\alpha), \; |t |<r \right\}.
	\]
Then
\label{loc-nc-algpage}
	\[
		\bigcap_{r>0} {\mathcal A}^{(\alpha)}_{r}(I_\alpha)= {\mathcal R} (I_\alpha) \;  .
	\]
(This is a special case of Theorem \ref{p6.1} below). This suggest to define
the local non-commutative von Neumann algebra ${\mathcal R} (I)$
as the intersection over the von Neumann algebras ${\mathcal A}^{(\alpha)}_{r}(I)$ generated by 
	\[
		\left\{ \widehat {\alpha}^{\circ}_{\Lambda^{(\alpha)} (t )} (A) 
		\mid A\in {\mathcal U} (I), \; |t |<r \right\}.
	\]
There is however the question, whether this definition depends on $\alpha$. This is not the case, 
as will be shown next.

\begin{theorem} 
\label{p6.1}
Consider the real subspace $\widehat{\mathfrak h} (I) \subset \widehat{\mathfrak h} (S^1)$,
	\begin{equation}
	\widehat{\mathfrak h} (I) 
	\doteq
	\left\{ h \in \widehat{\mathfrak h} (S^1)   \mid  
	{\rm supp\,} \left( \Re h \, ,  \, \omega^{-1}\Im h \right) \subset I \times I \right\}  \; , 
	\label{ekg4}
	\end{equation}
first introduced in Section \ref{Sect: canon-HS}. Then
	\begin{equation}
		\label{intersectBa2}
		{\mathcal R} (I) \doteq \bigcap_{r>0} {\mathcal A}^{(\alpha)}_{r}(I)
		= {\mathfrak W} ( \widehat{{\mathfrak h}} ( I) )'' \; , \qquad I  \subset I^{(\alpha)} \; .
	\end{equation}
In particular,  ${\mathcal R} (I)$ as defined in \eqref{intersectBa2}  does not dependent on $\alpha$.
\end{theorem}

\begin{proof} The following argument is similar to the one given in the proof of \cite[Theorem 6.5]{GeJII}. 
To simplify the notation, set
	\begin{equation}
	{\mathscr W} (I) \doteq {\mathfrak W} ( \widehat{\mathfrak h} ( I) )'' \; .
	\end{equation}
We first prove that $\bigcap_{r>0} {\mathcal A}^{(\alpha)}_{r}(I)\subset {\mathscr W} (I)$. 
Using ${\mathcal U} (I)\subset {\mathcal R} (I)$ and finite-speed-of-light (Theorem~\ref{fst-theorem}), we see that 
	\[
	 {\mathcal A}^{(\alpha)}_{r}(I)\subset {\mathscr W} \bigl( I (\alpha, r) \bigr)\; 
	 \qquad \forall r > 0 \; . 
	 \]
According to Proposition \ref{l6.1} the von Neumann algebras ${\mathscr W} (I)$, $I \subset S^1$, are regular from the outside. 
This implies~$\bigcap_{r>0} {\mathcal A}^{(\alpha)}_{r}(I)\subset {\mathscr W} (I)$.

Let us now prove that ${\mathscr W} (I) \subset \bigcap_{r>0} {\mathcal A}^{(\alpha)}_{r}(I)$. 
Using that the local time-zero algebras are regular from 
the inside (Proposition \ref{l6.1}), it suffices to show that 
for each~$\overline{J}\subset
I$ there exists some positive real number $r \ll 1$ such that  
	\begin{equation}
		\label{e6.02}
	{\mathscr W} (J)\subset {\mathcal A}^{(\alpha)}_{r}(I) \; .
	\end{equation}
To this end we fix $I$ and $J$ with $\overline{J}\subset I$ and set
$\delta={\frac{1}{2}}{\rm dist}(J, I^{c})$. We first note that 
	\begin{equation}
		\label{e6.03}
		{\rm e}^{i tL^{(\alpha)} }A{\rm e}^{-i tL^{(\alpha)} } \in {\mathcal A}^{(\alpha)}_{r}(I) \; , 
		\quad A\in
		{\mathcal U} (J) \; , \quad |t|< r \; ,
	\end{equation}
if $0< r<\delta$. Clearly, the Weyl operators ${\rm W}_F (h) = \exp(i \Phi^{\rm os} (0,h) )$, $h \in \widehat {\mathfrak h}(S^1)$ real valued, 
belong to~${\mathcal U} (J)$ if 
${\rm supp\,} h\in  J$ and hence to ${\mathcal A}^{(\alpha)}_{r}(I)$. 
Now~(\ref{e6.03}) implies 
	\begin{align}
		\label{e6.04}
		\widehat {\alpha}^{\, \circ}_{\Lambda^{(\alpha)} (t)} \bigl({\rm W}_F (h) \bigr) 
		&=		{\rm W}_F \bigl({\rm e}^{i t \omega r \, \cos_{\psi +\alpha}   }h \bigr) 
		\in {\mathcal A}^{(\alpha)}_{r}(I)\; , \qquad |t|<r \; . 
	\end{align}
Hence 
	\[
	{\rm W}_F \bigl(t^{-1}({\rm e}^{i t \omega r \, \cos_{\psi +\alpha} }h -h) \bigr)
	\in {\mathcal A}^{(\alpha)}_{r}(I)\; , \qquad 
	|t  |<\epsilon\; . 
	\]
Letting $t \to 0$ and using the fact that the map $h\mapsto {\rm W}_F (h)$ is continuous for the strong
operator topology, we obtain that ${\rm W}_F(i \omega \, r  \cos_{\psi +\alpha}  h)
\in {\mathcal A}^{(\alpha)}_{r}(I)$. 
But any vector 
$h \in \widehat{\mathfrak h} (J)$ can be
approximated in norm by vectors of the form 
	\[
	h_{1}+ i \omega \, r \cos_{\psi +\alpha} h_{2} \; , 
	\]
with ${\rm supp\,} h_{i}\in J$, $i=1,2$, real and $\cos_{\psi +\alpha} h_{2}\in {\mathscr D}(\omega)$. 
Thus for all $h\in \widehat{\mathfrak h} (J)$ 
the operators ${\rm W}_F (h)$ belong to ${\mathcal A}^{(\alpha)}_{r}(I)$ and hence 
${\mathscr W} (J) \subset {\mathcal A}^{(\alpha)}_{r}(I)$.  
\end{proof} 

\section{Finite speed of light for the ${\mathscr P}(\varphi)_2$ model}
\label{sec:fsol3}

The set	(see Proposition \ref{ialpha} for an explicit formula)
	\[
			I (\alpha , t) 
		= S^1 \cap \Bigl( \bigcup_{y \in \Lambda^{(\alpha)} (t) I }   
		\Gamma^- (y )\cup  \Gamma^+ (y ) \Bigr) 
	\]
describes the localisation region for the Cauchy data, which 
can influence space-time points in the set $\Lambda^{(\alpha)} (t) I$, $t \in \mathbb{R}$ fixed. 

\begin{theorem}
\label{fst-2-theorem}
Let $I \subset S^1$ 
be an open interval. Then 
	\begin{equation}
		\widehat 
		\alpha^{\rm \, int}_{\Lambda^{(\alpha)} (t)}
 		\colon  
		{\mathcal R} (I)\hookrightarrow 
		{\mathcal R} \bigl( I (\alpha , t)  \bigr) \; .
		\label{e6.1ef}
	\end{equation}
\end{theorem}

\begin{proof}
The following argument  is similar to the one given in the proof of \cite[Theorem 4.1.2]{GJcp}. 
We have seen in Theorem~\ref{fst-theorem}  that 
	\begin{equation}
		\label{e6.1d}
		\widehat  \alpha^{\, \circ}_{\Lambda^{(\alpha)} (t)} \colon 
		{\mathcal R} (I) \hookrightarrow 
		{\mathcal R} \bigl( I (\alpha , t) \bigr)
		\; .
	\end{equation}
We can now explore the fact that on the half-circle $I_\alpha$ the automorphism 
$\widehat \alpha^{\rm int}_{\Lambda^{(\alpha)} (t)}$ is unitarily implemented 
by~${\rm e}^{i tH^{(\alpha)}}$, where
$H^{(\alpha)}=\overline{L^{(\alpha)} + V^{(\alpha)} }$ with
	\[
		 V^{(\alpha)}=\int_{I_\alpha} r {\rm d} \psi \, \cos( \psi - \alpha)  \;   
		 {:}{\mathscr P}(\Phi^{\rm os} (0,\psi)){:}_{C_0} \; .
	\]
Trotter's product formula yields 
	\[
	{\rm e}^{i t H^{(\alpha)} }= s-\lim_{n\to\infty} 
	\left({\rm e}^{i tL^{(\alpha)} /n}{\rm e}^{i tV^{(\alpha)}/n}
	\right)^{n} \; . 
	\]
Hence
	\begin{equation}
	 \widehat \alpha^{\rm int}_{\Lambda^{(\alpha)} (t)} (A)
	 = s-\lim_{n\to \infty} \left( \widehat\alpha^{\circ}_{\Lambda^{(\alpha)} (t/n)}\circ
	\widehat {\gamma}^{(\alpha)}_{t/n} \right)^{n}(A) \; , 
	\qquad A\in {\mathcal R} (I_\alpha) \; ,
	\label{e6.1e}
	\end{equation} 
with 
	\[ 
	\widehat {\gamma}^{(\alpha)}_t  (A) = {\rm e}^{i  tV^{(\alpha)}  }A{\rm e}^{-i tV^{(\alpha)} } \; . 
	\]
Note that $\widehat {\gamma}^{(\alpha)}$ has zero propagation 
speed \cite{GJcp}, as for every open interval ${J} \subset I_\alpha$ 
there exists $V_{\rm loc}^{(\alpha)}$ affiliated\footnote{Let ${\mathcal R}$ be a von Neumann algebra acting on a Hilbert space 
${\mathcal H}$. A closed and densely defined operator $A$ is said to be affiliated with ${\mathcal R}$
if $A$ commutes with every unitary operator $U$ in the commutant of ${\mathcal R}$. }
to ${\mathcal U} ({J})$ such that  for all $ t \in {\mathbb R}  $
	\[
		{\rm e}^{it V^{(\alpha)}} A {\rm e}^{-it V^{(\alpha)}}
		= {\rm e}^{it V_{\rm loc}^{(\alpha)}} A {\rm e}^{-it V_{\rm loc}^{(\alpha)}} \; , 
		\qquad   A \in {\mathcal R}  ({J})   \; . 
	\]
Consequently,  
	\begin{equation}
	\label{zeroprop}
		\widehat {\gamma}^{(\alpha)}_t ({\mathcal R} (J)) 
		= {\mathcal R} (J)  
	\qquad  \forall t  \in {\mathbb R}  \; . 
	\end{equation}
Now (\ref{e6.1ef}) follows from \eqref{e6.1e} and  \eqref{e6.1d}.
\end{proof} 

\begin{theorem} \label{p6.1b}
For $I \subset S^1 $, let ${\mathcal B}^{(\alpha)}_{r}(I)$  
denote the
von Neumann algebra generated by 
	\[
		\left\{ \widehat \alpha^{\rm \, int}_{\Lambda^{(\alpha)} (t )} (A) 
		\mid A\in {\mathcal U} (I), \; |t |<r \right\}.
	\]
Then
	\begin{equation}
		\label{intersectBa}
		\bigcap_{r>0} {\mathcal B}^{(\alpha)}_{r}(I)
		= {\mathcal R}(I) \; , \qquad I  \subset S^1 \; .
	\end{equation}
Both sides in \eqref{intersectBa} are independent of $\alpha$.
\end{theorem}

\begin{proof} 
Let us first prove that $\bigcap_{r>0} {\mathcal B}^{(\alpha)}_{r}(I)\subset {\mathcal R}(I)$. 
Using
(\ref{e6.1f})  and ${\mathcal U}(I)\subset {\mathcal R}(I)$, we see that 
	\[
	 {\mathcal B}^{(\alpha)}_{r}(I)\subset {\mathcal R} \bigl( I (\alpha, r) \bigr)
	 \qquad r >0 \; .
	 \]
According to Proposition \ref{l6.1}, the local time-zero algebras are regular from 
the outside. This
implies~$\bigcap_{r>0} {\mathcal B}^{(\alpha)}_{r}(I)\subset {\mathcal R}(I)$.

Let us now prove that ${\mathcal R}(I)\subset \bigcap_{r>0} {\mathcal B}^{(\alpha)}_{r}(I)$. 
Using again Proposition \ref{l6.1} (this time using that the local time-zero algebras are regular from 
the inside),  it suffices to show that 
for each~$\overline{J}\subset
I$ there exists some positive real number $r \ll 1$ such that  
	\begin{equation}
		\label{e6.02b}
		{\mathcal R}({J})\subset {\mathcal B}^{(\alpha)}_{r}(I) \; .
	\end{equation}
To this end we fix ${J}$ and $I$ with $\overline{{J}}\subset I$ and set
$\delta={\frac{1}{2}}{\rm dist}({J}, I^{c})$. For $|t|\leq \delta$ the unitary group 
${\rm e}^{i t H^{(\alpha)} }$  with 
	\[
	H^{(\alpha)}=\overline{L^{(\alpha)} + V^{(\alpha)}}
	\]
induces the correct dynamics $\widehat \alpha^{\rm \, int}_{\Lambda^{(\alpha)} (t)}$ 
on ${\mathcal R} ({J})$.  Apply \cite[Proposition 2.5]{GeJII} to obtain 
	\[
		{\rm e}^{i tH^{(\alpha)} }=s-\lim_{n\to \infty}{\rm e}^{i tH^{(\alpha)}_n } \; ,  \qquad \: t\in {\mathbb R} \; ,
	\]
for $H^{(\alpha)}_n 
=\overline{L^{(\alpha)} + V^{(\alpha)}- V^{(\alpha)}_{n}}$, where $V^{(\alpha)}_n =
V^{(\alpha)} \mathbb{1}_{\{|V^{(\alpha)} |\leq n\}}$. Since $V^{(\alpha)}_n$ is bounded,
	\[
	H^{(\alpha)}_n
	= \overline{L^{(\alpha)} +V^{(\alpha)}- V^{(\alpha)}_n}
	= H^{(\alpha)}- V^{(\alpha)}_n \;  
	\]
and Trotter's formula yields
	\[
 	{\rm e}^{i tH^{(\alpha)}_n }= s-\lim_{p\to \infty} \left({\rm e}^{i t H^{(\alpha)}/p}{\rm e}^{-i t
		V^{(\alpha)}_n /p} \right)^{p} \; .
	\]
Hence, for $A\in {\mathcal R}({J})$,
	\begin{align*}
	& {\rm e}^{i tL^{(\alpha)} }A{\rm e}^{-i tL^{(\alpha)} } \\
	& \quad = s-\lim_{n\to \infty}s-\lim_{p\to \infty}
	\left({\rm e}^{i t H^{(\alpha)} /p}{\rm e}^{-i t V^{(\alpha)}_n/p} \bigr)^{p}A \bigl({\rm e}^{i t
	V^{(\alpha)}_n/p}{\rm e}^{-i t H^{(\alpha)} /p}\right)^{p} \; .
	\end{align*}
For $|t|<r$ and $p \in {\mathbb N}$
	\begin{multline*}
		\left({\rm e}^{i t H^{(\alpha)} /p}
		{\rm e}^{-i t V^{(\alpha)}_n /p} \right)^{p}
		A 
		\left({\rm e}^{i t V^{(\alpha)}_n /p}{\rm e}^{-i t H^{(\alpha)} /p}\right)^{p} = 
		 \\
		= \left( \widehat \alpha^{\rm \,  int}_{\Lambda^{(\alpha)} (t/p)} 
		\circ \widehat {\gamma}^{(\alpha)}_n (-t/p) \right)^{p}(A) \; ,
	\end{multline*}
where $\widehat {\gamma}^{(\alpha)}_n$ is the dynamics implemented by the unitary
group $t \mapsto {\rm e}^{-i t V^{(\alpha)}_n }$. This implies 
	\[
	\bigl( \widehat \alpha^{\rm \,  int}_{\Lambda^{(\alpha)} (t/p)} \circ 
	\widehat {\gamma}^{(\alpha)}_n (t/p) \bigr)^{p}(A)\in {\mathcal B}^{(\alpha)}_{r}(I) \; ,  
	\qquad |t|<r \; , \quad p \in {\mathbb N} \; . 
	\]
Take the limit $n \to \infty$ and recall from \eqref{zeroprop} that $ \widehat {\gamma}
= \lim_{n \to \infty} \widehat {\gamma}_n$ has zero propagation speed. 
Since~${\mathcal B}^{(\alpha)}_{r}(I)$ is weakly closed, we obtain (\ref{e6.02b}).
The result now follows from Proposition \ref{l6.1}. 
\end{proof} 

\begin{remark} Thus
	\[
		{\mathcal R}_{\rm int} (I)= {\mathcal R} (I) \; \qquad 
		\forall I \subset S^1 \; . 
	\]
We have seen earlier that ${\mathcal R}_{\rm int} (I_\alpha)= {\mathcal R}  (I_\alpha)$  for all 
half-circles $I_\alpha$. If $I$ is contained in some half-circle, then it follows that 
	\[
		{\mathcal R}_{\rm int} (I) = \bigcap_{I \subset I_\alpha} {\mathcal R}_{\rm int}  (I_\alpha)
	\]
can be identified with the intersection 
of all algebras ${\mathcal R}_{\rm int}  (I_\alpha)$ associated to the half-circles $I_\alpha$, which contain $I$.
\end{remark}

\section{The stress-energy tensor}

One may introducing canonical time-zero fields $\varphi$ and canonical momenta $\pi$: they can be defined in 
terms of the Fock fields $\Phi_F $ on $\Gamma \bigl( \widehat{\mathfrak h}(S^1)\bigr) $ (see, \emph{e.g.}, \cite{RS}): 
	\[
		\widetilde \varphi (h) \doteq \Phi_F (h) \;, 
		\qquad 
		\widetilde \pi (g) \doteq  \Phi_F (i \omega g)\; , 
		\qquad h, g \in \widehat{\mathfrak h}(S^1, \mathbb{R} ) \; . 
	\]
Thus $\widetilde \varphi (h) = \Phi^{\rm os} (0,h) $ and  
$ \widetilde \pi (g) = -i [ \Gamma (\omega ), \Phi^{\rm os} (0,g)]$.  
They satisfy the {\em canonical commutation relations}
	\begin{align*}
		[ \widetilde \varphi(\psi), \widetilde \pi (\psi') ] 
			&=    \frac{i}{r}   \delta (\psi - \psi')\;  , 
			\\ 
		[ \widetilde \varphi(\psi), \widetilde \varphi(\psi') ] 
		&=   [ \widetilde \pi(\psi), \widetilde \pi (\psi') ] = 0 \; , 
	\end{align*}
in the sense of quadratic forms on $\Gamma \bigl( \widehat{\mathfrak h} (S^1)\bigr) $. 

\goodbreak
However, it is now more convenient to work on the Fock space over $L^2 (S^1, r {\rm d} \psi)$, using the map
	\[
		\widehat{\mathfrak h} (S^1) \ni f \mapsto \frac{1}{\sqrt{2 \omega}} f \in L^2 (S^1, r {\rm d} \psi)   
	\]
to identify the two realisations of the Fock space. The canonical fields and the canonical momenta 
take the form 
	\begin{align*}
		\varphi(\psi) & =  \frac{1}{\sqrt{2}}
						\Big(\big(\omega^{-\frac{1}{2}} a\big)(\psi)^* + \big(\omega^{-\frac{1}{2}} a\big)(\psi)\Big)  \;,   \\
		\pi(\psi) & =  \frac{i}{\sqrt{2}}   \Big(\big(  \omega^{\frac{1}{2}}  a\big)(\psi)^* - \big(  \omega^{\frac{1}{2}}  
		a\big)(\psi)\Big)
	\end{align*}
with
	\[ 
		a(\psi) \doteq \sum_{k\in\mathbb{Z}} \frac{e^{-ik\psi}}{\sqrt{2\pi r}} a_k
			\quad \text{and} \quad
		a(\psi)^* \doteq \sum_{k\in\mathbb{Z}} \frac{e^{ik\psi}}{\sqrt{2\pi r}} a_k^* \; . 
	\]
Thus
	\[
		\big(\omega^{\pm \frac{1}{2}} a\big)(\psi) 
		=\sum_{k\in\mathbb{Z}} \frac{ \widetilde \omega(k)^{\pm \frac{1}{2}} e^{-ik\psi}}{\sqrt{2\pi r}}  a_k
			\quad \text{and} \quad
		\big(\omega^{\pm \frac{1}{2}} a\big)(\psi)^* 
		= \sum_{k\in\mathbb{Z}} \frac{ \widetilde \omega(k)^{\pm \frac{1}{2}}  e^{ik\psi}}{\sqrt{2\pi r}} a_k^*  
	\]
and  $\big[\pi(\psi'), \; \varphi(\psi)\big]=-   \frac{i}{r}  \delta(\psi-\psi')$ still holds.
Using
	\begin{align*}
		a(\psi) & =  \frac{1}{\sqrt{2}} \Big[ \big(\omega^{\frac{1}{2}}\varphi\big)(\psi) -i \big(\omega^{-\frac{1}{2}}\pi\big)(\psi)  \Big]  \;,   \\
		a(\psi)^* & =  \frac{1}{\sqrt{2}} \Big[ \big(\omega^{\frac{1}{2}}\varphi\big)(\psi) +i \big(\omega^{-\frac{1}{2}}\pi\big)(\psi) \Big]  \;.
	\end{align*}
one verifies that 
	\[
		\big[a(\psi')^*, \; a(\psi)\big] =  \frac{i}{2}  \left\{  \big[\pi(\psi'), \; \varphi(\psi)\big]
								+  \big[\pi(\psi), \; \varphi(\psi')\big]  \right\}
							=  \frac{1}{r} \delta(\psi-\psi')   
	\]
and $\big[a(\psi')^*, \; a(\psi)^*\big]=\big[a(\psi'), \; a(\psi)\big]=0$.

\begin{lemma} Consider the Fock space over $L^2 (S^1, {\rm d} \psi)$. It follows that\footnote{We note that normal ordering is not need 
at this point.} 
	\begin{align*}
		L^{(\alpha)} &= {\rm d}\Gamma 
		\bigl( \sqrt{ \omega } \, r \cos_{\psi + \alpha} \sqrt{\omega }  \bigr)  \\
		& =
		 \frac{1}{2} 
			\int_{S^1} r \cos (\psi+ \alpha) \, {\rm d} \psi \;   \Bigl( \pi^2 (\psi) + \tfrac{1}{r^2}
			\bigl(\tfrac{ \partial \varphi }{\partial \psi}\bigr)^2 (\psi) + \mu^2 \varphi^2 (\psi)  \Bigr) 
	\end{align*}
is the generator of the free boost $t \mapsto U (\Lambda^{(\alpha)} (t))$ first introduced\footnote{Note that 
Definition~\ref{weak-Ext} refers to the original Fock space
$\Gamma \bigl(\widehat{\mathfrak h} (S^1)\bigr)$.}  in Definition~\ref{weak-Ext}. 
\end{lemma}

\begin{proof}
We write, using the fact that $\widetilde \omega(k)=\widetilde \omega(-k)$,
	\begin{align*}
		\varphi(\psi) & =  \frac{1}{\sqrt{4\pi}} \sum_{k\in\mathbb{Z}} 
						\widetilde \omega(k)^{-\frac{1}{2}}\Big( e^{ik\psi}a_k^* + e^{-ik\psi}a_k \Big) \; , \\
		\tfrac{\partial \varphi}{\partial\psi}(\psi) & = \frac{i}{\sqrt{4\pi}} \sum_{k\in\mathbb{Z}} 
						\widetilde \omega(k)^{-\frac{1}{2}}k\Big( e^{ik\psi}a_k^* - e^{-ik\psi}a_k \Big) \; , \\
		\pi(\psi) & =  \frac{i}{\sqrt{4\pi}} \sum_{k\in\mathbb{Z}} 
						  \widetilde \omega(k)^{\frac{1}{2}} 
						 \Big( e^{ik\psi}a_k^* - e^{-ik\psi}a_k \Big) \; . 
	\end{align*}
One has
	\begin{align*}
		\mu^2 \varphi(\psi)^2 & =  \frac{\mu^2}{4\pi} \sum_{k\in\mathbb{Z}} \sum_{l\in\mathbb{Z}} 
						\widetilde \omega(k)^{-\frac{1}{2}}\widetilde \omega(l)^{-\frac{1}{2}} \\
					& \quad \times \Big( e^{i(k+l)\psi}a_k^*a_l^* + e^{i(k-l)\psi} a_k^* a_l 
							+ e^{-i(k-l)\psi} a_k a_l^* + e^{-i(k+l)\psi} a_k a_l \Big) \\
	\tfrac{1}{r^2} \bigl(\tfrac{ \partial \varphi }{\partial \psi}\bigr)^2 (\psi) & = 
			- \frac{1}{4\pi r^2} \sum_{k\in\mathbb{Z}}  \sum_{l\in\mathbb{Z}} 
			\widetilde \omega(k)^{-\frac{1}{2}}\widetilde \omega(l)^{-\frac{1}{2}}\, kl \\
					& \quad \times 
				\Big( e^{i(k+l)\psi}a_k^*a_l^* - e^{i(k-l)\psi} a_k^* a_l - e^{-i(k-l)\psi} a_k a_l^* 
						+ e^{-i(k+l)\psi} a_k a_l \Big)  \\
		\pi(\psi)^2 & = - \frac{1}{4\pi  } \sum_{k\in\mathbb{Z}}  \sum_{l\in\mathbb{Z}} 
			\widetilde \omega(k)^{\frac{1}{2}}\widetilde \omega(l)^{\frac{1}{2}} \\
					& \quad \times 
				\Big( e^{i(k+l)\psi}a_k^*a_l^* - e^{i(k-l)\psi} a_k^* a_l - e^{-i(k-l)\psi} a_k a_l^* 
						+ e^{-i(k+l)\psi} a_k a_l \Big)
	\end{align*}
Next, define, for $j\in \mathbb{Z}$,
	\[ 
		S_j \doteq \int_{S^1}{\rm d} \psi \; \cos\psi e^{ij\psi} = \frac{1}{2}\int_{S^1}{\rm d}  \psi \; e^{i(j+1)\psi}
				+ \frac{1}{2}\int_{S^1}{\rm d} \psi \; e^{i(j-1)\psi} = \pi\big(\delta_{j,-1}+\delta_{j,1}\big) \;.
	\]
It is clear that $S_j=S_{-j}$ for all $j\in \mathbb{Z}$. Hence, we may write
	\begin{align*}
		\frac{1}{2} \int_{S^1} & r \,  {\rm d} \psi \; \cos\psi \; \pi(\psi)^2  = - \frac{r}{8\pi  } \sum_{k \in \mathbb{Z}}  \sum_{l\in\mathbb{Z}} 
				\widetilde \omega(k)^{\frac{1}{2}}\widetilde \omega(l)^{\frac{1}{2}} \\
			& \qquad \qquad \qquad \qquad \times
				\Big( S_{k+l}a_k^*a_l^* - S_{k-l} a_k^* a_l - S_{k-l} a_k a_l^* + S_{k+l} a_k a_l  \Big)  \\
			& = \frac{r}{8  } \sum_{k\in\mathbb{Z}}  \widetilde \omega(k)^{\frac{1}{2}}
				\Big[ -\widetilde \omega(-k+1)^{\frac{1}{2}} a_k^*a_{-k+1}^* 
					- \widetilde \omega(-k-1)^{\frac{1}{2}} a_k^*a_{-k-1}^* \\
			&  \qquad \qquad \qquad  \qquad + \widetilde \omega(k+1)^{\frac{1}{2}} a_k^*a_{k+1} 
			+ \widetilde \omega(k-1)^{\frac{1}{2}} a_k^*a_{k-1} \\
			&  \qquad \qquad \qquad  \qquad + \widetilde \omega(k+1)^{\frac{1}{2}} a_ka_{k+1}^* 
			+ \widetilde \omega(k-1)^{\frac{1}{2}} a_ka_{k-1}^* \\
			& \qquad \qquad \qquad  \qquad - \widetilde \omega(-k+1)^{\frac{1}{2}} a_ka_{-k+1} 
			- \widetilde \omega(-k-1)^{\frac{1}{2}} a_ka_{-k-1} \Big]
			\\
			& \\
		\frac{1}{2} \int_{S^1} r \, & {\rm d} \psi \; \cos\psi  \tfrac{1}{r^2} \bigl(\tfrac{ \partial \varphi }{\partial \psi}\bigr)^2 (\psi)  
			= - \frac{1}{8\pi r} \sum_{k\in\mathbb{Z}}  \sum_{l\in\mathbb{Z}} 
					\widetilde \omega(k)^{-\frac{1}{2}}\widetilde \omega(l)^{-\frac{1}{2}} \; kl \\
			& \qquad \qquad \qquad \qquad \times
						\Big(S_{k+l}a_k^*a_l^* - S_{k-l} a_k^* a_l - S_{k-l} a_k a_l^* + S_{k+l} a_k a_l \Big) \\
			& = \frac{1}{8r} \sum_{k\in\mathbb{Z}}  \widetilde \omega(k)^{-\frac{1}{2}}\, k \\
			& \qquad \times
				\Big[ -\widetilde \omega(-k+1)^{-\frac{1}{2}}(-k+1) a_k^*a_{-k+1}^* 
					- \widetilde \omega(-k-1)^{-\frac{1}{2}}(-k-1) a_k^*a_{-k-1}^* \\
			& 	\qquad \qquad 	+ \widetilde \omega(k+1)^{-\frac{1}{2}}(k+1) a_k^*a_{k+1} 
					+ \widetilde \omega(k-1)^{-\frac{1}{2}}(k-1) a_k^*a_{k-1}  \\
			& 	\qquad \qquad  + \widetilde \omega(k+1)^{-\frac{1}{2}}(k+1) a_ka_{k+1}^* 
				+ \widetilde \omega(k-1)^{-\frac{1}{2}}(k-1) a_ka_{k-1}^* \\
			& 	\qquad \qquad  - \widetilde \omega(-k+1)^{-\frac{1}{2}}(-k+1) a_ka_{-k+1} 
				- \widetilde \omega(-k-1)^{-\frac{1}{2}}(-k-1) a_ka_{-k-1} \Big]
	\end{align*}
	\begin{align*}
		\frac{\mu^2}{2} \int_{S^1} r \, &  {\rm d} \psi \; \cos\psi  \big(\varphi(\psi) \big)^2 \\
			& = \frac{\mu^2 r}{8\pi} \sum_{k\in\mathbb{Z}}  \sum_{l\in\mathbb{Z}} 
					\widetilde \omega(k)^{-\frac{1}{2}}\widetilde \omega(l)^{-\frac{1}{2}} \\
			& \qquad \qquad \qquad \times
					\Big( S_{k+l}a_k^*a_l^* + S_{k-l} a_k^* a_l + S_{k-l} a_k a_l^* + S_{k+l} a_k a_l \Big) \\
			& = \frac{r}{8} \sum_{k\in\mathbb{Z}}  \mu^2\widetilde \omega(k)^{-\frac{1}{2}}
			\\
			& \qquad \times
					\Big[ \widetilde \omega(-k+1)^{-\frac{1}{2}} a_k^*a_{-k+1}^* 
						+ \widetilde \omega(-k-1)^{-\frac{1}{2}} a_k^*a_{-k-1}^* \\
			& \qquad \qquad  + \widetilde \omega(k+1)^{-\frac{1}{2}} a_k^*a_{k+1} 
				+ \widetilde \omega(k-1)^{-\frac{1}{2}} a_k^*a_{k-1} \\
			& \qquad \qquad  + \widetilde \omega(k+1)^{-\frac{1}{2}} a_ka_{k+1}^* 
				+ \widetilde \omega(k-1)^{-\frac{1}{2}} a_ka_{k-1}^* \\
			& \qquad \qquad  + \widetilde \omega(-k+1)^{-\frac{1}{2}} a_ka_{-k+1} 
				+ \widetilde \omega(-k-1)^{-\frac{1}{2}} a_ka_{-k-1} \Big]
	\end{align*}
Rearranging the terms in order to join terms having factors involving
the operators $a$ in common and using the fact that $\widetilde \omega(-k)=\widetilde \omega(k)$ for all $k\in\mathbb{Z}$, 
we get
	\begin{align}
		\label{eq:almost-there}
	L_1 & = \frac{r}{8 } \sum_{k\in\mathbb{Z}}  
			\Bigg[
					-\widetilde \omega(k)^{-\frac{1}{2}}\widetilde \omega(k-1)^{-\frac{1}{2}}
						\underbrace{ \Big( \widetilde \omega(k)\widetilde \omega(k-1) + {r}^{-2}
						(-k+1)k -  \mu^2  \Big)}_{A} 
						a_k^*a_{-k+1}^* 
				\nonumber\\
			&    \qquad \qquad \quad
					-\widetilde \omega(k)^{-\frac{1}{2}}\widetilde \omega(k+1)^{-\frac{1}{2}}
						\underbrace{ \Big( \widetilde \omega(k)\widetilde \omega(k+1) + {r}^{-2}
						(-k-1)k  - \mu^2  \Big) }_{B} 
						a_k^*a_{-k-1}^* 
				\nonumber\\
			&    \qquad \qquad \quad
					+\widetilde \omega(k)^{-\frac{1}{2}}\widetilde \omega(k+1)^{-\frac{1}{2}}
						\underbrace{ \Big(  \widetilde \omega(k)\widetilde \omega(k+1) + {r}^{-2}
						(k+1)k  + \mu^2  \Big) 
									}_{-B + 2\widetilde \omega(k)\widetilde \omega(k+1)} 
						a_k^*a_{k+1} 
				\nonumber\\
			&    \qquad \qquad \quad
					+\widetilde \omega(k)^{-\frac{1}{2}}\widetilde \omega(k-1)^{-\frac{1}{2}}
						\underbrace{ \Big(  \widetilde \omega(k)\widetilde \omega(k-1) + {r}^{-2}
						(k-1)k  + \mu^2  \Big) 
									}_{-A + 2\widetilde \omega(k)\widetilde \omega(k-1)} 
						a_k^*a_{k-1} 
				\nonumber\\
			&    \qquad \qquad \quad
					+\widetilde \omega(k)^{-\frac{1}{2}}\widetilde \omega(k+1)^{-\frac{1}{2}}
						\underbrace{ \Big(  \widetilde \omega(k)\widetilde \omega(k+1) + {r}^{-2}
						(k+1)k  + \mu^2  \Big)
									}_{-B + 2\widetilde \omega(k)\widetilde \omega(k+1)}  
						a_ka_{k+1}^* 
				\nonumber\\
			&    \qquad \qquad \quad
					+\widetilde \omega(k)^{-\frac{1}{2}}\widetilde \omega(k-1)^{-\frac{1}{2}}
						\underbrace{ \Big(  \widetilde \omega(k)\widetilde \omega(k-1) + {r}^{-2}
						(k-1)k  + \mu^2  \Big) 
									}_{-A + 2\widetilde \omega(k)\widetilde \omega(k-1)} 
						a_ka_{k-1}^* 
				\nonumber \\
			&    \qquad \qquad \quad
					-\widetilde \omega(k)^{-\frac{1}{2}}\widetilde \omega(k-1)^{-\frac{1}{2}}
						\underbrace{ \Big( \widetilde \omega(k)\widetilde \omega(k-1) + {r}^{-2}
						(-k+1)k  - \mu^2  \Big) }_{A} 
						a_ka_{-k+1} 
				\nonumber \\
			&    \qquad \qquad \quad
					-\widetilde \omega(k)^{-\frac{1}{2}}\widetilde \omega(k+1)^{-\frac{1}{2}}
						\underbrace{ \Big( \widetilde \omega(k)\widetilde \omega(k+1)  + {r}^{-2}
						(-k-1)k  - \mu^2  \Big) }_{B} 
						a_ka_{-k-1}
			\Bigg] 
	\end{align}
Due to \eqref{eq:useful-1} and \eqref{eq:useful-2} both $A$, $B$ vanish.

Returning with these informations to (\ref{eq:almost-there}), we get
	\begin{align}
		\label{eq:almost-there-2}
		L_1 & =  \frac{r}{4 } \sum_{k\in\mathbb{Z}}  
				\Bigg[
					\widetilde \omega(k)^{\frac{1}{2}}\widetilde \omega(k+1)^{\frac{1}{2}} a_k^*a_{k+1} 
					+\widetilde \omega(k)^{\frac{1}{2}}\widetilde \omega(k-1)^{\frac{1}{2}} a_k^*a_{k-1} 
				\nonumber \\
			& \qquad \qquad
					+\widetilde \omega(k)^{\frac{1}{2}}\widetilde \omega(k+1)^{\frac{1}{2}} a_ka_{k+1}^* 
					+\widetilde \omega(k)^{\frac{1}{2}}\widetilde \omega(k-1)^{\frac{1}{2}} a_ka_{k-1}^* 
				\Bigg]
		\end{align}
Now, we have
	\begin{align}
		\label{eq:upIUYiuy-1}
			\sum_{k\in\mathbb{Z}} \widetilde \omega(k)^{\frac{1}{2}}\widetilde \omega(k\pm1)^{\frac{1}{2}} a_k & a_{k\pm1}^* 
			  =   \sum_{k\in\mathbb{Z}} \widetilde \omega(k)^{\frac{1}{2}}\widetilde \omega(k\pm1)^{\frac{1}{2}} a_{k\pm1}^* a_k
			\nonumber \\
			& \quad \stackrel{k\to k\mp 1}{=}   \sum_{k\in\mathbb{Z}} \widetilde \omega(k)^{\frac{1}{2}} 
				\widetilde \omega(k\mp1)^{\frac{1}{2}} a_k^* a_{k\mp1} \; .
	\end{align}
In the first equality in (\ref{eq:upIUYiuy-1}) we have  used the fact that $a_k$ and $a_{k\pm1}^*$ commute.

Inserting (\ref{eq:upIUYiuy-1}) 
into (\ref{eq:almost-there-2}), we get, finally,
	\begin{equation}
		\label{eq:there}
	L_1 = \frac{r}{2 } \sum_{k\in\mathbb{Z}}  
		\bigg[
			\widetilde \omega(k)^{\frac{1}{2}}\widetilde \omega(k+1)^{\frac{1}{2}} a_k^*a_{k+1} 
			+\widetilde \omega(k)^{\frac{1}{2}}\widetilde \omega(k-1)^{\frac{1}{2}} a_k^*a_{k-1} 
		\bigg] \;.
	\end{equation}
On the other hand, 
one has
	\begin{align*}
		\int_{S^1} r \,  {\rm d} \psi\; & \big(\omega^{1/2} a\big)(\psi)^*\cos_\psi \big(\omega^{1/2} a\big)(\psi)
		\\
		& =  \frac{r}{2}\sum_{k\in\mathbb{Z}}\sum_{l\in\mathbb{Z}}
			\left(  \frac{1}{2\pi}\int_{S^1} {\rm d} \psi\; 
				e^{i(k-l+1)\psi} + e^{i(k-l-1)\psi} \right) \widetilde \omega(k)^{1/2}\widetilde \omega(l)^{1/2} a_k^* a_l
		\\
		& =  
			\frac{r}{2}\sum_{k\in\mathbb{Z}}\sum_{l\in\mathbb{Z}}
				\Big(\delta_{l,\, k+1}+\delta_{l,\,k-1}\Big) \widetilde \omega(k)^{1/2}\widetilde \omega(l)^{1/2} a_k^* a_l
		\\
		& =  
			\frac{r}{2}\sum_{k\in\mathbb{Z}}
				\Big( \widetilde \omega(k)^{1/2}\widetilde \omega(k+1)^{1/2} a_k^* a_{k+1}
				+ \widetilde \omega(k)^{1/2}\widetilde \omega(k-1)^{1/2} a_k^* a_{k-1} \Big)\;.
	\end{align*}
Therefore,
	\begin{align*}
		L_1 & = \frac{1}{2} \int_{S^1}  r \, {\rm d} \psi \cos\psi
				\left( \pi(\psi)^2 + \bigl( \tfrac{\partial \varphi}{\partial\psi}(\psi) \bigr)^2 
				+ \mu^2  \big(\varphi(\psi) \big)^2 \right) \\
		& 
		= \int_{S^1} r \, {\rm d}  \psi\; \big(  \omega^{1/2}  a\big)(\psi)^*\cos_\psi \big( \omega^{1/2}  a\big)(\psi) 
		\\
		& = {\rm d} \Gamma ( \sqrt{\omega} r \cos \psi \sqrt{\omega} ) 
		\;.
	\end{align*}
\end{proof}

The {\em energy density\/} ${\rm T}_{00}(\psi)$ is the restriction of the energy density in 
the time-zero plane (in the ambient Minkowski space) to the Cauchy surface $S^1$, \emph{i.e.}, for $\psi \in S^1$, 
	\begin{align}
	\label{energymomentumdensity2}
		{\rm T}_{00}  (\psi) &=  \frac{1}{2}  \left( \pi^2 (\psi)
						+ \tfrac{1}{r^2} \bigl( \tfrac{ \partial \varphi }{\partial \psi} (\psi) \bigr)^2 
						+ \mu^2 \varphi (\psi)  + {:}P (\varphi(\psi)) {:}_{C_0} \right)  \; .  
	\end{align}
The following formulas should be compared with the classical expressions derived in Section \ref{SET}.

\begin{lemma}
The following identities
hold in the sense of quadratic forms on the Hilbert space $\Gamma  \bigl( L^2( S^1, r {\rm d} \psi )\bigr) $:
	\begin{align*} 
			L_{\rm int}^{(\alpha)} &=  \int_{S^1} \,  r  \cos  (\psi + \alpha)  \, {\rm d} \psi  \;  {\rm T}_{00} \; , \\
		K_0 & = \int_{S^1} \,  r^2   \, |\cos\psi| \; {\rm d} \psi \;  T_{0\psi}
		\; ,
	\end{align*}
with $T_{0\psi} = (r \cos_\psi)^{-1} \mathbb{\pi} \, (\partial_\psi \mathbb{\Phi}) $. 
\end{lemma}

\begin{proof} Recall that $L^{(\alpha)}  = {\rm d} \Gamma \bigl(\sqrt{ \omega} \, r \, \cos_{\psi + \alpha} \sqrt{ \omega} \bigr)$.
Moreover, according to Theorem \ref{keyresult1}
	\[
	 L^{(\alpha)}_{\rm int}  = {\rm d} \Gamma ( \omega^{-1/2} ) \; 
	 \overline{{\rm d} \Gamma \bigl( \omega r \cos_{\psi + \alpha} \bigr)+V^{(\alpha)} } \; {\rm d} \Gamma ( \omega^{-1/2} )
	\] 
with $V^{(\alpha)}= \int_{S^1} r \,{\rm d} \psi' \,  \cos (\psi' + \alpha) \; {:} P (\varphi(\psi')){:}_{C_0} $. 

Next consider the angular momentum operator:
	\begin{align*} 
		K_0 & = \int_{S^1} \,  r \; {\rm d} \psi \;  {:} \mathbb{\pi} \,  (\partial_\psi \mathbb{\Phi}) {:}  \\
		& = \frac{i}{4 \pi}  \sum_{k\in\mathbb{Z}} \sum_{j\in\mathbb{Z}}
		\int_{S^1}  \,  {\rm d} \psi \;   
						{:} \; \widetilde \omega(k)^{\frac{1}{2}} \Big( e^{ik\psi}a_k^* - e^{-ik\psi}a_k \Big) 
						\\
		& \qquad \qquad \qquad \qquad \qquad \qquad \times  \partial_\psi   
						\widetilde \omega(j)^{-\frac{1}{2}}\Big( e^{ij\psi}a_j^* + e^{-ij\psi}a_j \Big) \; {:} 
		\\
		& = \frac{i}{4 \pi}  \sum_{k\in\mathbb{Z}} \sum_{j\in\mathbb{Z}}
		\frac{ \widetilde \omega(k)^{\frac{1}{2}} } {\widetilde \omega(j)^{-\frac{1}{2}}} 
		\Bigl( a_k^* a_j  \int_{S^1} \,  {\rm d} \psi \;   e^{ik\psi}  \partial_\psi  e^{-ij\psi} + a_j^* a_k  
		\int_{S^1} \,  {\rm d} \psi \;   e^{-ik\psi}  \partial_\psi  e^{ij\psi} \Bigr) \\
		& =  \frac{1}{2 \pi}  \sum_{k\in\mathbb{Z}} k a_k^* a_k  
		= {\rm d} \Gamma ( -i \partial _\psi) \; . 
	\end{align*}
\end{proof}

\color{black}
\begin{theorem}
The operator  $L_{\rm int}^{(0)}$ (and similar $L^{(0)}$) 	
	\begin{align*}
		L_{\rm int}^{(0)}					
		&=  \int_{I_+} {\rm d} \psi \; r \cos \psi   \; {\rm T}_{00}(\psi)
						+ \int_{I_-} {\rm d} \psi \; r \cos \psi  \; {\rm T}_{00} (\psi) \nonumber \\  
						&\equiv    L_{\rm int}^{  + } + L_{\rm int}^{  - }  + | \Omega \rangle \langle \Omega |\;  
	\end{align*}
splits into a positive and negative part,  
	\[
		{\rm Sp} (\pm L_{\rm int}^{ \pm }) = [0, \infty) \; . 
	\]
Moreover, ${\rm Sp} (L_{\rm int}^{ \pm })$ is absolutely continuous. 
\end{theorem}

\begin{proof}
The operator $\omega \cos_\psi = (\omega \cos_\psi)_{\upharpoonright I_+} + (\omega \cos_\psi)_{\upharpoonright I_-}$
is the sum of a positive operator $(\omega \cos_\psi )_{\upharpoonright I_+}$ acting on $\widehat{\mathfrak h} (I_+)$, 
and a negative operator $(\omega \cos_\psi)_{\upharpoonright I_-}$ on $\widehat{\mathfrak h} (I_-)$; 
see Proposition \ref{propfreeboost}. 

To show that 
	\[
		L_{\rm int \upharpoonright I_+}^{  + }  = {\rm d} \Gamma  (\omega \, r \cos_{\psi \upharpoonright I_+}  )
						+ \int_{I_+} {\rm d} \psi \;  r \cos \psi \;  {:} P (\varphi(\psi)) {:}_{C_0} \;  
	\]
is bounded from below on $\Gamma \bigl( \widehat{\mathfrak h} (S^1) \bigr)$, one can 
follow the arguments given in \cite[see, \emph{e.g.}, p.~276]{FHN}. The main idea is to
cut off the support of the interaction at the boundaries by some distance $\epsilon \ll \pi r$ 
form the boundary point. The bulk can then be bounded from below following standard arguments 
(see, \emph{e.g.}, \cite{N1} and also the proof of Proposition \ref{phi-bound} below.)
It remains to estimate the two contributions from the interaction in $\epsilon$-neighbourhoods of the boundary. 
This does not cause a problem, as the edges of the wedge are fixed points under the action of the boosts. 
A detailed analytic argument will be presented elsewhere. 
\end{proof}

\begin{remark}
We expect that the decomposition of $L_{\rm int}^{(0)}$ given above is relevant in context of the dethermalization 
discussed by Guido and Longo in \cite{GL}.
However, further work is need to clarify this question. 
\end{remark}

In connection with Remark \ref{classical-energy} it is worth while noting the following result:

\begin{proposition}[$\varphi$-bounds] 
\label{phi-bound}
For $c\gg 1$,
	\begin{equation}
		\label{e3.1b}
		\left\|\varphi (g) \left( \int_{S^1} r \, {\rm d} \psi \;   {\rm T}_{00}(\psi) + c \right)^{-\frac{1}{2}} \right\|
		\leq C\|g\|_{\widehat{\mathfrak h}(S^1)},
	\end{equation}
and
	\begin{equation}
		\label{e3.1c}
		\pm \varphi(g)\leq C \|g\|_{\widehat{\mathfrak h}(S^1)}
		\left( \int_{S^1} r \, {\rm d} \psi \;   {\rm T}_{00}(\psi) +c \right)^{\frac{1}{2}} \; ,
	\end{equation}
for all $g\in \widehat{\mathfrak h}(S^1) $. In particular, $ \int_{S^1} r \, {\rm d} \psi \;   {\rm T}_{00}(\psi) 
$ is bounded from below.
\label{3.1}
\end{proposition}

\begin{proof}
One easily obtains (see, \emph{e.g.},~\cite[Theorem~V.20]{S} or \cite[Theorem~6.4 (ii)]{DG}) that
	\begin{equation}
		\label{e3.1bb}
		\bigl({\rm d} \Gamma ( \omega )
		+1\bigr)\leq C \; \left( \int_{S^1} r \, {\rm d} \psi \;   {\rm T}_{00}(\psi) +c \right)  \: \: \hbox{for} \: \: c\gg 1 \; .
	\end{equation}
Since $ {\rm d} \Gamma  ( \omega )$ has compact resolvent on 
$\Gamma \bigl(\widehat{\mathfrak h}(S^1) \bigr)$,
it follows that 
	\begin{equation}
	\label{energy}
		 \int_{S^1} r \, {\rm d} \psi \;   {\rm T}_{00}(\psi) 
	\end{equation}
is bounded from below with a compact resolvent and hence has a
ground state. The uniqueness of this ground state of follows from a Perron-Frobenius 
argument (see e.g.~\cite[Theorem~V.17]{S}). Since 
	\[
		\omega  \geq m^\circ >0 
	\]
for some $m^\circ >0 $, we  see  that it suffices to check (\ref{e3.1b}) with
\eqref{energy} replaced by the number operator $N$, which is immediate.
To prove \eqref{e3.1c} we use   \eqref{e3.1bb} and the well known 
bound (see, e.g., \cite[Appendix]{Ge})
	\[
		\pm \varphi (g)\leq \| 
		g\|_{\widehat{\mathfrak h}(S^1)}
		\bigl({\rm d}\Gamma (\omega  )+ 1 \bigr) \; .
	\]
\end{proof}

\section{The equations of motion}
Equations of motion for interacting quantum fields on Minkowski space were derived by Glimm and  Jaffe \cite{GJ2},  
Schrader \cite{Sch} and, in $2+1$ space-time dimensions, by Feldman and~Raczka \cite{FR}. 
Formulas similar to the ones presented in this section were given 
in~\cite{FHN}.

Use the coordinate system 	
	\begin{equation*}
		x (t,  \psi) = \Lambda^{(\alpha)} \bigl(\tfrac{t}{r} \bigr)\,
				\left( \begin{array}{c}
						0 \\
						r \sin\psi  \\
						r \cos\psi   
				\end{array} \right) \; , 
		\qquad x \in W^{(\alpha)} \; , 
	\end{equation*}
with $t  \in \mathbb{C} $ and  $\psi \in   ( -\frac{\pi}{2}- \alpha,\frac{\pi}{2} -\alpha) $ and define
	\begin{equation}
		\label{covinteractingfield}
		{\Phi}_{\rm int} ( x ) 
			\doteq {\rm e}^{i \frac{t}{r} L_{\rm int}^{(\alpha)} }  \varphi ( \psi ) 
				{\rm e}^{- i \frac{t}{r} L_{\rm int}^{(\alpha)} } \;, 
				\qquad 
		x \equiv  x ( t , \psi) \; . 
	\end{equation}
We note that the interacting field ${\Phi}_{\rm int} ( x)$, defined as an operator-valued distribution at a 
space-time point $x \in dS$ does not depend on the choice of coordinates  in (\ref{covinteractingfield}).

\begin{theorem} 
The interacting quantum field ${\Phi}_{\rm int} ( x ) $, $x \in dS$,   satisfies the {\em covariant} equation of motion: 
	\[
		\Bigl( \square_{dS}+\mu^2 \Bigr) {\Phi}_{\rm int} ( x ) =   - {:} {\mathscr P}' ( {\Phi}_{\rm int} ( x )){:}_C \; .
	\]
\end{theorem}

\rm
\begin{proof} Without restriction of arbitrariness,  we may consider the case $L_{\rm int}^{(\alpha)}= L_{\rm int}^{(0)}$ and 
compute (following \cite[p.~224]{RS})
	\[ 
		[ L_{\rm int}^{(0)}, \varphi (\psi) ]  = [ L^{(0)}, \varphi (\psi) ]  
		 = - i r \,\cos (\psi) \, \pi (\psi) \;  
	\]
and	
	\[
		[ L_{\rm int}^{(0)}, [ L_{\rm int}^{(0)}, \varphi (\psi) ]]  
		= [ L^{(0)}, [ L_{\rm int}^{(0)}, \varphi (\psi) ]] 
			+ [ V^{(0)}, [ L_{\rm int}^{(0)}, \varphi (\psi) ]] \; . 
	\]
The first term on the right hand side yields 
	\begin{align}
			[ L^{(0)}, [ L_{\rm int}^{(0)}, \varphi (\psi) ]]
			&= - i r \, \cos (\psi) \, [ L^{(0)},  \pi (\psi) ]   \nonumber \\  
			&
			=  - \bigl(r\, \cos (\psi)  \partial_\psi \bigr)^2 \varphi (\psi) + r^2 \cos (\psi) \, \mu^2  \varphi (\psi) \; . 
	\end{align}
The second equality follows from partial integration (see \eqref{energymomentumdensity2}), \emph{i.e.}, 
	\begin{align*}
		\frac{1}{2} \int {\rm d} \psi' \; r \cos (\psi' ) \, & 
				\left[  \bigl( \partial_{\psi'} \varphi(\psi') \bigr)^2 ,  \pi (\psi) \right] \\
				 & = - 
				\int {\rm d} \psi' \;  \partial_{\psi'} r \cos (\psi' ) \partial_{\psi'} \varphi(\psi')   
				\left[  \varphi(\psi') ,  \pi (\psi) \right]   \; .
	\end{align*}
The second term yields
	\begin{align*} 
		[ V^{(\alpha)}, [ L_{\rm int}^{(0)}, \varphi (\psi) ] ] 
			&= - i r\, \cos (\psi) \, [ V^{(\alpha)},  \pi (\psi) ]  \nonumber \\ 
			&=  - i r^2 \cos (\psi) \int {\rm d} \psi' \cos (\psi' ) \, 
				\left[   \,{:}  P (\varphi(\psi'))  {:}_{C_0} \; ,  \pi (\psi) \right]  \nonumber \\  
			& = r^2 \cos^2 (\psi) \;  {:} P' (\varphi) {:}_{C_0}  \; . 
	\end{align*}
Set $x \equiv x ( t_\alpha, \psi)$. Use definition (\ref{covinteractingfield}) and compute
	\begin{align*}
		\frac{\partial^2}{\partial t^2} \Phi_{\rm int} (x)  
		&= - [ L_{\rm int}^{(0)}, [ L_{\rm int}^{(0)}, \Phi_{\rm int} (x) ]] 
		\\ 
		& = \Bigl( \bigl(\cos (\psi) \partial_\psi \bigr)^2 - \cos (\psi) \, \mu^2   \Bigr)
		{\rm e}^{iL_{\rm int}^{(0)} t}  \varphi(\psi) {\rm e}^{- iL_{\rm int}^{(0)} t}
			\nonumber
		\\ 
		& 	\qquad \qquad \qquad \qquad 
			 -  \cos^2 (\psi )\; {:} P' (  {\rm e}^{iL_{\rm int}^{(0)} t}  \varphi(\psi) {\rm e}^{- iL_{\rm int}^{(0)} t}) {:}_{C_0}   
			\nonumber
		\\ 
		&= \bigl(\cos (\psi) \partial_\psi \bigr)^2 \Phi_{\rm int}(x)
		- r^2 \cos^2 (\psi ) \,   \Bigr(  \mu^2 \Phi_{\rm int} (x)  -  {:} P' ( \Phi_{\rm int}(x){:}_{C} \; \Bigr) \; ,  \nonumber
	\end{align*}
\emph{i.e.}, $\bigl( \partial^2_t + \varepsilon^2 \bigr) \Phi_{\rm int}(x) = r^2 \cos^2 (\psi ) \,  {:} P' ( \Phi_{\rm int}(x){:}_{C} $ with 
	\[
		\varepsilon^2  \doteq  - (\cos \psi  \, \partial_\psi)^2 + (\cos \psi )^2 \, \mu^2 r^2 \; . 
	\]
Recalling from \eqref{varepsilon} that in the coordinates $x =x (x_0, \psi)$, 
	\[
			\square_{\mathbb{W}_1}+\mu^2
				=  \frac{1}{r^2 \cos^2 \psi}\,(\partial_t^2+  \varepsilon^2) \; , 
	\]
we arrive at the {\em equations of motion} in their covariant form
	\[
		\Bigl( \square_{dS}+\mu^2 \Bigr) {\Phi}_{\rm int} (  x  ) =   - \; {:} P' ( {\Phi}_{\rm int} ( x )) {:}_C \; .
	\]
\color{black}
\end{proof}

\appendix

\chapter{One particle structures}
\label{AKay}

Let $G$ be a group. A (classical) \emph{linear dynamical system}\index{dynamical system} 
$({\mathfrak k},\sigma, \{ T_g\}_{g \in G})$
is a real symplectic vector space\index{symplectic space} $({\mathfrak k},\sigma)$ together with a  group of \emph{symplectic 
transformations}\index{symplectic transformation} $\{ T_g\}_{g \in G}$. If~${\mathfrak h}$ is a complex Hilbert space with scalar product
$\langle \, .\, , \, .\, \rangle $, then $({\mathfrak h}, 2 \Im \langle \, .\, , \, .\, \rangle)$ is a symplectic space. 
If, in addition, a unitary representation $\{ u (g)  \}_{g \in G}$ of $G$
is given, then $( {\mathfrak h}, 2 \Im \langle \, .\, , \, .\, \rangle , \{ u (g)  \}_{g \in G}  )$
is a linear dynamical system.

\begin{definition}
Given a  linear dynamical system $({\mathfrak k},\sigma, \{ T_g\}_{g \in G})$, a symplectic transformation 
$K \colon {\mathfrak k}\to {\mathfrak h}$  defines  a \emph{one-particle quantum structure}\index{one-particle quantum structure} 
on a Hilbert space~${\mathfrak h}$, if there exists  a group of unitary operators such that 
the following diagram commutes

\bigskip
\label{page48}
\vskip -.8cm

\begin{picture}(180,140)


\put(120,100){$\longrightarrow$}
\put(60,100){$({\mathfrak k}, \sigma )$}
\put(125,110){$K$}
\put(55,70){$ {T_g}$}
\put(200,70){$u(g)$}
\put(125,50){$K$}

\put(170,100){$ ({\mathfrak h}, 2\Im \langle \, .\, , \, .\, \rangle)$}

\put(60,40){$({\mathfrak k}, \sigma )$}

\put(170,40){$({\mathfrak h}, 2\Im \langle \, .\, , \, .\, \rangle)$ \; \; . }

\put(120,40){$\longrightarrow$}

\put(75,85){\vector(0,-3){20}}
\put(190,85){\vector(0,-3){20}}

\end{picture}

\vskip -1cm
\end{definition}

\noindent
By definition,  $K$ is injective.  
Kay \cite{Kay0, Kay1, Kay2} has shown 
that one can associate several essentially unique
one-particle quantum structures to a given classical dynamical system.

\begin{definition}
Given a {\em linear dynamical system} $({\mathfrak k},\sigma, \{ T_t\}_{t \in \mathbb{R}})$, 
the symplectic transformation 
$K$  specifies
\begin {itemize}
\item [---] a \emph{one-particle structure with positive energy}\index{one-particle structure with positive energy}, 
if 
\begin{itemize}
\item[$ i.)$] $t \mapsto u(t)$ is strongly continuous and its generator $\varepsilon\ge 0$ is positive;   
\item[$ ii.)$]  $K{\mathfrak k}$ is dense in ${\mathfrak h}$.   
\end{itemize}
\goodbreak
\item [---] a \emph{one-particle $\beta$-KMS structure}\index{one-particle $\beta$-KMS structure}, if 
\begin{itemize}
\item[$ iii.)$] the map $t\mapsto \langle K {\mathfrak f},u(t)K {\mathfrak g} \rangle$, ${\mathfrak f}, {\mathfrak g} \in {\mathfrak k}$,  
is analytic in the strip $\{ t \in \mathbb{C} \mid 0< \Im t<\beta \} $, continuous at the boundary, and satisfies the one-particle 
$\beta$-KMS condition 
	\begin{equation} 
		\label{o-p-kms-condition}
			\quad \qquad \langle K{\mathfrak f},u(t+i\beta)K{\mathfrak g} \rangle 
			= \langle u(t) K{\mathfrak g},K{\mathfrak f} \rangle \; , \; \;   t\in\mathbb{R}  \; ,  
			\; {\mathfrak f}, {\mathfrak g} \in {\mathfrak k} \; ;
	\end{equation}
\item[$ iv.)$] $K{\mathfrak k}+iK{\mathfrak k}$ is dense in ${\mathfrak h}$.
\end{itemize}
\end{itemize}
Note that the Hilbert space ${\mathfrak h}$ and the one-parameter group $t \mapsto u(t)$ acting on it, 
although denoted by the same letters in $ i.)$--$ ii.)$ 
{\em and} $ iii.)$--$ iv.)$, are necessarily different in the two distinct cases.  
\end{definition}

\begin{proposition}
[Kay \cite{Kay1}, Theorems 1a \& 1b] 
\label{Kay Th}
There exists
a unique  (up to unitary equivalence) one-particle structure with positive energy 
for which  zero is not an eigenvalue of the generator of 
$t \mapsto u(t)$. Moreover, for each $\beta >0$ there exists a unique  (up to unitary equivalence) one-particle
$\beta$-KMS structure for which  zero is not an eigenvalue of the generator of 
$t \mapsto u(t)$.
\end{proposition} 

\paragraph{\it Notation.} If ${\mathfrak h}$ is a
complex vector space, then the \emph{conjugate vector space}\index{conjugate vector space} $\overline {\mathfrak h}$ is the
real vector space ${\mathfrak h}$ equipped with the complex structure $-i$. We 
denote by 
	\[ 
		{\mathfrak h} \ni h\mapsto  \overline{h} \in \overline{\mathfrak h} 
	\] 
the {\em linear} identity operator. If ${\mathfrak h}$ is a Hilbert
space, then the conjugate Hilbert space~$\bar {\mathfrak h}$ is equipped with the scalar 
product~$(\overline{h_{1}}, \overline{h_{2}})\doteq (h_{2}, h_{1})$. If $a\in {\mathcal  L}({\mathfrak h})$, 
then we denote by $\overline{a}\in {\mathcal  L}(\bar {\mathfrak h})$ the {\em linear} operator
$\overline{a}\overline{h}\doteq \overline{ah}$. 

\bigskip 
Given a one-particle structure with positive energy there exists an associated  one-particle $\beta$-KMS structure:

\begin{proposition}  \label{Kbeta} Let $(K, {\mathfrak h},\{ u(t)\}_{t \in \mathbb{R}})$ be a one-particle  structure with positive energy for a
classical dynamical system $({\mathfrak k},\sigma, \{ T_t \}_{t \in \mathbb{R}})$. If $K{\mathfrak k} \in {\mathcal  D}(\varepsilon^{-1/2})$, 
then  
\label{page49}
	\begin{align*}
		K_{\scriptscriptstyle \rm AW}  \mathfrak{f} 
		&\doteq (1+\varrho)^{\frac{1}{2}} K\mathfrak{f} \oplus  \varrho^{\frac{1}{2}} K\mathfrak{f}  \; , 
				\qquad \varrho\doteq ({\rm e}^{\beta\varepsilon}-1)^{-1} \; ,   \\  
		{\mathfrak h}_{\scriptscriptstyle \rm AW}  
		&\doteq  {\mathfrak h}\oplus \overline{{\mathfrak h}} 
		\; ,   \\  
		u_{\scriptscriptstyle \rm AW} (t) &\doteq  u(t)\oplus \overline{u(t)} \; ,  
	\end{align*}
defines a one particle $\beta$-KMS structure for  $({\mathfrak k},\sigma, \{ T_t \}_{t \in \mathbb{R}})$. 
\end{proposition}

\paragraph{\it  Remarks:} 
\begin{enumerate}
\item [$ i.)$] 
The subscripts used in $K_{\scriptscriptstyle \rm AW}$, ${\mathfrak h}_{\scriptscriptstyle \rm AW}$ and 
$u_{\scriptscriptstyle \rm AW} (t)$ pay tribute to the fundamental work of Araki and Woods~\cite{AW}.
\item [$ ii.)$] 
$({\mathfrak h}_{\scriptscriptstyle \rm AW},\{ u_{\scriptscriptstyle \rm AW} (t)\}_{t \in \mathbb{R}})$~is a 
one-particle $\beta$-KMS structure  for the dynamical system 
$({\mathfrak h}, \Im \langle \, . \,  ,  \, . \, \rangle ,\{ u(t)\}_{t \in \mathbb{R}})$,  
specified by  ${\mathcal K}_{\scriptscriptstyle \rm AW} \colon {\mathfrak h} \to {\mathfrak h}_{\scriptscriptstyle \rm AW}$, 
	\[ 
		h \mapsto (1+\varrho)^{\frac{1}{2}} h \oplus  \varrho^{\frac{1}{2}} h  \; . 
	\]
\item [$ iii.)$] 
$\overline{u(t)}=\overline{{\rm e}^{it\varepsilon}}={\rm e}^{-it 
\varepsilon}$, hence the generator of the one-parameter group
	\[
		t \mapsto \overline{u(t)} 
	\]
has  {\em negative}  spectrum. 
\item [$ iv.)$] The space
${\mathfrak h}^{\scriptscriptstyle \rm L}\doteq \{ K_{\scriptscriptstyle \rm AW} \mathfrak{f}\mid \mathfrak{f} \in  {\mathfrak k} \}$ 
is a real subspace in 
${\mathfrak h}_{\scriptscriptstyle \rm AW}$. Moreover, ${\mathfrak h}^{\scriptscriptstyle \rm L} 
+ i {\mathfrak h}^{\scriptscriptstyle \rm L}$ is dense in ${\mathfrak h}_{\scriptscriptstyle \rm AW}$
and ${\mathfrak h}^{\scriptscriptstyle \rm L}  \cap i {\mathfrak h}^{\scriptscriptstyle \rm L}  = \{0\}$. 
Thus one can define, following  
Eckmann and Osterwalder \cite{EO} (see also \cite{LRT}), a closeable operator 
	\begin{equation} 
		\begin{matrix}
			s: & {\mathfrak h}^{\scriptscriptstyle \rm L} & 
				+ & i {\mathfrak h}^{\scriptscriptstyle \rm L}  
				& \to&  {\mathfrak h}^{\scriptscriptstyle \rm L}  & 
				+ & i {\mathfrak h}^{\scriptscriptstyle \rm L}  \\
				& f & + & i g & \mapsto & f & - & i g
						\end{matrix} \; \; . 
	\end{equation}
The polar decomposition of its closure $\overline {s} = j \delta^{1/2}$
provides 
\begin{itemize}
\item[---] an anti-unitary involution (\emph{i.e.}, a \emph{conjugation}\index{conjugation})
	\begin{equation} 
		\label{eckmanj}
		\begin{matrix}
		j: & {\mathfrak h}\oplus \overline{{\mathfrak h}}  & \to&  {\mathfrak h}\oplus \overline{{\mathfrak h}}  \\
			& f  \oplus   g & \mapsto &   \overline{g}  \oplus  \overline{ f }
		\end{matrix}  \; \;  ; 
	\end{equation}
\item[---] a complex linear, positive operator $\delta^{1/2} $, such that
	\begin{equation}
		\label{eckmand}
		\delta^{it}=  u_{\scriptscriptstyle \rm AW} (- t \beta)\; ,  \qquad t \in \mathbb{R} \; . 
	\end{equation}
\end{itemize}
\eqref{eckmanj} implies $ j {\mathfrak h}^{\scriptscriptstyle \rm L}=   {\mathfrak h}^{\scriptscriptstyle \rm R}$
and  \eqref{eckmand} implies that
$\{ \delta^{it}\}_{t \in \mathbb{R}}$  leaves the subspaces~${\mathfrak h}^{\scriptscriptstyle \rm L}$ 
and~${\mathfrak h}^{\scriptscriptstyle \rm R}$ invariant. 
\item [$ v.)$] 
Sometimes we denote $K_{\scriptscriptstyle \rm AW}$ by
$K_{\scriptscriptstyle \rm AW}^{\scriptscriptstyle \rm L}$. This is useful as one encounters as well
the map 
$K^{\scriptscriptstyle \rm R}_{\scriptscriptstyle \rm AW} 
\colon {\mathfrak k} \to {\mathfrak h}_{\scriptscriptstyle \rm AW}$, 
	\begin{equation}
	\label{eqUbeta2}
	 	K^{\scriptscriptstyle \rm R}_{\scriptscriptstyle \rm AW} \mathfrak{g} 
		\doteq   \varrho^{\frac{1}{2}} K\mathfrak{g} \oplus  (1+\varrho)^{\frac{1}{2}} K\mathfrak{g} \;  ,
	 \end{equation} 
which maps ${\mathfrak k}$ to the symplectic complement 
${\mathfrak h}^{\scriptscriptstyle \rm R} \subset {\mathfrak h}_{\scriptscriptstyle \rm AW}$ 
of  ${\mathfrak h}^{\scriptscriptstyle \rm L}$.
\item [$ vi.)$] \label {sechs} 
The triple  $(K^{\scriptscriptstyle \rm R}_{\scriptscriptstyle \rm AW}, {\mathfrak h}_{\scriptscriptstyle \rm AW},
\{ u_{\scriptscriptstyle \rm AW} (t) \}_{t \in \mathbb{R} } )$ 
provides a $(-\beta)$-KMS
structure for
the linear dynamical system $({\mathfrak k},\sigma, \{T_t\}_{t \in \mathbb{R} } )$.
\end{enumerate}

\bigskip
The existence of  $vi.)$ motivated Kay \cite{Kay1, Kay2}  to investigate the possibility of doubling the 
classical dynamical system as well:

\begin{definition}
\label{dcldsy}
Let $\underline { {\mathfrak k}} ={\mathfrak k}_{\scriptscriptstyle \rm R} 
\oplus {\mathfrak k}_{\scriptscriptstyle \rm L}  $ be the direct sum of  two symplectic subspaces 
${\mathfrak k}_{\scriptscriptstyle \rm R} $ and~$  {\mathfrak k}_{\scriptscriptstyle \rm L}  $ 
such that 
	\[
		\underline { \sigma} (  \mathfrak{f},   \mathfrak{g}) 
			= 0 \quad \hbox{if}  \quad \mathfrak{f} \in {\mathfrak k}_{\scriptscriptstyle \rm L}
			\quad \hbox{and}  \quad  \mathfrak{g} \in 
			{\mathfrak k}_{\scriptscriptstyle \rm R} \; . 
	\]
Let $\{ \underline { T}_{\, t} \}_{t \in \mathbb{R}}$ be a one-parameter group of symplectic maps,  
which leaves~${\mathfrak k}_{\scriptscriptstyle \rm L} $ and ${\mathfrak k}_{\scriptscriptstyle \rm R}$ invariant.
Furthermore, let $ \underline{\imath}$ be an anti-symplectic involution such that 
	\[
		[ \, \underline  {T}_{\, t} , \underline {\imath}\, ]=0 \quad \hbox{and} \quad
		\underline {\imath} \, {\mathfrak k}_{\scriptscriptstyle \rm L} 
		= {\mathfrak k}_{\scriptscriptstyle \rm R} \; .
	\]
The quadruple $(\, \underline{{\mathfrak k}}, \underline {\sigma}, \{ \underline { T}_{\, t}\}_{t \in \mathbb{R} } , 
\underline {\imath}\,)$
is called a {\em double} (classical) linear dynamical system. 
\end{definition}

\goodbreak
It follows that  
$\underline {\imath}\, {\mathfrak k}_{\scriptscriptstyle \rm R} = {\mathfrak k}_{\scriptscriptstyle \rm L}$.
In other words, the following diagram commutes:
\vskip -.5cm
\begin{picture}(200,140)


\put(120,100){$\longrightarrow$}
\put(60,100){$({\mathfrak k}_{\scriptscriptstyle \rm L}, \underline{ \sigma})$}
\put(125,110){$\underline{\imath}$}
\put(55,70){$ {\underline{T}_{\, t}}$}
\put(200,70){$\underline{T}_{\, t}$}
\put(125,50){$\underline{\imath}$}

\put(170,100){$ ({\mathfrak k}_{\scriptscriptstyle \rm R}, \underline{\sigma})$}

\put(60,40){$({\mathfrak k}_{\scriptscriptstyle \rm L}, \underline{\sigma})$}

\put(170,40){$({\mathfrak k}_{\scriptscriptstyle \rm R}, \underline{\sigma})$ \; \; . }

\put(120,40){$\longleftarrow$}

\put(75,85){\vector(0,-3){20}}
\put(190,85){\vector(0,-3){20}}

\end{picture}

\vskip -.8cm
\goodbreak
\begin{definition} 
\label{dbops} (Kay \cite{Kay1}, Def.~3).
A \emph{double $\beta$-KMS one-particle structure}\index{double $\beta$-KMS one-particle structure}, 
\emph{i.e.}, a quadruple  $(\underline{K},  {\mathfrak h} ,  
\{ \underline{\delta}^{-i t / \beta}\}_{t \in \mathbb{R}}  , \underline{j} )$,
associated to a double  linear classical dynamical system 
$(\underline{{\mathfrak k}}, \underline{\sigma}, \{ \underline{T}_{\, t} \}_{t \in \mathbb{R}}, 
\underline{\imath}\, )$ consists of 
\begin{itemize}
\item [$ i.)$] a complex Hilbert space ${\mathfrak h}$; 
\item [$ ii.)$] a real linear symplectic map $\underline{K} \colon \underline{ {\mathfrak k} }\to  {\mathfrak h} $ 
such that $\underline{ K}{\mathfrak k}_{\scriptscriptstyle \rm L} 
+ i \underline{ K }{\mathfrak k}_{\scriptscriptstyle \rm L} $ is dense in ${\mathfrak h}$; 
\item [$ iii.)$] a strongly continuous unitary group 
$t \mapsto \underline{\delta}^{-i t / \beta} $ such that 
\begin{itemize}
\item[---] $\underline{\delta}^{-i t / \beta} \circ \underline{ K } = \underline{ K }\circ  \underline{ T}_{\, t}  $ 
for all\footnote{This is in agreement with 
\eqref{eckmand}.} $t \in \mathbb{R} $; 
\item[---] $\underline{K }{\mathfrak k}_{\scriptscriptstyle \rm L} 
+ i \underline{ K }{\mathfrak k}_{\scriptscriptstyle \rm L} 
\subset {\mathscr D} \bigl( \underline{\delta}^{1 /2} \bigr)$;
\end{itemize}
\item [$ iv.)$] an anti-unitary operator $\underline{j}$ such that $ \underline{j} \circ \underline{ K }
= \underline{ K} \circ \underline \imath$ on $\underline{ {\mathfrak k}}$ and 
	\[ 
		\underline{j} \underline{\delta}^{1 /2} f 
		=   f \qquad  \forall f \in \underline{ K }{\mathfrak k}_{\scriptscriptstyle \rm L} \; . 
	\]
\end{itemize}
\end{definition}

\bigskip
\noindent
The operator $\underline{\delta}$ is positive,  
$\underline {K} {\mathfrak k}_{\scriptscriptstyle \rm R} +  i \underline {K} {\mathfrak k}_{\scriptscriptstyle \rm R} $ 
is dense in~$ {\mathfrak h}$,  
	\[ 
		\underline { K} {\mathfrak k}_{\scriptscriptstyle \rm R} + i \underline {K} 
		{\mathfrak k}_{\scriptscriptstyle \rm R}  \subset {\mathscr D} \bigl( \underline{\delta}^{-1 /2}\bigr) 
	\]
and $\underline{j} \, \underline{\delta}^{-1/2}  g =  g $ for all 
$g \in \underline{ K }{\mathfrak k}_{\scriptscriptstyle \rm R} $. 

\goodbreak
\begin{theorem}[Kay \cite{Kay1}, Theorem 2] 
\label{ThB1}
There exists a unique, up to unitary equivalence, double $\beta$-KMS-structure for which
the generator~$\underline{\varepsilon}$ of 
the one parameter group 
	\[
\underline{\delta}^{i t } ={\rm e}^{-i t \beta \underline{\varepsilon}}, \qquad \beta >0 \; , 
	\] 
has no zero eigenvalue. 
\end{theorem}

\chapter{Fock space}
\label{Fockspace}

Consider\footnote{We follow \cite[Vol.~II]{BR}.} a  Hilbert space ${\mathfrak h} $ with 
scalar product $\langle\, . \, , \, . \, \rangle$.
Let $ \Gamma^{(n)} ( {\mathfrak h} ) $, $n \in \mathbb{N}$, 
be the n-fold totally symmetric tensor product $\otimes_s$ of~${\mathfrak h} $ with itself. The elements of $ \Gamma^{(n)} ( {\mathfrak h} ) $
are of the form
	\[
		P_+ ( f_1 \otimes \ldots \otimes f_n) 
		\doteq\frac{1}{n!} \sum_\pi  f_{\pi_1} \otimes f_{\pi_2} \otimes \ldots \otimes  f_{\pi_n}\; , 
		\qquad f_1, \ldots, f_n \in {\mathfrak h} \; .
	\]
The sum is over all 
permutations $\pi \colon (1, 2, \ldots, n) \mapsto (\pi_1 , \pi_2 , \ldots , \pi_n)$. 
In other words, the symmetrisation operator $P_+$  takes care of the necessary symmetrisation required. 

\begin{definition}
The symmetric Fock-space 
$\Gamma ( {\mathfrak h} ) $ over ${\mathfrak h} $ is  the 
direct sum of the $n$-particle spaces:
	\[
		\Gamma ( {\mathfrak h} )  \doteq\oplus_{n= 0}^{\infty}  \Gamma^{(n)} ( {\mathfrak h} ) \;  , 
	\]
with $ \Gamma^{(0)} ( {\mathfrak h} ) \doteq \mathbb{C}$.  
\end{definition}

The vectors with finitely many components unequal to zero form a dense subspace  
	\[
		\Gamma^\circ ({\mathfrak h} ) \doteq {\rm Span} \, \Bigl\{ \oplus_{n= 0}^{N} \Gamma^{(n)} 
		( {\mathfrak h} ) \mid  N \in \mathbb{N} \Bigr\}  
	\]
in~$\Gamma ( {\mathfrak h} )$. The vector $\Omega \doteq (1,0, 0, \ldots)$ is called the Fock vacuum vector. 

For $f\in  {\mathfrak h} $ define the {\it creation operator} ${ a}^* ({f}) \colon \Gamma^\circ ( {\mathfrak h} )
\to \Gamma^\circ ( {\mathfrak h} )$ by  
	\[
		{ a}^* ({f}) \Psi{(n)} \doteq   \sqrt {n+1} \, \, 
		f \otimes_s  \Psi{(n)} \;  , \qquad \Psi{(n)} \in \Gamma^{(n)} ( {\mathfrak h} ) \; .
	\]
The operator ${ a} ({f}) $ denotes the adjoint of  ${ a}^* ({f})$, and is called the annihilation operator. 
Both ${ a}({f})$ and ${ a}^* ({f})$ are defined on  $\Gamma^\circ ( {\mathfrak h} )$ and can be extended 
to densely defined closed, unbounded operators on 
$\Gamma ( {\mathfrak h} )$. 
 
The map ${f} \mapsto { a}^* ({f})$ is linear, while
the map ${f}  \mapsto { a} ({f})$ is anti-linear. They satisfy the canonical commutation relations:
	\[
		\bigl[ { a} ({f}) , { a} ({g})  \bigr] = \bigl[ { a}^* ({f}) , { a}^* ({g}) \bigr] = 0 
	\]
and
	\[
		\bigl[ { a} ({f}) , { a}^* ({g})  \bigr] 
 		= \langle f    ,  g \rangle  \qquad \forall  {f}, {g}  \in   {\mathfrak h}  \; . 
 	\]
By applying the creation operators ${ a}^* ({f})$ to $\Omega$ we 
get $\Gamma^\circ ({\mathfrak h} )$ and by closure
all of~$\Gamma ({\mathfrak h} )$:
	\[
		{ a}^* ({f}_1) \ldots { a}^* ({f}_n) \Omega 
		= \sqrt {n!} \, \Bigl(  {f}_1    \otimes_s \ldots \otimes_s {f}_n \Bigr) 
		\in \Gamma^{(n)} ( {\mathfrak h} )  \; , \quad {f}_1 , \ldots, {f}_n \in  {\mathfrak h} \; .
	\]
The symmetric operator ${ a}^*({f}) + { a}({f}) $ is essentially self-adjoint 
on~$\Gamma^\circ ( {\mathfrak h} )$, its closure is denoted by 
	\[ 
		\Phi (f) \doteq \frac{1}{ \sqrt{2}}\bigl ( { a}^*({f}) + { a}({f}) \bigr)^- \;  . 
	\]
The field operators $ \Phi (f)$ satisfy, in the sense of quadratic forms on ${\mathscr D}(\Phi (f))\cap {\mathscr D}(\Phi (g))$, 
the commutation relations 
	\[
		[\Phi(f),\Phi (g)]= i \,  \Im \langle f , g \rangle \; ,\qquad f , \:g \in {\mathfrak h} \; .
	\]
The Weyl operators ${ W}_F ({f}) \doteq {\rm e}^{ i  \Phi (f) }$ satisfy 
	\[ 
		{  W}_F ({f}) {  W}_F ({g} )  
		=  { \rm e}^{   i \Im \langle g ,   f \rangle } {  W}_F ( {f}  + {g} )  \;  ,\qquad f , \:g \in {\mathfrak h} \;  .
	\]
${  W}_F ( {f})$ is related to  the exponentials ${\rm e}^{i { a}^*  ( f ) } $ and ${\rm e}^{i { a} ( f ) } $ by
	\[
		 {  W}_F (f) = {\rm e}^{i  {   a}^*(f)}\; {\rm e}^{i  {  a}(f)} {\rm e}^{-\frac12 \|f \|^2}     
	\]
on $\Gamma^\circ ( {\mathfrak h} )$. Thus 
	\[
		(\Omega, {W}_F (f) \Omega) = {\rm e}^{-\frac12 \|f \|^2}\qquad  \forall f \in {\mathfrak h} \;  .
	\]
Moreover, the Weyl operators $\{ {W}_F (f) \mid f \in {\mathfrak h} (dS) \}$ 
generate all of $ {\mathcal B} \bigl(\Gamma ({\mathfrak h} )\bigr)$.

Given a selfadjoint operator $H$ acting on the one-particle space ${\mathfrak h} $, one can define operators $H_n$ acting
on the $n$-particle space $\Gamma^{(n)} ({\mathfrak h})$ by setting $H_0 \doteq 0$ and
	\[
		H_n \bigl( P_+ ( f_1 \otimes \ldots \otimes f_n) \bigr) \doteq
		P_+ \Bigl( \sum_i f_1 \otimes f_2 \otimes \ldots \otimes H f_i 
		\otimes \ldots \otimes f_n \Bigr) 
	\]
for all $f_i \in {\mathscr D} (H) \subset {\mathfrak h} $.  The operator $H_n $ extends to a 
selfadjoint operator on~$\Gamma^{(n)} ( {\mathfrak h})$. 
The direct sum of all $H_n$ is symmetric and therefore closable, and essentially selfadjoint, because there exists 
a dense set of analytic vectors, which is formed by the finite sums of symmetrised products of analytic 
vectors of $H$. The selfadjoint closure of the direct sum $\oplus_{n \in \mathbb{N}_\circ} H_n$ of  $H_n$ 
is called the second quantisation of~$H$.  It is denoted  by 
	\[
		{\rm d}\Gamma (H) \doteq \overline { \oplus_{n \in \mathbb{N}_\circ} H_n } \; . 
	\]

If $U$ is a {\em unitary operator} acting on ${\mathfrak h}$, then  $U_n$ 
acting on  $\Gamma^{(n)} ({\mathfrak h})$ is defined by 
$U_0 \doteq \mathbb{1} $ and
	\[
		U_n \bigl( P_+ ( f_1 \otimes \ldots \otimes f_n) \bigr) 
		\doteq
		P_+ \Bigl( U f_1 \otimes U f_2 \otimes \ldots  \otimes U f_n \Bigr) 
	\]
and continuous extension. The second quantisation of $U$ is  
	\[
		\Gamma (U) \doteq \oplus_{n \in \mathbb{N}_\circ} U_n \;  . 
	\]
$\Gamma (U)$ is a unitary operator acting on $\Gamma ({\mathfrak h})$. 

\begin{lemma}
If $t \mapsto U_t  = {\rm e}^{i t  H } $ is a strongly continuous group of unitary operators on ${\mathfrak h}$, 
then $\Gamma (U_t) = {\rm e}^{i t {\rm d}\Gamma (H)} $ holds on $\Gamma({\mathfrak h})$. 
\end{lemma}

For the convenience of the reader we also recall the following results. 

\begin{theorem}[Araki \cite{A1}, Theorem 1] 
\label{araki}
Let ${\mathfrak h}$ be a Hilbert space and let~${\mathfrak d}$ be a real subspace of ${\mathfrak h}$. Let 
${\mathfrak W}({\mathfrak d})\subset {\mathfrak W}({\mathfrak h})$ denote the abstract 
Weyl algebra\footnote{See, \emph{e.g.}, Eq.~\eqref{weylalgebra} with the symplectic form $\sigma$ 
given by twice the imaginary part of the scalar product in the Hilbert space ${\mathfrak h}$.}
generated by $\{W(h) \mid h\in {\mathfrak d}\}$  and let $\pi \colon W(h) \to W_F(h) \in 
{\mathcal B}(\Gamma({\mathfrak h}))$ be the Fock representation.    Then
	\begin{align}
	\bigcap_{\alpha} \pi  \bigl( {\mathfrak W}({\mathfrak d}_{\alpha}) \bigr)''
			& = \pi  \bigl({\mathfrak W}\left(\cap_{\alpha} {\mathfrak d}_{\alpha} \right) \bigr)''  \; ,
			\nonumber \\
		\label{et.2n}
		\bigvee_{\alpha}\pi  \bigl({\mathfrak W}({\mathfrak d}_{\alpha})\bigr)''
		& = \pi \bigl({\mathfrak W}\left(\vee_{\alpha} {\mathfrak d}_{\alpha} \right) \bigr)'' \; , 
	\end{align}
and
	\begin{align}
		\label{et.1}
		\pi  \bigl({\mathfrak W}({\mathfrak d})\bigr)'= \pi  \bigl({\mathfrak W}({\mathfrak d}^{\perp})\bigr)'',
	\end{align}
where ${\mathfrak d}_{\alpha}$ is a family of real subspaces of ${\mathfrak h}$ and ${\mathfrak d}^{\perp}$ is the vector space 
orthogonal to~${\mathfrak d}$ for the symplectic form $\sigma(h_{1}, h_{2})= 2 \Im \langle h_{1}, h_{2} \rangle_{\mathfrak h}$.
\end{theorem}

\begin{theorem}[Leyland, Roberts and Testard \cite{LRT}, Theorem I.3.2]
Let ${\mathfrak h}$ be a Hilbert space and let~${\mathfrak d}$ be a real subspace of ${\mathfrak h}$. Let 
$\overline{\mathfrak d}$ denote the closure of ${\mathfrak d}$. Furthermore, let
${\mathfrak W}({\mathfrak d})\subset {\mathfrak W}({\mathfrak h})$ denote the abstract Weyl algebra
generated by $\{W(h) \mid h\in {\mathfrak d}\}$  and let $\pi \colon W(h) \to W_F(h) \in 
{\mathcal B}(\Gamma({\mathfrak h}))$ be the Fock representation.  It follows that
\begin{itemize}
\item[$ i.)$] $\pi \bigl( {\mathfrak W}({\mathfrak d}) \bigr)"= \pi \bigl( {\mathfrak W}(\overline{\mathfrak d}) \bigr)$;
\item[$ ii.)$] $\Omega$ is cyclic for $\pi \bigl({\mathfrak W}({\mathfrak d}) \bigr)"$ if and only if ${\mathfrak d} + i {\mathfrak d}$ is dense in 
${\mathfrak h}$;
\item[$ iii.)$] $\Omega$ is separating for $\pi \bigl({\mathfrak W}({\mathfrak d}) \bigr)"$ 
if and only if $\overline{\mathfrak d} \cap i \overline{\mathfrak d} = \{0 \}$;
\item[$ iv.)$] $\pi \bigl({\mathfrak W}({\mathfrak d}) \bigr)'' = \pi \bigl({\mathfrak W}({\mathfrak d}^{\perp})\bigr)'$;
\item[$ v.)$] $\pi \bigl({\mathfrak W}({\mathfrak d}_1) \bigr)" \vee \pi \bigl({\mathfrak W}({\mathfrak d}_2)\bigr)"= 
\pi \bigl({\mathfrak W}({\mathfrak d}_1 + {\mathfrak d}_2) \bigr)"$;
\item[$ vi.)$] $\pi \bigl({\mathfrak W}({\mathfrak d}_1)\bigr)" 
\cap \pi \bigl({\mathfrak W}({\mathfrak d}_2) \bigr)"= 
\pi \bigl({\mathfrak W}(\overline{{\mathfrak d}_1} \cap \overline{{\mathfrak d}_2})\bigr)"$ and 
consequently $\pi \bigl({\mathfrak W}({\mathfrak d})\bigr)" $ is a factor if and only if 
$\overline{\mathfrak d} \cap  {\mathfrak d}^\perp = \{ 0 \}$;
\end{itemize}
As in the previous theorem, ${\mathfrak d}^{\perp}$ denotes the symplectic complement of~${\mathfrak d}$ in ${\mathfrak h}$.
\end{theorem}

\chapter{Tuboids}

\label{deftub}
\setcounter{equation}{0}
A \emph{tuboid}\index{tuboid} for the de Sitter space is a subset of $dS_{\mathbb{C}}$ which  
is $i.)$~bordered by  real de Sitter space $dS$ and $ii.)$ allows boundary values on $dS$ of functions 
holomorphic in the tuboid to be controlled by methods of complex analysis (in several variables). We 
proceed in several steps, following closely \cite{BM}. 

\goodbreak
\begin{definition}
A \emph{profile}\index{profile}~${\mathcal P}$ is an open subset of the tangent bundle $TdS$ of the form 
	\[
		{\mathcal P} = \bigcup_{ x \in dS} (  x , {\mathcal P}_{ x}) \; , 
	\]
where each fibre\index{fibre} ${\mathcal P}_{ x}$ is a non-empty cone with apex at the origin in $T_{ x}dS$.
\end{definition}

\goodbreak
\begin{definition}[Bros and Moschella \cite{BM}, Def.~A.1]
A diffeomorphism $\Xi$ is an \emph{admissible local diffeomorphism}\index{admissible local diffeomorphism}  
at a point $ x_\circ \in dS$, if it maps a neighbourhood ${\mathcal N}_{\scriptscriptstyle TdS} ( x_\circ,  0 ) $ 
of $(  x_\circ ,  0 )$ in $TdS$ onto a corresponding neighbourhood 
	\[
		{\mathcal N}_{\mathbb{C}}( x_\circ)
			\doteq \Xi \left({\mathcal N}_{\scriptscriptstyle TdS} ( x_\circ,  0 )\right)
	\]
in $dS_{\mathbb{C}}$, considered as a 4-dimensional $C^\infty$-manifold, in such a way that the 
following properties hold: 
\begin{itemize}
\item[$i.)$] $\Xi ( x,  0 ) =  x \in {\mathcal N}_{\mathbb{C}}( x_\circ)$ if $( x,  0) 
\in {\mathcal N}_{\scriptscriptstyle TdS} ( x_\circ,  0 )$;
\item[$ii.)$] for all $( x ,  y) \in {\mathcal N}_{\scriptscriptstyle TdS} ( x_\circ, 0 )  $ with $ y \ne  0$ 
the differentiable function 
	\[
		t \mapsto f(t) \doteq \Xi ( x,  t  y ) \in dS_\mathbb{C}
	\]
is such that 
	\[
		\frac{1}{i} \left(\frac{{\rm d}f}{{\rm d}t} \right)_{\upharpoonright t= 0} = \alpha   y\;  , 
		\qquad \alpha > 0 \; . 
	\]
\end{itemize}
\end{definition}

\bigskip
In order to define admissible local diffeomorphisms, which respect the causal structure of $dS$, 
it is convenient to consider the \emph{projective representation} 
	$
		\dot T dS \doteq \bigcup_{ x \in dS} \left(  x, \dot T_{ x} dS \right) 
	$
of $T dS$, with 
	\[
		\dot T_{ x} dS \doteq   \bigl( T_{ x} dS\setminus \{  0 \} \bigr) / \mathbb{R}^+ \; . 
	\]
The image of each point $ y \in T_{ x} dS$ in ${\dot T}_{ x} dS$ is 
$\dot { y}= \{ \lambda  y \mid \lambda >0 \}$. The complement $\dot {\mathcal P}'$
of $\dot {\mathcal P} \doteq   \bigl( {{\mathcal P}}_{ x} \setminus \{  0 \} \bigr) / \mathbb{R}^+ $ 
in $\dot T dS$ is the open set 
	\[
		\dot {\mathcal P}^c \doteq \dot T dS \setminus \overline{\dot {\mathcal P}}  \;  .
	\]
This set equals $\bigcup_{ x \in dS} ( x, \dot {\mathcal P}_{ x}^c ) $, where 
$\dot {\mathcal P}_{ x}^c = \dot T_{ x} dS \setminus \overline{\dot {\mathcal P}_{ x}} $. 
Taking advantage of these notions, one can now define tuboids for the de Sitter space:

\begin{definition}[Bros and Moschella \cite{BM}]
A domain ${\mathcal T}$ (\emph{i.e.}, a connected open subset) in~$dS_{\mathbb{C}}$ is called a 
{\em tuboid} with profile ${\mathcal P}$ above $dS$ if it satisfies the following property: for 
every point $ x_\circ$ in $dS$, there exists an admissible local diffeomorphism $\Xi$ at 
$ x_\circ$ such that 
\begin{itemize}
\item[$i.)$] every point $( x_\circ, \dot { y}_1)$   in $\dot {\mathcal P}$ admits a compact 
neighbourhood ${\mathcal K}( x_\circ, \dot { y}_1)$ in~$\dot {\mathcal P}$ and, 
in the sequel, a sufficiently small neighborhood ${\mathcal N}_{\scriptscriptstyle TdS} ( x_\circ,  0 )$ 
of~$( x_\circ,  0 )$ in $TdS$ such that 
	\[
		\Xi \left(   \{ ( x,   y) \in {\mathcal N}_{\scriptscriptstyle TdS} ( x_\circ,  0) 
		\mid ( x, \dot  y) \in {\mathcal K}( x_\circ, \dot  {y}_1)   \, \} \right) \subset {\mathcal T} \; ;
	\]
\item[$ii.)$] every point $( x_\circ, \dot { y}_2)$ in ${\dot {\mathcal P}}^c$ admits a compact 
neighbourhood ${\mathcal K}^{c}( x_\circ, \dot { y}_2) $ in~$\dot {{\mathcal P}}^c$ and, 
in the sequel, a sufficiently small neighborhood ${\mathcal N}_{\scriptscriptstyle TdS}^c ( x_\circ,  0 )$ 
of~$( x_\circ,  0 )$ in $TdS$ such that 
	\[
		\Xi \left( \bigl\{ ( x, { y}) \in {\mathcal N}_{\scriptscriptstyle TdS}^{c} ( x_\circ,  0 )
		\mid ( x, \dot { y}) \in {\mathcal K}^{c}( x_\circ, \dot { y}_2), \; \bigr\} \right) 
		\cap {\mathcal T} = \emptyset \; .
	\]
\end{itemize}
\end{definition}

In $i.)$ and $ii.)$ ${\mathcal N}_{\scriptscriptstyle TdS}  ( x_\circ,  0 )$ and  
${\mathcal N}_{\scriptscriptstyle TdS}^{c} ( x_\circ,  0 )$  may depend on $\dot { y}_1$ and $\dot { y}_2$, respectively.

\chapter{Sobolev spaces on the circle and on the sphere}
\label{apC}
If $h \in L^2 (S^1, {\rm d} \psi)$, then $h$ has a \emph{Fourier series}\index{Fourier series} 
	\begin{equation}
	\label{sobolev}
	h (\psi) = \sum_{k\in {\mathbb Z}} a_k {\rm e}^{ik \psi} \; , 
	\qquad a_k = \frac{1}{2 \pi} \int_{S^1} {\rm d} \psi \; h (\psi) {\rm e}^{-ik  \psi} \;  . 
	\end{equation}
The infinite sum on the r.h.s.~ converges  in $L^2 (S^1, {\rm d} \psi)$. In fact,  
the infinite sum $ \sum_{k\in {\mathbb Z}} a_k {\rm e}^{ik \psi} $ exists, iff $|a_k|= o(k^{-N})$ for all $N \in \mathbb{N}$.

By the Weierstra\ss'\ approximation theorem 
the polynomials  
	\[
	\sum_{k=-N}^N a_k {\rm e}^{ik \psi} \; , \qquad N \in {\mathbb N} \; , 
	\]
are dense in the sup norm in $C(S^1)$. Parseval's identity states that 
	\[
	 \sum_{k\in {\mathbb Z}} | a_k |^2 = \frac{1}{2 \pi} \int_{S^1} {\rm d} \psi \; | h (\psi) |^2 \; . 
	 \]
In case $h \in C^1 (S^1)$, this implies that the Fourier series converges uniformly and absolutely.

\begin{definition} Let $0 \le p \le \infty$.  
The Sobolev space of order $p$ is given by
	\[
		{\mathbb H}^p (S^1)  \doteq \Bigl\{ h \in L^2 (S^1) \mid 
		\sum_{k\in {\mathbb Z}}(1+k^2)^p |a_k |^2 < \infty \Bigr\} \; , 
	\]
where the $\{ a_k \}$ are the Fourier coefficients of $h$, see \eqref{sobolev}. 
\end{definition}	

${\mathbb H}^p (S^1)$ is a Hilbert space with the inner product 
	\[
	\Bigl\langle \, \sum_{j\in {\mathbb Z}} a_j {\rm e}^{i j \psi}  \; , 
	\; \sum_{k\in {\mathbb Z}} b_k {\rm e}^{ik \psi}g \Bigr\rangle_{{\mathbb H}^p (S^1)} 
	= \sum_{k\in {\mathbb Z}}(1+k^2)^p \; a_k \overline {b_k} 
	\]
for $h, g \in {\mathbb H}^p (S^1)$ with Fourier coefficients $\{a_j \}$, $\{ b_k \}$, respectively. The norm is given by 
	\[
	\| h \|_{{\mathbb H}^p (S^1)}  = \Bigl( \, \sum_{k\in {\mathbb Z}} (1+ k^2)^p \; |a_k|^2 \Bigr)^{1/2} \; . 
	\]
The trigonometric polynomials are dense in ${\mathbb H}^p (S^1)$.

\begin{definition} 
\label{sob-circ}
For $0 < p < \infty$, we denote by ${\mathbb H}^{-p}(S^1)$ the dual space of ${\mathbb H}^p (S^1)$, 
\emph{i.e.}, the space of bounded linear functionals on ${\mathbb H}^p (S^1)$. 
\end{definition}	

For $\xi \in {\mathbb H}^{-p}(S^1)$ we have 
	\[
	\| \xi \|_{{\mathbb H}^{-p} (S^1)}  = \Bigl( \, \sum_{k\in {\mathbb Z}} (1+ k^2)^{-p} \; |b_k|^2 \Bigr)^{1/2} \; , 
	\]
where $b_k = \xi ({\rm e}^{i k \psi})$. Furthermore, for each sequence $\{ b_k \}$ satisfying 
	\[
	 \sum_{k\in {\mathbb Z}} (1+ k^2)^{-p} \; |b_k|^2 < \infty \; ,   
	\]
there exists a bounded linear functional $\xi \in {\mathbb H}^{-p} (S^1)$ with $ b_k = \xi ({\rm e}^{ik \psi})$. 

\goodbreak
\begin{proposition}
The elements in ${\mathbb H}^p$ share the following properties:
\begin{itemize}
\item[$ i.)$] If $p <0$, then the elements in ${\mathbb H}^p$ are generalised functions;
\item[$ ii.)$] If $p >1/2$, then the functions $ f \in {\mathbb H}^p$ are continuous;
\item[$ iii.)$] If $p \ge 1$, then the functions $ f \in {\mathbb H}^p$ are differentiable almost everywhere.
\end{itemize}
\end{proposition}

Next we consider the sphere. The surface element is 
	\[
		{\rm d} \Omega = \cos \psi {\rm d} \psi  {\rm d} \theta \; . 
	\]
We denote by $L^2 (S^2,  {\rm d} \Omega )$ the set of measurable functions
$f$ on the sphere $S^2$ for which
	\[
		\| f \|_{L^2 (S^2,  {\rm d} \Omega )}^2  \doteq \int_{S^2} {\rm d} \Omega \; | f (\theta, \psi) |^2 < \infty \; . 
	\]
A function $f \in L^2 (S^2, \cos \psi {\rm d} \psi  {\rm d} \theta )$ can be expanded, in the $L^2$-sense, 
into its Fourier (Laplace) series (with respect to spherical harmonics)
where 
	\begin{equation}
	\label{F-S2}
		\widetilde f_{\ell, k} \doteq \int_{S^2} {\rm d} \Omega \; f (\theta, \psi) \, \overline{Y_{\ell, k} (\theta, \psi)} \; . 
	\end{equation}

\begin{definition}
\label{sobolev-S2}
The Sobolev $\mathbb{H}^p (S^2)$, $p \ge 0$, 
is the closure of the set of $C^\infty (S^2)$ functions with respect to the norm
	\[
		\| f \|_{\mathbb{H}^p (S^2)} \doteq \left( \sum_{\ell = 0}^\infty \sum_{k= - \ell}^{\ell} 
		\left( \ell + \tfrac{1}{2} \right)^{2p} | \widetilde f_{\ell, k} |^2 \right)^{1/2} \; . 
	\]
\end{definition}

The space $\mathbb{H}^p (S^2)$ is a Hilbert space with inner product 
	\[
		\langle f, g \rangle_{\mathbb{H}^p (S^2)} \doteq 
		\sum_{\ell = 0}^\infty \sum_{k= - \ell}^{\ell } 
		\left( \ell + \tfrac{1}{2} \right)^{2p} \; \overline{ \widetilde f_{\ell, k} } \, \widetilde g_{\ell, k} \; , \qquad f, g \in \mathbb{H}^p (S^2) \; .
	\]
By construction, $\mathbb{H}^0 (S^2)= L^2 (S^2,  {\rm d} \Omega )$. 	

\begin{definition} For $0 < p < \infty$, we denote by ${\mathbb H}^{-p}(S^2)$ the dual space of ${\mathbb H}^p (S^2)$, 
\emph{i.e.}, the space of bounded linear functionals on ${\mathbb H}^p (S^2)$. 
\end{definition}	

For $\xi \in {\mathbb H}^{-p}(S^2)$ we have 
	\[
		\| \xi \|_{{\mathbb H}^{-p} (S^2)}  = \Bigl( \, \sum_{\ell = 0}^\infty \sum_{k= - \ell}^{ \ell } 
	 	\left( \ell + \tfrac{1}{2} \right)^{-2p} \; | b_{\ell, k} |^2 \Bigr)^{1/2} \; , 
	\]
where $b_{\ell, k} = \xi (Y_{\ell, k})$. Furthermore, for each sequence $\{ b_{\ell, k} \}$ satisfying 
	\[
	 	\sum_{\ell = 0}^\infty \sum_{k= - \ell}^{ \ell } 
	 	\left( \ell + \tfrac{1}{2} \right)^{-2p} \; | b_{\ell, k} |^2 < \infty \; ,   
	\]
there exists a bounded linear functional $\xi \in {\mathbb H}^{-p} (S^2)$ with $b_{\ell, k} = \xi (Y_{\ell, k})$.

\chapter{Some identities involving Legendre functions}

In the sequel, we will use the following well-known properties of the Gamma function:
	\begin{align}
			\Gamma(z+1) & =  z  \Gamma(z) \; ,
				\label{eq:gamma-1} \\
			\Gamma(z)\Gamma(1-z) & =  \frac{\pi}{\sin(\pi z)} \; ,
				\label{eq:gamma-2} \\
			\Gamma(z) \Gamma(-z) & =  -\frac{\pi}{z\sin(\pi z)}\; ,
				\label{eq:gamma-3} \\
			\Gamma(2z)	& = \frac{2^{2z-1}}{\sqrt{\pi}}\; \Gamma(z)\Gamma\left(z+\tfrac{1}{2}\right) \; ,
				\label{eq:gamma-4} \\
			\Gamma\big(\overline{z}\big) & = \overline{\Gamma (z)} \; .
				\label{eq:gamma-5}
	\end{align}
They are valid except when the arguments are non-positive integers.

The \emph{Legendre function}\index{Legendre function} $P_s$ solves \cite[8.820]{Grad} the differential equation 
	\[
		\frac{{\rm d}}{{\rm d} z} (1- z^2)\frac{{\rm d}}{{\rm d} z} P_{s}(z) + s(s+1) P_{s}(z) =0 \; . 
	\] 
It is analytic in $\mathbb{C} \setminus (-\infty, \;
-1)$, that means, it has a cut in $(-\infty,\; -1)$.  

\begin{remark}
Setting $z= - \cos \psi$ we find
	\[ \frac{1}{\sin_\psi } 
						\frac{\partial}{\partial \psi}   \sin_\psi \frac{\partial}{\partial \psi} 	
						P_{s}( - \cos \psi) + s(s+1) P_{s}( - \cos \psi) =0 \; . 
	\] 
The \emph{associated Legendre functions} 
	\begin{align*}
		P_s^k (\cos \psi) & \doteq (-1)^k (\sin \psi)^k \tfrac{ {\rm d}^k }{ {\rm d} (\cos \psi)^k} \bigl( P_s (\cos \psi) \bigr) \\
		P_s^{-k} & \doteq (-1)^k \tfrac{(s-k) !}{(s+k)! } P_s^{k} \; ,  
		\qquad k= 0, 1,2, \ldots \; ,
	\end{align*}
are analytic in $\mathbb{C}\setminus (-\infty, \; +1)$. 
\end{remark}
\goodbreak

\begin{lemma}
The function 
	\[
		S(z)\doteq \sqrt{z^2-1}
	\]
is analytic in $\mathbb{C}\setminus (-\infty, \; 1)$ and one has 
	\begin{equation}
	\label{c-root}
		 \lim_{\epsilon \to 0_+}S\big(\epsilon(1\pm i)\big)= {\rm e}^{\pm i\pi/2} \; .
	\end{equation}
\end{lemma}

\begin{lemma} The Fourier series of the Legendre function is given by
\begin{equation}
P_s(-\cos\psi)
\; = \; 
p(0)+
2\sum_{k=1}^\infty p(k) \cos(k\psi)
\;,
\label{eq:Fourier-Pmuk-1}
\end{equation}
where, for $k\in \mathbb{N}_0$, 
	\begin{equation}
		p(k) \doteq (-1)^k \frac{\Gamma(s -k + 1)}{\Gamma(s+k+1)} 
				\left( \lim_{\epsilon\to 0_+}P_s^k \big(\epsilon(1 + i)\big)\right)
				\left(\lim_{\epsilon\to 0_+} P_s^k \big(\epsilon(1 - i)\big)\right) 	\;.
	\label{eq:pk-1}
\end{equation}
\end{lemma}

\begin{proof}
For $| \arg (z-1)| <\pi$ and $| \arg (w-1)| <\pi$ and $\Re z >0$ and $\Re w >0$, one has \cite[page 202]{Lebedev}
	\begin{align}
		P_s\Big(zw& -\sqrt{z^2-1} \sqrt{w^2-1} \cos\psi\Big)
		\label{eq:addition-formula}
		 \\
		&	= P_s(z)P_s(w) + 2 \sum_{k=1}^\infty (-1)^k \frac{\Gamma(s -k + 1)}{\Gamma(s+k+1)} 
 				P_s^k (z) P_s^k (w) \cos(k\psi)\;.
		\nonumber
	\end{align}
(This relation is also found in \cite[page 78]{snow}.) Hence, setting $z=\epsilon(1+i)$, $w=\epsilon(1-i)$, and taking
the limit $\epsilon \searrow 0$, we have for the l.h.s.\ of
(\ref{eq:addition-formula}),
	\[ 
		\lim_{\epsilon \searrow 0_+}  P_s\Big(2\epsilon^2 -S\big(\epsilon(1+i)\big)S\big(\epsilon(1+i)\big) \cos\psi\Big)
			\; = \; 
		P_s\left( - \cos\psi\right) \; .
	\]
Setting $z=i\epsilon$,  $w=-i\epsilon$ and taking the limit
$\epsilon \searrow 0$ on the  r.h.s.\ of (\ref{eq:addition-formula}), the lemma follows.
\end{proof}

\begin{lemma}
\begin{equation}
\lim_{\epsilon\to 0_+}P_s^k \big(\epsilon(1 \pm i)\big)
=  
\frac{ {\rm e}^{\pm ik\pi/2}\sqrt{\pi} }{2^k}\,
\frac{\Gamma(s + k + 1)}{  \Gamma(s -k + 1)} 
\frac{1}{\Gamma\left(\frac{k-s+1}{2}\right)\Gamma\left(\frac{k+s}{2}+1\right)}
\;.
\label{eq:Jinboiui-2}
\end{equation}
\end{lemma}

\begin{proof}
According to \cite[Eq.\ 7.12.27, page 198]{Lebedev}, one has,
for $k\in \mathbb{N}_0$, $|z-1|<2$ and $\mbox{arg}(z-1)<\pi$,
\begin{align*}
P_s^k (z)
& =  
\frac{ \big(z^2-1\big)^{k/2} \Gamma(s + k + 1)}{2^k \Gamma(k + 1) \Gamma(s -k + 1)} \;
F\left( k-s, \; k+s+1, \; k+1; \; \frac{1-z}{2}\right)
\\
& =  
\frac{S(z)^k  \Gamma(s + k + 1)}{2^k \Gamma(k + 1) \Gamma(s -k + 1)}
F\left( k-s, \; k+s+1, \; k+1; \; \frac{1-z}{2}\right)
\;,
\end{align*}
where $F$ is the hypergeometric function

\begin{align}
F(\alpha , \; \beta, \; \gamma,\; z )
& \doteq
1 + 
\sum_{n=1}^\infty
\frac{(\alpha)_{n} (\beta)_{n} }{n! (\gamma)_{n}} \; z^n 
\nonumber \\
& = 
\frac{\Gamma(\gamma)}{\Gamma(\alpha)\Gamma(\beta)}
\sum_{n=0}^\infty
\frac{\Gamma(\alpha+n)\Gamma(\beta+n)}{\Gamma(\gamma+n)}
\; \frac{z^n}{n!} \; ,
\label{eq:def-hipergeometrica}
\end{align}
valid for $|z|<1$. 
Here $(q)_n$ is the Pochhammer symbol, which is defined by
	\[
	(q)_n \doteq 
				\begin{cases} 1 & \text{if \ $n=0$} \; ; \\
							q(q+1) \cdots (q+ n-1) & \text{if \ $n>0$}\; . 
				\end{cases} 
	\]
$F$  is analytic in the whole open unit disk $|z|<1$. Therefore,
	\begin{align}
		\lim_{\epsilon\to 0_+}P_s^k \big(\epsilon(1 \pm i)\big)
			&=  
			\frac{{\rm e}^{\pm ik\pi/2}}{2^k}
			\frac{\Gamma(s + k + 1)}{ \Gamma(k + 1) \Gamma(s -k + 1)} 
			\nonumber \\
			&\qquad \qquad \times F\left( k-s, \; k+s+1, \; k+1; \; \frac{1}{2}\right)
			\;.
		\label{eq:Jinboiui}
	\end{align}

The value of $F(\alpha, \; \beta, \; \gamma; \; z)$ at the point
$z=1/2$ cannot be easily computed from the power series definition
\eqref{eq:def-hipergeometrica}. However, the hypergeometric function satisfies the following
relation (see \cite[Eq.\ 9.6.11, page 253]{Lebedev}):
\begin{align}
\label{eq:smart-identity}
& F\left( 2\alpha , \; 2\beta, \; \alpha + \beta + \tfrac{1}{2};\; \tfrac{1-z}{2} \right)
\\
&\qquad  \qquad =  
\frac{\Gamma\left(\alpha+\beta+\frac{1}{2}\right)\Gamma\left(\frac{1}{2}\right)
  }{
\Gamma\left(\alpha+\frac{1}{2}\right)\Gamma\left(\beta+\frac{1}{2}\right)}
\, F\left(\alpha, \; \beta, \; \tfrac{1}{2}; \; z^2\right)
\nonumber \\
& \qquad \qquad  \qquad \qquad +
z \frac{\Gamma\left(\alpha+\beta+\frac{1}{2}\right)\Gamma\left(-\frac{1}{2}\right)
  }{
\Gamma\left(\alpha\right)\Gamma\left(\beta\right)}
\, F\left(\alpha+\tfrac{1}{2}, \; \beta+\tfrac{1}{2}, \; \tfrac{3}{2}; \; z^2\right)
\;,
\nonumber
\end{align}
valid for all $z\in\mathbb{C}\setminus\big( (-\infty, \; -1)\cup(1, \;
\infty)\big)$ and for all $\alpha+\beta+\frac{1}{2}\not\in -\mathbb{N}_0$
(\emph{i.e.}, for all $\alpha+\beta+\frac{1}{2}\neq 0, \; -1, \; -2,
\ldots$).  Taking $z=0$ in (\ref{eq:smart-identity}), one gets
\begin{equation}
F\left(2\alpha , \; 2\beta, \; \alpha + \beta + \frac{1}{2};\; \frac{1}{2} \right)
\; = \; 
\frac{\Gamma\left(\alpha+\beta+\frac{1}{2}\right)\Gamma\left(\frac{1}{2}\right)
  }{
\Gamma\left(\alpha+\frac{1}{2}\right)\Gamma\left(\beta+\frac{1}{2}\right)}
\;,
\label{eq:oijosidcn}
\end{equation}
since $F\left(\alpha, \; \beta, \; \frac{1}{2}; \; 0\right)=1$ 
(see (\ref{eq:def-hipergeometrica})). By choosing 
$$
\alpha \; = \; \frac{k-s}{2}
\quad \mbox{ and } \quad
\beta \; = \; \frac{k+s+1}{2}
$$
one has 
$\alpha+\beta+\frac{1}{2}  =  k+1$ (which is non-zero for $k\in\mathbb{N}_0$)
and it follows from (\ref{eq:oijosidcn}) that
	\begin{equation}
	\label{d020}
		F\left( k-s, \; k+s+1, \;   k+1; \; \frac{1}{2} \right)
		\; = \; 
		\frac{\sqrt{\pi}\, \Gamma(k+1)}{\Gamma\left(\frac{k-s+1}{2}\right)\Gamma\left(\frac{k+s}{2}+1\right)}
		\;,
	\end{equation}
since $\Gamma\left(\frac{1}{2}\right)=\sqrt{\pi}$.
Inserting \eqref{d020} into (\ref{eq:Jinboiui}), one gets \eqref{eq:Jinboiui-2}.
\end{proof}

\begin{proposition} 
\label{w-coefficients}
The Fourier series of the Legendre function is given by
\begin{equation}
	P_s(-\cos\psi) = p(0)+ 2\sum_{k=1}^\infty p(k) \cos(k\psi) \;,
\end{equation}
where, for $k\in \mathbb{N}_0$, 
\begin{equation}
	\label{eq:pk-6}
		p(k)  = -\frac{\sin\big(\pi s \big)}{\pi} \frac{1}{ (k +s)} \; 
				\frac{\Gamma\left( \frac{k-s}{2} \right)}
					{\Gamma\left(\frac{k+s}{2}\right)}
				\frac{\Gamma\left( \frac{k+s+1}{2} \right)}{\Gamma\left(\frac{k-s+1}{2}\right)} \; ,
			\qquad \forall \, k\in\mathbb{N}_0 \; .
\end{equation}
\end{proposition}

\begin{proof}
Inserting (\ref{eq:Jinboiui-2}) into 
(\ref{eq:pk-1}), one gets
\begin{equation}
p(k)
= 
(-1)^k \frac{\pi}{2^{2k}}\,
\frac{\Gamma(s + k + 1)}{  \Gamma(s -k + 1)}
\; 
\frac{1}{\Gamma\left(\frac{k-s+1}{2}\right)^2
\Gamma\left(\frac{k+s}{2}+1\right)^2}
\; ,
\quad k\in \mathbb{N}_0
\;.
\label{eq:pk-2}
\end{equation}
Now, using the well-known properties
(\ref{eq:gamma-1})--(\ref{eq:gamma-5}) of the Gamma function, we start
a series of manipulations, in order to write $p(k)$ in a more
convenient fashion.

In (\ref{eq:pk-2}) we consider the factor
\[
	\frac{\Gamma(s + k + 1)}{\Gamma\left(\frac{k+s}{2}+1\right)}
		= 
	\frac{\Gamma(s + k + 1)}{\Gamma\left(\frac{k+s+1}{2}+\frac{1}{2}\right)}
		=  
	\frac{\Gamma(2z)}{\Gamma\left(z+\frac{1}{2}\right)} \;,
\]
by taking $z=\frac{k+s+1}{2}$. From (\ref{eq:gamma-4}), one has
$
\frac{\Gamma(2z)}{\Gamma\left(z+\frac{1}{2}\right)}
=
\frac{2^{2z-1}}{\sqrt{\pi}}\Gamma(z)
$. Hence,
\[
		\frac{\Gamma(s + k + 1)}{\Gamma\left(\frac{k+s}{2}+1\right)}	
			= 
		\frac{2^{2z-1}}{\sqrt{\pi}}\Gamma(z)
			=  
		\frac{2^{k+s}}{\sqrt{\pi}}\Gamma\left( \frac{k+s+1}{2} \right) \;.
\]
Inserting this into (\ref{eq:pk-2}), we get
\begin{align}
	p(k)  = \;  \; &(-1)^k \sqrt{\pi}\, 2^{s-k} \; 
			\frac{\Gamma\left( \frac{k+s+1}{2} \right)}{\Gamma(s -k + 1)\Gamma\left(\frac{k-s+1}{2}\right)^2
			\Gamma\left(\frac{k+s}{2}+1\right)}
		\nonumber \\
		 \stackrel{(\ref{eq:gamma-1})}{=}  &(-1)^k 
			\frac{\sqrt{\pi}\, 2^{s-k+1}}{s^2 - k^2} \; 
			\frac{\Gamma\left( \frac{k+s+1}{2} \right)}{\Gamma(s -k)\Gamma\left(\frac{k-s+1}{2}\right)^2
			\Gamma\left(\frac{k+s}{2}\right)} \;.
\label{eq:pk-3}
\end{align}
Now, we write $\Gamma\left( \frac{k-s+1}{2} \right) = \Gamma\left(
  z+\frac{1}{2} \right)$ with $z=\frac{k-s}{2}$ and, using
(\ref{eq:gamma-4}) we get
\begin{align*}
\Gamma\left( \frac{k-s+1}{2} \right)^2 
& = 
\Gamma\left( z+\frac{1}{2} \right)^2
\stackrel{(\ref{eq:gamma-4})}{=} 
\left( 
     \frac{\Gamma(2z)}{\Gamma(z)} \frac{\sqrt{\pi}}{2^{2z-1}}
\right)^2
\\
&= 
\left( 
\frac{\sqrt{\pi}}{2^{k-s-1}}
\frac{\Gamma(k-s)}{\Gamma\left( \frac{k-s}{2} \right)}
\right)^2
=  
\frac{\pi}{2^{2k-2s-2}}
\frac{\Gamma(k-s)^2}{\Gamma\left( \frac{k-s}{2} \right)^2}
\;.
\end{align*}
Returning to (\ref{eq:pk-3}), we get from this
\begin{equation}
p(k) \; = \; 
(-1)^k  
\frac{2^{ k-s-1} }{\sqrt{\pi}\big(s^2 - k^2\big)} \; 
\frac{\Gamma\left( \frac{k+s+1}{2} \right)
      \Gamma\left( \frac{k-s}{2} \right)^2
}{
\Gamma(s -k) \Gamma(k-s)^2\,
\Gamma\left(\frac{k+s}{2}\right)
}
\;.
\label{eq:pk-4}
\end{equation}
We now write
	\[ 
		\Gamma(s -k) \Gamma(k-s)
			\; \stackrel{(\ref{eq:gamma-3})}{=} \;
				-\frac{\pi}{(s -k)\sin\big(\pi (s -k)\big)}
			\; = \; 
				\frac{(-1)^{k}\pi}{(k -s)\sin\big(\pi s \big)} \; ,
	\]
and inserting this into (\ref{eq:pk-4}) we get
\begin{equation}
p(k) \; = \; 
-\sin\big(\pi s \big)
\frac{2^{ k-s-1} }{\pi^{3/2} (k +s)} \; 
\frac{
\Gamma\left( \frac{k+s+1}{2} \right)
      \Gamma\left( \frac{k-s}{2} \right)^2
}{
\Gamma(k-s)\,
\Gamma\left(\frac{k+s}{2}\right)
}
\;.
\label{eq:pk-5}
\end{equation}
Taking $z=\frac{k-s}{2}$, we have
\[ 
	\frac{\Gamma\left( \frac{k-s}{2} \right)}{\Gamma(k-s)}
\; = \;
\frac{\Gamma(z)}{\Gamma(2z)}
\; \stackrel{(\ref{eq:gamma-4})}{=} \; 
\frac{\sqrt{\pi}}{2^{2z-1}}\;
\frac{1}{\Gamma\left(z+\frac{1}{2}\right)} 
\; = \; 
\frac{\sqrt{\pi}}{2^{k-s-1}}\;
\frac{1}{\Gamma\left(\frac{k-s+1}{2}\right)}
\; .
\]
Returning with this to (\ref{eq:pk-5}), we find \eqref{eq:pk-6}.
\end{proof}

\begin{remark}
Comparing (\ref{eq:Fourier-Pmuk-1}) with (\ref{eq:Fourrier-Pmu}) we see that
$p_k=\sqrt{2\pi r} p\big(|k|\big)$, for all $k\in\mathbb{Z}$. Thus,
from the definition (\ref{eq:definition-omegak-0}) we get
\begin{equation}
\label{eq:symmetry-of-tildeomega}
\widetilde{\omega}(k)=\widetilde{\omega}(-k)
\end{equation}
for all $k\in\mathbb{Z}$.  Actually, we can directly establish that
$p(k)=p(-k)$ for all $k$. This is the content of the next lemma.
\end{remark}

\begin{lemma}For all $k\in\mathbb{Z}$, we have
$p(k)=p(-k)$.
\end{lemma}

\begin{proof}
Until now we considered $k\in\mathbb{N}_0$ but, for $s\not\in\mathbb{Z}$,
(\ref{eq:pk-6}) is well-defined for all $k\in\mathbb{Z}$ and will now show
that $p(k)=p(-k)$ for all $k\in\mathbb{Z}$. Let 
	\[ 
			{\mathscr F}(k) \doteq \frac{1}{ (k +s)} \; \frac{ \Gamma\left( \frac{k-s}{2} \right) }{
											\Gamma\left(\frac{k+s}{2}\right) }
										\frac{\Gamma\left( \frac{k+s+1}{2} \right)}
											{\Gamma\left(\frac{k-s+1}{2}\right)} \;.
	\]
Then,
	\[
			{\mathscr F}(-k) = \frac{1}{ (-k +s)} \; \frac{ \Gamma\left( \frac{-k-s}{2} \right) }{
											\Gamma\left(\frac{-k+s}{2}\right) }
					\frac{\Gamma\left( \frac{-k+s+1}{2} \right)}{\Gamma\left(\frac{-k-s+1}{2}\right)} \;.
	\]
Now,
	\[
		\frac{ \Gamma\left( \frac{-k-s}{2} \right) }{
			\Gamma\left(\frac{-k+s}{2}\right)}
		= \frac{ \left( -k+s \right) }{\left( -k-s \right)}
			\frac{  \sin\left(\pi\frac{-k+s}{2}\right) }{
      				\sin\left( \pi\frac{-k-s}{2} \right) }
			\frac{\Gamma\left(\frac{k-s}{2}\right)}
				{ \Gamma\left( \frac{k+s}{2} \right)}
	\]
and
	\[
		\frac{ \Gamma\left( \frac{-k+s+1}{2} \right) }{
				\Gamma\left(\frac{-k-s+1}{2}\right)}	
		= \frac{ \Gamma\left( 1- \frac{k-s+1}{2} \right) }{
				\Gamma\left( 1-\frac{k+s+1}{2}\right)}
		= \frac{\sin\left( \pi\frac{k+s+1}{2}\right)}{\sin\left( \pi\frac{k-s+1}{2} \right)}
			\frac{\Gamma\left( \frac{k+s+1}{2}\right)}{
				\Gamma\left( \frac{k-s+1}{2} \right)} \;.
	\]
Therefore,
	\begin{align}
		\label{eq:yubuuytFytf}
		{\mathscr F}(-k) & = \frac{1}{ (-k +s)} \; 
					\frac{ \left( -k+s \right) }{
 						     \left( -k-s \right)}
				\left( \frac{ \sin\left(\pi\frac{-k+s}{2}\right)}{
    						  \sin\left( \pi\frac{-k-s}{2} \right)}
					\frac{\sin\left( \pi\frac{k+s+1}{2}\right)}{
						\sin\left(\pi \frac{k-s+1}{2} \right)}
				\right)
								 \\
			& \qquad \qquad \qquad \qquad \qquad \qquad \qquad \qquad	\times\left[	\frac{\Gamma\left(\frac{k-s}{2}\right)}{
						    \Gamma\left( \frac{k+s}{2} \right)}
					\frac{\Gamma\left( \frac{k+s+1}{2}\right)}{
							\Gamma\left( \frac{k-s+1}{2} \right)}
				\right] \nonumber \\
			& = \frac{1}{ -k -s} \;  \left( \frac{ \sin\left(\pi\frac{-k+s}{2}\right)}{
      										\sin\left( \pi\frac{-k-s}{2} \right)}
									\frac{\sin\left( \pi\frac{k+s+1}{2}\right)}{
										\sin\left(\pi \frac{k-s+1}{2} \right)}
								\right)
				 \left[\frac{\Gamma\left(\frac{k-s}{2}\right)}{	
						\Gamma\left( \frac{k+s}{2} \right)}
					\frac{\Gamma\left( \frac{k+s+1}{2}\right)}{
						\Gamma\left( \frac{k-s+1}{2} \right)}
				\right] \;. \nonumber
	\end{align}
Using $\sin (a) \sin (b) = \frac{\cos(a-b)-\cos(a+b)}{2}$, we get
	\begin{align*} 
	\frac{ \sin\left(\pi\frac{-k+s}{2}\right)}{ \sin\left( \pi\frac{-k-s}{2} \right)}
	\frac{\sin\left( \pi\frac{k+s+1}{2}\right)}{\sin\left( \pi\frac{k-s+1}{2} \right)}
	&= \frac{\cos \left( -\pi k -\frac{\pi}{2} \right) - \cos\left( \pi s +\frac{\pi}{2} \right) }
		{\cos\left( -\pi k -\frac{\pi}{2}\right) - \cos\left( -\pi s +\frac{\pi}{2}\right)}
	\\
	&=  \frac{\cos\left( \pi s +\frac{\pi}{2}\right)}{\cos\left( \pi s -\frac{\pi}{2}\right)}
	=  -1 \;.
	\end{align*}
Hence, returning to (\ref{eq:yubuuytFytf}),
	\[
		{\mathscr F}(-k) =  \frac{1}{ k +s} \; 
				\frac{\Gamma\left(\frac{k-s}{2}\right)}
					{\Gamma\left( \frac{k+s}{2} \right)}
				\frac{\Gamma\left( \frac{k+s+1}{2}\right)}{
					\Gamma\left( \frac{k-s+1}{2} \right)}
			=  {\mathscr F}(k) \;.
	\]
This establishes that $p(k)=p(-k)$ for all $k\in\mathbb{Z}$.
\end{proof}

\begin{proposition} Let $\widehat{{\mathfrak h}}  (S^1)$ be defined in \eqref{eq:def-scalar-product-2}, and let 
$f, g \in {\mathscr D}(\omega)$. It follows that  
\label{prop:E5}
\begin{equation}
\langle \omega r f , \omega r g \rangle_{\widehat{{\mathfrak h}}  (S^1)} = -\frac{r^2}{2 \sin(\pi s^+)}
\int_{S^1 \times S^1} \kern-.3cm {\rm d} \psi\, {\rm d} \psi' \; \overline{f(\psi')} \, 
P_{s}'\big( - \cos(\psi' - \psi) \big) \, g(\psi)
\;.
\label{eq:def-scalar-product-2b-BB}
\end{equation}
\end{proposition}

\begin{proof} In what follows we will denote the coefficients $p(k)$, $p_k$ and
$\omega(k)$ by $p_s(k)$, $p_{k, \, s}$ and $\omega_s(k)$, respectively.

For $s\in \mathbb{C}\setminus\mathbb{Z}$, define
\begin{align*}
			\langle \langle f ,\, \; g  \rangle \rangle
			& \doteq   c_\nu \int_{S^1} r\, d\psi \int_{S^1} r \, d\psi' \; \overline{f(\psi')} \, g(\psi)
				P_{s}' \big( - \cos(\psi' - \psi) \big) 
				\\
			& = -\frac{1}{  2 
			\sin(\pi s^+)}  \int_{S^1} r\, d\psi \int_{S^1} r \, d\psi' \; \overline{f(\psi')} \, g(\psi)
							P_{s}'\big( - \cos(\psi' - \psi) \big)   \;.
\end{align*}
We write
	\begin{equation*}
		P_{s}'\big( - \cos(\varphi) \big) =  \sum_{k\in\mathbb{Z}} p_{k, \,s}^1 \; \frac{{\rm e}^{ik\varphi}}{\sqrt{2\pi r}} \; ,
	\end{equation*}
and, as in (\ref{eq:Jhoitrytfsd}), we get
	\begin{equation}
			\langle \langle f ,\, \; g  \rangle \rangle
		 =  -\frac{\sqrt{2\pi r}}{  2   \sin(\pi s^+)} \sum_{k\in\mathbb{Z}} p_{k, \,s}^1 \, \overline{f_k} \, g_k \;,
		\label{eq:Jhoitrytfsd-BB}
	\end{equation}
where $f_k$ and $g_k$ are the Fourier coefficients of $f$ and $g$, respectively, i.e.,
	\[
		f_k \doteq \int_{S^1} r \, d\psi' \;  f(\psi') \frac{{\rm e}^{-ik\psi'}}{\sqrt{2\pi}}
			\qquad \mbox{ and } \qquad
		g_k \doteq \int_{S^1} r \, d\psi \;  g(\psi) \frac{{\rm e}^{-ik\psi}}{\sqrt{2\pi}} \; .
	\]
Taking the mixed derivatives $\partial_z\partial_w$ of both sides in
(\ref{eq:addition-formula}), we get
\begin{align}
	&P_s''\bigl( zw-\sqrt{z^2-1} \sqrt{w^2-1} \cos\psi \bigr)
			\left(  w- \tfrac{z \sqrt{w^2-1}}{\sqrt{z^2-1}}  \cos\psi  \right)
			\left(  z- \tfrac{w\sqrt{z^2-1} }{\sqrt{w^2-1}}  \cos\psi  \right)
		\nonumber \\
		& \qquad + 	P_s' \bigl( zw-\sqrt{z^2-1} \sqrt{w^2-1} \cos\psi \bigr)
					\left(  1- \frac{zw}{\sqrt{z^2-1}\sqrt{w^2-1}} \cos\psi  \right)
		\nonumber \\
		&= P_s'(z)P_s'(w) + 2 \sum_{k=1}^\infty (-1)^k
				\frac{\Gamma(s -k + 1)}{\Gamma(s+k+1)} 
				{P_s^k }'(z) {P_s^k} '(w) \cos(k\psi)\;.
\label{eq:addition-formula-BB}
\end{align}
Writing  $z=\epsilon(1+i)$, $w=\epsilon(1-i)$ (with $\epsilon>0$) and taking
the limit $\epsilon\to0_+$, we get from (\ref{eq:addition-formula-BB})
\begin{align*}
	P_s' ( - \cos\psi ) & =  \bigl(P_s'(0)\bigr)^2 + 2 \sum_{k=1}^\infty 
								(-1)^k \frac{\Gamma(s -k + 1)}{\Gamma(s+k+1)} 
								\left( \lim_{\epsilon_\to 0_+} {P_s^k }'\big( \epsilon(1+i) \big) \right) 
						\\
						& \qquad \qquad \qquad \qquad \qquad \qquad \times 
								\left( \lim_{\epsilon_\to 0_+} {P_s^k} '\big( \epsilon(1-i) \big) \right) 
								\cos(k\psi)
						\\
						& =  p^1_s(0) + 2 \sum_{k=1}^\infty p^1_s(k)\cos(k\psi) \;,
\end{align*}
with
	\[
	p^1_s(k) \doteq (-1)^k \frac{\Gamma(s -k + 1)}{\Gamma(s+k+1)} 
						\left( \lim_{\epsilon_\to 0_+} {P_s^k }'\big( \epsilon(1+i) \big) \right) 
						\left( \lim_{\epsilon_\to 0_+} {P_s^k} '\big( \epsilon(1-i) \big) \right) \;.
	\]
Now, according to \cite[Eq.\ (7.12.16), page 195]{Lebedev}, one has
	\[
		\bigl( z^2 - 1 \bigr) {P_s^k }'\big( z \big) = s z P_s^k ( z )
											- (s+k)P_{s-1}^k ( z ) \;, \qquad k\in \mathbb{N}_0 \; . 
	\]
Hence,
	\[
		\lim_{\epsilon_\to 0_+} {P_s^k }'\bigl( \epsilon(1\pm i) \bigr) 
			= (s+k)\lim_{\epsilon_\to 0_+} P_{s-1}^k \bigr(  \epsilon(1\pm i) \bigr) \; ,
	\]
and from this we have
\begin{align*}
		p^1_s(k) & = (-1)^k (s+k)^2 \frac{\Gamma(s -k + 1)}{\Gamma(s+k+1)} 
					\left(  \lim_{\epsilon_\to 0_+} P_{s-1}^k \bigl(  \epsilon(1 + i) \bigr) \right)
					\\
				& \qquad \qquad \qquad \qquad \qquad \qquad \qquad \times 
					\left( \lim_{\epsilon_\to 0_+} P_{s-1}^k  \bigl(  \epsilon(1 - i) \bigr) \right)
					\\
				& = (-1)^k (s+k)(s-k)  \frac{\Gamma(s -k )}{\Gamma(s+k)} 
					\left( \lim_{\epsilon_\to 0_+} P_{s-1}^k \bigl(  \epsilon(1 + i) \bigr) \right)
					\\
				& \qquad \qquad \qquad \qquad \qquad \qquad \qquad \times 
					\left( \lim_{\epsilon_\to 0_+} P_{s-1}^k \bigl(  \epsilon(1 - i) \bigr) \right)
					\\
				& \stackrel{(\ref{eq:pk-1})}{=} (s+k)(s-k) p_{s-1}(k) \; .
\end{align*}
Since $p_{s-1}(k)=p_{s-1}(-k)$, $k\in\mathbb{Z}$, it follows from the last equality that
$p^1_s(k)=p^1_s(-k)$,  $k\in\mathbb{Z}$, and we have
\[
	P_s' ( - \cos\psi ) = \sum_{k\in\mathbb{Z}} p_{k, \, s}^1 \frac{ {\rm e}^{ik\psi} }{\sqrt{2\pi r}}
\]
with
\[
 	p_{k, \, s}^1 = \sqrt{2\pi r} \; p^1_s(k) = \sqrt{2\pi r} \; (s+k)(s-k) p_{s-1}(k) \; ,
	\qquad k\in\mathbb{Z} \; . 
\]
Now, we have
\begin{align}
		(s+k)(s-k) p_{s-1}(k) & \stackrel{(\ref{eq:pk-6})}{=} 
						 \frac{-\sin\big(\pi s -\pi\big)}{\pi} \frac{(s+k)(s-k)}{ (k +s-1)} \; 
						\label{eq:IUbuytytryxsga}
						\\
						& \qquad \qquad \qquad \qquad \times \frac{ \Gamma\left( \frac{k-s+1}{2} \right)}
										{\Gamma\left(\frac{k+s-1}{2}\right)}
									\frac{\Gamma\left( \frac{k+s}{2} \right)}
										{\Gamma\left(\frac{k-s+2}{2}\right)} \; . 
										\nonumber
\end{align}
Since $\sin (\pi s -\pi )=-\sin (\pi s )$, 
$\Gamma\left(\frac{k-s+2}{2}\right)=\Gamma\left(\frac{k-s}{2}+1\right)
= \frac{k-s}{2} \Gamma\left(\frac{k-s}{2}\right)$ 
and
$(k +s-1)\Gamma\left(\frac{k+s-1}{2}\right)=2\Gamma\left(\frac{k+s-1}{2}+1\right)
=
2\Gamma\left(\frac{k+s+1}{2}\right)
$, relation (\ref{eq:IUbuytytryxsga}) becomes
	\[
		(s+k)(s-k) p_{s-1}(k) = -(s+k) 
					\frac{\sin\big(\pi s\big)}{\pi}
					\frac{\Gamma\left( \frac{k-s+1}{2} \right)}
						{\Gamma\left(\frac{k+s+1}{2}\right)}
					\frac{\Gamma\left( \frac{k+s}{2} \right)}{\Gamma\left(\frac{k-s}{2}\right)} \;.
	\]
Comparing with 
(\ref{eq:omega-1}), we find
	\[
		(s+k)(s-k) p_{s-1}(k) =  - \frac{\sin (\pi s )}{\pi} \; \widetilde \omega_s(k) \, r \; .
	\]
Here $\omega_s \equiv \omega$, with the index indicating the 
dependence of $\omega$ on $s$. The latter  had been suppressed in the main text. 
Hence,
	\[  
		p_{k, \, s^+}^1 = -\sqrt{\frac{2 r }{\pi}}\sin (\pi s )\widetilde \omega_s(k) \, r
	\]
Returning to (\ref{eq:Jhoitrytfsd-BB}) we have
	\begin{align*}
			\langle \langle f ,\, \; g  \rangle \rangle & 
			 =  -\frac{\sqrt{2\pi r}}{ 2  \sin(\pi s^+)} \sum_{k\in\mathbb{Z}} p_{k, \,s}^1 \, \overline{f_k} \, g_k  \\
		&		=  r^2 \left\langle  f, \; \omega g \right\rangle_{L^2(S^1, \; r d\psi)} \\
		&  = 	 r^2  \left\langle \omega f, \; \omega g \right\rangle_{\widehat{{\mathfrak h}}  (S^1)} 
		\;.
	\end{align*}
\end{proof}

\backmatter

\bibliographystyle{amsalpha}

\printindex

\end{document}